\documentclass[12pt]{article}
\pdfoutput=1
\usepackage{amsmath}
\usepackage{graphicx,psfrag,epsf}
\usepackage{enumerate}
\usepackage{newclude} 
\usepackage[square,sort]{natbib}
\usepackage{url}

\newcommand{\blind}{1}

\RequirePackage[paper = a4paper,         %To be shrunk to 80% when printed
            hmargin = 25mm,          %Left and right margin
            vmargin = 25mm,          %Top and bottom margin
            dvipdfm]{geometry}

\clearpage{}
\RequirePackage[english]{babel}
\usepackage[autostyle=true,english=british]{csquotes}

\RequirePackage[]{color}
 \makeatletter \def\maxwidth{ 
   \ifdim\Gin@nat@width>\linewidth \linewidth \else \Gin@nat@width \fi
 } \makeatother

\definecolor{fgcolor}{rgb}{0, 0, 0}

\RequirePackage{framed}


\definecolor{shadecolor}{rgb}{.97, .97, .97}
\definecolor{messagecolor}{rgb}{0, 0, 0}
\definecolor{warningcolor}{rgb}{1, 0, 1}
\definecolor{errorcolor}{rgb}{1, 0, 0}

\RequirePackage{alltt}

\RequirePackage{pdflscape}
\RequirePackage{mathtools}

\RequirePackage{fancyhdr}
\RequirePackage{datetime} 
\RequirePackage{last page}

\numberwithin{equation}{section}

\RequirePackage{verbatim}  
\RequirePackage{bbm}     
\RequirePackage[final]{showkeys} 
\RequirePackage{bm} 
\RequirePackage{xspace} 
\RequirePackage{framed}

\RequirePackage{lineno}  

\RequirePackage{enumitem} 
\setlist{nolistsep} 

\RequirePackage[
            colorlinks,      
            citecolor=blue,  
            linkcolor=magenta,  
            urlcolor=blue,   
            ]{hyperref}

\hypersetup{citecolor=blue, urlcolor=magenta, colorlinks=true} 

\RequirePackage[sf,medium,bf,noindentafter]{titlesec}  
\titlespacing{\section}{0pt}{12pt plus 4pt minus 2pt}{0pt plus 2pt minus 2pt}
\titlespacing{\subsection}{0pt}{12pt plus 4pt minus 2pt}{0pt plus 2pt minus 2pt}
\titlespacing{\subsubsection}{0pt}{12pt plus 4pt minus 2pt}{0pt plus 2pt minus 2pt}
\titlespacing{\paragraph}{0pt}{\parskip}{-\parskip}

\RequirePackage{tensor} 
\RequirePackage{etoolbox} 

\RequirePackage{xparse} 

\RequirePackage{trimclip} 
\RequirePackage{adjustbox}

\RequirePackage{fouridx}

\IfFileExists{upquote.sty}{\RequirePackage{upquote}}{}
\RequirePackage{amsthm}

\makeatletter
\newcommand\appendix@section[1]{
  \refstepcounter{section}
  \orig@section*{Appendix \@Alph\c@section: #1}
  \addcontentsline{toc}{section}{Appendix \@Alph\c@section: #1}
}
\let\orig@section\section
\g@addto@macro\appendix{\let\section\appendix@section}
\makeatother

\clearpage{}
\clearpage{}

\pdfoutput=1

\graphicspath{{figure/}}
\graphicspath{{/}{../}}

\makeatletter{}

\RequirePackage{amssymb}
\RequirePackage{amsfonts}

\RequirePackage{xstring}

\newtheorem{theorem}{Theorem}[section]

{
  \theoremstyle{algorithm}
  \newtheorem{algorithm}{Algorithm}[section]
}

{
  \theoremstyle{assumption}
  \newtheorem{assumption}{Assumption}[section]
}

{
  \theoremstyle{definition}
  \newtheorem{definition}{Definition}[section]
}

{
  \theoremstyle{corollary}
  
}

{
  \theoremstyle{lemma}
  \newtheorem{lemma}{Lemma}[section]
}

{
  \theoremstyle{remark}
  
}

\RequirePackage[]{cleveref}

\newcommand{\myref}[2]{\mbox{\cref{#1}\eqref{#2}}\xspace
}

\newrobustcmd{\MacroNote}[1]{}

\MacroNote{Text written like this is contained in the macro MacroNote,
  and this is intended as a reminder concerning how to call the
  different commands. Might never be used...}

\newrobustcmd{\InternalNote}[2][Internal note:]{
  \rule{\linewidth}{.25pt} \\
  \textbf{#1} #2\\[-.25cm]
  \rule{\linewidth}{.25pt}
  \vspace*{\lineskip}
}

\newcommand{\E}[2][]{\ensuremath{
    \operatorname{E}_{#1}\! \left[#2\right]}}
\newcommand{\Prob}[1]{\ensuremath{
    \operatorname{P}\! \left(#1\right)}}

\newcommand{\Cov}[2]{\ensuremath{
    \operatorname{Cov}\!\left(#1,#2\right)}}

\newcommand{\Ind}[1]{\ensuremath{ \mathbbm{1}{\left\{#1\right\}}}}

\renewcommand{\d}[1]{\operatorname{d}\!#1}

\newcommand{\UVN}[2]{\ensuremath{
    \operatorname{N}\!\left(#1,#2\right)}}

\newcommand{\R}{\texttt{R}\xspace}
\newcommand{\Rfunction}{\mbox{\R-function}\xspace}
\newcommand{\Rpackage}{\mbox{\R-package}\xspace}

\newcommand{\Rref}[1]{\mbox{\texttt{#1}\xspace}}

\newcommand{\lgsdRpackage}{\Rref{localgaussSpec}}
\newcommand{\devtoolgithublgsdRpackage}{\Rref{remotes:\!:install\_github("LAJordanger/localgaussSpec")}}

\newcommand{\NN}{\ensuremath{\mathbb{N}}} \newcommand{\RR}{\ensuremath{\mathbb{R}}}  \newcommand{\ZZ}{\ensuremath{\mathbb{Z}}}

\newcommand{\defeq}{\coloneqq}

\newcommand{\ssi}[1]{\mbox{\footnotesize
    \IfStrEq{#1}{}
    {}{\ensuremath{{}_{#1}}}}}

\newcommand{\ceil}[1]{\ensuremath{\left\lceil #1 \right\rceil}}

\newcommand{\ssub}[3]{\mbox{\tiny
        \IfStrEq{#2}{}
                {}{\ensuremath{{#1#2#3}}}}} 
\newcommand{\ssup}[3]{\mbox{\tiny
        \IfStrEq{#2}{}
                {}{\ensuremath{{#1#2#3}}}}}

\newcommand{\SSI}[7]{\ensuremath{
        \tensor*{#1}{^{\ssup{#2}{#3}{#4}}_{\ssub{#5}{#6}{#7}}}}}

\newcommand{\pSSI}[9]{\ensuremath{
        \tensor*[^{\ssup{#8}{p}{#9}}]{#1}{^{\ssup{#2}{#3}{#4}}_{\ssub{#5}{#6}{#7}}}}}

\newcommand{\fs}[3][]{\SSI{f}{}{#1}{}{}{#2}{}
        \IfStrEq{#3}{}{}{\left(#3\right)}}
\newcommand{\Fs}[3][]{\SSI{F}{}{#1}{}{}{#2}{}
        \IfStrEq{#3}{}{}{\left(#3\right)}}
\newcommand{\fbs}[3][]{\SSI{\bm{f}}{}{#1}{}{}{#2}{}
        \IfStrEq{#3}{}{}{\left(#3\right)}}
\newcommand{\Fbs}[3][]{\SSI{\bm{F}}{}{#1}{}{}{#2}{}
        \IfStrEq{#3}{}{}{\!\! \left(#3\right)}}

\newcommand{\cs}[3][]{\SSI{c}{}{#1}{}{}{#2}{}
        \IfStrEq{#3}{}{}{\left(#3\right)}}
\newcommand{\Cs}[3][]{\SSI{C}{}{#1}{}{}{#2}{}
        \IfStrEq{#3}{}{}{\left(#3\right)}}
\newcommand{\cbs}[3][]{\SSI{\bm{c}}{}{#1}{}{}{#2}{}
        \IfStrEq{#3}{}{}{\left(#3\right)}}
\newcommand{\Cbs}[3][]{\SSI{\bm{C}}{}{#1}{}{}{#2}{}
        \IfStrEq{#3}{}{}{\left(#3\right)}}

\newcommand{\pFs}[3][]{\pSSI{F}{}{#1}{}{}{#2}{}{}{}
        \IfStrEq{#3}{}{}{\left(#3\right)}}

\newrobustcmd{\subp}[5]{ 
      \def\somethingThere{#2#3#4#5}
  \IfStrEq{\somethingThere}{}{#1}{
    \ensuremath{#1_{#2#3}^{#4#5}}}
}

\newrobustcmd{\partialz}[2][]{
  \ensuremath{\subp{\partial}{}{#2}{}{#1}}}

\newrobustcmd{\onez}[2][]{
  \ensuremath{\subp{1}{}{#2}{}{#1}}}
\newrobustcmd{\bmonez}[2][]{
  \ensuremath{\subp{\bm{1}}{}{#2}{}{#1}}}
\newrobustcmd{\zeroz}[2][]{
  \ensuremath{\subp{0}{}{#2}{}{#1}}}
\newrobustcmd{\bmzeroz}[2][]{
  \ensuremath{\subp{\bm{0}}{}{#2}{}{#1}}}

\newrobustcmd{\az}[2][]{\ensuremath{\subp{a}{}{#2}{}{#1}}}
\newrobustcmd{\hataz}[2][]{\ensuremath{\subp{\hat{a}}{}{#2}{}{#1}}}
\newrobustcmd{\widehataz}[2][]{\ensuremath{\subp{\widehat{a}}{}{#2}{}{#1}}}
\newrobustcmd{\checkaz}[2][]{\ensuremath{\subp{\check{a}}{}{#2}{}{#1}}}
\newrobustcmd{\tildeaz}[2][]{\ensuremath{\subp{\tilde{a}}{}{#2}{}{#1}}}
\newrobustcmd{\widetildeaz}[2][]{\ensuremath{\subp{\widetilde{a}}{}{#2}{}{#1}}}
\newrobustcmd{\acuteaz}[2][]{\ensuremath{\subp{\acute{a}}{}{#2}{}{#1}}}
\newrobustcmd{\graveaz}[2][]{\ensuremath{\subp{\grave{a}}{}{#2}{}{#1}}}
\newrobustcmd{\dotaz}[2][]{\ensuremath{\subp{\dot{a}}{}{#2}{}{#1}}}
\newrobustcmd{\ddotaz}[2][]{\ensuremath{\subp{\ddot{a}}{}{#2}{}{#1}}}
\newrobustcmd{\breveaz}[2][]{\ensuremath{\subp{\breve{a}}{}{#2}{}{#1}}}
\newrobustcmd{\baraz}[2][]{\ensuremath{\subp{\bar{a}}{}{#2}{}{#1}}}
\newrobustcmd{\vecaz}[2][]{\ensuremath{\subp{\vec{a}}{}{#2}{}{#1}}}
\newrobustcmd{\bmaz}[2][]{\ensuremath{\subp{\bm{a}}{}{#2}{}{#1}}}
\newrobustcmd{\hatbmaz}[2][]{\ensuremath{\subp{\hat{\bm{a}}}{}{#2}{}{#1}}}
\newrobustcmd{\widehatbmaz}[2][]{\ensuremath{\subp{\widehat{\bm{a}}}{}{#2}{}{#1}}}
\newrobustcmd{\checkbmaz}[2][]{\ensuremath{\subp{\check{\bm{a}}}{}{#2}{}{#1}}}
\newrobustcmd{\tildebmaz}[2][]{\ensuremath{\subp{\tilde{\bm{a}}}{}{#2}{}{#1}}}
\newrobustcmd{\widetildebmaz}[2][]{\ensuremath{\subp{\widetilde{\bm{a}}}{}{#2}{}{#1}}}
\newrobustcmd{\acutebmaz}[2][]{\ensuremath{\subp{\acute{\bm{a}}}{}{#2}{}{#1}}}
\newrobustcmd{\gravebmaz}[2][]{\ensuremath{\subp{\grave{\bm{a}}}{}{#2}{}{#1}}}
\newrobustcmd{\dotbmaz}[2][]{\ensuremath{\subp{\dot{\bm{a}}}{}{#2}{}{#1}}}
\newrobustcmd{\ddotbmaz}[2][]{\ensuremath{\subp{\ddot{\bm{a}}}{}{#2}{}{#1}}}
\newrobustcmd{\brevebmaz}[2][]{\ensuremath{\subp{\breve{\bm{a}}}{}{#2}{}{#1}}}
\newrobustcmd{\barbmaz}[2][]{\ensuremath{\subp{\bar{\bm{a}}}{}{#2}{}{#1}}}
\newrobustcmd{\vecbmaz}[2][]{\ensuremath{\subp{\vec{\bm{a}}}{}{#2}{}{#1}}}
\newrobustcmd{\bz}[2][]{\ensuremath{\subp{b}{}{#2}{}{#1}}}
\newrobustcmd{\hatbz}[2][]{\ensuremath{\subp{\hat{b}}{}{#2}{}{#1}}}
\newrobustcmd{\widehatbz}[2][]{\ensuremath{\subp{\widehat{b}}{}{#2}{}{#1}}}
\newrobustcmd{\checkbz}[2][]{\ensuremath{\subp{\check{b}}{}{#2}{}{#1}}}
\newrobustcmd{\tildebz}[2][]{\ensuremath{\subp{\tilde{b}}{}{#2}{}{#1}}}
\newrobustcmd{\widetildebz}[2][]{\ensuremath{\subp{\widetilde{b}}{}{#2}{}{#1}}}
\newrobustcmd{\acutebz}[2][]{\ensuremath{\subp{\acute{b}}{}{#2}{}{#1}}}
\newrobustcmd{\gravebz}[2][]{\ensuremath{\subp{\grave{b}}{}{#2}{}{#1}}}
\newrobustcmd{\dotbz}[2][]{\ensuremath{\subp{\dot{b}}{}{#2}{}{#1}}}
\newrobustcmd{\ddotbz}[2][]{\ensuremath{\subp{\ddot{b}}{}{#2}{}{#1}}}
\newrobustcmd{\brevebz}[2][]{\ensuremath{\subp{\breve{b}}{}{#2}{}{#1}}}
\newrobustcmd{\barbz}[2][]{\ensuremath{\subp{\bar{b}}{}{#2}{}{#1}}}
\newrobustcmd{\vecbz}[2][]{\ensuremath{\subp{\vec{b}}{}{#2}{}{#1}}}
\newrobustcmd{\bmbz}[2][]{\ensuremath{\subp{\bm{b}}{}{#2}{}{#1}}}
\newrobustcmd{\hatbmbz}[2][]{\ensuremath{\subp{\hat{\bm{b}}}{}{#2}{}{#1}}}
\newrobustcmd{\widehatbmbz}[2][]{\ensuremath{\subp{\widehat{\bm{b}}}{}{#2}{}{#1}}}
\newrobustcmd{\checkbmbz}[2][]{\ensuremath{\subp{\check{\bm{b}}}{}{#2}{}{#1}}}
\newrobustcmd{\tildebmbz}[2][]{\ensuremath{\subp{\tilde{\bm{b}}}{}{#2}{}{#1}}}
\newrobustcmd{\widetildebmbz}[2][]{\ensuremath{\subp{\widetilde{\bm{b}}}{}{#2}{}{#1}}}
\newrobustcmd{\acutebmbz}[2][]{\ensuremath{\subp{\acute{\bm{b}}}{}{#2}{}{#1}}}
\newrobustcmd{\gravebmbz}[2][]{\ensuremath{\subp{\grave{\bm{b}}}{}{#2}{}{#1}}}
\newrobustcmd{\dotbmbz}[2][]{\ensuremath{\subp{\dot{\bm{b}}}{}{#2}{}{#1}}}
\newrobustcmd{\ddotbmbz}[2][]{\ensuremath{\subp{\ddot{\bm{b}}}{}{#2}{}{#1}}}
\newrobustcmd{\brevebmbz}[2][]{\ensuremath{\subp{\breve{\bm{b}}}{}{#2}{}{#1}}}
\newrobustcmd{\barbmbz}[2][]{\ensuremath{\subp{\bar{\bm{b}}}{}{#2}{}{#1}}}
\newrobustcmd{\vecbmbz}[2][]{\ensuremath{\subp{\vec{\bm{b}}}{}{#2}{}{#1}}}
\newrobustcmd{\cz}[2][]{\ensuremath{\subp{c}{}{#2}{}{#1}}}
\newrobustcmd{\hatcz}[2][]{\ensuremath{\subp{\hat{c}}{}{#2}{}{#1}}}
\newrobustcmd{\widehatcz}[2][]{\ensuremath{\subp{\widehat{c}}{}{#2}{}{#1}}}
\newrobustcmd{\checkcz}[2][]{\ensuremath{\subp{\check{c}}{}{#2}{}{#1}}}
\newrobustcmd{\tildecz}[2][]{\ensuremath{\subp{\tilde{c}}{}{#2}{}{#1}}}
\newrobustcmd{\widetildecz}[2][]{\ensuremath{\subp{\widetilde{c}}{}{#2}{}{#1}}}
\newrobustcmd{\acutecz}[2][]{\ensuremath{\subp{\acute{c}}{}{#2}{}{#1}}}
\newrobustcmd{\gravecz}[2][]{\ensuremath{\subp{\grave{c}}{}{#2}{}{#1}}}
\newrobustcmd{\dotcz}[2][]{\ensuremath{\subp{\dot{c}}{}{#2}{}{#1}}}
\newrobustcmd{\ddotcz}[2][]{\ensuremath{\subp{\ddot{c}}{}{#2}{}{#1}}}
\newrobustcmd{\brevecz}[2][]{\ensuremath{\subp{\breve{c}}{}{#2}{}{#1}}}
\newrobustcmd{\barcz}[2][]{\ensuremath{\subp{\bar{c}}{}{#2}{}{#1}}}
\newrobustcmd{\veccz}[2][]{\ensuremath{\subp{\vec{c}}{}{#2}{}{#1}}}
\newrobustcmd{\bmcz}[2][]{\ensuremath{\subp{\bm{c}}{}{#2}{}{#1}}}
\newrobustcmd{\hatbmcz}[2][]{\ensuremath{\subp{\hat{\bm{c}}}{}{#2}{}{#1}}}
\newrobustcmd{\widehatbmcz}[2][]{\ensuremath{\subp{\widehat{\bm{c}}}{}{#2}{}{#1}}}
\newrobustcmd{\checkbmcz}[2][]{\ensuremath{\subp{\check{\bm{c}}}{}{#2}{}{#1}}}
\newrobustcmd{\tildebmcz}[2][]{\ensuremath{\subp{\tilde{\bm{c}}}{}{#2}{}{#1}}}
\newrobustcmd{\widetildebmcz}[2][]{\ensuremath{\subp{\widetilde{\bm{c}}}{}{#2}{}{#1}}}
\newrobustcmd{\acutebmcz}[2][]{\ensuremath{\subp{\acute{\bm{c}}}{}{#2}{}{#1}}}
\newrobustcmd{\gravebmcz}[2][]{\ensuremath{\subp{\grave{\bm{c}}}{}{#2}{}{#1}}}
\newrobustcmd{\dotbmcz}[2][]{\ensuremath{\subp{\dot{\bm{c}}}{}{#2}{}{#1}}}
\newrobustcmd{\ddotbmcz}[2][]{\ensuremath{\subp{\ddot{\bm{c}}}{}{#2}{}{#1}}}
\newrobustcmd{\brevebmcz}[2][]{\ensuremath{\subp{\breve{\bm{c}}}{}{#2}{}{#1}}}
\newrobustcmd{\barbmcz}[2][]{\ensuremath{\subp{\bar{\bm{c}}}{}{#2}{}{#1}}}
\newrobustcmd{\vecbmcz}[2][]{\ensuremath{\subp{\vec{\bm{c}}}{}{#2}{}{#1}}}
\newrobustcmd{\dz}[2][]{\ensuremath{\subp{d}{}{#2}{}{#1}}}
\newrobustcmd{\hatdz}[2][]{\ensuremath{\subp{\hat{d}}{}{#2}{}{#1}}}
\newrobustcmd{\widehatdz}[2][]{\ensuremath{\subp{\widehat{d}}{}{#2}{}{#1}}}
\newrobustcmd{\checkdz}[2][]{\ensuremath{\subp{\check{d}}{}{#2}{}{#1}}}
\newrobustcmd{\tildedz}[2][]{\ensuremath{\subp{\tilde{d}}{}{#2}{}{#1}}}
\newrobustcmd{\widetildedz}[2][]{\ensuremath{\subp{\widetilde{d}}{}{#2}{}{#1}}}
\newrobustcmd{\acutedz}[2][]{\ensuremath{\subp{\acute{d}}{}{#2}{}{#1}}}
\newrobustcmd{\gravedz}[2][]{\ensuremath{\subp{\grave{d}}{}{#2}{}{#1}}}
\newrobustcmd{\dotdz}[2][]{\ensuremath{\subp{\dot{d}}{}{#2}{}{#1}}}
\newrobustcmd{\ddotdz}[2][]{\ensuremath{\subp{\ddot{d}}{}{#2}{}{#1}}}
\newrobustcmd{\brevedz}[2][]{\ensuremath{\subp{\breve{d}}{}{#2}{}{#1}}}
\newrobustcmd{\bardz}[2][]{\ensuremath{\subp{\bar{d}}{}{#2}{}{#1}}}
\newrobustcmd{\vecdz}[2][]{\ensuremath{\subp{\vec{d}}{}{#2}{}{#1}}}
\newrobustcmd{\bmdz}[2][]{\ensuremath{\subp{\bm{d}}{}{#2}{}{#1}}}
\newrobustcmd{\hatbmdz}[2][]{\ensuremath{\subp{\hat{\bm{d}}}{}{#2}{}{#1}}}
\newrobustcmd{\widehatbmdz}[2][]{\ensuremath{\subp{\widehat{\bm{d}}}{}{#2}{}{#1}}}
\newrobustcmd{\checkbmdz}[2][]{\ensuremath{\subp{\check{\bm{d}}}{}{#2}{}{#1}}}
\newrobustcmd{\tildebmdz}[2][]{\ensuremath{\subp{\tilde{\bm{d}}}{}{#2}{}{#1}}}
\newrobustcmd{\widetildebmdz}[2][]{\ensuremath{\subp{\widetilde{\bm{d}}}{}{#2}{}{#1}}}
\newrobustcmd{\acutebmdz}[2][]{\ensuremath{\subp{\acute{\bm{d}}}{}{#2}{}{#1}}}
\newrobustcmd{\gravebmdz}[2][]{\ensuremath{\subp{\grave{\bm{d}}}{}{#2}{}{#1}}}
\newrobustcmd{\dotbmdz}[2][]{\ensuremath{\subp{\dot{\bm{d}}}{}{#2}{}{#1}}}
\newrobustcmd{\ddotbmdz}[2][]{\ensuremath{\subp{\ddot{\bm{d}}}{}{#2}{}{#1}}}
\newrobustcmd{\brevebmdz}[2][]{\ensuremath{\subp{\breve{\bm{d}}}{}{#2}{}{#1}}}
\newrobustcmd{\barbmdz}[2][]{\ensuremath{\subp{\bar{\bm{d}}}{}{#2}{}{#1}}}
\newrobustcmd{\vecbmdz}[2][]{\ensuremath{\subp{\vec{\bm{d}}}{}{#2}{}{#1}}}
\newrobustcmd{\ez}[2][]{\ensuremath{\subp{e}{}{#2}{}{#1}}}
\newrobustcmd{\hatez}[2][]{\ensuremath{\subp{\hat{e}}{}{#2}{}{#1}}}
\newrobustcmd{\widehatez}[2][]{\ensuremath{\subp{\widehat{e}}{}{#2}{}{#1}}}
\newrobustcmd{\checkez}[2][]{\ensuremath{\subp{\check{e}}{}{#2}{}{#1}}}
\newrobustcmd{\tildeez}[2][]{\ensuremath{\subp{\tilde{e}}{}{#2}{}{#1}}}
\newrobustcmd{\widetildeez}[2][]{\ensuremath{\subp{\widetilde{e}}{}{#2}{}{#1}}}
\newrobustcmd{\acuteez}[2][]{\ensuremath{\subp{\acute{e}}{}{#2}{}{#1}}}
\newrobustcmd{\graveez}[2][]{\ensuremath{\subp{\grave{e}}{}{#2}{}{#1}}}
\newrobustcmd{\dotez}[2][]{\ensuremath{\subp{\dot{e}}{}{#2}{}{#1}}}
\newrobustcmd{\ddotez}[2][]{\ensuremath{\subp{\ddot{e}}{}{#2}{}{#1}}}
\newrobustcmd{\breveez}[2][]{\ensuremath{\subp{\breve{e}}{}{#2}{}{#1}}}
\newrobustcmd{\barez}[2][]{\ensuremath{\subp{\bar{e}}{}{#2}{}{#1}}}
\newrobustcmd{\vecez}[2][]{\ensuremath{\subp{\vec{e}}{}{#2}{}{#1}}}
\newrobustcmd{\bmez}[2][]{\ensuremath{\subp{\bm{e}}{}{#2}{}{#1}}}
\newrobustcmd{\hatbmez}[2][]{\ensuremath{\subp{\hat{\bm{e}}}{}{#2}{}{#1}}}
\newrobustcmd{\widehatbmez}[2][]{\ensuremath{\subp{\widehat{\bm{e}}}{}{#2}{}{#1}}}
\newrobustcmd{\checkbmez}[2][]{\ensuremath{\subp{\check{\bm{e}}}{}{#2}{}{#1}}}
\newrobustcmd{\tildebmez}[2][]{\ensuremath{\subp{\tilde{\bm{e}}}{}{#2}{}{#1}}}
\newrobustcmd{\widetildebmez}[2][]{\ensuremath{\subp{\widetilde{\bm{e}}}{}{#2}{}{#1}}}
\newrobustcmd{\acutebmez}[2][]{\ensuremath{\subp{\acute{\bm{e}}}{}{#2}{}{#1}}}
\newrobustcmd{\gravebmez}[2][]{\ensuremath{\subp{\grave{\bm{e}}}{}{#2}{}{#1}}}
\newrobustcmd{\dotbmez}[2][]{\ensuremath{\subp{\dot{\bm{e}}}{}{#2}{}{#1}}}
\newrobustcmd{\ddotbmez}[2][]{\ensuremath{\subp{\ddot{\bm{e}}}{}{#2}{}{#1}}}
\newrobustcmd{\brevebmez}[2][]{\ensuremath{\subp{\breve{\bm{e}}}{}{#2}{}{#1}}}
\newrobustcmd{\barbmez}[2][]{\ensuremath{\subp{\bar{\bm{e}}}{}{#2}{}{#1}}}
\newrobustcmd{\vecbmez}[2][]{\ensuremath{\subp{\vec{\bm{e}}}{}{#2}{}{#1}}}
\newrobustcmd{\fz}[2][]{\ensuremath{\subp{f}{}{#2}{}{#1}}}
\newrobustcmd{\hatfz}[2][]{\ensuremath{\subp{\hat{f}}{}{#2}{}{#1}}}
\newrobustcmd{\widehatfz}[2][]{\ensuremath{\subp{\widehat{f}}{}{#2}{}{#1}}}
\newrobustcmd{\checkfz}[2][]{\ensuremath{\subp{\check{f}}{}{#2}{}{#1}}}
\newrobustcmd{\tildefz}[2][]{\ensuremath{\subp{\tilde{f}}{}{#2}{}{#1}}}
\newrobustcmd{\widetildefz}[2][]{\ensuremath{\subp{\widetilde{f}}{}{#2}{}{#1}}}
\newrobustcmd{\acutefz}[2][]{\ensuremath{\subp{\acute{f}}{}{#2}{}{#1}}}
\newrobustcmd{\gravefz}[2][]{\ensuremath{\subp{\grave{f}}{}{#2}{}{#1}}}
\newrobustcmd{\dotfz}[2][]{\ensuremath{\subp{\dot{f}}{}{#2}{}{#1}}}
\newrobustcmd{\ddotfz}[2][]{\ensuremath{\subp{\ddot{f}}{}{#2}{}{#1}}}
\newrobustcmd{\brevefz}[2][]{\ensuremath{\subp{\breve{f}}{}{#2}{}{#1}}}
\newrobustcmd{\barfz}[2][]{\ensuremath{\subp{\bar{f}}{}{#2}{}{#1}}}
\newrobustcmd{\vecfz}[2][]{\ensuremath{\subp{\vec{f}}{}{#2}{}{#1}}}
\newrobustcmd{\bmfz}[2][]{\ensuremath{\subp{\bm{f}}{}{#2}{}{#1}}}
\newrobustcmd{\hatbmfz}[2][]{\ensuremath{\subp{\hat{\bm{f}}}{}{#2}{}{#1}}}
\newrobustcmd{\widehatbmfz}[2][]{\ensuremath{\subp{\widehat{\bm{f}}}{}{#2}{}{#1}}}
\newrobustcmd{\checkbmfz}[2][]{\ensuremath{\subp{\check{\bm{f}}}{}{#2}{}{#1}}}
\newrobustcmd{\tildebmfz}[2][]{\ensuremath{\subp{\tilde{\bm{f}}}{}{#2}{}{#1}}}
\newrobustcmd{\widetildebmfz}[2][]{\ensuremath{\subp{\widetilde{\bm{f}}}{}{#2}{}{#1}}}
\newrobustcmd{\acutebmfz}[2][]{\ensuremath{\subp{\acute{\bm{f}}}{}{#2}{}{#1}}}
\newrobustcmd{\gravebmfz}[2][]{\ensuremath{\subp{\grave{\bm{f}}}{}{#2}{}{#1}}}
\newrobustcmd{\dotbmfz}[2][]{\ensuremath{\subp{\dot{\bm{f}}}{}{#2}{}{#1}}}
\newrobustcmd{\ddotbmfz}[2][]{\ensuremath{\subp{\ddot{\bm{f}}}{}{#2}{}{#1}}}
\newrobustcmd{\brevebmfz}[2][]{\ensuremath{\subp{\breve{\bm{f}}}{}{#2}{}{#1}}}
\newrobustcmd{\barbmfz}[2][]{\ensuremath{\subp{\bar{\bm{f}}}{}{#2}{}{#1}}}
\newrobustcmd{\vecbmfz}[2][]{\ensuremath{\subp{\vec{\bm{f}}}{}{#2}{}{#1}}}
\newrobustcmd{\gz}[2][]{\ensuremath{\subp{g}{}{#2}{}{#1}}}
\newrobustcmd{\hatgz}[2][]{\ensuremath{\subp{\hat{g}}{}{#2}{}{#1}}}
\newrobustcmd{\widehatgz}[2][]{\ensuremath{\subp{\widehat{g}}{}{#2}{}{#1}}}
\newrobustcmd{\checkgz}[2][]{\ensuremath{\subp{\check{g}}{}{#2}{}{#1}}}
\newrobustcmd{\tildegz}[2][]{\ensuremath{\subp{\tilde{g}}{}{#2}{}{#1}}}
\newrobustcmd{\widetildegz}[2][]{\ensuremath{\subp{\widetilde{g}}{}{#2}{}{#1}}}
\newrobustcmd{\acutegz}[2][]{\ensuremath{\subp{\acute{g}}{}{#2}{}{#1}}}
\newrobustcmd{\gravegz}[2][]{\ensuremath{\subp{\grave{g}}{}{#2}{}{#1}}}
\newrobustcmd{\dotgz}[2][]{\ensuremath{\subp{\dot{g}}{}{#2}{}{#1}}}
\newrobustcmd{\ddotgz}[2][]{\ensuremath{\subp{\ddot{g}}{}{#2}{}{#1}}}
\newrobustcmd{\brevegz}[2][]{\ensuremath{\subp{\breve{g}}{}{#2}{}{#1}}}
\newrobustcmd{\bargz}[2][]{\ensuremath{\subp{\bar{g}}{}{#2}{}{#1}}}
\newrobustcmd{\vecgz}[2][]{\ensuremath{\subp{\vec{g}}{}{#2}{}{#1}}}
\newrobustcmd{\bmgz}[2][]{\ensuremath{\subp{\bm{g}}{}{#2}{}{#1}}}
\newrobustcmd{\hatbmgz}[2][]{\ensuremath{\subp{\hat{\bm{g}}}{}{#2}{}{#1}}}
\newrobustcmd{\widehatbmgz}[2][]{\ensuremath{\subp{\widehat{\bm{g}}}{}{#2}{}{#1}}}
\newrobustcmd{\checkbmgz}[2][]{\ensuremath{\subp{\check{\bm{g}}}{}{#2}{}{#1}}}
\newrobustcmd{\tildebmgz}[2][]{\ensuremath{\subp{\tilde{\bm{g}}}{}{#2}{}{#1}}}
\newrobustcmd{\widetildebmgz}[2][]{\ensuremath{\subp{\widetilde{\bm{g}}}{}{#2}{}{#1}}}
\newrobustcmd{\acutebmgz}[2][]{\ensuremath{\subp{\acute{\bm{g}}}{}{#2}{}{#1}}}
\newrobustcmd{\gravebmgz}[2][]{\ensuremath{\subp{\grave{\bm{g}}}{}{#2}{}{#1}}}
\newrobustcmd{\dotbmgz}[2][]{\ensuremath{\subp{\dot{\bm{g}}}{}{#2}{}{#1}}}
\newrobustcmd{\ddotbmgz}[2][]{\ensuremath{\subp{\ddot{\bm{g}}}{}{#2}{}{#1}}}
\newrobustcmd{\brevebmgz}[2][]{\ensuremath{\subp{\breve{\bm{g}}}{}{#2}{}{#1}}}
\newrobustcmd{\barbmgz}[2][]{\ensuremath{\subp{\bar{\bm{g}}}{}{#2}{}{#1}}}
\newrobustcmd{\vecbmgz}[2][]{\ensuremath{\subp{\vec{\bm{g}}}{}{#2}{}{#1}}}
\newrobustcmd{\hz}[2][]{\ensuremath{\subp{h}{}{#2}{}{#1}}}
\newrobustcmd{\hathz}[2][]{\ensuremath{\subp{\hat{h}}{}{#2}{}{#1}}}
\newrobustcmd{\widehathz}[2][]{\ensuremath{\subp{\widehat{h}}{}{#2}{}{#1}}}
\newrobustcmd{\checkhz}[2][]{\ensuremath{\subp{\check{h}}{}{#2}{}{#1}}}
\newrobustcmd{\tildehz}[2][]{\ensuremath{\subp{\tilde{h}}{}{#2}{}{#1}}}
\newrobustcmd{\widetildehz}[2][]{\ensuremath{\subp{\widetilde{h}}{}{#2}{}{#1}}}
\newrobustcmd{\acutehz}[2][]{\ensuremath{\subp{\acute{h}}{}{#2}{}{#1}}}
\newrobustcmd{\gravehz}[2][]{\ensuremath{\subp{\grave{h}}{}{#2}{}{#1}}}
\newrobustcmd{\dothz}[2][]{\ensuremath{\subp{\dot{h}}{}{#2}{}{#1}}}
\newrobustcmd{\ddothz}[2][]{\ensuremath{\subp{\ddot{h}}{}{#2}{}{#1}}}
\newrobustcmd{\brevehz}[2][]{\ensuremath{\subp{\breve{h}}{}{#2}{}{#1}}}
\newrobustcmd{\barhz}[2][]{\ensuremath{\subp{\bar{h}}{}{#2}{}{#1}}}
\newrobustcmd{\vechz}[2][]{\ensuremath{\subp{\vec{h}}{}{#2}{}{#1}}}
\newrobustcmd{\bmhz}[2][]{\ensuremath{\subp{\bm{h}}{}{#2}{}{#1}}}
\newrobustcmd{\hatbmhz}[2][]{\ensuremath{\subp{\hat{\bm{h}}}{}{#2}{}{#1}}}
\newrobustcmd{\widehatbmhz}[2][]{\ensuremath{\subp{\widehat{\bm{h}}}{}{#2}{}{#1}}}
\newrobustcmd{\checkbmhz}[2][]{\ensuremath{\subp{\check{\bm{h}}}{}{#2}{}{#1}}}
\newrobustcmd{\tildebmhz}[2][]{\ensuremath{\subp{\tilde{\bm{h}}}{}{#2}{}{#1}}}
\newrobustcmd{\widetildebmhz}[2][]{\ensuremath{\subp{\widetilde{\bm{h}}}{}{#2}{}{#1}}}
\newrobustcmd{\acutebmhz}[2][]{\ensuremath{\subp{\acute{\bm{h}}}{}{#2}{}{#1}}}
\newrobustcmd{\gravebmhz}[2][]{\ensuremath{\subp{\grave{\bm{h}}}{}{#2}{}{#1}}}
\newrobustcmd{\dotbmhz}[2][]{\ensuremath{\subp{\dot{\bm{h}}}{}{#2}{}{#1}}}
\newrobustcmd{\ddotbmhz}[2][]{\ensuremath{\subp{\ddot{\bm{h}}}{}{#2}{}{#1}}}
\newrobustcmd{\brevebmhz}[2][]{\ensuremath{\subp{\breve{\bm{h}}}{}{#2}{}{#1}}}
\newrobustcmd{\barbmhz}[2][]{\ensuremath{\subp{\bar{\bm{h}}}{}{#2}{}{#1}}}
\newrobustcmd{\vecbmhz}[2][]{\ensuremath{\subp{\vec{\bm{h}}}{}{#2}{}{#1}}}
\newrobustcmd{\iz}[2][]{\ensuremath{\subp{i}{}{#2}{}{#1}}}
\newrobustcmd{\hatiz}[2][]{\ensuremath{\subp{\hat{\imath}}{}{#2}{}{#1}}}
\newrobustcmd{\widehatiz}[2][]{\ensuremath{\subp{\widehat{\imath}}{}{#2}{}{#1}}}
\newrobustcmd{\checkiz}[2][]{\ensuremath{\subp{\check{\imath}}{}{#2}{}{#1}}}
\newrobustcmd{\tildeiz}[2][]{\ensuremath{\subp{\tilde{\imath}}{}{#2}{}{#1}}}
\newrobustcmd{\widetildeiz}[2][]{\ensuremath{\subp{\widetilde{\imath}}{}{#2}{}{#1}}}
\newrobustcmd{\acuteiz}[2][]{\ensuremath{\subp{\acute{\imath}}{}{#2}{}{#1}}}
\newrobustcmd{\graveiz}[2][]{\ensuremath{\subp{\grave{\imath}}{}{#2}{}{#1}}}
\newrobustcmd{\dotiz}[2][]{\ensuremath{\subp{\dot{\imath}}{}{#2}{}{#1}}}
\newrobustcmd{\ddotiz}[2][]{\ensuremath{\subp{\ddot{\imath}}{}{#2}{}{#1}}}
\newrobustcmd{\breveiz}[2][]{\ensuremath{\subp{\breve{\imath}}{}{#2}{}{#1}}}
\newrobustcmd{\bariz}[2][]{\ensuremath{\subp{\bar{\imath}}{}{#2}{}{#1}}}
\newrobustcmd{\veciz}[2][]{\ensuremath{\subp{\vec{\imath}}{}{#2}{}{#1}}}
\newrobustcmd{\bmiz}[2][]{\ensuremath{\subp{\bm{i}}{}{#2}{}{#1}}}
\newrobustcmd{\hatbmiz}[2][]{\ensuremath{\subp{\hat{\bm{\imath}}}{}{#2}{}{#1}}}
\newrobustcmd{\widehatbmiz}[2][]{\ensuremath{\subp{\widehat{\bm{\imath}}}{}{#2}{}{#1}}}
\newrobustcmd{\checkbmiz}[2][]{\ensuremath{\subp{\check{\bm{\imath}}}{}{#2}{}{#1}}}
\newrobustcmd{\tildebmiz}[2][]{\ensuremath{\subp{\tilde{\bm{\imath}}}{}{#2}{}{#1}}}
\newrobustcmd{\widetildebmiz}[2][]{\ensuremath{\subp{\widetilde{\bm{\imath}}}{}{#2}{}{#1}}}
\newrobustcmd{\acutebmiz}[2][]{\ensuremath{\subp{\acute{\bm{\imath}}}{}{#2}{}{#1}}}
\newrobustcmd{\gravebmiz}[2][]{\ensuremath{\subp{\grave{\bm{\imath}}}{}{#2}{}{#1}}}
\newrobustcmd{\dotbmiz}[2][]{\ensuremath{\subp{\dot{\bm{\imath}}}{}{#2}{}{#1}}}
\newrobustcmd{\ddotbmiz}[2][]{\ensuremath{\subp{\ddot{\bm{\imath}}}{}{#2}{}{#1}}}
\newrobustcmd{\brevebmiz}[2][]{\ensuremath{\subp{\breve{\bm{\imath}}}{}{#2}{}{#1}}}
\newrobustcmd{\barbmiz}[2][]{\ensuremath{\subp{\bar{\bm{\imath}}}{}{#2}{}{#1}}}
\newrobustcmd{\vecbmiz}[2][]{\ensuremath{\subp{\vec{\bm{\imath}}}{}{#2}{}{#1}}}
\newrobustcmd{\jz}[2][]{\ensuremath{\subp{j}{}{#2}{}{#1}}}
\newrobustcmd{\hatjz}[2][]{\ensuremath{\subp{\hat{\jmath}}{}{#2}{}{#1}}}
\newrobustcmd{\widehatjz}[2][]{\ensuremath{\subp{\widehat{\jmath}}{}{#2}{}{#1}}}
\newrobustcmd{\checkjz}[2][]{\ensuremath{\subp{\check{\jmath}}{}{#2}{}{#1}}}
\newrobustcmd{\tildejz}[2][]{\ensuremath{\subp{\tilde{\jmath}}{}{#2}{}{#1}}}
\newrobustcmd{\widetildejz}[2][]{\ensuremath{\subp{\widetilde{\jmath}}{}{#2}{}{#1}}}
\newrobustcmd{\acutejz}[2][]{\ensuremath{\subp{\acute{\jmath}}{}{#2}{}{#1}}}
\newrobustcmd{\gravejz}[2][]{\ensuremath{\subp{\grave{\jmath}}{}{#2}{}{#1}}}
\newrobustcmd{\dotjz}[2][]{\ensuremath{\subp{\dot{\jmath}}{}{#2}{}{#1}}}
\newrobustcmd{\ddotjz}[2][]{\ensuremath{\subp{\ddot{\jmath}}{}{#2}{}{#1}}}
\newrobustcmd{\brevejz}[2][]{\ensuremath{\subp{\breve{\jmath}}{}{#2}{}{#1}}}
\newrobustcmd{\barjz}[2][]{\ensuremath{\subp{\bar{\jmath}}{}{#2}{}{#1}}}
\newrobustcmd{\vecjz}[2][]{\ensuremath{\subp{\vec{\jmath}}{}{#2}{}{#1}}}
\newrobustcmd{\bmjz}[2][]{\ensuremath{\subp{\bm{j}}{}{#2}{}{#1}}}
\newrobustcmd{\hatbmjz}[2][]{\ensuremath{\subp{\hat{\bm{\jmath}}}{}{#2}{}{#1}}}
\newrobustcmd{\widehatbmjz}[2][]{\ensuremath{\subp{\widehat{\bm{\jmath}}}{}{#2}{}{#1}}}
\newrobustcmd{\checkbmjz}[2][]{\ensuremath{\subp{\check{\bm{\jmath}}}{}{#2}{}{#1}}}
\newrobustcmd{\tildebmjz}[2][]{\ensuremath{\subp{\tilde{\bm{\jmath}}}{}{#2}{}{#1}}}
\newrobustcmd{\widetildebmjz}[2][]{\ensuremath{\subp{\widetilde{\bm{\jmath}}}{}{#2}{}{#1}}}
\newrobustcmd{\acutebmjz}[2][]{\ensuremath{\subp{\acute{\bm{\jmath}}}{}{#2}{}{#1}}}
\newrobustcmd{\gravebmjz}[2][]{\ensuremath{\subp{\grave{\bm{\jmath}}}{}{#2}{}{#1}}}
\newrobustcmd{\dotbmjz}[2][]{\ensuremath{\subp{\dot{\bm{\jmath}}}{}{#2}{}{#1}}}
\newrobustcmd{\ddotbmjz}[2][]{\ensuremath{\subp{\ddot{\bm{\jmath}}}{}{#2}{}{#1}}}
\newrobustcmd{\brevebmjz}[2][]{\ensuremath{\subp{\breve{\bm{\jmath}}}{}{#2}{}{#1}}}
\newrobustcmd{\barbmjz}[2][]{\ensuremath{\subp{\bar{\bm{\jmath}}}{}{#2}{}{#1}}}
\newrobustcmd{\vecbmjz}[2][]{\ensuremath{\subp{\vec{\bm{\jmath}}}{}{#2}{}{#1}}}
\newrobustcmd{\kz}[2][]{\ensuremath{\subp{k}{}{#2}{}{#1}}}
\newrobustcmd{\hatkz}[2][]{\ensuremath{\subp{\hat{k}}{}{#2}{}{#1}}}
\newrobustcmd{\widehatkz}[2][]{\ensuremath{\subp{\widehat{k}}{}{#2}{}{#1}}}
\newrobustcmd{\checkkz}[2][]{\ensuremath{\subp{\check{k}}{}{#2}{}{#1}}}
\newrobustcmd{\tildekz}[2][]{\ensuremath{\subp{\tilde{k}}{}{#2}{}{#1}}}
\newrobustcmd{\widetildekz}[2][]{\ensuremath{\subp{\widetilde{k}}{}{#2}{}{#1}}}
\newrobustcmd{\acutekz}[2][]{\ensuremath{\subp{\acute{k}}{}{#2}{}{#1}}}
\newrobustcmd{\gravekz}[2][]{\ensuremath{\subp{\grave{k}}{}{#2}{}{#1}}}
\newrobustcmd{\dotkz}[2][]{\ensuremath{\subp{\dot{k}}{}{#2}{}{#1}}}
\newrobustcmd{\ddotkz}[2][]{\ensuremath{\subp{\ddot{k}}{}{#2}{}{#1}}}
\newrobustcmd{\brevekz}[2][]{\ensuremath{\subp{\breve{k}}{}{#2}{}{#1}}}
\newrobustcmd{\barkz}[2][]{\ensuremath{\subp{\bar{k}}{}{#2}{}{#1}}}
\newrobustcmd{\veckz}[2][]{\ensuremath{\subp{\vec{k}}{}{#2}{}{#1}}}
\newrobustcmd{\bmkz}[2][]{\ensuremath{\subp{\bm{k}}{}{#2}{}{#1}}}
\newrobustcmd{\hatbmkz}[2][]{\ensuremath{\subp{\hat{\bm{k}}}{}{#2}{}{#1}}}
\newrobustcmd{\widehatbmkz}[2][]{\ensuremath{\subp{\widehat{\bm{k}}}{}{#2}{}{#1}}}
\newrobustcmd{\checkbmkz}[2][]{\ensuremath{\subp{\check{\bm{k}}}{}{#2}{}{#1}}}
\newrobustcmd{\tildebmkz}[2][]{\ensuremath{\subp{\tilde{\bm{k}}}{}{#2}{}{#1}}}
\newrobustcmd{\widetildebmkz}[2][]{\ensuremath{\subp{\widetilde{\bm{k}}}{}{#2}{}{#1}}}
\newrobustcmd{\acutebmkz}[2][]{\ensuremath{\subp{\acute{\bm{k}}}{}{#2}{}{#1}}}
\newrobustcmd{\gravebmkz}[2][]{\ensuremath{\subp{\grave{\bm{k}}}{}{#2}{}{#1}}}
\newrobustcmd{\dotbmkz}[2][]{\ensuremath{\subp{\dot{\bm{k}}}{}{#2}{}{#1}}}
\newrobustcmd{\ddotbmkz}[2][]{\ensuremath{\subp{\ddot{\bm{k}}}{}{#2}{}{#1}}}
\newrobustcmd{\brevebmkz}[2][]{\ensuremath{\subp{\breve{\bm{k}}}{}{#2}{}{#1}}}
\newrobustcmd{\barbmkz}[2][]{\ensuremath{\subp{\bar{\bm{k}}}{}{#2}{}{#1}}}
\newrobustcmd{\vecbmkz}[2][]{\ensuremath{\subp{\vec{\bm{k}}}{}{#2}{}{#1}}}
\newrobustcmd{\lz}[2][]{\ensuremath{\subp{l}{}{#2}{}{#1}}}
\newrobustcmd{\hatlz}[2][]{\ensuremath{\subp{\hat{l}}{}{#2}{}{#1}}}
\newrobustcmd{\widehatlz}[2][]{\ensuremath{\subp{\widehat{l}}{}{#2}{}{#1}}}
\newrobustcmd{\checklz}[2][]{\ensuremath{\subp{\check{l}}{}{#2}{}{#1}}}
\newrobustcmd{\tildelz}[2][]{\ensuremath{\subp{\tilde{l}}{}{#2}{}{#1}}}
\newrobustcmd{\widetildelz}[2][]{\ensuremath{\subp{\widetilde{l}}{}{#2}{}{#1}}}
\newrobustcmd{\acutelz}[2][]{\ensuremath{\subp{\acute{l}}{}{#2}{}{#1}}}
\newrobustcmd{\gravelz}[2][]{\ensuremath{\subp{\grave{l}}{}{#2}{}{#1}}}
\newrobustcmd{\dotlz}[2][]{\ensuremath{\subp{\dot{l}}{}{#2}{}{#1}}}
\newrobustcmd{\ddotlz}[2][]{\ensuremath{\subp{\ddot{l}}{}{#2}{}{#1}}}
\newrobustcmd{\brevelz}[2][]{\ensuremath{\subp{\breve{l}}{}{#2}{}{#1}}}
\newrobustcmd{\barlz}[2][]{\ensuremath{\subp{\bar{l}}{}{#2}{}{#1}}}
\newrobustcmd{\veclz}[2][]{\ensuremath{\subp{\vec{l}}{}{#2}{}{#1}}}
\newrobustcmd{\bmlz}[2][]{\ensuremath{\subp{\bm{l}}{}{#2}{}{#1}}}
\newrobustcmd{\hatbmlz}[2][]{\ensuremath{\subp{\hat{\bm{l}}}{}{#2}{}{#1}}}
\newrobustcmd{\widehatbmlz}[2][]{\ensuremath{\subp{\widehat{\bm{l}}}{}{#2}{}{#1}}}
\newrobustcmd{\checkbmlz}[2][]{\ensuremath{\subp{\check{\bm{l}}}{}{#2}{}{#1}}}
\newrobustcmd{\tildebmlz}[2][]{\ensuremath{\subp{\tilde{\bm{l}}}{}{#2}{}{#1}}}
\newrobustcmd{\widetildebmlz}[2][]{\ensuremath{\subp{\widetilde{\bm{l}}}{}{#2}{}{#1}}}
\newrobustcmd{\acutebmlz}[2][]{\ensuremath{\subp{\acute{\bm{l}}}{}{#2}{}{#1}}}
\newrobustcmd{\gravebmlz}[2][]{\ensuremath{\subp{\grave{\bm{l}}}{}{#2}{}{#1}}}
\newrobustcmd{\dotbmlz}[2][]{\ensuremath{\subp{\dot{\bm{l}}}{}{#2}{}{#1}}}
\newrobustcmd{\ddotbmlz}[2][]{\ensuremath{\subp{\ddot{\bm{l}}}{}{#2}{}{#1}}}
\newrobustcmd{\brevebmlz}[2][]{\ensuremath{\subp{\breve{\bm{l}}}{}{#2}{}{#1}}}
\newrobustcmd{\barbmlz}[2][]{\ensuremath{\subp{\bar{\bm{l}}}{}{#2}{}{#1}}}
\newrobustcmd{\vecbmlz}[2][]{\ensuremath{\subp{\vec{\bm{l}}}{}{#2}{}{#1}}}
\newrobustcmd{\mz}[2][]{\ensuremath{\subp{m}{}{#2}{}{#1}}}
\newrobustcmd{\hatmz}[2][]{\ensuremath{\subp{\hat{m}}{}{#2}{}{#1}}}
\newrobustcmd{\widehatmz}[2][]{\ensuremath{\subp{\widehat{m}}{}{#2}{}{#1}}}
\newrobustcmd{\checkmz}[2][]{\ensuremath{\subp{\check{m}}{}{#2}{}{#1}}}
\newrobustcmd{\tildemz}[2][]{\ensuremath{\subp{\tilde{m}}{}{#2}{}{#1}}}
\newrobustcmd{\widetildemz}[2][]{\ensuremath{\subp{\widetilde{m}}{}{#2}{}{#1}}}
\newrobustcmd{\acutemz}[2][]{\ensuremath{\subp{\acute{m}}{}{#2}{}{#1}}}
\newrobustcmd{\gravemz}[2][]{\ensuremath{\subp{\grave{m}}{}{#2}{}{#1}}}
\newrobustcmd{\dotmz}[2][]{\ensuremath{\subp{\dot{m}}{}{#2}{}{#1}}}
\newrobustcmd{\ddotmz}[2][]{\ensuremath{\subp{\ddot{m}}{}{#2}{}{#1}}}
\newrobustcmd{\brevemz}[2][]{\ensuremath{\subp{\breve{m}}{}{#2}{}{#1}}}
\newrobustcmd{\barmz}[2][]{\ensuremath{\subp{\bar{m}}{}{#2}{}{#1}}}
\newrobustcmd{\vecmz}[2][]{\ensuremath{\subp{\vec{m}}{}{#2}{}{#1}}}
\newrobustcmd{\bmmz}[2][]{\ensuremath{\subp{\bm{m}}{}{#2}{}{#1}}}
\newrobustcmd{\hatbmmz}[2][]{\ensuremath{\subp{\hat{\bm{m}}}{}{#2}{}{#1}}}
\newrobustcmd{\widehatbmmz}[2][]{\ensuremath{\subp{\widehat{\bm{m}}}{}{#2}{}{#1}}}
\newrobustcmd{\checkbmmz}[2][]{\ensuremath{\subp{\check{\bm{m}}}{}{#2}{}{#1}}}
\newrobustcmd{\tildebmmz}[2][]{\ensuremath{\subp{\tilde{\bm{m}}}{}{#2}{}{#1}}}
\newrobustcmd{\widetildebmmz}[2][]{\ensuremath{\subp{\widetilde{\bm{m}}}{}{#2}{}{#1}}}
\newrobustcmd{\acutebmmz}[2][]{\ensuremath{\subp{\acute{\bm{m}}}{}{#2}{}{#1}}}
\newrobustcmd{\gravebmmz}[2][]{\ensuremath{\subp{\grave{\bm{m}}}{}{#2}{}{#1}}}
\newrobustcmd{\dotbmmz}[2][]{\ensuremath{\subp{\dot{\bm{m}}}{}{#2}{}{#1}}}
\newrobustcmd{\ddotbmmz}[2][]{\ensuremath{\subp{\ddot{\bm{m}}}{}{#2}{}{#1}}}
\newrobustcmd{\brevebmmz}[2][]{\ensuremath{\subp{\breve{\bm{m}}}{}{#2}{}{#1}}}
\newrobustcmd{\barbmmz}[2][]{\ensuremath{\subp{\bar{\bm{m}}}{}{#2}{}{#1}}}
\newrobustcmd{\vecbmmz}[2][]{\ensuremath{\subp{\vec{\bm{m}}}{}{#2}{}{#1}}}
\newrobustcmd{\nz}[2][]{\ensuremath{\subp{n}{}{#2}{}{#1}}}
\newrobustcmd{\hatnz}[2][]{\ensuremath{\subp{\hat{n}}{}{#2}{}{#1}}}
\newrobustcmd{\widehatnz}[2][]{\ensuremath{\subp{\widehat{n}}{}{#2}{}{#1}}}
\newrobustcmd{\checknz}[2][]{\ensuremath{\subp{\check{n}}{}{#2}{}{#1}}}
\newrobustcmd{\tildenz}[2][]{\ensuremath{\subp{\tilde{n}}{}{#2}{}{#1}}}
\newrobustcmd{\widetildenz}[2][]{\ensuremath{\subp{\widetilde{n}}{}{#2}{}{#1}}}
\newrobustcmd{\acutenz}[2][]{\ensuremath{\subp{\acute{n}}{}{#2}{}{#1}}}
\newrobustcmd{\gravenz}[2][]{\ensuremath{\subp{\grave{n}}{}{#2}{}{#1}}}
\newrobustcmd{\dotnz}[2][]{\ensuremath{\subp{\dot{n}}{}{#2}{}{#1}}}
\newrobustcmd{\ddotnz}[2][]{\ensuremath{\subp{\ddot{n}}{}{#2}{}{#1}}}
\newrobustcmd{\brevenz}[2][]{\ensuremath{\subp{\breve{n}}{}{#2}{}{#1}}}
\newrobustcmd{\barnz}[2][]{\ensuremath{\subp{\bar{n}}{}{#2}{}{#1}}}
\newrobustcmd{\vecnz}[2][]{\ensuremath{\subp{\vec{n}}{}{#2}{}{#1}}}
\newrobustcmd{\bmnz}[2][]{\ensuremath{\subp{\bm{n}}{}{#2}{}{#1}}}
\newrobustcmd{\hatbmnz}[2][]{\ensuremath{\subp{\hat{\bm{n}}}{}{#2}{}{#1}}}
\newrobustcmd{\widehatbmnz}[2][]{\ensuremath{\subp{\widehat{\bm{n}}}{}{#2}{}{#1}}}
\newrobustcmd{\checkbmnz}[2][]{\ensuremath{\subp{\check{\bm{n}}}{}{#2}{}{#1}}}
\newrobustcmd{\tildebmnz}[2][]{\ensuremath{\subp{\tilde{\bm{n}}}{}{#2}{}{#1}}}
\newrobustcmd{\widetildebmnz}[2][]{\ensuremath{\subp{\widetilde{\bm{n}}}{}{#2}{}{#1}}}
\newrobustcmd{\acutebmnz}[2][]{\ensuremath{\subp{\acute{\bm{n}}}{}{#2}{}{#1}}}
\newrobustcmd{\gravebmnz}[2][]{\ensuremath{\subp{\grave{\bm{n}}}{}{#2}{}{#1}}}
\newrobustcmd{\dotbmnz}[2][]{\ensuremath{\subp{\dot{\bm{n}}}{}{#2}{}{#1}}}
\newrobustcmd{\ddotbmnz}[2][]{\ensuremath{\subp{\ddot{\bm{n}}}{}{#2}{}{#1}}}
\newrobustcmd{\brevebmnz}[2][]{\ensuremath{\subp{\breve{\bm{n}}}{}{#2}{}{#1}}}
\newrobustcmd{\barbmnz}[2][]{\ensuremath{\subp{\bar{\bm{n}}}{}{#2}{}{#1}}}
\newrobustcmd{\vecbmnz}[2][]{\ensuremath{\subp{\vec{\bm{n}}}{}{#2}{}{#1}}}
\newrobustcmd{\oz}[2][]{\ensuremath{\subp{o}{}{#2}{}{#1}}}
\newrobustcmd{\hatoz}[2][]{\ensuremath{\subp{\hat{o}}{}{#2}{}{#1}}}
\newrobustcmd{\widehatoz}[2][]{\ensuremath{\subp{\widehat{o}}{}{#2}{}{#1}}}
\newrobustcmd{\checkoz}[2][]{\ensuremath{\subp{\check{o}}{}{#2}{}{#1}}}
\newrobustcmd{\tildeoz}[2][]{\ensuremath{\subp{\tilde{o}}{}{#2}{}{#1}}}
\newrobustcmd{\widetildeoz}[2][]{\ensuremath{\subp{\widetilde{o}}{}{#2}{}{#1}}}
\newrobustcmd{\acuteoz}[2][]{\ensuremath{\subp{\acute{o}}{}{#2}{}{#1}}}
\newrobustcmd{\graveoz}[2][]{\ensuremath{\subp{\grave{o}}{}{#2}{}{#1}}}
\newrobustcmd{\dotoz}[2][]{\ensuremath{\subp{\dot{o}}{}{#2}{}{#1}}}
\newrobustcmd{\ddotoz}[2][]{\ensuremath{\subp{\ddot{o}}{}{#2}{}{#1}}}
\newrobustcmd{\breveoz}[2][]{\ensuremath{\subp{\breve{o}}{}{#2}{}{#1}}}
\newrobustcmd{\baroz}[2][]{\ensuremath{\subp{\bar{o}}{}{#2}{}{#1}}}
\newrobustcmd{\vecoz}[2][]{\ensuremath{\subp{\vec{o}}{}{#2}{}{#1}}}
\newrobustcmd{\bmoz}[2][]{\ensuremath{\subp{\bm{o}}{}{#2}{}{#1}}}
\newrobustcmd{\hatbmoz}[2][]{\ensuremath{\subp{\hat{\bm{o}}}{}{#2}{}{#1}}}
\newrobustcmd{\widehatbmoz}[2][]{\ensuremath{\subp{\widehat{\bm{o}}}{}{#2}{}{#1}}}
\newrobustcmd{\checkbmoz}[2][]{\ensuremath{\subp{\check{\bm{o}}}{}{#2}{}{#1}}}
\newrobustcmd{\tildebmoz}[2][]{\ensuremath{\subp{\tilde{\bm{o}}}{}{#2}{}{#1}}}
\newrobustcmd{\widetildebmoz}[2][]{\ensuremath{\subp{\widetilde{\bm{o}}}{}{#2}{}{#1}}}
\newrobustcmd{\acutebmoz}[2][]{\ensuremath{\subp{\acute{\bm{o}}}{}{#2}{}{#1}}}
\newrobustcmd{\gravebmoz}[2][]{\ensuremath{\subp{\grave{\bm{o}}}{}{#2}{}{#1}}}
\newrobustcmd{\dotbmoz}[2][]{\ensuremath{\subp{\dot{\bm{o}}}{}{#2}{}{#1}}}
\newrobustcmd{\ddotbmoz}[2][]{\ensuremath{\subp{\ddot{\bm{o}}}{}{#2}{}{#1}}}
\newrobustcmd{\brevebmoz}[2][]{\ensuremath{\subp{\breve{\bm{o}}}{}{#2}{}{#1}}}
\newrobustcmd{\barbmoz}[2][]{\ensuremath{\subp{\bar{\bm{o}}}{}{#2}{}{#1}}}
\newrobustcmd{\vecbmoz}[2][]{\ensuremath{\subp{\vec{\bm{o}}}{}{#2}{}{#1}}}
\newrobustcmd{\pz}[2][]{\ensuremath{\subp{p}{}{#2}{}{#1}}}
\newrobustcmd{\hatpz}[2][]{\ensuremath{\subp{\hat{p}}{}{#2}{}{#1}}}
\newrobustcmd{\widehatpz}[2][]{\ensuremath{\subp{\widehat{p}}{}{#2}{}{#1}}}
\newrobustcmd{\checkpz}[2][]{\ensuremath{\subp{\check{p}}{}{#2}{}{#1}}}
\newrobustcmd{\tildepz}[2][]{\ensuremath{\subp{\tilde{p}}{}{#2}{}{#1}}}
\newrobustcmd{\widetildepz}[2][]{\ensuremath{\subp{\widetilde{p}}{}{#2}{}{#1}}}
\newrobustcmd{\acutepz}[2][]{\ensuremath{\subp{\acute{p}}{}{#2}{}{#1}}}
\newrobustcmd{\gravepz}[2][]{\ensuremath{\subp{\grave{p}}{}{#2}{}{#1}}}
\newrobustcmd{\dotpz}[2][]{\ensuremath{\subp{\dot{p}}{}{#2}{}{#1}}}
\newrobustcmd{\ddotpz}[2][]{\ensuremath{\subp{\ddot{p}}{}{#2}{}{#1}}}
\newrobustcmd{\brevepz}[2][]{\ensuremath{\subp{\breve{p}}{}{#2}{}{#1}}}
\newrobustcmd{\barpz}[2][]{\ensuremath{\subp{\bar{p}}{}{#2}{}{#1}}}
\newrobustcmd{\vecpz}[2][]{\ensuremath{\subp{\vec{p}}{}{#2}{}{#1}}}
\newrobustcmd{\bmpz}[2][]{\ensuremath{\subp{\bm{p}}{}{#2}{}{#1}}}
\newrobustcmd{\hatbmpz}[2][]{\ensuremath{\subp{\hat{\bm{p}}}{}{#2}{}{#1}}}
\newrobustcmd{\widehatbmpz}[2][]{\ensuremath{\subp{\widehat{\bm{p}}}{}{#2}{}{#1}}}
\newrobustcmd{\checkbmpz}[2][]{\ensuremath{\subp{\check{\bm{p}}}{}{#2}{}{#1}}}
\newrobustcmd{\tildebmpz}[2][]{\ensuremath{\subp{\tilde{\bm{p}}}{}{#2}{}{#1}}}
\newrobustcmd{\widetildebmpz}[2][]{\ensuremath{\subp{\widetilde{\bm{p}}}{}{#2}{}{#1}}}
\newrobustcmd{\acutebmpz}[2][]{\ensuremath{\subp{\acute{\bm{p}}}{}{#2}{}{#1}}}
\newrobustcmd{\gravebmpz}[2][]{\ensuremath{\subp{\grave{\bm{p}}}{}{#2}{}{#1}}}
\newrobustcmd{\dotbmpz}[2][]{\ensuremath{\subp{\dot{\bm{p}}}{}{#2}{}{#1}}}
\newrobustcmd{\ddotbmpz}[2][]{\ensuremath{\subp{\ddot{\bm{p}}}{}{#2}{}{#1}}}
\newrobustcmd{\brevebmpz}[2][]{\ensuremath{\subp{\breve{\bm{p}}}{}{#2}{}{#1}}}
\newrobustcmd{\barbmpz}[2][]{\ensuremath{\subp{\bar{\bm{p}}}{}{#2}{}{#1}}}
\newrobustcmd{\vecbmpz}[2][]{\ensuremath{\subp{\vec{\bm{p}}}{}{#2}{}{#1}}}
\newrobustcmd{\qz}[2][]{\ensuremath{\subp{q}{}{#2}{}{#1}}}
\newrobustcmd{\hatqz}[2][]{\ensuremath{\subp{\hat{q}}{}{#2}{}{#1}}}
\newrobustcmd{\widehatqz}[2][]{\ensuremath{\subp{\widehat{q}}{}{#2}{}{#1}}}
\newrobustcmd{\checkqz}[2][]{\ensuremath{\subp{\check{q}}{}{#2}{}{#1}}}
\newrobustcmd{\tildeqz}[2][]{\ensuremath{\subp{\tilde{q}}{}{#2}{}{#1}}}
\newrobustcmd{\widetildeqz}[2][]{\ensuremath{\subp{\widetilde{q}}{}{#2}{}{#1}}}
\newrobustcmd{\acuteqz}[2][]{\ensuremath{\subp{\acute{q}}{}{#2}{}{#1}}}
\newrobustcmd{\graveqz}[2][]{\ensuremath{\subp{\grave{q}}{}{#2}{}{#1}}}
\newrobustcmd{\dotqz}[2][]{\ensuremath{\subp{\dot{q}}{}{#2}{}{#1}}}
\newrobustcmd{\ddotqz}[2][]{\ensuremath{\subp{\ddot{q}}{}{#2}{}{#1}}}
\newrobustcmd{\breveqz}[2][]{\ensuremath{\subp{\breve{q}}{}{#2}{}{#1}}}
\newrobustcmd{\barqz}[2][]{\ensuremath{\subp{\bar{q}}{}{#2}{}{#1}}}
\newrobustcmd{\vecqz}[2][]{\ensuremath{\subp{\vec{q}}{}{#2}{}{#1}}}
\newrobustcmd{\bmqz}[2][]{\ensuremath{\subp{\bm{q}}{}{#2}{}{#1}}}
\newrobustcmd{\hatbmqz}[2][]{\ensuremath{\subp{\hat{\bm{q}}}{}{#2}{}{#1}}}
\newrobustcmd{\widehatbmqz}[2][]{\ensuremath{\subp{\widehat{\bm{q}}}{}{#2}{}{#1}}}
\newrobustcmd{\checkbmqz}[2][]{\ensuremath{\subp{\check{\bm{q}}}{}{#2}{}{#1}}}
\newrobustcmd{\tildebmqz}[2][]{\ensuremath{\subp{\tilde{\bm{q}}}{}{#2}{}{#1}}}
\newrobustcmd{\widetildebmqz}[2][]{\ensuremath{\subp{\widetilde{\bm{q}}}{}{#2}{}{#1}}}
\newrobustcmd{\acutebmqz}[2][]{\ensuremath{\subp{\acute{\bm{q}}}{}{#2}{}{#1}}}
\newrobustcmd{\gravebmqz}[2][]{\ensuremath{\subp{\grave{\bm{q}}}{}{#2}{}{#1}}}
\newrobustcmd{\dotbmqz}[2][]{\ensuremath{\subp{\dot{\bm{q}}}{}{#2}{}{#1}}}
\newrobustcmd{\ddotbmqz}[2][]{\ensuremath{\subp{\ddot{\bm{q}}}{}{#2}{}{#1}}}
\newrobustcmd{\brevebmqz}[2][]{\ensuremath{\subp{\breve{\bm{q}}}{}{#2}{}{#1}}}
\newrobustcmd{\barbmqz}[2][]{\ensuremath{\subp{\bar{\bm{q}}}{}{#2}{}{#1}}}
\newrobustcmd{\vecbmqz}[2][]{\ensuremath{\subp{\vec{\bm{q}}}{}{#2}{}{#1}}}
\newrobustcmd{\rz}[2][]{\ensuremath{\subp{r}{}{#2}{}{#1}}}
\newrobustcmd{\hatrz}[2][]{\ensuremath{\subp{\hat{r}}{}{#2}{}{#1}}}
\newrobustcmd{\widehatrz}[2][]{\ensuremath{\subp{\widehat{r}}{}{#2}{}{#1}}}
\newrobustcmd{\checkrz}[2][]{\ensuremath{\subp{\check{r}}{}{#2}{}{#1}}}
\newrobustcmd{\tilderz}[2][]{\ensuremath{\subp{\tilde{r}}{}{#2}{}{#1}}}
\newrobustcmd{\widetilderz}[2][]{\ensuremath{\subp{\widetilde{r}}{}{#2}{}{#1}}}
\newrobustcmd{\acuterz}[2][]{\ensuremath{\subp{\acute{r}}{}{#2}{}{#1}}}
\newrobustcmd{\graverz}[2][]{\ensuremath{\subp{\grave{r}}{}{#2}{}{#1}}}
\newrobustcmd{\dotrz}[2][]{\ensuremath{\subp{\dot{r}}{}{#2}{}{#1}}}
\newrobustcmd{\ddotrz}[2][]{\ensuremath{\subp{\ddot{r}}{}{#2}{}{#1}}}
\newrobustcmd{\breverz}[2][]{\ensuremath{\subp{\breve{r}}{}{#2}{}{#1}}}
\newrobustcmd{\barrz}[2][]{\ensuremath{\subp{\bar{r}}{}{#2}{}{#1}}}
\newrobustcmd{\vecrz}[2][]{\ensuremath{\subp{\vec{r}}{}{#2}{}{#1}}}
\newrobustcmd{\bmrz}[2][]{\ensuremath{\subp{\bm{r}}{}{#2}{}{#1}}}
\newrobustcmd{\hatbmrz}[2][]{\ensuremath{\subp{\hat{\bm{r}}}{}{#2}{}{#1}}}
\newrobustcmd{\widehatbmrz}[2][]{\ensuremath{\subp{\widehat{\bm{r}}}{}{#2}{}{#1}}}
\newrobustcmd{\checkbmrz}[2][]{\ensuremath{\subp{\check{\bm{r}}}{}{#2}{}{#1}}}
\newrobustcmd{\tildebmrz}[2][]{\ensuremath{\subp{\tilde{\bm{r}}}{}{#2}{}{#1}}}
\newrobustcmd{\widetildebmrz}[2][]{\ensuremath{\subp{\widetilde{\bm{r}}}{}{#2}{}{#1}}}
\newrobustcmd{\acutebmrz}[2][]{\ensuremath{\subp{\acute{\bm{r}}}{}{#2}{}{#1}}}
\newrobustcmd{\gravebmrz}[2][]{\ensuremath{\subp{\grave{\bm{r}}}{}{#2}{}{#1}}}
\newrobustcmd{\dotbmrz}[2][]{\ensuremath{\subp{\dot{\bm{r}}}{}{#2}{}{#1}}}
\newrobustcmd{\ddotbmrz}[2][]{\ensuremath{\subp{\ddot{\bm{r}}}{}{#2}{}{#1}}}
\newrobustcmd{\brevebmrz}[2][]{\ensuremath{\subp{\breve{\bm{r}}}{}{#2}{}{#1}}}
\newrobustcmd{\barbmrz}[2][]{\ensuremath{\subp{\bar{\bm{r}}}{}{#2}{}{#1}}}
\newrobustcmd{\vecbmrz}[2][]{\ensuremath{\subp{\vec{\bm{r}}}{}{#2}{}{#1}}}
\newrobustcmd{\sz}[2][]{\ensuremath{\subp{s}{}{#2}{}{#1}}}
\newrobustcmd{\hatsz}[2][]{\ensuremath{\subp{\hat{s}}{}{#2}{}{#1}}}
\newrobustcmd{\widehatsz}[2][]{\ensuremath{\subp{\widehat{s}}{}{#2}{}{#1}}}
\newrobustcmd{\checksz}[2][]{\ensuremath{\subp{\check{s}}{}{#2}{}{#1}}}
\newrobustcmd{\tildesz}[2][]{\ensuremath{\subp{\tilde{s}}{}{#2}{}{#1}}}
\newrobustcmd{\widetildesz}[2][]{\ensuremath{\subp{\widetilde{s}}{}{#2}{}{#1}}}
\newrobustcmd{\acutesz}[2][]{\ensuremath{\subp{\acute{s}}{}{#2}{}{#1}}}
\newrobustcmd{\gravesz}[2][]{\ensuremath{\subp{\grave{s}}{}{#2}{}{#1}}}
\newrobustcmd{\dotsz}[2][]{\ensuremath{\subp{\dot{s}}{}{#2}{}{#1}}}
\newrobustcmd{\ddotsz}[2][]{\ensuremath{\subp{\ddot{s}}{}{#2}{}{#1}}}
\newrobustcmd{\brevesz}[2][]{\ensuremath{\subp{\breve{s}}{}{#2}{}{#1}}}
\newrobustcmd{\barsz}[2][]{\ensuremath{\subp{\bar{s}}{}{#2}{}{#1}}}
\newrobustcmd{\vecsz}[2][]{\ensuremath{\subp{\vec{s}}{}{#2}{}{#1}}}
\newrobustcmd{\bmsz}[2][]{\ensuremath{\subp{\bm{s}}{}{#2}{}{#1}}}
\newrobustcmd{\hatbmsz}[2][]{\ensuremath{\subp{\hat{\bm{s}}}{}{#2}{}{#1}}}
\newrobustcmd{\widehatbmsz}[2][]{\ensuremath{\subp{\widehat{\bm{s}}}{}{#2}{}{#1}}}
\newrobustcmd{\checkbmsz}[2][]{\ensuremath{\subp{\check{\bm{s}}}{}{#2}{}{#1}}}
\newrobustcmd{\tildebmsz}[2][]{\ensuremath{\subp{\tilde{\bm{s}}}{}{#2}{}{#1}}}
\newrobustcmd{\widetildebmsz}[2][]{\ensuremath{\subp{\widetilde{\bm{s}}}{}{#2}{}{#1}}}
\newrobustcmd{\acutebmsz}[2][]{\ensuremath{\subp{\acute{\bm{s}}}{}{#2}{}{#1}}}
\newrobustcmd{\gravebmsz}[2][]{\ensuremath{\subp{\grave{\bm{s}}}{}{#2}{}{#1}}}
\newrobustcmd{\dotbmsz}[2][]{\ensuremath{\subp{\dot{\bm{s}}}{}{#2}{}{#1}}}
\newrobustcmd{\ddotbmsz}[2][]{\ensuremath{\subp{\ddot{\bm{s}}}{}{#2}{}{#1}}}
\newrobustcmd{\brevebmsz}[2][]{\ensuremath{\subp{\breve{\bm{s}}}{}{#2}{}{#1}}}
\newrobustcmd{\barbmsz}[2][]{\ensuremath{\subp{\bar{\bm{s}}}{}{#2}{}{#1}}}
\newrobustcmd{\vecbmsz}[2][]{\ensuremath{\subp{\vec{\bm{s}}}{}{#2}{}{#1}}}
\newrobustcmd{\tz}[2][]{\ensuremath{\subp{t}{}{#2}{}{#1}}}
\newrobustcmd{\hattz}[2][]{\ensuremath{\subp{\hat{t}}{}{#2}{}{#1}}}
\newrobustcmd{\widehattz}[2][]{\ensuremath{\subp{\widehat{t}}{}{#2}{}{#1}}}
\newrobustcmd{\checktz}[2][]{\ensuremath{\subp{\check{t}}{}{#2}{}{#1}}}
\newrobustcmd{\tildetz}[2][]{\ensuremath{\subp{\tilde{t}}{}{#2}{}{#1}}}
\newrobustcmd{\widetildetz}[2][]{\ensuremath{\subp{\widetilde{t}}{}{#2}{}{#1}}}
\newrobustcmd{\acutetz}[2][]{\ensuremath{\subp{\acute{t}}{}{#2}{}{#1}}}
\newrobustcmd{\gravetz}[2][]{\ensuremath{\subp{\grave{t}}{}{#2}{}{#1}}}
\newrobustcmd{\dottz}[2][]{\ensuremath{\subp{\dot{t}}{}{#2}{}{#1}}}
\newrobustcmd{\ddottz}[2][]{\ensuremath{\subp{\ddot{t}}{}{#2}{}{#1}}}
\newrobustcmd{\brevetz}[2][]{\ensuremath{\subp{\breve{t}}{}{#2}{}{#1}}}
\newrobustcmd{\bartz}[2][]{\ensuremath{\subp{\bar{t}}{}{#2}{}{#1}}}
\newrobustcmd{\vectz}[2][]{\ensuremath{\subp{\vec{t}}{}{#2}{}{#1}}}
\newrobustcmd{\bmtz}[2][]{\ensuremath{\subp{\bm{t}}{}{#2}{}{#1}}}
\newrobustcmd{\hatbmtz}[2][]{\ensuremath{\subp{\hat{\bm{t}}}{}{#2}{}{#1}}}
\newrobustcmd{\widehatbmtz}[2][]{\ensuremath{\subp{\widehat{\bm{t}}}{}{#2}{}{#1}}}
\newrobustcmd{\checkbmtz}[2][]{\ensuremath{\subp{\check{\bm{t}}}{}{#2}{}{#1}}}
\newrobustcmd{\tildebmtz}[2][]{\ensuremath{\subp{\tilde{\bm{t}}}{}{#2}{}{#1}}}
\newrobustcmd{\widetildebmtz}[2][]{\ensuremath{\subp{\widetilde{\bm{t}}}{}{#2}{}{#1}}}
\newrobustcmd{\acutebmtz}[2][]{\ensuremath{\subp{\acute{\bm{t}}}{}{#2}{}{#1}}}
\newrobustcmd{\gravebmtz}[2][]{\ensuremath{\subp{\grave{\bm{t}}}{}{#2}{}{#1}}}
\newrobustcmd{\dotbmtz}[2][]{\ensuremath{\subp{\dot{\bm{t}}}{}{#2}{}{#1}}}
\newrobustcmd{\ddotbmtz}[2][]{\ensuremath{\subp{\ddot{\bm{t}}}{}{#2}{}{#1}}}
\newrobustcmd{\brevebmtz}[2][]{\ensuremath{\subp{\breve{\bm{t}}}{}{#2}{}{#1}}}
\newrobustcmd{\barbmtz}[2][]{\ensuremath{\subp{\bar{\bm{t}}}{}{#2}{}{#1}}}
\newrobustcmd{\vecbmtz}[2][]{\ensuremath{\subp{\vec{\bm{t}}}{}{#2}{}{#1}}}
\newrobustcmd{\uz}[2][]{\ensuremath{\subp{u}{}{#2}{}{#1}}}
\newrobustcmd{\hatuz}[2][]{\ensuremath{\subp{\hat{u}}{}{#2}{}{#1}}}
\newrobustcmd{\widehatuz}[2][]{\ensuremath{\subp{\widehat{u}}{}{#2}{}{#1}}}
\newrobustcmd{\checkuz}[2][]{\ensuremath{\subp{\check{u}}{}{#2}{}{#1}}}
\newrobustcmd{\tildeuz}[2][]{\ensuremath{\subp{\tilde{u}}{}{#2}{}{#1}}}
\newrobustcmd{\widetildeuz}[2][]{\ensuremath{\subp{\widetilde{u}}{}{#2}{}{#1}}}
\newrobustcmd{\acuteuz}[2][]{\ensuremath{\subp{\acute{u}}{}{#2}{}{#1}}}
\newrobustcmd{\graveuz}[2][]{\ensuremath{\subp{\grave{u}}{}{#2}{}{#1}}}
\newrobustcmd{\dotuz}[2][]{\ensuremath{\subp{\dot{u}}{}{#2}{}{#1}}}
\newrobustcmd{\ddotuz}[2][]{\ensuremath{\subp{\ddot{u}}{}{#2}{}{#1}}}
\newrobustcmd{\breveuz}[2][]{\ensuremath{\subp{\breve{u}}{}{#2}{}{#1}}}
\newrobustcmd{\baruz}[2][]{\ensuremath{\subp{\bar{u}}{}{#2}{}{#1}}}
\newrobustcmd{\vecuz}[2][]{\ensuremath{\subp{\vec{u}}{}{#2}{}{#1}}}
\newrobustcmd{\bmuz}[2][]{\ensuremath{\subp{\bm{u}}{}{#2}{}{#1}}}
\newrobustcmd{\hatbmuz}[2][]{\ensuremath{\subp{\hat{\bm{u}}}{}{#2}{}{#1}}}
\newrobustcmd{\widehatbmuz}[2][]{\ensuremath{\subp{\widehat{\bm{u}}}{}{#2}{}{#1}}}
\newrobustcmd{\checkbmuz}[2][]{\ensuremath{\subp{\check{\bm{u}}}{}{#2}{}{#1}}}
\newrobustcmd{\tildebmuz}[2][]{\ensuremath{\subp{\tilde{\bm{u}}}{}{#2}{}{#1}}}
\newrobustcmd{\widetildebmuz}[2][]{\ensuremath{\subp{\widetilde{\bm{u}}}{}{#2}{}{#1}}}
\newrobustcmd{\acutebmuz}[2][]{\ensuremath{\subp{\acute{\bm{u}}}{}{#2}{}{#1}}}
\newrobustcmd{\gravebmuz}[2][]{\ensuremath{\subp{\grave{\bm{u}}}{}{#2}{}{#1}}}
\newrobustcmd{\dotbmuz}[2][]{\ensuremath{\subp{\dot{\bm{u}}}{}{#2}{}{#1}}}
\newrobustcmd{\ddotbmuz}[2][]{\ensuremath{\subp{\ddot{\bm{u}}}{}{#2}{}{#1}}}
\newrobustcmd{\brevebmuz}[2][]{\ensuremath{\subp{\breve{\bm{u}}}{}{#2}{}{#1}}}
\newrobustcmd{\barbmuz}[2][]{\ensuremath{\subp{\bar{\bm{u}}}{}{#2}{}{#1}}}
\newrobustcmd{\vecbmuz}[2][]{\ensuremath{\subp{\vec{\bm{u}}}{}{#2}{}{#1}}}
\newrobustcmd{\vz}[2][]{\ensuremath{\subp{v}{}{#2}{}{#1}}}
\newrobustcmd{\hatvz}[2][]{\ensuremath{\subp{\hat{v}}{}{#2}{}{#1}}}
\newrobustcmd{\widehatvz}[2][]{\ensuremath{\subp{\widehat{v}}{}{#2}{}{#1}}}
\newrobustcmd{\checkvz}[2][]{\ensuremath{\subp{\check{v}}{}{#2}{}{#1}}}
\newrobustcmd{\tildevz}[2][]{\ensuremath{\subp{\tilde{v}}{}{#2}{}{#1}}}
\newrobustcmd{\widetildevz}[2][]{\ensuremath{\subp{\widetilde{v}}{}{#2}{}{#1}}}
\newrobustcmd{\acutevz}[2][]{\ensuremath{\subp{\acute{v}}{}{#2}{}{#1}}}
\newrobustcmd{\gravevz}[2][]{\ensuremath{\subp{\grave{v}}{}{#2}{}{#1}}}
\newrobustcmd{\dotvz}[2][]{\ensuremath{\subp{\dot{v}}{}{#2}{}{#1}}}
\newrobustcmd{\ddotvz}[2][]{\ensuremath{\subp{\ddot{v}}{}{#2}{}{#1}}}
\newrobustcmd{\brevevz}[2][]{\ensuremath{\subp{\breve{v}}{}{#2}{}{#1}}}
\newrobustcmd{\barvz}[2][]{\ensuremath{\subp{\bar{v}}{}{#2}{}{#1}}}
\newrobustcmd{\vecvz}[2][]{\ensuremath{\subp{\vec{v}}{}{#2}{}{#1}}}
\newrobustcmd{\bmvz}[2][]{\ensuremath{\subp{\bm{v}}{}{#2}{}{#1}}}
\newrobustcmd{\hatbmvz}[2][]{\ensuremath{\subp{\hat{\bm{v}}}{}{#2}{}{#1}}}
\newrobustcmd{\widehatbmvz}[2][]{\ensuremath{\subp{\widehat{\bm{v}}}{}{#2}{}{#1}}}
\newrobustcmd{\checkbmvz}[2][]{\ensuremath{\subp{\check{\bm{v}}}{}{#2}{}{#1}}}
\newrobustcmd{\tildebmvz}[2][]{\ensuremath{\subp{\tilde{\bm{v}}}{}{#2}{}{#1}}}
\newrobustcmd{\widetildebmvz}[2][]{\ensuremath{\subp{\widetilde{\bm{v}}}{}{#2}{}{#1}}}
\newrobustcmd{\acutebmvz}[2][]{\ensuremath{\subp{\acute{\bm{v}}}{}{#2}{}{#1}}}
\newrobustcmd{\gravebmvz}[2][]{\ensuremath{\subp{\grave{\bm{v}}}{}{#2}{}{#1}}}
\newrobustcmd{\dotbmvz}[2][]{\ensuremath{\subp{\dot{\bm{v}}}{}{#2}{}{#1}}}
\newrobustcmd{\ddotbmvz}[2][]{\ensuremath{\subp{\ddot{\bm{v}}}{}{#2}{}{#1}}}
\newrobustcmd{\brevebmvz}[2][]{\ensuremath{\subp{\breve{\bm{v}}}{}{#2}{}{#1}}}
\newrobustcmd{\barbmvz}[2][]{\ensuremath{\subp{\bar{\bm{v}}}{}{#2}{}{#1}}}
\newrobustcmd{\vecbmvz}[2][]{\ensuremath{\subp{\vec{\bm{v}}}{}{#2}{}{#1}}}
\newrobustcmd{\wz}[2][]{\ensuremath{\subp{w}{}{#2}{}{#1}}}
\newrobustcmd{\hatwz}[2][]{\ensuremath{\subp{\hat{w}}{}{#2}{}{#1}}}
\newrobustcmd{\widehatwz}[2][]{\ensuremath{\subp{\widehat{w}}{}{#2}{}{#1}}}
\newrobustcmd{\checkwz}[2][]{\ensuremath{\subp{\check{w}}{}{#2}{}{#1}}}
\newrobustcmd{\tildewz}[2][]{\ensuremath{\subp{\tilde{w}}{}{#2}{}{#1}}}
\newrobustcmd{\widetildewz}[2][]{\ensuremath{\subp{\widetilde{w}}{}{#2}{}{#1}}}
\newrobustcmd{\acutewz}[2][]{\ensuremath{\subp{\acute{w}}{}{#2}{}{#1}}}
\newrobustcmd{\gravewz}[2][]{\ensuremath{\subp{\grave{w}}{}{#2}{}{#1}}}
\newrobustcmd{\dotwz}[2][]{\ensuremath{\subp{\dot{w}}{}{#2}{}{#1}}}
\newrobustcmd{\ddotwz}[2][]{\ensuremath{\subp{\ddot{w}}{}{#2}{}{#1}}}
\newrobustcmd{\brevewz}[2][]{\ensuremath{\subp{\breve{w}}{}{#2}{}{#1}}}
\newrobustcmd{\barwz}[2][]{\ensuremath{\subp{\bar{w}}{}{#2}{}{#1}}}
\newrobustcmd{\vecwz}[2][]{\ensuremath{\subp{\vec{w}}{}{#2}{}{#1}}}
\newrobustcmd{\bmwz}[2][]{\ensuremath{\subp{\bm{w}}{}{#2}{}{#1}}}
\newrobustcmd{\hatbmwz}[2][]{\ensuremath{\subp{\hat{\bm{w}}}{}{#2}{}{#1}}}
\newrobustcmd{\widehatbmwz}[2][]{\ensuremath{\subp{\widehat{\bm{w}}}{}{#2}{}{#1}}}
\newrobustcmd{\checkbmwz}[2][]{\ensuremath{\subp{\check{\bm{w}}}{}{#2}{}{#1}}}
\newrobustcmd{\tildebmwz}[2][]{\ensuremath{\subp{\tilde{\bm{w}}}{}{#2}{}{#1}}}
\newrobustcmd{\widetildebmwz}[2][]{\ensuremath{\subp{\widetilde{\bm{w}}}{}{#2}{}{#1}}}
\newrobustcmd{\acutebmwz}[2][]{\ensuremath{\subp{\acute{\bm{w}}}{}{#2}{}{#1}}}
\newrobustcmd{\gravebmwz}[2][]{\ensuremath{\subp{\grave{\bm{w}}}{}{#2}{}{#1}}}
\newrobustcmd{\dotbmwz}[2][]{\ensuremath{\subp{\dot{\bm{w}}}{}{#2}{}{#1}}}
\newrobustcmd{\ddotbmwz}[2][]{\ensuremath{\subp{\ddot{\bm{w}}}{}{#2}{}{#1}}}
\newrobustcmd{\brevebmwz}[2][]{\ensuremath{\subp{\breve{\bm{w}}}{}{#2}{}{#1}}}
\newrobustcmd{\barbmwz}[2][]{\ensuremath{\subp{\bar{\bm{w}}}{}{#2}{}{#1}}}
\newrobustcmd{\vecbmwz}[2][]{\ensuremath{\subp{\vec{\bm{w}}}{}{#2}{}{#1}}}
\newrobustcmd{\xz}[2][]{\ensuremath{\subp{x}{}{#2}{}{#1}}}
\newrobustcmd{\hatxz}[2][]{\ensuremath{\subp{\hat{x}}{}{#2}{}{#1}}}
\newrobustcmd{\widehatxz}[2][]{\ensuremath{\subp{\widehat{x}}{}{#2}{}{#1}}}
\newrobustcmd{\checkxz}[2][]{\ensuremath{\subp{\check{x}}{}{#2}{}{#1}}}
\newrobustcmd{\tildexz}[2][]{\ensuremath{\subp{\tilde{x}}{}{#2}{}{#1}}}
\newrobustcmd{\widetildexz}[2][]{\ensuremath{\subp{\widetilde{x}}{}{#2}{}{#1}}}
\newrobustcmd{\acutexz}[2][]{\ensuremath{\subp{\acute{x}}{}{#2}{}{#1}}}
\newrobustcmd{\gravexz}[2][]{\ensuremath{\subp{\grave{x}}{}{#2}{}{#1}}}
\newrobustcmd{\dotxz}[2][]{\ensuremath{\subp{\dot{x}}{}{#2}{}{#1}}}
\newrobustcmd{\ddotxz}[2][]{\ensuremath{\subp{\ddot{x}}{}{#2}{}{#1}}}
\newrobustcmd{\brevexz}[2][]{\ensuremath{\subp{\breve{x}}{}{#2}{}{#1}}}
\newrobustcmd{\barxz}[2][]{\ensuremath{\subp{\bar{x}}{}{#2}{}{#1}}}
\newrobustcmd{\vecxz}[2][]{\ensuremath{\subp{\vec{x}}{}{#2}{}{#1}}}
\newrobustcmd{\bmxz}[2][]{\ensuremath{\subp{\bm{x}}{}{#2}{}{#1}}}
\newrobustcmd{\hatbmxz}[2][]{\ensuremath{\subp{\hat{\bm{x}}}{}{#2}{}{#1}}}
\newrobustcmd{\widehatbmxz}[2][]{\ensuremath{\subp{\widehat{\bm{x}}}{}{#2}{}{#1}}}
\newrobustcmd{\checkbmxz}[2][]{\ensuremath{\subp{\check{\bm{x}}}{}{#2}{}{#1}}}
\newrobustcmd{\tildebmxz}[2][]{\ensuremath{\subp{\tilde{\bm{x}}}{}{#2}{}{#1}}}
\newrobustcmd{\widetildebmxz}[2][]{\ensuremath{\subp{\widetilde{\bm{x}}}{}{#2}{}{#1}}}
\newrobustcmd{\acutebmxz}[2][]{\ensuremath{\subp{\acute{\bm{x}}}{}{#2}{}{#1}}}
\newrobustcmd{\gravebmxz}[2][]{\ensuremath{\subp{\grave{\bm{x}}}{}{#2}{}{#1}}}
\newrobustcmd{\dotbmxz}[2][]{\ensuremath{\subp{\dot{\bm{x}}}{}{#2}{}{#1}}}
\newrobustcmd{\ddotbmxz}[2][]{\ensuremath{\subp{\ddot{\bm{x}}}{}{#2}{}{#1}}}
\newrobustcmd{\brevebmxz}[2][]{\ensuremath{\subp{\breve{\bm{x}}}{}{#2}{}{#1}}}
\newrobustcmd{\barbmxz}[2][]{\ensuremath{\subp{\bar{\bm{x}}}{}{#2}{}{#1}}}
\newrobustcmd{\vecbmxz}[2][]{\ensuremath{\subp{\vec{\bm{x}}}{}{#2}{}{#1}}}
\newrobustcmd{\yz}[2][]{\ensuremath{\subp{y}{}{#2}{}{#1}}}
\newrobustcmd{\hatyz}[2][]{\ensuremath{\subp{\hat{y}}{}{#2}{}{#1}}}
\newrobustcmd{\widehatyz}[2][]{\ensuremath{\subp{\widehat{y}}{}{#2}{}{#1}}}
\newrobustcmd{\checkyz}[2][]{\ensuremath{\subp{\check{y}}{}{#2}{}{#1}}}
\newrobustcmd{\tildeyz}[2][]{\ensuremath{\subp{\tilde{y}}{}{#2}{}{#1}}}
\newrobustcmd{\widetildeyz}[2][]{\ensuremath{\subp{\widetilde{y}}{}{#2}{}{#1}}}
\newrobustcmd{\acuteyz}[2][]{\ensuremath{\subp{\acute{y}}{}{#2}{}{#1}}}
\newrobustcmd{\graveyz}[2][]{\ensuremath{\subp{\grave{y}}{}{#2}{}{#1}}}
\newrobustcmd{\dotyz}[2][]{\ensuremath{\subp{\dot{y}}{}{#2}{}{#1}}}
\newrobustcmd{\ddotyz}[2][]{\ensuremath{\subp{\ddot{y}}{}{#2}{}{#1}}}
\newrobustcmd{\breveyz}[2][]{\ensuremath{\subp{\breve{y}}{}{#2}{}{#1}}}
\newrobustcmd{\baryz}[2][]{\ensuremath{\subp{\bar{y}}{}{#2}{}{#1}}}
\newrobustcmd{\vecyz}[2][]{\ensuremath{\subp{\vec{y}}{}{#2}{}{#1}}}
\newrobustcmd{\bmyz}[2][]{\ensuremath{\subp{\bm{y}}{}{#2}{}{#1}}}
\newrobustcmd{\hatbmyz}[2][]{\ensuremath{\subp{\hat{\bm{y}}}{}{#2}{}{#1}}}
\newrobustcmd{\widehatbmyz}[2][]{\ensuremath{\subp{\widehat{\bm{y}}}{}{#2}{}{#1}}}
\newrobustcmd{\checkbmyz}[2][]{\ensuremath{\subp{\check{\bm{y}}}{}{#2}{}{#1}}}
\newrobustcmd{\tildebmyz}[2][]{\ensuremath{\subp{\tilde{\bm{y}}}{}{#2}{}{#1}}}
\newrobustcmd{\widetildebmyz}[2][]{\ensuremath{\subp{\widetilde{\bm{y}}}{}{#2}{}{#1}}}
\newrobustcmd{\acutebmyz}[2][]{\ensuremath{\subp{\acute{\bm{y}}}{}{#2}{}{#1}}}
\newrobustcmd{\gravebmyz}[2][]{\ensuremath{\subp{\grave{\bm{y}}}{}{#2}{}{#1}}}
\newrobustcmd{\dotbmyz}[2][]{\ensuremath{\subp{\dot{\bm{y}}}{}{#2}{}{#1}}}
\newrobustcmd{\ddotbmyz}[2][]{\ensuremath{\subp{\ddot{\bm{y}}}{}{#2}{}{#1}}}
\newrobustcmd{\brevebmyz}[2][]{\ensuremath{\subp{\breve{\bm{y}}}{}{#2}{}{#1}}}
\newrobustcmd{\barbmyz}[2][]{\ensuremath{\subp{\bar{\bm{y}}}{}{#2}{}{#1}}}
\newrobustcmd{\vecbmyz}[2][]{\ensuremath{\subp{\vec{\bm{y}}}{}{#2}{}{#1}}}
\newrobustcmd{\zz}[2][]{\ensuremath{\subp{z}{}{#2}{}{#1}}}
\newrobustcmd{\hatzz}[2][]{\ensuremath{\subp{\hat{z}}{}{#2}{}{#1}}}
\newrobustcmd{\widehatzz}[2][]{\ensuremath{\subp{\widehat{z}}{}{#2}{}{#1}}}
\newrobustcmd{\checkzz}[2][]{\ensuremath{\subp{\check{z}}{}{#2}{}{#1}}}
\newrobustcmd{\tildezz}[2][]{\ensuremath{\subp{\tilde{z}}{}{#2}{}{#1}}}
\newrobustcmd{\widetildezz}[2][]{\ensuremath{\subp{\widetilde{z}}{}{#2}{}{#1}}}
\newrobustcmd{\acutezz}[2][]{\ensuremath{\subp{\acute{z}}{}{#2}{}{#1}}}
\newrobustcmd{\gravezz}[2][]{\ensuremath{\subp{\grave{z}}{}{#2}{}{#1}}}
\newrobustcmd{\dotzz}[2][]{\ensuremath{\subp{\dot{z}}{}{#2}{}{#1}}}
\newrobustcmd{\ddotzz}[2][]{\ensuremath{\subp{\ddot{z}}{}{#2}{}{#1}}}
\newrobustcmd{\brevezz}[2][]{\ensuremath{\subp{\breve{z}}{}{#2}{}{#1}}}
\newrobustcmd{\barzz}[2][]{\ensuremath{\subp{\bar{z}}{}{#2}{}{#1}}}
\newrobustcmd{\veczz}[2][]{\ensuremath{\subp{\vec{z}}{}{#2}{}{#1}}}
\newrobustcmd{\bmzz}[2][]{\ensuremath{\subp{\bm{z}}{}{#2}{}{#1}}}
\newrobustcmd{\hatbmzz}[2][]{\ensuremath{\subp{\hat{\bm{z}}}{}{#2}{}{#1}}}
\newrobustcmd{\widehatbmzz}[2][]{\ensuremath{\subp{\widehat{\bm{z}}}{}{#2}{}{#1}}}
\newrobustcmd{\checkbmzz}[2][]{\ensuremath{\subp{\check{\bm{z}}}{}{#2}{}{#1}}}
\newrobustcmd{\tildebmzz}[2][]{\ensuremath{\subp{\tilde{\bm{z}}}{}{#2}{}{#1}}}
\newrobustcmd{\widetildebmzz}[2][]{\ensuremath{\subp{\widetilde{\bm{z}}}{}{#2}{}{#1}}}
\newrobustcmd{\acutebmzz}[2][]{\ensuremath{\subp{\acute{\bm{z}}}{}{#2}{}{#1}}}
\newrobustcmd{\gravebmzz}[2][]{\ensuremath{\subp{\grave{\bm{z}}}{}{#2}{}{#1}}}
\newrobustcmd{\dotbmzz}[2][]{\ensuremath{\subp{\dot{\bm{z}}}{}{#2}{}{#1}}}
\newrobustcmd{\ddotbmzz}[2][]{\ensuremath{\subp{\ddot{\bm{z}}}{}{#2}{}{#1}}}
\newrobustcmd{\brevebmzz}[2][]{\ensuremath{\subp{\breve{\bm{z}}}{}{#2}{}{#1}}}
\newrobustcmd{\barbmzz}[2][]{\ensuremath{\subp{\bar{\bm{z}}}{}{#2}{}{#1}}}
\newrobustcmd{\vecbmzz}[2][]{\ensuremath{\subp{\vec{\bm{z}}}{}{#2}{}{#1}}}

\newrobustcmd{\Az}[2][]{\ensuremath{\subp{A}{}{#2}{}{#1}}}
\newrobustcmd{\hatAz}[2][]{\ensuremath{\subp{\hat{A}}{}{#2}{}{#1}}}
\newrobustcmd{\widehatAz}[2][]{\ensuremath{\subp{\widehat{A}}{}{#2}{}{#1}}}
\newrobustcmd{\checkAz}[2][]{\ensuremath{\subp{\check{A}}{}{#2}{}{#1}}}
\newrobustcmd{\tildeAz}[2][]{\ensuremath{\subp{\tilde{A}}{}{#2}{}{#1}}}
\newrobustcmd{\widetildeAz}[2][]{\ensuremath{\subp{\widetilde{A}}{}{#2}{}{#1}}}
\newrobustcmd{\acuteAz}[2][]{\ensuremath{\subp{\acute{A}}{}{#2}{}{#1}}}
\newrobustcmd{\graveAz}[2][]{\ensuremath{\subp{\grave{A}}{}{#2}{}{#1}}}
\newrobustcmd{\dotAz}[2][]{\ensuremath{\subp{\dot{A}}{}{#2}{}{#1}}}
\newrobustcmd{\ddotAz}[2][]{\ensuremath{\subp{\ddot{A}}{}{#2}{}{#1}}}
\newrobustcmd{\breveAz}[2][]{\ensuremath{\subp{\breve{A}}{}{#2}{}{#1}}}
\newrobustcmd{\barAz}[2][]{\ensuremath{\subp{\bar{A}}{}{#2}{}{#1}}}
\newrobustcmd{\vecAz}[2][]{\ensuremath{\subp{\vec{A}}{}{#2}{}{#1}}}
\newrobustcmd{\bmAz}[2][]{\ensuremath{\subp{\bm{A}}{}{#2}{}{#1}}}
\newrobustcmd{\hatbmAz}[2][]{\ensuremath{\subp{\hat{\bm{A}}}{}{#2}{}{#1}}}
\newrobustcmd{\widehatbmAz}[2][]{\ensuremath{\subp{\widehat{\bm{A}}}{}{#2}{}{#1}}}
\newrobustcmd{\checkbmAz}[2][]{\ensuremath{\subp{\check{\bm{A}}}{}{#2}{}{#1}}}
\newrobustcmd{\tildebmAz}[2][]{\ensuremath{\subp{\tilde{\bm{A}}}{}{#2}{}{#1}}}
\newrobustcmd{\widetildebmAz}[2][]{\ensuremath{\subp{\widetilde{\bm{A}}}{}{#2}{}{#1}}}
\newrobustcmd{\acutebmAz}[2][]{\ensuremath{\subp{\acute{\bm{A}}}{}{#2}{}{#1}}}
\newrobustcmd{\gravebmAz}[2][]{\ensuremath{\subp{\grave{\bm{A}}}{}{#2}{}{#1}}}
\newrobustcmd{\dotbmAz}[2][]{\ensuremath{\subp{\dot{\bm{A}}}{}{#2}{}{#1}}}
\newrobustcmd{\ddotbmAz}[2][]{\ensuremath{\subp{\ddot{\bm{A}}}{}{#2}{}{#1}}}
\newrobustcmd{\brevebmAz}[2][]{\ensuremath{\subp{\breve{\bm{A}}}{}{#2}{}{#1}}}
\newrobustcmd{\barbmAz}[2][]{\ensuremath{\subp{\bar{\bm{A}}}{}{#2}{}{#1}}}
\newrobustcmd{\vecbmAz}[2][]{\ensuremath{\subp{\vec{\bm{A}}}{}{#2}{}{#1}}}
\newrobustcmd{\Bz}[2][]{\ensuremath{\subp{B}{}{#2}{}{#1}}}
\newrobustcmd{\hatBz}[2][]{\ensuremath{\subp{\hat{B}}{}{#2}{}{#1}}}
\newrobustcmd{\widehatBz}[2][]{\ensuremath{\subp{\widehat{B}}{}{#2}{}{#1}}}
\newrobustcmd{\checkBz}[2][]{\ensuremath{\subp{\check{B}}{}{#2}{}{#1}}}
\newrobustcmd{\tildeBz}[2][]{\ensuremath{\subp{\tilde{B}}{}{#2}{}{#1}}}
\newrobustcmd{\widetildeBz}[2][]{\ensuremath{\subp{\widetilde{B}}{}{#2}{}{#1}}}
\newrobustcmd{\acuteBz}[2][]{\ensuremath{\subp{\acute{B}}{}{#2}{}{#1}}}
\newrobustcmd{\graveBz}[2][]{\ensuremath{\subp{\grave{B}}{}{#2}{}{#1}}}
\newrobustcmd{\dotBz}[2][]{\ensuremath{\subp{\dot{B}}{}{#2}{}{#1}}}
\newrobustcmd{\ddotBz}[2][]{\ensuremath{\subp{\ddot{B}}{}{#2}{}{#1}}}
\newrobustcmd{\breveBz}[2][]{\ensuremath{\subp{\breve{B}}{}{#2}{}{#1}}}
\newrobustcmd{\barBz}[2][]{\ensuremath{\subp{\bar{B}}{}{#2}{}{#1}}}
\newrobustcmd{\vecBz}[2][]{\ensuremath{\subp{\vec{B}}{}{#2}{}{#1}}}
\newrobustcmd{\bmBz}[2][]{\ensuremath{\subp{\bm{B}}{}{#2}{}{#1}}}
\newrobustcmd{\hatbmBz}[2][]{\ensuremath{\subp{\hat{\bm{B}}}{}{#2}{}{#1}}}
\newrobustcmd{\widehatbmBz}[2][]{\ensuremath{\subp{\widehat{\bm{B}}}{}{#2}{}{#1}}}
\newrobustcmd{\checkbmBz}[2][]{\ensuremath{\subp{\check{\bm{B}}}{}{#2}{}{#1}}}
\newrobustcmd{\tildebmBz}[2][]{\ensuremath{\subp{\tilde{\bm{B}}}{}{#2}{}{#1}}}
\newrobustcmd{\widetildebmBz}[2][]{\ensuremath{\subp{\widetilde{\bm{B}}}{}{#2}{}{#1}}}
\newrobustcmd{\acutebmBz}[2][]{\ensuremath{\subp{\acute{\bm{B}}}{}{#2}{}{#1}}}
\newrobustcmd{\gravebmBz}[2][]{\ensuremath{\subp{\grave{\bm{B}}}{}{#2}{}{#1}}}
\newrobustcmd{\dotbmBz}[2][]{\ensuremath{\subp{\dot{\bm{B}}}{}{#2}{}{#1}}}
\newrobustcmd{\ddotbmBz}[2][]{\ensuremath{\subp{\ddot{\bm{B}}}{}{#2}{}{#1}}}
\newrobustcmd{\brevebmBz}[2][]{\ensuremath{\subp{\breve{\bm{B}}}{}{#2}{}{#1}}}
\newrobustcmd{\barbmBz}[2][]{\ensuremath{\subp{\bar{\bm{B}}}{}{#2}{}{#1}}}
\newrobustcmd{\vecbmBz}[2][]{\ensuremath{\subp{\vec{\bm{B}}}{}{#2}{}{#1}}}
\newrobustcmd{\Cz}[2][]{\ensuremath{\subp{C}{}{#2}{}{#1}}}
\newrobustcmd{\hatCz}[2][]{\ensuremath{\subp{\hat{C}}{}{#2}{}{#1}}}
\newrobustcmd{\widehatCz}[2][]{\ensuremath{\subp{\widehat{C}}{}{#2}{}{#1}}}
\newrobustcmd{\checkCz}[2][]{\ensuremath{\subp{\check{C}}{}{#2}{}{#1}}}
\newrobustcmd{\tildeCz}[2][]{\ensuremath{\subp{\tilde{C}}{}{#2}{}{#1}}}
\newrobustcmd{\widetildeCz}[2][]{\ensuremath{\subp{\widetilde{C}}{}{#2}{}{#1}}}
\newrobustcmd{\acuteCz}[2][]{\ensuremath{\subp{\acute{C}}{}{#2}{}{#1}}}
\newrobustcmd{\graveCz}[2][]{\ensuremath{\subp{\grave{C}}{}{#2}{}{#1}}}
\newrobustcmd{\dotCz}[2][]{\ensuremath{\subp{\dot{C}}{}{#2}{}{#1}}}
\newrobustcmd{\ddotCz}[2][]{\ensuremath{\subp{\ddot{C}}{}{#2}{}{#1}}}
\newrobustcmd{\breveCz}[2][]{\ensuremath{\subp{\breve{C}}{}{#2}{}{#1}}}
\newrobustcmd{\barCz}[2][]{\ensuremath{\subp{\bar{C}}{}{#2}{}{#1}}}
\newrobustcmd{\vecCz}[2][]{\ensuremath{\subp{\vec{C}}{}{#2}{}{#1}}}
\newrobustcmd{\bmCz}[2][]{\ensuremath{\subp{\bm{C}}{}{#2}{}{#1}}}
\newrobustcmd{\hatbmCz}[2][]{\ensuremath{\subp{\hat{\bm{C}}}{}{#2}{}{#1}}}
\newrobustcmd{\widehatbmCz}[2][]{\ensuremath{\subp{\widehat{\bm{C}}}{}{#2}{}{#1}}}
\newrobustcmd{\checkbmCz}[2][]{\ensuremath{\subp{\check{\bm{C}}}{}{#2}{}{#1}}}
\newrobustcmd{\tildebmCz}[2][]{\ensuremath{\subp{\tilde{\bm{C}}}{}{#2}{}{#1}}}
\newrobustcmd{\widetildebmCz}[2][]{\ensuremath{\subp{\widetilde{\bm{C}}}{}{#2}{}{#1}}}
\newrobustcmd{\acutebmCz}[2][]{\ensuremath{\subp{\acute{\bm{C}}}{}{#2}{}{#1}}}
\newrobustcmd{\gravebmCz}[2][]{\ensuremath{\subp{\grave{\bm{C}}}{}{#2}{}{#1}}}
\newrobustcmd{\dotbmCz}[2][]{\ensuremath{\subp{\dot{\bm{C}}}{}{#2}{}{#1}}}
\newrobustcmd{\ddotbmCz}[2][]{\ensuremath{\subp{\ddot{\bm{C}}}{}{#2}{}{#1}}}
\newrobustcmd{\brevebmCz}[2][]{\ensuremath{\subp{\breve{\bm{C}}}{}{#2}{}{#1}}}
\newrobustcmd{\barbmCz}[2][]{\ensuremath{\subp{\bar{\bm{C}}}{}{#2}{}{#1}}}
\newrobustcmd{\vecbmCz}[2][]{\ensuremath{\subp{\vec{\bm{C}}}{}{#2}{}{#1}}}

\newrobustcmd{\Dz}[2][]{\ensuremath{\subp{D}{}{#2}{}{#1}}}
\newrobustcmd{\hatDz}[2][]{\ensuremath{\subp{\hat{D}}{}{#2}{}{#1}}}
\newrobustcmd{\widehatDz}[2][]{\ensuremath{\subp{\widehat{D}}{}{#2}{}{#1}}}
\newrobustcmd{\checkDz}[2][]{\ensuremath{\subp{\check{D}}{}{#2}{}{#1}}}
\newrobustcmd{\tildeDz}[2][]{\ensuremath{\subp{\tilde{D}}{}{#2}{}{#1}}}
\newrobustcmd{\widetildeDz}[2][]{\ensuremath{\subp{\widetilde{D}}{}{#2}{}{#1}}}
\newrobustcmd{\acuteDz}[2][]{\ensuremath{\subp{\acute{D}}{}{#2}{}{#1}}}
\newrobustcmd{\graveDz}[2][]{\ensuremath{\subp{\grave{D}}{}{#2}{}{#1}}}
\newrobustcmd{\dotDz}[2][]{\ensuremath{\subp{\dot{D}}{}{#2}{}{#1}}}
\newrobustcmd{\ddotDz}[2][]{\ensuremath{\subp{\ddot{D}}{}{#2}{}{#1}}}
\newrobustcmd{\breveDz}[2][]{\ensuremath{\subp{\breve{D}}{}{#2}{}{#1}}}
\newrobustcmd{\barDz}[2][]{\ensuremath{\subp{\bar{D}}{}{#2}{}{#1}}}
\newrobustcmd{\vecDz}[2][]{\ensuremath{\subp{\vec{D}}{}{#2}{}{#1}}}
\newrobustcmd{\bmDz}[2][]{\ensuremath{\subp{\bm{D}}{}{#2}{}{#1}}}
\newrobustcmd{\hatbmDz}[2][]{\ensuremath{\subp{\hat{\bm{D}}}{}{#2}{}{#1}}}
\newrobustcmd{\widehatbmDz}[2][]{\ensuremath{\subp{\widehat{\bm{D}}}{}{#2}{}{#1}}}
\newrobustcmd{\checkbmDz}[2][]{\ensuremath{\subp{\check{\bm{D}}}{}{#2}{}{#1}}}
\newrobustcmd{\tildebmDz}[2][]{\ensuremath{\subp{\tilde{\bm{D}}}{}{#2}{}{#1}}}
\newrobustcmd{\widetildebmDz}[2][]{\ensuremath{\subp{\widetilde{\bm{D}}}{}{#2}{}{#1}}}
\newrobustcmd{\acutebmDz}[2][]{\ensuremath{\subp{\acute{\bm{D}}}{}{#2}{}{#1}}}
\newrobustcmd{\gravebmDz}[2][]{\ensuremath{\subp{\grave{\bm{D}}}{}{#2}{}{#1}}}
\newrobustcmd{\dotbmDz}[2][]{\ensuremath{\subp{\dot{\bm{D}}}{}{#2}{}{#1}}}
\newrobustcmd{\ddotbmDz}[2][]{\ensuremath{\subp{\ddot{\bm{D}}}{}{#2}{}{#1}}}
\newrobustcmd{\brevebmDz}[2][]{\ensuremath{\subp{\breve{\bm{D}}}{}{#2}{}{#1}}}
\newrobustcmd{\barbmDz}[2][]{\ensuremath{\subp{\bar{\bm{D}}}{}{#2}{}{#1}}}
\newrobustcmd{\vecbmDz}[2][]{\ensuremath{\subp{\vec{\bm{D}}}{}{#2}{}{#1}}}
\newrobustcmd{\Ez}[2][]{\ensuremath{\subp{E}{}{#2}{}{#1}}}
\newrobustcmd{\hatEz}[2][]{\ensuremath{\subp{\hat{E}}{}{#2}{}{#1}}}
\newrobustcmd{\widehatEz}[2][]{\ensuremath{\subp{\widehat{E}}{}{#2}{}{#1}}}
\newrobustcmd{\checkEz}[2][]{\ensuremath{\subp{\check{E}}{}{#2}{}{#1}}}
\newrobustcmd{\tildeEz}[2][]{\ensuremath{\subp{\tilde{E}}{}{#2}{}{#1}}}
\newrobustcmd{\widetildeEz}[2][]{\ensuremath{\subp{\widetilde{E}}{}{#2}{}{#1}}}
\newrobustcmd{\acuteEz}[2][]{\ensuremath{\subp{\acute{E}}{}{#2}{}{#1}}}
\newrobustcmd{\graveEz}[2][]{\ensuremath{\subp{\grave{E}}{}{#2}{}{#1}}}
\newrobustcmd{\dotEz}[2][]{\ensuremath{\subp{\dot{E}}{}{#2}{}{#1}}}
\newrobustcmd{\ddotEz}[2][]{\ensuremath{\subp{\ddot{E}}{}{#2}{}{#1}}}
\newrobustcmd{\breveEz}[2][]{\ensuremath{\subp{\breve{E}}{}{#2}{}{#1}}}
\newrobustcmd{\barEz}[2][]{\ensuremath{\subp{\bar{E}}{}{#2}{}{#1}}}
\newrobustcmd{\vecEz}[2][]{\ensuremath{\subp{\vec{E}}{}{#2}{}{#1}}}
\newrobustcmd{\bmEz}[2][]{\ensuremath{\subp{\bm{E}}{}{#2}{}{#1}}}
\newrobustcmd{\hatbmEz}[2][]{\ensuremath{\subp{\hat{\bm{E}}}{}{#2}{}{#1}}}
\newrobustcmd{\widehatbmEz}[2][]{\ensuremath{\subp{\widehat{\bm{E}}}{}{#2}{}{#1}}}
\newrobustcmd{\checkbmEz}[2][]{\ensuremath{\subp{\check{\bm{E}}}{}{#2}{}{#1}}}
\newrobustcmd{\tildebmEz}[2][]{\ensuremath{\subp{\tilde{\bm{E}}}{}{#2}{}{#1}}}
\newrobustcmd{\widetildebmEz}[2][]{\ensuremath{\subp{\widetilde{\bm{E}}}{}{#2}{}{#1}}}
\newrobustcmd{\acutebmEz}[2][]{\ensuremath{\subp{\acute{\bm{E}}}{}{#2}{}{#1}}}
\newrobustcmd{\gravebmEz}[2][]{\ensuremath{\subp{\grave{\bm{E}}}{}{#2}{}{#1}}}
\newrobustcmd{\dotbmEz}[2][]{\ensuremath{\subp{\dot{\bm{E}}}{}{#2}{}{#1}}}
\newrobustcmd{\ddotbmEz}[2][]{\ensuremath{\subp{\ddot{\bm{E}}}{}{#2}{}{#1}}}
\newrobustcmd{\brevebmEz}[2][]{\ensuremath{\subp{\breve{\bm{E}}}{}{#2}{}{#1}}}
\newrobustcmd{\barbmEz}[2][]{\ensuremath{\subp{\bar{\bm{E}}}{}{#2}{}{#1}}}
\newrobustcmd{\vecbmEz}[2][]{\ensuremath{\subp{\vec{\bm{E}}}{}{#2}{}{#1}}}
\newrobustcmd{\Fz}[2][]{\ensuremath{\subp{F}{}{#2}{}{#1}}}
\newrobustcmd{\hatFz}[2][]{\ensuremath{\subp{\hat{F}}{}{#2}{}{#1}}}
\newrobustcmd{\widehatFz}[2][]{\ensuremath{\subp{\widehat{F}}{}{#2}{}{#1}}}
\newrobustcmd{\checkFz}[2][]{\ensuremath{\subp{\check{F}}{}{#2}{}{#1}}}
\newrobustcmd{\tildeFz}[2][]{\ensuremath{\subp{\tilde{F}}{}{#2}{}{#1}}}
\newrobustcmd{\widetildeFz}[2][]{\ensuremath{\subp{\widetilde{F}}{}{#2}{}{#1}}}
\newrobustcmd{\acuteFz}[2][]{\ensuremath{\subp{\acute{F}}{}{#2}{}{#1}}}
\newrobustcmd{\graveFz}[2][]{\ensuremath{\subp{\grave{F}}{}{#2}{}{#1}}}
\newrobustcmd{\dotFz}[2][]{\ensuremath{\subp{\dot{F}}{}{#2}{}{#1}}}
\newrobustcmd{\ddotFz}[2][]{\ensuremath{\subp{\ddot{F}}{}{#2}{}{#1}}}
\newrobustcmd{\breveFz}[2][]{\ensuremath{\subp{\breve{F}}{}{#2}{}{#1}}}
\newrobustcmd{\barFz}[2][]{\ensuremath{\subp{\bar{F}}{}{#2}{}{#1}}}
\newrobustcmd{\vecFz}[2][]{\ensuremath{\subp{\vec{F}}{}{#2}{}{#1}}}
\newrobustcmd{\bmFz}[2][]{\ensuremath{\subp{\bm{F}}{}{#2}{}{#1}}}
\newrobustcmd{\hatbmFz}[2][]{\ensuremath{\subp{\hat{\bm{F}}}{}{#2}{}{#1}}}
\newrobustcmd{\widehatbmFz}[2][]{\ensuremath{\subp{\widehat{\bm{F}}}{}{#2}{}{#1}}}
\newrobustcmd{\checkbmFz}[2][]{\ensuremath{\subp{\check{\bm{F}}}{}{#2}{}{#1}}}
\newrobustcmd{\tildebmFz}[2][]{\ensuremath{\subp{\tilde{\bm{F}}}{}{#2}{}{#1}}}
\newrobustcmd{\widetildebmFz}[2][]{\ensuremath{\subp{\widetilde{\bm{F}}}{}{#2}{}{#1}}}
\newrobustcmd{\acutebmFz}[2][]{\ensuremath{\subp{\acute{\bm{F}}}{}{#2}{}{#1}}}
\newrobustcmd{\gravebmFz}[2][]{\ensuremath{\subp{\grave{\bm{F}}}{}{#2}{}{#1}}}
\newrobustcmd{\dotbmFz}[2][]{\ensuremath{\subp{\dot{\bm{F}}}{}{#2}{}{#1}}}
\newrobustcmd{\ddotbmFz}[2][]{\ensuremath{\subp{\ddot{\bm{F}}}{}{#2}{}{#1}}}
\newrobustcmd{\brevebmFz}[2][]{\ensuremath{\subp{\breve{\bm{F}}}{}{#2}{}{#1}}}
\newrobustcmd{\barbmFz}[2][]{\ensuremath{\subp{\bar{\bm{F}}}{}{#2}{}{#1}}}
\newrobustcmd{\vecbmFz}[2][]{\ensuremath{\subp{\vec{\bm{F}}}{}{#2}{}{#1}}}
\newrobustcmd{\Gz}[2][]{\ensuremath{\subp{G}{}{#2}{}{#1}}}
\newrobustcmd{\hatGz}[2][]{\ensuremath{\subp{\hat{G}}{}{#2}{}{#1}}}
\newrobustcmd{\widehatGz}[2][]{\ensuremath{\subp{\widehat{G}}{}{#2}{}{#1}}}
\newrobustcmd{\checkGz}[2][]{\ensuremath{\subp{\check{G}}{}{#2}{}{#1}}}
\newrobustcmd{\tildeGz}[2][]{\ensuremath{\subp{\tilde{G}}{}{#2}{}{#1}}}
\newrobustcmd{\widetildeGz}[2][]{\ensuremath{\subp{\widetilde{G}}{}{#2}{}{#1}}}
\newrobustcmd{\acuteGz}[2][]{\ensuremath{\subp{\acute{G}}{}{#2}{}{#1}}}
\newrobustcmd{\graveGz}[2][]{\ensuremath{\subp{\grave{G}}{}{#2}{}{#1}}}
\newrobustcmd{\dotGz}[2][]{\ensuremath{\subp{\dot{G}}{}{#2}{}{#1}}}
\newrobustcmd{\ddotGz}[2][]{\ensuremath{\subp{\ddot{G}}{}{#2}{}{#1}}}
\newrobustcmd{\breveGz}[2][]{\ensuremath{\subp{\breve{G}}{}{#2}{}{#1}}}
\newrobustcmd{\barGz}[2][]{\ensuremath{\subp{\bar{G}}{}{#2}{}{#1}}}
\newrobustcmd{\vecGz}[2][]{\ensuremath{\subp{\vec{G}}{}{#2}{}{#1}}}
\newrobustcmd{\bmGz}[2][]{\ensuremath{\subp{\bm{G}}{}{#2}{}{#1}}}
\newrobustcmd{\hatbmGz}[2][]{\ensuremath{\subp{\hat{\bm{G}}}{}{#2}{}{#1}}}
\newrobustcmd{\widehatbmGz}[2][]{\ensuremath{\subp{\widehat{\bm{G}}}{}{#2}{}{#1}}}
\newrobustcmd{\checkbmGz}[2][]{\ensuremath{\subp{\check{\bm{G}}}{}{#2}{}{#1}}}
\newrobustcmd{\tildebmGz}[2][]{\ensuremath{\subp{\tilde{\bm{G}}}{}{#2}{}{#1}}}
\newrobustcmd{\widetildebmGz}[2][]{\ensuremath{\subp{\widetilde{\bm{G}}}{}{#2}{}{#1}}}
\newrobustcmd{\acutebmGz}[2][]{\ensuremath{\subp{\acute{\bm{G}}}{}{#2}{}{#1}}}
\newrobustcmd{\gravebmGz}[2][]{\ensuremath{\subp{\grave{\bm{G}}}{}{#2}{}{#1}}}
\newrobustcmd{\dotbmGz}[2][]{\ensuremath{\subp{\dot{\bm{G}}}{}{#2}{}{#1}}}
\newrobustcmd{\ddotbmGz}[2][]{\ensuremath{\subp{\ddot{\bm{G}}}{}{#2}{}{#1}}}
\newrobustcmd{\brevebmGz}[2][]{\ensuremath{\subp{\breve{\bm{G}}}{}{#2}{}{#1}}}
\newrobustcmd{\barbmGz}[2][]{\ensuremath{\subp{\bar{\bm{G}}}{}{#2}{}{#1}}}
\newrobustcmd{\vecbmGz}[2][]{\ensuremath{\subp{\vec{\bm{G}}}{}{#2}{}{#1}}}
\newrobustcmd{\Hz}[2][]{\ensuremath{\subp{H}{}{#2}{}{#1}}}
\newrobustcmd{\hatHz}[2][]{\ensuremath{\subp{\hat{H}}{}{#2}{}{#1}}}
\newrobustcmd{\widehatHz}[2][]{\ensuremath{\subp{\widehat{H}}{}{#2}{}{#1}}}
\newrobustcmd{\checkHz}[2][]{\ensuremath{\subp{\check{H}}{}{#2}{}{#1}}}
\newrobustcmd{\tildeHz}[2][]{\ensuremath{\subp{\tilde{H}}{}{#2}{}{#1}}}
\newrobustcmd{\widetildeHz}[2][]{\ensuremath{\subp{\widetilde{H}}{}{#2}{}{#1}}}
\newrobustcmd{\acuteHz}[2][]{\ensuremath{\subp{\acute{H}}{}{#2}{}{#1}}}
\newrobustcmd{\graveHz}[2][]{\ensuremath{\subp{\grave{H}}{}{#2}{}{#1}}}
\newrobustcmd{\dotHz}[2][]{\ensuremath{\subp{\dot{H}}{}{#2}{}{#1}}}
\newrobustcmd{\ddotHz}[2][]{\ensuremath{\subp{\ddot{H}}{}{#2}{}{#1}}}
\newrobustcmd{\breveHz}[2][]{\ensuremath{\subp{\breve{H}}{}{#2}{}{#1}}}
\newrobustcmd{\barHz}[2][]{\ensuremath{\subp{\bar{H}}{}{#2}{}{#1}}}
\newrobustcmd{\vecHz}[2][]{\ensuremath{\subp{\vec{H}}{}{#2}{}{#1}}}
\newrobustcmd{\bmHz}[2][]{\ensuremath{\subp{\bm{H}}{}{#2}{}{#1}}}
\newrobustcmd{\hatbmHz}[2][]{\ensuremath{\subp{\hat{\bm{H}}}{}{#2}{}{#1}}}
\newrobustcmd{\widehatbmHz}[2][]{\ensuremath{\subp{\widehat{\bm{H}}}{}{#2}{}{#1}}}
\newrobustcmd{\checkbmHz}[2][]{\ensuremath{\subp{\check{\bm{H}}}{}{#2}{}{#1}}}
\newrobustcmd{\tildebmHz}[2][]{\ensuremath{\subp{\tilde{\bm{H}}}{}{#2}{}{#1}}}
\newrobustcmd{\widetildebmHz}[2][]{\ensuremath{\subp{\widetilde{\bm{H}}}{}{#2}{}{#1}}}
\newrobustcmd{\acutebmHz}[2][]{\ensuremath{\subp{\acute{\bm{H}}}{}{#2}{}{#1}}}
\newrobustcmd{\gravebmHz}[2][]{\ensuremath{\subp{\grave{\bm{H}}}{}{#2}{}{#1}}}
\newrobustcmd{\dotbmHz}[2][]{\ensuremath{\subp{\dot{\bm{H}}}{}{#2}{}{#1}}}
\newrobustcmd{\ddotbmHz}[2][]{\ensuremath{\subp{\ddot{\bm{H}}}{}{#2}{}{#1}}}
\newrobustcmd{\brevebmHz}[2][]{\ensuremath{\subp{\breve{\bm{H}}}{}{#2}{}{#1}}}
\newrobustcmd{\barbmHz}[2][]{\ensuremath{\subp{\bar{\bm{H}}}{}{#2}{}{#1}}}
\newrobustcmd{\vecbmHz}[2][]{\ensuremath{\subp{\vec{\bm{H}}}{}{#2}{}{#1}}}
\newrobustcmd{\Iz}[2][]{\ensuremath{\subp{I}{}{#2}{}{#1}}}
\newrobustcmd{\hatIz}[2][]{\ensuremath{\subp{\hat{I}}{}{#2}{}{#1}}}
\newrobustcmd{\widehatIz}[2][]{\ensuremath{\subp{\widehat{I}}{}{#2}{}{#1}}}
\newrobustcmd{\checkIz}[2][]{\ensuremath{\subp{\check{I}}{}{#2}{}{#1}}}
\newrobustcmd{\tildeIz}[2][]{\ensuremath{\subp{\tilde{I}}{}{#2}{}{#1}}}
\newrobustcmd{\widetildeIz}[2][]{\ensuremath{\subp{\widetilde{I}}{}{#2}{}{#1}}}
\newrobustcmd{\acuteIz}[2][]{\ensuremath{\subp{\acute{I}}{}{#2}{}{#1}}}
\newrobustcmd{\graveIz}[2][]{\ensuremath{\subp{\grave{I}}{}{#2}{}{#1}}}
\newrobustcmd{\dotIz}[2][]{\ensuremath{\subp{\dot{I}}{}{#2}{}{#1}}}
\newrobustcmd{\ddotIz}[2][]{\ensuremath{\subp{\ddot{I}}{}{#2}{}{#1}}}
\newrobustcmd{\breveIz}[2][]{\ensuremath{\subp{\breve{I}}{}{#2}{}{#1}}}
\newrobustcmd{\barIz}[2][]{\ensuremath{\subp{\bar{I}}{}{#2}{}{#1}}}
\newrobustcmd{\vecIz}[2][]{\ensuremath{\subp{\vec{I}}{}{#2}{}{#1}}}
\newrobustcmd{\bmIz}[2][]{\ensuremath{\subp{\bm{I}}{}{#2}{}{#1}}}
\newrobustcmd{\hatbmIz}[2][]{\ensuremath{\subp{\hat{\bm{I}}}{}{#2}{}{#1}}}
\newrobustcmd{\widehatbmIz}[2][]{\ensuremath{\subp{\widehat{\bm{I}}}{}{#2}{}{#1}}}
\newrobustcmd{\checkbmIz}[2][]{\ensuremath{\subp{\check{\bm{I}}}{}{#2}{}{#1}}}
\newrobustcmd{\tildebmIz}[2][]{\ensuremath{\subp{\tilde{\bm{I}}}{}{#2}{}{#1}}}
\newrobustcmd{\widetildebmIz}[2][]{\ensuremath{\subp{\widetilde{\bm{I}}}{}{#2}{}{#1}}}
\newrobustcmd{\acutebmIz}[2][]{\ensuremath{\subp{\acute{\bm{I}}}{}{#2}{}{#1}}}
\newrobustcmd{\gravebmIz}[2][]{\ensuremath{\subp{\grave{\bm{I}}}{}{#2}{}{#1}}}
\newrobustcmd{\dotbmIz}[2][]{\ensuremath{\subp{\dot{\bm{I}}}{}{#2}{}{#1}}}
\newrobustcmd{\ddotbmIz}[2][]{\ensuremath{\subp{\ddot{\bm{I}}}{}{#2}{}{#1}}}
\newrobustcmd{\brevebmIz}[2][]{\ensuremath{\subp{\breve{\bm{I}}}{}{#2}{}{#1}}}
\newrobustcmd{\barbmIz}[2][]{\ensuremath{\subp{\bar{\bm{I}}}{}{#2}{}{#1}}}
\newrobustcmd{\vecbmIz}[2][]{\ensuremath{\subp{\vec{\bm{I}}}{}{#2}{}{#1}}}
\newrobustcmd{\Jz}[2][]{\ensuremath{\subp{J}{}{#2}{}{#1}}}
\newrobustcmd{\hatJz}[2][]{\ensuremath{\subp{\hat{J}}{}{#2}{}{#1}}}
\newrobustcmd{\widehatJz}[2][]{\ensuremath{\subp{\widehat{J}}{}{#2}{}{#1}}}
\newrobustcmd{\checkJz}[2][]{\ensuremath{\subp{\check{J}}{}{#2}{}{#1}}}
\newrobustcmd{\tildeJz}[2][]{\ensuremath{\subp{\tilde{J}}{}{#2}{}{#1}}}
\newrobustcmd{\widetildeJz}[2][]{\ensuremath{\subp{\widetilde{J}}{}{#2}{}{#1}}}
\newrobustcmd{\acuteJz}[2][]{\ensuremath{\subp{\acute{J}}{}{#2}{}{#1}}}
\newrobustcmd{\graveJz}[2][]{\ensuremath{\subp{\grave{J}}{}{#2}{}{#1}}}
\newrobustcmd{\dotJz}[2][]{\ensuremath{\subp{\dot{J}}{}{#2}{}{#1}}}
\newrobustcmd{\ddotJz}[2][]{\ensuremath{\subp{\ddot{J}}{}{#2}{}{#1}}}
\newrobustcmd{\breveJz}[2][]{\ensuremath{\subp{\breve{J}}{}{#2}{}{#1}}}
\newrobustcmd{\barJz}[2][]{\ensuremath{\subp{\bar{J}}{}{#2}{}{#1}}}
\newrobustcmd{\vecJz}[2][]{\ensuremath{\subp{\vec{J}}{}{#2}{}{#1}}}
\newrobustcmd{\bmJz}[2][]{\ensuremath{\subp{\bm{J}}{}{#2}{}{#1}}}
\newrobustcmd{\hatbmJz}[2][]{\ensuremath{\subp{\hat{\bm{J}}}{}{#2}{}{#1}}}
\newrobustcmd{\widehatbmJz}[2][]{\ensuremath{\subp{\widehat{\bm{J}}}{}{#2}{}{#1}}}
\newrobustcmd{\checkbmJz}[2][]{\ensuremath{\subp{\check{\bm{J}}}{}{#2}{}{#1}}}
\newrobustcmd{\tildebmJz}[2][]{\ensuremath{\subp{\tilde{\bm{J}}}{}{#2}{}{#1}}}
\newrobustcmd{\widetildebmJz}[2][]{\ensuremath{\subp{\widetilde{\bm{J}}}{}{#2}{}{#1}}}
\newrobustcmd{\acutebmJz}[2][]{\ensuremath{\subp{\acute{\bm{J}}}{}{#2}{}{#1}}}
\newrobustcmd{\gravebmJz}[2][]{\ensuremath{\subp{\grave{\bm{J}}}{}{#2}{}{#1}}}
\newrobustcmd{\dotbmJz}[2][]{\ensuremath{\subp{\dot{\bm{J}}}{}{#2}{}{#1}}}
\newrobustcmd{\ddotbmJz}[2][]{\ensuremath{\subp{\ddot{\bm{J}}}{}{#2}{}{#1}}}
\newrobustcmd{\brevebmJz}[2][]{\ensuremath{\subp{\breve{\bm{J}}}{}{#2}{}{#1}}}
\newrobustcmd{\barbmJz}[2][]{\ensuremath{\subp{\bar{\bm{J}}}{}{#2}{}{#1}}}
\newrobustcmd{\vecbmJz}[2][]{\ensuremath{\subp{\vec{\bm{J}}}{}{#2}{}{#1}}}
\newrobustcmd{\Kz}[2][]{\ensuremath{\subp{K}{}{#2}{}{#1}}}
\newrobustcmd{\hatKz}[2][]{\ensuremath{\subp{\hat{K}}{}{#2}{}{#1}}}
\newrobustcmd{\widehatKz}[2][]{\ensuremath{\subp{\widehat{K}}{}{#2}{}{#1}}}
\newrobustcmd{\checkKz}[2][]{\ensuremath{\subp{\check{K}}{}{#2}{}{#1}}}
\newrobustcmd{\tildeKz}[2][]{\ensuremath{\subp{\tilde{K}}{}{#2}{}{#1}}}
\newrobustcmd{\widetildeKz}[2][]{\ensuremath{\subp{\widetilde{K}}{}{#2}{}{#1}}}
\newrobustcmd{\acuteKz}[2][]{\ensuremath{\subp{\acute{K}}{}{#2}{}{#1}}}
\newrobustcmd{\graveKz}[2][]{\ensuremath{\subp{\grave{K}}{}{#2}{}{#1}}}
\newrobustcmd{\dotKz}[2][]{\ensuremath{\subp{\dot{K}}{}{#2}{}{#1}}}
\newrobustcmd{\ddotKz}[2][]{\ensuremath{\subp{\ddot{K}}{}{#2}{}{#1}}}
\newrobustcmd{\breveKz}[2][]{\ensuremath{\subp{\breve{K}}{}{#2}{}{#1}}}
\newrobustcmd{\barKz}[2][]{\ensuremath{\subp{\bar{K}}{}{#2}{}{#1}}}
\newrobustcmd{\vecKz}[2][]{\ensuremath{\subp{\vec{K}}{}{#2}{}{#1}}}
\newrobustcmd{\bmKz}[2][]{\ensuremath{\subp{\bm{K}}{}{#2}{}{#1}}}
\newrobustcmd{\hatbmKz}[2][]{\ensuremath{\subp{\hat{\bm{K}}}{}{#2}{}{#1}}}
\newrobustcmd{\widehatbmKz}[2][]{\ensuremath{\subp{\widehat{\bm{K}}}{}{#2}{}{#1}}}
\newrobustcmd{\checkbmKz}[2][]{\ensuremath{\subp{\check{\bm{K}}}{}{#2}{}{#1}}}
\newrobustcmd{\tildebmKz}[2][]{\ensuremath{\subp{\tilde{\bm{K}}}{}{#2}{}{#1}}}
\newrobustcmd{\widetildebmKz}[2][]{\ensuremath{\subp{\widetilde{\bm{K}}}{}{#2}{}{#1}}}
\newrobustcmd{\acutebmKz}[2][]{\ensuremath{\subp{\acute{\bm{K}}}{}{#2}{}{#1}}}
\newrobustcmd{\gravebmKz}[2][]{\ensuremath{\subp{\grave{\bm{K}}}{}{#2}{}{#1}}}
\newrobustcmd{\dotbmKz}[2][]{\ensuremath{\subp{\dot{\bm{K}}}{}{#2}{}{#1}}}
\newrobustcmd{\ddotbmKz}[2][]{\ensuremath{\subp{\ddot{\bm{K}}}{}{#2}{}{#1}}}
\newrobustcmd{\brevebmKz}[2][]{\ensuremath{\subp{\breve{\bm{K}}}{}{#2}{}{#1}}}
\newrobustcmd{\barbmKz}[2][]{\ensuremath{\subp{\bar{\bm{K}}}{}{#2}{}{#1}}}
\newrobustcmd{\vecbmKz}[2][]{\ensuremath{\subp{\vec{\bm{K}}}{}{#2}{}{#1}}}
\newrobustcmd{\Lz}[2][]{\ensuremath{\subp{L}{}{#2}{}{#1}}}
\newrobustcmd{\hatLz}[2][]{\ensuremath{\subp{\hat{L}}{}{#2}{}{#1}}}
\newrobustcmd{\widehatLz}[2][]{\ensuremath{\subp{\widehat{L}}{}{#2}{}{#1}}}
\newrobustcmd{\checkLz}[2][]{\ensuremath{\subp{\check{L}}{}{#2}{}{#1}}}
\newrobustcmd{\tildeLz}[2][]{\ensuremath{\subp{\tilde{L}}{}{#2}{}{#1}}}
\newrobustcmd{\widetildeLz}[2][]{\ensuremath{\subp{\widetilde{L}}{}{#2}{}{#1}}}
\newrobustcmd{\acuteLz}[2][]{\ensuremath{\subp{\acute{L}}{}{#2}{}{#1}}}
\newrobustcmd{\graveLz}[2][]{\ensuremath{\subp{\grave{L}}{}{#2}{}{#1}}}
\newrobustcmd{\dotLz}[2][]{\ensuremath{\subp{\dot{L}}{}{#2}{}{#1}}}
\newrobustcmd{\ddotLz}[2][]{\ensuremath{\subp{\ddot{L}}{}{#2}{}{#1}}}
\newrobustcmd{\breveLz}[2][]{\ensuremath{\subp{\breve{L}}{}{#2}{}{#1}}}
\newrobustcmd{\barLz}[2][]{\ensuremath{\subp{\bar{L}}{}{#2}{}{#1}}}
\newrobustcmd{\vecLz}[2][]{\ensuremath{\subp{\vec{L}}{}{#2}{}{#1}}}
\newrobustcmd{\bmLz}[2][]{\ensuremath{\subp{\bm{L}}{}{#2}{}{#1}}}
\newrobustcmd{\hatbmLz}[2][]{\ensuremath{\subp{\hat{\bm{L}}}{}{#2}{}{#1}}}
\newrobustcmd{\widehatbmLz}[2][]{\ensuremath{\subp{\widehat{\bm{L}}}{}{#2}{}{#1}}}
\newrobustcmd{\checkbmLz}[2][]{\ensuremath{\subp{\check{\bm{L}}}{}{#2}{}{#1}}}
\newrobustcmd{\tildebmLz}[2][]{\ensuremath{\subp{\tilde{\bm{L}}}{}{#2}{}{#1}}}
\newrobustcmd{\widetildebmLz}[2][]{\ensuremath{\subp{\widetilde{\bm{L}}}{}{#2}{}{#1}}}
\newrobustcmd{\acutebmLz}[2][]{\ensuremath{\subp{\acute{\bm{L}}}{}{#2}{}{#1}}}
\newrobustcmd{\gravebmLz}[2][]{\ensuremath{\subp{\grave{\bm{L}}}{}{#2}{}{#1}}}
\newrobustcmd{\dotbmLz}[2][]{\ensuremath{\subp{\dot{\bm{L}}}{}{#2}{}{#1}}}
\newrobustcmd{\ddotbmLz}[2][]{\ensuremath{\subp{\ddot{\bm{L}}}{}{#2}{}{#1}}}
\newrobustcmd{\brevebmLz}[2][]{\ensuremath{\subp{\breve{\bm{L}}}{}{#2}{}{#1}}}
\newrobustcmd{\barbmLz}[2][]{\ensuremath{\subp{\bar{\bm{L}}}{}{#2}{}{#1}}}
\newrobustcmd{\vecbmLz}[2][]{\ensuremath{\subp{\vec{\bm{L}}}{}{#2}{}{#1}}}
\newrobustcmd{\Mz}[2][]{\ensuremath{\subp{M}{}{#2}{}{#1}}}
\newrobustcmd{\hatMz}[2][]{\ensuremath{\subp{\hat{M}}{}{#2}{}{#1}}}
\newrobustcmd{\widehatMz}[2][]{\ensuremath{\subp{\widehat{M}}{}{#2}{}{#1}}}
\newrobustcmd{\checkMz}[2][]{\ensuremath{\subp{\check{M}}{}{#2}{}{#1}}}
\newrobustcmd{\tildeMz}[2][]{\ensuremath{\subp{\tilde{M}}{}{#2}{}{#1}}}
\newrobustcmd{\widetildeMz}[2][]{\ensuremath{\subp{\widetilde{M}}{}{#2}{}{#1}}}
\newrobustcmd{\acuteMz}[2][]{\ensuremath{\subp{\acute{M}}{}{#2}{}{#1}}}
\newrobustcmd{\graveMz}[2][]{\ensuremath{\subp{\grave{M}}{}{#2}{}{#1}}}
\newrobustcmd{\dotMz}[2][]{\ensuremath{\subp{\dot{M}}{}{#2}{}{#1}}}
\newrobustcmd{\ddotMz}[2][]{\ensuremath{\subp{\ddot{M}}{}{#2}{}{#1}}}
\newrobustcmd{\breveMz}[2][]{\ensuremath{\subp{\breve{M}}{}{#2}{}{#1}}}
\newrobustcmd{\barMz}[2][]{\ensuremath{\subp{\bar{M}}{}{#2}{}{#1}}}
\newrobustcmd{\vecMz}[2][]{\ensuremath{\subp{\vec{M}}{}{#2}{}{#1}}}
\newrobustcmd{\bmMz}[2][]{\ensuremath{\subp{\bm{M}}{}{#2}{}{#1}}}
\newrobustcmd{\hatbmMz}[2][]{\ensuremath{\subp{\hat{\bm{M}}}{}{#2}{}{#1}}}
\newrobustcmd{\widehatbmMz}[2][]{\ensuremath{\subp{\widehat{\bm{M}}}{}{#2}{}{#1}}}
\newrobustcmd{\checkbmMz}[2][]{\ensuremath{\subp{\check{\bm{M}}}{}{#2}{}{#1}}}
\newrobustcmd{\tildebmMz}[2][]{\ensuremath{\subp{\tilde{\bm{M}}}{}{#2}{}{#1}}}
\newrobustcmd{\widetildebmMz}[2][]{\ensuremath{\subp{\widetilde{\bm{M}}}{}{#2}{}{#1}}}
\newrobustcmd{\acutebmMz}[2][]{\ensuremath{\subp{\acute{\bm{M}}}{}{#2}{}{#1}}}
\newrobustcmd{\gravebmMz}[2][]{\ensuremath{\subp{\grave{\bm{M}}}{}{#2}{}{#1}}}
\newrobustcmd{\dotbmMz}[2][]{\ensuremath{\subp{\dot{\bm{M}}}{}{#2}{}{#1}}}
\newrobustcmd{\ddotbmMz}[2][]{\ensuremath{\subp{\ddot{\bm{M}}}{}{#2}{}{#1}}}
\newrobustcmd{\brevebmMz}[2][]{\ensuremath{\subp{\breve{\bm{M}}}{}{#2}{}{#1}}}
\newrobustcmd{\barbmMz}[2][]{\ensuremath{\subp{\bar{\bm{M}}}{}{#2}{}{#1}}}
\newrobustcmd{\vecbmMz}[2][]{\ensuremath{\subp{\vec{\bm{M}}}{}{#2}{}{#1}}}
\newrobustcmd{\Nz}[2][]{\ensuremath{\subp{N}{}{#2}{}{#1}}}
\newrobustcmd{\hatNz}[2][]{\ensuremath{\subp{\hat{N}}{}{#2}{}{#1}}}
\newrobustcmd{\widehatNz}[2][]{\ensuremath{\subp{\widehat{N}}{}{#2}{}{#1}}}
\newrobustcmd{\checkNz}[2][]{\ensuremath{\subp{\check{N}}{}{#2}{}{#1}}}
\newrobustcmd{\tildeNz}[2][]{\ensuremath{\subp{\tilde{N}}{}{#2}{}{#1}}}
\newrobustcmd{\widetildeNz}[2][]{\ensuremath{\subp{\widetilde{N}}{}{#2}{}{#1}}}
\newrobustcmd{\acuteNz}[2][]{\ensuremath{\subp{\acute{N}}{}{#2}{}{#1}}}
\newrobustcmd{\graveNz}[2][]{\ensuremath{\subp{\grave{N}}{}{#2}{}{#1}}}
\newrobustcmd{\dotNz}[2][]{\ensuremath{\subp{\dot{N}}{}{#2}{}{#1}}}
\newrobustcmd{\ddotNz}[2][]{\ensuremath{\subp{\ddot{N}}{}{#2}{}{#1}}}
\newrobustcmd{\breveNz}[2][]{\ensuremath{\subp{\breve{N}}{}{#2}{}{#1}}}
\newrobustcmd{\barNz}[2][]{\ensuremath{\subp{\bar{N}}{}{#2}{}{#1}}}
\newrobustcmd{\vecNz}[2][]{\ensuremath{\subp{\vec{N}}{}{#2}{}{#1}}}
\newrobustcmd{\bmNz}[2][]{\ensuremath{\subp{\bm{N}}{}{#2}{}{#1}}}
\newrobustcmd{\hatbmNz}[2][]{\ensuremath{\subp{\hat{\bm{N}}}{}{#2}{}{#1}}}
\newrobustcmd{\widehatbmNz}[2][]{\ensuremath{\subp{\widehat{\bm{N}}}{}{#2}{}{#1}}}
\newrobustcmd{\checkbmNz}[2][]{\ensuremath{\subp{\check{\bm{N}}}{}{#2}{}{#1}}}
\newrobustcmd{\tildebmNz}[2][]{\ensuremath{\subp{\tilde{\bm{N}}}{}{#2}{}{#1}}}
\newrobustcmd{\widetildebmNz}[2][]{\ensuremath{\subp{\widetilde{\bm{N}}}{}{#2}{}{#1}}}
\newrobustcmd{\acutebmNz}[2][]{\ensuremath{\subp{\acute{\bm{N}}}{}{#2}{}{#1}}}
\newrobustcmd{\gravebmNz}[2][]{\ensuremath{\subp{\grave{\bm{N}}}{}{#2}{}{#1}}}
\newrobustcmd{\dotbmNz}[2][]{\ensuremath{\subp{\dot{\bm{N}}}{}{#2}{}{#1}}}
\newrobustcmd{\ddotbmNz}[2][]{\ensuremath{\subp{\ddot{\bm{N}}}{}{#2}{}{#1}}}
\newrobustcmd{\brevebmNz}[2][]{\ensuremath{\subp{\breve{\bm{N}}}{}{#2}{}{#1}}}
\newrobustcmd{\barbmNz}[2][]{\ensuremath{\subp{\bar{\bm{N}}}{}{#2}{}{#1}}}
\newrobustcmd{\vecbmNz}[2][]{\ensuremath{\subp{\vec{\bm{N}}}{}{#2}{}{#1}}}
\newrobustcmd{\Oz}[2][]{\ensuremath{\subp{O}{}{#2}{}{#1}}}
\newrobustcmd{\hatOz}[2][]{\ensuremath{\subp{\hat{O}}{}{#2}{}{#1}}}
\newrobustcmd{\widehatOz}[2][]{\ensuremath{\subp{\widehat{O}}{}{#2}{}{#1}}}
\newrobustcmd{\checkOz}[2][]{\ensuremath{\subp{\check{O}}{}{#2}{}{#1}}}
\newrobustcmd{\tildeOz}[2][]{\ensuremath{\subp{\tilde{O}}{}{#2}{}{#1}}}
\newrobustcmd{\widetildeOz}[2][]{\ensuremath{\subp{\widetilde{O}}{}{#2}{}{#1}}}
\newrobustcmd{\acuteOz}[2][]{\ensuremath{\subp{\acute{O}}{}{#2}{}{#1}}}
\newrobustcmd{\graveOz}[2][]{\ensuremath{\subp{\grave{O}}{}{#2}{}{#1}}}
\newrobustcmd{\dotOz}[2][]{\ensuremath{\subp{\dot{O}}{}{#2}{}{#1}}}
\newrobustcmd{\ddotOz}[2][]{\ensuremath{\subp{\ddot{O}}{}{#2}{}{#1}}}
\newrobustcmd{\breveOz}[2][]{\ensuremath{\subp{\breve{O}}{}{#2}{}{#1}}}
\newrobustcmd{\barOz}[2][]{\ensuremath{\subp{\bar{O}}{}{#2}{}{#1}}}
\newrobustcmd{\vecOz}[2][]{\ensuremath{\subp{\vec{O}}{}{#2}{}{#1}}}
\newrobustcmd{\bmOz}[2][]{\ensuremath{\subp{\bm{O}}{}{#2}{}{#1}}}
\newrobustcmd{\hatbmOz}[2][]{\ensuremath{\subp{\hat{\bm{O}}}{}{#2}{}{#1}}}
\newrobustcmd{\widehatbmOz}[2][]{\ensuremath{\subp{\widehat{\bm{O}}}{}{#2}{}{#1}}}
\newrobustcmd{\checkbmOz}[2][]{\ensuremath{\subp{\check{\bm{O}}}{}{#2}{}{#1}}}
\newrobustcmd{\tildebmOz}[2][]{\ensuremath{\subp{\tilde{\bm{O}}}{}{#2}{}{#1}}}
\newrobustcmd{\widetildebmOz}[2][]{\ensuremath{\subp{\widetilde{\bm{O}}}{}{#2}{}{#1}}}
\newrobustcmd{\acutebmOz}[2][]{\ensuremath{\subp{\acute{\bm{O}}}{}{#2}{}{#1}}}
\newrobustcmd{\gravebmOz}[2][]{\ensuremath{\subp{\grave{\bm{O}}}{}{#2}{}{#1}}}
\newrobustcmd{\dotbmOz}[2][]{\ensuremath{\subp{\dot{\bm{O}}}{}{#2}{}{#1}}}
\newrobustcmd{\ddotbmOz}[2][]{\ensuremath{\subp{\ddot{\bm{O}}}{}{#2}{}{#1}}}
\newrobustcmd{\brevebmOz}[2][]{\ensuremath{\subp{\breve{\bm{O}}}{}{#2}{}{#1}}}
\newrobustcmd{\barbmOz}[2][]{\ensuremath{\subp{\bar{\bm{O}}}{}{#2}{}{#1}}}
\newrobustcmd{\vecbmOz}[2][]{\ensuremath{\subp{\vec{\bm{O}}}{}{#2}{}{#1}}}
\newrobustcmd{\Pz}[2][]{\ensuremath{\subp{P}{}{#2}{}{#1}}}
\newrobustcmd{\hatPz}[2][]{\ensuremath{\subp{\hat{P}}{}{#2}{}{#1}}}
\newrobustcmd{\widehatPz}[2][]{\ensuremath{\subp{\widehat{P}}{}{#2}{}{#1}}}
\newrobustcmd{\checkPz}[2][]{\ensuremath{\subp{\check{P}}{}{#2}{}{#1}}}
\newrobustcmd{\tildePz}[2][]{\ensuremath{\subp{\tilde{P}}{}{#2}{}{#1}}}
\newrobustcmd{\widetildePz}[2][]{\ensuremath{\subp{\widetilde{P}}{}{#2}{}{#1}}}
\newrobustcmd{\acutePz}[2][]{\ensuremath{\subp{\acute{P}}{}{#2}{}{#1}}}
\newrobustcmd{\gravePz}[2][]{\ensuremath{\subp{\grave{P}}{}{#2}{}{#1}}}
\newrobustcmd{\dotPz}[2][]{\ensuremath{\subp{\dot{P}}{}{#2}{}{#1}}}
\newrobustcmd{\ddotPz}[2][]{\ensuremath{\subp{\ddot{P}}{}{#2}{}{#1}}}
\newrobustcmd{\brevePz}[2][]{\ensuremath{\subp{\breve{P}}{}{#2}{}{#1}}}
\newrobustcmd{\barPz}[2][]{\ensuremath{\subp{\bar{P}}{}{#2}{}{#1}}}
\newrobustcmd{\vecPz}[2][]{\ensuremath{\subp{\vec{P}}{}{#2}{}{#1}}}
\newrobustcmd{\bmPz}[2][]{\ensuremath{\subp{\bm{P}}{}{#2}{}{#1}}}
\newrobustcmd{\hatbmPz}[2][]{\ensuremath{\subp{\hat{\bm{P}}}{}{#2}{}{#1}}}
\newrobustcmd{\widehatbmPz}[2][]{\ensuremath{\subp{\widehat{\bm{P}}}{}{#2}{}{#1}}}
\newrobustcmd{\checkbmPz}[2][]{\ensuremath{\subp{\check{\bm{P}}}{}{#2}{}{#1}}}
\newrobustcmd{\tildebmPz}[2][]{\ensuremath{\subp{\tilde{\bm{P}}}{}{#2}{}{#1}}}
\newrobustcmd{\widetildebmPz}[2][]{\ensuremath{\subp{\widetilde{\bm{P}}}{}{#2}{}{#1}}}
\newrobustcmd{\acutebmPz}[2][]{\ensuremath{\subp{\acute{\bm{P}}}{}{#2}{}{#1}}}
\newrobustcmd{\gravebmPz}[2][]{\ensuremath{\subp{\grave{\bm{P}}}{}{#2}{}{#1}}}
\newrobustcmd{\dotbmPz}[2][]{\ensuremath{\subp{\dot{\bm{P}}}{}{#2}{}{#1}}}
\newrobustcmd{\ddotbmPz}[2][]{\ensuremath{\subp{\ddot{\bm{P}}}{}{#2}{}{#1}}}
\newrobustcmd{\brevebmPz}[2][]{\ensuremath{\subp{\breve{\bm{P}}}{}{#2}{}{#1}}}
\newrobustcmd{\barbmPz}[2][]{\ensuremath{\subp{\bar{\bm{P}}}{}{#2}{}{#1}}}
\newrobustcmd{\vecbmPz}[2][]{\ensuremath{\subp{\vec{\bm{P}}}{}{#2}{}{#1}}}
\newrobustcmd{\Qz}[2][]{\ensuremath{\subp{Q}{}{#2}{}{#1}}}
\newrobustcmd{\hatQz}[2][]{\ensuremath{\subp{\hat{Q}}{}{#2}{}{#1}}}
\newrobustcmd{\widehatQz}[2][]{\ensuremath{\subp{\widehat{Q}}{}{#2}{}{#1}}}
\newrobustcmd{\checkQz}[2][]{\ensuremath{\subp{\check{Q}}{}{#2}{}{#1}}}
\newrobustcmd{\tildeQz}[2][]{\ensuremath{\subp{\tilde{Q}}{}{#2}{}{#1}}}
\newrobustcmd{\widetildeQz}[2][]{\ensuremath{\subp{\widetilde{Q}}{}{#2}{}{#1}}}
\newrobustcmd{\acuteQz}[2][]{\ensuremath{\subp{\acute{Q}}{}{#2}{}{#1}}}
\newrobustcmd{\graveQz}[2][]{\ensuremath{\subp{\grave{Q}}{}{#2}{}{#1}}}
\newrobustcmd{\dotQz}[2][]{\ensuremath{\subp{\dot{Q}}{}{#2}{}{#1}}}
\newrobustcmd{\ddotQz}[2][]{\ensuremath{\subp{\ddot{Q}}{}{#2}{}{#1}}}
\newrobustcmd{\breveQz}[2][]{\ensuremath{\subp{\breve{Q}}{}{#2}{}{#1}}}
\newrobustcmd{\barQz}[2][]{\ensuremath{\subp{\bar{Q}}{}{#2}{}{#1}}}
\newrobustcmd{\vecQz}[2][]{\ensuremath{\subp{\vec{Q}}{}{#2}{}{#1}}}
\newrobustcmd{\bmQz}[2][]{\ensuremath{\subp{\bm{Q}}{}{#2}{}{#1}}}
\newrobustcmd{\hatbmQz}[2][]{\ensuremath{\subp{\hat{\bm{Q}}}{}{#2}{}{#1}}}
\newrobustcmd{\widehatbmQz}[2][]{\ensuremath{\subp{\widehat{\bm{Q}}}{}{#2}{}{#1}}}
\newrobustcmd{\checkbmQz}[2][]{\ensuremath{\subp{\check{\bm{Q}}}{}{#2}{}{#1}}}
\newrobustcmd{\tildebmQz}[2][]{\ensuremath{\subp{\tilde{\bm{Q}}}{}{#2}{}{#1}}}
\newrobustcmd{\widetildebmQz}[2][]{\ensuremath{\subp{\widetilde{\bm{Q}}}{}{#2}{}{#1}}}
\newrobustcmd{\acutebmQz}[2][]{\ensuremath{\subp{\acute{\bm{Q}}}{}{#2}{}{#1}}}
\newrobustcmd{\gravebmQz}[2][]{\ensuremath{\subp{\grave{\bm{Q}}}{}{#2}{}{#1}}}
\newrobustcmd{\dotbmQz}[2][]{\ensuremath{\subp{\dot{\bm{Q}}}{}{#2}{}{#1}}}
\newrobustcmd{\ddotbmQz}[2][]{\ensuremath{\subp{\ddot{\bm{Q}}}{}{#2}{}{#1}}}
\newrobustcmd{\brevebmQz}[2][]{\ensuremath{\subp{\breve{\bm{Q}}}{}{#2}{}{#1}}}
\newrobustcmd{\barbmQz}[2][]{\ensuremath{\subp{\bar{\bm{Q}}}{}{#2}{}{#1}}}
\newrobustcmd{\vecbmQz}[2][]{\ensuremath{\subp{\vec{\bm{Q}}}{}{#2}{}{#1}}}
\newrobustcmd{\Rz}[2][]{\ensuremath{\subp{R}{}{#2}{}{#1}}}
\newrobustcmd{\hatRz}[2][]{\ensuremath{\subp{\hat{R}}{}{#2}{}{#1}}}
\newrobustcmd{\widehatRz}[2][]{\ensuremath{\subp{\widehat{R}}{}{#2}{}{#1}}}
\newrobustcmd{\checkRz}[2][]{\ensuremath{\subp{\check{R}}{}{#2}{}{#1}}}
\newrobustcmd{\tildeRz}[2][]{\ensuremath{\subp{\tilde{R}}{}{#2}{}{#1}}}
\newrobustcmd{\widetildeRz}[2][]{\ensuremath{\subp{\widetilde{R}}{}{#2}{}{#1}}}
\newrobustcmd{\acuteRz}[2][]{\ensuremath{\subp{\acute{R}}{}{#2}{}{#1}}}
\newrobustcmd{\graveRz}[2][]{\ensuremath{\subp{\grave{R}}{}{#2}{}{#1}}}
\newrobustcmd{\dotRz}[2][]{\ensuremath{\subp{\dot{R}}{}{#2}{}{#1}}}
\newrobustcmd{\ddotRz}[2][]{\ensuremath{\subp{\ddot{R}}{}{#2}{}{#1}}}
\newrobustcmd{\breveRz}[2][]{\ensuremath{\subp{\breve{R}}{}{#2}{}{#1}}}
\newrobustcmd{\barRz}[2][]{\ensuremath{\subp{\bar{R}}{}{#2}{}{#1}}}
\newrobustcmd{\vecRz}[2][]{\ensuremath{\subp{\vec{R}}{}{#2}{}{#1}}}
\newrobustcmd{\bmRz}[2][]{\ensuremath{\subp{\bm{R}}{}{#2}{}{#1}}}
\newrobustcmd{\hatbmRz}[2][]{\ensuremath{\subp{\hat{\bm{R}}}{}{#2}{}{#1}}}
\newrobustcmd{\widehatbmRz}[2][]{\ensuremath{\subp{\widehat{\bm{R}}}{}{#2}{}{#1}}}
\newrobustcmd{\checkbmRz}[2][]{\ensuremath{\subp{\check{\bm{R}}}{}{#2}{}{#1}}}
\newrobustcmd{\tildebmRz}[2][]{\ensuremath{\subp{\tilde{\bm{R}}}{}{#2}{}{#1}}}
\newrobustcmd{\widetildebmRz}[2][]{\ensuremath{\subp{\widetilde{\bm{R}}}{}{#2}{}{#1}}}
\newrobustcmd{\acutebmRz}[2][]{\ensuremath{\subp{\acute{\bm{R}}}{}{#2}{}{#1}}}
\newrobustcmd{\gravebmRz}[2][]{\ensuremath{\subp{\grave{\bm{R}}}{}{#2}{}{#1}}}
\newrobustcmd{\dotbmRz}[2][]{\ensuremath{\subp{\dot{\bm{R}}}{}{#2}{}{#1}}}
\newrobustcmd{\ddotbmRz}[2][]{\ensuremath{\subp{\ddot{\bm{R}}}{}{#2}{}{#1}}}
\newrobustcmd{\brevebmRz}[2][]{\ensuremath{\subp{\breve{\bm{R}}}{}{#2}{}{#1}}}
\newrobustcmd{\barbmRz}[2][]{\ensuremath{\subp{\bar{\bm{R}}}{}{#2}{}{#1}}}
\newrobustcmd{\vecbmRz}[2][]{\ensuremath{\subp{\vec{\bm{R}}}{}{#2}{}{#1}}}
\newrobustcmd{\Sz}[2][]{\ensuremath{\subp{S}{}{#2}{}{#1}}}
\newrobustcmd{\hatSz}[2][]{\ensuremath{\subp{\hat{S}}{}{#2}{}{#1}}}
\newrobustcmd{\widehatSz}[2][]{\ensuremath{\subp{\widehat{S}}{}{#2}{}{#1}}}
\newrobustcmd{\checkSz}[2][]{\ensuremath{\subp{\check{S}}{}{#2}{}{#1}}}
\newrobustcmd{\tildeSz}[2][]{\ensuremath{\subp{\tilde{S}}{}{#2}{}{#1}}}
\newrobustcmd{\widetildeSz}[2][]{\ensuremath{\subp{\widetilde{S}}{}{#2}{}{#1}}}
\newrobustcmd{\acuteSz}[2][]{\ensuremath{\subp{\acute{S}}{}{#2}{}{#1}}}
\newrobustcmd{\graveSz}[2][]{\ensuremath{\subp{\grave{S}}{}{#2}{}{#1}}}
\newrobustcmd{\dotSz}[2][]{\ensuremath{\subp{\dot{S}}{}{#2}{}{#1}}}
\newrobustcmd{\ddotSz}[2][]{\ensuremath{\subp{\ddot{S}}{}{#2}{}{#1}}}
\newrobustcmd{\breveSz}[2][]{\ensuremath{\subp{\breve{S}}{}{#2}{}{#1}}}
\newrobustcmd{\barSz}[2][]{\ensuremath{\subp{\bar{S}}{}{#2}{}{#1}}}
\newrobustcmd{\vecSz}[2][]{\ensuremath{\subp{\vec{S}}{}{#2}{}{#1}}}
\newrobustcmd{\bmSz}[2][]{\ensuremath{\subp{\bm{S}}{}{#2}{}{#1}}}
\newrobustcmd{\hatbmSz}[2][]{\ensuremath{\subp{\hat{\bm{S}}}{}{#2}{}{#1}}}
\newrobustcmd{\widehatbmSz}[2][]{\ensuremath{\subp{\widehat{\bm{S}}}{}{#2}{}{#1}}}
\newrobustcmd{\checkbmSz}[2][]{\ensuremath{\subp{\check{\bm{S}}}{}{#2}{}{#1}}}
\newrobustcmd{\tildebmSz}[2][]{\ensuremath{\subp{\tilde{\bm{S}}}{}{#2}{}{#1}}}
\newrobustcmd{\widetildebmSz}[2][]{\ensuremath{\subp{\widetilde{\bm{S}}}{}{#2}{}{#1}}}
\newrobustcmd{\acutebmSz}[2][]{\ensuremath{\subp{\acute{\bm{S}}}{}{#2}{}{#1}}}
\newrobustcmd{\gravebmSz}[2][]{\ensuremath{\subp{\grave{\bm{S}}}{}{#2}{}{#1}}}
\newrobustcmd{\dotbmSz}[2][]{\ensuremath{\subp{\dot{\bm{S}}}{}{#2}{}{#1}}}
\newrobustcmd{\ddotbmSz}[2][]{\ensuremath{\subp{\ddot{\bm{S}}}{}{#2}{}{#1}}}
\newrobustcmd{\brevebmSz}[2][]{\ensuremath{\subp{\breve{\bm{S}}}{}{#2}{}{#1}}}
\newrobustcmd{\barbmSz}[2][]{\ensuremath{\subp{\bar{\bm{S}}}{}{#2}{}{#1}}}
\newrobustcmd{\vecbmSz}[2][]{\ensuremath{\subp{\vec{\bm{S}}}{}{#2}{}{#1}}}
\newrobustcmd{\Tz}[2][]{\ensuremath{\subp{T}{}{#2}{}{#1}}}
\newrobustcmd{\hatTz}[2][]{\ensuremath{\subp{\hat{T}}{}{#2}{}{#1}}}
\newrobustcmd{\widehatTz}[2][]{\ensuremath{\subp{\widehat{T}}{}{#2}{}{#1}}}
\newrobustcmd{\checkTz}[2][]{\ensuremath{\subp{\check{T}}{}{#2}{}{#1}}}
\newrobustcmd{\tildeTz}[2][]{\ensuremath{\subp{\tilde{T}}{}{#2}{}{#1}}}
\newrobustcmd{\widetildeTz}[2][]{\ensuremath{\subp{\widetilde{T}}{}{#2}{}{#1}}}
\newrobustcmd{\acuteTz}[2][]{\ensuremath{\subp{\acute{T}}{}{#2}{}{#1}}}
\newrobustcmd{\graveTz}[2][]{\ensuremath{\subp{\grave{T}}{}{#2}{}{#1}}}
\newrobustcmd{\dotTz}[2][]{\ensuremath{\subp{\dot{T}}{}{#2}{}{#1}}}
\newrobustcmd{\ddotTz}[2][]{\ensuremath{\subp{\ddot{T}}{}{#2}{}{#1}}}
\newrobustcmd{\breveTz}[2][]{\ensuremath{\subp{\breve{T}}{}{#2}{}{#1}}}
\newrobustcmd{\barTz}[2][]{\ensuremath{\subp{\bar{T}}{}{#2}{}{#1}}}
\newrobustcmd{\vecTz}[2][]{\ensuremath{\subp{\vec{T}}{}{#2}{}{#1}}}
\newrobustcmd{\bmTz}[2][]{\ensuremath{\subp{\bm{T}}{}{#2}{}{#1}}}
\newrobustcmd{\hatbmTz}[2][]{\ensuremath{\subp{\hat{\bm{T}}}{}{#2}{}{#1}}}
\newrobustcmd{\widehatbmTz}[2][]{\ensuremath{\subp{\widehat{\bm{T}}}{}{#2}{}{#1}}}
\newrobustcmd{\checkbmTz}[2][]{\ensuremath{\subp{\check{\bm{T}}}{}{#2}{}{#1}}}
\newrobustcmd{\tildebmTz}[2][]{\ensuremath{\subp{\tilde{\bm{T}}}{}{#2}{}{#1}}}
\newrobustcmd{\widetildebmTz}[2][]{\ensuremath{\subp{\widetilde{\bm{T}}}{}{#2}{}{#1}}}
\newrobustcmd{\acutebmTz}[2][]{\ensuremath{\subp{\acute{\bm{T}}}{}{#2}{}{#1}}}
\newrobustcmd{\gravebmTz}[2][]{\ensuremath{\subp{\grave{\bm{T}}}{}{#2}{}{#1}}}
\newrobustcmd{\dotbmTz}[2][]{\ensuremath{\subp{\dot{\bm{T}}}{}{#2}{}{#1}}}
\newrobustcmd{\ddotbmTz}[2][]{\ensuremath{\subp{\ddot{\bm{T}}}{}{#2}{}{#1}}}
\newrobustcmd{\brevebmTz}[2][]{\ensuremath{\subp{\breve{\bm{T}}}{}{#2}{}{#1}}}
\newrobustcmd{\barbmTz}[2][]{\ensuremath{\subp{\bar{\bm{T}}}{}{#2}{}{#1}}}
\newrobustcmd{\vecbmTz}[2][]{\ensuremath{\subp{\vec{\bm{T}}}{}{#2}{}{#1}}}
\newrobustcmd{\Uz}[2][]{\ensuremath{\subp{U}{}{#2}{}{#1}}}
\newrobustcmd{\hatUz}[2][]{\ensuremath{\subp{\hat{U}}{}{#2}{}{#1}}}
\newrobustcmd{\widehatUz}[2][]{\ensuremath{\subp{\widehat{U}}{}{#2}{}{#1}}}
\newrobustcmd{\checkUz}[2][]{\ensuremath{\subp{\check{U}}{}{#2}{}{#1}}}
\newrobustcmd{\tildeUz}[2][]{\ensuremath{\subp{\tilde{U}}{}{#2}{}{#1}}}
\newrobustcmd{\widetildeUz}[2][]{\ensuremath{\subp{\widetilde{U}}{}{#2}{}{#1}}}
\newrobustcmd{\acuteUz}[2][]{\ensuremath{\subp{\acute{U}}{}{#2}{}{#1}}}
\newrobustcmd{\graveUz}[2][]{\ensuremath{\subp{\grave{U}}{}{#2}{}{#1}}}
\newrobustcmd{\dotUz}[2][]{\ensuremath{\subp{\dot{U}}{}{#2}{}{#1}}}
\newrobustcmd{\ddotUz}[2][]{\ensuremath{\subp{\ddot{U}}{}{#2}{}{#1}}}
\newrobustcmd{\breveUz}[2][]{\ensuremath{\subp{\breve{U}}{}{#2}{}{#1}}}
\newrobustcmd{\barUz}[2][]{\ensuremath{\subp{\bar{U}}{}{#2}{}{#1}}}
\newrobustcmd{\vecUz}[2][]{\ensuremath{\subp{\vec{U}}{}{#2}{}{#1}}}
\newrobustcmd{\bmUz}[2][]{\ensuremath{\subp{\bm{U}}{}{#2}{}{#1}}}
\newrobustcmd{\hatbmUz}[2][]{\ensuremath{\subp{\hat{\bm{U}}}{}{#2}{}{#1}}}
\newrobustcmd{\widehatbmUz}[2][]{\ensuremath{\subp{\widehat{\bm{U}}}{}{#2}{}{#1}}}
\newrobustcmd{\checkbmUz}[2][]{\ensuremath{\subp{\check{\bm{U}}}{}{#2}{}{#1}}}
\newrobustcmd{\tildebmUz}[2][]{\ensuremath{\subp{\tilde{\bm{U}}}{}{#2}{}{#1}}}
\newrobustcmd{\widetildebmUz}[2][]{\ensuremath{\subp{\widetilde{\bm{U}}}{}{#2}{}{#1}}}
\newrobustcmd{\acutebmUz}[2][]{\ensuremath{\subp{\acute{\bm{U}}}{}{#2}{}{#1}}}
\newrobustcmd{\gravebmUz}[2][]{\ensuremath{\subp{\grave{\bm{U}}}{}{#2}{}{#1}}}
\newrobustcmd{\dotbmUz}[2][]{\ensuremath{\subp{\dot{\bm{U}}}{}{#2}{}{#1}}}
\newrobustcmd{\ddotbmUz}[2][]{\ensuremath{\subp{\ddot{\bm{U}}}{}{#2}{}{#1}}}
\newrobustcmd{\brevebmUz}[2][]{\ensuremath{\subp{\breve{\bm{U}}}{}{#2}{}{#1}}}
\newrobustcmd{\barbmUz}[2][]{\ensuremath{\subp{\bar{\bm{U}}}{}{#2}{}{#1}}}
\newrobustcmd{\vecbmUz}[2][]{\ensuremath{\subp{\vec{\bm{U}}}{}{#2}{}{#1}}}
\newrobustcmd{\Vz}[2][]{\ensuremath{\subp{V}{}{#2}{}{#1}}}
\newrobustcmd{\hatVz}[2][]{\ensuremath{\subp{\hat{V}}{}{#2}{}{#1}}}
\newrobustcmd{\widehatVz}[2][]{\ensuremath{\subp{\widehat{V}}{}{#2}{}{#1}}}
\newrobustcmd{\checkVz}[2][]{\ensuremath{\subp{\check{V}}{}{#2}{}{#1}}}
\newrobustcmd{\tildeVz}[2][]{\ensuremath{\subp{\tilde{V}}{}{#2}{}{#1}}}
\newrobustcmd{\widetildeVz}[2][]{\ensuremath{\subp{\widetilde{V}}{}{#2}{}{#1}}}
\newrobustcmd{\acuteVz}[2][]{\ensuremath{\subp{\acute{V}}{}{#2}{}{#1}}}
\newrobustcmd{\graveVz}[2][]{\ensuremath{\subp{\grave{V}}{}{#2}{}{#1}}}
\newrobustcmd{\dotVz}[2][]{\ensuremath{\subp{\dot{V}}{}{#2}{}{#1}}}
\newrobustcmd{\ddotVz}[2][]{\ensuremath{\subp{\ddot{V}}{}{#2}{}{#1}}}
\newrobustcmd{\breveVz}[2][]{\ensuremath{\subp{\breve{V}}{}{#2}{}{#1}}}
\newrobustcmd{\barVz}[2][]{\ensuremath{\subp{\bar{V}}{}{#2}{}{#1}}}
\newrobustcmd{\vecVz}[2][]{\ensuremath{\subp{\vec{V}}{}{#2}{}{#1}}}
\newrobustcmd{\bmVz}[2][]{\ensuremath{\subp{\bm{V}}{}{#2}{}{#1}}}
\newrobustcmd{\hatbmVz}[2][]{\ensuremath{\subp{\hat{\bm{V}}}{}{#2}{}{#1}}}
\newrobustcmd{\widehatbmVz}[2][]{\ensuremath{\subp{\widehat{\bm{V}}}{}{#2}{}{#1}}}
\newrobustcmd{\checkbmVz}[2][]{\ensuremath{\subp{\check{\bm{V}}}{}{#2}{}{#1}}}
\newrobustcmd{\tildebmVz}[2][]{\ensuremath{\subp{\tilde{\bm{V}}}{}{#2}{}{#1}}}
\newrobustcmd{\widetildebmVz}[2][]{\ensuremath{\subp{\widetilde{\bm{V}}}{}{#2}{}{#1}}}
\newrobustcmd{\acutebmVz}[2][]{\ensuremath{\subp{\acute{\bm{V}}}{}{#2}{}{#1}}}
\newrobustcmd{\gravebmVz}[2][]{\ensuremath{\subp{\grave{\bm{V}}}{}{#2}{}{#1}}}
\newrobustcmd{\dotbmVz}[2][]{\ensuremath{\subp{\dot{\bm{V}}}{}{#2}{}{#1}}}
\newrobustcmd{\ddotbmVz}[2][]{\ensuremath{\subp{\ddot{\bm{V}}}{}{#2}{}{#1}}}
\newrobustcmd{\brevebmVz}[2][]{\ensuremath{\subp{\breve{\bm{V}}}{}{#2}{}{#1}}}
\newrobustcmd{\barbmVz}[2][]{\ensuremath{\subp{\bar{\bm{V}}}{}{#2}{}{#1}}}
\newrobustcmd{\vecbmVz}[2][]{\ensuremath{\subp{\vec{\bm{V}}}{}{#2}{}{#1}}}
\newrobustcmd{\Wz}[2][]{\ensuremath{\subp{W}{}{#2}{}{#1}}}
\newrobustcmd{\hatWz}[2][]{\ensuremath{\subp{\hat{W}}{}{#2}{}{#1}}}
\newrobustcmd{\widehatWz}[2][]{\ensuremath{\subp{\widehat{W}}{}{#2}{}{#1}}}
\newrobustcmd{\checkWz}[2][]{\ensuremath{\subp{\check{W}}{}{#2}{}{#1}}}
\newrobustcmd{\tildeWz}[2][]{\ensuremath{\subp{\tilde{W}}{}{#2}{}{#1}}}
\newrobustcmd{\widetildeWz}[2][]{\ensuremath{\subp{\widetilde{W}}{}{#2}{}{#1}}}
\newrobustcmd{\acuteWz}[2][]{\ensuremath{\subp{\acute{W}}{}{#2}{}{#1}}}
\newrobustcmd{\graveWz}[2][]{\ensuremath{\subp{\grave{W}}{}{#2}{}{#1}}}
\newrobustcmd{\dotWz}[2][]{\ensuremath{\subp{\dot{W}}{}{#2}{}{#1}}}
\newrobustcmd{\ddotWz}[2][]{\ensuremath{\subp{\ddot{W}}{}{#2}{}{#1}}}
\newrobustcmd{\breveWz}[2][]{\ensuremath{\subp{\breve{W}}{}{#2}{}{#1}}}
\newrobustcmd{\barWz}[2][]{\ensuremath{\subp{\bar{W}}{}{#2}{}{#1}}}
\newrobustcmd{\vecWz}[2][]{\ensuremath{\subp{\vec{W}}{}{#2}{}{#1}}}
\newrobustcmd{\bmWz}[2][]{\ensuremath{\subp{\bm{W}}{}{#2}{}{#1}}}
\newrobustcmd{\hatbmWz}[2][]{\ensuremath{\subp{\hat{\bm{W}}}{}{#2}{}{#1}}}
\newrobustcmd{\widehatbmWz}[2][]{\ensuremath{\subp{\widehat{\bm{W}}}{}{#2}{}{#1}}}
\newrobustcmd{\checkbmWz}[2][]{\ensuremath{\subp{\check{\bm{W}}}{}{#2}{}{#1}}}
\newrobustcmd{\tildebmWz}[2][]{\ensuremath{\subp{\tilde{\bm{W}}}{}{#2}{}{#1}}}
\newrobustcmd{\widetildebmWz}[2][]{\ensuremath{\subp{\widetilde{\bm{W}}}{}{#2}{}{#1}}}
\newrobustcmd{\acutebmWz}[2][]{\ensuremath{\subp{\acute{\bm{W}}}{}{#2}{}{#1}}}
\newrobustcmd{\gravebmWz}[2][]{\ensuremath{\subp{\grave{\bm{W}}}{}{#2}{}{#1}}}
\newrobustcmd{\dotbmWz}[2][]{\ensuremath{\subp{\dot{\bm{W}}}{}{#2}{}{#1}}}
\newrobustcmd{\ddotbmWz}[2][]{\ensuremath{\subp{\ddot{\bm{W}}}{}{#2}{}{#1}}}
\newrobustcmd{\brevebmWz}[2][]{\ensuremath{\subp{\breve{\bm{W}}}{}{#2}{}{#1}}}
\newrobustcmd{\barbmWz}[2][]{\ensuremath{\subp{\bar{\bm{W}}}{}{#2}{}{#1}}}
\newrobustcmd{\vecbmWz}[2][]{\ensuremath{\subp{\vec{\bm{W}}}{}{#2}{}{#1}}}
\newrobustcmd{\Xz}[2][]{\ensuremath{\subp{X}{}{#2}{}{#1}}}
\newrobustcmd{\hatXz}[2][]{\ensuremath{\subp{\hat{X}}{}{#2}{}{#1}}}
\newrobustcmd{\widehatXz}[2][]{\ensuremath{\subp{\widehat{X}}{}{#2}{}{#1}}}
\newrobustcmd{\checkXz}[2][]{\ensuremath{\subp{\check{X}}{}{#2}{}{#1}}}
\newrobustcmd{\tildeXz}[2][]{\ensuremath{\subp{\tilde{X}}{}{#2}{}{#1}}}
\newrobustcmd{\widetildeXz}[2][]{\ensuremath{\subp{\widetilde{X}}{}{#2}{}{#1}}}
\newrobustcmd{\acuteXz}[2][]{\ensuremath{\subp{\acute{X}}{}{#2}{}{#1}}}
\newrobustcmd{\graveXz}[2][]{\ensuremath{\subp{\grave{X}}{}{#2}{}{#1}}}
\newrobustcmd{\dotXz}[2][]{\ensuremath{\subp{\dot{X}}{}{#2}{}{#1}}}
\newrobustcmd{\ddotXz}[2][]{\ensuremath{\subp{\ddot{X}}{}{#2}{}{#1}}}
\newrobustcmd{\breveXz}[2][]{\ensuremath{\subp{\breve{X}}{}{#2}{}{#1}}}
\newrobustcmd{\barXz}[2][]{\ensuremath{\subp{\bar{X}}{}{#2}{}{#1}}}
\newrobustcmd{\vecXz}[2][]{\ensuremath{\subp{\vec{X}}{}{#2}{}{#1}}}
\newrobustcmd{\bmXz}[2][]{\ensuremath{\subp{\bm{X}}{}{#2}{}{#1}}}
\newrobustcmd{\hatbmXz}[2][]{\ensuremath{\subp{\hat{\bm{X}}}{}{#2}{}{#1}}}
\newrobustcmd{\widehatbmXz}[2][]{\ensuremath{\subp{\widehat{\bm{X}}}{}{#2}{}{#1}}}
\newrobustcmd{\checkbmXz}[2][]{\ensuremath{\subp{\check{\bm{X}}}{}{#2}{}{#1}}}
\newrobustcmd{\tildebmXz}[2][]{\ensuremath{\subp{\tilde{\bm{X}}}{}{#2}{}{#1}}}
\newrobustcmd{\widetildebmXz}[2][]{\ensuremath{\subp{\widetilde{\bm{X}}}{}{#2}{}{#1}}}
\newrobustcmd{\acutebmXz}[2][]{\ensuremath{\subp{\acute{\bm{X}}}{}{#2}{}{#1}}}
\newrobustcmd{\gravebmXz}[2][]{\ensuremath{\subp{\grave{\bm{X}}}{}{#2}{}{#1}}}
\newrobustcmd{\dotbmXz}[2][]{\ensuremath{\subp{\dot{\bm{X}}}{}{#2}{}{#1}}}
\newrobustcmd{\ddotbmXz}[2][]{\ensuremath{\subp{\ddot{\bm{X}}}{}{#2}{}{#1}}}
\newrobustcmd{\brevebmXz}[2][]{\ensuremath{\subp{\breve{\bm{X}}}{}{#2}{}{#1}}}
\newrobustcmd{\barbmXz}[2][]{\ensuremath{\subp{\bar{\bm{X}}}{}{#2}{}{#1}}}
\newrobustcmd{\vecbmXz}[2][]{\ensuremath{\subp{\vec{\bm{X}}}{}{#2}{}{#1}}}
\newrobustcmd{\Yz}[2][]{\ensuremath{\subp{Y}{}{#2}{}{#1}}}
\newrobustcmd{\hatYz}[2][]{\ensuremath{\subp{\hat{Y}}{}{#2}{}{#1}}}
\newrobustcmd{\widehatYz}[2][]{\ensuremath{\subp{\widehat{Y}}{}{#2}{}{#1}}}
\newrobustcmd{\checkYz}[2][]{\ensuremath{\subp{\check{Y}}{}{#2}{}{#1}}}
\newrobustcmd{\tildeYz}[2][]{\ensuremath{\subp{\tilde{Y}}{}{#2}{}{#1}}}
\newrobustcmd{\widetildeYz}[2][]{\ensuremath{\subp{\widetilde{Y}}{}{#2}{}{#1}}}
\newrobustcmd{\acuteYz}[2][]{\ensuremath{\subp{\acute{Y}}{}{#2}{}{#1}}}
\newrobustcmd{\graveYz}[2][]{\ensuremath{\subp{\grave{Y}}{}{#2}{}{#1}}}
\newrobustcmd{\dotYz}[2][]{\ensuremath{\subp{\dot{Y}}{}{#2}{}{#1}}}
\newrobustcmd{\ddotYz}[2][]{\ensuremath{\subp{\ddot{Y}}{}{#2}{}{#1}}}
\newrobustcmd{\breveYz}[2][]{\ensuremath{\subp{\breve{Y}}{}{#2}{}{#1}}}
\newrobustcmd{\barYz}[2][]{\ensuremath{\subp{\bar{Y}}{}{#2}{}{#1}}}
\newrobustcmd{\vecYz}[2][]{\ensuremath{\subp{\vec{Y}}{}{#2}{}{#1}}}
\newrobustcmd{\bmYz}[2][]{\ensuremath{\subp{\bm{Y}}{}{#2}{}{#1}}}
\newrobustcmd{\hatbmYz}[2][]{\ensuremath{\subp{\hat{\bm{Y}}}{}{#2}{}{#1}}}
\newrobustcmd{\widehatbmYz}[2][]{\ensuremath{\subp{\widehat{\bm{Y}}}{}{#2}{}{#1}}}
\newrobustcmd{\checkbmYz}[2][]{\ensuremath{\subp{\check{\bm{Y}}}{}{#2}{}{#1}}}
\newrobustcmd{\tildebmYz}[2][]{\ensuremath{\subp{\tilde{\bm{Y}}}{}{#2}{}{#1}}}
\newrobustcmd{\widetildebmYz}[2][]{\ensuremath{\subp{\widetilde{\bm{Y}}}{}{#2}{}{#1}}}
\newrobustcmd{\acutebmYz}[2][]{\ensuremath{\subp{\acute{\bm{Y}}}{}{#2}{}{#1}}}
\newrobustcmd{\gravebmYz}[2][]{\ensuremath{\subp{\grave{\bm{Y}}}{}{#2}{}{#1}}}
\newrobustcmd{\dotbmYz}[2][]{\ensuremath{\subp{\dot{\bm{Y}}}{}{#2}{}{#1}}}
\newrobustcmd{\ddotbmYz}[2][]{\ensuremath{\subp{\ddot{\bm{Y}}}{}{#2}{}{#1}}}
\newrobustcmd{\brevebmYz}[2][]{\ensuremath{\subp{\breve{\bm{Y}}}{}{#2}{}{#1}}}
\newrobustcmd{\barbmYz}[2][]{\ensuremath{\subp{\bar{\bm{Y}}}{}{#2}{}{#1}}}
\newrobustcmd{\vecbmYz}[2][]{\ensuremath{\subp{\vec{\bm{Y}}}{}{#2}{}{#1}}}
\newrobustcmd{\Zz}[2][]{\ensuremath{\subp{Z}{}{#2}{}{#1}}}
\newrobustcmd{\hatZz}[2][]{\ensuremath{\subp{\hat{Z}}{}{#2}{}{#1}}}
\newrobustcmd{\widehatZz}[2][]{\ensuremath{\subp{\widehat{Z}}{}{#2}{}{#1}}}
\newrobustcmd{\checkZz}[2][]{\ensuremath{\subp{\check{Z}}{}{#2}{}{#1}}}
\newrobustcmd{\tildeZz}[2][]{\ensuremath{\subp{\tilde{Z}}{}{#2}{}{#1}}}
\newrobustcmd{\widetildeZz}[2][]{\ensuremath{\subp{\widetilde{Z}}{}{#2}{}{#1}}}
\newrobustcmd{\acuteZz}[2][]{\ensuremath{\subp{\acute{Z}}{}{#2}{}{#1}}}
\newrobustcmd{\graveZz}[2][]{\ensuremath{\subp{\grave{Z}}{}{#2}{}{#1}}}
\newrobustcmd{\dotZz}[2][]{\ensuremath{\subp{\dot{Z}}{}{#2}{}{#1}}}
\newrobustcmd{\ddotZz}[2][]{\ensuremath{\subp{\ddot{Z}}{}{#2}{}{#1}}}
\newrobustcmd{\breveZz}[2][]{\ensuremath{\subp{\breve{Z}}{}{#2}{}{#1}}}
\newrobustcmd{\barZz}[2][]{\ensuremath{\subp{\bar{Z}}{}{#2}{}{#1}}}
\newrobustcmd{\vecZz}[2][]{\ensuremath{\subp{\vec{Z}}{}{#2}{}{#1}}}
\newrobustcmd{\bmZz}[2][]{\ensuremath{\subp{\bm{Z}}{}{#2}{}{#1}}}
\newrobustcmd{\hatbmZz}[2][]{\ensuremath{\subp{\hat{\bm{Z}}}{}{#2}{}{#1}}}
\newrobustcmd{\widehatbmZz}[2][]{\ensuremath{\subp{\widehat{\bm{Z}}}{}{#2}{}{#1}}}
\newrobustcmd{\checkbmZz}[2][]{\ensuremath{\subp{\check{\bm{Z}}}{}{#2}{}{#1}}}
\newrobustcmd{\tildebmZz}[2][]{\ensuremath{\subp{\tilde{\bm{Z}}}{}{#2}{}{#1}}}
\newrobustcmd{\widetildebmZz}[2][]{\ensuremath{\subp{\widetilde{\bm{Z}}}{}{#2}{}{#1}}}
\newrobustcmd{\acutebmZz}[2][]{\ensuremath{\subp{\acute{\bm{Z}}}{}{#2}{}{#1}}}
\newrobustcmd{\gravebmZz}[2][]{\ensuremath{\subp{\grave{\bm{Z}}}{}{#2}{}{#1}}}
\newrobustcmd{\dotbmZz}[2][]{\ensuremath{\subp{\dot{\bm{Z}}}{}{#2}{}{#1}}}
\newrobustcmd{\ddotbmZz}[2][]{\ensuremath{\subp{\ddot{\bm{Z}}}{}{#2}{}{#1}}}
\newrobustcmd{\brevebmZz}[2][]{\ensuremath{\subp{\breve{\bm{Z}}}{}{#2}{}{#1}}}
\newrobustcmd{\barbmZz}[2][]{\ensuremath{\subp{\bar{\bm{Z}}}{}{#2}{}{#1}}}
\newrobustcmd{\vecbmZz}[2][]{\ensuremath{\subp{\vec{\bm{Z}}}{}{#2}{}{#1}}}

\newrobustcmd{\alphaz}[2][]{\ensuremath{\subp{\alpha}{}{#2}{}{#1}}}
\newrobustcmd{\hatalphaz}[2][]{\ensuremath{\subp{\hat{\alpha}}{}{#2}{}{#1}}}
\newrobustcmd{\widehatalphaz}[2][]{\ensuremath{\subp{\widehat{\alpha}}{}{#2}{}{#1}}}
\newrobustcmd{\checkalphaz}[2][]{\ensuremath{\subp{\check{\alpha}}{}{#2}{}{#1}}}
\newrobustcmd{\tildealphaz}[2][]{\ensuremath{\subp{\tilde{\alpha}}{}{#2}{}{#1}}}
\newrobustcmd{\widetildealphaz}[2][]{\ensuremath{\subp{\widetilde{\alpha}}{}{#2}{}{#1}}}
\newrobustcmd{\acutealphaz}[2][]{\ensuremath{\subp{\acute{\alpha}}{}{#2}{}{#1}}}
\newrobustcmd{\gravealphaz}[2][]{\ensuremath{\subp{\grave{\alpha}}{}{#2}{}{#1}}}
\newrobustcmd{\dotalphaz}[2][]{\ensuremath{\subp{\dot{\alpha}}{}{#2}{}{#1}}}
\newrobustcmd{\ddotalphaz}[2][]{\ensuremath{\subp{\ddot{\alpha}}{}{#2}{}{#1}}}
\newrobustcmd{\brevealphaz}[2][]{\ensuremath{\subp{\breve{\alpha}}{}{#2}{}{#1}}}
\newrobustcmd{\baralphaz}[2][]{\ensuremath{\subp{\bar{\alpha}}{}{#2}{}{#1}}}
\newrobustcmd{\vecalphaz}[2][]{\ensuremath{\subp{\vec{\alpha}}{}{#2}{}{#1}}}
\newrobustcmd{\bmalphaz}[2][]{\ensuremath{\subp{\bm{\alpha}}{}{#2}{}{#1}}}
\newrobustcmd{\hatbmalphaz}[2][]{\ensuremath{\subp{\hat{\bm{\alpha}}}{}{#2}{}{#1}}}
\newrobustcmd{\widehatbmalphaz}[2][]{\ensuremath{\subp{\widehat{\bm{\alpha}}}{}{#2}{}{#1}}}
\newrobustcmd{\checkbmalphaz}[2][]{\ensuremath{\subp{\check{\bm{\alpha}}}{}{#2}{}{#1}}}
\newrobustcmd{\tildebmalphaz}[2][]{\ensuremath{\subp{\tilde{\bm{\alpha}}}{}{#2}{}{#1}}}
\newrobustcmd{\widetildebmalphaz}[2][]{\ensuremath{\subp{\widetilde{\bm{\alpha}}}{}{#2}{}{#1}}}
\newrobustcmd{\acutebmalphaz}[2][]{\ensuremath{\subp{\acute{\bm{\alpha}}}{}{#2}{}{#1}}}
\newrobustcmd{\gravebmalphaz}[2][]{\ensuremath{\subp{\grave{\bm{\alpha}}}{}{#2}{}{#1}}}
\newrobustcmd{\dotbmalphaz}[2][]{\ensuremath{\subp{\dot{\bm{\alpha}}}{}{#2}{}{#1}}}
\newrobustcmd{\ddotbmalphaz}[2][]{\ensuremath{\subp{\ddot{\bm{\alpha}}}{}{#2}{}{#1}}}
\newrobustcmd{\brevebmalphaz}[2][]{\ensuremath{\subp{\breve{\bm{\alpha}}}{}{#2}{}{#1}}}
\newrobustcmd{\barbmalphaz}[2][]{\ensuremath{\subp{\bar{\bm{\alpha}}}{}{#2}{}{#1}}}
\newrobustcmd{\vecbmalphaz}[2][]{\ensuremath{\subp{\vec{\bm{\alpha}}}{}{#2}{}{#1}}}
\newrobustcmd{\betaz}[2][]{\ensuremath{\subp{\beta}{}{#2}{}{#1}}}
\newrobustcmd{\hatbetaz}[2][]{\ensuremath{\subp{\hat{\beta}}{}{#2}{}{#1}}}
\newrobustcmd{\widehatbetaz}[2][]{\ensuremath{\subp{\widehat{\beta}}{}{#2}{}{#1}}}
\newrobustcmd{\checkbetaz}[2][]{\ensuremath{\subp{\check{\beta}}{}{#2}{}{#1}}}
\newrobustcmd{\tildebetaz}[2][]{\ensuremath{\subp{\tilde{\beta}}{}{#2}{}{#1}}}
\newrobustcmd{\widetildebetaz}[2][]{\ensuremath{\subp{\widetilde{\beta}}{}{#2}{}{#1}}}
\newrobustcmd{\acutebetaz}[2][]{\ensuremath{\subp{\acute{\beta}}{}{#2}{}{#1}}}
\newrobustcmd{\gravebetaz}[2][]{\ensuremath{\subp{\grave{\beta}}{}{#2}{}{#1}}}
\newrobustcmd{\dotbetaz}[2][]{\ensuremath{\subp{\dot{\beta}}{}{#2}{}{#1}}}
\newrobustcmd{\ddotbetaz}[2][]{\ensuremath{\subp{\ddot{\beta}}{}{#2}{}{#1}}}
\newrobustcmd{\brevebetaz}[2][]{\ensuremath{\subp{\breve{\beta}}{}{#2}{}{#1}}}
\newrobustcmd{\barbetaz}[2][]{\ensuremath{\subp{\bar{\beta}}{}{#2}{}{#1}}}
\newrobustcmd{\vecbetaz}[2][]{\ensuremath{\subp{\vec{\beta}}{}{#2}{}{#1}}}
\newrobustcmd{\bmbetaz}[2][]{\ensuremath{\subp{\bm{\beta}}{}{#2}{}{#1}}}
\newrobustcmd{\hatbmbetaz}[2][]{\ensuremath{\subp{\hat{\bm{\beta}}}{}{#2}{}{#1}}}
\newrobustcmd{\widehatbmbetaz}[2][]{\ensuremath{\subp{\widehat{\bm{\beta}}}{}{#2}{}{#1}}}
\newrobustcmd{\checkbmbetaz}[2][]{\ensuremath{\subp{\check{\bm{\beta}}}{}{#2}{}{#1}}}
\newrobustcmd{\tildebmbetaz}[2][]{\ensuremath{\subp{\tilde{\bm{\beta}}}{}{#2}{}{#1}}}
\newrobustcmd{\widetildebmbetaz}[2][]{\ensuremath{\subp{\widetilde{\bm{\beta}}}{}{#2}{}{#1}}}
\newrobustcmd{\acutebmbetaz}[2][]{\ensuremath{\subp{\acute{\bm{\beta}}}{}{#2}{}{#1}}}
\newrobustcmd{\gravebmbetaz}[2][]{\ensuremath{\subp{\grave{\bm{\beta}}}{}{#2}{}{#1}}}
\newrobustcmd{\dotbmbetaz}[2][]{\ensuremath{\subp{\dot{\bm{\beta}}}{}{#2}{}{#1}}}
\newrobustcmd{\ddotbmbetaz}[2][]{\ensuremath{\subp{\ddot{\bm{\beta}}}{}{#2}{}{#1}}}
\newrobustcmd{\brevebmbetaz}[2][]{\ensuremath{\subp{\breve{\bm{\beta}}}{}{#2}{}{#1}}}
\newrobustcmd{\barbmbetaz}[2][]{\ensuremath{\subp{\bar{\bm{\beta}}}{}{#2}{}{#1}}}
\newrobustcmd{\vecbmbetaz}[2][]{\ensuremath{\subp{\vec{\bm{\beta}}}{}{#2}{}{#1}}}
\newrobustcmd{\gammaz}[2][]{\ensuremath{\subp{\gamma}{}{#2}{}{#1}}}
\newrobustcmd{\hatgammaz}[2][]{\ensuremath{\subp{\hat{\gamma}}{}{#2}{}{#1}}}
\newrobustcmd{\widehatgammaz}[2][]{\ensuremath{\subp{\widehat{\gamma}}{}{#2}{}{#1}}}
\newrobustcmd{\checkgammaz}[2][]{\ensuremath{\subp{\check{\gamma}}{}{#2}{}{#1}}}
\newrobustcmd{\tildegammaz}[2][]{\ensuremath{\subp{\tilde{\gamma}}{}{#2}{}{#1}}}
\newrobustcmd{\widetildegammaz}[2][]{\ensuremath{\subp{\widetilde{\gamma}}{}{#2}{}{#1}}}
\newrobustcmd{\acutegammaz}[2][]{\ensuremath{\subp{\acute{\gamma}}{}{#2}{}{#1}}}
\newrobustcmd{\gravegammaz}[2][]{\ensuremath{\subp{\grave{\gamma}}{}{#2}{}{#1}}}
\newrobustcmd{\dotgammaz}[2][]{\ensuremath{\subp{\dot{\gamma}}{}{#2}{}{#1}}}
\newrobustcmd{\ddotgammaz}[2][]{\ensuremath{\subp{\ddot{\gamma}}{}{#2}{}{#1}}}
\newrobustcmd{\brevegammaz}[2][]{\ensuremath{\subp{\breve{\gamma}}{}{#2}{}{#1}}}
\newrobustcmd{\bargammaz}[2][]{\ensuremath{\subp{\bar{\gamma}}{}{#2}{}{#1}}}
\newrobustcmd{\vecgammaz}[2][]{\ensuremath{\subp{\vec{\gamma}}{}{#2}{}{#1}}}
\newrobustcmd{\bmgammaz}[2][]{\ensuremath{\subp{\bm{\gamma}}{}{#2}{}{#1}}}
\newrobustcmd{\hatbmgammaz}[2][]{\ensuremath{\subp{\hat{\bm{\gamma}}}{}{#2}{}{#1}}}
\newrobustcmd{\widehatbmgammaz}[2][]{\ensuremath{\subp{\widehat{\bm{\gamma}}}{}{#2}{}{#1}}}
\newrobustcmd{\checkbmgammaz}[2][]{\ensuremath{\subp{\check{\bm{\gamma}}}{}{#2}{}{#1}}}
\newrobustcmd{\tildebmgammaz}[2][]{\ensuremath{\subp{\tilde{\bm{\gamma}}}{}{#2}{}{#1}}}
\newrobustcmd{\widetildebmgammaz}[2][]{\ensuremath{\subp{\widetilde{\bm{\gamma}}}{}{#2}{}{#1}}}
\newrobustcmd{\acutebmgammaz}[2][]{\ensuremath{\subp{\acute{\bm{\gamma}}}{}{#2}{}{#1}}}
\newrobustcmd{\gravebmgammaz}[2][]{\ensuremath{\subp{\grave{\bm{\gamma}}}{}{#2}{}{#1}}}
\newrobustcmd{\dotbmgammaz}[2][]{\ensuremath{\subp{\dot{\bm{\gamma}}}{}{#2}{}{#1}}}
\newrobustcmd{\ddotbmgammaz}[2][]{\ensuremath{\subp{\ddot{\bm{\gamma}}}{}{#2}{}{#1}}}
\newrobustcmd{\brevebmgammaz}[2][]{\ensuremath{\subp{\breve{\bm{\gamma}}}{}{#2}{}{#1}}}
\newrobustcmd{\barbmgammaz}[2][]{\ensuremath{\subp{\bar{\bm{\gamma}}}{}{#2}{}{#1}}}
\newrobustcmd{\vecbmgammaz}[2][]{\ensuremath{\subp{\vec{\bm{\gamma}}}{}{#2}{}{#1}}}
\newrobustcmd{\deltaz}[2][]{\ensuremath{\subp{\delta}{}{#2}{}{#1}}}
\newrobustcmd{\hatdeltaz}[2][]{\ensuremath{\subp{\hat{\delta}}{}{#2}{}{#1}}}
\newrobustcmd{\widehatdeltaz}[2][]{\ensuremath{\subp{\widehat{\delta}}{}{#2}{}{#1}}}
\newrobustcmd{\checkdeltaz}[2][]{\ensuremath{\subp{\check{\delta}}{}{#2}{}{#1}}}
\newrobustcmd{\tildedeltaz}[2][]{\ensuremath{\subp{\tilde{\delta}}{}{#2}{}{#1}}}
\newrobustcmd{\widetildedeltaz}[2][]{\ensuremath{\subp{\widetilde{\delta}}{}{#2}{}{#1}}}
\newrobustcmd{\acutedeltaz}[2][]{\ensuremath{\subp{\acute{\delta}}{}{#2}{}{#1}}}
\newrobustcmd{\gravedeltaz}[2][]{\ensuremath{\subp{\grave{\delta}}{}{#2}{}{#1}}}
\newrobustcmd{\dotdeltaz}[2][]{\ensuremath{\subp{\dot{\delta}}{}{#2}{}{#1}}}
\newrobustcmd{\ddotdeltaz}[2][]{\ensuremath{\subp{\ddot{\delta}}{}{#2}{}{#1}}}
\newrobustcmd{\brevedeltaz}[2][]{\ensuremath{\subp{\breve{\delta}}{}{#2}{}{#1}}}
\newrobustcmd{\bardeltaz}[2][]{\ensuremath{\subp{\bar{\delta}}{}{#2}{}{#1}}}
\newrobustcmd{\vecdeltaz}[2][]{\ensuremath{\subp{\vec{\delta}}{}{#2}{}{#1}}}
\newrobustcmd{\bmdeltaz}[2][]{\ensuremath{\subp{\bm{\delta}}{}{#2}{}{#1}}}
\newrobustcmd{\hatbmdeltaz}[2][]{\ensuremath{\subp{\hat{\bm{\delta}}}{}{#2}{}{#1}}}
\newrobustcmd{\widehatbmdeltaz}[2][]{\ensuremath{\subp{\widehat{\bm{\delta}}}{}{#2}{}{#1}}}
\newrobustcmd{\checkbmdeltaz}[2][]{\ensuremath{\subp{\check{\bm{\delta}}}{}{#2}{}{#1}}}
\newrobustcmd{\tildebmdeltaz}[2][]{\ensuremath{\subp{\tilde{\bm{\delta}}}{}{#2}{}{#1}}}
\newrobustcmd{\widetildebmdeltaz}[2][]{\ensuremath{\subp{\widetilde{\bm{\delta}}}{}{#2}{}{#1}}}
\newrobustcmd{\acutebmdeltaz}[2][]{\ensuremath{\subp{\acute{\bm{\delta}}}{}{#2}{}{#1}}}
\newrobustcmd{\gravebmdeltaz}[2][]{\ensuremath{\subp{\grave{\bm{\delta}}}{}{#2}{}{#1}}}
\newrobustcmd{\dotbmdeltaz}[2][]{\ensuremath{\subp{\dot{\bm{\delta}}}{}{#2}{}{#1}}}
\newrobustcmd{\ddotbmdeltaz}[2][]{\ensuremath{\subp{\ddot{\bm{\delta}}}{}{#2}{}{#1}}}
\newrobustcmd{\brevebmdeltaz}[2][]{\ensuremath{\subp{\breve{\bm{\delta}}}{}{#2}{}{#1}}}
\newrobustcmd{\barbmdeltaz}[2][]{\ensuremath{\subp{\bar{\bm{\delta}}}{}{#2}{}{#1}}}
\newrobustcmd{\vecbmdeltaz}[2][]{\ensuremath{\subp{\vec{\bm{\delta}}}{}{#2}{}{#1}}}
\newrobustcmd{\epsilonz}[2][]{\ensuremath{\subp{\epsilon}{}{#2}{}{#1}}}
\newrobustcmd{\hatepsilonz}[2][]{\ensuremath{\subp{\hat{\epsilon}}{}{#2}{}{#1}}}
\newrobustcmd{\widehatepsilonz}[2][]{\ensuremath{\subp{\widehat{\epsilon}}{}{#2}{}{#1}}}
\newrobustcmd{\checkepsilonz}[2][]{\ensuremath{\subp{\check{\epsilon}}{}{#2}{}{#1}}}
\newrobustcmd{\tildeepsilonz}[2][]{\ensuremath{\subp{\tilde{\epsilon}}{}{#2}{}{#1}}}
\newrobustcmd{\widetildeepsilonz}[2][]{\ensuremath{\subp{\widetilde{\epsilon}}{}{#2}{}{#1}}}
\newrobustcmd{\acuteepsilonz}[2][]{\ensuremath{\subp{\acute{\epsilon}}{}{#2}{}{#1}}}
\newrobustcmd{\graveepsilonz}[2][]{\ensuremath{\subp{\grave{\epsilon}}{}{#2}{}{#1}}}
\newrobustcmd{\dotepsilonz}[2][]{\ensuremath{\subp{\dot{\epsilon}}{}{#2}{}{#1}}}
\newrobustcmd{\ddotepsilonz}[2][]{\ensuremath{\subp{\ddot{\epsilon}}{}{#2}{}{#1}}}
\newrobustcmd{\breveepsilonz}[2][]{\ensuremath{\subp{\breve{\epsilon}}{}{#2}{}{#1}}}
\newrobustcmd{\barepsilonz}[2][]{\ensuremath{\subp{\bar{\epsilon}}{}{#2}{}{#1}}}
\newrobustcmd{\vecepsilonz}[2][]{\ensuremath{\subp{\vec{\epsilon}}{}{#2}{}{#1}}}
\newrobustcmd{\bmepsilonz}[2][]{\ensuremath{\subp{\bm{\epsilon}}{}{#2}{}{#1}}}
\newrobustcmd{\hatbmepsilonz}[2][]{\ensuremath{\subp{\hat{\bm{\epsilon}}}{}{#2}{}{#1}}}
\newrobustcmd{\widehatbmepsilonz}[2][]{\ensuremath{\subp{\widehat{\bm{\epsilon}}}{}{#2}{}{#1}}}
\newrobustcmd{\checkbmepsilonz}[2][]{\ensuremath{\subp{\check{\bm{\epsilon}}}{}{#2}{}{#1}}}
\newrobustcmd{\tildebmepsilonz}[2][]{\ensuremath{\subp{\tilde{\bm{\epsilon}}}{}{#2}{}{#1}}}
\newrobustcmd{\widetildebmepsilonz}[2][]{\ensuremath{\subp{\widetilde{\bm{\epsilon}}}{}{#2}{}{#1}}}
\newrobustcmd{\acutebmepsilonz}[2][]{\ensuremath{\subp{\acute{\bm{\epsilon}}}{}{#2}{}{#1}}}
\newrobustcmd{\gravebmepsilonz}[2][]{\ensuremath{\subp{\grave{\bm{\epsilon}}}{}{#2}{}{#1}}}
\newrobustcmd{\dotbmepsilonz}[2][]{\ensuremath{\subp{\dot{\bm{\epsilon}}}{}{#2}{}{#1}}}
\newrobustcmd{\ddotbmepsilonz}[2][]{\ensuremath{\subp{\ddot{\bm{\epsilon}}}{}{#2}{}{#1}}}
\newrobustcmd{\brevebmepsilonz}[2][]{\ensuremath{\subp{\breve{\bm{\epsilon}}}{}{#2}{}{#1}}}
\newrobustcmd{\barbmepsilonz}[2][]{\ensuremath{\subp{\bar{\bm{\epsilon}}}{}{#2}{}{#1}}}
\newrobustcmd{\vecbmepsilonz}[2][]{\ensuremath{\subp{\vec{\bm{\epsilon}}}{}{#2}{}{#1}}}
\newrobustcmd{\varepsilonz}[2][]{\ensuremath{\subp{\varepsilon}{}{#2}{}{#1}}}
\newrobustcmd{\hatvarepsilonz}[2][]{\ensuremath{\subp{\hat{\varepsilon}}{}{#2}{}{#1}}}
\newrobustcmd{\widehatvarepsilonz}[2][]{\ensuremath{\subp{\widehat{\varepsilon}}{}{#2}{}{#1}}}
\newrobustcmd{\checkvarepsilonz}[2][]{\ensuremath{\subp{\check{\varepsilon}}{}{#2}{}{#1}}}
\newrobustcmd{\tildevarepsilonz}[2][]{\ensuremath{\subp{\tilde{\varepsilon}}{}{#2}{}{#1}}}
\newrobustcmd{\widetildevarepsilonz}[2][]{\ensuremath{\subp{\widetilde{\varepsilon}}{}{#2}{}{#1}}}
\newrobustcmd{\acutevarepsilonz}[2][]{\ensuremath{\subp{\acute{\varepsilon}}{}{#2}{}{#1}}}
\newrobustcmd{\gravevarepsilonz}[2][]{\ensuremath{\subp{\grave{\varepsilon}}{}{#2}{}{#1}}}
\newrobustcmd{\dotvarepsilonz}[2][]{\ensuremath{\subp{\dot{\varepsilon}}{}{#2}{}{#1}}}
\newrobustcmd{\ddotvarepsilonz}[2][]{\ensuremath{\subp{\ddot{\varepsilon}}{}{#2}{}{#1}}}
\newrobustcmd{\brevevarepsilonz}[2][]{\ensuremath{\subp{\breve{\varepsilon}}{}{#2}{}{#1}}}
\newrobustcmd{\barvarepsilonz}[2][]{\ensuremath{\subp{\bar{\varepsilon}}{}{#2}{}{#1}}}
\newrobustcmd{\vecvarepsilonz}[2][]{\ensuremath{\subp{\vec{\varepsilon}}{}{#2}{}{#1}}}
\newrobustcmd{\bmvarepsilonz}[2][]{\ensuremath{\subp{\bm{\varepsilon}}{}{#2}{}{#1}}}
\newrobustcmd{\hatbmvarepsilonz}[2][]{\ensuremath{\subp{\hat{\bm{\varepsilon}}}{}{#2}{}{#1}}}
\newrobustcmd{\widehatbmvarepsilonz}[2][]{\ensuremath{\subp{\widehat{\bm{\varepsilon}}}{}{#2}{}{#1}}}
\newrobustcmd{\checkbmvarepsilonz}[2][]{\ensuremath{\subp{\check{\bm{\varepsilon}}}{}{#2}{}{#1}}}
\newrobustcmd{\tildebmvarepsilonz}[2][]{\ensuremath{\subp{\tilde{\bm{\varepsilon}}}{}{#2}{}{#1}}}
\newrobustcmd{\widetildebmvarepsilonz}[2][]{\ensuremath{\subp{\widetilde{\bm{\varepsilon}}}{}{#2}{}{#1}}}
\newrobustcmd{\acutebmvarepsilonz}[2][]{\ensuremath{\subp{\acute{\bm{\varepsilon}}}{}{#2}{}{#1}}}
\newrobustcmd{\gravebmvarepsilonz}[2][]{\ensuremath{\subp{\grave{\bm{\varepsilon}}}{}{#2}{}{#1}}}
\newrobustcmd{\dotbmvarepsilonz}[2][]{\ensuremath{\subp{\dot{\bm{\varepsilon}}}{}{#2}{}{#1}}}
\newrobustcmd{\ddotbmvarepsilonz}[2][]{\ensuremath{\subp{\ddot{\bm{\varepsilon}}}{}{#2}{}{#1}}}
\newrobustcmd{\brevebmvarepsilonz}[2][]{\ensuremath{\subp{\breve{\bm{\varepsilon}}}{}{#2}{}{#1}}}
\newrobustcmd{\barbmvarepsilonz}[2][]{\ensuremath{\subp{\bar{\bm{\varepsilon}}}{}{#2}{}{#1}}}
\newrobustcmd{\vecbmvarepsilonz}[2][]{\ensuremath{\subp{\vec{\bm{\varepsilon}}}{}{#2}{}{#1}}}
\newrobustcmd{\zetaz}[2][]{\ensuremath{\subp{\zeta}{}{#2}{}{#1}}}
\newrobustcmd{\hatzetaz}[2][]{\ensuremath{\subp{\hat{\zeta}}{}{#2}{}{#1}}}
\newrobustcmd{\widehatzetaz}[2][]{\ensuremath{\subp{\widehat{\zeta}}{}{#2}{}{#1}}}
\newrobustcmd{\checkzetaz}[2][]{\ensuremath{\subp{\check{\zeta}}{}{#2}{}{#1}}}
\newrobustcmd{\tildezetaz}[2][]{\ensuremath{\subp{\tilde{\zeta}}{}{#2}{}{#1}}}
\newrobustcmd{\widetildezetaz}[2][]{\ensuremath{\subp{\widetilde{\zeta}}{}{#2}{}{#1}}}
\newrobustcmd{\acutezetaz}[2][]{\ensuremath{\subp{\acute{\zeta}}{}{#2}{}{#1}}}
\newrobustcmd{\gravezetaz}[2][]{\ensuremath{\subp{\grave{\zeta}}{}{#2}{}{#1}}}
\newrobustcmd{\dotzetaz}[2][]{\ensuremath{\subp{\dot{\zeta}}{}{#2}{}{#1}}}
\newrobustcmd{\ddotzetaz}[2][]{\ensuremath{\subp{\ddot{\zeta}}{}{#2}{}{#1}}}
\newrobustcmd{\brevezetaz}[2][]{\ensuremath{\subp{\breve{\zeta}}{}{#2}{}{#1}}}
\newrobustcmd{\barzetaz}[2][]{\ensuremath{\subp{\bar{\zeta}}{}{#2}{}{#1}}}
\newrobustcmd{\veczetaz}[2][]{\ensuremath{\subp{\vec{\zeta}}{}{#2}{}{#1}}}
\newrobustcmd{\bmzetaz}[2][]{\ensuremath{\subp{\bm{\zeta}}{}{#2}{}{#1}}}
\newrobustcmd{\hatbmzetaz}[2][]{\ensuremath{\subp{\hat{\bm{\zeta}}}{}{#2}{}{#1}}}
\newrobustcmd{\widehatbmzetaz}[2][]{\ensuremath{\subp{\widehat{\bm{\zeta}}}{}{#2}{}{#1}}}
\newrobustcmd{\checkbmzetaz}[2][]{\ensuremath{\subp{\check{\bm{\zeta}}}{}{#2}{}{#1}}}
\newrobustcmd{\tildebmzetaz}[2][]{\ensuremath{\subp{\tilde{\bm{\zeta}}}{}{#2}{}{#1}}}
\newrobustcmd{\widetildebmzetaz}[2][]{\ensuremath{\subp{\widetilde{\bm{\zeta}}}{}{#2}{}{#1}}}
\newrobustcmd{\acutebmzetaz}[2][]{\ensuremath{\subp{\acute{\bm{\zeta}}}{}{#2}{}{#1}}}
\newrobustcmd{\gravebmzetaz}[2][]{\ensuremath{\subp{\grave{\bm{\zeta}}}{}{#2}{}{#1}}}
\newrobustcmd{\dotbmzetaz}[2][]{\ensuremath{\subp{\dot{\bm{\zeta}}}{}{#2}{}{#1}}}
\newrobustcmd{\ddotbmzetaz}[2][]{\ensuremath{\subp{\ddot{\bm{\zeta}}}{}{#2}{}{#1}}}
\newrobustcmd{\brevebmzetaz}[2][]{\ensuremath{\subp{\breve{\bm{\zeta}}}{}{#2}{}{#1}}}
\newrobustcmd{\barbmzetaz}[2][]{\ensuremath{\subp{\bar{\bm{\zeta}}}{}{#2}{}{#1}}}
\newrobustcmd{\vecbmzetaz}[2][]{\ensuremath{\subp{\vec{\bm{\zeta}}}{}{#2}{}{#1}}}
\newrobustcmd{\etaz}[2][]{\ensuremath{\subp{\eta}{}{#2}{}{#1}}}
\newrobustcmd{\hatetaz}[2][]{\ensuremath{\subp{\hat{\eta}}{}{#2}{}{#1}}}
\newrobustcmd{\widehatetaz}[2][]{\ensuremath{\subp{\widehat{\eta}}{}{#2}{}{#1}}}
\newrobustcmd{\checketaz}[2][]{\ensuremath{\subp{\check{\eta}}{}{#2}{}{#1}}}
\newrobustcmd{\tildeetaz}[2][]{\ensuremath{\subp{\tilde{\eta}}{}{#2}{}{#1}}}
\newrobustcmd{\widetildeetaz}[2][]{\ensuremath{\subp{\widetilde{\eta}}{}{#2}{}{#1}}}
\newrobustcmd{\acuteetaz}[2][]{\ensuremath{\subp{\acute{\eta}}{}{#2}{}{#1}}}
\newrobustcmd{\graveetaz}[2][]{\ensuremath{\subp{\grave{\eta}}{}{#2}{}{#1}}}
\newrobustcmd{\dotetaz}[2][]{\ensuremath{\subp{\dot{\eta}}{}{#2}{}{#1}}}
\newrobustcmd{\ddotetaz}[2][]{\ensuremath{\subp{\ddot{\eta}}{}{#2}{}{#1}}}
\newrobustcmd{\breveetaz}[2][]{\ensuremath{\subp{\breve{\eta}}{}{#2}{}{#1}}}
\newrobustcmd{\baretaz}[2][]{\ensuremath{\subp{\bar{\eta}}{}{#2}{}{#1}}}
\newrobustcmd{\vecetaz}[2][]{\ensuremath{\subp{\vec{\eta}}{}{#2}{}{#1}}}
\newrobustcmd{\bmetaz}[2][]{\ensuremath{\subp{\bm{\eta}}{}{#2}{}{#1}}}
\newrobustcmd{\hatbmetaz}[2][]{\ensuremath{\subp{\hat{\bm{\eta}}}{}{#2}{}{#1}}}
\newrobustcmd{\widehatbmetaz}[2][]{\ensuremath{\subp{\widehat{\bm{\eta}}}{}{#2}{}{#1}}}
\newrobustcmd{\checkbmetaz}[2][]{\ensuremath{\subp{\check{\bm{\eta}}}{}{#2}{}{#1}}}
\newrobustcmd{\tildebmetaz}[2][]{\ensuremath{\subp{\tilde{\bm{\eta}}}{}{#2}{}{#1}}}
\newrobustcmd{\widetildebmetaz}[2][]{\ensuremath{\subp{\widetilde{\bm{\eta}}}{}{#2}{}{#1}}}
\newrobustcmd{\acutebmetaz}[2][]{\ensuremath{\subp{\acute{\bm{\eta}}}{}{#2}{}{#1}}}
\newrobustcmd{\gravebmetaz}[2][]{\ensuremath{\subp{\grave{\bm{\eta}}}{}{#2}{}{#1}}}
\newrobustcmd{\dotbmetaz}[2][]{\ensuremath{\subp{\dot{\bm{\eta}}}{}{#2}{}{#1}}}
\newrobustcmd{\ddotbmetaz}[2][]{\ensuremath{\subp{\ddot{\bm{\eta}}}{}{#2}{}{#1}}}
\newrobustcmd{\brevebmetaz}[2][]{\ensuremath{\subp{\breve{\bm{\eta}}}{}{#2}{}{#1}}}
\newrobustcmd{\barbmetaz}[2][]{\ensuremath{\subp{\bar{\bm{\eta}}}{}{#2}{}{#1}}}
\newrobustcmd{\vecbmetaz}[2][]{\ensuremath{\subp{\vec{\bm{\eta}}}{}{#2}{}{#1}}}
\newrobustcmd{\thetaz}[2][]{\ensuremath{\subp{\theta}{}{#2}{}{#1}}}
\newrobustcmd{\hatthetaz}[2][]{\ensuremath{\subp{\hat{\theta}}{}{#2}{}{#1}}}
\newrobustcmd{\widehatthetaz}[2][]{\ensuremath{\subp{\widehat{\theta}}{}{#2}{}{#1}}}
\newrobustcmd{\checkthetaz}[2][]{\ensuremath{\subp{\check{\theta}}{}{#2}{}{#1}}}
\newrobustcmd{\tildethetaz}[2][]{\ensuremath{\subp{\tilde{\theta}}{}{#2}{}{#1}}}
\newrobustcmd{\widetildethetaz}[2][]{\ensuremath{\subp{\widetilde{\theta}}{}{#2}{}{#1}}}
\newrobustcmd{\acutethetaz}[2][]{\ensuremath{\subp{\acute{\theta}}{}{#2}{}{#1}}}
\newrobustcmd{\gravethetaz}[2][]{\ensuremath{\subp{\grave{\theta}}{}{#2}{}{#1}}}
\newrobustcmd{\dotthetaz}[2][]{\ensuremath{\subp{\dot{\theta}}{}{#2}{}{#1}}}
\newrobustcmd{\ddotthetaz}[2][]{\ensuremath{\subp{\ddot{\theta}}{}{#2}{}{#1}}}
\newrobustcmd{\brevethetaz}[2][]{\ensuremath{\subp{\breve{\theta}}{}{#2}{}{#1}}}
\newrobustcmd{\barthetaz}[2][]{\ensuremath{\subp{\bar{\theta}}{}{#2}{}{#1}}}
\newrobustcmd{\vecthetaz}[2][]{\ensuremath{\subp{\vec{\theta}}{}{#2}{}{#1}}}
\newrobustcmd{\bmthetaz}[2][]{\ensuremath{\subp{\bm{\theta}}{}{#2}{}{#1}}}
\newrobustcmd{\hatbmthetaz}[2][]{\ensuremath{\subp{\hat{\bm{\theta}}}{}{#2}{}{#1}}}
\newrobustcmd{\widehatbmthetaz}[2][]{\ensuremath{\subp{\widehat{\bm{\theta}}}{}{#2}{}{#1}}}
\newrobustcmd{\checkbmthetaz}[2][]{\ensuremath{\subp{\check{\bm{\theta}}}{}{#2}{}{#1}}}
\newrobustcmd{\tildebmthetaz}[2][]{\ensuremath{\subp{\tilde{\bm{\theta}}}{}{#2}{}{#1}}}
\newrobustcmd{\widetildebmthetaz}[2][]{\ensuremath{\subp{\widetilde{\bm{\theta}}}{}{#2}{}{#1}}}
\newrobustcmd{\acutebmthetaz}[2][]{\ensuremath{\subp{\acute{\bm{\theta}}}{}{#2}{}{#1}}}
\newrobustcmd{\gravebmthetaz}[2][]{\ensuremath{\subp{\grave{\bm{\theta}}}{}{#2}{}{#1}}}
\newrobustcmd{\dotbmthetaz}[2][]{\ensuremath{\subp{\dot{\bm{\theta}}}{}{#2}{}{#1}}}
\newrobustcmd{\ddotbmthetaz}[2][]{\ensuremath{\subp{\ddot{\bm{\theta}}}{}{#2}{}{#1}}}
\newrobustcmd{\brevebmthetaz}[2][]{\ensuremath{\subp{\breve{\bm{\theta}}}{}{#2}{}{#1}}}
\newrobustcmd{\barbmthetaz}[2][]{\ensuremath{\subp{\bar{\bm{\theta}}}{}{#2}{}{#1}}}
\newrobustcmd{\vecbmthetaz}[2][]{\ensuremath{\subp{\vec{\bm{\theta}}}{}{#2}{}{#1}}}
\newrobustcmd{\varthetaz}[2][]{\ensuremath{\subp{\vartheta}{}{#2}{}{#1}}}
\newrobustcmd{\hatvarthetaz}[2][]{\ensuremath{\subp{\hat{\vartheta}}{}{#2}{}{#1}}}
\newrobustcmd{\widehatvarthetaz}[2][]{\ensuremath{\subp{\widehat{\vartheta}}{}{#2}{}{#1}}}
\newrobustcmd{\checkvarthetaz}[2][]{\ensuremath{\subp{\check{\vartheta}}{}{#2}{}{#1}}}
\newrobustcmd{\tildevarthetaz}[2][]{\ensuremath{\subp{\tilde{\vartheta}}{}{#2}{}{#1}}}
\newrobustcmd{\widetildevarthetaz}[2][]{\ensuremath{\subp{\widetilde{\vartheta}}{}{#2}{}{#1}}}
\newrobustcmd{\acutevarthetaz}[2][]{\ensuremath{\subp{\acute{\vartheta}}{}{#2}{}{#1}}}
\newrobustcmd{\gravevarthetaz}[2][]{\ensuremath{\subp{\grave{\vartheta}}{}{#2}{}{#1}}}
\newrobustcmd{\dotvarthetaz}[2][]{\ensuremath{\subp{\dot{\vartheta}}{}{#2}{}{#1}}}
\newrobustcmd{\ddotvarthetaz}[2][]{\ensuremath{\subp{\ddot{\vartheta}}{}{#2}{}{#1}}}
\newrobustcmd{\brevevarthetaz}[2][]{\ensuremath{\subp{\breve{\vartheta}}{}{#2}{}{#1}}}
\newrobustcmd{\barvarthetaz}[2][]{\ensuremath{\subp{\bar{\vartheta}}{}{#2}{}{#1}}}
\newrobustcmd{\vecvarthetaz}[2][]{\ensuremath{\subp{\vec{\vartheta}}{}{#2}{}{#1}}}
\newrobustcmd{\bmvarthetaz}[2][]{\ensuremath{\subp{\bm{\vartheta}}{}{#2}{}{#1}}}
\newrobustcmd{\hatbmvarthetaz}[2][]{\ensuremath{\subp{\hat{\bm{\vartheta}}}{}{#2}{}{#1}}}
\newrobustcmd{\widehatbmvarthetaz}[2][]{\ensuremath{\subp{\widehat{\bm{\vartheta}}}{}{#2}{}{#1}}}
\newrobustcmd{\checkbmvarthetaz}[2][]{\ensuremath{\subp{\check{\bm{\vartheta}}}{}{#2}{}{#1}}}
\newrobustcmd{\tildebmvarthetaz}[2][]{\ensuremath{\subp{\tilde{\bm{\vartheta}}}{}{#2}{}{#1}}}
\newrobustcmd{\widetildebmvarthetaz}[2][]{\ensuremath{\subp{\widetilde{\bm{\vartheta}}}{}{#2}{}{#1}}}
\newrobustcmd{\acutebmvarthetaz}[2][]{\ensuremath{\subp{\acute{\bm{\vartheta}}}{}{#2}{}{#1}}}
\newrobustcmd{\gravebmvarthetaz}[2][]{\ensuremath{\subp{\grave{\bm{\vartheta}}}{}{#2}{}{#1}}}
\newrobustcmd{\dotbmvarthetaz}[2][]{\ensuremath{\subp{\dot{\bm{\vartheta}}}{}{#2}{}{#1}}}
\newrobustcmd{\ddotbmvarthetaz}[2][]{\ensuremath{\subp{\ddot{\bm{\vartheta}}}{}{#2}{}{#1}}}
\newrobustcmd{\brevebmvarthetaz}[2][]{\ensuremath{\subp{\breve{\bm{\vartheta}}}{}{#2}{}{#1}}}
\newrobustcmd{\barbmvarthetaz}[2][]{\ensuremath{\subp{\bar{\bm{\vartheta}}}{}{#2}{}{#1}}}
\newrobustcmd{\vecbmvarthetaz}[2][]{\ensuremath{\subp{\vec{\bm{\vartheta}}}{}{#2}{}{#1}}}
\newrobustcmd{\iotaz}[2][]{\ensuremath{\subp{\iota}{}{#2}{}{#1}}}
\newrobustcmd{\hatiotaz}[2][]{\ensuremath{\subp{\hat{\iota}}{}{#2}{}{#1}}}
\newrobustcmd{\widehatiotaz}[2][]{\ensuremath{\subp{\widehat{\iota}}{}{#2}{}{#1}}}
\newrobustcmd{\checkiotaz}[2][]{\ensuremath{\subp{\check{\iota}}{}{#2}{}{#1}}}
\newrobustcmd{\tildeiotaz}[2][]{\ensuremath{\subp{\tilde{\iota}}{}{#2}{}{#1}}}
\newrobustcmd{\widetildeiotaz}[2][]{\ensuremath{\subp{\widetilde{\iota}}{}{#2}{}{#1}}}
\newrobustcmd{\acuteiotaz}[2][]{\ensuremath{\subp{\acute{\iota}}{}{#2}{}{#1}}}
\newrobustcmd{\graveiotaz}[2][]{\ensuremath{\subp{\grave{\iota}}{}{#2}{}{#1}}}
\newrobustcmd{\dotiotaz}[2][]{\ensuremath{\subp{\dot{\iota}}{}{#2}{}{#1}}}
\newrobustcmd{\ddotiotaz}[2][]{\ensuremath{\subp{\ddot{\iota}}{}{#2}{}{#1}}}
\newrobustcmd{\breveiotaz}[2][]{\ensuremath{\subp{\breve{\iota}}{}{#2}{}{#1}}}
\newrobustcmd{\bariotaz}[2][]{\ensuremath{\subp{\bar{\iota}}{}{#2}{}{#1}}}
\newrobustcmd{\veciotaz}[2][]{\ensuremath{\subp{\vec{\iota}}{}{#2}{}{#1}}}
\newrobustcmd{\bmiotaz}[2][]{\ensuremath{\subp{\bm{\iota}}{}{#2}{}{#1}}}
\newrobustcmd{\hatbmiotaz}[2][]{\ensuremath{\subp{\hat{\bm{\iota}}}{}{#2}{}{#1}}}
\newrobustcmd{\widehatbmiotaz}[2][]{\ensuremath{\subp{\widehat{\bm{\iota}}}{}{#2}{}{#1}}}
\newrobustcmd{\checkbmiotaz}[2][]{\ensuremath{\subp{\check{\bm{\iota}}}{}{#2}{}{#1}}}
\newrobustcmd{\tildebmiotaz}[2][]{\ensuremath{\subp{\tilde{\bm{\iota}}}{}{#2}{}{#1}}}
\newrobustcmd{\widetildebmiotaz}[2][]{\ensuremath{\subp{\widetilde{\bm{\iota}}}{}{#2}{}{#1}}}
\newrobustcmd{\acutebmiotaz}[2][]{\ensuremath{\subp{\acute{\bm{\iota}}}{}{#2}{}{#1}}}
\newrobustcmd{\gravebmiotaz}[2][]{\ensuremath{\subp{\grave{\bm{\iota}}}{}{#2}{}{#1}}}
\newrobustcmd{\dotbmiotaz}[2][]{\ensuremath{\subp{\dot{\bm{\iota}}}{}{#2}{}{#1}}}
\newrobustcmd{\ddotbmiotaz}[2][]{\ensuremath{\subp{\ddot{\bm{\iota}}}{}{#2}{}{#1}}}
\newrobustcmd{\brevebmiotaz}[2][]{\ensuremath{\subp{\breve{\bm{\iota}}}{}{#2}{}{#1}}}
\newrobustcmd{\barbmiotaz}[2][]{\ensuremath{\subp{\bar{\bm{\iota}}}{}{#2}{}{#1}}}
\newrobustcmd{\vecbmiotaz}[2][]{\ensuremath{\subp{\vec{\bm{\iota}}}{}{#2}{}{#1}}}
\newrobustcmd{\kappaz}[2][]{\ensuremath{\subp{\kappa}{}{#2}{}{#1}}}
\newrobustcmd{\hatkappaz}[2][]{\ensuremath{\subp{\hat{\kappa}}{}{#2}{}{#1}}}
\newrobustcmd{\widehatkappaz}[2][]{\ensuremath{\subp{\widehat{\kappa}}{}{#2}{}{#1}}}
\newrobustcmd{\checkkappaz}[2][]{\ensuremath{\subp{\check{\kappa}}{}{#2}{}{#1}}}
\newrobustcmd{\tildekappaz}[2][]{\ensuremath{\subp{\tilde{\kappa}}{}{#2}{}{#1}}}
\newrobustcmd{\widetildekappaz}[2][]{\ensuremath{\subp{\widetilde{\kappa}}{}{#2}{}{#1}}}
\newrobustcmd{\acutekappaz}[2][]{\ensuremath{\subp{\acute{\kappa}}{}{#2}{}{#1}}}
\newrobustcmd{\gravekappaz}[2][]{\ensuremath{\subp{\grave{\kappa}}{}{#2}{}{#1}}}
\newrobustcmd{\dotkappaz}[2][]{\ensuremath{\subp{\dot{\kappa}}{}{#2}{}{#1}}}
\newrobustcmd{\ddotkappaz}[2][]{\ensuremath{\subp{\ddot{\kappa}}{}{#2}{}{#1}}}
\newrobustcmd{\brevekappaz}[2][]{\ensuremath{\subp{\breve{\kappa}}{}{#2}{}{#1}}}
\newrobustcmd{\barkappaz}[2][]{\ensuremath{\subp{\bar{\kappa}}{}{#2}{}{#1}}}
\newrobustcmd{\veckappaz}[2][]{\ensuremath{\subp{\vec{\kappa}}{}{#2}{}{#1}}}
\newrobustcmd{\bmkappaz}[2][]{\ensuremath{\subp{\bm{\kappa}}{}{#2}{}{#1}}}
\newrobustcmd{\hatbmkappaz}[2][]{\ensuremath{\subp{\hat{\bm{\kappa}}}{}{#2}{}{#1}}}
\newrobustcmd{\widehatbmkappaz}[2][]{\ensuremath{\subp{\widehat{\bm{\kappa}}}{}{#2}{}{#1}}}
\newrobustcmd{\checkbmkappaz}[2][]{\ensuremath{\subp{\check{\bm{\kappa}}}{}{#2}{}{#1}}}
\newrobustcmd{\tildebmkappaz}[2][]{\ensuremath{\subp{\tilde{\bm{\kappa}}}{}{#2}{}{#1}}}
\newrobustcmd{\widetildebmkappaz}[2][]{\ensuremath{\subp{\widetilde{\bm{\kappa}}}{}{#2}{}{#1}}}
\newrobustcmd{\acutebmkappaz}[2][]{\ensuremath{\subp{\acute{\bm{\kappa}}}{}{#2}{}{#1}}}
\newrobustcmd{\gravebmkappaz}[2][]{\ensuremath{\subp{\grave{\bm{\kappa}}}{}{#2}{}{#1}}}
\newrobustcmd{\dotbmkappaz}[2][]{\ensuremath{\subp{\dot{\bm{\kappa}}}{}{#2}{}{#1}}}
\newrobustcmd{\ddotbmkappaz}[2][]{\ensuremath{\subp{\ddot{\bm{\kappa}}}{}{#2}{}{#1}}}
\newrobustcmd{\brevebmkappaz}[2][]{\ensuremath{\subp{\breve{\bm{\kappa}}}{}{#2}{}{#1}}}
\newrobustcmd{\barbmkappaz}[2][]{\ensuremath{\subp{\bar{\bm{\kappa}}}{}{#2}{}{#1}}}
\newrobustcmd{\vecbmkappaz}[2][]{\ensuremath{\subp{\vec{\bm{\kappa}}}{}{#2}{}{#1}}}
\newrobustcmd{\varkappaz}[2][]{\ensuremath{\subp{\varkappa}{}{#2}{}{#1}}}
\newrobustcmd{\hatvarkappaz}[2][]{\ensuremath{\subp{\hat{\varkappa}}{}{#2}{}{#1}}}
\newrobustcmd{\widehatvarkappaz}[2][]{\ensuremath{\subp{\widehat{\varkappa}}{}{#2}{}{#1}}}
\newrobustcmd{\checkvarkappaz}[2][]{\ensuremath{\subp{\check{\varkappa}}{}{#2}{}{#1}}}
\newrobustcmd{\tildevarkappaz}[2][]{\ensuremath{\subp{\tilde{\varkappa}}{}{#2}{}{#1}}}
\newrobustcmd{\widetildevarkappaz}[2][]{\ensuremath{\subp{\widetilde{\varkappa}}{}{#2}{}{#1}}}
\newrobustcmd{\acutevarkappaz}[2][]{\ensuremath{\subp{\acute{\varkappa}}{}{#2}{}{#1}}}
\newrobustcmd{\gravevarkappaz}[2][]{\ensuremath{\subp{\grave{\varkappa}}{}{#2}{}{#1}}}
\newrobustcmd{\dotvarkappaz}[2][]{\ensuremath{\subp{\dot{\varkappa}}{}{#2}{}{#1}}}
\newrobustcmd{\ddotvarkappaz}[2][]{\ensuremath{\subp{\ddot{\varkappa}}{}{#2}{}{#1}}}
\newrobustcmd{\brevevarkappaz}[2][]{\ensuremath{\subp{\breve{\varkappa}}{}{#2}{}{#1}}}
\newrobustcmd{\barvarkappaz}[2][]{\ensuremath{\subp{\bar{\varkappa}}{}{#2}{}{#1}}}
\newrobustcmd{\vecvarkappaz}[2][]{\ensuremath{\subp{\vec{\varkappa}}{}{#2}{}{#1}}}
\newrobustcmd{\bmvarkappaz}[2][]{\ensuremath{\subp{\bm{\varkappa}}{}{#2}{}{#1}}}
\newrobustcmd{\hatbmvarkappaz}[2][]{\ensuremath{\subp{\hat{\bm{\varkappa}}}{}{#2}{}{#1}}}
\newrobustcmd{\widehatbmvarkappaz}[2][]{\ensuremath{\subp{\widehat{\bm{\varkappa}}}{}{#2}{}{#1}}}
\newrobustcmd{\checkbmvarkappaz}[2][]{\ensuremath{\subp{\check{\bm{\varkappa}}}{}{#2}{}{#1}}}
\newrobustcmd{\tildebmvarkappaz}[2][]{\ensuremath{\subp{\tilde{\bm{\varkappa}}}{}{#2}{}{#1}}}
\newrobustcmd{\widetildebmvarkappaz}[2][]{\ensuremath{\subp{\widetilde{\bm{\varkappa}}}{}{#2}{}{#1}}}
\newrobustcmd{\acutebmvarkappaz}[2][]{\ensuremath{\subp{\acute{\bm{\varkappa}}}{}{#2}{}{#1}}}
\newrobustcmd{\gravebmvarkappaz}[2][]{\ensuremath{\subp{\grave{\bm{\varkappa}}}{}{#2}{}{#1}}}
\newrobustcmd{\dotbmvarkappaz}[2][]{\ensuremath{\subp{\dot{\bm{\varkappa}}}{}{#2}{}{#1}}}
\newrobustcmd{\ddotbmvarkappaz}[2][]{\ensuremath{\subp{\ddot{\bm{\varkappa}}}{}{#2}{}{#1}}}
\newrobustcmd{\brevebmvarkappaz}[2][]{\ensuremath{\subp{\breve{\bm{\varkappa}}}{}{#2}{}{#1}}}
\newrobustcmd{\barbmvarkappaz}[2][]{\ensuremath{\subp{\bar{\bm{\varkappa}}}{}{#2}{}{#1}}}
\newrobustcmd{\vecbmvarkappaz}[2][]{\ensuremath{\subp{\vec{\bm{\varkappa}}}{}{#2}{}{#1}}}
\newrobustcmd{\lambdaz}[2][]{\ensuremath{\subp{\lambda}{}{#2}{}{#1}}}
\newrobustcmd{\hatlambdaz}[2][]{\ensuremath{\subp{\hat{\lambda}}{}{#2}{}{#1}}}
\newrobustcmd{\widehatlambdaz}[2][]{\ensuremath{\subp{\widehat{\lambda}}{}{#2}{}{#1}}}
\newrobustcmd{\checklambdaz}[2][]{\ensuremath{\subp{\check{\lambda}}{}{#2}{}{#1}}}
\newrobustcmd{\tildelambdaz}[2][]{\ensuremath{\subp{\tilde{\lambda}}{}{#2}{}{#1}}}
\newrobustcmd{\widetildelambdaz}[2][]{\ensuremath{\subp{\widetilde{\lambda}}{}{#2}{}{#1}}}
\newrobustcmd{\acutelambdaz}[2][]{\ensuremath{\subp{\acute{\lambda}}{}{#2}{}{#1}}}
\newrobustcmd{\gravelambdaz}[2][]{\ensuremath{\subp{\grave{\lambda}}{}{#2}{}{#1}}}
\newrobustcmd{\dotlambdaz}[2][]{\ensuremath{\subp{\dot{\lambda}}{}{#2}{}{#1}}}
\newrobustcmd{\ddotlambdaz}[2][]{\ensuremath{\subp{\ddot{\lambda}}{}{#2}{}{#1}}}
\newrobustcmd{\brevelambdaz}[2][]{\ensuremath{\subp{\breve{\lambda}}{}{#2}{}{#1}}}
\newrobustcmd{\barlambdaz}[2][]{\ensuremath{\subp{\bar{\lambda}}{}{#2}{}{#1}}}
\newrobustcmd{\veclambdaz}[2][]{\ensuremath{\subp{\vec{\lambda}}{}{#2}{}{#1}}}
\newrobustcmd{\bmlambdaz}[2][]{\ensuremath{\subp{\bm{\lambda}}{}{#2}{}{#1}}}
\newrobustcmd{\hatbmlambdaz}[2][]{\ensuremath{\subp{\hat{\bm{\lambda}}}{}{#2}{}{#1}}}
\newrobustcmd{\widehatbmlambdaz}[2][]{\ensuremath{\subp{\widehat{\bm{\lambda}}}{}{#2}{}{#1}}}
\newrobustcmd{\checkbmlambdaz}[2][]{\ensuremath{\subp{\check{\bm{\lambda}}}{}{#2}{}{#1}}}
\newrobustcmd{\tildebmlambdaz}[2][]{\ensuremath{\subp{\tilde{\bm{\lambda}}}{}{#2}{}{#1}}}
\newrobustcmd{\widetildebmlambdaz}[2][]{\ensuremath{\subp{\widetilde{\bm{\lambda}}}{}{#2}{}{#1}}}
\newrobustcmd{\acutebmlambdaz}[2][]{\ensuremath{\subp{\acute{\bm{\lambda}}}{}{#2}{}{#1}}}
\newrobustcmd{\gravebmlambdaz}[2][]{\ensuremath{\subp{\grave{\bm{\lambda}}}{}{#2}{}{#1}}}
\newrobustcmd{\dotbmlambdaz}[2][]{\ensuremath{\subp{\dot{\bm{\lambda}}}{}{#2}{}{#1}}}
\newrobustcmd{\ddotbmlambdaz}[2][]{\ensuremath{\subp{\ddot{\bm{\lambda}}}{}{#2}{}{#1}}}
\newrobustcmd{\brevebmlambdaz}[2][]{\ensuremath{\subp{\breve{\bm{\lambda}}}{}{#2}{}{#1}}}
\newrobustcmd{\barbmlambdaz}[2][]{\ensuremath{\subp{\bar{\bm{\lambda}}}{}{#2}{}{#1}}}
\newrobustcmd{\vecbmlambdaz}[2][]{\ensuremath{\subp{\vec{\bm{\lambda}}}{}{#2}{}{#1}}}
\newrobustcmd{\muz}[2][]{\ensuremath{\subp{\mu}{}{#2}{}{#1}}}
\newrobustcmd{\hatmuz}[2][]{\ensuremath{\subp{\hat{\mu}}{}{#2}{}{#1}}}
\newrobustcmd{\widehatmuz}[2][]{\ensuremath{\subp{\widehat{\mu}}{}{#2}{}{#1}}}
\newrobustcmd{\checkmuz}[2][]{\ensuremath{\subp{\check{\mu}}{}{#2}{}{#1}}}
\newrobustcmd{\tildemuz}[2][]{\ensuremath{\subp{\tilde{\mu}}{}{#2}{}{#1}}}
\newrobustcmd{\widetildemuz}[2][]{\ensuremath{\subp{\widetilde{\mu}}{}{#2}{}{#1}}}
\newrobustcmd{\acutemuz}[2][]{\ensuremath{\subp{\acute{\mu}}{}{#2}{}{#1}}}
\newrobustcmd{\gravemuz}[2][]{\ensuremath{\subp{\grave{\mu}}{}{#2}{}{#1}}}
\newrobustcmd{\dotmuz}[2][]{\ensuremath{\subp{\dot{\mu}}{}{#2}{}{#1}}}
\newrobustcmd{\ddotmuz}[2][]{\ensuremath{\subp{\ddot{\mu}}{}{#2}{}{#1}}}
\newrobustcmd{\brevemuz}[2][]{\ensuremath{\subp{\breve{\mu}}{}{#2}{}{#1}}}
\newrobustcmd{\barmuz}[2][]{\ensuremath{\subp{\bar{\mu}}{}{#2}{}{#1}}}
\newrobustcmd{\vecmuz}[2][]{\ensuremath{\subp{\vec{\mu}}{}{#2}{}{#1}}}
\newrobustcmd{\bmmuz}[2][]{\ensuremath{\subp{\bm{\mu}}{}{#2}{}{#1}}}
\newrobustcmd{\hatbmmuz}[2][]{\ensuremath{\subp{\hat{\bm{\mu}}}{}{#2}{}{#1}}}
\newrobustcmd{\widehatbmmuz}[2][]{\ensuremath{\subp{\widehat{\bm{\mu}}}{}{#2}{}{#1}}}
\newrobustcmd{\checkbmmuz}[2][]{\ensuremath{\subp{\check{\bm{\mu}}}{}{#2}{}{#1}}}
\newrobustcmd{\tildebmmuz}[2][]{\ensuremath{\subp{\tilde{\bm{\mu}}}{}{#2}{}{#1}}}
\newrobustcmd{\widetildebmmuz}[2][]{\ensuremath{\subp{\widetilde{\bm{\mu}}}{}{#2}{}{#1}}}
\newrobustcmd{\acutebmmuz}[2][]{\ensuremath{\subp{\acute{\bm{\mu}}}{}{#2}{}{#1}}}
\newrobustcmd{\gravebmmuz}[2][]{\ensuremath{\subp{\grave{\bm{\mu}}}{}{#2}{}{#1}}}
\newrobustcmd{\dotbmmuz}[2][]{\ensuremath{\subp{\dot{\bm{\mu}}}{}{#2}{}{#1}}}
\newrobustcmd{\ddotbmmuz}[2][]{\ensuremath{\subp{\ddot{\bm{\mu}}}{}{#2}{}{#1}}}
\newrobustcmd{\brevebmmuz}[2][]{\ensuremath{\subp{\breve{\bm{\mu}}}{}{#2}{}{#1}}}
\newrobustcmd{\barbmmuz}[2][]{\ensuremath{\subp{\bar{\bm{\mu}}}{}{#2}{}{#1}}}
\newrobustcmd{\vecbmmuz}[2][]{\ensuremath{\subp{\vec{\bm{\mu}}}{}{#2}{}{#1}}}
\newrobustcmd{\nuz}[2][]{\ensuremath{\subp{\nu}{}{#2}{}{#1}}}
\newrobustcmd{\hatnuz}[2][]{\ensuremath{\subp{\hat{\nu}}{}{#2}{}{#1}}}
\newrobustcmd{\widehatnuz}[2][]{\ensuremath{\subp{\widehat{\nu}}{}{#2}{}{#1}}}
\newrobustcmd{\checknuz}[2][]{\ensuremath{\subp{\check{\nu}}{}{#2}{}{#1}}}
\newrobustcmd{\tildenuz}[2][]{\ensuremath{\subp{\tilde{\nu}}{}{#2}{}{#1}}}
\newrobustcmd{\widetildenuz}[2][]{\ensuremath{\subp{\widetilde{\nu}}{}{#2}{}{#1}}}
\newrobustcmd{\acutenuz}[2][]{\ensuremath{\subp{\acute{\nu}}{}{#2}{}{#1}}}
\newrobustcmd{\gravenuz}[2][]{\ensuremath{\subp{\grave{\nu}}{}{#2}{}{#1}}}
\newrobustcmd{\dotnuz}[2][]{\ensuremath{\subp{\dot{\nu}}{}{#2}{}{#1}}}
\newrobustcmd{\ddotnuz}[2][]{\ensuremath{\subp{\ddot{\nu}}{}{#2}{}{#1}}}
\newrobustcmd{\brevenuz}[2][]{\ensuremath{\subp{\breve{\nu}}{}{#2}{}{#1}}}
\newrobustcmd{\barnuz}[2][]{\ensuremath{\subp{\bar{\nu}}{}{#2}{}{#1}}}
\newrobustcmd{\vecnuz}[2][]{\ensuremath{\subp{\vec{\nu}}{}{#2}{}{#1}}}
\newrobustcmd{\bmnuz}[2][]{\ensuremath{\subp{\bm{\nu}}{}{#2}{}{#1}}}
\newrobustcmd{\hatbmnuz}[2][]{\ensuremath{\subp{\hat{\bm{\nu}}}{}{#2}{}{#1}}}
\newrobustcmd{\widehatbmnuz}[2][]{\ensuremath{\subp{\widehat{\bm{\nu}}}{}{#2}{}{#1}}}
\newrobustcmd{\checkbmnuz}[2][]{\ensuremath{\subp{\check{\bm{\nu}}}{}{#2}{}{#1}}}
\newrobustcmd{\tildebmnuz}[2][]{\ensuremath{\subp{\tilde{\bm{\nu}}}{}{#2}{}{#1}}}
\newrobustcmd{\widetildebmnuz}[2][]{\ensuremath{\subp{\widetilde{\bm{\nu}}}{}{#2}{}{#1}}}
\newrobustcmd{\acutebmnuz}[2][]{\ensuremath{\subp{\acute{\bm{\nu}}}{}{#2}{}{#1}}}
\newrobustcmd{\gravebmnuz}[2][]{\ensuremath{\subp{\grave{\bm{\nu}}}{}{#2}{}{#1}}}
\newrobustcmd{\dotbmnuz}[2][]{\ensuremath{\subp{\dot{\bm{\nu}}}{}{#2}{}{#1}}}
\newrobustcmd{\ddotbmnuz}[2][]{\ensuremath{\subp{\ddot{\bm{\nu}}}{}{#2}{}{#1}}}
\newrobustcmd{\brevebmnuz}[2][]{\ensuremath{\subp{\breve{\bm{\nu}}}{}{#2}{}{#1}}}
\newrobustcmd{\barbmnuz}[2][]{\ensuremath{\subp{\bar{\bm{\nu}}}{}{#2}{}{#1}}}
\newrobustcmd{\vecbmnuz}[2][]{\ensuremath{\subp{\vec{\bm{\nu}}}{}{#2}{}{#1}}}
\newrobustcmd{\xiz}[2][]{\ensuremath{\subp{\xi}{}{#2}{}{#1}}}
\newrobustcmd{\hatxiz}[2][]{\ensuremath{\subp{\hat{\xi}}{}{#2}{}{#1}}}
\newrobustcmd{\widehatxiz}[2][]{\ensuremath{\subp{\widehat{\xi}}{}{#2}{}{#1}}}
\newrobustcmd{\checkxiz}[2][]{\ensuremath{\subp{\check{\xi}}{}{#2}{}{#1}}}
\newrobustcmd{\tildexiz}[2][]{\ensuremath{\subp{\tilde{\xi}}{}{#2}{}{#1}}}
\newrobustcmd{\widetildexiz}[2][]{\ensuremath{\subp{\widetilde{\xi}}{}{#2}{}{#1}}}
\newrobustcmd{\acutexiz}[2][]{\ensuremath{\subp{\acute{\xi}}{}{#2}{}{#1}}}
\newrobustcmd{\gravexiz}[2][]{\ensuremath{\subp{\grave{\xi}}{}{#2}{}{#1}}}
\newrobustcmd{\dotxiz}[2][]{\ensuremath{\subp{\dot{\xi}}{}{#2}{}{#1}}}
\newrobustcmd{\ddotxiz}[2][]{\ensuremath{\subp{\ddot{\xi}}{}{#2}{}{#1}}}
\newrobustcmd{\brevexiz}[2][]{\ensuremath{\subp{\breve{\xi}}{}{#2}{}{#1}}}
\newrobustcmd{\barxiz}[2][]{\ensuremath{\subp{\bar{\xi}}{}{#2}{}{#1}}}
\newrobustcmd{\vecxiz}[2][]{\ensuremath{\subp{\vec{\xi}}{}{#2}{}{#1}}}
\newrobustcmd{\bmxiz}[2][]{\ensuremath{\subp{\bm{\xi}}{}{#2}{}{#1}}}
\newrobustcmd{\hatbmxiz}[2][]{\ensuremath{\subp{\hat{\bm{\xi}}}{}{#2}{}{#1}}}
\newrobustcmd{\widehatbmxiz}[2][]{\ensuremath{\subp{\widehat{\bm{\xi}}}{}{#2}{}{#1}}}
\newrobustcmd{\checkbmxiz}[2][]{\ensuremath{\subp{\check{\bm{\xi}}}{}{#2}{}{#1}}}
\newrobustcmd{\tildebmxiz}[2][]{\ensuremath{\subp{\tilde{\bm{\xi}}}{}{#2}{}{#1}}}
\newrobustcmd{\widetildebmxiz}[2][]{\ensuremath{\subp{\widetilde{\bm{\xi}}}{}{#2}{}{#1}}}
\newrobustcmd{\acutebmxiz}[2][]{\ensuremath{\subp{\acute{\bm{\xi}}}{}{#2}{}{#1}}}
\newrobustcmd{\gravebmxiz}[2][]{\ensuremath{\subp{\grave{\bm{\xi}}}{}{#2}{}{#1}}}
\newrobustcmd{\dotbmxiz}[2][]{\ensuremath{\subp{\dot{\bm{\xi}}}{}{#2}{}{#1}}}
\newrobustcmd{\ddotbmxiz}[2][]{\ensuremath{\subp{\ddot{\bm{\xi}}}{}{#2}{}{#1}}}
\newrobustcmd{\brevebmxiz}[2][]{\ensuremath{\subp{\breve{\bm{\xi}}}{}{#2}{}{#1}}}
\newrobustcmd{\barbmxiz}[2][]{\ensuremath{\subp{\bar{\bm{\xi}}}{}{#2}{}{#1}}}
\newrobustcmd{\vecbmxiz}[2][]{\ensuremath{\subp{\vec{\bm{\xi}}}{}{#2}{}{#1}}}
\newrobustcmd{\piz}[2][]{\ensuremath{\subp{\pi}{}{#2}{}{#1}}}
\newrobustcmd{\hatpiz}[2][]{\ensuremath{\subp{\hat{\pi}}{}{#2}{}{#1}}}
\newrobustcmd{\widehatpiz}[2][]{\ensuremath{\subp{\widehat{\pi}}{}{#2}{}{#1}}}
\newrobustcmd{\checkpiz}[2][]{\ensuremath{\subp{\check{\pi}}{}{#2}{}{#1}}}
\newrobustcmd{\tildepiz}[2][]{\ensuremath{\subp{\tilde{\pi}}{}{#2}{}{#1}}}
\newrobustcmd{\widetildepiz}[2][]{\ensuremath{\subp{\widetilde{\pi}}{}{#2}{}{#1}}}
\newrobustcmd{\acutepiz}[2][]{\ensuremath{\subp{\acute{\pi}}{}{#2}{}{#1}}}
\newrobustcmd{\gravepiz}[2][]{\ensuremath{\subp{\grave{\pi}}{}{#2}{}{#1}}}
\newrobustcmd{\dotpiz}[2][]{\ensuremath{\subp{\dot{\pi}}{}{#2}{}{#1}}}
\newrobustcmd{\ddotpiz}[2][]{\ensuremath{\subp{\ddot{\pi}}{}{#2}{}{#1}}}
\newrobustcmd{\brevepiz}[2][]{\ensuremath{\subp{\breve{\pi}}{}{#2}{}{#1}}}
\newrobustcmd{\barpiz}[2][]{\ensuremath{\subp{\bar{\pi}}{}{#2}{}{#1}}}
\newrobustcmd{\vecpiz}[2][]{\ensuremath{\subp{\vec{\pi}}{}{#2}{}{#1}}}
\newrobustcmd{\bmpiz}[2][]{\ensuremath{\subp{\bm{\pi}}{}{#2}{}{#1}}}
\newrobustcmd{\hatbmpiz}[2][]{\ensuremath{\subp{\hat{\bm{\pi}}}{}{#2}{}{#1}}}
\newrobustcmd{\widehatbmpiz}[2][]{\ensuremath{\subp{\widehat{\bm{\pi}}}{}{#2}{}{#1}}}
\newrobustcmd{\checkbmpiz}[2][]{\ensuremath{\subp{\check{\bm{\pi}}}{}{#2}{}{#1}}}
\newrobustcmd{\tildebmpiz}[2][]{\ensuremath{\subp{\tilde{\bm{\pi}}}{}{#2}{}{#1}}}
\newrobustcmd{\widetildebmpiz}[2][]{\ensuremath{\subp{\widetilde{\bm{\pi}}}{}{#2}{}{#1}}}
\newrobustcmd{\acutebmpiz}[2][]{\ensuremath{\subp{\acute{\bm{\pi}}}{}{#2}{}{#1}}}
\newrobustcmd{\gravebmpiz}[2][]{\ensuremath{\subp{\grave{\bm{\pi}}}{}{#2}{}{#1}}}
\newrobustcmd{\dotbmpiz}[2][]{\ensuremath{\subp{\dot{\bm{\pi}}}{}{#2}{}{#1}}}
\newrobustcmd{\ddotbmpiz}[2][]{\ensuremath{\subp{\ddot{\bm{\pi}}}{}{#2}{}{#1}}}
\newrobustcmd{\brevebmpiz}[2][]{\ensuremath{\subp{\breve{\bm{\pi}}}{}{#2}{}{#1}}}
\newrobustcmd{\barbmpiz}[2][]{\ensuremath{\subp{\bar{\bm{\pi}}}{}{#2}{}{#1}}}
\newrobustcmd{\vecbmpiz}[2][]{\ensuremath{\subp{\vec{\bm{\pi}}}{}{#2}{}{#1}}}
\newrobustcmd{\varpiz}[2][]{\ensuremath{\subp{\varpi}{}{#2}{}{#1}}}
\newrobustcmd{\hatvarpiz}[2][]{\ensuremath{\subp{\hat{\varpi}}{}{#2}{}{#1}}}
\newrobustcmd{\widehatvarpiz}[2][]{\ensuremath{\subp{\widehat{\varpi}}{}{#2}{}{#1}}}
\newrobustcmd{\checkvarpiz}[2][]{\ensuremath{\subp{\check{\varpi}}{}{#2}{}{#1}}}
\newrobustcmd{\tildevarpiz}[2][]{\ensuremath{\subp{\tilde{\varpi}}{}{#2}{}{#1}}}
\newrobustcmd{\widetildevarpiz}[2][]{\ensuremath{\subp{\widetilde{\varpi}}{}{#2}{}{#1}}}
\newrobustcmd{\acutevarpiz}[2][]{\ensuremath{\subp{\acute{\varpi}}{}{#2}{}{#1}}}
\newrobustcmd{\gravevarpiz}[2][]{\ensuremath{\subp{\grave{\varpi}}{}{#2}{}{#1}}}
\newrobustcmd{\dotvarpiz}[2][]{\ensuremath{\subp{\dot{\varpi}}{}{#2}{}{#1}}}
\newrobustcmd{\ddotvarpiz}[2][]{\ensuremath{\subp{\ddot{\varpi}}{}{#2}{}{#1}}}
\newrobustcmd{\brevevarpiz}[2][]{\ensuremath{\subp{\breve{\varpi}}{}{#2}{}{#1}}}
\newrobustcmd{\barvarpiz}[2][]{\ensuremath{\subp{\bar{\varpi}}{}{#2}{}{#1}}}
\newrobustcmd{\vecvarpiz}[2][]{\ensuremath{\subp{\vec{\varpi}}{}{#2}{}{#1}}}
\newrobustcmd{\bmvarpiz}[2][]{\ensuremath{\subp{\bm{\varpi}}{}{#2}{}{#1}}}
\newrobustcmd{\hatbmvarpiz}[2][]{\ensuremath{\subp{\hat{\bm{\varpi}}}{}{#2}{}{#1}}}
\newrobustcmd{\widehatbmvarpiz}[2][]{\ensuremath{\subp{\widehat{\bm{\varpi}}}{}{#2}{}{#1}}}
\newrobustcmd{\checkbmvarpiz}[2][]{\ensuremath{\subp{\check{\bm{\varpi}}}{}{#2}{}{#1}}}
\newrobustcmd{\tildebmvarpiz}[2][]{\ensuremath{\subp{\tilde{\bm{\varpi}}}{}{#2}{}{#1}}}
\newrobustcmd{\widetildebmvarpiz}[2][]{\ensuremath{\subp{\widetilde{\bm{\varpi}}}{}{#2}{}{#1}}}
\newrobustcmd{\acutebmvarpiz}[2][]{\ensuremath{\subp{\acute{\bm{\varpi}}}{}{#2}{}{#1}}}
\newrobustcmd{\gravebmvarpiz}[2][]{\ensuremath{\subp{\grave{\bm{\varpi}}}{}{#2}{}{#1}}}
\newrobustcmd{\dotbmvarpiz}[2][]{\ensuremath{\subp{\dot{\bm{\varpi}}}{}{#2}{}{#1}}}
\newrobustcmd{\ddotbmvarpiz}[2][]{\ensuremath{\subp{\ddot{\bm{\varpi}}}{}{#2}{}{#1}}}
\newrobustcmd{\brevebmvarpiz}[2][]{\ensuremath{\subp{\breve{\bm{\varpi}}}{}{#2}{}{#1}}}
\newrobustcmd{\barbmvarpiz}[2][]{\ensuremath{\subp{\bar{\bm{\varpi}}}{}{#2}{}{#1}}}
\newrobustcmd{\vecbmvarpiz}[2][]{\ensuremath{\subp{\vec{\bm{\varpi}}}{}{#2}{}{#1}}}
\newrobustcmd{\rhoz}[2][]{\ensuremath{\subp{\rho}{}{#2}{}{#1}}}
\newrobustcmd{\hatrhoz}[2][]{\ensuremath{\subp{\hat{\rho}}{}{#2}{}{#1}}}
\newrobustcmd{\widehatrhoz}[2][]{\ensuremath{\subp{\widehat{\rho}}{}{#2}{}{#1}}}
\newrobustcmd{\checkrhoz}[2][]{\ensuremath{\subp{\check{\rho}}{}{#2}{}{#1}}}
\newrobustcmd{\tilderhoz}[2][]{\ensuremath{\subp{\tilde{\rho}}{}{#2}{}{#1}}}
\newrobustcmd{\widetilderhoz}[2][]{\ensuremath{\subp{\widetilde{\rho}}{}{#2}{}{#1}}}
\newrobustcmd{\acuterhoz}[2][]{\ensuremath{\subp{\acute{\rho}}{}{#2}{}{#1}}}
\newrobustcmd{\graverhoz}[2][]{\ensuremath{\subp{\grave{\rho}}{}{#2}{}{#1}}}
\newrobustcmd{\dotrhoz}[2][]{\ensuremath{\subp{\dot{\rho}}{}{#2}{}{#1}}}
\newrobustcmd{\ddotrhoz}[2][]{\ensuremath{\subp{\ddot{\rho}}{}{#2}{}{#1}}}
\newrobustcmd{\breverhoz}[2][]{\ensuremath{\subp{\breve{\rho}}{}{#2}{}{#1}}}
\newrobustcmd{\barrhoz}[2][]{\ensuremath{\subp{\bar{\rho}}{}{#2}{}{#1}}}
\newrobustcmd{\vecrhoz}[2][]{\ensuremath{\subp{\vec{\rho}}{}{#2}{}{#1}}}
\newrobustcmd{\bmrhoz}[2][]{\ensuremath{\subp{\bm{\rho}}{}{#2}{}{#1}}}
\newrobustcmd{\hatbmrhoz}[2][]{\ensuremath{\subp{\hat{\bm{\rho}}}{}{#2}{}{#1}}}
\newrobustcmd{\widehatbmrhoz}[2][]{\ensuremath{\subp{\widehat{\bm{\rho}}}{}{#2}{}{#1}}}
\newrobustcmd{\checkbmrhoz}[2][]{\ensuremath{\subp{\check{\bm{\rho}}}{}{#2}{}{#1}}}
\newrobustcmd{\tildebmrhoz}[2][]{\ensuremath{\subp{\tilde{\bm{\rho}}}{}{#2}{}{#1}}}
\newrobustcmd{\widetildebmrhoz}[2][]{\ensuremath{\subp{\widetilde{\bm{\rho}}}{}{#2}{}{#1}}}
\newrobustcmd{\acutebmrhoz}[2][]{\ensuremath{\subp{\acute{\bm{\rho}}}{}{#2}{}{#1}}}
\newrobustcmd{\gravebmrhoz}[2][]{\ensuremath{\subp{\grave{\bm{\rho}}}{}{#2}{}{#1}}}
\newrobustcmd{\dotbmrhoz}[2][]{\ensuremath{\subp{\dot{\bm{\rho}}}{}{#2}{}{#1}}}
\newrobustcmd{\ddotbmrhoz}[2][]{\ensuremath{\subp{\ddot{\bm{\rho}}}{}{#2}{}{#1}}}
\newrobustcmd{\brevebmrhoz}[2][]{\ensuremath{\subp{\breve{\bm{\rho}}}{}{#2}{}{#1}}}
\newrobustcmd{\barbmrhoz}[2][]{\ensuremath{\subp{\bar{\bm{\rho}}}{}{#2}{}{#1}}}
\newrobustcmd{\vecbmrhoz}[2][]{\ensuremath{\subp{\vec{\bm{\rho}}}{}{#2}{}{#1}}}
\newrobustcmd{\varrhoz}[2][]{\ensuremath{\subp{\varrho}{}{#2}{}{#1}}}
\newrobustcmd{\hatvarrhoz}[2][]{\ensuremath{\subp{\hat{\varrho}}{}{#2}{}{#1}}}
\newrobustcmd{\widehatvarrhoz}[2][]{\ensuremath{\subp{\widehat{\varrho}}{}{#2}{}{#1}}}
\newrobustcmd{\checkvarrhoz}[2][]{\ensuremath{\subp{\check{\varrho}}{}{#2}{}{#1}}}
\newrobustcmd{\tildevarrhoz}[2][]{\ensuremath{\subp{\tilde{\varrho}}{}{#2}{}{#1}}}
\newrobustcmd{\widetildevarrhoz}[2][]{\ensuremath{\subp{\widetilde{\varrho}}{}{#2}{}{#1}}}
\newrobustcmd{\acutevarrhoz}[2][]{\ensuremath{\subp{\acute{\varrho}}{}{#2}{}{#1}}}
\newrobustcmd{\gravevarrhoz}[2][]{\ensuremath{\subp{\grave{\varrho}}{}{#2}{}{#1}}}
\newrobustcmd{\dotvarrhoz}[2][]{\ensuremath{\subp{\dot{\varrho}}{}{#2}{}{#1}}}
\newrobustcmd{\ddotvarrhoz}[2][]{\ensuremath{\subp{\ddot{\varrho}}{}{#2}{}{#1}}}
\newrobustcmd{\brevevarrhoz}[2][]{\ensuremath{\subp{\breve{\varrho}}{}{#2}{}{#1}}}
\newrobustcmd{\barvarrhoz}[2][]{\ensuremath{\subp{\bar{\varrho}}{}{#2}{}{#1}}}
\newrobustcmd{\vecvarrhoz}[2][]{\ensuremath{\subp{\vec{\varrho}}{}{#2}{}{#1}}}
\newrobustcmd{\bmvarrhoz}[2][]{\ensuremath{\subp{\bm{\varrho}}{}{#2}{}{#1}}}
\newrobustcmd{\hatbmvarrhoz}[2][]{\ensuremath{\subp{\hat{\bm{\varrho}}}{}{#2}{}{#1}}}
\newrobustcmd{\widehatbmvarrhoz}[2][]{\ensuremath{\subp{\widehat{\bm{\varrho}}}{}{#2}{}{#1}}}
\newrobustcmd{\checkbmvarrhoz}[2][]{\ensuremath{\subp{\check{\bm{\varrho}}}{}{#2}{}{#1}}}
\newrobustcmd{\tildebmvarrhoz}[2][]{\ensuremath{\subp{\tilde{\bm{\varrho}}}{}{#2}{}{#1}}}
\newrobustcmd{\widetildebmvarrhoz}[2][]{\ensuremath{\subp{\widetilde{\bm{\varrho}}}{}{#2}{}{#1}}}
\newrobustcmd{\acutebmvarrhoz}[2][]{\ensuremath{\subp{\acute{\bm{\varrho}}}{}{#2}{}{#1}}}
\newrobustcmd{\gravebmvarrhoz}[2][]{\ensuremath{\subp{\grave{\bm{\varrho}}}{}{#2}{}{#1}}}
\newrobustcmd{\dotbmvarrhoz}[2][]{\ensuremath{\subp{\dot{\bm{\varrho}}}{}{#2}{}{#1}}}
\newrobustcmd{\ddotbmvarrhoz}[2][]{\ensuremath{\subp{\ddot{\bm{\varrho}}}{}{#2}{}{#1}}}
\newrobustcmd{\brevebmvarrhoz}[2][]{\ensuremath{\subp{\breve{\bm{\varrho}}}{}{#2}{}{#1}}}
\newrobustcmd{\barbmvarrhoz}[2][]{\ensuremath{\subp{\bar{\bm{\varrho}}}{}{#2}{}{#1}}}
\newrobustcmd{\vecbmvarrhoz}[2][]{\ensuremath{\subp{\vec{\bm{\varrho}}}{}{#2}{}{#1}}}
\newrobustcmd{\sigmaz}[2][]{\ensuremath{\subp{\sigma}{}{#2}{}{#1}}}
\newrobustcmd{\hatsigmaz}[2][]{\ensuremath{\subp{\hat{\sigma}}{}{#2}{}{#1}}}
\newrobustcmd{\widehatsigmaz}[2][]{\ensuremath{\subp{\widehat{\sigma}}{}{#2}{}{#1}}}
\newrobustcmd{\checksigmaz}[2][]{\ensuremath{\subp{\check{\sigma}}{}{#2}{}{#1}}}
\newrobustcmd{\tildesigmaz}[2][]{\ensuremath{\subp{\tilde{\sigma}}{}{#2}{}{#1}}}
\newrobustcmd{\widetildesigmaz}[2][]{\ensuremath{\subp{\widetilde{\sigma}}{}{#2}{}{#1}}}
\newrobustcmd{\acutesigmaz}[2][]{\ensuremath{\subp{\acute{\sigma}}{}{#2}{}{#1}}}
\newrobustcmd{\gravesigmaz}[2][]{\ensuremath{\subp{\grave{\sigma}}{}{#2}{}{#1}}}
\newrobustcmd{\dotsigmaz}[2][]{\ensuremath{\subp{\dot{\sigma}}{}{#2}{}{#1}}}
\newrobustcmd{\ddotsigmaz}[2][]{\ensuremath{\subp{\ddot{\sigma}}{}{#2}{}{#1}}}
\newrobustcmd{\brevesigmaz}[2][]{\ensuremath{\subp{\breve{\sigma}}{}{#2}{}{#1}}}
\newrobustcmd{\barsigmaz}[2][]{\ensuremath{\subp{\bar{\sigma}}{}{#2}{}{#1}}}
\newrobustcmd{\vecsigmaz}[2][]{\ensuremath{\subp{\vec{\sigma}}{}{#2}{}{#1}}}
\newrobustcmd{\bmsigmaz}[2][]{\ensuremath{\subp{\bm{\sigma}}{}{#2}{}{#1}}}
\newrobustcmd{\hatbmsigmaz}[2][]{\ensuremath{\subp{\hat{\bm{\sigma}}}{}{#2}{}{#1}}}
\newrobustcmd{\widehatbmsigmaz}[2][]{\ensuremath{\subp{\widehat{\bm{\sigma}}}{}{#2}{}{#1}}}
\newrobustcmd{\checkbmsigmaz}[2][]{\ensuremath{\subp{\check{\bm{\sigma}}}{}{#2}{}{#1}}}
\newrobustcmd{\tildebmsigmaz}[2][]{\ensuremath{\subp{\tilde{\bm{\sigma}}}{}{#2}{}{#1}}}
\newrobustcmd{\widetildebmsigmaz}[2][]{\ensuremath{\subp{\widetilde{\bm{\sigma}}}{}{#2}{}{#1}}}
\newrobustcmd{\acutebmsigmaz}[2][]{\ensuremath{\subp{\acute{\bm{\sigma}}}{}{#2}{}{#1}}}
\newrobustcmd{\gravebmsigmaz}[2][]{\ensuremath{\subp{\grave{\bm{\sigma}}}{}{#2}{}{#1}}}
\newrobustcmd{\dotbmsigmaz}[2][]{\ensuremath{\subp{\dot{\bm{\sigma}}}{}{#2}{}{#1}}}
\newrobustcmd{\ddotbmsigmaz}[2][]{\ensuremath{\subp{\ddot{\bm{\sigma}}}{}{#2}{}{#1}}}
\newrobustcmd{\brevebmsigmaz}[2][]{\ensuremath{\subp{\breve{\bm{\sigma}}}{}{#2}{}{#1}}}
\newrobustcmd{\barbmsigmaz}[2][]{\ensuremath{\subp{\bar{\bm{\sigma}}}{}{#2}{}{#1}}}
\newrobustcmd{\vecbmsigmaz}[2][]{\ensuremath{\subp{\vec{\bm{\sigma}}}{}{#2}{}{#1}}}
\newrobustcmd{\varsigmaz}[2][]{\ensuremath{\subp{\varsigma}{}{#2}{}{#1}}}
\newrobustcmd{\hatvarsigmaz}[2][]{\ensuremath{\subp{\hat{\varsigma}}{}{#2}{}{#1}}}
\newrobustcmd{\widehatvarsigmaz}[2][]{\ensuremath{\subp{\widehat{\varsigma}}{}{#2}{}{#1}}}
\newrobustcmd{\checkvarsigmaz}[2][]{\ensuremath{\subp{\check{\varsigma}}{}{#2}{}{#1}}}
\newrobustcmd{\tildevarsigmaz}[2][]{\ensuremath{\subp{\tilde{\varsigma}}{}{#2}{}{#1}}}
\newrobustcmd{\widetildevarsigmaz}[2][]{\ensuremath{\subp{\widetilde{\varsigma}}{}{#2}{}{#1}}}
\newrobustcmd{\acutevarsigmaz}[2][]{\ensuremath{\subp{\acute{\varsigma}}{}{#2}{}{#1}}}
\newrobustcmd{\gravevarsigmaz}[2][]{\ensuremath{\subp{\grave{\varsigma}}{}{#2}{}{#1}}}
\newrobustcmd{\dotvarsigmaz}[2][]{\ensuremath{\subp{\dot{\varsigma}}{}{#2}{}{#1}}}
\newrobustcmd{\ddotvarsigmaz}[2][]{\ensuremath{\subp{\ddot{\varsigma}}{}{#2}{}{#1}}}
\newrobustcmd{\brevevarsigmaz}[2][]{\ensuremath{\subp{\breve{\varsigma}}{}{#2}{}{#1}}}
\newrobustcmd{\barvarsigmaz}[2][]{\ensuremath{\subp{\bar{\varsigma}}{}{#2}{}{#1}}}
\newrobustcmd{\vecvarsigmaz}[2][]{\ensuremath{\subp{\vec{\varsigma}}{}{#2}{}{#1}}}
\newrobustcmd{\bmvarsigmaz}[2][]{\ensuremath{\subp{\bm{\varsigma}}{}{#2}{}{#1}}}
\newrobustcmd{\hatbmvarsigmaz}[2][]{\ensuremath{\subp{\hat{\bm{\varsigma}}}{}{#2}{}{#1}}}
\newrobustcmd{\widehatbmvarsigmaz}[2][]{\ensuremath{\subp{\widehat{\bm{\varsigma}}}{}{#2}{}{#1}}}
\newrobustcmd{\checkbmvarsigmaz}[2][]{\ensuremath{\subp{\check{\bm{\varsigma}}}{}{#2}{}{#1}}}
\newrobustcmd{\tildebmvarsigmaz}[2][]{\ensuremath{\subp{\tilde{\bm{\varsigma}}}{}{#2}{}{#1}}}
\newrobustcmd{\widetildebmvarsigmaz}[2][]{\ensuremath{\subp{\widetilde{\bm{\varsigma}}}{}{#2}{}{#1}}}
\newrobustcmd{\acutebmvarsigmaz}[2][]{\ensuremath{\subp{\acute{\bm{\varsigma}}}{}{#2}{}{#1}}}
\newrobustcmd{\gravebmvarsigmaz}[2][]{\ensuremath{\subp{\grave{\bm{\varsigma}}}{}{#2}{}{#1}}}
\newrobustcmd{\dotbmvarsigmaz}[2][]{\ensuremath{\subp{\dot{\bm{\varsigma}}}{}{#2}{}{#1}}}
\newrobustcmd{\ddotbmvarsigmaz}[2][]{\ensuremath{\subp{\ddot{\bm{\varsigma}}}{}{#2}{}{#1}}}
\newrobustcmd{\brevebmvarsigmaz}[2][]{\ensuremath{\subp{\breve{\bm{\varsigma}}}{}{#2}{}{#1}}}
\newrobustcmd{\barbmvarsigmaz}[2][]{\ensuremath{\subp{\bar{\bm{\varsigma}}}{}{#2}{}{#1}}}
\newrobustcmd{\vecbmvarsigmaz}[2][]{\ensuremath{\subp{\vec{\bm{\varsigma}}}{}{#2}{}{#1}}}
\newrobustcmd{\tauz}[2][]{\ensuremath{\subp{\tau}{}{#2}{}{#1}}}
\newrobustcmd{\hattauz}[2][]{\ensuremath{\subp{\hat{\tau}}{}{#2}{}{#1}}}
\newrobustcmd{\widehattauz}[2][]{\ensuremath{\subp{\widehat{\tau}}{}{#2}{}{#1}}}
\newrobustcmd{\checktauz}[2][]{\ensuremath{\subp{\check{\tau}}{}{#2}{}{#1}}}
\newrobustcmd{\tildetauz}[2][]{\ensuremath{\subp{\tilde{\tau}}{}{#2}{}{#1}}}
\newrobustcmd{\widetildetauz}[2][]{\ensuremath{\subp{\widetilde{\tau}}{}{#2}{}{#1}}}
\newrobustcmd{\acutetauz}[2][]{\ensuremath{\subp{\acute{\tau}}{}{#2}{}{#1}}}
\newrobustcmd{\gravetauz}[2][]{\ensuremath{\subp{\grave{\tau}}{}{#2}{}{#1}}}
\newrobustcmd{\dottauz}[2][]{\ensuremath{\subp{\dot{\tau}}{}{#2}{}{#1}}}
\newrobustcmd{\ddottauz}[2][]{\ensuremath{\subp{\ddot{\tau}}{}{#2}{}{#1}}}
\newrobustcmd{\brevetauz}[2][]{\ensuremath{\subp{\breve{\tau}}{}{#2}{}{#1}}}
\newrobustcmd{\bartauz}[2][]{\ensuremath{\subp{\bar{\tau}}{}{#2}{}{#1}}}
\newrobustcmd{\vectauz}[2][]{\ensuremath{\subp{\vec{\tau}}{}{#2}{}{#1}}}
\newrobustcmd{\bmtauz}[2][]{\ensuremath{\subp{\bm{\tau}}{}{#2}{}{#1}}}
\newrobustcmd{\hatbmtauz}[2][]{\ensuremath{\subp{\hat{\bm{\tau}}}{}{#2}{}{#1}}}
\newrobustcmd{\widehatbmtauz}[2][]{\ensuremath{\subp{\widehat{\bm{\tau}}}{}{#2}{}{#1}}}
\newrobustcmd{\checkbmtauz}[2][]{\ensuremath{\subp{\check{\bm{\tau}}}{}{#2}{}{#1}}}
\newrobustcmd{\tildebmtauz}[2][]{\ensuremath{\subp{\tilde{\bm{\tau}}}{}{#2}{}{#1}}}
\newrobustcmd{\widetildebmtauz}[2][]{\ensuremath{\subp{\widetilde{\bm{\tau}}}{}{#2}{}{#1}}}
\newrobustcmd{\acutebmtauz}[2][]{\ensuremath{\subp{\acute{\bm{\tau}}}{}{#2}{}{#1}}}
\newrobustcmd{\gravebmtauz}[2][]{\ensuremath{\subp{\grave{\bm{\tau}}}{}{#2}{}{#1}}}
\newrobustcmd{\dotbmtauz}[2][]{\ensuremath{\subp{\dot{\bm{\tau}}}{}{#2}{}{#1}}}
\newrobustcmd{\ddotbmtauz}[2][]{\ensuremath{\subp{\ddot{\bm{\tau}}}{}{#2}{}{#1}}}
\newrobustcmd{\brevebmtauz}[2][]{\ensuremath{\subp{\breve{\bm{\tau}}}{}{#2}{}{#1}}}
\newrobustcmd{\barbmtauz}[2][]{\ensuremath{\subp{\bar{\bm{\tau}}}{}{#2}{}{#1}}}
\newrobustcmd{\vecbmtauz}[2][]{\ensuremath{\subp{\vec{\bm{\tau}}}{}{#2}{}{#1}}}
\newrobustcmd{\upsilonz}[2][]{\ensuremath{\subp{\upsilon}{}{#2}{}{#1}}}
\newrobustcmd{\hatupsilonz}[2][]{\ensuremath{\subp{\hat{\upsilon}}{}{#2}{}{#1}}}
\newrobustcmd{\widehatupsilonz}[2][]{\ensuremath{\subp{\widehat{\upsilon}}{}{#2}{}{#1}}}
\newrobustcmd{\checkupsilonz}[2][]{\ensuremath{\subp{\check{\upsilon}}{}{#2}{}{#1}}}
\newrobustcmd{\tildeupsilonz}[2][]{\ensuremath{\subp{\tilde{\upsilon}}{}{#2}{}{#1}}}
\newrobustcmd{\widetildeupsilonz}[2][]{\ensuremath{\subp{\widetilde{\upsilon}}{}{#2}{}{#1}}}
\newrobustcmd{\acuteupsilonz}[2][]{\ensuremath{\subp{\acute{\upsilon}}{}{#2}{}{#1}}}
\newrobustcmd{\graveupsilonz}[2][]{\ensuremath{\subp{\grave{\upsilon}}{}{#2}{}{#1}}}
\newrobustcmd{\dotupsilonz}[2][]{\ensuremath{\subp{\dot{\upsilon}}{}{#2}{}{#1}}}
\newrobustcmd{\ddotupsilonz}[2][]{\ensuremath{\subp{\ddot{\upsilon}}{}{#2}{}{#1}}}
\newrobustcmd{\breveupsilonz}[2][]{\ensuremath{\subp{\breve{\upsilon}}{}{#2}{}{#1}}}
\newrobustcmd{\barupsilonz}[2][]{\ensuremath{\subp{\bar{\upsilon}}{}{#2}{}{#1}}}
\newrobustcmd{\vecupsilonz}[2][]{\ensuremath{\subp{\vec{\upsilon}}{}{#2}{}{#1}}}
\newrobustcmd{\bmupsilonz}[2][]{\ensuremath{\subp{\bm{\upsilon}}{}{#2}{}{#1}}}
\newrobustcmd{\hatbmupsilonz}[2][]{\ensuremath{\subp{\hat{\bm{\upsilon}}}{}{#2}{}{#1}}}
\newrobustcmd{\widehatbmupsilonz}[2][]{\ensuremath{\subp{\widehat{\bm{\upsilon}}}{}{#2}{}{#1}}}
\newrobustcmd{\checkbmupsilonz}[2][]{\ensuremath{\subp{\check{\bm{\upsilon}}}{}{#2}{}{#1}}}
\newrobustcmd{\tildebmupsilonz}[2][]{\ensuremath{\subp{\tilde{\bm{\upsilon}}}{}{#2}{}{#1}}}
\newrobustcmd{\widetildebmupsilonz}[2][]{\ensuremath{\subp{\widetilde{\bm{\upsilon}}}{}{#2}{}{#1}}}
\newrobustcmd{\acutebmupsilonz}[2][]{\ensuremath{\subp{\acute{\bm{\upsilon}}}{}{#2}{}{#1}}}
\newrobustcmd{\gravebmupsilonz}[2][]{\ensuremath{\subp{\grave{\bm{\upsilon}}}{}{#2}{}{#1}}}
\newrobustcmd{\dotbmupsilonz}[2][]{\ensuremath{\subp{\dot{\bm{\upsilon}}}{}{#2}{}{#1}}}
\newrobustcmd{\ddotbmupsilonz}[2][]{\ensuremath{\subp{\ddot{\bm{\upsilon}}}{}{#2}{}{#1}}}
\newrobustcmd{\brevebmupsilonz}[2][]{\ensuremath{\subp{\breve{\bm{\upsilon}}}{}{#2}{}{#1}}}
\newrobustcmd{\barbmupsilonz}[2][]{\ensuremath{\subp{\bar{\bm{\upsilon}}}{}{#2}{}{#1}}}
\newrobustcmd{\vecbmupsilonz}[2][]{\ensuremath{\subp{\vec{\bm{\upsilon}}}{}{#2}{}{#1}}}
\newrobustcmd{\phiz}[2][]{\ensuremath{\subp{\phi}{}{#2}{}{#1}}}
\newrobustcmd{\hatphiz}[2][]{\ensuremath{\subp{\hat{\phi}}{}{#2}{}{#1}}}
\newrobustcmd{\widehatphiz}[2][]{\ensuremath{\subp{\widehat{\phi}}{}{#2}{}{#1}}}
\newrobustcmd{\checkphiz}[2][]{\ensuremath{\subp{\check{\phi}}{}{#2}{}{#1}}}
\newrobustcmd{\tildephiz}[2][]{\ensuremath{\subp{\tilde{\phi}}{}{#2}{}{#1}}}
\newrobustcmd{\widetildephiz}[2][]{\ensuremath{\subp{\widetilde{\phi}}{}{#2}{}{#1}}}
\newrobustcmd{\acutephiz}[2][]{\ensuremath{\subp{\acute{\phi}}{}{#2}{}{#1}}}
\newrobustcmd{\gravephiz}[2][]{\ensuremath{\subp{\grave{\phi}}{}{#2}{}{#1}}}
\newrobustcmd{\dotphiz}[2][]{\ensuremath{\subp{\dot{\phi}}{}{#2}{}{#1}}}
\newrobustcmd{\ddotphiz}[2][]{\ensuremath{\subp{\ddot{\phi}}{}{#2}{}{#1}}}
\newrobustcmd{\brevephiz}[2][]{\ensuremath{\subp{\breve{\phi}}{}{#2}{}{#1}}}
\newrobustcmd{\barphiz}[2][]{\ensuremath{\subp{\bar{\phi}}{}{#2}{}{#1}}}
\newrobustcmd{\vecphiz}[2][]{\ensuremath{\subp{\vec{\phi}}{}{#2}{}{#1}}}
\newrobustcmd{\bmphiz}[2][]{\ensuremath{\subp{\bm{\phi}}{}{#2}{}{#1}}}
\newrobustcmd{\hatbmphiz}[2][]{\ensuremath{\subp{\hat{\bm{\phi}}}{}{#2}{}{#1}}}
\newrobustcmd{\widehatbmphiz}[2][]{\ensuremath{\subp{\widehat{\bm{\phi}}}{}{#2}{}{#1}}}
\newrobustcmd{\checkbmphiz}[2][]{\ensuremath{\subp{\check{\bm{\phi}}}{}{#2}{}{#1}}}
\newrobustcmd{\tildebmphiz}[2][]{\ensuremath{\subp{\tilde{\bm{\phi}}}{}{#2}{}{#1}}}
\newrobustcmd{\widetildebmphiz}[2][]{\ensuremath{\subp{\widetilde{\bm{\phi}}}{}{#2}{}{#1}}}
\newrobustcmd{\acutebmphiz}[2][]{\ensuremath{\subp{\acute{\bm{\phi}}}{}{#2}{}{#1}}}
\newrobustcmd{\gravebmphiz}[2][]{\ensuremath{\subp{\grave{\bm{\phi}}}{}{#2}{}{#1}}}
\newrobustcmd{\dotbmphiz}[2][]{\ensuremath{\subp{\dot{\bm{\phi}}}{}{#2}{}{#1}}}
\newrobustcmd{\ddotbmphiz}[2][]{\ensuremath{\subp{\ddot{\bm{\phi}}}{}{#2}{}{#1}}}
\newrobustcmd{\brevebmphiz}[2][]{\ensuremath{\subp{\breve{\bm{\phi}}}{}{#2}{}{#1}}}
\newrobustcmd{\barbmphiz}[2][]{\ensuremath{\subp{\bar{\bm{\phi}}}{}{#2}{}{#1}}}
\newrobustcmd{\vecbmphiz}[2][]{\ensuremath{\subp{\vec{\bm{\phi}}}{}{#2}{}{#1}}}
\newrobustcmd{\varphiz}[2][]{\ensuremath{\subp{\varphi}{}{#2}{}{#1}}}
\newrobustcmd{\hatvarphiz}[2][]{\ensuremath{\subp{\hat{\varphi}}{}{#2}{}{#1}}}
\newrobustcmd{\widehatvarphiz}[2][]{\ensuremath{\subp{\widehat{\varphi}}{}{#2}{}{#1}}}
\newrobustcmd{\checkvarphiz}[2][]{\ensuremath{\subp{\check{\varphi}}{}{#2}{}{#1}}}
\newrobustcmd{\tildevarphiz}[2][]{\ensuremath{\subp{\tilde{\varphi}}{}{#2}{}{#1}}}
\newrobustcmd{\widetildevarphiz}[2][]{\ensuremath{\subp{\widetilde{\varphi}}{}{#2}{}{#1}}}
\newrobustcmd{\acutevarphiz}[2][]{\ensuremath{\subp{\acute{\varphi}}{}{#2}{}{#1}}}
\newrobustcmd{\gravevarphiz}[2][]{\ensuremath{\subp{\grave{\varphi}}{}{#2}{}{#1}}}
\newrobustcmd{\dotvarphiz}[2][]{\ensuremath{\subp{\dot{\varphi}}{}{#2}{}{#1}}}
\newrobustcmd{\ddotvarphiz}[2][]{\ensuremath{\subp{\ddot{\varphi}}{}{#2}{}{#1}}}
\newrobustcmd{\brevevarphiz}[2][]{\ensuremath{\subp{\breve{\varphi}}{}{#2}{}{#1}}}
\newrobustcmd{\barvarphiz}[2][]{\ensuremath{\subp{\bar{\varphi}}{}{#2}{}{#1}}}
\newrobustcmd{\vecvarphiz}[2][]{\ensuremath{\subp{\vec{\varphi}}{}{#2}{}{#1}}}
\newrobustcmd{\bmvarphiz}[2][]{\ensuremath{\subp{\bm{\varphi}}{}{#2}{}{#1}}}
\newrobustcmd{\hatbmvarphiz}[2][]{\ensuremath{\subp{\hat{\bm{\varphi}}}{}{#2}{}{#1}}}
\newrobustcmd{\widehatbmvarphiz}[2][]{\ensuremath{\subp{\widehat{\bm{\varphi}}}{}{#2}{}{#1}}}
\newrobustcmd{\checkbmvarphiz}[2][]{\ensuremath{\subp{\check{\bm{\varphi}}}{}{#2}{}{#1}}}
\newrobustcmd{\tildebmvarphiz}[2][]{\ensuremath{\subp{\tilde{\bm{\varphi}}}{}{#2}{}{#1}}}
\newrobustcmd{\widetildebmvarphiz}[2][]{\ensuremath{\subp{\widetilde{\bm{\varphi}}}{}{#2}{}{#1}}}
\newrobustcmd{\acutebmvarphiz}[2][]{\ensuremath{\subp{\acute{\bm{\varphi}}}{}{#2}{}{#1}}}
\newrobustcmd{\gravebmvarphiz}[2][]{\ensuremath{\subp{\grave{\bm{\varphi}}}{}{#2}{}{#1}}}
\newrobustcmd{\dotbmvarphiz}[2][]{\ensuremath{\subp{\dot{\bm{\varphi}}}{}{#2}{}{#1}}}
\newrobustcmd{\ddotbmvarphiz}[2][]{\ensuremath{\subp{\ddot{\bm{\varphi}}}{}{#2}{}{#1}}}
\newrobustcmd{\brevebmvarphiz}[2][]{\ensuremath{\subp{\breve{\bm{\varphi}}}{}{#2}{}{#1}}}
\newrobustcmd{\barbmvarphiz}[2][]{\ensuremath{\subp{\bar{\bm{\varphi}}}{}{#2}{}{#1}}}
\newrobustcmd{\vecbmvarphiz}[2][]{\ensuremath{\subp{\vec{\bm{\varphi}}}{}{#2}{}{#1}}}
\newrobustcmd{\chiz}[2][]{\ensuremath{\subp{\chi}{}{#2}{}{#1}}}
\newrobustcmd{\hatchiz}[2][]{\ensuremath{\subp{\hat{\chi}}{}{#2}{}{#1}}}
\newrobustcmd{\widehatchiz}[2][]{\ensuremath{\subp{\widehat{\chi}}{}{#2}{}{#1}}}
\newrobustcmd{\checkchiz}[2][]{\ensuremath{\subp{\check{\chi}}{}{#2}{}{#1}}}
\newrobustcmd{\tildechiz}[2][]{\ensuremath{\subp{\tilde{\chi}}{}{#2}{}{#1}}}
\newrobustcmd{\widetildechiz}[2][]{\ensuremath{\subp{\widetilde{\chi}}{}{#2}{}{#1}}}
\newrobustcmd{\acutechiz}[2][]{\ensuremath{\subp{\acute{\chi}}{}{#2}{}{#1}}}
\newrobustcmd{\gravechiz}[2][]{\ensuremath{\subp{\grave{\chi}}{}{#2}{}{#1}}}
\newrobustcmd{\dotchiz}[2][]{\ensuremath{\subp{\dot{\chi}}{}{#2}{}{#1}}}
\newrobustcmd{\ddotchiz}[2][]{\ensuremath{\subp{\ddot{\chi}}{}{#2}{}{#1}}}
\newrobustcmd{\brevechiz}[2][]{\ensuremath{\subp{\breve{\chi}}{}{#2}{}{#1}}}
\newrobustcmd{\barchiz}[2][]{\ensuremath{\subp{\bar{\chi}}{}{#2}{}{#1}}}
\newrobustcmd{\vecchiz}[2][]{\ensuremath{\subp{\vec{\chi}}{}{#2}{}{#1}}}
\newrobustcmd{\bmchiz}[2][]{\ensuremath{\subp{\bm{\chi}}{}{#2}{}{#1}}}
\newrobustcmd{\hatbmchiz}[2][]{\ensuremath{\subp{\hat{\bm{\chi}}}{}{#2}{}{#1}}}
\newrobustcmd{\widehatbmchiz}[2][]{\ensuremath{\subp{\widehat{\bm{\chi}}}{}{#2}{}{#1}}}
\newrobustcmd{\checkbmchiz}[2][]{\ensuremath{\subp{\check{\bm{\chi}}}{}{#2}{}{#1}}}
\newrobustcmd{\tildebmchiz}[2][]{\ensuremath{\subp{\tilde{\bm{\chi}}}{}{#2}{}{#1}}}
\newrobustcmd{\widetildebmchiz}[2][]{\ensuremath{\subp{\widetilde{\bm{\chi}}}{}{#2}{}{#1}}}
\newrobustcmd{\acutebmchiz}[2][]{\ensuremath{\subp{\acute{\bm{\chi}}}{}{#2}{}{#1}}}
\newrobustcmd{\gravebmchiz}[2][]{\ensuremath{\subp{\grave{\bm{\chi}}}{}{#2}{}{#1}}}
\newrobustcmd{\dotbmchiz}[2][]{\ensuremath{\subp{\dot{\bm{\chi}}}{}{#2}{}{#1}}}
\newrobustcmd{\ddotbmchiz}[2][]{\ensuremath{\subp{\ddot{\bm{\chi}}}{}{#2}{}{#1}}}
\newrobustcmd{\brevebmchiz}[2][]{\ensuremath{\subp{\breve{\bm{\chi}}}{}{#2}{}{#1}}}
\newrobustcmd{\barbmchiz}[2][]{\ensuremath{\subp{\bar{\bm{\chi}}}{}{#2}{}{#1}}}
\newrobustcmd{\vecbmchiz}[2][]{\ensuremath{\subp{\vec{\bm{\chi}}}{}{#2}{}{#1}}}
\newrobustcmd{\psiz}[2][]{\ensuremath{\subp{\psi}{}{#2}{}{#1}}}
\newrobustcmd{\hatpsiz}[2][]{\ensuremath{\subp{\hat{\psi}}{}{#2}{}{#1}}}
\newrobustcmd{\widehatpsiz}[2][]{\ensuremath{\subp{\widehat{\psi}}{}{#2}{}{#1}}}
\newrobustcmd{\checkpsiz}[2][]{\ensuremath{\subp{\check{\psi}}{}{#2}{}{#1}}}
\newrobustcmd{\tildepsiz}[2][]{\ensuremath{\subp{\tilde{\psi}}{}{#2}{}{#1}}}
\newrobustcmd{\widetildepsiz}[2][]{\ensuremath{\subp{\widetilde{\psi}}{}{#2}{}{#1}}}
\newrobustcmd{\acutepsiz}[2][]{\ensuremath{\subp{\acute{\psi}}{}{#2}{}{#1}}}
\newrobustcmd{\gravepsiz}[2][]{\ensuremath{\subp{\grave{\psi}}{}{#2}{}{#1}}}
\newrobustcmd{\dotpsiz}[2][]{\ensuremath{\subp{\dot{\psi}}{}{#2}{}{#1}}}
\newrobustcmd{\ddotpsiz}[2][]{\ensuremath{\subp{\ddot{\psi}}{}{#2}{}{#1}}}
\newrobustcmd{\brevepsiz}[2][]{\ensuremath{\subp{\breve{\psi}}{}{#2}{}{#1}}}
\newrobustcmd{\barpsiz}[2][]{\ensuremath{\subp{\bar{\psi}}{}{#2}{}{#1}}}
\newrobustcmd{\vecpsiz}[2][]{\ensuremath{\subp{\vec{\psi}}{}{#2}{}{#1}}}
\newrobustcmd{\bmpsiz}[2][]{\ensuremath{\subp{\bm{\psi}}{}{#2}{}{#1}}}
\newrobustcmd{\hatbmpsiz}[2][]{\ensuremath{\subp{\hat{\bm{\psi}}}{}{#2}{}{#1}}}
\newrobustcmd{\widehatbmpsiz}[2][]{\ensuremath{\subp{\widehat{\bm{\psi}}}{}{#2}{}{#1}}}
\newrobustcmd{\checkbmpsiz}[2][]{\ensuremath{\subp{\check{\bm{\psi}}}{}{#2}{}{#1}}}
\newrobustcmd{\tildebmpsiz}[2][]{\ensuremath{\subp{\tilde{\bm{\psi}}}{}{#2}{}{#1}}}
\newrobustcmd{\widetildebmpsiz}[2][]{\ensuremath{\subp{\widetilde{\bm{\psi}}}{}{#2}{}{#1}}}
\newrobustcmd{\acutebmpsiz}[2][]{\ensuremath{\subp{\acute{\bm{\psi}}}{}{#2}{}{#1}}}
\newrobustcmd{\gravebmpsiz}[2][]{\ensuremath{\subp{\grave{\bm{\psi}}}{}{#2}{}{#1}}}
\newrobustcmd{\dotbmpsiz}[2][]{\ensuremath{\subp{\dot{\bm{\psi}}}{}{#2}{}{#1}}}
\newrobustcmd{\ddotbmpsiz}[2][]{\ensuremath{\subp{\ddot{\bm{\psi}}}{}{#2}{}{#1}}}
\newrobustcmd{\brevebmpsiz}[2][]{\ensuremath{\subp{\breve{\bm{\psi}}}{}{#2}{}{#1}}}
\newrobustcmd{\barbmpsiz}[2][]{\ensuremath{\subp{\bar{\bm{\psi}}}{}{#2}{}{#1}}}
\newrobustcmd{\vecbmpsiz}[2][]{\ensuremath{\subp{\vec{\bm{\psi}}}{}{#2}{}{#1}}}
\newrobustcmd{\omegaz}[2][]{\ensuremath{\subp{\omega}{}{#2}{}{#1}}}
\newrobustcmd{\hatomegaz}[2][]{\ensuremath{\subp{\hat{\omega}}{}{#2}{}{#1}}}
\newrobustcmd{\widehatomegaz}[2][]{\ensuremath{\subp{\widehat{\omega}}{}{#2}{}{#1}}}
\newrobustcmd{\checkomegaz}[2][]{\ensuremath{\subp{\check{\omega}}{}{#2}{}{#1}}}
\newrobustcmd{\tildeomegaz}[2][]{\ensuremath{\subp{\tilde{\omega}}{}{#2}{}{#1}}}
\newrobustcmd{\widetildeomegaz}[2][]{\ensuremath{\subp{\widetilde{\omega}}{}{#2}{}{#1}}}
\newrobustcmd{\acuteomegaz}[2][]{\ensuremath{\subp{\acute{\omega}}{}{#2}{}{#1}}}
\newrobustcmd{\graveomegaz}[2][]{\ensuremath{\subp{\grave{\omega}}{}{#2}{}{#1}}}
\newrobustcmd{\dotomegaz}[2][]{\ensuremath{\subp{\dot{\omega}}{}{#2}{}{#1}}}
\newrobustcmd{\ddotomegaz}[2][]{\ensuremath{\subp{\ddot{\omega}}{}{#2}{}{#1}}}
\newrobustcmd{\breveomegaz}[2][]{\ensuremath{\subp{\breve{\omega}}{}{#2}{}{#1}}}
\newrobustcmd{\baromegaz}[2][]{\ensuremath{\subp{\bar{\omega}}{}{#2}{}{#1}}}
\newrobustcmd{\vecomegaz}[2][]{\ensuremath{\subp{\vec{\omega}}{}{#2}{}{#1}}}
\newrobustcmd{\bmomegaz}[2][]{\ensuremath{\subp{\bm{\omega}}{}{#2}{}{#1}}}
\newrobustcmd{\hatbmomegaz}[2][]{\ensuremath{\subp{\hat{\bm{\omega}}}{}{#2}{}{#1}}}
\newrobustcmd{\widehatbmomegaz}[2][]{\ensuremath{\subp{\widehat{\bm{\omega}}}{}{#2}{}{#1}}}
\newrobustcmd{\checkbmomegaz}[2][]{\ensuremath{\subp{\check{\bm{\omega}}}{}{#2}{}{#1}}}
\newrobustcmd{\tildebmomegaz}[2][]{\ensuremath{\subp{\tilde{\bm{\omega}}}{}{#2}{}{#1}}}
\newrobustcmd{\widetildebmomegaz}[2][]{\ensuremath{\subp{\widetilde{\bm{\omega}}}{}{#2}{}{#1}}}
\newrobustcmd{\acutebmomegaz}[2][]{\ensuremath{\subp{\acute{\bm{\omega}}}{}{#2}{}{#1}}}
\newrobustcmd{\gravebmomegaz}[2][]{\ensuremath{\subp{\grave{\bm{\omega}}}{}{#2}{}{#1}}}
\newrobustcmd{\dotbmomegaz}[2][]{\ensuremath{\subp{\dot{\bm{\omega}}}{}{#2}{}{#1}}}
\newrobustcmd{\ddotbmomegaz}[2][]{\ensuremath{\subp{\ddot{\bm{\omega}}}{}{#2}{}{#1}}}
\newrobustcmd{\brevebmomegaz}[2][]{\ensuremath{\subp{\breve{\bm{\omega}}}{}{#2}{}{#1}}}
\newrobustcmd{\barbmomegaz}[2][]{\ensuremath{\subp{\bar{\bm{\omega}}}{}{#2}{}{#1}}}
\newrobustcmd{\vecbmomegaz}[2][]{\ensuremath{\subp{\vec{\bm{\omega}}}{}{#2}{}{#1}}}

\newrobustcmd{\Alphaz}[2][]{\ensuremath{\subp{A}{}{#2}{}{#1}}}
\newrobustcmd{\hatAlphaz}[2][]{\ensuremath{\subp{\hat{A}}{}{#2}{}{#1}}}
\newrobustcmd{\widehatAlphaz}[2][]{\ensuremath{\subp{\widehat{A}}{}{#2}{}{#1}}}
\newrobustcmd{\checkAlphaz}[2][]{\ensuremath{\subp{\check{A}}{}{#2}{}{#1}}}
\newrobustcmd{\tildeAlphaz}[2][]{\ensuremath{\subp{\tilde{A}}{}{#2}{}{#1}}}
\newrobustcmd{\widetildeAlphaz}[2][]{\ensuremath{\subp{\widetilde{A}}{}{#2}{}{#1}}}
\newrobustcmd{\acuteAlphaz}[2][]{\ensuremath{\subp{\acute{A}}{}{#2}{}{#1}}}
\newrobustcmd{\graveAlphaz}[2][]{\ensuremath{\subp{\grave{A}}{}{#2}{}{#1}}}
\newrobustcmd{\dotAlphaz}[2][]{\ensuremath{\subp{\dot{A}}{}{#2}{}{#1}}}
\newrobustcmd{\ddotAlphaz}[2][]{\ensuremath{\subp{\ddot{A}}{}{#2}{}{#1}}}
\newrobustcmd{\breveAlphaz}[2][]{\ensuremath{\subp{\breve{A}}{}{#2}{}{#1}}}
\newrobustcmd{\barAlphaz}[2][]{\ensuremath{\subp{\bar{A}}{}{#2}{}{#1}}}
\newrobustcmd{\vecAlphaz}[2][]{\ensuremath{\subp{\vec{A}}{}{#2}{}{#1}}}
\newrobustcmd{\bmAlphaz}[2][]{\ensuremath{\subp{\bm{A}}{}{#2}{}{#1}}}
\newrobustcmd{\hatbmAlphaz}[2][]{\ensuremath{\subp{\hat{\bm{A}}}{}{#2}{}{#1}}}
\newrobustcmd{\widehatbmAlphaz}[2][]{\ensuremath{\subp{\widehat{\bm{A}}}{}{#2}{}{#1}}}
\newrobustcmd{\checkbmAlphaz}[2][]{\ensuremath{\subp{\check{\bm{A}}}{}{#2}{}{#1}}}
\newrobustcmd{\tildebmAlphaz}[2][]{\ensuremath{\subp{\tilde{\bm{A}}}{}{#2}{}{#1}}}
\newrobustcmd{\widetildebmAlphaz}[2][]{\ensuremath{\subp{\widetilde{\bm{A}}}{}{#2}{}{#1}}}
\newrobustcmd{\acutebmAlphaz}[2][]{\ensuremath{\subp{\acute{\bm{A}}}{}{#2}{}{#1}}}
\newrobustcmd{\gravebmAlphaz}[2][]{\ensuremath{\subp{\grave{\bm{A}}}{}{#2}{}{#1}}}
\newrobustcmd{\dotbmAlphaz}[2][]{\ensuremath{\subp{\dot{\bm{A}}}{}{#2}{}{#1}}}
\newrobustcmd{\ddotbmAlphaz}[2][]{\ensuremath{\subp{\ddot{\bm{A}}}{}{#2}{}{#1}}}
\newrobustcmd{\brevebmAlphaz}[2][]{\ensuremath{\subp{\breve{\bm{A}}}{}{#2}{}{#1}}}
\newrobustcmd{\barbmAlphaz}[2][]{\ensuremath{\subp{\bar{\bm{A}}}{}{#2}{}{#1}}}
\newrobustcmd{\vecbmAlphaz}[2][]{\ensuremath{\subp{\vec{\bm{A}}}{}{#2}{}{#1}}}
\newrobustcmd{\Betaz}[2][]{\ensuremath{\subp{B}{}{#2}{}{#1}}}
\newrobustcmd{\hatBetaz}[2][]{\ensuremath{\subp{\hat{B}}{}{#2}{}{#1}}}
\newrobustcmd{\widehatBetaz}[2][]{\ensuremath{\subp{\widehat{B}}{}{#2}{}{#1}}}
\newrobustcmd{\checkBetaz}[2][]{\ensuremath{\subp{\check{B}}{}{#2}{}{#1}}}
\newrobustcmd{\tildeBetaz}[2][]{\ensuremath{\subp{\tilde{B}}{}{#2}{}{#1}}}
\newrobustcmd{\widetildeBetaz}[2][]{\ensuremath{\subp{\widetilde{B}}{}{#2}{}{#1}}}
\newrobustcmd{\acuteBetaz}[2][]{\ensuremath{\subp{\acute{B}}{}{#2}{}{#1}}}
\newrobustcmd{\graveBetaz}[2][]{\ensuremath{\subp{\grave{B}}{}{#2}{}{#1}}}
\newrobustcmd{\dotBetaz}[2][]{\ensuremath{\subp{\dot{B}}{}{#2}{}{#1}}}
\newrobustcmd{\ddotBetaz}[2][]{\ensuremath{\subp{\ddot{B}}{}{#2}{}{#1}}}
\newrobustcmd{\breveBetaz}[2][]{\ensuremath{\subp{\breve{B}}{}{#2}{}{#1}}}
\newrobustcmd{\barBetaz}[2][]{\ensuremath{\subp{\bar{B}}{}{#2}{}{#1}}}
\newrobustcmd{\vecBetaz}[2][]{\ensuremath{\subp{\vec{B}}{}{#2}{}{#1}}}
\newrobustcmd{\bmBetaz}[2][]{\ensuremath{\subp{\bm{B}}{}{#2}{}{#1}}}
\newrobustcmd{\hatbmBetaz}[2][]{\ensuremath{\subp{\hat{\bm{B}}}{}{#2}{}{#1}}}
\newrobustcmd{\widehatbmBetaz}[2][]{\ensuremath{\subp{\widehat{\bm{B}}}{}{#2}{}{#1}}}
\newrobustcmd{\checkbmBetaz}[2][]{\ensuremath{\subp{\check{\bm{B}}}{}{#2}{}{#1}}}
\newrobustcmd{\tildebmBetaz}[2][]{\ensuremath{\subp{\tilde{\bm{B}}}{}{#2}{}{#1}}}
\newrobustcmd{\widetildebmBetaz}[2][]{\ensuremath{\subp{\widetilde{\bm{B}}}{}{#2}{}{#1}}}
\newrobustcmd{\acutebmBetaz}[2][]{\ensuremath{\subp{\acute{\bm{B}}}{}{#2}{}{#1}}}
\newrobustcmd{\gravebmBetaz}[2][]{\ensuremath{\subp{\grave{\bm{B}}}{}{#2}{}{#1}}}
\newrobustcmd{\dotbmBetaz}[2][]{\ensuremath{\subp{\dot{\bm{B}}}{}{#2}{}{#1}}}
\newrobustcmd{\ddotbmBetaz}[2][]{\ensuremath{\subp{\ddot{\bm{B}}}{}{#2}{}{#1}}}
\newrobustcmd{\brevebmBetaz}[2][]{\ensuremath{\subp{\breve{\bm{B}}}{}{#2}{}{#1}}}
\newrobustcmd{\barbmBetaz}[2][]{\ensuremath{\subp{\bar{\bm{B}}}{}{#2}{}{#1}}}
\newrobustcmd{\vecbmBetaz}[2][]{\ensuremath{\subp{\vec{\bm{B}}}{}{#2}{}{#1}}}
\newrobustcmd{\Gammaz}[2][]{\ensuremath{\subp{\Gamma}{}{#2}{}{#1}}}
\newrobustcmd{\hatGammaz}[2][]{\ensuremath{\subp{\hat{\Gamma}}{}{#2}{}{#1}}}
\newrobustcmd{\widehatGammaz}[2][]{\ensuremath{\subp{\widehat{\Gamma}}{}{#2}{}{#1}}}
\newrobustcmd{\checkGammaz}[2][]{\ensuremath{\subp{\check{\Gamma}}{}{#2}{}{#1}}}
\newrobustcmd{\tildeGammaz}[2][]{\ensuremath{\subp{\tilde{\Gamma}}{}{#2}{}{#1}}}
\newrobustcmd{\widetildeGammaz}[2][]{\ensuremath{\subp{\widetilde{\Gamma}}{}{#2}{}{#1}}}
\newrobustcmd{\acuteGammaz}[2][]{\ensuremath{\subp{\acute{\Gamma}}{}{#2}{}{#1}}}
\newrobustcmd{\graveGammaz}[2][]{\ensuremath{\subp{\grave{\Gamma}}{}{#2}{}{#1}}}
\newrobustcmd{\dotGammaz}[2][]{\ensuremath{\subp{\dot{\Gamma}}{}{#2}{}{#1}}}
\newrobustcmd{\ddotGammaz}[2][]{\ensuremath{\subp{\ddot{\Gamma}}{}{#2}{}{#1}}}
\newrobustcmd{\breveGammaz}[2][]{\ensuremath{\subp{\breve{\Gamma}}{}{#2}{}{#1}}}
\newrobustcmd{\barGammaz}[2][]{\ensuremath{\subp{\bar{\Gamma}}{}{#2}{}{#1}}}
\newrobustcmd{\vecGammaz}[2][]{\ensuremath{\subp{\vec{\Gamma}}{}{#2}{}{#1}}}
\newrobustcmd{\bmGammaz}[2][]{\ensuremath{\subp{\bm{\Gamma}}{}{#2}{}{#1}}}
\newrobustcmd{\hatbmGammaz}[2][]{\ensuremath{\subp{\hat{\bm{\Gamma}}}{}{#2}{}{#1}}}
\newrobustcmd{\widehatbmGammaz}[2][]{\ensuremath{\subp{\widehat{\bm{\Gamma}}}{}{#2}{}{#1}}}
\newrobustcmd{\checkbmGammaz}[2][]{\ensuremath{\subp{\check{\bm{\Gamma}}}{}{#2}{}{#1}}}
\newrobustcmd{\tildebmGammaz}[2][]{\ensuremath{\subp{\tilde{\bm{\Gamma}}}{}{#2}{}{#1}}}
\newrobustcmd{\widetildebmGammaz}[2][]{\ensuremath{\subp{\widetilde{\bm{\Gamma}}}{}{#2}{}{#1}}}
\newrobustcmd{\acutebmGammaz}[2][]{\ensuremath{\subp{\acute{\bm{\Gamma}}}{}{#2}{}{#1}}}
\newrobustcmd{\gravebmGammaz}[2][]{\ensuremath{\subp{\grave{\bm{\Gamma}}}{}{#2}{}{#1}}}
\newrobustcmd{\dotbmGammaz}[2][]{\ensuremath{\subp{\dot{\bm{\Gamma}}}{}{#2}{}{#1}}}
\newrobustcmd{\ddotbmGammaz}[2][]{\ensuremath{\subp{\ddot{\bm{\Gamma}}}{}{#2}{}{#1}}}
\newrobustcmd{\brevebmGammaz}[2][]{\ensuremath{\subp{\breve{\bm{\Gamma}}}{}{#2}{}{#1}}}
\newrobustcmd{\barbmGammaz}[2][]{\ensuremath{\subp{\bar{\bm{\Gamma}}}{}{#2}{}{#1}}}
\newrobustcmd{\vecbmGammaz}[2][]{\ensuremath{\subp{\vec{\bm{\Gamma}}}{}{#2}{}{#1}}}
\newrobustcmd{\Deltaz}[2][]{\ensuremath{\subp{\Delta}{}{#2}{}{#1}}}
\newrobustcmd{\hatDeltaz}[2][]{\ensuremath{\subp{\hat{\Delta}}{}{#2}{}{#1}}}
\newrobustcmd{\widehatDeltaz}[2][]{\ensuremath{\subp{\widehat{\Delta}}{}{#2}{}{#1}}}
\newrobustcmd{\checkDeltaz}[2][]{\ensuremath{\subp{\check{\Delta}}{}{#2}{}{#1}}}
\newrobustcmd{\tildeDeltaz}[2][]{\ensuremath{\subp{\tilde{\Delta}}{}{#2}{}{#1}}}
\newrobustcmd{\widetildeDeltaz}[2][]{\ensuremath{\subp{\widetilde{\Delta}}{}{#2}{}{#1}}}
\newrobustcmd{\acuteDeltaz}[2][]{\ensuremath{\subp{\acute{\Delta}}{}{#2}{}{#1}}}
\newrobustcmd{\graveDeltaz}[2][]{\ensuremath{\subp{\grave{\Delta}}{}{#2}{}{#1}}}
\newrobustcmd{\dotDeltaz}[2][]{\ensuremath{\subp{\dot{\Delta}}{}{#2}{}{#1}}}
\newrobustcmd{\ddotDeltaz}[2][]{\ensuremath{\subp{\ddot{\Delta}}{}{#2}{}{#1}}}
\newrobustcmd{\breveDeltaz}[2][]{\ensuremath{\subp{\breve{\Delta}}{}{#2}{}{#1}}}
\newrobustcmd{\barDeltaz}[2][]{\ensuremath{\subp{\bar{\Delta}}{}{#2}{}{#1}}}
\newrobustcmd{\vecDeltaz}[2][]{\ensuremath{\subp{\vec{\Delta}}{}{#2}{}{#1}}}
\newrobustcmd{\bmDeltaz}[2][]{\ensuremath{\subp{\bm{\Delta}}{}{#2}{}{#1}}}
\newrobustcmd{\hatbmDeltaz}[2][]{\ensuremath{\subp{\hat{\bm{\Delta}}}{}{#2}{}{#1}}}
\newrobustcmd{\widehatbmDeltaz}[2][]{\ensuremath{\subp{\widehat{\bm{\Delta}}}{}{#2}{}{#1}}}
\newrobustcmd{\checkbmDeltaz}[2][]{\ensuremath{\subp{\check{\bm{\Delta}}}{}{#2}{}{#1}}}
\newrobustcmd{\tildebmDeltaz}[2][]{\ensuremath{\subp{\tilde{\bm{\Delta}}}{}{#2}{}{#1}}}
\newrobustcmd{\widetildebmDeltaz}[2][]{\ensuremath{\subp{\widetilde{\bm{\Delta}}}{}{#2}{}{#1}}}
\newrobustcmd{\acutebmDeltaz}[2][]{\ensuremath{\subp{\acute{\bm{\Delta}}}{}{#2}{}{#1}}}
\newrobustcmd{\gravebmDeltaz}[2][]{\ensuremath{\subp{\grave{\bm{\Delta}}}{}{#2}{}{#1}}}
\newrobustcmd{\dotbmDeltaz}[2][]{\ensuremath{\subp{\dot{\bm{\Delta}}}{}{#2}{}{#1}}}
\newrobustcmd{\ddotbmDeltaz}[2][]{\ensuremath{\subp{\ddot{\bm{\Delta}}}{}{#2}{}{#1}}}
\newrobustcmd{\brevebmDeltaz}[2][]{\ensuremath{\subp{\breve{\bm{\Delta}}}{}{#2}{}{#1}}}
\newrobustcmd{\barbmDeltaz}[2][]{\ensuremath{\subp{\bar{\bm{\Delta}}}{}{#2}{}{#1}}}
\newrobustcmd{\vecbmDeltaz}[2][]{\ensuremath{\subp{\vec{\bm{\Delta}}}{}{#2}{}{#1}}}
\newrobustcmd{\Epsilonz}[2][]{\ensuremath{\subp{\Epsilon}{}{#2}{}{#1}}}
\newrobustcmd{\hatEpsilonz}[2][]{\ensuremath{\subp{\hat{\Epsilon}}{}{#2}{}{#1}}}
\newrobustcmd{\widehatEpsilonz}[2][]{\ensuremath{\subp{\widehat{\Epsilon}}{}{#2}{}{#1}}}
\newrobustcmd{\checkEpsilonz}[2][]{\ensuremath{\subp{\check{\Epsilon}}{}{#2}{}{#1}}}
\newrobustcmd{\tildeEpsilonz}[2][]{\ensuremath{\subp{\tilde{\Epsilon}}{}{#2}{}{#1}}}
\newrobustcmd{\widetildeEpsilonz}[2][]{\ensuremath{\subp{\widetilde{\Epsilon}}{}{#2}{}{#1}}}
\newrobustcmd{\acuteEpsilonz}[2][]{\ensuremath{\subp{\acute{\Epsilon}}{}{#2}{}{#1}}}
\newrobustcmd{\graveEpsilonz}[2][]{\ensuremath{\subp{\grave{\Epsilon}}{}{#2}{}{#1}}}
\newrobustcmd{\dotEpsilonz}[2][]{\ensuremath{\subp{\dot{\Epsilon}}{}{#2}{}{#1}}}
\newrobustcmd{\ddotEpsilonz}[2][]{\ensuremath{\subp{\ddot{\Epsilon}}{}{#2}{}{#1}}}
\newrobustcmd{\breveEpsilonz}[2][]{\ensuremath{\subp{\breve{\Epsilon}}{}{#2}{}{#1}}}
\newrobustcmd{\barEpsilonz}[2][]{\ensuremath{\subp{\bar{\Epsilon}}{}{#2}{}{#1}}}
\newrobustcmd{\vecEpsilonz}[2][]{\ensuremath{\subp{\vec{\Epsilon}}{}{#2}{}{#1}}}
\newrobustcmd{\bmEpsilonz}[2][]{\ensuremath{\subp{\bm{\Epsilon}}{}{#2}{}{#1}}}
\newrobustcmd{\hatbmEpsilonz}[2][]{\ensuremath{\subp{\hat{\bm{\Epsilon}}}{}{#2}{}{#1}}}
\newrobustcmd{\widehatbmEpsilonz}[2][]{\ensuremath{\subp{\widehat{\bm{\Epsilon}}}{}{#2}{}{#1}}}
\newrobustcmd{\checkbmEpsilonz}[2][]{\ensuremath{\subp{\check{\bm{\Epsilon}}}{}{#2}{}{#1}}}
\newrobustcmd{\tildebmEpsilonz}[2][]{\ensuremath{\subp{\tilde{\bm{\Epsilon}}}{}{#2}{}{#1}}}
\newrobustcmd{\widetildebmEpsilonz}[2][]{\ensuremath{\subp{\widetilde{\bm{\Epsilon}}}{}{#2}{}{#1}}}
\newrobustcmd{\acutebmEpsilonz}[2][]{\ensuremath{\subp{\acute{\bm{\Epsilon}}}{}{#2}{}{#1}}}
\newrobustcmd{\gravebmEpsilonz}[2][]{\ensuremath{\subp{\grave{\bm{\Epsilon}}}{}{#2}{}{#1}}}
\newrobustcmd{\dotbmEpsilonz}[2][]{\ensuremath{\subp{\dot{\bm{\Epsilon}}}{}{#2}{}{#1}}}
\newrobustcmd{\ddotbmEpsilonz}[2][]{\ensuremath{\subp{\ddot{\bm{\Epsilon}}}{}{#2}{}{#1}}}
\newrobustcmd{\brevebmEpsilonz}[2][]{\ensuremath{\subp{\breve{\bm{\Epsilon}}}{}{#2}{}{#1}}}
\newrobustcmd{\barbmEpsilonz}[2][]{\ensuremath{\subp{\bar{\bm{\Epsilon}}}{}{#2}{}{#1}}}
\newrobustcmd{\vecbmEpsilonz}[2][]{\ensuremath{\subp{\vec{\bm{\Epsilon}}}{}{#2}{}{#1}}}
\newrobustcmd{\Zetaz}[2][]{\ensuremath{\subp{Z}{}{#2}{}{#1}}}
\newrobustcmd{\hatZetaz}[2][]{\ensuremath{\subp{\hat{Z}}{}{#2}{}{#1}}}
\newrobustcmd{\widehatZetaz}[2][]{\ensuremath{\subp{\widehat{Z}}{}{#2}{}{#1}}}
\newrobustcmd{\checkZetaz}[2][]{\ensuremath{\subp{\check{Z}}{}{#2}{}{#1}}}
\newrobustcmd{\tildeZetaz}[2][]{\ensuremath{\subp{\tilde{Z}}{}{#2}{}{#1}}}
\newrobustcmd{\widetildeZetaz}[2][]{\ensuremath{\subp{\widetilde{Z}}{}{#2}{}{#1}}}
\newrobustcmd{\acuteZetaz}[2][]{\ensuremath{\subp{\acute{Z}}{}{#2}{}{#1}}}
\newrobustcmd{\graveZetaz}[2][]{\ensuremath{\subp{\grave{Z}}{}{#2}{}{#1}}}
\newrobustcmd{\dotZetaz}[2][]{\ensuremath{\subp{\dot{Z}}{}{#2}{}{#1}}}
\newrobustcmd{\ddotZetaz}[2][]{\ensuremath{\subp{\ddot{Z}}{}{#2}{}{#1}}}
\newrobustcmd{\breveZetaz}[2][]{\ensuremath{\subp{\breve{Z}}{}{#2}{}{#1}}}
\newrobustcmd{\barZetaz}[2][]{\ensuremath{\subp{\bar{Z}}{}{#2}{}{#1}}}
\newrobustcmd{\vecZetaz}[2][]{\ensuremath{\subp{\vec{Z}}{}{#2}{}{#1}}}
\newrobustcmd{\bmZetaz}[2][]{\ensuremath{\subp{\bm{Z}}{}{#2}{}{#1}}}
\newrobustcmd{\hatbmZetaz}[2][]{\ensuremath{\subp{\hat{\bm{Z}}}{}{#2}{}{#1}}}
\newrobustcmd{\widehatbmZetaz}[2][]{\ensuremath{\subp{\widehat{\bm{Z}}}{}{#2}{}{#1}}}
\newrobustcmd{\checkbmZetaz}[2][]{\ensuremath{\subp{\check{\bm{Z}}}{}{#2}{}{#1}}}
\newrobustcmd{\tildebmZetaz}[2][]{\ensuremath{\subp{\tilde{\bm{Z}}}{}{#2}{}{#1}}}
\newrobustcmd{\widetildebmZetaz}[2][]{\ensuremath{\subp{\widetilde{\bm{Z}}}{}{#2}{}{#1}}}
\newrobustcmd{\acutebmZetaz}[2][]{\ensuremath{\subp{\acute{\bm{Z}}}{}{#2}{}{#1}}}
\newrobustcmd{\gravebmZetaz}[2][]{\ensuremath{\subp{\grave{\bm{Z}}}{}{#2}{}{#1}}}
\newrobustcmd{\dotbmZetaz}[2][]{\ensuremath{\subp{\dot{\bm{Z}}}{}{#2}{}{#1}}}
\newrobustcmd{\ddotbmZetaz}[2][]{\ensuremath{\subp{\ddot{\bm{Z}}}{}{#2}{}{#1}}}
\newrobustcmd{\brevebmZetaz}[2][]{\ensuremath{\subp{\breve{\bm{Z}}}{}{#2}{}{#1}}}
\newrobustcmd{\barbmZetaz}[2][]{\ensuremath{\subp{\bar{\bm{Z}}}{}{#2}{}{#1}}}
\newrobustcmd{\vecbmZetaz}[2][]{\ensuremath{\subp{\vec{\bm{Z}}}{}{#2}{}{#1}}}
\newrobustcmd{\Etaz}[2][]{\ensuremath{\subp{H}{}{#2}{}{#1}}}
\newrobustcmd{\hatEtaz}[2][]{\ensuremath{\subp{\hat{H}}{}{#2}{}{#1}}}
\newrobustcmd{\widehatEtaz}[2][]{\ensuremath{\subp{\widehat{H}}{}{#2}{}{#1}}}
\newrobustcmd{\checkEtaz}[2][]{\ensuremath{\subp{\check{H}}{}{#2}{}{#1}}}
\newrobustcmd{\tildeEtaz}[2][]{\ensuremath{\subp{\tilde{H}}{}{#2}{}{#1}}}
\newrobustcmd{\widetildeEtaz}[2][]{\ensuremath{\subp{\widetilde{H}}{}{#2}{}{#1}}}
\newrobustcmd{\acuteEtaz}[2][]{\ensuremath{\subp{\acute{H}}{}{#2}{}{#1}}}
\newrobustcmd{\graveEtaz}[2][]{\ensuremath{\subp{\grave{H}}{}{#2}{}{#1}}}
\newrobustcmd{\dotEtaz}[2][]{\ensuremath{\subp{\dot{H}}{}{#2}{}{#1}}}
\newrobustcmd{\ddotEtaz}[2][]{\ensuremath{\subp{\ddot{H}}{}{#2}{}{#1}}}
\newrobustcmd{\breveEtaz}[2][]{\ensuremath{\subp{\breve{H}}{}{#2}{}{#1}}}
\newrobustcmd{\barEtaz}[2][]{\ensuremath{\subp{\bar{H}}{}{#2}{}{#1}}}
\newrobustcmd{\vecEtaz}[2][]{\ensuremath{\subp{\vec{H}}{}{#2}{}{#1}}}
\newrobustcmd{\bmEtaz}[2][]{\ensuremath{\subp{\bm{H}}{}{#2}{}{#1}}}
\newrobustcmd{\hatbmEtaz}[2][]{\ensuremath{\subp{\hat{\bm{H}}}{}{#2}{}{#1}}}
\newrobustcmd{\widehatbmEtaz}[2][]{\ensuremath{\subp{\widehat{\bm{H}}}{}{#2}{}{#1}}}
\newrobustcmd{\checkbmEtaz}[2][]{\ensuremath{\subp{\check{\bm{H}}}{}{#2}{}{#1}}}
\newrobustcmd{\tildebmEtaz}[2][]{\ensuremath{\subp{\tilde{\bm{H}}}{}{#2}{}{#1}}}
\newrobustcmd{\widetildebmEtaz}[2][]{\ensuremath{\subp{\widetilde{\bm{H}}}{}{#2}{}{#1}}}
\newrobustcmd{\acutebmEtaz}[2][]{\ensuremath{\subp{\acute{\bm{H}}}{}{#2}{}{#1}}}
\newrobustcmd{\gravebmEtaz}[2][]{\ensuremath{\subp{\grave{\bm{H}}}{}{#2}{}{#1}}}
\newrobustcmd{\dotbmEtaz}[2][]{\ensuremath{\subp{\dot{\bm{H}}}{}{#2}{}{#1}}}
\newrobustcmd{\ddotbmEtaz}[2][]{\ensuremath{\subp{\ddot{\bm{H}}}{}{#2}{}{#1}}}
\newrobustcmd{\brevebmEtaz}[2][]{\ensuremath{\subp{\breve{\bm{H}}}{}{#2}{}{#1}}}
\newrobustcmd{\barbmEtaz}[2][]{\ensuremath{\subp{\bar{\bm{H}}}{}{#2}{}{#1}}}
\newrobustcmd{\vecbmEtaz}[2][]{\ensuremath{\subp{\vec{\bm{H}}}{}{#2}{}{#1}}}
\newrobustcmd{\Thetaz}[2][]{\ensuremath{\subp{\Theta}{}{#2}{}{#1}}}
\newrobustcmd{\hatThetaz}[2][]{\ensuremath{\subp{\hat{\Theta}}{}{#2}{}{#1}}}
\newrobustcmd{\widehatThetaz}[2][]{\ensuremath{\subp{\widehat{\Theta}}{}{#2}{}{#1}}}
\newrobustcmd{\checkThetaz}[2][]{\ensuremath{\subp{\check{\Theta}}{}{#2}{}{#1}}}
\newrobustcmd{\tildeThetaz}[2][]{\ensuremath{\subp{\tilde{\Theta}}{}{#2}{}{#1}}}
\newrobustcmd{\widetildeThetaz}[2][]{\ensuremath{\subp{\widetilde{\Theta}}{}{#2}{}{#1}}}
\newrobustcmd{\acuteThetaz}[2][]{\ensuremath{\subp{\acute{\Theta}}{}{#2}{}{#1}}}
\newrobustcmd{\graveThetaz}[2][]{\ensuremath{\subp{\grave{\Theta}}{}{#2}{}{#1}}}
\newrobustcmd{\dotThetaz}[2][]{\ensuremath{\subp{\dot{\Theta}}{}{#2}{}{#1}}}
\newrobustcmd{\ddotThetaz}[2][]{\ensuremath{\subp{\ddot{\Theta}}{}{#2}{}{#1}}}
\newrobustcmd{\breveThetaz}[2][]{\ensuremath{\subp{\breve{\Theta}}{}{#2}{}{#1}}}
\newrobustcmd{\barThetaz}[2][]{\ensuremath{\subp{\bar{\Theta}}{}{#2}{}{#1}}}
\newrobustcmd{\vecThetaz}[2][]{\ensuremath{\subp{\vec{\Theta}}{}{#2}{}{#1}}}
\newrobustcmd{\bmThetaz}[2][]{\ensuremath{\subp{\bm{\Theta}}{}{#2}{}{#1}}}
\newrobustcmd{\hatbmThetaz}[2][]{\ensuremath{\subp{\hat{\bm{\Theta}}}{}{#2}{}{#1}}}
\newrobustcmd{\widehatbmThetaz}[2][]{\ensuremath{\subp{\widehat{\bm{\Theta}}}{}{#2}{}{#1}}}
\newrobustcmd{\checkbmThetaz}[2][]{\ensuremath{\subp{\check{\bm{\Theta}}}{}{#2}{}{#1}}}
\newrobustcmd{\tildebmThetaz}[2][]{\ensuremath{\subp{\tilde{\bm{\Theta}}}{}{#2}{}{#1}}}
\newrobustcmd{\widetildebmThetaz}[2][]{\ensuremath{\subp{\widetilde{\bm{\Theta}}}{}{#2}{}{#1}}}
\newrobustcmd{\acutebmThetaz}[2][]{\ensuremath{\subp{\acute{\bm{\Theta}}}{}{#2}{}{#1}}}
\newrobustcmd{\gravebmThetaz}[2][]{\ensuremath{\subp{\grave{\bm{\Theta}}}{}{#2}{}{#1}}}
\newrobustcmd{\dotbmThetaz}[2][]{\ensuremath{\subp{\dot{\bm{\Theta}}}{}{#2}{}{#1}}}
\newrobustcmd{\ddotbmThetaz}[2][]{\ensuremath{\subp{\ddot{\bm{\Theta}}}{}{#2}{}{#1}}}
\newrobustcmd{\brevebmThetaz}[2][]{\ensuremath{\subp{\breve{\bm{\Theta}}}{}{#2}{}{#1}}}
\newrobustcmd{\barbmThetaz}[2][]{\ensuremath{\subp{\bar{\bm{\Theta}}}{}{#2}{}{#1}}}
\newrobustcmd{\vecbmThetaz}[2][]{\ensuremath{\subp{\vec{\bm{\Theta}}}{}{#2}{}{#1}}}
\newrobustcmd{\Iotaz}[2][]{\ensuremath{\subp{I}{}{#2}{}{#1}}}
\newrobustcmd{\hatIotaz}[2][]{\ensuremath{\subp{\hat{I}}{}{#2}{}{#1}}}
\newrobustcmd{\widehatIotaz}[2][]{\ensuremath{\subp{\widehat{I}}{}{#2}{}{#1}}}
\newrobustcmd{\checkIotaz}[2][]{\ensuremath{\subp{\check{I}}{}{#2}{}{#1}}}
\newrobustcmd{\tildeIotaz}[2][]{\ensuremath{\subp{\tilde{I}}{}{#2}{}{#1}}}
\newrobustcmd{\widetildeIotaz}[2][]{\ensuremath{\subp{\widetilde{I}}{}{#2}{}{#1}}}
\newrobustcmd{\acuteIotaz}[2][]{\ensuremath{\subp{\acute{I}}{}{#2}{}{#1}}}
\newrobustcmd{\graveIotaz}[2][]{\ensuremath{\subp{\grave{I}}{}{#2}{}{#1}}}
\newrobustcmd{\dotIotaz}[2][]{\ensuremath{\subp{\dot{I}}{}{#2}{}{#1}}}
\newrobustcmd{\ddotIotaz}[2][]{\ensuremath{\subp{\ddot{I}}{}{#2}{}{#1}}}
\newrobustcmd{\breveIotaz}[2][]{\ensuremath{\subp{\breve{I}}{}{#2}{}{#1}}}
\newrobustcmd{\barIotaz}[2][]{\ensuremath{\subp{\bar{I}}{}{#2}{}{#1}}}
\newrobustcmd{\vecIotaz}[2][]{\ensuremath{\subp{\vec{I}}{}{#2}{}{#1}}}
\newrobustcmd{\bmIotaz}[2][]{\ensuremath{\subp{\bm{I}}{}{#2}{}{#1}}}
\newrobustcmd{\hatbmIotaz}[2][]{\ensuremath{\subp{\hat{\bm{I}}}{}{#2}{}{#1}}}
\newrobustcmd{\widehatbmIotaz}[2][]{\ensuremath{\subp{\widehat{\bm{I}}}{}{#2}{}{#1}}}
\newrobustcmd{\checkbmIotaz}[2][]{\ensuremath{\subp{\check{\bm{I}}}{}{#2}{}{#1}}}
\newrobustcmd{\tildebmIotaz}[2][]{\ensuremath{\subp{\tilde{\bm{I}}}{}{#2}{}{#1}}}
\newrobustcmd{\widetildebmIotaz}[2][]{\ensuremath{\subp{\widetilde{\bm{I}}}{}{#2}{}{#1}}}
\newrobustcmd{\acutebmIotaz}[2][]{\ensuremath{\subp{\acute{\bm{I}}}{}{#2}{}{#1}}}
\newrobustcmd{\gravebmIotaz}[2][]{\ensuremath{\subp{\grave{\bm{I}}}{}{#2}{}{#1}}}
\newrobustcmd{\dotbmIotaz}[2][]{\ensuremath{\subp{\dot{\bm{I}}}{}{#2}{}{#1}}}
\newrobustcmd{\ddotbmIotaz}[2][]{\ensuremath{\subp{\ddot{\bm{I}}}{}{#2}{}{#1}}}
\newrobustcmd{\brevebmIotaz}[2][]{\ensuremath{\subp{\breve{\bm{I}}}{}{#2}{}{#1}}}
\newrobustcmd{\barbmIotaz}[2][]{\ensuremath{\subp{\bar{\bm{I}}}{}{#2}{}{#1}}}
\newrobustcmd{\vecbmIotaz}[2][]{\ensuremath{\subp{\vec{\bm{I}}}{}{#2}{}{#1}}}
\newrobustcmd{\Kappaz}[2][]{\ensuremath{\subp{K}{}{#2}{}{#1}}}
\newrobustcmd{\hatKappaz}[2][]{\ensuremath{\subp{\hat{K}}{}{#2}{}{#1}}}
\newrobustcmd{\widehatKappaz}[2][]{\ensuremath{\subp{\widehat{K}}{}{#2}{}{#1}}}
\newrobustcmd{\checkKappaz}[2][]{\ensuremath{\subp{\check{K}}{}{#2}{}{#1}}}
\newrobustcmd{\tildeKappaz}[2][]{\ensuremath{\subp{\tilde{K}}{}{#2}{}{#1}}}
\newrobustcmd{\widetildeKappaz}[2][]{\ensuremath{\subp{\widetilde{K}}{}{#2}{}{#1}}}
\newrobustcmd{\acuteKappaz}[2][]{\ensuremath{\subp{\acute{K}}{}{#2}{}{#1}}}
\newrobustcmd{\graveKappaz}[2][]{\ensuremath{\subp{\grave{K}}{}{#2}{}{#1}}}
\newrobustcmd{\dotKappaz}[2][]{\ensuremath{\subp{\dot{K}}{}{#2}{}{#1}}}
\newrobustcmd{\ddotKappaz}[2][]{\ensuremath{\subp{\ddot{K}}{}{#2}{}{#1}}}
\newrobustcmd{\breveKappaz}[2][]{\ensuremath{\subp{\breve{K}}{}{#2}{}{#1}}}
\newrobustcmd{\barKappaz}[2][]{\ensuremath{\subp{\bar{K}}{}{#2}{}{#1}}}
\newrobustcmd{\vecKappaz}[2][]{\ensuremath{\subp{\vec{K}}{}{#2}{}{#1}}}
\newrobustcmd{\bmKappaz}[2][]{\ensuremath{\subp{\bm{K}}{}{#2}{}{#1}}}
\newrobustcmd{\hatbmKappaz}[2][]{\ensuremath{\subp{\hat{\bm{K}}}{}{#2}{}{#1}}}
\newrobustcmd{\widehatbmKappaz}[2][]{\ensuremath{\subp{\widehat{\bm{K}}}{}{#2}{}{#1}}}
\newrobustcmd{\checkbmKappaz}[2][]{\ensuremath{\subp{\check{\bm{K}}}{}{#2}{}{#1}}}
\newrobustcmd{\tildebmKappaz}[2][]{\ensuremath{\subp{\tilde{\bm{K}}}{}{#2}{}{#1}}}
\newrobustcmd{\widetildebmKappaz}[2][]{\ensuremath{\subp{\widetilde{\bm{K}}}{}{#2}{}{#1}}}
\newrobustcmd{\acutebmKappaz}[2][]{\ensuremath{\subp{\acute{\bm{K}}}{}{#2}{}{#1}}}
\newrobustcmd{\gravebmKappaz}[2][]{\ensuremath{\subp{\grave{\bm{K}}}{}{#2}{}{#1}}}
\newrobustcmd{\dotbmKappaz}[2][]{\ensuremath{\subp{\dot{\bm{K}}}{}{#2}{}{#1}}}
\newrobustcmd{\ddotbmKappaz}[2][]{\ensuremath{\subp{\ddot{\bm{K}}}{}{#2}{}{#1}}}
\newrobustcmd{\brevebmKappaz}[2][]{\ensuremath{\subp{\breve{\bm{K}}}{}{#2}{}{#1}}}
\newrobustcmd{\barbmKappaz}[2][]{\ensuremath{\subp{\bar{\bm{K}}}{}{#2}{}{#1}}}
\newrobustcmd{\vecbmKappaz}[2][]{\ensuremath{\subp{\vec{\bm{K}}}{}{#2}{}{#1}}}
\newrobustcmd{\Lambdaz}[2][]{\ensuremath{\subp{\Lambda}{}{#2}{}{#1}}}
\newrobustcmd{\hatLambdaz}[2][]{\ensuremath{\subp{\hat{\Lambda}}{}{#2}{}{#1}}}
\newrobustcmd{\widehatLambdaz}[2][]{\ensuremath{\subp{\widehat{\Lambda}}{}{#2}{}{#1}}}
\newrobustcmd{\checkLambdaz}[2][]{\ensuremath{\subp{\check{\Lambda}}{}{#2}{}{#1}}}
\newrobustcmd{\tildeLambdaz}[2][]{\ensuremath{\subp{\tilde{\Lambda}}{}{#2}{}{#1}}}
\newrobustcmd{\widetildeLambdaz}[2][]{\ensuremath{\subp{\widetilde{\Lambda}}{}{#2}{}{#1}}}
\newrobustcmd{\acuteLambdaz}[2][]{\ensuremath{\subp{\acute{\Lambda}}{}{#2}{}{#1}}}
\newrobustcmd{\graveLambdaz}[2][]{\ensuremath{\subp{\grave{\Lambda}}{}{#2}{}{#1}}}
\newrobustcmd{\dotLambdaz}[2][]{\ensuremath{\subp{\dot{\Lambda}}{}{#2}{}{#1}}}
\newrobustcmd{\ddotLambdaz}[2][]{\ensuremath{\subp{\ddot{\Lambda}}{}{#2}{}{#1}}}
\newrobustcmd{\breveLambdaz}[2][]{\ensuremath{\subp{\breve{\Lambda}}{}{#2}{}{#1}}}
\newrobustcmd{\barLambdaz}[2][]{\ensuremath{\subp{\bar{\Lambda}}{}{#2}{}{#1}}}
\newrobustcmd{\vecLambdaz}[2][]{\ensuremath{\subp{\vec{\Lambda}}{}{#2}{}{#1}}}
\newrobustcmd{\bmLambdaz}[2][]{\ensuremath{\subp{\bm{\Lambda}}{}{#2}{}{#1}}}
\newrobustcmd{\hatbmLambdaz}[2][]{\ensuremath{\subp{\hat{\bm{\Lambda}}}{}{#2}{}{#1}}}
\newrobustcmd{\widehatbmLambdaz}[2][]{\ensuremath{\subp{\widehat{\bm{\Lambda}}}{}{#2}{}{#1}}}
\newrobustcmd{\checkbmLambdaz}[2][]{\ensuremath{\subp{\check{\bm{\Lambda}}}{}{#2}{}{#1}}}
\newrobustcmd{\tildebmLambdaz}[2][]{\ensuremath{\subp{\tilde{\bm{\Lambda}}}{}{#2}{}{#1}}}
\newrobustcmd{\widetildebmLambdaz}[2][]{\ensuremath{\subp{\widetilde{\bm{\Lambda}}}{}{#2}{}{#1}}}
\newrobustcmd{\acutebmLambdaz}[2][]{\ensuremath{\subp{\acute{\bm{\Lambda}}}{}{#2}{}{#1}}}
\newrobustcmd{\gravebmLambdaz}[2][]{\ensuremath{\subp{\grave{\bm{\Lambda}}}{}{#2}{}{#1}}}
\newrobustcmd{\dotbmLambdaz}[2][]{\ensuremath{\subp{\dot{\bm{\Lambda}}}{}{#2}{}{#1}}}
\newrobustcmd{\ddotbmLambdaz}[2][]{\ensuremath{\subp{\ddot{\bm{\Lambda}}}{}{#2}{}{#1}}}
\newrobustcmd{\brevebmLambdaz}[2][]{\ensuremath{\subp{\breve{\bm{\Lambda}}}{}{#2}{}{#1}}}
\newrobustcmd{\barbmLambdaz}[2][]{\ensuremath{\subp{\bar{\bm{\Lambda}}}{}{#2}{}{#1}}}
\newrobustcmd{\vecbmLambdaz}[2][]{\ensuremath{\subp{\vec{\bm{\Lambda}}}{}{#2}{}{#1}}}
\newrobustcmd{\Muz}[2][]{\ensuremath{\subp{M}{}{#2}{}{#1}}}
\newrobustcmd{\hatMuz}[2][]{\ensuremath{\subp{\hat{M}}{}{#2}{}{#1}}}
\newrobustcmd{\widehatMuz}[2][]{\ensuremath{\subp{\widehat{M}}{}{#2}{}{#1}}}
\newrobustcmd{\checkMuz}[2][]{\ensuremath{\subp{\check{M}}{}{#2}{}{#1}}}
\newrobustcmd{\tildeMuz}[2][]{\ensuremath{\subp{\tilde{M}}{}{#2}{}{#1}}}
\newrobustcmd{\widetildeMuz}[2][]{\ensuremath{\subp{\widetilde{M}}{}{#2}{}{#1}}}
\newrobustcmd{\acuteMuz}[2][]{\ensuremath{\subp{\acute{M}}{}{#2}{}{#1}}}
\newrobustcmd{\graveMuz}[2][]{\ensuremath{\subp{\grave{M}}{}{#2}{}{#1}}}
\newrobustcmd{\dotMuz}[2][]{\ensuremath{\subp{\dot{M}}{}{#2}{}{#1}}}
\newrobustcmd{\ddotMuz}[2][]{\ensuremath{\subp{\ddot{M}}{}{#2}{}{#1}}}
\newrobustcmd{\breveMuz}[2][]{\ensuremath{\subp{\breve{M}}{}{#2}{}{#1}}}
\newrobustcmd{\barMuz}[2][]{\ensuremath{\subp{\bar{M}}{}{#2}{}{#1}}}
\newrobustcmd{\vecMuz}[2][]{\ensuremath{\subp{\vec{M}}{}{#2}{}{#1}}}
\newrobustcmd{\bmMuz}[2][]{\ensuremath{\subp{\bm{M}}{}{#2}{}{#1}}}
\newrobustcmd{\hatbmMuz}[2][]{\ensuremath{\subp{\hat{\bm{M}}}{}{#2}{}{#1}}}
\newrobustcmd{\widehatbmMuz}[2][]{\ensuremath{\subp{\widehat{\bm{M}}}{}{#2}{}{#1}}}
\newrobustcmd{\checkbmMuz}[2][]{\ensuremath{\subp{\check{\bm{M}}}{}{#2}{}{#1}}}
\newrobustcmd{\tildebmMuz}[2][]{\ensuremath{\subp{\tilde{\bm{M}}}{}{#2}{}{#1}}}
\newrobustcmd{\widetildebmMuz}[2][]{\ensuremath{\subp{\widetilde{\bm{M}}}{}{#2}{}{#1}}}
\newrobustcmd{\acutebmMuz}[2][]{\ensuremath{\subp{\acute{\bm{M}}}{}{#2}{}{#1}}}
\newrobustcmd{\gravebmMuz}[2][]{\ensuremath{\subp{\grave{\bm{M}}}{}{#2}{}{#1}}}
\newrobustcmd{\dotbmMuz}[2][]{\ensuremath{\subp{\dot{\bm{M}}}{}{#2}{}{#1}}}
\newrobustcmd{\ddotbmMuz}[2][]{\ensuremath{\subp{\ddot{\bm{M}}}{}{#2}{}{#1}}}
\newrobustcmd{\brevebmMuz}[2][]{\ensuremath{\subp{\breve{\bm{M}}}{}{#2}{}{#1}}}
\newrobustcmd{\barbmMuz}[2][]{\ensuremath{\subp{\bar{\bm{M}}}{}{#2}{}{#1}}}
\newrobustcmd{\vecbmMuz}[2][]{\ensuremath{\subp{\vec{\bm{M}}}{}{#2}{}{#1}}}
\newrobustcmd{\Nuz}[2][]{\ensuremath{\subp{N}{}{#2}{}{#1}}}
\newrobustcmd{\hatNuz}[2][]{\ensuremath{\subp{\hat{N}}{}{#2}{}{#1}}}
\newrobustcmd{\widehatNuz}[2][]{\ensuremath{\subp{\widehat{N}}{}{#2}{}{#1}}}
\newrobustcmd{\checkNuz}[2][]{\ensuremath{\subp{\check{N}}{}{#2}{}{#1}}}
\newrobustcmd{\tildeNuz}[2][]{\ensuremath{\subp{\tilde{N}}{}{#2}{}{#1}}}
\newrobustcmd{\widetildeNuz}[2][]{\ensuremath{\subp{\widetilde{N}}{}{#2}{}{#1}}}
\newrobustcmd{\acuteNuz}[2][]{\ensuremath{\subp{\acute{N}}{}{#2}{}{#1}}}
\newrobustcmd{\graveNuz}[2][]{\ensuremath{\subp{\grave{N}}{}{#2}{}{#1}}}
\newrobustcmd{\dotNuz}[2][]{\ensuremath{\subp{\dot{N}}{}{#2}{}{#1}}}
\newrobustcmd{\ddotNuz}[2][]{\ensuremath{\subp{\ddot{N}}{}{#2}{}{#1}}}
\newrobustcmd{\breveNuz}[2][]{\ensuremath{\subp{\breve{N}}{}{#2}{}{#1}}}
\newrobustcmd{\barNuz}[2][]{\ensuremath{\subp{\bar{N}}{}{#2}{}{#1}}}
\newrobustcmd{\vecNuz}[2][]{\ensuremath{\subp{\vec{N}}{}{#2}{}{#1}}}
\newrobustcmd{\bmNuz}[2][]{\ensuremath{\subp{\bm{N}}{}{#2}{}{#1}}}
\newrobustcmd{\hatbmNuz}[2][]{\ensuremath{\subp{\hat{\bm{N}}}{}{#2}{}{#1}}}
\newrobustcmd{\widehatbmNuz}[2][]{\ensuremath{\subp{\widehat{\bm{N}}}{}{#2}{}{#1}}}
\newrobustcmd{\checkbmNuz}[2][]{\ensuremath{\subp{\check{\bm{N}}}{}{#2}{}{#1}}}
\newrobustcmd{\tildebmNuz}[2][]{\ensuremath{\subp{\tilde{\bm{N}}}{}{#2}{}{#1}}}
\newrobustcmd{\widetildebmNuz}[2][]{\ensuremath{\subp{\widetilde{\bm{N}}}{}{#2}{}{#1}}}
\newrobustcmd{\acutebmNuz}[2][]{\ensuremath{\subp{\acute{\bm{N}}}{}{#2}{}{#1}}}
\newrobustcmd{\gravebmNuz}[2][]{\ensuremath{\subp{\grave{\bm{N}}}{}{#2}{}{#1}}}
\newrobustcmd{\dotbmNuz}[2][]{\ensuremath{\subp{\dot{\bm{N}}}{}{#2}{}{#1}}}
\newrobustcmd{\ddotbmNuz}[2][]{\ensuremath{\subp{\ddot{\bm{N}}}{}{#2}{}{#1}}}
\newrobustcmd{\brevebmNuz}[2][]{\ensuremath{\subp{\breve{\bm{N}}}{}{#2}{}{#1}}}
\newrobustcmd{\barbmNuz}[2][]{\ensuremath{\subp{\bar{\bm{N}}}{}{#2}{}{#1}}}
\newrobustcmd{\vecbmNuz}[2][]{\ensuremath{\subp{\vec{\bm{N}}}{}{#2}{}{#1}}}
\newrobustcmd{\Xiz}[2][]{\ensuremath{\subp{\Xi}{}{#2}{}{#1}}}
\newrobustcmd{\hatXiz}[2][]{\ensuremath{\subp{\hat{\Xi}}{}{#2}{}{#1}}}
\newrobustcmd{\widehatXiz}[2][]{\ensuremath{\subp{\widehat{\Xi}}{}{#2}{}{#1}}}
\newrobustcmd{\checkXiz}[2][]{\ensuremath{\subp{\check{\Xi}}{}{#2}{}{#1}}}
\newrobustcmd{\tildeXiz}[2][]{\ensuremath{\subp{\tilde{\Xi}}{}{#2}{}{#1}}}
\newrobustcmd{\widetildeXiz}[2][]{\ensuremath{\subp{\widetilde{\Xi}}{}{#2}{}{#1}}}
\newrobustcmd{\acuteXiz}[2][]{\ensuremath{\subp{\acute{\Xi}}{}{#2}{}{#1}}}
\newrobustcmd{\graveXiz}[2][]{\ensuremath{\subp{\grave{\Xi}}{}{#2}{}{#1}}}
\newrobustcmd{\dotXiz}[2][]{\ensuremath{\subp{\dot{\Xi}}{}{#2}{}{#1}}}
\newrobustcmd{\ddotXiz}[2][]{\ensuremath{\subp{\ddot{\Xi}}{}{#2}{}{#1}}}
\newrobustcmd{\breveXiz}[2][]{\ensuremath{\subp{\breve{\Xi}}{}{#2}{}{#1}}}
\newrobustcmd{\barXiz}[2][]{\ensuremath{\subp{\bar{\Xi}}{}{#2}{}{#1}}}
\newrobustcmd{\vecXiz}[2][]{\ensuremath{\subp{\vec{\Xi}}{}{#2}{}{#1}}}
\newrobustcmd{\bmXiz}[2][]{\ensuremath{\subp{\bm{\Xi}}{}{#2}{}{#1}}}
\newrobustcmd{\hatbmXiz}[2][]{\ensuremath{\subp{\hat{\bm{\Xi}}}{}{#2}{}{#1}}}
\newrobustcmd{\widehatbmXiz}[2][]{\ensuremath{\subp{\widehat{\bm{\Xi}}}{}{#2}{}{#1}}}
\newrobustcmd{\checkbmXiz}[2][]{\ensuremath{\subp{\check{\bm{\Xi}}}{}{#2}{}{#1}}}
\newrobustcmd{\tildebmXiz}[2][]{\ensuremath{\subp{\tilde{\bm{\Xi}}}{}{#2}{}{#1}}}
\newrobustcmd{\widetildebmXiz}[2][]{\ensuremath{\subp{\widetilde{\bm{\Xi}}}{}{#2}{}{#1}}}
\newrobustcmd{\acutebmXiz}[2][]{\ensuremath{\subp{\acute{\bm{\Xi}}}{}{#2}{}{#1}}}
\newrobustcmd{\gravebmXiz}[2][]{\ensuremath{\subp{\grave{\bm{\Xi}}}{}{#2}{}{#1}}}
\newrobustcmd{\dotbmXiz}[2][]{\ensuremath{\subp{\dot{\bm{\Xi}}}{}{#2}{}{#1}}}
\newrobustcmd{\ddotbmXiz}[2][]{\ensuremath{\subp{\ddot{\bm{\Xi}}}{}{#2}{}{#1}}}
\newrobustcmd{\brevebmXiz}[2][]{\ensuremath{\subp{\breve{\bm{\Xi}}}{}{#2}{}{#1}}}
\newrobustcmd{\barbmXiz}[2][]{\ensuremath{\subp{\bar{\bm{\Xi}}}{}{#2}{}{#1}}}
\newrobustcmd{\vecbmXiz}[2][]{\ensuremath{\subp{\vec{\bm{\Xi}}}{}{#2}{}{#1}}}
\newrobustcmd{\Piz}[2][]{\ensuremath{\subp{\Pi}{}{#2}{}{#1}}}
\newrobustcmd{\hatPiz}[2][]{\ensuremath{\subp{\hat{\Pi}}{}{#2}{}{#1}}}
\newrobustcmd{\widehatPiz}[2][]{\ensuremath{\subp{\widehat{\Pi}}{}{#2}{}{#1}}}
\newrobustcmd{\checkPiz}[2][]{\ensuremath{\subp{\check{\Pi}}{}{#2}{}{#1}}}
\newrobustcmd{\tildePiz}[2][]{\ensuremath{\subp{\tilde{\Pi}}{}{#2}{}{#1}}}
\newrobustcmd{\widetildePiz}[2][]{\ensuremath{\subp{\widetilde{\Pi}}{}{#2}{}{#1}}}
\newrobustcmd{\acutePiz}[2][]{\ensuremath{\subp{\acute{\Pi}}{}{#2}{}{#1}}}
\newrobustcmd{\gravePiz}[2][]{\ensuremath{\subp{\grave{\Pi}}{}{#2}{}{#1}}}
\newrobustcmd{\dotPiz}[2][]{\ensuremath{\subp{\dot{\Pi}}{}{#2}{}{#1}}}
\newrobustcmd{\ddotPiz}[2][]{\ensuremath{\subp{\ddot{\Pi}}{}{#2}{}{#1}}}
\newrobustcmd{\brevePiz}[2][]{\ensuremath{\subp{\breve{\Pi}}{}{#2}{}{#1}}}
\newrobustcmd{\barPiz}[2][]{\ensuremath{\subp{\bar{\Pi}}{}{#2}{}{#1}}}
\newrobustcmd{\vecPiz}[2][]{\ensuremath{\subp{\vec{\Pi}}{}{#2}{}{#1}}}
\newrobustcmd{\bmPiz}[2][]{\ensuremath{\subp{\bm{\Pi}}{}{#2}{}{#1}}}
\newrobustcmd{\hatbmPiz}[2][]{\ensuremath{\subp{\hat{\bm{\Pi}}}{}{#2}{}{#1}}}
\newrobustcmd{\widehatbmPiz}[2][]{\ensuremath{\subp{\widehat{\bm{\Pi}}}{}{#2}{}{#1}}}
\newrobustcmd{\checkbmPiz}[2][]{\ensuremath{\subp{\check{\bm{\Pi}}}{}{#2}{}{#1}}}
\newrobustcmd{\tildebmPiz}[2][]{\ensuremath{\subp{\tilde{\bm{\Pi}}}{}{#2}{}{#1}}}
\newrobustcmd{\widetildebmPiz}[2][]{\ensuremath{\subp{\widetilde{\bm{\Pi}}}{}{#2}{}{#1}}}
\newrobustcmd{\acutebmPiz}[2][]{\ensuremath{\subp{\acute{\bm{\Pi}}}{}{#2}{}{#1}}}
\newrobustcmd{\gravebmPiz}[2][]{\ensuremath{\subp{\grave{\bm{\Pi}}}{}{#2}{}{#1}}}
\newrobustcmd{\dotbmPiz}[2][]{\ensuremath{\subp{\dot{\bm{\Pi}}}{}{#2}{}{#1}}}
\newrobustcmd{\ddotbmPiz}[2][]{\ensuremath{\subp{\ddot{\bm{\Pi}}}{}{#2}{}{#1}}}
\newrobustcmd{\brevebmPiz}[2][]{\ensuremath{\subp{\breve{\bm{\Pi}}}{}{#2}{}{#1}}}
\newrobustcmd{\barbmPiz}[2][]{\ensuremath{\subp{\bar{\bm{\Pi}}}{}{#2}{}{#1}}}
\newrobustcmd{\vecbmPiz}[2][]{\ensuremath{\subp{\vec{\bm{\Pi}}}{}{#2}{}{#1}}}
\newrobustcmd{\Rhoz}[2][]{\ensuremath{\subp{R}{}{#2}{}{#1}}}
\newrobustcmd{\hatRhoz}[2][]{\ensuremath{\subp{\hat{R}}{}{#2}{}{#1}}}
\newrobustcmd{\widehatRhoz}[2][]{\ensuremath{\subp{\widehat{R}}{}{#2}{}{#1}}}
\newrobustcmd{\checkRhoz}[2][]{\ensuremath{\subp{\check{R}}{}{#2}{}{#1}}}
\newrobustcmd{\tildeRhoz}[2][]{\ensuremath{\subp{\tilde{R}}{}{#2}{}{#1}}}
\newrobustcmd{\widetildeRhoz}[2][]{\ensuremath{\subp{\widetilde{R}}{}{#2}{}{#1}}}
\newrobustcmd{\acuteRhoz}[2][]{\ensuremath{\subp{\acute{R}}{}{#2}{}{#1}}}
\newrobustcmd{\graveRhoz}[2][]{\ensuremath{\subp{\grave{R}}{}{#2}{}{#1}}}
\newrobustcmd{\dotRhoz}[2][]{\ensuremath{\subp{\dot{R}}{}{#2}{}{#1}}}
\newrobustcmd{\ddotRhoz}[2][]{\ensuremath{\subp{\ddot{R}}{}{#2}{}{#1}}}
\newrobustcmd{\breveRhoz}[2][]{\ensuremath{\subp{\breve{R}}{}{#2}{}{#1}}}
\newrobustcmd{\barRhoz}[2][]{\ensuremath{\subp{\bar{R}}{}{#2}{}{#1}}}
\newrobustcmd{\vecRhoz}[2][]{\ensuremath{\subp{\vec{R}}{}{#2}{}{#1}}}
\newrobustcmd{\bmRhoz}[2][]{\ensuremath{\subp{\bm{R}}{}{#2}{}{#1}}}
\newrobustcmd{\hatbmRhoz}[2][]{\ensuremath{\subp{\hat{\bm{R}}}{}{#2}{}{#1}}}
\newrobustcmd{\widehatbmRhoz}[2][]{\ensuremath{\subp{\widehat{\bm{R}}}{}{#2}{}{#1}}}
\newrobustcmd{\checkbmRhoz}[2][]{\ensuremath{\subp{\check{\bm{R}}}{}{#2}{}{#1}}}
\newrobustcmd{\tildebmRhoz}[2][]{\ensuremath{\subp{\tilde{\bm{R}}}{}{#2}{}{#1}}}
\newrobustcmd{\widetildebmRhoz}[2][]{\ensuremath{\subp{\widetilde{\bm{R}}}{}{#2}{}{#1}}}
\newrobustcmd{\acutebmRhoz}[2][]{\ensuremath{\subp{\acute{\bm{R}}}{}{#2}{}{#1}}}
\newrobustcmd{\gravebmRhoz}[2][]{\ensuremath{\subp{\grave{\bm{R}}}{}{#2}{}{#1}}}
\newrobustcmd{\dotbmRhoz}[2][]{\ensuremath{\subp{\dot{\bm{R}}}{}{#2}{}{#1}}}
\newrobustcmd{\ddotbmRhoz}[2][]{\ensuremath{\subp{\ddot{\bm{R}}}{}{#2}{}{#1}}}
\newrobustcmd{\brevebmRhoz}[2][]{\ensuremath{\subp{\breve{\bm{R}}}{}{#2}{}{#1}}}
\newrobustcmd{\barbmRhoz}[2][]{\ensuremath{\subp{\bar{\bm{R}}}{}{#2}{}{#1}}}
\newrobustcmd{\vecbmRhoz}[2][]{\ensuremath{\subp{\vec{\bm{R}}}{}{#2}{}{#1}}}
\newrobustcmd{\Sigmaz}[2][]{\ensuremath{\subp{\Sigma}{}{#2}{}{#1}}}
\newrobustcmd{\hatSigmaz}[2][]{\ensuremath{\subp{\hat{\Sigma}}{}{#2}{}{#1}}}
\newrobustcmd{\widehatSigmaz}[2][]{\ensuremath{\subp{\widehat{\Sigma}}{}{#2}{}{#1}}}
\newrobustcmd{\checkSigmaz}[2][]{\ensuremath{\subp{\check{\Sigma}}{}{#2}{}{#1}}}
\newrobustcmd{\tildeSigmaz}[2][]{\ensuremath{\subp{\tilde{\Sigma}}{}{#2}{}{#1}}}
\newrobustcmd{\widetildeSigmaz}[2][]{\ensuremath{\subp{\widetilde{\Sigma}}{}{#2}{}{#1}}}
\newrobustcmd{\acuteSigmaz}[2][]{\ensuremath{\subp{\acute{\Sigma}}{}{#2}{}{#1}}}
\newrobustcmd{\graveSigmaz}[2][]{\ensuremath{\subp{\grave{\Sigma}}{}{#2}{}{#1}}}
\newrobustcmd{\dotSigmaz}[2][]{\ensuremath{\subp{\dot{\Sigma}}{}{#2}{}{#1}}}
\newrobustcmd{\ddotSigmaz}[2][]{\ensuremath{\subp{\ddot{\Sigma}}{}{#2}{}{#1}}}
\newrobustcmd{\breveSigmaz}[2][]{\ensuremath{\subp{\breve{\Sigma}}{}{#2}{}{#1}}}
\newrobustcmd{\barSigmaz}[2][]{\ensuremath{\subp{\bar{\Sigma}}{}{#2}{}{#1}}}
\newrobustcmd{\vecSigmaz}[2][]{\ensuremath{\subp{\vec{\Sigma}}{}{#2}{}{#1}}}
\newrobustcmd{\bmSigmaz}[2][]{\ensuremath{\subp{\bm{\Sigma}}{}{#2}{}{#1}}}
\newrobustcmd{\hatbmSigmaz}[2][]{\ensuremath{\subp{\hat{\bm{\Sigma}}}{}{#2}{}{#1}}}
\newrobustcmd{\widehatbmSigmaz}[2][]{\ensuremath{\subp{\widehat{\bm{\Sigma}}}{}{#2}{}{#1}}}
\newrobustcmd{\checkbmSigmaz}[2][]{\ensuremath{\subp{\check{\bm{\Sigma}}}{}{#2}{}{#1}}}
\newrobustcmd{\tildebmSigmaz}[2][]{\ensuremath{\subp{\tilde{\bm{\Sigma}}}{}{#2}{}{#1}}}
\newrobustcmd{\widetildebmSigmaz}[2][]{\ensuremath{\subp{\widetilde{\bm{\Sigma}}}{}{#2}{}{#1}}}
\newrobustcmd{\acutebmSigmaz}[2][]{\ensuremath{\subp{\acute{\bm{\Sigma}}}{}{#2}{}{#1}}}
\newrobustcmd{\gravebmSigmaz}[2][]{\ensuremath{\subp{\grave{\bm{\Sigma}}}{}{#2}{}{#1}}}
\newrobustcmd{\dotbmSigmaz}[2][]{\ensuremath{\subp{\dot{\bm{\Sigma}}}{}{#2}{}{#1}}}
\newrobustcmd{\ddotbmSigmaz}[2][]{\ensuremath{\subp{\ddot{\bm{\Sigma}}}{}{#2}{}{#1}}}
\newrobustcmd{\brevebmSigmaz}[2][]{\ensuremath{\subp{\breve{\bm{\Sigma}}}{}{#2}{}{#1}}}
\newrobustcmd{\barbmSigmaz}[2][]{\ensuremath{\subp{\bar{\bm{\Sigma}}}{}{#2}{}{#1}}}
\newrobustcmd{\vecbmSigmaz}[2][]{\ensuremath{\subp{\vec{\bm{\Sigma}}}{}{#2}{}{#1}}}
\newrobustcmd{\Tauz}[2][]{\ensuremath{\subp{T}{}{#2}{}{#1}}}
\newrobustcmd{\hatTauz}[2][]{\ensuremath{\subp{\hat{T}}{}{#2}{}{#1}}}
\newrobustcmd{\widehatTauz}[2][]{\ensuremath{\subp{\widehat{T}}{}{#2}{}{#1}}}
\newrobustcmd{\checkTauz}[2][]{\ensuremath{\subp{\check{T}}{}{#2}{}{#1}}}
\newrobustcmd{\tildeTauz}[2][]{\ensuremath{\subp{\tilde{T}}{}{#2}{}{#1}}}
\newrobustcmd{\widetildeTauz}[2][]{\ensuremath{\subp{\widetilde{T}}{}{#2}{}{#1}}}
\newrobustcmd{\acuteTauz}[2][]{\ensuremath{\subp{\acute{T}}{}{#2}{}{#1}}}
\newrobustcmd{\graveTauz}[2][]{\ensuremath{\subp{\grave{T}}{}{#2}{}{#1}}}
\newrobustcmd{\dotTauz}[2][]{\ensuremath{\subp{\dot{T}}{}{#2}{}{#1}}}
\newrobustcmd{\ddotTauz}[2][]{\ensuremath{\subp{\ddot{T}}{}{#2}{}{#1}}}
\newrobustcmd{\breveTauz}[2][]{\ensuremath{\subp{\breve{T}}{}{#2}{}{#1}}}
\newrobustcmd{\barTauz}[2][]{\ensuremath{\subp{\bar{T}}{}{#2}{}{#1}}}
\newrobustcmd{\vecTauz}[2][]{\ensuremath{\subp{\vec{T}}{}{#2}{}{#1}}}
\newrobustcmd{\bmTauz}[2][]{\ensuremath{\subp{\bm{T}}{}{#2}{}{#1}}}
\newrobustcmd{\hatbmTauz}[2][]{\ensuremath{\subp{\hat{\bm{T}}}{}{#2}{}{#1}}}
\newrobustcmd{\widehatbmTauz}[2][]{\ensuremath{\subp{\widehat{\bm{T}}}{}{#2}{}{#1}}}
\newrobustcmd{\checkbmTauz}[2][]{\ensuremath{\subp{\check{\bm{T}}}{}{#2}{}{#1}}}
\newrobustcmd{\tildebmTauz}[2][]{\ensuremath{\subp{\tilde{\bm{T}}}{}{#2}{}{#1}}}
\newrobustcmd{\widetildebmTauz}[2][]{\ensuremath{\subp{\widetilde{\bm{T}}}{}{#2}{}{#1}}}
\newrobustcmd{\acutebmTauz}[2][]{\ensuremath{\subp{\acute{\bm{T}}}{}{#2}{}{#1}}}
\newrobustcmd{\gravebmTauz}[2][]{\ensuremath{\subp{\grave{\bm{T}}}{}{#2}{}{#1}}}
\newrobustcmd{\dotbmTauz}[2][]{\ensuremath{\subp{\dot{\bm{T}}}{}{#2}{}{#1}}}
\newrobustcmd{\ddotbmTauz}[2][]{\ensuremath{\subp{\ddot{\bm{T}}}{}{#2}{}{#1}}}
\newrobustcmd{\brevebmTauz}[2][]{\ensuremath{\subp{\breve{\bm{T}}}{}{#2}{}{#1}}}
\newrobustcmd{\barbmTauz}[2][]{\ensuremath{\subp{\bar{\bm{T}}}{}{#2}{}{#1}}}
\newrobustcmd{\vecbmTauz}[2][]{\ensuremath{\subp{\vec{\bm{T}}}{}{#2}{}{#1}}}
\newrobustcmd{\Upsilonz}[2][]{\ensuremath{\subp{\Upsilon}{}{#2}{}{#1}}}
\newrobustcmd{\hatUpsilonz}[2][]{\ensuremath{\subp{\hat{\Upsilon}}{}{#2}{}{#1}}}
\newrobustcmd{\widehatUpsilonz}[2][]{\ensuremath{\subp{\widehat{\Upsilon}}{}{#2}{}{#1}}}
\newrobustcmd{\checkUpsilonz}[2][]{\ensuremath{\subp{\check{\Upsilon}}{}{#2}{}{#1}}}
\newrobustcmd{\tildeUpsilonz}[2][]{\ensuremath{\subp{\tilde{\Upsilon}}{}{#2}{}{#1}}}
\newrobustcmd{\widetildeUpsilonz}[2][]{\ensuremath{\subp{\widetilde{\Upsilon}}{}{#2}{}{#1}}}
\newrobustcmd{\acuteUpsilonz}[2][]{\ensuremath{\subp{\acute{\Upsilon}}{}{#2}{}{#1}}}
\newrobustcmd{\graveUpsilonz}[2][]{\ensuremath{\subp{\grave{\Upsilon}}{}{#2}{}{#1}}}
\newrobustcmd{\dotUpsilonz}[2][]{\ensuremath{\subp{\dot{\Upsilon}}{}{#2}{}{#1}}}
\newrobustcmd{\ddotUpsilonz}[2][]{\ensuremath{\subp{\ddot{\Upsilon}}{}{#2}{}{#1}}}
\newrobustcmd{\breveUpsilonz}[2][]{\ensuremath{\subp{\breve{\Upsilon}}{}{#2}{}{#1}}}
\newrobustcmd{\barUpsilonz}[2][]{\ensuremath{\subp{\bar{\Upsilon}}{}{#2}{}{#1}}}
\newrobustcmd{\vecUpsilonz}[2][]{\ensuremath{\subp{\vec{\Upsilon}}{}{#2}{}{#1}}}
\newrobustcmd{\bmUpsilonz}[2][]{\ensuremath{\subp{\bm{\Upsilon}}{}{#2}{}{#1}}}
\newrobustcmd{\hatbmUpsilonz}[2][]{\ensuremath{\subp{\hat{\bm{\Upsilon}}}{}{#2}{}{#1}}}
\newrobustcmd{\widehatbmUpsilonz}[2][]{\ensuremath{\subp{\widehat{\bm{\Upsilon}}}{}{#2}{}{#1}}}
\newrobustcmd{\checkbmUpsilonz}[2][]{\ensuremath{\subp{\check{\bm{\Upsilon}}}{}{#2}{}{#1}}}
\newrobustcmd{\tildebmUpsilonz}[2][]{\ensuremath{\subp{\tilde{\bm{\Upsilon}}}{}{#2}{}{#1}}}
\newrobustcmd{\widetildebmUpsilonz}[2][]{\ensuremath{\subp{\widetilde{\bm{\Upsilon}}}{}{#2}{}{#1}}}
\newrobustcmd{\acutebmUpsilonz}[2][]{\ensuremath{\subp{\acute{\bm{\Upsilon}}}{}{#2}{}{#1}}}
\newrobustcmd{\gravebmUpsilonz}[2][]{\ensuremath{\subp{\grave{\bm{\Upsilon}}}{}{#2}{}{#1}}}
\newrobustcmd{\dotbmUpsilonz}[2][]{\ensuremath{\subp{\dot{\bm{\Upsilon}}}{}{#2}{}{#1}}}
\newrobustcmd{\ddotbmUpsilonz}[2][]{\ensuremath{\subp{\ddot{\bm{\Upsilon}}}{}{#2}{}{#1}}}
\newrobustcmd{\brevebmUpsilonz}[2][]{\ensuremath{\subp{\breve{\bm{\Upsilon}}}{}{#2}{}{#1}}}
\newrobustcmd{\barbmUpsilonz}[2][]{\ensuremath{\subp{\bar{\bm{\Upsilon}}}{}{#2}{}{#1}}}
\newrobustcmd{\vecbmUpsilonz}[2][]{\ensuremath{\subp{\vec{\bm{\Upsilon}}}{}{#2}{}{#1}}}
\newrobustcmd{\Phiz}[2][]{\ensuremath{\subp{\Phi}{}{#2}{}{#1}}}
\newrobustcmd{\hatPhiz}[2][]{\ensuremath{\subp{\hat{\Phi}}{}{#2}{}{#1}}}
\newrobustcmd{\widehatPhiz}[2][]{\ensuremath{\subp{\widehat{\Phi}}{}{#2}{}{#1}}}
\newrobustcmd{\checkPhiz}[2][]{\ensuremath{\subp{\check{\Phi}}{}{#2}{}{#1}}}
\newrobustcmd{\tildePhiz}[2][]{\ensuremath{\subp{\tilde{\Phi}}{}{#2}{}{#1}}}
\newrobustcmd{\widetildePhiz}[2][]{\ensuremath{\subp{\widetilde{\Phi}}{}{#2}{}{#1}}}
\newrobustcmd{\acutePhiz}[2][]{\ensuremath{\subp{\acute{\Phi}}{}{#2}{}{#1}}}
\newrobustcmd{\gravePhiz}[2][]{\ensuremath{\subp{\grave{\Phi}}{}{#2}{}{#1}}}
\newrobustcmd{\dotPhiz}[2][]{\ensuremath{\subp{\dot{\Phi}}{}{#2}{}{#1}}}
\newrobustcmd{\ddotPhiz}[2][]{\ensuremath{\subp{\ddot{\Phi}}{}{#2}{}{#1}}}
\newrobustcmd{\brevePhiz}[2][]{\ensuremath{\subp{\breve{\Phi}}{}{#2}{}{#1}}}
\newrobustcmd{\barPhiz}[2][]{\ensuremath{\subp{\bar{\Phi}}{}{#2}{}{#1}}}
\newrobustcmd{\vecPhiz}[2][]{\ensuremath{\subp{\vec{\Phi}}{}{#2}{}{#1}}}
\newrobustcmd{\bmPhiz}[2][]{\ensuremath{\subp{\bm{\Phi}}{}{#2}{}{#1}}}
\newrobustcmd{\hatbmPhiz}[2][]{\ensuremath{\subp{\hat{\bm{\Phi}}}{}{#2}{}{#1}}}
\newrobustcmd{\widehatbmPhiz}[2][]{\ensuremath{\subp{\widehat{\bm{\Phi}}}{}{#2}{}{#1}}}
\newrobustcmd{\checkbmPhiz}[2][]{\ensuremath{\subp{\check{\bm{\Phi}}}{}{#2}{}{#1}}}
\newrobustcmd{\tildebmPhiz}[2][]{\ensuremath{\subp{\tilde{\bm{\Phi}}}{}{#2}{}{#1}}}
\newrobustcmd{\widetildebmPhiz}[2][]{\ensuremath{\subp{\widetilde{\bm{\Phi}}}{}{#2}{}{#1}}}
\newrobustcmd{\acutebmPhiz}[2][]{\ensuremath{\subp{\acute{\bm{\Phi}}}{}{#2}{}{#1}}}
\newrobustcmd{\gravebmPhiz}[2][]{\ensuremath{\subp{\grave{\bm{\Phi}}}{}{#2}{}{#1}}}
\newrobustcmd{\dotbmPhiz}[2][]{\ensuremath{\subp{\dot{\bm{\Phi}}}{}{#2}{}{#1}}}
\newrobustcmd{\ddotbmPhiz}[2][]{\ensuremath{\subp{\ddot{\bm{\Phi}}}{}{#2}{}{#1}}}
\newrobustcmd{\brevebmPhiz}[2][]{\ensuremath{\subp{\breve{\bm{\Phi}}}{}{#2}{}{#1}}}
\newrobustcmd{\barbmPhiz}[2][]{\ensuremath{\subp{\bar{\bm{\Phi}}}{}{#2}{}{#1}}}
\newrobustcmd{\vecbmPhiz}[2][]{\ensuremath{\subp{\vec{\bm{\Phi}}}{}{#2}{}{#1}}}
\newrobustcmd{\Chiz}[2][]{\ensuremath{\subp{\Chi}{}{#2}{}{#1}}}
\newrobustcmd{\hatChiz}[2][]{\ensuremath{\subp{\hat{\Chi}}{}{#2}{}{#1}}}
\newrobustcmd{\widehatChiz}[2][]{\ensuremath{\subp{\widehat{\Chi}}{}{#2}{}{#1}}}
\newrobustcmd{\checkChiz}[2][]{\ensuremath{\subp{\check{\Chi}}{}{#2}{}{#1}}}
\newrobustcmd{\tildeChiz}[2][]{\ensuremath{\subp{\tilde{\Chi}}{}{#2}{}{#1}}}
\newrobustcmd{\widetildeChiz}[2][]{\ensuremath{\subp{\widetilde{\Chi}}{}{#2}{}{#1}}}
\newrobustcmd{\acuteChiz}[2][]{\ensuremath{\subp{\acute{\Chi}}{}{#2}{}{#1}}}
\newrobustcmd{\graveChiz}[2][]{\ensuremath{\subp{\grave{\Chi}}{}{#2}{}{#1}}}
\newrobustcmd{\dotChiz}[2][]{\ensuremath{\subp{\dot{\Chi}}{}{#2}{}{#1}}}
\newrobustcmd{\ddotChiz}[2][]{\ensuremath{\subp{\ddot{\Chi}}{}{#2}{}{#1}}}
\newrobustcmd{\breveChiz}[2][]{\ensuremath{\subp{\breve{\Chi}}{}{#2}{}{#1}}}
\newrobustcmd{\barChiz}[2][]{\ensuremath{\subp{\bar{\Chi}}{}{#2}{}{#1}}}
\newrobustcmd{\vecChiz}[2][]{\ensuremath{\subp{\vec{\Chi}}{}{#2}{}{#1}}}
\newrobustcmd{\bmChiz}[2][]{\ensuremath{\subp{\bm{\Chi}}{}{#2}{}{#1}}}
\newrobustcmd{\hatbmChiz}[2][]{\ensuremath{\subp{\hat{\bm{\Chi}}}{}{#2}{}{#1}}}
\newrobustcmd{\widehatbmChiz}[2][]{\ensuremath{\subp{\widehat{\bm{\Chi}}}{}{#2}{}{#1}}}
\newrobustcmd{\checkbmChiz}[2][]{\ensuremath{\subp{\check{\bm{\Chi}}}{}{#2}{}{#1}}}
\newrobustcmd{\tildebmChiz}[2][]{\ensuremath{\subp{\tilde{\bm{\Chi}}}{}{#2}{}{#1}}}
\newrobustcmd{\widetildebmChiz}[2][]{\ensuremath{\subp{\widetilde{\bm{\Chi}}}{}{#2}{}{#1}}}
\newrobustcmd{\acutebmChiz}[2][]{\ensuremath{\subp{\acute{\bm{\Chi}}}{}{#2}{}{#1}}}
\newrobustcmd{\gravebmChiz}[2][]{\ensuremath{\subp{\grave{\bm{\Chi}}}{}{#2}{}{#1}}}
\newrobustcmd{\dotbmChiz}[2][]{\ensuremath{\subp{\dot{\bm{\Chi}}}{}{#2}{}{#1}}}
\newrobustcmd{\ddotbmChiz}[2][]{\ensuremath{\subp{\ddot{\bm{\Chi}}}{}{#2}{}{#1}}}
\newrobustcmd{\brevebmChiz}[2][]{\ensuremath{\subp{\breve{\bm{\Chi}}}{}{#2}{}{#1}}}
\newrobustcmd{\barbmChiz}[2][]{\ensuremath{\subp{\bar{\bm{\Chi}}}{}{#2}{}{#1}}}
\newrobustcmd{\vecbmChiz}[2][]{\ensuremath{\subp{\vec{\bm{\Chi}}}{}{#2}{}{#1}}}
\newrobustcmd{\Psiz}[2][]{\ensuremath{\subp{\Psi}{}{#2}{}{#1}}}
\newrobustcmd{\hatPsiz}[2][]{\ensuremath{\subp{\hat{\Psi}}{}{#2}{}{#1}}}
\newrobustcmd{\widehatPsiz}[2][]{\ensuremath{\subp{\widehat{\Psi}}{}{#2}{}{#1}}}
\newrobustcmd{\checkPsiz}[2][]{\ensuremath{\subp{\check{\Psi}}{}{#2}{}{#1}}}
\newrobustcmd{\tildePsiz}[2][]{\ensuremath{\subp{\tilde{\Psi}}{}{#2}{}{#1}}}
\newrobustcmd{\widetildePsiz}[2][]{\ensuremath{\subp{\widetilde{\Psi}}{}{#2}{}{#1}}}
\newrobustcmd{\acutePsiz}[2][]{\ensuremath{\subp{\acute{\Psi}}{}{#2}{}{#1}}}
\newrobustcmd{\gravePsiz}[2][]{\ensuremath{\subp{\grave{\Psi}}{}{#2}{}{#1}}}
\newrobustcmd{\dotPsiz}[2][]{\ensuremath{\subp{\dot{\Psi}}{}{#2}{}{#1}}}
\newrobustcmd{\ddotPsiz}[2][]{\ensuremath{\subp{\ddot{\Psi}}{}{#2}{}{#1}}}
\newrobustcmd{\brevePsiz}[2][]{\ensuremath{\subp{\breve{\Psi}}{}{#2}{}{#1}}}
\newrobustcmd{\barPsiz}[2][]{\ensuremath{\subp{\bar{\Psi}}{}{#2}{}{#1}}}
\newrobustcmd{\vecPsiz}[2][]{\ensuremath{\subp{\vec{\Psi}}{}{#2}{}{#1}}}
\newrobustcmd{\bmPsiz}[2][]{\ensuremath{\subp{\bm{\Psi}}{}{#2}{}{#1}}}
\newrobustcmd{\hatbmPsiz}[2][]{\ensuremath{\subp{\hat{\bm{\Psi}}}{}{#2}{}{#1}}}
\newrobustcmd{\widehatbmPsiz}[2][]{\ensuremath{\subp{\widehat{\bm{\Psi}}}{}{#2}{}{#1}}}
\newrobustcmd{\checkbmPsiz}[2][]{\ensuremath{\subp{\check{\bm{\Psi}}}{}{#2}{}{#1}}}
\newrobustcmd{\tildebmPsiz}[2][]{\ensuremath{\subp{\tilde{\bm{\Psi}}}{}{#2}{}{#1}}}
\newrobustcmd{\widetildebmPsiz}[2][]{\ensuremath{\subp{\widetilde{\bm{\Psi}}}{}{#2}{}{#1}}}
\newrobustcmd{\acutebmPsiz}[2][]{\ensuremath{\subp{\acute{\bm{\Psi}}}{}{#2}{}{#1}}}
\newrobustcmd{\gravebmPsiz}[2][]{\ensuremath{\subp{\grave{\bm{\Psi}}}{}{#2}{}{#1}}}
\newrobustcmd{\dotbmPsiz}[2][]{\ensuremath{\subp{\dot{\bm{\Psi}}}{}{#2}{}{#1}}}
\newrobustcmd{\ddotbmPsiz}[2][]{\ensuremath{\subp{\ddot{\bm{\Psi}}}{}{#2}{}{#1}}}
\newrobustcmd{\brevebmPsiz}[2][]{\ensuremath{\subp{\breve{\bm{\Psi}}}{}{#2}{}{#1}}}
\newrobustcmd{\barbmPsiz}[2][]{\ensuremath{\subp{\bar{\bm{\Psi}}}{}{#2}{}{#1}}}
\newrobustcmd{\vecbmPsiz}[2][]{\ensuremath{\subp{\vec{\bm{\Psi}}}{}{#2}{}{#1}}}
\newrobustcmd{\Omegaz}[2][]{\ensuremath{\subp{\Omega}{}{#2}{}{#1}}}
\newrobustcmd{\hatOmegaz}[2][]{\ensuremath{\subp{\hat{\Omega}}{}{#2}{}{#1}}}
\newrobustcmd{\widehatOmegaz}[2][]{\ensuremath{\subp{\widehat{\Omega}}{}{#2}{}{#1}}}
\newrobustcmd{\checkOmegaz}[2][]{\ensuremath{\subp{\check{\Omega}}{}{#2}{}{#1}}}
\newrobustcmd{\tildeOmegaz}[2][]{\ensuremath{\subp{\tilde{\Omega}}{}{#2}{}{#1}}}
\newrobustcmd{\widetildeOmegaz}[2][]{\ensuremath{\subp{\widetilde{\Omega}}{}{#2}{}{#1}}}
\newrobustcmd{\acuteOmegaz}[2][]{\ensuremath{\subp{\acute{\Omega}}{}{#2}{}{#1}}}
\newrobustcmd{\graveOmegaz}[2][]{\ensuremath{\subp{\grave{\Omega}}{}{#2}{}{#1}}}
\newrobustcmd{\dotOmegaz}[2][]{\ensuremath{\subp{\dot{\Omega}}{}{#2}{}{#1}}}
\newrobustcmd{\ddotOmegaz}[2][]{\ensuremath{\subp{\ddot{\Omega}}{}{#2}{}{#1}}}
\newrobustcmd{\breveOmegaz}[2][]{\ensuremath{\subp{\breve{\Omega}}{}{#2}{}{#1}}}
\newrobustcmd{\barOmegaz}[2][]{\ensuremath{\subp{\bar{\Omega}}{}{#2}{}{#1}}}
\newrobustcmd{\vecOmegaz}[2][]{\ensuremath{\subp{\vec{\Omega}}{}{#2}{}{#1}}}
\newrobustcmd{\bmOmegaz}[2][]{\ensuremath{\subp{\bm{\Omega}}{}{#2}{}{#1}}}
\newrobustcmd{\hatbmOmegaz}[2][]{\ensuremath{\subp{\hat{\bm{\Omega}}}{}{#2}{}{#1}}}
\newrobustcmd{\widehatbmOmegaz}[2][]{\ensuremath{\subp{\widehat{\bm{\Omega}}}{}{#2}{}{#1}}}
\newrobustcmd{\checkbmOmegaz}[2][]{\ensuremath{\subp{\check{\bm{\Omega}}}{}{#2}{}{#1}}}
\newrobustcmd{\tildebmOmegaz}[2][]{\ensuremath{\subp{\tilde{\bm{\Omega}}}{}{#2}{}{#1}}}
\newrobustcmd{\widetildebmOmegaz}[2][]{\ensuremath{\subp{\widetilde{\bm{\Omega}}}{}{#2}{}{#1}}}
\newrobustcmd{\acutebmOmegaz}[2][]{\ensuremath{\subp{\acute{\bm{\Omega}}}{}{#2}{}{#1}}}
\newrobustcmd{\gravebmOmegaz}[2][]{\ensuremath{\subp{\grave{\bm{\Omega}}}{}{#2}{}{#1}}}
\newrobustcmd{\dotbmOmegaz}[2][]{\ensuremath{\subp{\dot{\bm{\Omega}}}{}{#2}{}{#1}}}
\newrobustcmd{\ddotbmOmegaz}[2][]{\ensuremath{\subp{\ddot{\bm{\Omega}}}{}{#2}{}{#1}}}
\newrobustcmd{\brevebmOmegaz}[2][]{\ensuremath{\subp{\breve{\bm{\Omega}}}{}{#2}{}{#1}}}
\newrobustcmd{\barbmOmegaz}[2][]{\ensuremath{\subp{\bar{\bm{\Omega}}}{}{#2}{}{#1}}}
\newrobustcmd{\vecbmOmegaz}[2][]{\ensuremath{\subp{\vec{\bm{\Omega}}}{}{#2}{}{#1}}}

\newrobustcmd{\mathfrakaz}[2][]{\ensuremath{\subp{\mathfrak{a}}{}{#2}{}{#1}}}
\newrobustcmd{\hatmathfrakaz}[2][]{\ensuremath{\subp{\hat{\mathfrak{a}}}{}{#2}{}{#1}}}
\newrobustcmd{\widehatmathfrakaz}[2][]{\ensuremath{\subp{\widehat{\mathfrak{a}}}{}{#2}{}{#1}}}
\newrobustcmd{\checkmathfrakaz}[2][]{\ensuremath{\subp{\check{\mathfrak{a}}}{}{#2}{}{#1}}}
\newrobustcmd{\tildemathfrakaz}[2][]{\ensuremath{\subp{\tilde{\mathfrak{a}}}{}{#2}{}{#1}}}
\newrobustcmd{\widetildemathfrakaz}[2][]{\ensuremath{\subp{\widetilde{\mathfrak{a}}}{}{#2}{}{#1}}}
\newrobustcmd{\acutemathfrakaz}[2][]{\ensuremath{\subp{\acute{\mathfrak{a}}}{}{#2}{}{#1}}}
\newrobustcmd{\gravemathfrakaz}[2][]{\ensuremath{\subp{\grave{\mathfrak{a}}}{}{#2}{}{#1}}}
\newrobustcmd{\dotmathfrakaz}[2][]{\ensuremath{\subp{\dot{\mathfrak{a}}}{}{#2}{}{#1}}}
\newrobustcmd{\ddotmathfrakaz}[2][]{\ensuremath{\subp{\ddot{\mathfrak{a}}}{}{#2}{}{#1}}}
\newrobustcmd{\brevemathfrakaz}[2][]{\ensuremath{\subp{\breve{\mathfrak{a}}}{}{#2}{}{#1}}}
\newrobustcmd{\barmathfrakaz}[2][]{\ensuremath{\subp{\bar{\mathfrak{a}}}{}{#2}{}{#1}}}
\newrobustcmd{\vecmathfrakaz}[2][]{\ensuremath{\subp{\vec{\mathfrak{a}}}{}{#2}{}{#1}}}
\newrobustcmd{\bmmathfrakaz}[2][]{\ensuremath{\subp{\bm{\mathfrak{a}}}{}{#2}{}{#1}}}
\newrobustcmd{\hatbmmathfrakaz}[2][]{\ensuremath{\subp{\hat{\bm{\mathfrak{a}}}}{}{#2}{}{#1}}}
\newrobustcmd{\widehatbmmathfrakaz}[2][]{\ensuremath{\subp{\widehat{\bm{\mathfrak{a}}}}{}{#2}{}{#1}}}
\newrobustcmd{\checkbmmathfrakaz}[2][]{\ensuremath{\subp{\check{\bm{\mathfrak{a}}}}{}{#2}{}{#1}}}
\newrobustcmd{\tildebmmathfrakaz}[2][]{\ensuremath{\subp{\tilde{\bm{\mathfrak{a}}}}{}{#2}{}{#1}}}
\newrobustcmd{\widetildebmmathfrakaz}[2][]{\ensuremath{\subp{\widetilde{\bm{\mathfrak{a}}}}{}{#2}{}{#1}}}
\newrobustcmd{\acutebmmathfrakaz}[2][]{\ensuremath{\subp{\acute{\bm{\mathfrak{a}}}}{}{#2}{}{#1}}}
\newrobustcmd{\gravebmmathfrakaz}[2][]{\ensuremath{\subp{\grave{\bm{\mathfrak{a}}}}{}{#2}{}{#1}}}
\newrobustcmd{\dotbmmathfrakaz}[2][]{\ensuremath{\subp{\dot{\bm{\mathfrak{a}}}}{}{#2}{}{#1}}}
\newrobustcmd{\ddotbmmathfrakaz}[2][]{\ensuremath{\subp{\ddot{\bm{\mathfrak{a}}}}{}{#2}{}{#1}}}
\newrobustcmd{\brevebmmathfrakaz}[2][]{\ensuremath{\subp{\breve{\bm{\mathfrak{a}}}}{}{#2}{}{#1}}}
\newrobustcmd{\barbmmathfrakaz}[2][]{\ensuremath{\subp{\bar{\bm{\mathfrak{a}}}}{}{#2}{}{#1}}}
\newrobustcmd{\vecbmmathfrakaz}[2][]{\ensuremath{\subp{\vec{\bm{\mathfrak{a}}}}{}{#2}{}{#1}}}
\newrobustcmd{\mathfrakbz}[2][]{\ensuremath{\subp{\mathfrak{b}}{}{#2}{}{#1}}}
\newrobustcmd{\hatmathfrakbz}[2][]{\ensuremath{\subp{\hat{\mathfrak{b}}}{}{#2}{}{#1}}}
\newrobustcmd{\widehatmathfrakbz}[2][]{\ensuremath{\subp{\widehat{\mathfrak{b}}}{}{#2}{}{#1}}}
\newrobustcmd{\checkmathfrakbz}[2][]{\ensuremath{\subp{\check{\mathfrak{b}}}{}{#2}{}{#1}}}
\newrobustcmd{\tildemathfrakbz}[2][]{\ensuremath{\subp{\tilde{\mathfrak{b}}}{}{#2}{}{#1}}}
\newrobustcmd{\widetildemathfrakbz}[2][]{\ensuremath{\subp{\widetilde{\mathfrak{b}}}{}{#2}{}{#1}}}
\newrobustcmd{\acutemathfrakbz}[2][]{\ensuremath{\subp{\acute{\mathfrak{b}}}{}{#2}{}{#1}}}
\newrobustcmd{\gravemathfrakbz}[2][]{\ensuremath{\subp{\grave{\mathfrak{b}}}{}{#2}{}{#1}}}
\newrobustcmd{\dotmathfrakbz}[2][]{\ensuremath{\subp{\dot{\mathfrak{b}}}{}{#2}{}{#1}}}
\newrobustcmd{\ddotmathfrakbz}[2][]{\ensuremath{\subp{\ddot{\mathfrak{b}}}{}{#2}{}{#1}}}
\newrobustcmd{\brevemathfrakbz}[2][]{\ensuremath{\subp{\breve{\mathfrak{b}}}{}{#2}{}{#1}}}
\newrobustcmd{\barmathfrakbz}[2][]{\ensuremath{\subp{\bar{\mathfrak{b}}}{}{#2}{}{#1}}}
\newrobustcmd{\vecmathfrakbz}[2][]{\ensuremath{\subp{\vec{\mathfrak{b}}}{}{#2}{}{#1}}}
\newrobustcmd{\bmmathfrakbz}[2][]{\ensuremath{\subp{\bm{\mathfrak{b}}}{}{#2}{}{#1}}}
\newrobustcmd{\hatbmmathfrakbz}[2][]{\ensuremath{\subp{\hat{\bm{\mathfrak{b}}}}{}{#2}{}{#1}}}
\newrobustcmd{\widehatbmmathfrakbz}[2][]{\ensuremath{\subp{\widehat{\bm{\mathfrak{b}}}}{}{#2}{}{#1}}}
\newrobustcmd{\checkbmmathfrakbz}[2][]{\ensuremath{\subp{\check{\bm{\mathfrak{b}}}}{}{#2}{}{#1}}}
\newrobustcmd{\tildebmmathfrakbz}[2][]{\ensuremath{\subp{\tilde{\bm{\mathfrak{b}}}}{}{#2}{}{#1}}}
\newrobustcmd{\widetildebmmathfrakbz}[2][]{\ensuremath{\subp{\widetilde{\bm{\mathfrak{b}}}}{}{#2}{}{#1}}}
\newrobustcmd{\acutebmmathfrakbz}[2][]{\ensuremath{\subp{\acute{\bm{\mathfrak{b}}}}{}{#2}{}{#1}}}
\newrobustcmd{\gravebmmathfrakbz}[2][]{\ensuremath{\subp{\grave{\bm{\mathfrak{b}}}}{}{#2}{}{#1}}}
\newrobustcmd{\dotbmmathfrakbz}[2][]{\ensuremath{\subp{\dot{\bm{\mathfrak{b}}}}{}{#2}{}{#1}}}
\newrobustcmd{\ddotbmmathfrakbz}[2][]{\ensuremath{\subp{\ddot{\bm{\mathfrak{b}}}}{}{#2}{}{#1}}}
\newrobustcmd{\brevebmmathfrakbz}[2][]{\ensuremath{\subp{\breve{\bm{\mathfrak{b}}}}{}{#2}{}{#1}}}
\newrobustcmd{\barbmmathfrakbz}[2][]{\ensuremath{\subp{\bar{\bm{\mathfrak{b}}}}{}{#2}{}{#1}}}
\newrobustcmd{\vecbmmathfrakbz}[2][]{\ensuremath{\subp{\vec{\bm{\mathfrak{b}}}}{}{#2}{}{#1}}}
\newrobustcmd{\mathfrakcz}[2][]{\ensuremath{\subp{\mathfrak{c}}{}{#2}{}{#1}}}
\newrobustcmd{\hatmathfrakcz}[2][]{\ensuremath{\subp{\hat{\mathfrak{c}}}{}{#2}{}{#1}}}
\newrobustcmd{\widehatmathfrakcz}[2][]{\ensuremath{\subp{\widehat{\mathfrak{c}}}{}{#2}{}{#1}}}
\newrobustcmd{\checkmathfrakcz}[2][]{\ensuremath{\subp{\check{\mathfrak{c}}}{}{#2}{}{#1}}}
\newrobustcmd{\tildemathfrakcz}[2][]{\ensuremath{\subp{\tilde{\mathfrak{c}}}{}{#2}{}{#1}}}
\newrobustcmd{\widetildemathfrakcz}[2][]{\ensuremath{\subp{\widetilde{\mathfrak{c}}}{}{#2}{}{#1}}}
\newrobustcmd{\acutemathfrakcz}[2][]{\ensuremath{\subp{\acute{\mathfrak{c}}}{}{#2}{}{#1}}}
\newrobustcmd{\gravemathfrakcz}[2][]{\ensuremath{\subp{\grave{\mathfrak{c}}}{}{#2}{}{#1}}}
\newrobustcmd{\dotmathfrakcz}[2][]{\ensuremath{\subp{\dot{\mathfrak{c}}}{}{#2}{}{#1}}}
\newrobustcmd{\ddotmathfrakcz}[2][]{\ensuremath{\subp{\ddot{\mathfrak{c}}}{}{#2}{}{#1}}}
\newrobustcmd{\brevemathfrakcz}[2][]{\ensuremath{\subp{\breve{\mathfrak{c}}}{}{#2}{}{#1}}}
\newrobustcmd{\barmathfrakcz}[2][]{\ensuremath{\subp{\bar{\mathfrak{c}}}{}{#2}{}{#1}}}
\newrobustcmd{\vecmathfrakcz}[2][]{\ensuremath{\subp{\vec{\mathfrak{c}}}{}{#2}{}{#1}}}
\newrobustcmd{\bmmathfrakcz}[2][]{\ensuremath{\subp{\bm{\mathfrak{c}}}{}{#2}{}{#1}}}
\newrobustcmd{\hatbmmathfrakcz}[2][]{\ensuremath{\subp{\hat{\bm{\mathfrak{c}}}}{}{#2}{}{#1}}}
\newrobustcmd{\widehatbmmathfrakcz}[2][]{\ensuremath{\subp{\widehat{\bm{\mathfrak{c}}}}{}{#2}{}{#1}}}
\newrobustcmd{\checkbmmathfrakcz}[2][]{\ensuremath{\subp{\check{\bm{\mathfrak{c}}}}{}{#2}{}{#1}}}
\newrobustcmd{\tildebmmathfrakcz}[2][]{\ensuremath{\subp{\tilde{\bm{\mathfrak{c}}}}{}{#2}{}{#1}}}
\newrobustcmd{\widetildebmmathfrakcz}[2][]{\ensuremath{\subp{\widetilde{\bm{\mathfrak{c}}}}{}{#2}{}{#1}}}
\newrobustcmd{\acutebmmathfrakcz}[2][]{\ensuremath{\subp{\acute{\bm{\mathfrak{c}}}}{}{#2}{}{#1}}}
\newrobustcmd{\gravebmmathfrakcz}[2][]{\ensuremath{\subp{\grave{\bm{\mathfrak{c}}}}{}{#2}{}{#1}}}
\newrobustcmd{\dotbmmathfrakcz}[2][]{\ensuremath{\subp{\dot{\bm{\mathfrak{c}}}}{}{#2}{}{#1}}}
\newrobustcmd{\ddotbmmathfrakcz}[2][]{\ensuremath{\subp{\ddot{\bm{\mathfrak{c}}}}{}{#2}{}{#1}}}
\newrobustcmd{\brevebmmathfrakcz}[2][]{\ensuremath{\subp{\breve{\bm{\mathfrak{c}}}}{}{#2}{}{#1}}}
\newrobustcmd{\barbmmathfrakcz}[2][]{\ensuremath{\subp{\bar{\bm{\mathfrak{c}}}}{}{#2}{}{#1}}}
\newrobustcmd{\vecbmmathfrakcz}[2][]{\ensuremath{\subp{\vec{\bm{\mathfrak{c}}}}{}{#2}{}{#1}}}
\newrobustcmd{\mathfrakdz}[2][]{\ensuremath{\subp{\mathfrak{d}}{}{#2}{}{#1}}}
\newrobustcmd{\hatmathfrakdz}[2][]{\ensuremath{\subp{\hat{\mathfrak{d}}}{}{#2}{}{#1}}}
\newrobustcmd{\widehatmathfrakdz}[2][]{\ensuremath{\subp{\widehat{\mathfrak{d}}}{}{#2}{}{#1}}}
\newrobustcmd{\checkmathfrakdz}[2][]{\ensuremath{\subp{\check{\mathfrak{d}}}{}{#2}{}{#1}}}
\newrobustcmd{\tildemathfrakdz}[2][]{\ensuremath{\subp{\tilde{\mathfrak{d}}}{}{#2}{}{#1}}}
\newrobustcmd{\widetildemathfrakdz}[2][]{\ensuremath{\subp{\widetilde{\mathfrak{d}}}{}{#2}{}{#1}}}
\newrobustcmd{\acutemathfrakdz}[2][]{\ensuremath{\subp{\acute{\mathfrak{d}}}{}{#2}{}{#1}}}
\newrobustcmd{\gravemathfrakdz}[2][]{\ensuremath{\subp{\grave{\mathfrak{d}}}{}{#2}{}{#1}}}
\newrobustcmd{\dotmathfrakdz}[2][]{\ensuremath{\subp{\dot{\mathfrak{d}}}{}{#2}{}{#1}}}
\newrobustcmd{\ddotmathfrakdz}[2][]{\ensuremath{\subp{\ddot{\mathfrak{d}}}{}{#2}{}{#1}}}
\newrobustcmd{\brevemathfrakdz}[2][]{\ensuremath{\subp{\breve{\mathfrak{d}}}{}{#2}{}{#1}}}
\newrobustcmd{\barmathfrakdz}[2][]{\ensuremath{\subp{\bar{\mathfrak{d}}}{}{#2}{}{#1}}}
\newrobustcmd{\vecmathfrakdz}[2][]{\ensuremath{\subp{\vec{\mathfrak{d}}}{}{#2}{}{#1}}}
\newrobustcmd{\bmmathfrakdz}[2][]{\ensuremath{\subp{\bm{\mathfrak{d}}}{}{#2}{}{#1}}}
\newrobustcmd{\hatbmmathfrakdz}[2][]{\ensuremath{\subp{\hat{\bm{\mathfrak{d}}}}{}{#2}{}{#1}}}
\newrobustcmd{\widehatbmmathfrakdz}[2][]{\ensuremath{\subp{\widehat{\bm{\mathfrak{d}}}}{}{#2}{}{#1}}}
\newrobustcmd{\checkbmmathfrakdz}[2][]{\ensuremath{\subp{\check{\bm{\mathfrak{d}}}}{}{#2}{}{#1}}}
\newrobustcmd{\tildebmmathfrakdz}[2][]{\ensuremath{\subp{\tilde{\bm{\mathfrak{d}}}}{}{#2}{}{#1}}}
\newrobustcmd{\widetildebmmathfrakdz}[2][]{\ensuremath{\subp{\widetilde{\bm{\mathfrak{d}}}}{}{#2}{}{#1}}}
\newrobustcmd{\acutebmmathfrakdz}[2][]{\ensuremath{\subp{\acute{\bm{\mathfrak{d}}}}{}{#2}{}{#1}}}
\newrobustcmd{\gravebmmathfrakdz}[2][]{\ensuremath{\subp{\grave{\bm{\mathfrak{d}}}}{}{#2}{}{#1}}}
\newrobustcmd{\dotbmmathfrakdz}[2][]{\ensuremath{\subp{\dot{\bm{\mathfrak{d}}}}{}{#2}{}{#1}}}
\newrobustcmd{\ddotbmmathfrakdz}[2][]{\ensuremath{\subp{\ddot{\bm{\mathfrak{d}}}}{}{#2}{}{#1}}}
\newrobustcmd{\brevebmmathfrakdz}[2][]{\ensuremath{\subp{\breve{\bm{\mathfrak{d}}}}{}{#2}{}{#1}}}
\newrobustcmd{\barbmmathfrakdz}[2][]{\ensuremath{\subp{\bar{\bm{\mathfrak{d}}}}{}{#2}{}{#1}}}
\newrobustcmd{\vecbmmathfrakdz}[2][]{\ensuremath{\subp{\vec{\bm{\mathfrak{d}}}}{}{#2}{}{#1}}}
\newrobustcmd{\mathfrakez}[2][]{\ensuremath{\subp{\mathfrak{e}}{}{#2}{}{#1}}}
\newrobustcmd{\hatmathfrakez}[2][]{\ensuremath{\subp{\hat{\mathfrak{e}}}{}{#2}{}{#1}}}
\newrobustcmd{\widehatmathfrakez}[2][]{\ensuremath{\subp{\widehat{\mathfrak{e}}}{}{#2}{}{#1}}}
\newrobustcmd{\checkmathfrakez}[2][]{\ensuremath{\subp{\check{\mathfrak{e}}}{}{#2}{}{#1}}}
\newrobustcmd{\tildemathfrakez}[2][]{\ensuremath{\subp{\tilde{\mathfrak{e}}}{}{#2}{}{#1}}}
\newrobustcmd{\widetildemathfrakez}[2][]{\ensuremath{\subp{\widetilde{\mathfrak{e}}}{}{#2}{}{#1}}}
\newrobustcmd{\acutemathfrakez}[2][]{\ensuremath{\subp{\acute{\mathfrak{e}}}{}{#2}{}{#1}}}
\newrobustcmd{\gravemathfrakez}[2][]{\ensuremath{\subp{\grave{\mathfrak{e}}}{}{#2}{}{#1}}}
\newrobustcmd{\dotmathfrakez}[2][]{\ensuremath{\subp{\dot{\mathfrak{e}}}{}{#2}{}{#1}}}
\newrobustcmd{\ddotmathfrakez}[2][]{\ensuremath{\subp{\ddot{\mathfrak{e}}}{}{#2}{}{#1}}}
\newrobustcmd{\brevemathfrakez}[2][]{\ensuremath{\subp{\breve{\mathfrak{e}}}{}{#2}{}{#1}}}
\newrobustcmd{\barmathfrakez}[2][]{\ensuremath{\subp{\bar{\mathfrak{e}}}{}{#2}{}{#1}}}
\newrobustcmd{\vecmathfrakez}[2][]{\ensuremath{\subp{\vec{\mathfrak{e}}}{}{#2}{}{#1}}}
\newrobustcmd{\bmmathfrakez}[2][]{\ensuremath{\subp{\bm{\mathfrak{e}}}{}{#2}{}{#1}}}
\newrobustcmd{\hatbmmathfrakez}[2][]{\ensuremath{\subp{\hat{\bm{\mathfrak{e}}}}{}{#2}{}{#1}}}
\newrobustcmd{\widehatbmmathfrakez}[2][]{\ensuremath{\subp{\widehat{\bm{\mathfrak{e}}}}{}{#2}{}{#1}}}
\newrobustcmd{\checkbmmathfrakez}[2][]{\ensuremath{\subp{\check{\bm{\mathfrak{e}}}}{}{#2}{}{#1}}}
\newrobustcmd{\tildebmmathfrakez}[2][]{\ensuremath{\subp{\tilde{\bm{\mathfrak{e}}}}{}{#2}{}{#1}}}
\newrobustcmd{\widetildebmmathfrakez}[2][]{\ensuremath{\subp{\widetilde{\bm{\mathfrak{e}}}}{}{#2}{}{#1}}}
\newrobustcmd{\acutebmmathfrakez}[2][]{\ensuremath{\subp{\acute{\bm{\mathfrak{e}}}}{}{#2}{}{#1}}}
\newrobustcmd{\gravebmmathfrakez}[2][]{\ensuremath{\subp{\grave{\bm{\mathfrak{e}}}}{}{#2}{}{#1}}}
\newrobustcmd{\dotbmmathfrakez}[2][]{\ensuremath{\subp{\dot{\bm{\mathfrak{e}}}}{}{#2}{}{#1}}}
\newrobustcmd{\ddotbmmathfrakez}[2][]{\ensuremath{\subp{\ddot{\bm{\mathfrak{e}}}}{}{#2}{}{#1}}}
\newrobustcmd{\brevebmmathfrakez}[2][]{\ensuremath{\subp{\breve{\bm{\mathfrak{e}}}}{}{#2}{}{#1}}}
\newrobustcmd{\barbmmathfrakez}[2][]{\ensuremath{\subp{\bar{\bm{\mathfrak{e}}}}{}{#2}{}{#1}}}
\newrobustcmd{\vecbmmathfrakez}[2][]{\ensuremath{\subp{\vec{\bm{\mathfrak{e}}}}{}{#2}{}{#1}}}
\newrobustcmd{\mathfrakfz}[2][]{\ensuremath{\subp{\mathfrak{f}}{}{#2}{}{#1}}}
\newrobustcmd{\hatmathfrakfz}[2][]{\ensuremath{\subp{\hat{\mathfrak{f}}}{}{#2}{}{#1}}}
\newrobustcmd{\widehatmathfrakfz}[2][]{\ensuremath{\subp{\widehat{\mathfrak{f}}}{}{#2}{}{#1}}}
\newrobustcmd{\checkmathfrakfz}[2][]{\ensuremath{\subp{\check{\mathfrak{f}}}{}{#2}{}{#1}}}
\newrobustcmd{\tildemathfrakfz}[2][]{\ensuremath{\subp{\tilde{\mathfrak{f}}}{}{#2}{}{#1}}}
\newrobustcmd{\widetildemathfrakfz}[2][]{\ensuremath{\subp{\widetilde{\mathfrak{f}}}{}{#2}{}{#1}}}
\newrobustcmd{\acutemathfrakfz}[2][]{\ensuremath{\subp{\acute{\mathfrak{f}}}{}{#2}{}{#1}}}
\newrobustcmd{\gravemathfrakfz}[2][]{\ensuremath{\subp{\grave{\mathfrak{f}}}{}{#2}{}{#1}}}
\newrobustcmd{\dotmathfrakfz}[2][]{\ensuremath{\subp{\dot{\mathfrak{f}}}{}{#2}{}{#1}}}
\newrobustcmd{\ddotmathfrakfz}[2][]{\ensuremath{\subp{\ddot{\mathfrak{f}}}{}{#2}{}{#1}}}
\newrobustcmd{\brevemathfrakfz}[2][]{\ensuremath{\subp{\breve{\mathfrak{f}}}{}{#2}{}{#1}}}
\newrobustcmd{\barmathfrakfz}[2][]{\ensuremath{\subp{\bar{\mathfrak{f}}}{}{#2}{}{#1}}}
\newrobustcmd{\vecmathfrakfz}[2][]{\ensuremath{\subp{\vec{\mathfrak{f}}}{}{#2}{}{#1}}}
\newrobustcmd{\bmmathfrakfz}[2][]{\ensuremath{\subp{\bm{\mathfrak{f}}}{}{#2}{}{#1}}}
\newrobustcmd{\hatbmmathfrakfz}[2][]{\ensuremath{\subp{\hat{\bm{\mathfrak{f}}}}{}{#2}{}{#1}}}
\newrobustcmd{\widehatbmmathfrakfz}[2][]{\ensuremath{\subp{\widehat{\bm{\mathfrak{f}}}}{}{#2}{}{#1}}}
\newrobustcmd{\checkbmmathfrakfz}[2][]{\ensuremath{\subp{\check{\bm{\mathfrak{f}}}}{}{#2}{}{#1}}}
\newrobustcmd{\tildebmmathfrakfz}[2][]{\ensuremath{\subp{\tilde{\bm{\mathfrak{f}}}}{}{#2}{}{#1}}}
\newrobustcmd{\widetildebmmathfrakfz}[2][]{\ensuremath{\subp{\widetilde{\bm{\mathfrak{f}}}}{}{#2}{}{#1}}}
\newrobustcmd{\acutebmmathfrakfz}[2][]{\ensuremath{\subp{\acute{\bm{\mathfrak{f}}}}{}{#2}{}{#1}}}
\newrobustcmd{\gravebmmathfrakfz}[2][]{\ensuremath{\subp{\grave{\bm{\mathfrak{f}}}}{}{#2}{}{#1}}}
\newrobustcmd{\dotbmmathfrakfz}[2][]{\ensuremath{\subp{\dot{\bm{\mathfrak{f}}}}{}{#2}{}{#1}}}
\newrobustcmd{\ddotbmmathfrakfz}[2][]{\ensuremath{\subp{\ddot{\bm{\mathfrak{f}}}}{}{#2}{}{#1}}}
\newrobustcmd{\brevebmmathfrakfz}[2][]{\ensuremath{\subp{\breve{\bm{\mathfrak{f}}}}{}{#2}{}{#1}}}
\newrobustcmd{\barbmmathfrakfz}[2][]{\ensuremath{\subp{\bar{\bm{\mathfrak{f}}}}{}{#2}{}{#1}}}
\newrobustcmd{\vecbmmathfrakfz}[2][]{\ensuremath{\subp{\vec{\bm{\mathfrak{f}}}}{}{#2}{}{#1}}}
\newrobustcmd{\mathfrakgz}[2][]{\ensuremath{\subp{\mathfrak{g}}{}{#2}{}{#1}}}
\newrobustcmd{\hatmathfrakgz}[2][]{\ensuremath{\subp{\hat{\mathfrak{g}}}{}{#2}{}{#1}}}
\newrobustcmd{\widehatmathfrakgz}[2][]{\ensuremath{\subp{\widehat{\mathfrak{g}}}{}{#2}{}{#1}}}
\newrobustcmd{\checkmathfrakgz}[2][]{\ensuremath{\subp{\check{\mathfrak{g}}}{}{#2}{}{#1}}}
\newrobustcmd{\tildemathfrakgz}[2][]{\ensuremath{\subp{\tilde{\mathfrak{g}}}{}{#2}{}{#1}}}
\newrobustcmd{\widetildemathfrakgz}[2][]{\ensuremath{\subp{\widetilde{\mathfrak{g}}}{}{#2}{}{#1}}}
\newrobustcmd{\acutemathfrakgz}[2][]{\ensuremath{\subp{\acute{\mathfrak{g}}}{}{#2}{}{#1}}}
\newrobustcmd{\gravemathfrakgz}[2][]{\ensuremath{\subp{\grave{\mathfrak{g}}}{}{#2}{}{#1}}}
\newrobustcmd{\dotmathfrakgz}[2][]{\ensuremath{\subp{\dot{\mathfrak{g}}}{}{#2}{}{#1}}}
\newrobustcmd{\ddotmathfrakgz}[2][]{\ensuremath{\subp{\ddot{\mathfrak{g}}}{}{#2}{}{#1}}}
\newrobustcmd{\brevemathfrakgz}[2][]{\ensuremath{\subp{\breve{\mathfrak{g}}}{}{#2}{}{#1}}}
\newrobustcmd{\barmathfrakgz}[2][]{\ensuremath{\subp{\bar{\mathfrak{g}}}{}{#2}{}{#1}}}
\newrobustcmd{\vecmathfrakgz}[2][]{\ensuremath{\subp{\vec{\mathfrak{g}}}{}{#2}{}{#1}}}
\newrobustcmd{\bmmathfrakgz}[2][]{\ensuremath{\subp{\bm{\mathfrak{g}}}{}{#2}{}{#1}}}
\newrobustcmd{\hatbmmathfrakgz}[2][]{\ensuremath{\subp{\hat{\bm{\mathfrak{g}}}}{}{#2}{}{#1}}}
\newrobustcmd{\widehatbmmathfrakgz}[2][]{\ensuremath{\subp{\widehat{\bm{\mathfrak{g}}}}{}{#2}{}{#1}}}
\newrobustcmd{\checkbmmathfrakgz}[2][]{\ensuremath{\subp{\check{\bm{\mathfrak{g}}}}{}{#2}{}{#1}}}
\newrobustcmd{\tildebmmathfrakgz}[2][]{\ensuremath{\subp{\tilde{\bm{\mathfrak{g}}}}{}{#2}{}{#1}}}
\newrobustcmd{\widetildebmmathfrakgz}[2][]{\ensuremath{\subp{\widetilde{\bm{\mathfrak{g}}}}{}{#2}{}{#1}}}
\newrobustcmd{\acutebmmathfrakgz}[2][]{\ensuremath{\subp{\acute{\bm{\mathfrak{g}}}}{}{#2}{}{#1}}}
\newrobustcmd{\gravebmmathfrakgz}[2][]{\ensuremath{\subp{\grave{\bm{\mathfrak{g}}}}{}{#2}{}{#1}}}
\newrobustcmd{\dotbmmathfrakgz}[2][]{\ensuremath{\subp{\dot{\bm{\mathfrak{g}}}}{}{#2}{}{#1}}}
\newrobustcmd{\ddotbmmathfrakgz}[2][]{\ensuremath{\subp{\ddot{\bm{\mathfrak{g}}}}{}{#2}{}{#1}}}
\newrobustcmd{\brevebmmathfrakgz}[2][]{\ensuremath{\subp{\breve{\bm{\mathfrak{g}}}}{}{#2}{}{#1}}}
\newrobustcmd{\barbmmathfrakgz}[2][]{\ensuremath{\subp{\bar{\bm{\mathfrak{g}}}}{}{#2}{}{#1}}}
\newrobustcmd{\vecbmmathfrakgz}[2][]{\ensuremath{\subp{\vec{\bm{\mathfrak{g}}}}{}{#2}{}{#1}}}
\newrobustcmd{\mathfrakhz}[2][]{\ensuremath{\subp{\mathfrak{h}}{}{#2}{}{#1}}}
\newrobustcmd{\hatmathfrakhz}[2][]{\ensuremath{\subp{\hat{\mathfrak{h}}}{}{#2}{}{#1}}}
\newrobustcmd{\widehatmathfrakhz}[2][]{\ensuremath{\subp{\widehat{\mathfrak{h}}}{}{#2}{}{#1}}}
\newrobustcmd{\checkmathfrakhz}[2][]{\ensuremath{\subp{\check{\mathfrak{h}}}{}{#2}{}{#1}}}
\newrobustcmd{\tildemathfrakhz}[2][]{\ensuremath{\subp{\tilde{\mathfrak{h}}}{}{#2}{}{#1}}}
\newrobustcmd{\widetildemathfrakhz}[2][]{\ensuremath{\subp{\widetilde{\mathfrak{h}}}{}{#2}{}{#1}}}
\newrobustcmd{\acutemathfrakhz}[2][]{\ensuremath{\subp{\acute{\mathfrak{h}}}{}{#2}{}{#1}}}
\newrobustcmd{\gravemathfrakhz}[2][]{\ensuremath{\subp{\grave{\mathfrak{h}}}{}{#2}{}{#1}}}
\newrobustcmd{\dotmathfrakhz}[2][]{\ensuremath{\subp{\dot{\mathfrak{h}}}{}{#2}{}{#1}}}
\newrobustcmd{\ddotmathfrakhz}[2][]{\ensuremath{\subp{\ddot{\mathfrak{h}}}{}{#2}{}{#1}}}
\newrobustcmd{\brevemathfrakhz}[2][]{\ensuremath{\subp{\breve{\mathfrak{h}}}{}{#2}{}{#1}}}
\newrobustcmd{\barmathfrakhz}[2][]{\ensuremath{\subp{\bar{\mathfrak{h}}}{}{#2}{}{#1}}}
\newrobustcmd{\vecmathfrakhz}[2][]{\ensuremath{\subp{\vec{\mathfrak{h}}}{}{#2}{}{#1}}}
\newrobustcmd{\bmmathfrakhz}[2][]{\ensuremath{\subp{\bm{\mathfrak{h}}}{}{#2}{}{#1}}}
\newrobustcmd{\hatbmmathfrakhz}[2][]{\ensuremath{\subp{\hat{\bm{\mathfrak{h}}}}{}{#2}{}{#1}}}
\newrobustcmd{\widehatbmmathfrakhz}[2][]{\ensuremath{\subp{\widehat{\bm{\mathfrak{h}}}}{}{#2}{}{#1}}}
\newrobustcmd{\checkbmmathfrakhz}[2][]{\ensuremath{\subp{\check{\bm{\mathfrak{h}}}}{}{#2}{}{#1}}}
\newrobustcmd{\tildebmmathfrakhz}[2][]{\ensuremath{\subp{\tilde{\bm{\mathfrak{h}}}}{}{#2}{}{#1}}}
\newrobustcmd{\widetildebmmathfrakhz}[2][]{\ensuremath{\subp{\widetilde{\bm{\mathfrak{h}}}}{}{#2}{}{#1}}}
\newrobustcmd{\acutebmmathfrakhz}[2][]{\ensuremath{\subp{\acute{\bm{\mathfrak{h}}}}{}{#2}{}{#1}}}
\newrobustcmd{\gravebmmathfrakhz}[2][]{\ensuremath{\subp{\grave{\bm{\mathfrak{h}}}}{}{#2}{}{#1}}}
\newrobustcmd{\dotbmmathfrakhz}[2][]{\ensuremath{\subp{\dot{\bm{\mathfrak{h}}}}{}{#2}{}{#1}}}
\newrobustcmd{\ddotbmmathfrakhz}[2][]{\ensuremath{\subp{\ddot{\bm{\mathfrak{h}}}}{}{#2}{}{#1}}}
\newrobustcmd{\brevebmmathfrakhz}[2][]{\ensuremath{\subp{\breve{\bm{\mathfrak{h}}}}{}{#2}{}{#1}}}
\newrobustcmd{\barbmmathfrakhz}[2][]{\ensuremath{\subp{\bar{\bm{\mathfrak{h}}}}{}{#2}{}{#1}}}
\newrobustcmd{\vecbmmathfrakhz}[2][]{\ensuremath{\subp{\vec{\bm{\mathfrak{h}}}}{}{#2}{}{#1}}}
\newrobustcmd{\mathfrakiz}[2][]{\ensuremath{\subp{\mathfrak{i}}{}{#2}{}{#1}}}
\newrobustcmd{\hatmathfrakiz}[2][]{\ensuremath{\subp{\hat{\mathfrak{i}}}{}{#2}{}{#1}}}
\newrobustcmd{\widehatmathfrakiz}[2][]{\ensuremath{\subp{\widehat{\mathfrak{i}}}{}{#2}{}{#1}}}
\newrobustcmd{\checkmathfrakiz}[2][]{\ensuremath{\subp{\check{\mathfrak{i}}}{}{#2}{}{#1}}}
\newrobustcmd{\tildemathfrakiz}[2][]{\ensuremath{\subp{\tilde{\mathfrak{i}}}{}{#2}{}{#1}}}
\newrobustcmd{\widetildemathfrakiz}[2][]{\ensuremath{\subp{\widetilde{\mathfrak{i}}}{}{#2}{}{#1}}}
\newrobustcmd{\acutemathfrakiz}[2][]{\ensuremath{\subp{\acute{\mathfrak{i}}}{}{#2}{}{#1}}}
\newrobustcmd{\gravemathfrakiz}[2][]{\ensuremath{\subp{\grave{\mathfrak{i}}}{}{#2}{}{#1}}}
\newrobustcmd{\dotmathfrakiz}[2][]{\ensuremath{\subp{\dot{\mathfrak{i}}}{}{#2}{}{#1}}}
\newrobustcmd{\ddotmathfrakiz}[2][]{\ensuremath{\subp{\ddot{\mathfrak{i}}}{}{#2}{}{#1}}}
\newrobustcmd{\brevemathfrakiz}[2][]{\ensuremath{\subp{\breve{\mathfrak{i}}}{}{#2}{}{#1}}}
\newrobustcmd{\barmathfrakiz}[2][]{\ensuremath{\subp{\bar{\mathfrak{i}}}{}{#2}{}{#1}}}
\newrobustcmd{\vecmathfrakiz}[2][]{\ensuremath{\subp{\vec{\mathfrak{i}}}{}{#2}{}{#1}}}
\newrobustcmd{\bmmathfrakiz}[2][]{\ensuremath{\subp{\bm{\mathfrak{i}}}{}{#2}{}{#1}}}
\newrobustcmd{\hatbmmathfrakiz}[2][]{\ensuremath{\subp{\hat{\bm{\mathfrak{i}}}}{}{#2}{}{#1}}}
\newrobustcmd{\widehatbmmathfrakiz}[2][]{\ensuremath{\subp{\widehat{\bm{\mathfrak{i}}}}{}{#2}{}{#1}}}
\newrobustcmd{\checkbmmathfrakiz}[2][]{\ensuremath{\subp{\check{\bm{\mathfrak{i}}}}{}{#2}{}{#1}}}
\newrobustcmd{\tildebmmathfrakiz}[2][]{\ensuremath{\subp{\tilde{\bm{\mathfrak{i}}}}{}{#2}{}{#1}}}
\newrobustcmd{\widetildebmmathfrakiz}[2][]{\ensuremath{\subp{\widetilde{\bm{\mathfrak{i}}}}{}{#2}{}{#1}}}
\newrobustcmd{\acutebmmathfrakiz}[2][]{\ensuremath{\subp{\acute{\bm{\mathfrak{i}}}}{}{#2}{}{#1}}}
\newrobustcmd{\gravebmmathfrakiz}[2][]{\ensuremath{\subp{\grave{\bm{\mathfrak{i}}}}{}{#2}{}{#1}}}
\newrobustcmd{\dotbmmathfrakiz}[2][]{\ensuremath{\subp{\dot{\bm{\mathfrak{i}}}}{}{#2}{}{#1}}}
\newrobustcmd{\ddotbmmathfrakiz}[2][]{\ensuremath{\subp{\ddot{\bm{\mathfrak{i}}}}{}{#2}{}{#1}}}
\newrobustcmd{\brevebmmathfrakiz}[2][]{\ensuremath{\subp{\breve{\bm{\mathfrak{i}}}}{}{#2}{}{#1}}}
\newrobustcmd{\barbmmathfrakiz}[2][]{\ensuremath{\subp{\bar{\bm{\mathfrak{i}}}}{}{#2}{}{#1}}}
\newrobustcmd{\vecbmmathfrakiz}[2][]{\ensuremath{\subp{\vec{\bm{\mathfrak{i}}}}{}{#2}{}{#1}}}
\newrobustcmd{\mathfrakjz}[2][]{\ensuremath{\subp{\mathfrak{j}}{}{#2}{}{#1}}}
\newrobustcmd{\hatmathfrakjz}[2][]{\ensuremath{\subp{\hat{\mathfrak{j}}}{}{#2}{}{#1}}}
\newrobustcmd{\widehatmathfrakjz}[2][]{\ensuremath{\subp{\widehat{\mathfrak{j}}}{}{#2}{}{#1}}}
\newrobustcmd{\checkmathfrakjz}[2][]{\ensuremath{\subp{\check{\mathfrak{j}}}{}{#2}{}{#1}}}
\newrobustcmd{\tildemathfrakjz}[2][]{\ensuremath{\subp{\tilde{\mathfrak{j}}}{}{#2}{}{#1}}}
\newrobustcmd{\widetildemathfrakjz}[2][]{\ensuremath{\subp{\widetilde{\mathfrak{j}}}{}{#2}{}{#1}}}
\newrobustcmd{\acutemathfrakjz}[2][]{\ensuremath{\subp{\acute{\mathfrak{j}}}{}{#2}{}{#1}}}
\newrobustcmd{\gravemathfrakjz}[2][]{\ensuremath{\subp{\grave{\mathfrak{j}}}{}{#2}{}{#1}}}
\newrobustcmd{\dotmathfrakjz}[2][]{\ensuremath{\subp{\dot{\mathfrak{j}}}{}{#2}{}{#1}}}
\newrobustcmd{\ddotmathfrakjz}[2][]{\ensuremath{\subp{\ddot{\mathfrak{j}}}{}{#2}{}{#1}}}
\newrobustcmd{\brevemathfrakjz}[2][]{\ensuremath{\subp{\breve{\mathfrak{j}}}{}{#2}{}{#1}}}
\newrobustcmd{\barmathfrakjz}[2][]{\ensuremath{\subp{\bar{\mathfrak{j}}}{}{#2}{}{#1}}}
\newrobustcmd{\vecmathfrakjz}[2][]{\ensuremath{\subp{\vec{\mathfrak{j}}}{}{#2}{}{#1}}}
\newrobustcmd{\bmmathfrakjz}[2][]{\ensuremath{\subp{\bm{\mathfrak{j}}}{}{#2}{}{#1}}}
\newrobustcmd{\hatbmmathfrakjz}[2][]{\ensuremath{\subp{\hat{\bm{\mathfrak{j}}}}{}{#2}{}{#1}}}
\newrobustcmd{\widehatbmmathfrakjz}[2][]{\ensuremath{\subp{\widehat{\bm{\mathfrak{j}}}}{}{#2}{}{#1}}}
\newrobustcmd{\checkbmmathfrakjz}[2][]{\ensuremath{\subp{\check{\bm{\mathfrak{j}}}}{}{#2}{}{#1}}}
\newrobustcmd{\tildebmmathfrakjz}[2][]{\ensuremath{\subp{\tilde{\bm{\mathfrak{j}}}}{}{#2}{}{#1}}}
\newrobustcmd{\widetildebmmathfrakjz}[2][]{\ensuremath{\subp{\widetilde{\bm{\mathfrak{j}}}}{}{#2}{}{#1}}}
\newrobustcmd{\acutebmmathfrakjz}[2][]{\ensuremath{\subp{\acute{\bm{\mathfrak{j}}}}{}{#2}{}{#1}}}
\newrobustcmd{\gravebmmathfrakjz}[2][]{\ensuremath{\subp{\grave{\bm{\mathfrak{j}}}}{}{#2}{}{#1}}}
\newrobustcmd{\dotbmmathfrakjz}[2][]{\ensuremath{\subp{\dot{\bm{\mathfrak{j}}}}{}{#2}{}{#1}}}
\newrobustcmd{\ddotbmmathfrakjz}[2][]{\ensuremath{\subp{\ddot{\bm{\mathfrak{j}}}}{}{#2}{}{#1}}}
\newrobustcmd{\brevebmmathfrakjz}[2][]{\ensuremath{\subp{\breve{\bm{\mathfrak{j}}}}{}{#2}{}{#1}}}
\newrobustcmd{\barbmmathfrakjz}[2][]{\ensuremath{\subp{\bar{\bm{\mathfrak{j}}}}{}{#2}{}{#1}}}
\newrobustcmd{\vecbmmathfrakjz}[2][]{\ensuremath{\subp{\vec{\bm{\mathfrak{j}}}}{}{#2}{}{#1}}}
\newrobustcmd{\mathfrakkz}[2][]{\ensuremath{\subp{\mathfrak{k}}{}{#2}{}{#1}}}
\newrobustcmd{\hatmathfrakkz}[2][]{\ensuremath{\subp{\hat{\mathfrak{k}}}{}{#2}{}{#1}}}
\newrobustcmd{\widehatmathfrakkz}[2][]{\ensuremath{\subp{\widehat{\mathfrak{k}}}{}{#2}{}{#1}}}
\newrobustcmd{\checkmathfrakkz}[2][]{\ensuremath{\subp{\check{\mathfrak{k}}}{}{#2}{}{#1}}}
\newrobustcmd{\tildemathfrakkz}[2][]{\ensuremath{\subp{\tilde{\mathfrak{k}}}{}{#2}{}{#1}}}
\newrobustcmd{\widetildemathfrakkz}[2][]{\ensuremath{\subp{\widetilde{\mathfrak{k}}}{}{#2}{}{#1}}}
\newrobustcmd{\acutemathfrakkz}[2][]{\ensuremath{\subp{\acute{\mathfrak{k}}}{}{#2}{}{#1}}}
\newrobustcmd{\gravemathfrakkz}[2][]{\ensuremath{\subp{\grave{\mathfrak{k}}}{}{#2}{}{#1}}}
\newrobustcmd{\dotmathfrakkz}[2][]{\ensuremath{\subp{\dot{\mathfrak{k}}}{}{#2}{}{#1}}}
\newrobustcmd{\ddotmathfrakkz}[2][]{\ensuremath{\subp{\ddot{\mathfrak{k}}}{}{#2}{}{#1}}}
\newrobustcmd{\brevemathfrakkz}[2][]{\ensuremath{\subp{\breve{\mathfrak{k}}}{}{#2}{}{#1}}}
\newrobustcmd{\barmathfrakkz}[2][]{\ensuremath{\subp{\bar{\mathfrak{k}}}{}{#2}{}{#1}}}
\newrobustcmd{\vecmathfrakkz}[2][]{\ensuremath{\subp{\vec{\mathfrak{k}}}{}{#2}{}{#1}}}
\newrobustcmd{\bmmathfrakkz}[2][]{\ensuremath{\subp{\bm{\mathfrak{k}}}{}{#2}{}{#1}}}
\newrobustcmd{\hatbmmathfrakkz}[2][]{\ensuremath{\subp{\hat{\bm{\mathfrak{k}}}}{}{#2}{}{#1}}}
\newrobustcmd{\widehatbmmathfrakkz}[2][]{\ensuremath{\subp{\widehat{\bm{\mathfrak{k}}}}{}{#2}{}{#1}}}
\newrobustcmd{\checkbmmathfrakkz}[2][]{\ensuremath{\subp{\check{\bm{\mathfrak{k}}}}{}{#2}{}{#1}}}
\newrobustcmd{\tildebmmathfrakkz}[2][]{\ensuremath{\subp{\tilde{\bm{\mathfrak{k}}}}{}{#2}{}{#1}}}
\newrobustcmd{\widetildebmmathfrakkz}[2][]{\ensuremath{\subp{\widetilde{\bm{\mathfrak{k}}}}{}{#2}{}{#1}}}
\newrobustcmd{\acutebmmathfrakkz}[2][]{\ensuremath{\subp{\acute{\bm{\mathfrak{k}}}}{}{#2}{}{#1}}}
\newrobustcmd{\gravebmmathfrakkz}[2][]{\ensuremath{\subp{\grave{\bm{\mathfrak{k}}}}{}{#2}{}{#1}}}
\newrobustcmd{\dotbmmathfrakkz}[2][]{\ensuremath{\subp{\dot{\bm{\mathfrak{k}}}}{}{#2}{}{#1}}}
\newrobustcmd{\ddotbmmathfrakkz}[2][]{\ensuremath{\subp{\ddot{\bm{\mathfrak{k}}}}{}{#2}{}{#1}}}
\newrobustcmd{\brevebmmathfrakkz}[2][]{\ensuremath{\subp{\breve{\bm{\mathfrak{k}}}}{}{#2}{}{#1}}}
\newrobustcmd{\barbmmathfrakkz}[2][]{\ensuremath{\subp{\bar{\bm{\mathfrak{k}}}}{}{#2}{}{#1}}}
\newrobustcmd{\vecbmmathfrakkz}[2][]{\ensuremath{\subp{\vec{\bm{\mathfrak{k}}}}{}{#2}{}{#1}}}
\newrobustcmd{\mathfraklz}[2][]{\ensuremath{\subp{\mathfrak{l}}{}{#2}{}{#1}}}
\newrobustcmd{\hatmathfraklz}[2][]{\ensuremath{\subp{\hat{\mathfrak{l}}}{}{#2}{}{#1}}}
\newrobustcmd{\widehatmathfraklz}[2][]{\ensuremath{\subp{\widehat{\mathfrak{l}}}{}{#2}{}{#1}}}
\newrobustcmd{\checkmathfraklz}[2][]{\ensuremath{\subp{\check{\mathfrak{l}}}{}{#2}{}{#1}}}
\newrobustcmd{\tildemathfraklz}[2][]{\ensuremath{\subp{\tilde{\mathfrak{l}}}{}{#2}{}{#1}}}
\newrobustcmd{\widetildemathfraklz}[2][]{\ensuremath{\subp{\widetilde{\mathfrak{l}}}{}{#2}{}{#1}}}
\newrobustcmd{\acutemathfraklz}[2][]{\ensuremath{\subp{\acute{\mathfrak{l}}}{}{#2}{}{#1}}}
\newrobustcmd{\gravemathfraklz}[2][]{\ensuremath{\subp{\grave{\mathfrak{l}}}{}{#2}{}{#1}}}
\newrobustcmd{\dotmathfraklz}[2][]{\ensuremath{\subp{\dot{\mathfrak{l}}}{}{#2}{}{#1}}}
\newrobustcmd{\ddotmathfraklz}[2][]{\ensuremath{\subp{\ddot{\mathfrak{l}}}{}{#2}{}{#1}}}
\newrobustcmd{\brevemathfraklz}[2][]{\ensuremath{\subp{\breve{\mathfrak{l}}}{}{#2}{}{#1}}}
\newrobustcmd{\barmathfraklz}[2][]{\ensuremath{\subp{\bar{\mathfrak{l}}}{}{#2}{}{#1}}}
\newrobustcmd{\vecmathfraklz}[2][]{\ensuremath{\subp{\vec{\mathfrak{l}}}{}{#2}{}{#1}}}
\newrobustcmd{\bmmathfraklz}[2][]{\ensuremath{\subp{\bm{\mathfrak{l}}}{}{#2}{}{#1}}}
\newrobustcmd{\hatbmmathfraklz}[2][]{\ensuremath{\subp{\hat{\bm{\mathfrak{l}}}}{}{#2}{}{#1}}}
\newrobustcmd{\widehatbmmathfraklz}[2][]{\ensuremath{\subp{\widehat{\bm{\mathfrak{l}}}}{}{#2}{}{#1}}}
\newrobustcmd{\checkbmmathfraklz}[2][]{\ensuremath{\subp{\check{\bm{\mathfrak{l}}}}{}{#2}{}{#1}}}
\newrobustcmd{\tildebmmathfraklz}[2][]{\ensuremath{\subp{\tilde{\bm{\mathfrak{l}}}}{}{#2}{}{#1}}}
\newrobustcmd{\widetildebmmathfraklz}[2][]{\ensuremath{\subp{\widetilde{\bm{\mathfrak{l}}}}{}{#2}{}{#1}}}
\newrobustcmd{\acutebmmathfraklz}[2][]{\ensuremath{\subp{\acute{\bm{\mathfrak{l}}}}{}{#2}{}{#1}}}
\newrobustcmd{\gravebmmathfraklz}[2][]{\ensuremath{\subp{\grave{\bm{\mathfrak{l}}}}{}{#2}{}{#1}}}
\newrobustcmd{\dotbmmathfraklz}[2][]{\ensuremath{\subp{\dot{\bm{\mathfrak{l}}}}{}{#2}{}{#1}}}
\newrobustcmd{\ddotbmmathfraklz}[2][]{\ensuremath{\subp{\ddot{\bm{\mathfrak{l}}}}{}{#2}{}{#1}}}
\newrobustcmd{\brevebmmathfraklz}[2][]{\ensuremath{\subp{\breve{\bm{\mathfrak{l}}}}{}{#2}{}{#1}}}
\newrobustcmd{\barbmmathfraklz}[2][]{\ensuremath{\subp{\bar{\bm{\mathfrak{l}}}}{}{#2}{}{#1}}}
\newrobustcmd{\vecbmmathfraklz}[2][]{\ensuremath{\subp{\vec{\bm{\mathfrak{l}}}}{}{#2}{}{#1}}}
\newrobustcmd{\mathfrakmz}[2][]{\ensuremath{\subp{\mathfrak{m}}{}{#2}{}{#1}}}
\newrobustcmd{\hatmathfrakmz}[2][]{\ensuremath{\subp{\hat{\mathfrak{m}}}{}{#2}{}{#1}}}
\newrobustcmd{\widehatmathfrakmz}[2][]{\ensuremath{\subp{\widehat{\mathfrak{m}}}{}{#2}{}{#1}}}
\newrobustcmd{\checkmathfrakmz}[2][]{\ensuremath{\subp{\check{\mathfrak{m}}}{}{#2}{}{#1}}}
\newrobustcmd{\tildemathfrakmz}[2][]{\ensuremath{\subp{\tilde{\mathfrak{m}}}{}{#2}{}{#1}}}
\newrobustcmd{\widetildemathfrakmz}[2][]{\ensuremath{\subp{\widetilde{\mathfrak{m}}}{}{#2}{}{#1}}}
\newrobustcmd{\acutemathfrakmz}[2][]{\ensuremath{\subp{\acute{\mathfrak{m}}}{}{#2}{}{#1}}}
\newrobustcmd{\gravemathfrakmz}[2][]{\ensuremath{\subp{\grave{\mathfrak{m}}}{}{#2}{}{#1}}}
\newrobustcmd{\dotmathfrakmz}[2][]{\ensuremath{\subp{\dot{\mathfrak{m}}}{}{#2}{}{#1}}}
\newrobustcmd{\ddotmathfrakmz}[2][]{\ensuremath{\subp{\ddot{\mathfrak{m}}}{}{#2}{}{#1}}}
\newrobustcmd{\brevemathfrakmz}[2][]{\ensuremath{\subp{\breve{\mathfrak{m}}}{}{#2}{}{#1}}}
\newrobustcmd{\barmathfrakmz}[2][]{\ensuremath{\subp{\bar{\mathfrak{m}}}{}{#2}{}{#1}}}
\newrobustcmd{\vecmathfrakmz}[2][]{\ensuremath{\subp{\vec{\mathfrak{m}}}{}{#2}{}{#1}}}
\newrobustcmd{\bmmathfrakmz}[2][]{\ensuremath{\subp{\bm{\mathfrak{m}}}{}{#2}{}{#1}}}
\newrobustcmd{\hatbmmathfrakmz}[2][]{\ensuremath{\subp{\hat{\bm{\mathfrak{m}}}}{}{#2}{}{#1}}}
\newrobustcmd{\widehatbmmathfrakmz}[2][]{\ensuremath{\subp{\widehat{\bm{\mathfrak{m}}}}{}{#2}{}{#1}}}
\newrobustcmd{\checkbmmathfrakmz}[2][]{\ensuremath{\subp{\check{\bm{\mathfrak{m}}}}{}{#2}{}{#1}}}
\newrobustcmd{\tildebmmathfrakmz}[2][]{\ensuremath{\subp{\tilde{\bm{\mathfrak{m}}}}{}{#2}{}{#1}}}
\newrobustcmd{\widetildebmmathfrakmz}[2][]{\ensuremath{\subp{\widetilde{\bm{\mathfrak{m}}}}{}{#2}{}{#1}}}
\newrobustcmd{\acutebmmathfrakmz}[2][]{\ensuremath{\subp{\acute{\bm{\mathfrak{m}}}}{}{#2}{}{#1}}}
\newrobustcmd{\gravebmmathfrakmz}[2][]{\ensuremath{\subp{\grave{\bm{\mathfrak{m}}}}{}{#2}{}{#1}}}
\newrobustcmd{\dotbmmathfrakmz}[2][]{\ensuremath{\subp{\dot{\bm{\mathfrak{m}}}}{}{#2}{}{#1}}}
\newrobustcmd{\ddotbmmathfrakmz}[2][]{\ensuremath{\subp{\ddot{\bm{\mathfrak{m}}}}{}{#2}{}{#1}}}
\newrobustcmd{\brevebmmathfrakmz}[2][]{\ensuremath{\subp{\breve{\bm{\mathfrak{m}}}}{}{#2}{}{#1}}}
\newrobustcmd{\barbmmathfrakmz}[2][]{\ensuremath{\subp{\bar{\bm{\mathfrak{m}}}}{}{#2}{}{#1}}}
\newrobustcmd{\vecbmmathfrakmz}[2][]{\ensuremath{\subp{\vec{\bm{\mathfrak{m}}}}{}{#2}{}{#1}}}
\newrobustcmd{\mathfraknz}[2][]{\ensuremath{\subp{\mathfrak{n}}{}{#2}{}{#1}}}
\newrobustcmd{\hatmathfraknz}[2][]{\ensuremath{\subp{\hat{\mathfrak{n}}}{}{#2}{}{#1}}}
\newrobustcmd{\widehatmathfraknz}[2][]{\ensuremath{\subp{\widehat{\mathfrak{n}}}{}{#2}{}{#1}}}
\newrobustcmd{\checkmathfraknz}[2][]{\ensuremath{\subp{\check{\mathfrak{n}}}{}{#2}{}{#1}}}
\newrobustcmd{\tildemathfraknz}[2][]{\ensuremath{\subp{\tilde{\mathfrak{n}}}{}{#2}{}{#1}}}
\newrobustcmd{\widetildemathfraknz}[2][]{\ensuremath{\subp{\widetilde{\mathfrak{n}}}{}{#2}{}{#1}}}
\newrobustcmd{\acutemathfraknz}[2][]{\ensuremath{\subp{\acute{\mathfrak{n}}}{}{#2}{}{#1}}}
\newrobustcmd{\gravemathfraknz}[2][]{\ensuremath{\subp{\grave{\mathfrak{n}}}{}{#2}{}{#1}}}
\newrobustcmd{\dotmathfraknz}[2][]{\ensuremath{\subp{\dot{\mathfrak{n}}}{}{#2}{}{#1}}}
\newrobustcmd{\ddotmathfraknz}[2][]{\ensuremath{\subp{\ddot{\mathfrak{n}}}{}{#2}{}{#1}}}
\newrobustcmd{\brevemathfraknz}[2][]{\ensuremath{\subp{\breve{\mathfrak{n}}}{}{#2}{}{#1}}}
\newrobustcmd{\barmathfraknz}[2][]{\ensuremath{\subp{\bar{\mathfrak{n}}}{}{#2}{}{#1}}}
\newrobustcmd{\vecmathfraknz}[2][]{\ensuremath{\subp{\vec{\mathfrak{n}}}{}{#2}{}{#1}}}
\newrobustcmd{\bmmathfraknz}[2][]{\ensuremath{\subp{\bm{\mathfrak{n}}}{}{#2}{}{#1}}}
\newrobustcmd{\hatbmmathfraknz}[2][]{\ensuremath{\subp{\hat{\bm{\mathfrak{n}}}}{}{#2}{}{#1}}}
\newrobustcmd{\widehatbmmathfraknz}[2][]{\ensuremath{\subp{\widehat{\bm{\mathfrak{n}}}}{}{#2}{}{#1}}}
\newrobustcmd{\checkbmmathfraknz}[2][]{\ensuremath{\subp{\check{\bm{\mathfrak{n}}}}{}{#2}{}{#1}}}
\newrobustcmd{\tildebmmathfraknz}[2][]{\ensuremath{\subp{\tilde{\bm{\mathfrak{n}}}}{}{#2}{}{#1}}}
\newrobustcmd{\widetildebmmathfraknz}[2][]{\ensuremath{\subp{\widetilde{\bm{\mathfrak{n}}}}{}{#2}{}{#1}}}
\newrobustcmd{\acutebmmathfraknz}[2][]{\ensuremath{\subp{\acute{\bm{\mathfrak{n}}}}{}{#2}{}{#1}}}
\newrobustcmd{\gravebmmathfraknz}[2][]{\ensuremath{\subp{\grave{\bm{\mathfrak{n}}}}{}{#2}{}{#1}}}
\newrobustcmd{\dotbmmathfraknz}[2][]{\ensuremath{\subp{\dot{\bm{\mathfrak{n}}}}{}{#2}{}{#1}}}
\newrobustcmd{\ddotbmmathfraknz}[2][]{\ensuremath{\subp{\ddot{\bm{\mathfrak{n}}}}{}{#2}{}{#1}}}
\newrobustcmd{\brevebmmathfraknz}[2][]{\ensuremath{\subp{\breve{\bm{\mathfrak{n}}}}{}{#2}{}{#1}}}
\newrobustcmd{\barbmmathfraknz}[2][]{\ensuremath{\subp{\bar{\bm{\mathfrak{n}}}}{}{#2}{}{#1}}}
\newrobustcmd{\vecbmmathfraknz}[2][]{\ensuremath{\subp{\vec{\bm{\mathfrak{n}}}}{}{#2}{}{#1}}}
\newrobustcmd{\mathfrakoz}[2][]{\ensuremath{\subp{\mathfrak{o}}{}{#2}{}{#1}}}
\newrobustcmd{\hatmathfrakoz}[2][]{\ensuremath{\subp{\hat{\mathfrak{o}}}{}{#2}{}{#1}}}
\newrobustcmd{\widehatmathfrakoz}[2][]{\ensuremath{\subp{\widehat{\mathfrak{o}}}{}{#2}{}{#1}}}
\newrobustcmd{\checkmathfrakoz}[2][]{\ensuremath{\subp{\check{\mathfrak{o}}}{}{#2}{}{#1}}}
\newrobustcmd{\tildemathfrakoz}[2][]{\ensuremath{\subp{\tilde{\mathfrak{o}}}{}{#2}{}{#1}}}
\newrobustcmd{\widetildemathfrakoz}[2][]{\ensuremath{\subp{\widetilde{\mathfrak{o}}}{}{#2}{}{#1}}}
\newrobustcmd{\acutemathfrakoz}[2][]{\ensuremath{\subp{\acute{\mathfrak{o}}}{}{#2}{}{#1}}}
\newrobustcmd{\gravemathfrakoz}[2][]{\ensuremath{\subp{\grave{\mathfrak{o}}}{}{#2}{}{#1}}}
\newrobustcmd{\dotmathfrakoz}[2][]{\ensuremath{\subp{\dot{\mathfrak{o}}}{}{#2}{}{#1}}}
\newrobustcmd{\ddotmathfrakoz}[2][]{\ensuremath{\subp{\ddot{\mathfrak{o}}}{}{#2}{}{#1}}}
\newrobustcmd{\brevemathfrakoz}[2][]{\ensuremath{\subp{\breve{\mathfrak{o}}}{}{#2}{}{#1}}}
\newrobustcmd{\barmathfrakoz}[2][]{\ensuremath{\subp{\bar{\mathfrak{o}}}{}{#2}{}{#1}}}
\newrobustcmd{\vecmathfrakoz}[2][]{\ensuremath{\subp{\vec{\mathfrak{o}}}{}{#2}{}{#1}}}
\newrobustcmd{\bmmathfrakoz}[2][]{\ensuremath{\subp{\bm{\mathfrak{o}}}{}{#2}{}{#1}}}
\newrobustcmd{\hatbmmathfrakoz}[2][]{\ensuremath{\subp{\hat{\bm{\mathfrak{o}}}}{}{#2}{}{#1}}}
\newrobustcmd{\widehatbmmathfrakoz}[2][]{\ensuremath{\subp{\widehat{\bm{\mathfrak{o}}}}{}{#2}{}{#1}}}
\newrobustcmd{\checkbmmathfrakoz}[2][]{\ensuremath{\subp{\check{\bm{\mathfrak{o}}}}{}{#2}{}{#1}}}
\newrobustcmd{\tildebmmathfrakoz}[2][]{\ensuremath{\subp{\tilde{\bm{\mathfrak{o}}}}{}{#2}{}{#1}}}
\newrobustcmd{\widetildebmmathfrakoz}[2][]{\ensuremath{\subp{\widetilde{\bm{\mathfrak{o}}}}{}{#2}{}{#1}}}
\newrobustcmd{\acutebmmathfrakoz}[2][]{\ensuremath{\subp{\acute{\bm{\mathfrak{o}}}}{}{#2}{}{#1}}}
\newrobustcmd{\gravebmmathfrakoz}[2][]{\ensuremath{\subp{\grave{\bm{\mathfrak{o}}}}{}{#2}{}{#1}}}
\newrobustcmd{\dotbmmathfrakoz}[2][]{\ensuremath{\subp{\dot{\bm{\mathfrak{o}}}}{}{#2}{}{#1}}}
\newrobustcmd{\ddotbmmathfrakoz}[2][]{\ensuremath{\subp{\ddot{\bm{\mathfrak{o}}}}{}{#2}{}{#1}}}
\newrobustcmd{\brevebmmathfrakoz}[2][]{\ensuremath{\subp{\breve{\bm{\mathfrak{o}}}}{}{#2}{}{#1}}}
\newrobustcmd{\barbmmathfrakoz}[2][]{\ensuremath{\subp{\bar{\bm{\mathfrak{o}}}}{}{#2}{}{#1}}}
\newrobustcmd{\vecbmmathfrakoz}[2][]{\ensuremath{\subp{\vec{\bm{\mathfrak{o}}}}{}{#2}{}{#1}}}
\newrobustcmd{\mathfrakpz}[2][]{\ensuremath{\subp{\mathfrak{p}}{}{#2}{}{#1}}}
\newrobustcmd{\hatmathfrakpz}[2][]{\ensuremath{\subp{\hat{\mathfrak{p}}}{}{#2}{}{#1}}}
\newrobustcmd{\widehatmathfrakpz}[2][]{\ensuremath{\subp{\widehat{\mathfrak{p}}}{}{#2}{}{#1}}}
\newrobustcmd{\checkmathfrakpz}[2][]{\ensuremath{\subp{\check{\mathfrak{p}}}{}{#2}{}{#1}}}
\newrobustcmd{\tildemathfrakpz}[2][]{\ensuremath{\subp{\tilde{\mathfrak{p}}}{}{#2}{}{#1}}}
\newrobustcmd{\widetildemathfrakpz}[2][]{\ensuremath{\subp{\widetilde{\mathfrak{p}}}{}{#2}{}{#1}}}
\newrobustcmd{\acutemathfrakpz}[2][]{\ensuremath{\subp{\acute{\mathfrak{p}}}{}{#2}{}{#1}}}
\newrobustcmd{\gravemathfrakpz}[2][]{\ensuremath{\subp{\grave{\mathfrak{p}}}{}{#2}{}{#1}}}
\newrobustcmd{\dotmathfrakpz}[2][]{\ensuremath{\subp{\dot{\mathfrak{p}}}{}{#2}{}{#1}}}
\newrobustcmd{\ddotmathfrakpz}[2][]{\ensuremath{\subp{\ddot{\mathfrak{p}}}{}{#2}{}{#1}}}
\newrobustcmd{\brevemathfrakpz}[2][]{\ensuremath{\subp{\breve{\mathfrak{p}}}{}{#2}{}{#1}}}
\newrobustcmd{\barmathfrakpz}[2][]{\ensuremath{\subp{\bar{\mathfrak{p}}}{}{#2}{}{#1}}}
\newrobustcmd{\vecmathfrakpz}[2][]{\ensuremath{\subp{\vec{\mathfrak{p}}}{}{#2}{}{#1}}}
\newrobustcmd{\bmmathfrakpz}[2][]{\ensuremath{\subp{\bm{\mathfrak{p}}}{}{#2}{}{#1}}}
\newrobustcmd{\hatbmmathfrakpz}[2][]{\ensuremath{\subp{\hat{\bm{\mathfrak{p}}}}{}{#2}{}{#1}}}
\newrobustcmd{\widehatbmmathfrakpz}[2][]{\ensuremath{\subp{\widehat{\bm{\mathfrak{p}}}}{}{#2}{}{#1}}}
\newrobustcmd{\checkbmmathfrakpz}[2][]{\ensuremath{\subp{\check{\bm{\mathfrak{p}}}}{}{#2}{}{#1}}}
\newrobustcmd{\tildebmmathfrakpz}[2][]{\ensuremath{\subp{\tilde{\bm{\mathfrak{p}}}}{}{#2}{}{#1}}}
\newrobustcmd{\widetildebmmathfrakpz}[2][]{\ensuremath{\subp{\widetilde{\bm{\mathfrak{p}}}}{}{#2}{}{#1}}}
\newrobustcmd{\acutebmmathfrakpz}[2][]{\ensuremath{\subp{\acute{\bm{\mathfrak{p}}}}{}{#2}{}{#1}}}
\newrobustcmd{\gravebmmathfrakpz}[2][]{\ensuremath{\subp{\grave{\bm{\mathfrak{p}}}}{}{#2}{}{#1}}}
\newrobustcmd{\dotbmmathfrakpz}[2][]{\ensuremath{\subp{\dot{\bm{\mathfrak{p}}}}{}{#2}{}{#1}}}
\newrobustcmd{\ddotbmmathfrakpz}[2][]{\ensuremath{\subp{\ddot{\bm{\mathfrak{p}}}}{}{#2}{}{#1}}}
\newrobustcmd{\brevebmmathfrakpz}[2][]{\ensuremath{\subp{\breve{\bm{\mathfrak{p}}}}{}{#2}{}{#1}}}
\newrobustcmd{\barbmmathfrakpz}[2][]{\ensuremath{\subp{\bar{\bm{\mathfrak{p}}}}{}{#2}{}{#1}}}
\newrobustcmd{\vecbmmathfrakpz}[2][]{\ensuremath{\subp{\vec{\bm{\mathfrak{p}}}}{}{#2}{}{#1}}}
\newrobustcmd{\mathfrakqz}[2][]{\ensuremath{\subp{\mathfrak{q}}{}{#2}{}{#1}}}
\newrobustcmd{\hatmathfrakqz}[2][]{\ensuremath{\subp{\hat{\mathfrak{q}}}{}{#2}{}{#1}}}
\newrobustcmd{\widehatmathfrakqz}[2][]{\ensuremath{\subp{\widehat{\mathfrak{q}}}{}{#2}{}{#1}}}
\newrobustcmd{\checkmathfrakqz}[2][]{\ensuremath{\subp{\check{\mathfrak{q}}}{}{#2}{}{#1}}}
\newrobustcmd{\tildemathfrakqz}[2][]{\ensuremath{\subp{\tilde{\mathfrak{q}}}{}{#2}{}{#1}}}
\newrobustcmd{\widetildemathfrakqz}[2][]{\ensuremath{\subp{\widetilde{\mathfrak{q}}}{}{#2}{}{#1}}}
\newrobustcmd{\acutemathfrakqz}[2][]{\ensuremath{\subp{\acute{\mathfrak{q}}}{}{#2}{}{#1}}}
\newrobustcmd{\gravemathfrakqz}[2][]{\ensuremath{\subp{\grave{\mathfrak{q}}}{}{#2}{}{#1}}}
\newrobustcmd{\dotmathfrakqz}[2][]{\ensuremath{\subp{\dot{\mathfrak{q}}}{}{#2}{}{#1}}}
\newrobustcmd{\ddotmathfrakqz}[2][]{\ensuremath{\subp{\ddot{\mathfrak{q}}}{}{#2}{}{#1}}}
\newrobustcmd{\brevemathfrakqz}[2][]{\ensuremath{\subp{\breve{\mathfrak{q}}}{}{#2}{}{#1}}}
\newrobustcmd{\barmathfrakqz}[2][]{\ensuremath{\subp{\bar{\mathfrak{q}}}{}{#2}{}{#1}}}
\newrobustcmd{\vecmathfrakqz}[2][]{\ensuremath{\subp{\vec{\mathfrak{q}}}{}{#2}{}{#1}}}
\newrobustcmd{\bmmathfrakqz}[2][]{\ensuremath{\subp{\bm{\mathfrak{q}}}{}{#2}{}{#1}}}
\newrobustcmd{\hatbmmathfrakqz}[2][]{\ensuremath{\subp{\hat{\bm{\mathfrak{q}}}}{}{#2}{}{#1}}}
\newrobustcmd{\widehatbmmathfrakqz}[2][]{\ensuremath{\subp{\widehat{\bm{\mathfrak{q}}}}{}{#2}{}{#1}}}
\newrobustcmd{\checkbmmathfrakqz}[2][]{\ensuremath{\subp{\check{\bm{\mathfrak{q}}}}{}{#2}{}{#1}}}
\newrobustcmd{\tildebmmathfrakqz}[2][]{\ensuremath{\subp{\tilde{\bm{\mathfrak{q}}}}{}{#2}{}{#1}}}
\newrobustcmd{\widetildebmmathfrakqz}[2][]{\ensuremath{\subp{\widetilde{\bm{\mathfrak{q}}}}{}{#2}{}{#1}}}
\newrobustcmd{\acutebmmathfrakqz}[2][]{\ensuremath{\subp{\acute{\bm{\mathfrak{q}}}}{}{#2}{}{#1}}}
\newrobustcmd{\gravebmmathfrakqz}[2][]{\ensuremath{\subp{\grave{\bm{\mathfrak{q}}}}{}{#2}{}{#1}}}
\newrobustcmd{\dotbmmathfrakqz}[2][]{\ensuremath{\subp{\dot{\bm{\mathfrak{q}}}}{}{#2}{}{#1}}}
\newrobustcmd{\ddotbmmathfrakqz}[2][]{\ensuremath{\subp{\ddot{\bm{\mathfrak{q}}}}{}{#2}{}{#1}}}
\newrobustcmd{\brevebmmathfrakqz}[2][]{\ensuremath{\subp{\breve{\bm{\mathfrak{q}}}}{}{#2}{}{#1}}}
\newrobustcmd{\barbmmathfrakqz}[2][]{\ensuremath{\subp{\bar{\bm{\mathfrak{q}}}}{}{#2}{}{#1}}}
\newrobustcmd{\vecbmmathfrakqz}[2][]{\ensuremath{\subp{\vec{\bm{\mathfrak{q}}}}{}{#2}{}{#1}}}
\newrobustcmd{\mathfrakrz}[2][]{\ensuremath{\subp{\mathfrak{r}}{}{#2}{}{#1}}}
\newrobustcmd{\hatmathfrakrz}[2][]{\ensuremath{\subp{\hat{\mathfrak{r}}}{}{#2}{}{#1}}}
\newrobustcmd{\widehatmathfrakrz}[2][]{\ensuremath{\subp{\widehat{\mathfrak{r}}}{}{#2}{}{#1}}}
\newrobustcmd{\checkmathfrakrz}[2][]{\ensuremath{\subp{\check{\mathfrak{r}}}{}{#2}{}{#1}}}
\newrobustcmd{\tildemathfrakrz}[2][]{\ensuremath{\subp{\tilde{\mathfrak{r}}}{}{#2}{}{#1}}}
\newrobustcmd{\widetildemathfrakrz}[2][]{\ensuremath{\subp{\widetilde{\mathfrak{r}}}{}{#2}{}{#1}}}
\newrobustcmd{\acutemathfrakrz}[2][]{\ensuremath{\subp{\acute{\mathfrak{r}}}{}{#2}{}{#1}}}
\newrobustcmd{\gravemathfrakrz}[2][]{\ensuremath{\subp{\grave{\mathfrak{r}}}{}{#2}{}{#1}}}
\newrobustcmd{\dotmathfrakrz}[2][]{\ensuremath{\subp{\dot{\mathfrak{r}}}{}{#2}{}{#1}}}
\newrobustcmd{\ddotmathfrakrz}[2][]{\ensuremath{\subp{\ddot{\mathfrak{r}}}{}{#2}{}{#1}}}
\newrobustcmd{\brevemathfrakrz}[2][]{\ensuremath{\subp{\breve{\mathfrak{r}}}{}{#2}{}{#1}}}
\newrobustcmd{\barmathfrakrz}[2][]{\ensuremath{\subp{\bar{\mathfrak{r}}}{}{#2}{}{#1}}}
\newrobustcmd{\vecmathfrakrz}[2][]{\ensuremath{\subp{\vec{\mathfrak{r}}}{}{#2}{}{#1}}}
\newrobustcmd{\bmmathfrakrz}[2][]{\ensuremath{\subp{\bm{\mathfrak{r}}}{}{#2}{}{#1}}}
\newrobustcmd{\hatbmmathfrakrz}[2][]{\ensuremath{\subp{\hat{\bm{\mathfrak{r}}}}{}{#2}{}{#1}}}
\newrobustcmd{\widehatbmmathfrakrz}[2][]{\ensuremath{\subp{\widehat{\bm{\mathfrak{r}}}}{}{#2}{}{#1}}}
\newrobustcmd{\checkbmmathfrakrz}[2][]{\ensuremath{\subp{\check{\bm{\mathfrak{r}}}}{}{#2}{}{#1}}}
\newrobustcmd{\tildebmmathfrakrz}[2][]{\ensuremath{\subp{\tilde{\bm{\mathfrak{r}}}}{}{#2}{}{#1}}}
\newrobustcmd{\widetildebmmathfrakrz}[2][]{\ensuremath{\subp{\widetilde{\bm{\mathfrak{r}}}}{}{#2}{}{#1}}}
\newrobustcmd{\acutebmmathfrakrz}[2][]{\ensuremath{\subp{\acute{\bm{\mathfrak{r}}}}{}{#2}{}{#1}}}
\newrobustcmd{\gravebmmathfrakrz}[2][]{\ensuremath{\subp{\grave{\bm{\mathfrak{r}}}}{}{#2}{}{#1}}}
\newrobustcmd{\dotbmmathfrakrz}[2][]{\ensuremath{\subp{\dot{\bm{\mathfrak{r}}}}{}{#2}{}{#1}}}
\newrobustcmd{\ddotbmmathfrakrz}[2][]{\ensuremath{\subp{\ddot{\bm{\mathfrak{r}}}}{}{#2}{}{#1}}}
\newrobustcmd{\brevebmmathfrakrz}[2][]{\ensuremath{\subp{\breve{\bm{\mathfrak{r}}}}{}{#2}{}{#1}}}
\newrobustcmd{\barbmmathfrakrz}[2][]{\ensuremath{\subp{\bar{\bm{\mathfrak{r}}}}{}{#2}{}{#1}}}
\newrobustcmd{\vecbmmathfrakrz}[2][]{\ensuremath{\subp{\vec{\bm{\mathfrak{r}}}}{}{#2}{}{#1}}}
\newrobustcmd{\mathfraksz}[2][]{\ensuremath{\subp{\mathfrak{s}}{}{#2}{}{#1}}}
\newrobustcmd{\hatmathfraksz}[2][]{\ensuremath{\subp{\hat{\mathfrak{s}}}{}{#2}{}{#1}}}
\newrobustcmd{\widehatmathfraksz}[2][]{\ensuremath{\subp{\widehat{\mathfrak{s}}}{}{#2}{}{#1}}}
\newrobustcmd{\checkmathfraksz}[2][]{\ensuremath{\subp{\check{\mathfrak{s}}}{}{#2}{}{#1}}}
\newrobustcmd{\tildemathfraksz}[2][]{\ensuremath{\subp{\tilde{\mathfrak{s}}}{}{#2}{}{#1}}}
\newrobustcmd{\widetildemathfraksz}[2][]{\ensuremath{\subp{\widetilde{\mathfrak{s}}}{}{#2}{}{#1}}}
\newrobustcmd{\acutemathfraksz}[2][]{\ensuremath{\subp{\acute{\mathfrak{s}}}{}{#2}{}{#1}}}
\newrobustcmd{\gravemathfraksz}[2][]{\ensuremath{\subp{\grave{\mathfrak{s}}}{}{#2}{}{#1}}}
\newrobustcmd{\dotmathfraksz}[2][]{\ensuremath{\subp{\dot{\mathfrak{s}}}{}{#2}{}{#1}}}
\newrobustcmd{\ddotmathfraksz}[2][]{\ensuremath{\subp{\ddot{\mathfrak{s}}}{}{#2}{}{#1}}}
\newrobustcmd{\brevemathfraksz}[2][]{\ensuremath{\subp{\breve{\mathfrak{s}}}{}{#2}{}{#1}}}
\newrobustcmd{\barmathfraksz}[2][]{\ensuremath{\subp{\bar{\mathfrak{s}}}{}{#2}{}{#1}}}
\newrobustcmd{\vecmathfraksz}[2][]{\ensuremath{\subp{\vec{\mathfrak{s}}}{}{#2}{}{#1}}}
\newrobustcmd{\bmmathfraksz}[2][]{\ensuremath{\subp{\bm{\mathfrak{s}}}{}{#2}{}{#1}}}
\newrobustcmd{\hatbmmathfraksz}[2][]{\ensuremath{\subp{\hat{\bm{\mathfrak{s}}}}{}{#2}{}{#1}}}
\newrobustcmd{\widehatbmmathfraksz}[2][]{\ensuremath{\subp{\widehat{\bm{\mathfrak{s}}}}{}{#2}{}{#1}}}
\newrobustcmd{\checkbmmathfraksz}[2][]{\ensuremath{\subp{\check{\bm{\mathfrak{s}}}}{}{#2}{}{#1}}}
\newrobustcmd{\tildebmmathfraksz}[2][]{\ensuremath{\subp{\tilde{\bm{\mathfrak{s}}}}{}{#2}{}{#1}}}
\newrobustcmd{\widetildebmmathfraksz}[2][]{\ensuremath{\subp{\widetilde{\bm{\mathfrak{s}}}}{}{#2}{}{#1}}}
\newrobustcmd{\acutebmmathfraksz}[2][]{\ensuremath{\subp{\acute{\bm{\mathfrak{s}}}}{}{#2}{}{#1}}}
\newrobustcmd{\gravebmmathfraksz}[2][]{\ensuremath{\subp{\grave{\bm{\mathfrak{s}}}}{}{#2}{}{#1}}}
\newrobustcmd{\dotbmmathfraksz}[2][]{\ensuremath{\subp{\dot{\bm{\mathfrak{s}}}}{}{#2}{}{#1}}}
\newrobustcmd{\ddotbmmathfraksz}[2][]{\ensuremath{\subp{\ddot{\bm{\mathfrak{s}}}}{}{#2}{}{#1}}}
\newrobustcmd{\brevebmmathfraksz}[2][]{\ensuremath{\subp{\breve{\bm{\mathfrak{s}}}}{}{#2}{}{#1}}}
\newrobustcmd{\barbmmathfraksz}[2][]{\ensuremath{\subp{\bar{\bm{\mathfrak{s}}}}{}{#2}{}{#1}}}
\newrobustcmd{\vecbmmathfraksz}[2][]{\ensuremath{\subp{\vec{\bm{\mathfrak{s}}}}{}{#2}{}{#1}}}
\newrobustcmd{\mathfraktz}[2][]{\ensuremath{\subp{\mathfrak{t}}{}{#2}{}{#1}}}
\newrobustcmd{\hatmathfraktz}[2][]{\ensuremath{\subp{\hat{\mathfrak{t}}}{}{#2}{}{#1}}}
\newrobustcmd{\widehatmathfraktz}[2][]{\ensuremath{\subp{\widehat{\mathfrak{t}}}{}{#2}{}{#1}}}
\newrobustcmd{\checkmathfraktz}[2][]{\ensuremath{\subp{\check{\mathfrak{t}}}{}{#2}{}{#1}}}
\newrobustcmd{\tildemathfraktz}[2][]{\ensuremath{\subp{\tilde{\mathfrak{t}}}{}{#2}{}{#1}}}
\newrobustcmd{\widetildemathfraktz}[2][]{\ensuremath{\subp{\widetilde{\mathfrak{t}}}{}{#2}{}{#1}}}
\newrobustcmd{\acutemathfraktz}[2][]{\ensuremath{\subp{\acute{\mathfrak{t}}}{}{#2}{}{#1}}}
\newrobustcmd{\gravemathfraktz}[2][]{\ensuremath{\subp{\grave{\mathfrak{t}}}{}{#2}{}{#1}}}
\newrobustcmd{\dotmathfraktz}[2][]{\ensuremath{\subp{\dot{\mathfrak{t}}}{}{#2}{}{#1}}}
\newrobustcmd{\ddotmathfraktz}[2][]{\ensuremath{\subp{\ddot{\mathfrak{t}}}{}{#2}{}{#1}}}
\newrobustcmd{\brevemathfraktz}[2][]{\ensuremath{\subp{\breve{\mathfrak{t}}}{}{#2}{}{#1}}}
\newrobustcmd{\barmathfraktz}[2][]{\ensuremath{\subp{\bar{\mathfrak{t}}}{}{#2}{}{#1}}}
\newrobustcmd{\vecmathfraktz}[2][]{\ensuremath{\subp{\vec{\mathfrak{t}}}{}{#2}{}{#1}}}
\newrobustcmd{\bmmathfraktz}[2][]{\ensuremath{\subp{\bm{\mathfrak{t}}}{}{#2}{}{#1}}}
\newrobustcmd{\hatbmmathfraktz}[2][]{\ensuremath{\subp{\hat{\bm{\mathfrak{t}}}}{}{#2}{}{#1}}}
\newrobustcmd{\widehatbmmathfraktz}[2][]{\ensuremath{\subp{\widehat{\bm{\mathfrak{t}}}}{}{#2}{}{#1}}}
\newrobustcmd{\checkbmmathfraktz}[2][]{\ensuremath{\subp{\check{\bm{\mathfrak{t}}}}{}{#2}{}{#1}}}
\newrobustcmd{\tildebmmathfraktz}[2][]{\ensuremath{\subp{\tilde{\bm{\mathfrak{t}}}}{}{#2}{}{#1}}}
\newrobustcmd{\widetildebmmathfraktz}[2][]{\ensuremath{\subp{\widetilde{\bm{\mathfrak{t}}}}{}{#2}{}{#1}}}
\newrobustcmd{\acutebmmathfraktz}[2][]{\ensuremath{\subp{\acute{\bm{\mathfrak{t}}}}{}{#2}{}{#1}}}
\newrobustcmd{\gravebmmathfraktz}[2][]{\ensuremath{\subp{\grave{\bm{\mathfrak{t}}}}{}{#2}{}{#1}}}
\newrobustcmd{\dotbmmathfraktz}[2][]{\ensuremath{\subp{\dot{\bm{\mathfrak{t}}}}{}{#2}{}{#1}}}
\newrobustcmd{\ddotbmmathfraktz}[2][]{\ensuremath{\subp{\ddot{\bm{\mathfrak{t}}}}{}{#2}{}{#1}}}
\newrobustcmd{\brevebmmathfraktz}[2][]{\ensuremath{\subp{\breve{\bm{\mathfrak{t}}}}{}{#2}{}{#1}}}
\newrobustcmd{\barbmmathfraktz}[2][]{\ensuremath{\subp{\bar{\bm{\mathfrak{t}}}}{}{#2}{}{#1}}}
\newrobustcmd{\vecbmmathfraktz}[2][]{\ensuremath{\subp{\vec{\bm{\mathfrak{t}}}}{}{#2}{}{#1}}}
\newrobustcmd{\mathfrakuz}[2][]{\ensuremath{\subp{\mathfrak{u}}{}{#2}{}{#1}}}
\newrobustcmd{\hatmathfrakuz}[2][]{\ensuremath{\subp{\hat{\mathfrak{u}}}{}{#2}{}{#1}}}
\newrobustcmd{\widehatmathfrakuz}[2][]{\ensuremath{\subp{\widehat{\mathfrak{u}}}{}{#2}{}{#1}}}
\newrobustcmd{\checkmathfrakuz}[2][]{\ensuremath{\subp{\check{\mathfrak{u}}}{}{#2}{}{#1}}}
\newrobustcmd{\tildemathfrakuz}[2][]{\ensuremath{\subp{\tilde{\mathfrak{u}}}{}{#2}{}{#1}}}
\newrobustcmd{\widetildemathfrakuz}[2][]{\ensuremath{\subp{\widetilde{\mathfrak{u}}}{}{#2}{}{#1}}}
\newrobustcmd{\acutemathfrakuz}[2][]{\ensuremath{\subp{\acute{\mathfrak{u}}}{}{#2}{}{#1}}}
\newrobustcmd{\gravemathfrakuz}[2][]{\ensuremath{\subp{\grave{\mathfrak{u}}}{}{#2}{}{#1}}}
\newrobustcmd{\dotmathfrakuz}[2][]{\ensuremath{\subp{\dot{\mathfrak{u}}}{}{#2}{}{#1}}}
\newrobustcmd{\ddotmathfrakuz}[2][]{\ensuremath{\subp{\ddot{\mathfrak{u}}}{}{#2}{}{#1}}}
\newrobustcmd{\brevemathfrakuz}[2][]{\ensuremath{\subp{\breve{\mathfrak{u}}}{}{#2}{}{#1}}}
\newrobustcmd{\barmathfrakuz}[2][]{\ensuremath{\subp{\bar{\mathfrak{u}}}{}{#2}{}{#1}}}
\newrobustcmd{\vecmathfrakuz}[2][]{\ensuremath{\subp{\vec{\mathfrak{u}}}{}{#2}{}{#1}}}
\newrobustcmd{\bmmathfrakuz}[2][]{\ensuremath{\subp{\bm{\mathfrak{u}}}{}{#2}{}{#1}}}
\newrobustcmd{\hatbmmathfrakuz}[2][]{\ensuremath{\subp{\hat{\bm{\mathfrak{u}}}}{}{#2}{}{#1}}}
\newrobustcmd{\widehatbmmathfrakuz}[2][]{\ensuremath{\subp{\widehat{\bm{\mathfrak{u}}}}{}{#2}{}{#1}}}
\newrobustcmd{\checkbmmathfrakuz}[2][]{\ensuremath{\subp{\check{\bm{\mathfrak{u}}}}{}{#2}{}{#1}}}
\newrobustcmd{\tildebmmathfrakuz}[2][]{\ensuremath{\subp{\tilde{\bm{\mathfrak{u}}}}{}{#2}{}{#1}}}
\newrobustcmd{\widetildebmmathfrakuz}[2][]{\ensuremath{\subp{\widetilde{\bm{\mathfrak{u}}}}{}{#2}{}{#1}}}
\newrobustcmd{\acutebmmathfrakuz}[2][]{\ensuremath{\subp{\acute{\bm{\mathfrak{u}}}}{}{#2}{}{#1}}}
\newrobustcmd{\gravebmmathfrakuz}[2][]{\ensuremath{\subp{\grave{\bm{\mathfrak{u}}}}{}{#2}{}{#1}}}
\newrobustcmd{\dotbmmathfrakuz}[2][]{\ensuremath{\subp{\dot{\bm{\mathfrak{u}}}}{}{#2}{}{#1}}}
\newrobustcmd{\ddotbmmathfrakuz}[2][]{\ensuremath{\subp{\ddot{\bm{\mathfrak{u}}}}{}{#2}{}{#1}}}
\newrobustcmd{\brevebmmathfrakuz}[2][]{\ensuremath{\subp{\breve{\bm{\mathfrak{u}}}}{}{#2}{}{#1}}}
\newrobustcmd{\barbmmathfrakuz}[2][]{\ensuremath{\subp{\bar{\bm{\mathfrak{u}}}}{}{#2}{}{#1}}}
\newrobustcmd{\vecbmmathfrakuz}[2][]{\ensuremath{\subp{\vec{\bm{\mathfrak{u}}}}{}{#2}{}{#1}}}
\newrobustcmd{\mathfrakvz}[2][]{\ensuremath{\subp{\mathfrak{v}}{}{#2}{}{#1}}}
\newrobustcmd{\hatmathfrakvz}[2][]{\ensuremath{\subp{\hat{\mathfrak{v}}}{}{#2}{}{#1}}}
\newrobustcmd{\widehatmathfrakvz}[2][]{\ensuremath{\subp{\widehat{\mathfrak{v}}}{}{#2}{}{#1}}}
\newrobustcmd{\checkmathfrakvz}[2][]{\ensuremath{\subp{\check{\mathfrak{v}}}{}{#2}{}{#1}}}
\newrobustcmd{\tildemathfrakvz}[2][]{\ensuremath{\subp{\tilde{\mathfrak{v}}}{}{#2}{}{#1}}}
\newrobustcmd{\widetildemathfrakvz}[2][]{\ensuremath{\subp{\widetilde{\mathfrak{v}}}{}{#2}{}{#1}}}
\newrobustcmd{\acutemathfrakvz}[2][]{\ensuremath{\subp{\acute{\mathfrak{v}}}{}{#2}{}{#1}}}
\newrobustcmd{\gravemathfrakvz}[2][]{\ensuremath{\subp{\grave{\mathfrak{v}}}{}{#2}{}{#1}}}
\newrobustcmd{\dotmathfrakvz}[2][]{\ensuremath{\subp{\dot{\mathfrak{v}}}{}{#2}{}{#1}}}
\newrobustcmd{\ddotmathfrakvz}[2][]{\ensuremath{\subp{\ddot{\mathfrak{v}}}{}{#2}{}{#1}}}
\newrobustcmd{\brevemathfrakvz}[2][]{\ensuremath{\subp{\breve{\mathfrak{v}}}{}{#2}{}{#1}}}
\newrobustcmd{\barmathfrakvz}[2][]{\ensuremath{\subp{\bar{\mathfrak{v}}}{}{#2}{}{#1}}}
\newrobustcmd{\vecmathfrakvz}[2][]{\ensuremath{\subp{\vec{\mathfrak{v}}}{}{#2}{}{#1}}}
\newrobustcmd{\bmmathfrakvz}[2][]{\ensuremath{\subp{\bm{\mathfrak{v}}}{}{#2}{}{#1}}}
\newrobustcmd{\hatbmmathfrakvz}[2][]{\ensuremath{\subp{\hat{\bm{\mathfrak{v}}}}{}{#2}{}{#1}}}
\newrobustcmd{\widehatbmmathfrakvz}[2][]{\ensuremath{\subp{\widehat{\bm{\mathfrak{v}}}}{}{#2}{}{#1}}}
\newrobustcmd{\checkbmmathfrakvz}[2][]{\ensuremath{\subp{\check{\bm{\mathfrak{v}}}}{}{#2}{}{#1}}}
\newrobustcmd{\tildebmmathfrakvz}[2][]{\ensuremath{\subp{\tilde{\bm{\mathfrak{v}}}}{}{#2}{}{#1}}}
\newrobustcmd{\widetildebmmathfrakvz}[2][]{\ensuremath{\subp{\widetilde{\bm{\mathfrak{v}}}}{}{#2}{}{#1}}}
\newrobustcmd{\acutebmmathfrakvz}[2][]{\ensuremath{\subp{\acute{\bm{\mathfrak{v}}}}{}{#2}{}{#1}}}
\newrobustcmd{\gravebmmathfrakvz}[2][]{\ensuremath{\subp{\grave{\bm{\mathfrak{v}}}}{}{#2}{}{#1}}}
\newrobustcmd{\dotbmmathfrakvz}[2][]{\ensuremath{\subp{\dot{\bm{\mathfrak{v}}}}{}{#2}{}{#1}}}
\newrobustcmd{\ddotbmmathfrakvz}[2][]{\ensuremath{\subp{\ddot{\bm{\mathfrak{v}}}}{}{#2}{}{#1}}}
\newrobustcmd{\brevebmmathfrakvz}[2][]{\ensuremath{\subp{\breve{\bm{\mathfrak{v}}}}{}{#2}{}{#1}}}
\newrobustcmd{\barbmmathfrakvz}[2][]{\ensuremath{\subp{\bar{\bm{\mathfrak{v}}}}{}{#2}{}{#1}}}
\newrobustcmd{\vecbmmathfrakvz}[2][]{\ensuremath{\subp{\vec{\bm{\mathfrak{v}}}}{}{#2}{}{#1}}}
\newrobustcmd{\mathfrakwz}[2][]{\ensuremath{\subp{\mathfrak{w}}{}{#2}{}{#1}}}
\newrobustcmd{\hatmathfrakwz}[2][]{\ensuremath{\subp{\hat{\mathfrak{w}}}{}{#2}{}{#1}}}
\newrobustcmd{\widehatmathfrakwz}[2][]{\ensuremath{\subp{\widehat{\mathfrak{w}}}{}{#2}{}{#1}}}
\newrobustcmd{\checkmathfrakwz}[2][]{\ensuremath{\subp{\check{\mathfrak{w}}}{}{#2}{}{#1}}}
\newrobustcmd{\tildemathfrakwz}[2][]{\ensuremath{\subp{\tilde{\mathfrak{w}}}{}{#2}{}{#1}}}
\newrobustcmd{\widetildemathfrakwz}[2][]{\ensuremath{\subp{\widetilde{\mathfrak{w}}}{}{#2}{}{#1}}}
\newrobustcmd{\acutemathfrakwz}[2][]{\ensuremath{\subp{\acute{\mathfrak{w}}}{}{#2}{}{#1}}}
\newrobustcmd{\gravemathfrakwz}[2][]{\ensuremath{\subp{\grave{\mathfrak{w}}}{}{#2}{}{#1}}}
\newrobustcmd{\dotmathfrakwz}[2][]{\ensuremath{\subp{\dot{\mathfrak{w}}}{}{#2}{}{#1}}}
\newrobustcmd{\ddotmathfrakwz}[2][]{\ensuremath{\subp{\ddot{\mathfrak{w}}}{}{#2}{}{#1}}}
\newrobustcmd{\brevemathfrakwz}[2][]{\ensuremath{\subp{\breve{\mathfrak{w}}}{}{#2}{}{#1}}}
\newrobustcmd{\barmathfrakwz}[2][]{\ensuremath{\subp{\bar{\mathfrak{w}}}{}{#2}{}{#1}}}
\newrobustcmd{\vecmathfrakwz}[2][]{\ensuremath{\subp{\vec{\mathfrak{w}}}{}{#2}{}{#1}}}
\newrobustcmd{\bmmathfrakwz}[2][]{\ensuremath{\subp{\bm{\mathfrak{w}}}{}{#2}{}{#1}}}
\newrobustcmd{\hatbmmathfrakwz}[2][]{\ensuremath{\subp{\hat{\bm{\mathfrak{w}}}}{}{#2}{}{#1}}}
\newrobustcmd{\widehatbmmathfrakwz}[2][]{\ensuremath{\subp{\widehat{\bm{\mathfrak{w}}}}{}{#2}{}{#1}}}
\newrobustcmd{\checkbmmathfrakwz}[2][]{\ensuremath{\subp{\check{\bm{\mathfrak{w}}}}{}{#2}{}{#1}}}
\newrobustcmd{\tildebmmathfrakwz}[2][]{\ensuremath{\subp{\tilde{\bm{\mathfrak{w}}}}{}{#2}{}{#1}}}
\newrobustcmd{\widetildebmmathfrakwz}[2][]{\ensuremath{\subp{\widetilde{\bm{\mathfrak{w}}}}{}{#2}{}{#1}}}
\newrobustcmd{\acutebmmathfrakwz}[2][]{\ensuremath{\subp{\acute{\bm{\mathfrak{w}}}}{}{#2}{}{#1}}}
\newrobustcmd{\gravebmmathfrakwz}[2][]{\ensuremath{\subp{\grave{\bm{\mathfrak{w}}}}{}{#2}{}{#1}}}
\newrobustcmd{\dotbmmathfrakwz}[2][]{\ensuremath{\subp{\dot{\bm{\mathfrak{w}}}}{}{#2}{}{#1}}}
\newrobustcmd{\ddotbmmathfrakwz}[2][]{\ensuremath{\subp{\ddot{\bm{\mathfrak{w}}}}{}{#2}{}{#1}}}
\newrobustcmd{\brevebmmathfrakwz}[2][]{\ensuremath{\subp{\breve{\bm{\mathfrak{w}}}}{}{#2}{}{#1}}}
\newrobustcmd{\barbmmathfrakwz}[2][]{\ensuremath{\subp{\bar{\bm{\mathfrak{w}}}}{}{#2}{}{#1}}}
\newrobustcmd{\vecbmmathfrakwz}[2][]{\ensuremath{\subp{\vec{\bm{\mathfrak{w}}}}{}{#2}{}{#1}}}
\newrobustcmd{\mathfrakxz}[2][]{\ensuremath{\subp{\mathfrak{x}}{}{#2}{}{#1}}}
\newrobustcmd{\hatmathfrakxz}[2][]{\ensuremath{\subp{\hat{\mathfrak{x}}}{}{#2}{}{#1}}}
\newrobustcmd{\widehatmathfrakxz}[2][]{\ensuremath{\subp{\widehat{\mathfrak{x}}}{}{#2}{}{#1}}}
\newrobustcmd{\checkmathfrakxz}[2][]{\ensuremath{\subp{\check{\mathfrak{x}}}{}{#2}{}{#1}}}
\newrobustcmd{\tildemathfrakxz}[2][]{\ensuremath{\subp{\tilde{\mathfrak{x}}}{}{#2}{}{#1}}}
\newrobustcmd{\widetildemathfrakxz}[2][]{\ensuremath{\subp{\widetilde{\mathfrak{x}}}{}{#2}{}{#1}}}
\newrobustcmd{\acutemathfrakxz}[2][]{\ensuremath{\subp{\acute{\mathfrak{x}}}{}{#2}{}{#1}}}
\newrobustcmd{\gravemathfrakxz}[2][]{\ensuremath{\subp{\grave{\mathfrak{x}}}{}{#2}{}{#1}}}
\newrobustcmd{\dotmathfrakxz}[2][]{\ensuremath{\subp{\dot{\mathfrak{x}}}{}{#2}{}{#1}}}
\newrobustcmd{\ddotmathfrakxz}[2][]{\ensuremath{\subp{\ddot{\mathfrak{x}}}{}{#2}{}{#1}}}
\newrobustcmd{\brevemathfrakxz}[2][]{\ensuremath{\subp{\breve{\mathfrak{x}}}{}{#2}{}{#1}}}
\newrobustcmd{\barmathfrakxz}[2][]{\ensuremath{\subp{\bar{\mathfrak{x}}}{}{#2}{}{#1}}}
\newrobustcmd{\vecmathfrakxz}[2][]{\ensuremath{\subp{\vec{\mathfrak{x}}}{}{#2}{}{#1}}}
\newrobustcmd{\bmmathfrakxz}[2][]{\ensuremath{\subp{\bm{\mathfrak{x}}}{}{#2}{}{#1}}}
\newrobustcmd{\hatbmmathfrakxz}[2][]{\ensuremath{\subp{\hat{\bm{\mathfrak{x}}}}{}{#2}{}{#1}}}
\newrobustcmd{\widehatbmmathfrakxz}[2][]{\ensuremath{\subp{\widehat{\bm{\mathfrak{x}}}}{}{#2}{}{#1}}}
\newrobustcmd{\checkbmmathfrakxz}[2][]{\ensuremath{\subp{\check{\bm{\mathfrak{x}}}}{}{#2}{}{#1}}}
\newrobustcmd{\tildebmmathfrakxz}[2][]{\ensuremath{\subp{\tilde{\bm{\mathfrak{x}}}}{}{#2}{}{#1}}}
\newrobustcmd{\widetildebmmathfrakxz}[2][]{\ensuremath{\subp{\widetilde{\bm{\mathfrak{x}}}}{}{#2}{}{#1}}}
\newrobustcmd{\acutebmmathfrakxz}[2][]{\ensuremath{\subp{\acute{\bm{\mathfrak{x}}}}{}{#2}{}{#1}}}
\newrobustcmd{\gravebmmathfrakxz}[2][]{\ensuremath{\subp{\grave{\bm{\mathfrak{x}}}}{}{#2}{}{#1}}}
\newrobustcmd{\dotbmmathfrakxz}[2][]{\ensuremath{\subp{\dot{\bm{\mathfrak{x}}}}{}{#2}{}{#1}}}
\newrobustcmd{\ddotbmmathfrakxz}[2][]{\ensuremath{\subp{\ddot{\bm{\mathfrak{x}}}}{}{#2}{}{#1}}}
\newrobustcmd{\brevebmmathfrakxz}[2][]{\ensuremath{\subp{\breve{\bm{\mathfrak{x}}}}{}{#2}{}{#1}}}
\newrobustcmd{\barbmmathfrakxz}[2][]{\ensuremath{\subp{\bar{\bm{\mathfrak{x}}}}{}{#2}{}{#1}}}
\newrobustcmd{\vecbmmathfrakxz}[2][]{\ensuremath{\subp{\vec{\bm{\mathfrak{x}}}}{}{#2}{}{#1}}}
\newrobustcmd{\mathfrakyz}[2][]{\ensuremath{\subp{\mathfrak{y}}{}{#2}{}{#1}}}
\newrobustcmd{\hatmathfrakyz}[2][]{\ensuremath{\subp{\hat{\mathfrak{y}}}{}{#2}{}{#1}}}
\newrobustcmd{\widehatmathfrakyz}[2][]{\ensuremath{\subp{\widehat{\mathfrak{y}}}{}{#2}{}{#1}}}
\newrobustcmd{\checkmathfrakyz}[2][]{\ensuremath{\subp{\check{\mathfrak{y}}}{}{#2}{}{#1}}}
\newrobustcmd{\tildemathfrakyz}[2][]{\ensuremath{\subp{\tilde{\mathfrak{y}}}{}{#2}{}{#1}}}
\newrobustcmd{\widetildemathfrakyz}[2][]{\ensuremath{\subp{\widetilde{\mathfrak{y}}}{}{#2}{}{#1}}}
\newrobustcmd{\acutemathfrakyz}[2][]{\ensuremath{\subp{\acute{\mathfrak{y}}}{}{#2}{}{#1}}}
\newrobustcmd{\gravemathfrakyz}[2][]{\ensuremath{\subp{\grave{\mathfrak{y}}}{}{#2}{}{#1}}}
\newrobustcmd{\dotmathfrakyz}[2][]{\ensuremath{\subp{\dot{\mathfrak{y}}}{}{#2}{}{#1}}}
\newrobustcmd{\ddotmathfrakyz}[2][]{\ensuremath{\subp{\ddot{\mathfrak{y}}}{}{#2}{}{#1}}}
\newrobustcmd{\brevemathfrakyz}[2][]{\ensuremath{\subp{\breve{\mathfrak{y}}}{}{#2}{}{#1}}}
\newrobustcmd{\barmathfrakyz}[2][]{\ensuremath{\subp{\bar{\mathfrak{y}}}{}{#2}{}{#1}}}
\newrobustcmd{\vecmathfrakyz}[2][]{\ensuremath{\subp{\vec{\mathfrak{y}}}{}{#2}{}{#1}}}
\newrobustcmd{\bmmathfrakyz}[2][]{\ensuremath{\subp{\bm{\mathfrak{y}}}{}{#2}{}{#1}}}
\newrobustcmd{\hatbmmathfrakyz}[2][]{\ensuremath{\subp{\hat{\bm{\mathfrak{y}}}}{}{#2}{}{#1}}}
\newrobustcmd{\widehatbmmathfrakyz}[2][]{\ensuremath{\subp{\widehat{\bm{\mathfrak{y}}}}{}{#2}{}{#1}}}
\newrobustcmd{\checkbmmathfrakyz}[2][]{\ensuremath{\subp{\check{\bm{\mathfrak{y}}}}{}{#2}{}{#1}}}
\newrobustcmd{\tildebmmathfrakyz}[2][]{\ensuremath{\subp{\tilde{\bm{\mathfrak{y}}}}{}{#2}{}{#1}}}
\newrobustcmd{\widetildebmmathfrakyz}[2][]{\ensuremath{\subp{\widetilde{\bm{\mathfrak{y}}}}{}{#2}{}{#1}}}
\newrobustcmd{\acutebmmathfrakyz}[2][]{\ensuremath{\subp{\acute{\bm{\mathfrak{y}}}}{}{#2}{}{#1}}}
\newrobustcmd{\gravebmmathfrakyz}[2][]{\ensuremath{\subp{\grave{\bm{\mathfrak{y}}}}{}{#2}{}{#1}}}
\newrobustcmd{\dotbmmathfrakyz}[2][]{\ensuremath{\subp{\dot{\bm{\mathfrak{y}}}}{}{#2}{}{#1}}}
\newrobustcmd{\ddotbmmathfrakyz}[2][]{\ensuremath{\subp{\ddot{\bm{\mathfrak{y}}}}{}{#2}{}{#1}}}
\newrobustcmd{\brevebmmathfrakyz}[2][]{\ensuremath{\subp{\breve{\bm{\mathfrak{y}}}}{}{#2}{}{#1}}}
\newrobustcmd{\barbmmathfrakyz}[2][]{\ensuremath{\subp{\bar{\bm{\mathfrak{y}}}}{}{#2}{}{#1}}}
\newrobustcmd{\vecbmmathfrakyz}[2][]{\ensuremath{\subp{\vec{\bm{\mathfrak{y}}}}{}{#2}{}{#1}}}
\newrobustcmd{\mathfrakzz}[2][]{\ensuremath{\subp{\mathfrak{z}}{}{#2}{}{#1}}}
\newrobustcmd{\hatmathfrakzz}[2][]{\ensuremath{\subp{\hat{\mathfrak{z}}}{}{#2}{}{#1}}}
\newrobustcmd{\widehatmathfrakzz}[2][]{\ensuremath{\subp{\widehat{\mathfrak{z}}}{}{#2}{}{#1}}}
\newrobustcmd{\checkmathfrakzz}[2][]{\ensuremath{\subp{\check{\mathfrak{z}}}{}{#2}{}{#1}}}
\newrobustcmd{\tildemathfrakzz}[2][]{\ensuremath{\subp{\tilde{\mathfrak{z}}}{}{#2}{}{#1}}}
\newrobustcmd{\widetildemathfrakzz}[2][]{\ensuremath{\subp{\widetilde{\mathfrak{z}}}{}{#2}{}{#1}}}
\newrobustcmd{\acutemathfrakzz}[2][]{\ensuremath{\subp{\acute{\mathfrak{z}}}{}{#2}{}{#1}}}
\newrobustcmd{\gravemathfrakzz}[2][]{\ensuremath{\subp{\grave{\mathfrak{z}}}{}{#2}{}{#1}}}
\newrobustcmd{\dotmathfrakzz}[2][]{\ensuremath{\subp{\dot{\mathfrak{z}}}{}{#2}{}{#1}}}
\newrobustcmd{\ddotmathfrakzz}[2][]{\ensuremath{\subp{\ddot{\mathfrak{z}}}{}{#2}{}{#1}}}
\newrobustcmd{\brevemathfrakzz}[2][]{\ensuremath{\subp{\breve{\mathfrak{z}}}{}{#2}{}{#1}}}
\newrobustcmd{\barmathfrakzz}[2][]{\ensuremath{\subp{\bar{\mathfrak{z}}}{}{#2}{}{#1}}}
\newrobustcmd{\vecmathfrakzz}[2][]{\ensuremath{\subp{\vec{\mathfrak{z}}}{}{#2}{}{#1}}}
\newrobustcmd{\bmmathfrakzz}[2][]{\ensuremath{\subp{\bm{\mathfrak{z}}}{}{#2}{}{#1}}}
\newrobustcmd{\hatbmmathfrakzz}[2][]{\ensuremath{\subp{\hat{\bm{\mathfrak{z}}}}{}{#2}{}{#1}}}
\newrobustcmd{\widehatbmmathfrakzz}[2][]{\ensuremath{\subp{\widehat{\bm{\mathfrak{z}}}}{}{#2}{}{#1}}}
\newrobustcmd{\checkbmmathfrakzz}[2][]{\ensuremath{\subp{\check{\bm{\mathfrak{z}}}}{}{#2}{}{#1}}}
\newrobustcmd{\tildebmmathfrakzz}[2][]{\ensuremath{\subp{\tilde{\bm{\mathfrak{z}}}}{}{#2}{}{#1}}}
\newrobustcmd{\widetildebmmathfrakzz}[2][]{\ensuremath{\subp{\widetilde{\bm{\mathfrak{z}}}}{}{#2}{}{#1}}}
\newrobustcmd{\acutebmmathfrakzz}[2][]{\ensuremath{\subp{\acute{\bm{\mathfrak{z}}}}{}{#2}{}{#1}}}
\newrobustcmd{\gravebmmathfrakzz}[2][]{\ensuremath{\subp{\grave{\bm{\mathfrak{z}}}}{}{#2}{}{#1}}}
\newrobustcmd{\dotbmmathfrakzz}[2][]{\ensuremath{\subp{\dot{\bm{\mathfrak{z}}}}{}{#2}{}{#1}}}
\newrobustcmd{\ddotbmmathfrakzz}[2][]{\ensuremath{\subp{\ddot{\bm{\mathfrak{z}}}}{}{#2}{}{#1}}}
\newrobustcmd{\brevebmmathfrakzz}[2][]{\ensuremath{\subp{\breve{\bm{\mathfrak{z}}}}{}{#2}{}{#1}}}
\newrobustcmd{\barbmmathfrakzz}[2][]{\ensuremath{\subp{\bar{\bm{\mathfrak{z}}}}{}{#2}{}{#1}}}
\newrobustcmd{\vecbmmathfrakzz}[2][]{\ensuremath{\subp{\vec{\bm{\mathfrak{z}}}}{}{#2}{}{#1}}}

\newrobustcmd{\mathfrakAz}[2][]{\ensuremath{\subp{\mathfrak{A}}{}{#2}{}{#1}}}
\newrobustcmd{\hatmathfrakAz}[2][]{\ensuremath{\subp{\hat{\mathfrak{A}}}{}{#2}{}{#1}}}
\newrobustcmd{\widehatmathfrakAz}[2][]{\ensuremath{\subp{\widehat{\mathfrak{A}}}{}{#2}{}{#1}}}
\newrobustcmd{\checkmathfrakAz}[2][]{\ensuremath{\subp{\check{\mathfrak{A}}}{}{#2}{}{#1}}}
\newrobustcmd{\tildemathfrakAz}[2][]{\ensuremath{\subp{\tilde{\mathfrak{A}}}{}{#2}{}{#1}}}
\newrobustcmd{\widetildemathfrakAz}[2][]{\ensuremath{\subp{\widetilde{\mathfrak{A}}}{}{#2}{}{#1}}}
\newrobustcmd{\acutemathfrakAz}[2][]{\ensuremath{\subp{\acute{\mathfrak{A}}}{}{#2}{}{#1}}}
\newrobustcmd{\gravemathfrakAz}[2][]{\ensuremath{\subp{\grave{\mathfrak{A}}}{}{#2}{}{#1}}}
\newrobustcmd{\dotmathfrakAz}[2][]{\ensuremath{\subp{\dot{\mathfrak{A}}}{}{#2}{}{#1}}}
\newrobustcmd{\ddotmathfrakAz}[2][]{\ensuremath{\subp{\ddot{\mathfrak{A}}}{}{#2}{}{#1}}}
\newrobustcmd{\brevemathfrakAz}[2][]{\ensuremath{\subp{\breve{\mathfrak{A}}}{}{#2}{}{#1}}}
\newrobustcmd{\barmathfrakAz}[2][]{\ensuremath{\subp{\bar{\mathfrak{A}}}{}{#2}{}{#1}}}
\newrobustcmd{\vecmathfrakAz}[2][]{\ensuremath{\subp{\vec{\mathfrak{A}}}{}{#2}{}{#1}}}
\newrobustcmd{\bmmathfrakAz}[2][]{\ensuremath{\subp{\bm{\mathfrak{A}}}{}{#2}{}{#1}}}
\newrobustcmd{\hatbmmathfrakAz}[2][]{\ensuremath{\subp{\hat{\bm{\mathfrak{A}}}}{}{#2}{}{#1}}}
\newrobustcmd{\widehatbmmathfrakAz}[2][]{\ensuremath{\subp{\widehat{\bm{\mathfrak{A}}}}{}{#2}{}{#1}}}
\newrobustcmd{\checkbmmathfrakAz}[2][]{\ensuremath{\subp{\check{\bm{\mathfrak{A}}}}{}{#2}{}{#1}}}
\newrobustcmd{\tildebmmathfrakAz}[2][]{\ensuremath{\subp{\tilde{\bm{\mathfrak{A}}}}{}{#2}{}{#1}}}
\newrobustcmd{\widetildebmmathfrakAz}[2][]{\ensuremath{\subp{\widetilde{\bm{\mathfrak{A}}}}{}{#2}{}{#1}}}
\newrobustcmd{\acutebmmathfrakAz}[2][]{\ensuremath{\subp{\acute{\bm{\mathfrak{A}}}}{}{#2}{}{#1}}}
\newrobustcmd{\gravebmmathfrakAz}[2][]{\ensuremath{\subp{\grave{\bm{\mathfrak{A}}}}{}{#2}{}{#1}}}
\newrobustcmd{\dotbmmathfrakAz}[2][]{\ensuremath{\subp{\dot{\bm{\mathfrak{A}}}}{}{#2}{}{#1}}}
\newrobustcmd{\ddotbmmathfrakAz}[2][]{\ensuremath{\subp{\ddot{\bm{\mathfrak{A}}}}{}{#2}{}{#1}}}
\newrobustcmd{\brevebmmathfrakAz}[2][]{\ensuremath{\subp{\breve{\bm{\mathfrak{A}}}}{}{#2}{}{#1}}}
\newrobustcmd{\barbmmathfrakAz}[2][]{\ensuremath{\subp{\bar{\bm{\mathfrak{A}}}}{}{#2}{}{#1}}}
\newrobustcmd{\vecbmmathfrakAz}[2][]{\ensuremath{\subp{\vec{\bm{\mathfrak{A}}}}{}{#2}{}{#1}}}
\newrobustcmd{\mathfrakBz}[2][]{\ensuremath{\subp{\mathfrak{B}}{}{#2}{}{#1}}}
\newrobustcmd{\hatmathfrakBz}[2][]{\ensuremath{\subp{\hat{\mathfrak{B}}}{}{#2}{}{#1}}}
\newrobustcmd{\widehatmathfrakBz}[2][]{\ensuremath{\subp{\widehat{\mathfrak{B}}}{}{#2}{}{#1}}}
\newrobustcmd{\checkmathfrakBz}[2][]{\ensuremath{\subp{\check{\mathfrak{B}}}{}{#2}{}{#1}}}
\newrobustcmd{\tildemathfrakBz}[2][]{\ensuremath{\subp{\tilde{\mathfrak{B}}}{}{#2}{}{#1}}}
\newrobustcmd{\widetildemathfrakBz}[2][]{\ensuremath{\subp{\widetilde{\mathfrak{B}}}{}{#2}{}{#1}}}
\newrobustcmd{\acutemathfrakBz}[2][]{\ensuremath{\subp{\acute{\mathfrak{B}}}{}{#2}{}{#1}}}
\newrobustcmd{\gravemathfrakBz}[2][]{\ensuremath{\subp{\grave{\mathfrak{B}}}{}{#2}{}{#1}}}
\newrobustcmd{\dotmathfrakBz}[2][]{\ensuremath{\subp{\dot{\mathfrak{B}}}{}{#2}{}{#1}}}
\newrobustcmd{\ddotmathfrakBz}[2][]{\ensuremath{\subp{\ddot{\mathfrak{B}}}{}{#2}{}{#1}}}
\newrobustcmd{\brevemathfrakBz}[2][]{\ensuremath{\subp{\breve{\mathfrak{B}}}{}{#2}{}{#1}}}
\newrobustcmd{\barmathfrakBz}[2][]{\ensuremath{\subp{\bar{\mathfrak{B}}}{}{#2}{}{#1}}}
\newrobustcmd{\vecmathfrakBz}[2][]{\ensuremath{\subp{\vec{\mathfrak{B}}}{}{#2}{}{#1}}}
\newrobustcmd{\bmmathfrakBz}[2][]{\ensuremath{\subp{\bm{\mathfrak{B}}}{}{#2}{}{#1}}}
\newrobustcmd{\hatbmmathfrakBz}[2][]{\ensuremath{\subp{\hat{\bm{\mathfrak{B}}}}{}{#2}{}{#1}}}
\newrobustcmd{\widehatbmmathfrakBz}[2][]{\ensuremath{\subp{\widehat{\bm{\mathfrak{B}}}}{}{#2}{}{#1}}}
\newrobustcmd{\checkbmmathfrakBz}[2][]{\ensuremath{\subp{\check{\bm{\mathfrak{B}}}}{}{#2}{}{#1}}}
\newrobustcmd{\tildebmmathfrakBz}[2][]{\ensuremath{\subp{\tilde{\bm{\mathfrak{B}}}}{}{#2}{}{#1}}}
\newrobustcmd{\widetildebmmathfrakBz}[2][]{\ensuremath{\subp{\widetilde{\bm{\mathfrak{B}}}}{}{#2}{}{#1}}}
\newrobustcmd{\acutebmmathfrakBz}[2][]{\ensuremath{\subp{\acute{\bm{\mathfrak{B}}}}{}{#2}{}{#1}}}
\newrobustcmd{\gravebmmathfrakBz}[2][]{\ensuremath{\subp{\grave{\bm{\mathfrak{B}}}}{}{#2}{}{#1}}}
\newrobustcmd{\dotbmmathfrakBz}[2][]{\ensuremath{\subp{\dot{\bm{\mathfrak{B}}}}{}{#2}{}{#1}}}
\newrobustcmd{\ddotbmmathfrakBz}[2][]{\ensuremath{\subp{\ddot{\bm{\mathfrak{B}}}}{}{#2}{}{#1}}}
\newrobustcmd{\brevebmmathfrakBz}[2][]{\ensuremath{\subp{\breve{\bm{\mathfrak{B}}}}{}{#2}{}{#1}}}
\newrobustcmd{\barbmmathfrakBz}[2][]{\ensuremath{\subp{\bar{\bm{\mathfrak{B}}}}{}{#2}{}{#1}}}
\newrobustcmd{\vecbmmathfrakBz}[2][]{\ensuremath{\subp{\vec{\bm{\mathfrak{B}}}}{}{#2}{}{#1}}}
\newrobustcmd{\mathfrakCz}[2][]{\ensuremath{\subp{\mathfrak{C}}{}{#2}{}{#1}}}
\newrobustcmd{\hatmathfrakCz}[2][]{\ensuremath{\subp{\hat{\mathfrak{C}}}{}{#2}{}{#1}}}
\newrobustcmd{\widehatmathfrakCz}[2][]{\ensuremath{\subp{\widehat{\mathfrak{C}}}{}{#2}{}{#1}}}
\newrobustcmd{\checkmathfrakCz}[2][]{\ensuremath{\subp{\check{\mathfrak{C}}}{}{#2}{}{#1}}}
\newrobustcmd{\tildemathfrakCz}[2][]{\ensuremath{\subp{\tilde{\mathfrak{C}}}{}{#2}{}{#1}}}
\newrobustcmd{\widetildemathfrakCz}[2][]{\ensuremath{\subp{\widetilde{\mathfrak{C}}}{}{#2}{}{#1}}}
\newrobustcmd{\acutemathfrakCz}[2][]{\ensuremath{\subp{\acute{\mathfrak{C}}}{}{#2}{}{#1}}}
\newrobustcmd{\gravemathfrakCz}[2][]{\ensuremath{\subp{\grave{\mathfrak{C}}}{}{#2}{}{#1}}}
\newrobustcmd{\dotmathfrakCz}[2][]{\ensuremath{\subp{\dot{\mathfrak{C}}}{}{#2}{}{#1}}}
\newrobustcmd{\ddotmathfrakCz}[2][]{\ensuremath{\subp{\ddot{\mathfrak{C}}}{}{#2}{}{#1}}}
\newrobustcmd{\brevemathfrakCz}[2][]{\ensuremath{\subp{\breve{\mathfrak{C}}}{}{#2}{}{#1}}}
\newrobustcmd{\barmathfrakCz}[2][]{\ensuremath{\subp{\bar{\mathfrak{C}}}{}{#2}{}{#1}}}
\newrobustcmd{\vecmathfrakCz}[2][]{\ensuremath{\subp{\vec{\mathfrak{C}}}{}{#2}{}{#1}}}
\newrobustcmd{\bmmathfrakCz}[2][]{\ensuremath{\subp{\bm{\mathfrak{C}}}{}{#2}{}{#1}}}
\newrobustcmd{\hatbmmathfrakCz}[2][]{\ensuremath{\subp{\hat{\bm{\mathfrak{C}}}}{}{#2}{}{#1}}}
\newrobustcmd{\widehatbmmathfrakCz}[2][]{\ensuremath{\subp{\widehat{\bm{\mathfrak{C}}}}{}{#2}{}{#1}}}
\newrobustcmd{\checkbmmathfrakCz}[2][]{\ensuremath{\subp{\check{\bm{\mathfrak{C}}}}{}{#2}{}{#1}}}
\newrobustcmd{\tildebmmathfrakCz}[2][]{\ensuremath{\subp{\tilde{\bm{\mathfrak{C}}}}{}{#2}{}{#1}}}
\newrobustcmd{\widetildebmmathfrakCz}[2][]{\ensuremath{\subp{\widetilde{\bm{\mathfrak{C}}}}{}{#2}{}{#1}}}
\newrobustcmd{\acutebmmathfrakCz}[2][]{\ensuremath{\subp{\acute{\bm{\mathfrak{C}}}}{}{#2}{}{#1}}}
\newrobustcmd{\gravebmmathfrakCz}[2][]{\ensuremath{\subp{\grave{\bm{\mathfrak{C}}}}{}{#2}{}{#1}}}
\newrobustcmd{\dotbmmathfrakCz}[2][]{\ensuremath{\subp{\dot{\bm{\mathfrak{C}}}}{}{#2}{}{#1}}}
\newrobustcmd{\ddotbmmathfrakCz}[2][]{\ensuremath{\subp{\ddot{\bm{\mathfrak{C}}}}{}{#2}{}{#1}}}
\newrobustcmd{\brevebmmathfrakCz}[2][]{\ensuremath{\subp{\breve{\bm{\mathfrak{C}}}}{}{#2}{}{#1}}}
\newrobustcmd{\barbmmathfrakCz}[2][]{\ensuremath{\subp{\bar{\bm{\mathfrak{C}}}}{}{#2}{}{#1}}}
\newrobustcmd{\vecbmmathfrakCz}[2][]{\ensuremath{\subp{\vec{\bm{\mathfrak{C}}}}{}{#2}{}{#1}}}
\newrobustcmd{\mathfrakDz}[2][]{\ensuremath{\subp{\mathfrak{D}}{}{#2}{}{#1}}}
\newrobustcmd{\hatmathfrakDz}[2][]{\ensuremath{\subp{\hat{\mathfrak{D}}}{}{#2}{}{#1}}}
\newrobustcmd{\widehatmathfrakDz}[2][]{\ensuremath{\subp{\widehat{\mathfrak{D}}}{}{#2}{}{#1}}}
\newrobustcmd{\checkmathfrakDz}[2][]{\ensuremath{\subp{\check{\mathfrak{D}}}{}{#2}{}{#1}}}
\newrobustcmd{\tildemathfrakDz}[2][]{\ensuremath{\subp{\tilde{\mathfrak{D}}}{}{#2}{}{#1}}}
\newrobustcmd{\widetildemathfrakDz}[2][]{\ensuremath{\subp{\widetilde{\mathfrak{D}}}{}{#2}{}{#1}}}
\newrobustcmd{\acutemathfrakDz}[2][]{\ensuremath{\subp{\acute{\mathfrak{D}}}{}{#2}{}{#1}}}
\newrobustcmd{\gravemathfrakDz}[2][]{\ensuremath{\subp{\grave{\mathfrak{D}}}{}{#2}{}{#1}}}
\newrobustcmd{\dotmathfrakDz}[2][]{\ensuremath{\subp{\dot{\mathfrak{D}}}{}{#2}{}{#1}}}
\newrobustcmd{\ddotmathfrakDz}[2][]{\ensuremath{\subp{\ddot{\mathfrak{D}}}{}{#2}{}{#1}}}
\newrobustcmd{\brevemathfrakDz}[2][]{\ensuremath{\subp{\breve{\mathfrak{D}}}{}{#2}{}{#1}}}
\newrobustcmd{\barmathfrakDz}[2][]{\ensuremath{\subp{\bar{\mathfrak{D}}}{}{#2}{}{#1}}}
\newrobustcmd{\vecmathfrakDz}[2][]{\ensuremath{\subp{\vec{\mathfrak{D}}}{}{#2}{}{#1}}}
\newrobustcmd{\bmmathfrakDz}[2][]{\ensuremath{\subp{\bm{\mathfrak{D}}}{}{#2}{}{#1}}}
\newrobustcmd{\hatbmmathfrakDz}[2][]{\ensuremath{\subp{\hat{\bm{\mathfrak{D}}}}{}{#2}{}{#1}}}
\newrobustcmd{\widehatbmmathfrakDz}[2][]{\ensuremath{\subp{\widehat{\bm{\mathfrak{D}}}}{}{#2}{}{#1}}}
\newrobustcmd{\checkbmmathfrakDz}[2][]{\ensuremath{\subp{\check{\bm{\mathfrak{D}}}}{}{#2}{}{#1}}}
\newrobustcmd{\tildebmmathfrakDz}[2][]{\ensuremath{\subp{\tilde{\bm{\mathfrak{D}}}}{}{#2}{}{#1}}}
\newrobustcmd{\widetildebmmathfrakDz}[2][]{\ensuremath{\subp{\widetilde{\bm{\mathfrak{D}}}}{}{#2}{}{#1}}}
\newrobustcmd{\acutebmmathfrakDz}[2][]{\ensuremath{\subp{\acute{\bm{\mathfrak{D}}}}{}{#2}{}{#1}}}
\newrobustcmd{\gravebmmathfrakDz}[2][]{\ensuremath{\subp{\grave{\bm{\mathfrak{D}}}}{}{#2}{}{#1}}}
\newrobustcmd{\dotbmmathfrakDz}[2][]{\ensuremath{\subp{\dot{\bm{\mathfrak{D}}}}{}{#2}{}{#1}}}
\newrobustcmd{\ddotbmmathfrakDz}[2][]{\ensuremath{\subp{\ddot{\bm{\mathfrak{D}}}}{}{#2}{}{#1}}}
\newrobustcmd{\brevebmmathfrakDz}[2][]{\ensuremath{\subp{\breve{\bm{\mathfrak{D}}}}{}{#2}{}{#1}}}
\newrobustcmd{\barbmmathfrakDz}[2][]{\ensuremath{\subp{\bar{\bm{\mathfrak{D}}}}{}{#2}{}{#1}}}
\newrobustcmd{\vecbmmathfrakDz}[2][]{\ensuremath{\subp{\vec{\bm{\mathfrak{D}}}}{}{#2}{}{#1}}}
\newrobustcmd{\mathfrakEz}[2][]{\ensuremath{\subp{\mathfrak{E}}{}{#2}{}{#1}}}
\newrobustcmd{\hatmathfrakEz}[2][]{\ensuremath{\subp{\hat{\mathfrak{E}}}{}{#2}{}{#1}}}
\newrobustcmd{\widehatmathfrakEz}[2][]{\ensuremath{\subp{\widehat{\mathfrak{E}}}{}{#2}{}{#1}}}
\newrobustcmd{\checkmathfrakEz}[2][]{\ensuremath{\subp{\check{\mathfrak{E}}}{}{#2}{}{#1}}}
\newrobustcmd{\tildemathfrakEz}[2][]{\ensuremath{\subp{\tilde{\mathfrak{E}}}{}{#2}{}{#1}}}
\newrobustcmd{\widetildemathfrakEz}[2][]{\ensuremath{\subp{\widetilde{\mathfrak{E}}}{}{#2}{}{#1}}}
\newrobustcmd{\acutemathfrakEz}[2][]{\ensuremath{\subp{\acute{\mathfrak{E}}}{}{#2}{}{#1}}}
\newrobustcmd{\gravemathfrakEz}[2][]{\ensuremath{\subp{\grave{\mathfrak{E}}}{}{#2}{}{#1}}}
\newrobustcmd{\dotmathfrakEz}[2][]{\ensuremath{\subp{\dot{\mathfrak{E}}}{}{#2}{}{#1}}}
\newrobustcmd{\ddotmathfrakEz}[2][]{\ensuremath{\subp{\ddot{\mathfrak{E}}}{}{#2}{}{#1}}}
\newrobustcmd{\brevemathfrakEz}[2][]{\ensuremath{\subp{\breve{\mathfrak{E}}}{}{#2}{}{#1}}}
\newrobustcmd{\barmathfrakEz}[2][]{\ensuremath{\subp{\bar{\mathfrak{E}}}{}{#2}{}{#1}}}
\newrobustcmd{\vecmathfrakEz}[2][]{\ensuremath{\subp{\vec{\mathfrak{E}}}{}{#2}{}{#1}}}
\newrobustcmd{\bmmathfrakEz}[2][]{\ensuremath{\subp{\bm{\mathfrak{E}}}{}{#2}{}{#1}}}
\newrobustcmd{\hatbmmathfrakEz}[2][]{\ensuremath{\subp{\hat{\bm{\mathfrak{E}}}}{}{#2}{}{#1}}}
\newrobustcmd{\widehatbmmathfrakEz}[2][]{\ensuremath{\subp{\widehat{\bm{\mathfrak{E}}}}{}{#2}{}{#1}}}
\newrobustcmd{\checkbmmathfrakEz}[2][]{\ensuremath{\subp{\check{\bm{\mathfrak{E}}}}{}{#2}{}{#1}}}
\newrobustcmd{\tildebmmathfrakEz}[2][]{\ensuremath{\subp{\tilde{\bm{\mathfrak{E}}}}{}{#2}{}{#1}}}
\newrobustcmd{\widetildebmmathfrakEz}[2][]{\ensuremath{\subp{\widetilde{\bm{\mathfrak{E}}}}{}{#2}{}{#1}}}
\newrobustcmd{\acutebmmathfrakEz}[2][]{\ensuremath{\subp{\acute{\bm{\mathfrak{E}}}}{}{#2}{}{#1}}}
\newrobustcmd{\gravebmmathfrakEz}[2][]{\ensuremath{\subp{\grave{\bm{\mathfrak{E}}}}{}{#2}{}{#1}}}
\newrobustcmd{\dotbmmathfrakEz}[2][]{\ensuremath{\subp{\dot{\bm{\mathfrak{E}}}}{}{#2}{}{#1}}}
\newrobustcmd{\ddotbmmathfrakEz}[2][]{\ensuremath{\subp{\ddot{\bm{\mathfrak{E}}}}{}{#2}{}{#1}}}
\newrobustcmd{\brevebmmathfrakEz}[2][]{\ensuremath{\subp{\breve{\bm{\mathfrak{E}}}}{}{#2}{}{#1}}}
\newrobustcmd{\barbmmathfrakEz}[2][]{\ensuremath{\subp{\bar{\bm{\mathfrak{E}}}}{}{#2}{}{#1}}}
\newrobustcmd{\vecbmmathfrakEz}[2][]{\ensuremath{\subp{\vec{\bm{\mathfrak{E}}}}{}{#2}{}{#1}}}
\newrobustcmd{\mathfrakFz}[2][]{\ensuremath{\subp{\mathfrak{F}}{}{#2}{}{#1}}}
\newrobustcmd{\hatmathfrakFz}[2][]{\ensuremath{\subp{\hat{\mathfrak{F}}}{}{#2}{}{#1}}}
\newrobustcmd{\widehatmathfrakFz}[2][]{\ensuremath{\subp{\widehat{\mathfrak{F}}}{}{#2}{}{#1}}}
\newrobustcmd{\checkmathfrakFz}[2][]{\ensuremath{\subp{\check{\mathfrak{F}}}{}{#2}{}{#1}}}
\newrobustcmd{\tildemathfrakFz}[2][]{\ensuremath{\subp{\tilde{\mathfrak{F}}}{}{#2}{}{#1}}}
\newrobustcmd{\widetildemathfrakFz}[2][]{\ensuremath{\subp{\widetilde{\mathfrak{F}}}{}{#2}{}{#1}}}
\newrobustcmd{\acutemathfrakFz}[2][]{\ensuremath{\subp{\acute{\mathfrak{F}}}{}{#2}{}{#1}}}
\newrobustcmd{\gravemathfrakFz}[2][]{\ensuremath{\subp{\grave{\mathfrak{F}}}{}{#2}{}{#1}}}
\newrobustcmd{\dotmathfrakFz}[2][]{\ensuremath{\subp{\dot{\mathfrak{F}}}{}{#2}{}{#1}}}
\newrobustcmd{\ddotmathfrakFz}[2][]{\ensuremath{\subp{\ddot{\mathfrak{F}}}{}{#2}{}{#1}}}
\newrobustcmd{\brevemathfrakFz}[2][]{\ensuremath{\subp{\breve{\mathfrak{F}}}{}{#2}{}{#1}}}
\newrobustcmd{\barmathfrakFz}[2][]{\ensuremath{\subp{\bar{\mathfrak{F}}}{}{#2}{}{#1}}}
\newrobustcmd{\vecmathfrakFz}[2][]{\ensuremath{\subp{\vec{\mathfrak{F}}}{}{#2}{}{#1}}}
\newrobustcmd{\bmmathfrakFz}[2][]{\ensuremath{\subp{\bm{\mathfrak{F}}}{}{#2}{}{#1}}}
\newrobustcmd{\hatbmmathfrakFz}[2][]{\ensuremath{\subp{\hat{\bm{\mathfrak{F}}}}{}{#2}{}{#1}}}
\newrobustcmd{\widehatbmmathfrakFz}[2][]{\ensuremath{\subp{\widehat{\bm{\mathfrak{F}}}}{}{#2}{}{#1}}}
\newrobustcmd{\checkbmmathfrakFz}[2][]{\ensuremath{\subp{\check{\bm{\mathfrak{F}}}}{}{#2}{}{#1}}}
\newrobustcmd{\tildebmmathfrakFz}[2][]{\ensuremath{\subp{\tilde{\bm{\mathfrak{F}}}}{}{#2}{}{#1}}}
\newrobustcmd{\widetildebmmathfrakFz}[2][]{\ensuremath{\subp{\widetilde{\bm{\mathfrak{F}}}}{}{#2}{}{#1}}}
\newrobustcmd{\acutebmmathfrakFz}[2][]{\ensuremath{\subp{\acute{\bm{\mathfrak{F}}}}{}{#2}{}{#1}}}
\newrobustcmd{\gravebmmathfrakFz}[2][]{\ensuremath{\subp{\grave{\bm{\mathfrak{F}}}}{}{#2}{}{#1}}}
\newrobustcmd{\dotbmmathfrakFz}[2][]{\ensuremath{\subp{\dot{\bm{\mathfrak{F}}}}{}{#2}{}{#1}}}
\newrobustcmd{\ddotbmmathfrakFz}[2][]{\ensuremath{\subp{\ddot{\bm{\mathfrak{F}}}}{}{#2}{}{#1}}}
\newrobustcmd{\brevebmmathfrakFz}[2][]{\ensuremath{\subp{\breve{\bm{\mathfrak{F}}}}{}{#2}{}{#1}}}
\newrobustcmd{\barbmmathfrakFz}[2][]{\ensuremath{\subp{\bar{\bm{\mathfrak{F}}}}{}{#2}{}{#1}}}
\newrobustcmd{\vecbmmathfrakFz}[2][]{\ensuremath{\subp{\vec{\bm{\mathfrak{F}}}}{}{#2}{}{#1}}}
\newrobustcmd{\mathfrakGz}[2][]{\ensuremath{\subp{\mathfrak{G}}{}{#2}{}{#1}}}
\newrobustcmd{\hatmathfrakGz}[2][]{\ensuremath{\subp{\hat{\mathfrak{G}}}{}{#2}{}{#1}}}
\newrobustcmd{\widehatmathfrakGz}[2][]{\ensuremath{\subp{\widehat{\mathfrak{G}}}{}{#2}{}{#1}}}
\newrobustcmd{\checkmathfrakGz}[2][]{\ensuremath{\subp{\check{\mathfrak{G}}}{}{#2}{}{#1}}}
\newrobustcmd{\tildemathfrakGz}[2][]{\ensuremath{\subp{\tilde{\mathfrak{G}}}{}{#2}{}{#1}}}
\newrobustcmd{\widetildemathfrakGz}[2][]{\ensuremath{\subp{\widetilde{\mathfrak{G}}}{}{#2}{}{#1}}}
\newrobustcmd{\acutemathfrakGz}[2][]{\ensuremath{\subp{\acute{\mathfrak{G}}}{}{#2}{}{#1}}}
\newrobustcmd{\gravemathfrakGz}[2][]{\ensuremath{\subp{\grave{\mathfrak{G}}}{}{#2}{}{#1}}}
\newrobustcmd{\dotmathfrakGz}[2][]{\ensuremath{\subp{\dot{\mathfrak{G}}}{}{#2}{}{#1}}}
\newrobustcmd{\ddotmathfrakGz}[2][]{\ensuremath{\subp{\ddot{\mathfrak{G}}}{}{#2}{}{#1}}}
\newrobustcmd{\brevemathfrakGz}[2][]{\ensuremath{\subp{\breve{\mathfrak{G}}}{}{#2}{}{#1}}}
\newrobustcmd{\barmathfrakGz}[2][]{\ensuremath{\subp{\bar{\mathfrak{G}}}{}{#2}{}{#1}}}
\newrobustcmd{\vecmathfrakGz}[2][]{\ensuremath{\subp{\vec{\mathfrak{G}}}{}{#2}{}{#1}}}
\newrobustcmd{\bmmathfrakGz}[2][]{\ensuremath{\subp{\bm{\mathfrak{G}}}{}{#2}{}{#1}}}
\newrobustcmd{\hatbmmathfrakGz}[2][]{\ensuremath{\subp{\hat{\bm{\mathfrak{G}}}}{}{#2}{}{#1}}}
\newrobustcmd{\widehatbmmathfrakGz}[2][]{\ensuremath{\subp{\widehat{\bm{\mathfrak{G}}}}{}{#2}{}{#1}}}
\newrobustcmd{\checkbmmathfrakGz}[2][]{\ensuremath{\subp{\check{\bm{\mathfrak{G}}}}{}{#2}{}{#1}}}
\newrobustcmd{\tildebmmathfrakGz}[2][]{\ensuremath{\subp{\tilde{\bm{\mathfrak{G}}}}{}{#2}{}{#1}}}
\newrobustcmd{\widetildebmmathfrakGz}[2][]{\ensuremath{\subp{\widetilde{\bm{\mathfrak{G}}}}{}{#2}{}{#1}}}
\newrobustcmd{\acutebmmathfrakGz}[2][]{\ensuremath{\subp{\acute{\bm{\mathfrak{G}}}}{}{#2}{}{#1}}}
\newrobustcmd{\gravebmmathfrakGz}[2][]{\ensuremath{\subp{\grave{\bm{\mathfrak{G}}}}{}{#2}{}{#1}}}
\newrobustcmd{\dotbmmathfrakGz}[2][]{\ensuremath{\subp{\dot{\bm{\mathfrak{G}}}}{}{#2}{}{#1}}}
\newrobustcmd{\ddotbmmathfrakGz}[2][]{\ensuremath{\subp{\ddot{\bm{\mathfrak{G}}}}{}{#2}{}{#1}}}
\newrobustcmd{\brevebmmathfrakGz}[2][]{\ensuremath{\subp{\breve{\bm{\mathfrak{G}}}}{}{#2}{}{#1}}}
\newrobustcmd{\barbmmathfrakGz}[2][]{\ensuremath{\subp{\bar{\bm{\mathfrak{G}}}}{}{#2}{}{#1}}}
\newrobustcmd{\vecbmmathfrakGz}[2][]{\ensuremath{\subp{\vec{\bm{\mathfrak{G}}}}{}{#2}{}{#1}}}
\newrobustcmd{\mathfrakHz}[2][]{\ensuremath{\subp{\mathfrak{H}}{}{#2}{}{#1}}}
\newrobustcmd{\hatmathfrakHz}[2][]{\ensuremath{\subp{\hat{\mathfrak{H}}}{}{#2}{}{#1}}}
\newrobustcmd{\widehatmathfrakHz}[2][]{\ensuremath{\subp{\widehat{\mathfrak{H}}}{}{#2}{}{#1}}}
\newrobustcmd{\checkmathfrakHz}[2][]{\ensuremath{\subp{\check{\mathfrak{H}}}{}{#2}{}{#1}}}
\newrobustcmd{\tildemathfrakHz}[2][]{\ensuremath{\subp{\tilde{\mathfrak{H}}}{}{#2}{}{#1}}}
\newrobustcmd{\widetildemathfrakHz}[2][]{\ensuremath{\subp{\widetilde{\mathfrak{H}}}{}{#2}{}{#1}}}
\newrobustcmd{\acutemathfrakHz}[2][]{\ensuremath{\subp{\acute{\mathfrak{H}}}{}{#2}{}{#1}}}
\newrobustcmd{\gravemathfrakHz}[2][]{\ensuremath{\subp{\grave{\mathfrak{H}}}{}{#2}{}{#1}}}
\newrobustcmd{\dotmathfrakHz}[2][]{\ensuremath{\subp{\dot{\mathfrak{H}}}{}{#2}{}{#1}}}
\newrobustcmd{\ddotmathfrakHz}[2][]{\ensuremath{\subp{\ddot{\mathfrak{H}}}{}{#2}{}{#1}}}
\newrobustcmd{\brevemathfrakHz}[2][]{\ensuremath{\subp{\breve{\mathfrak{H}}}{}{#2}{}{#1}}}
\newrobustcmd{\barmathfrakHz}[2][]{\ensuremath{\subp{\bar{\mathfrak{H}}}{}{#2}{}{#1}}}
\newrobustcmd{\vecmathfrakHz}[2][]{\ensuremath{\subp{\vec{\mathfrak{H}}}{}{#2}{}{#1}}}
\newrobustcmd{\bmmathfrakHz}[2][]{\ensuremath{\subp{\bm{\mathfrak{H}}}{}{#2}{}{#1}}}
\newrobustcmd{\hatbmmathfrakHz}[2][]{\ensuremath{\subp{\hat{\bm{\mathfrak{H}}}}{}{#2}{}{#1}}}
\newrobustcmd{\widehatbmmathfrakHz}[2][]{\ensuremath{\subp{\widehat{\bm{\mathfrak{H}}}}{}{#2}{}{#1}}}
\newrobustcmd{\checkbmmathfrakHz}[2][]{\ensuremath{\subp{\check{\bm{\mathfrak{H}}}}{}{#2}{}{#1}}}
\newrobustcmd{\tildebmmathfrakHz}[2][]{\ensuremath{\subp{\tilde{\bm{\mathfrak{H}}}}{}{#2}{}{#1}}}
\newrobustcmd{\widetildebmmathfrakHz}[2][]{\ensuremath{\subp{\widetilde{\bm{\mathfrak{H}}}}{}{#2}{}{#1}}}
\newrobustcmd{\acutebmmathfrakHz}[2][]{\ensuremath{\subp{\acute{\bm{\mathfrak{H}}}}{}{#2}{}{#1}}}
\newrobustcmd{\gravebmmathfrakHz}[2][]{\ensuremath{\subp{\grave{\bm{\mathfrak{H}}}}{}{#2}{}{#1}}}
\newrobustcmd{\dotbmmathfrakHz}[2][]{\ensuremath{\subp{\dot{\bm{\mathfrak{H}}}}{}{#2}{}{#1}}}
\newrobustcmd{\ddotbmmathfrakHz}[2][]{\ensuremath{\subp{\ddot{\bm{\mathfrak{H}}}}{}{#2}{}{#1}}}
\newrobustcmd{\brevebmmathfrakHz}[2][]{\ensuremath{\subp{\breve{\bm{\mathfrak{H}}}}{}{#2}{}{#1}}}
\newrobustcmd{\barbmmathfrakHz}[2][]{\ensuremath{\subp{\bar{\bm{\mathfrak{H}}}}{}{#2}{}{#1}}}
\newrobustcmd{\vecbmmathfrakHz}[2][]{\ensuremath{\subp{\vec{\bm{\mathfrak{H}}}}{}{#2}{}{#1}}}
\newrobustcmd{\mathfrakIz}[2][]{\ensuremath{\subp{\mathfrak{I}}{}{#2}{}{#1}}}
\newrobustcmd{\hatmathfrakIz}[2][]{\ensuremath{\subp{\hat{\mathfrak{I}}}{}{#2}{}{#1}}}
\newrobustcmd{\widehatmathfrakIz}[2][]{\ensuremath{\subp{\widehat{\mathfrak{I}}}{}{#2}{}{#1}}}
\newrobustcmd{\checkmathfrakIz}[2][]{\ensuremath{\subp{\check{\mathfrak{I}}}{}{#2}{}{#1}}}
\newrobustcmd{\tildemathfrakIz}[2][]{\ensuremath{\subp{\tilde{\mathfrak{I}}}{}{#2}{}{#1}}}
\newrobustcmd{\widetildemathfrakIz}[2][]{\ensuremath{\subp{\widetilde{\mathfrak{I}}}{}{#2}{}{#1}}}
\newrobustcmd{\acutemathfrakIz}[2][]{\ensuremath{\subp{\acute{\mathfrak{I}}}{}{#2}{}{#1}}}
\newrobustcmd{\gravemathfrakIz}[2][]{\ensuremath{\subp{\grave{\mathfrak{I}}}{}{#2}{}{#1}}}
\newrobustcmd{\dotmathfrakIz}[2][]{\ensuremath{\subp{\dot{\mathfrak{I}}}{}{#2}{}{#1}}}
\newrobustcmd{\ddotmathfrakIz}[2][]{\ensuremath{\subp{\ddot{\mathfrak{I}}}{}{#2}{}{#1}}}
\newrobustcmd{\brevemathfrakIz}[2][]{\ensuremath{\subp{\breve{\mathfrak{I}}}{}{#2}{}{#1}}}
\newrobustcmd{\barmathfrakIz}[2][]{\ensuremath{\subp{\bar{\mathfrak{I}}}{}{#2}{}{#1}}}
\newrobustcmd{\vecmathfrakIz}[2][]{\ensuremath{\subp{\vec{\mathfrak{I}}}{}{#2}{}{#1}}}
\newrobustcmd{\bmmathfrakIz}[2][]{\ensuremath{\subp{\bm{\mathfrak{I}}}{}{#2}{}{#1}}}
\newrobustcmd{\hatbmmathfrakIz}[2][]{\ensuremath{\subp{\hat{\bm{\mathfrak{I}}}}{}{#2}{}{#1}}}
\newrobustcmd{\widehatbmmathfrakIz}[2][]{\ensuremath{\subp{\widehat{\bm{\mathfrak{I}}}}{}{#2}{}{#1}}}
\newrobustcmd{\checkbmmathfrakIz}[2][]{\ensuremath{\subp{\check{\bm{\mathfrak{I}}}}{}{#2}{}{#1}}}
\newrobustcmd{\tildebmmathfrakIz}[2][]{\ensuremath{\subp{\tilde{\bm{\mathfrak{I}}}}{}{#2}{}{#1}}}
\newrobustcmd{\widetildebmmathfrakIz}[2][]{\ensuremath{\subp{\widetilde{\bm{\mathfrak{I}}}}{}{#2}{}{#1}}}
\newrobustcmd{\acutebmmathfrakIz}[2][]{\ensuremath{\subp{\acute{\bm{\mathfrak{I}}}}{}{#2}{}{#1}}}
\newrobustcmd{\gravebmmathfrakIz}[2][]{\ensuremath{\subp{\grave{\bm{\mathfrak{I}}}}{}{#2}{}{#1}}}
\newrobustcmd{\dotbmmathfrakIz}[2][]{\ensuremath{\subp{\dot{\bm{\mathfrak{I}}}}{}{#2}{}{#1}}}
\newrobustcmd{\ddotbmmathfrakIz}[2][]{\ensuremath{\subp{\ddot{\bm{\mathfrak{I}}}}{}{#2}{}{#1}}}
\newrobustcmd{\brevebmmathfrakIz}[2][]{\ensuremath{\subp{\breve{\bm{\mathfrak{I}}}}{}{#2}{}{#1}}}
\newrobustcmd{\barbmmathfrakIz}[2][]{\ensuremath{\subp{\bar{\bm{\mathfrak{I}}}}{}{#2}{}{#1}}}
\newrobustcmd{\vecbmmathfrakIz}[2][]{\ensuremath{\subp{\vec{\bm{\mathfrak{I}}}}{}{#2}{}{#1}}}
\newrobustcmd{\mathfrakJz}[2][]{\ensuremath{\subp{\mathfrak{J}}{}{#2}{}{#1}}}
\newrobustcmd{\hatmathfrakJz}[2][]{\ensuremath{\subp{\hat{\mathfrak{J}}}{}{#2}{}{#1}}}
\newrobustcmd{\widehatmathfrakJz}[2][]{\ensuremath{\subp{\widehat{\mathfrak{J}}}{}{#2}{}{#1}}}
\newrobustcmd{\checkmathfrakJz}[2][]{\ensuremath{\subp{\check{\mathfrak{J}}}{}{#2}{}{#1}}}
\newrobustcmd{\tildemathfrakJz}[2][]{\ensuremath{\subp{\tilde{\mathfrak{J}}}{}{#2}{}{#1}}}
\newrobustcmd{\widetildemathfrakJz}[2][]{\ensuremath{\subp{\widetilde{\mathfrak{J}}}{}{#2}{}{#1}}}
\newrobustcmd{\acutemathfrakJz}[2][]{\ensuremath{\subp{\acute{\mathfrak{J}}}{}{#2}{}{#1}}}
\newrobustcmd{\gravemathfrakJz}[2][]{\ensuremath{\subp{\grave{\mathfrak{J}}}{}{#2}{}{#1}}}
\newrobustcmd{\dotmathfrakJz}[2][]{\ensuremath{\subp{\dot{\mathfrak{J}}}{}{#2}{}{#1}}}
\newrobustcmd{\ddotmathfrakJz}[2][]{\ensuremath{\subp{\ddot{\mathfrak{J}}}{}{#2}{}{#1}}}
\newrobustcmd{\brevemathfrakJz}[2][]{\ensuremath{\subp{\breve{\mathfrak{J}}}{}{#2}{}{#1}}}
\newrobustcmd{\barmathfrakJz}[2][]{\ensuremath{\subp{\bar{\mathfrak{J}}}{}{#2}{}{#1}}}
\newrobustcmd{\vecmathfrakJz}[2][]{\ensuremath{\subp{\vec{\mathfrak{J}}}{}{#2}{}{#1}}}
\newrobustcmd{\bmmathfrakJz}[2][]{\ensuremath{\subp{\bm{\mathfrak{J}}}{}{#2}{}{#1}}}
\newrobustcmd{\hatbmmathfrakJz}[2][]{\ensuremath{\subp{\hat{\bm{\mathfrak{J}}}}{}{#2}{}{#1}}}
\newrobustcmd{\widehatbmmathfrakJz}[2][]{\ensuremath{\subp{\widehat{\bm{\mathfrak{J}}}}{}{#2}{}{#1}}}
\newrobustcmd{\checkbmmathfrakJz}[2][]{\ensuremath{\subp{\check{\bm{\mathfrak{J}}}}{}{#2}{}{#1}}}
\newrobustcmd{\tildebmmathfrakJz}[2][]{\ensuremath{\subp{\tilde{\bm{\mathfrak{J}}}}{}{#2}{}{#1}}}
\newrobustcmd{\widetildebmmathfrakJz}[2][]{\ensuremath{\subp{\widetilde{\bm{\mathfrak{J}}}}{}{#2}{}{#1}}}
\newrobustcmd{\acutebmmathfrakJz}[2][]{\ensuremath{\subp{\acute{\bm{\mathfrak{J}}}}{}{#2}{}{#1}}}
\newrobustcmd{\gravebmmathfrakJz}[2][]{\ensuremath{\subp{\grave{\bm{\mathfrak{J}}}}{}{#2}{}{#1}}}
\newrobustcmd{\dotbmmathfrakJz}[2][]{\ensuremath{\subp{\dot{\bm{\mathfrak{J}}}}{}{#2}{}{#1}}}
\newrobustcmd{\ddotbmmathfrakJz}[2][]{\ensuremath{\subp{\ddot{\bm{\mathfrak{J}}}}{}{#2}{}{#1}}}
\newrobustcmd{\brevebmmathfrakJz}[2][]{\ensuremath{\subp{\breve{\bm{\mathfrak{J}}}}{}{#2}{}{#1}}}
\newrobustcmd{\barbmmathfrakJz}[2][]{\ensuremath{\subp{\bar{\bm{\mathfrak{J}}}}{}{#2}{}{#1}}}
\newrobustcmd{\vecbmmathfrakJz}[2][]{\ensuremath{\subp{\vec{\bm{\mathfrak{J}}}}{}{#2}{}{#1}}}
\newrobustcmd{\mathfrakKz}[2][]{\ensuremath{\subp{\mathfrak{K}}{}{#2}{}{#1}}}
\newrobustcmd{\hatmathfrakKz}[2][]{\ensuremath{\subp{\hat{\mathfrak{K}}}{}{#2}{}{#1}}}
\newrobustcmd{\widehatmathfrakKz}[2][]{\ensuremath{\subp{\widehat{\mathfrak{K}}}{}{#2}{}{#1}}}
\newrobustcmd{\checkmathfrakKz}[2][]{\ensuremath{\subp{\check{\mathfrak{K}}}{}{#2}{}{#1}}}
\newrobustcmd{\tildemathfrakKz}[2][]{\ensuremath{\subp{\tilde{\mathfrak{K}}}{}{#2}{}{#1}}}
\newrobustcmd{\widetildemathfrakKz}[2][]{\ensuremath{\subp{\widetilde{\mathfrak{K}}}{}{#2}{}{#1}}}
\newrobustcmd{\acutemathfrakKz}[2][]{\ensuremath{\subp{\acute{\mathfrak{K}}}{}{#2}{}{#1}}}
\newrobustcmd{\gravemathfrakKz}[2][]{\ensuremath{\subp{\grave{\mathfrak{K}}}{}{#2}{}{#1}}}
\newrobustcmd{\dotmathfrakKz}[2][]{\ensuremath{\subp{\dot{\mathfrak{K}}}{}{#2}{}{#1}}}
\newrobustcmd{\ddotmathfrakKz}[2][]{\ensuremath{\subp{\ddot{\mathfrak{K}}}{}{#2}{}{#1}}}
\newrobustcmd{\brevemathfrakKz}[2][]{\ensuremath{\subp{\breve{\mathfrak{K}}}{}{#2}{}{#1}}}
\newrobustcmd{\barmathfrakKz}[2][]{\ensuremath{\subp{\bar{\mathfrak{K}}}{}{#2}{}{#1}}}
\newrobustcmd{\vecmathfrakKz}[2][]{\ensuremath{\subp{\vec{\mathfrak{K}}}{}{#2}{}{#1}}}
\newrobustcmd{\bmmathfrakKz}[2][]{\ensuremath{\subp{\bm{\mathfrak{K}}}{}{#2}{}{#1}}}
\newrobustcmd{\hatbmmathfrakKz}[2][]{\ensuremath{\subp{\hat{\bm{\mathfrak{K}}}}{}{#2}{}{#1}}}
\newrobustcmd{\widehatbmmathfrakKz}[2][]{\ensuremath{\subp{\widehat{\bm{\mathfrak{K}}}}{}{#2}{}{#1}}}
\newrobustcmd{\checkbmmathfrakKz}[2][]{\ensuremath{\subp{\check{\bm{\mathfrak{K}}}}{}{#2}{}{#1}}}
\newrobustcmd{\tildebmmathfrakKz}[2][]{\ensuremath{\subp{\tilde{\bm{\mathfrak{K}}}}{}{#2}{}{#1}}}
\newrobustcmd{\widetildebmmathfrakKz}[2][]{\ensuremath{\subp{\widetilde{\bm{\mathfrak{K}}}}{}{#2}{}{#1}}}
\newrobustcmd{\acutebmmathfrakKz}[2][]{\ensuremath{\subp{\acute{\bm{\mathfrak{K}}}}{}{#2}{}{#1}}}
\newrobustcmd{\gravebmmathfrakKz}[2][]{\ensuremath{\subp{\grave{\bm{\mathfrak{K}}}}{}{#2}{}{#1}}}
\newrobustcmd{\dotbmmathfrakKz}[2][]{\ensuremath{\subp{\dot{\bm{\mathfrak{K}}}}{}{#2}{}{#1}}}
\newrobustcmd{\ddotbmmathfrakKz}[2][]{\ensuremath{\subp{\ddot{\bm{\mathfrak{K}}}}{}{#2}{}{#1}}}
\newrobustcmd{\brevebmmathfrakKz}[2][]{\ensuremath{\subp{\breve{\bm{\mathfrak{K}}}}{}{#2}{}{#1}}}
\newrobustcmd{\barbmmathfrakKz}[2][]{\ensuremath{\subp{\bar{\bm{\mathfrak{K}}}}{}{#2}{}{#1}}}
\newrobustcmd{\vecbmmathfrakKz}[2][]{\ensuremath{\subp{\vec{\bm{\mathfrak{K}}}}{}{#2}{}{#1}}}
\newrobustcmd{\mathfrakLz}[2][]{\ensuremath{\subp{\mathfrak{L}}{}{#2}{}{#1}}}
\newrobustcmd{\hatmathfrakLz}[2][]{\ensuremath{\subp{\hat{\mathfrak{L}}}{}{#2}{}{#1}}}
\newrobustcmd{\widehatmathfrakLz}[2][]{\ensuremath{\subp{\widehat{\mathfrak{L}}}{}{#2}{}{#1}}}
\newrobustcmd{\checkmathfrakLz}[2][]{\ensuremath{\subp{\check{\mathfrak{L}}}{}{#2}{}{#1}}}
\newrobustcmd{\tildemathfrakLz}[2][]{\ensuremath{\subp{\tilde{\mathfrak{L}}}{}{#2}{}{#1}}}
\newrobustcmd{\widetildemathfrakLz}[2][]{\ensuremath{\subp{\widetilde{\mathfrak{L}}}{}{#2}{}{#1}}}
\newrobustcmd{\acutemathfrakLz}[2][]{\ensuremath{\subp{\acute{\mathfrak{L}}}{}{#2}{}{#1}}}
\newrobustcmd{\gravemathfrakLz}[2][]{\ensuremath{\subp{\grave{\mathfrak{L}}}{}{#2}{}{#1}}}
\newrobustcmd{\dotmathfrakLz}[2][]{\ensuremath{\subp{\dot{\mathfrak{L}}}{}{#2}{}{#1}}}
\newrobustcmd{\ddotmathfrakLz}[2][]{\ensuremath{\subp{\ddot{\mathfrak{L}}}{}{#2}{}{#1}}}
\newrobustcmd{\brevemathfrakLz}[2][]{\ensuremath{\subp{\breve{\mathfrak{L}}}{}{#2}{}{#1}}}
\newrobustcmd{\barmathfrakLz}[2][]{\ensuremath{\subp{\bar{\mathfrak{L}}}{}{#2}{}{#1}}}
\newrobustcmd{\vecmathfrakLz}[2][]{\ensuremath{\subp{\vec{\mathfrak{L}}}{}{#2}{}{#1}}}
\newrobustcmd{\bmmathfrakLz}[2][]{\ensuremath{\subp{\bm{\mathfrak{L}}}{}{#2}{}{#1}}}
\newrobustcmd{\hatbmmathfrakLz}[2][]{\ensuremath{\subp{\hat{\bm{\mathfrak{L}}}}{}{#2}{}{#1}}}
\newrobustcmd{\widehatbmmathfrakLz}[2][]{\ensuremath{\subp{\widehat{\bm{\mathfrak{L}}}}{}{#2}{}{#1}}}
\newrobustcmd{\checkbmmathfrakLz}[2][]{\ensuremath{\subp{\check{\bm{\mathfrak{L}}}}{}{#2}{}{#1}}}
\newrobustcmd{\tildebmmathfrakLz}[2][]{\ensuremath{\subp{\tilde{\bm{\mathfrak{L}}}}{}{#2}{}{#1}}}
\newrobustcmd{\widetildebmmathfrakLz}[2][]{\ensuremath{\subp{\widetilde{\bm{\mathfrak{L}}}}{}{#2}{}{#1}}}
\newrobustcmd{\acutebmmathfrakLz}[2][]{\ensuremath{\subp{\acute{\bm{\mathfrak{L}}}}{}{#2}{}{#1}}}
\newrobustcmd{\gravebmmathfrakLz}[2][]{\ensuremath{\subp{\grave{\bm{\mathfrak{L}}}}{}{#2}{}{#1}}}
\newrobustcmd{\dotbmmathfrakLz}[2][]{\ensuremath{\subp{\dot{\bm{\mathfrak{L}}}}{}{#2}{}{#1}}}
\newrobustcmd{\ddotbmmathfrakLz}[2][]{\ensuremath{\subp{\ddot{\bm{\mathfrak{L}}}}{}{#2}{}{#1}}}
\newrobustcmd{\brevebmmathfrakLz}[2][]{\ensuremath{\subp{\breve{\bm{\mathfrak{L}}}}{}{#2}{}{#1}}}
\newrobustcmd{\barbmmathfrakLz}[2][]{\ensuremath{\subp{\bar{\bm{\mathfrak{L}}}}{}{#2}{}{#1}}}
\newrobustcmd{\vecbmmathfrakLz}[2][]{\ensuremath{\subp{\vec{\bm{\mathfrak{L}}}}{}{#2}{}{#1}}}
\newrobustcmd{\mathfrakMz}[2][]{\ensuremath{\subp{\mathfrak{M}}{}{#2}{}{#1}}}
\newrobustcmd{\hatmathfrakMz}[2][]{\ensuremath{\subp{\hat{\mathfrak{M}}}{}{#2}{}{#1}}}
\newrobustcmd{\widehatmathfrakMz}[2][]{\ensuremath{\subp{\widehat{\mathfrak{M}}}{}{#2}{}{#1}}}
\newrobustcmd{\checkmathfrakMz}[2][]{\ensuremath{\subp{\check{\mathfrak{M}}}{}{#2}{}{#1}}}
\newrobustcmd{\tildemathfrakMz}[2][]{\ensuremath{\subp{\tilde{\mathfrak{M}}}{}{#2}{}{#1}}}
\newrobustcmd{\widetildemathfrakMz}[2][]{\ensuremath{\subp{\widetilde{\mathfrak{M}}}{}{#2}{}{#1}}}
\newrobustcmd{\acutemathfrakMz}[2][]{\ensuremath{\subp{\acute{\mathfrak{M}}}{}{#2}{}{#1}}}
\newrobustcmd{\gravemathfrakMz}[2][]{\ensuremath{\subp{\grave{\mathfrak{M}}}{}{#2}{}{#1}}}
\newrobustcmd{\dotmathfrakMz}[2][]{\ensuremath{\subp{\dot{\mathfrak{M}}}{}{#2}{}{#1}}}
\newrobustcmd{\ddotmathfrakMz}[2][]{\ensuremath{\subp{\ddot{\mathfrak{M}}}{}{#2}{}{#1}}}
\newrobustcmd{\brevemathfrakMz}[2][]{\ensuremath{\subp{\breve{\mathfrak{M}}}{}{#2}{}{#1}}}
\newrobustcmd{\barmathfrakMz}[2][]{\ensuremath{\subp{\bar{\mathfrak{M}}}{}{#2}{}{#1}}}
\newrobustcmd{\vecmathfrakMz}[2][]{\ensuremath{\subp{\vec{\mathfrak{M}}}{}{#2}{}{#1}}}
\newrobustcmd{\bmmathfrakMz}[2][]{\ensuremath{\subp{\bm{\mathfrak{M}}}{}{#2}{}{#1}}}
\newrobustcmd{\hatbmmathfrakMz}[2][]{\ensuremath{\subp{\hat{\bm{\mathfrak{M}}}}{}{#2}{}{#1}}}
\newrobustcmd{\widehatbmmathfrakMz}[2][]{\ensuremath{\subp{\widehat{\bm{\mathfrak{M}}}}{}{#2}{}{#1}}}
\newrobustcmd{\checkbmmathfrakMz}[2][]{\ensuremath{\subp{\check{\bm{\mathfrak{M}}}}{}{#2}{}{#1}}}
\newrobustcmd{\tildebmmathfrakMz}[2][]{\ensuremath{\subp{\tilde{\bm{\mathfrak{M}}}}{}{#2}{}{#1}}}
\newrobustcmd{\widetildebmmathfrakMz}[2][]{\ensuremath{\subp{\widetilde{\bm{\mathfrak{M}}}}{}{#2}{}{#1}}}
\newrobustcmd{\acutebmmathfrakMz}[2][]{\ensuremath{\subp{\acute{\bm{\mathfrak{M}}}}{}{#2}{}{#1}}}
\newrobustcmd{\gravebmmathfrakMz}[2][]{\ensuremath{\subp{\grave{\bm{\mathfrak{M}}}}{}{#2}{}{#1}}}
\newrobustcmd{\dotbmmathfrakMz}[2][]{\ensuremath{\subp{\dot{\bm{\mathfrak{M}}}}{}{#2}{}{#1}}}
\newrobustcmd{\ddotbmmathfrakMz}[2][]{\ensuremath{\subp{\ddot{\bm{\mathfrak{M}}}}{}{#2}{}{#1}}}
\newrobustcmd{\brevebmmathfrakMz}[2][]{\ensuremath{\subp{\breve{\bm{\mathfrak{M}}}}{}{#2}{}{#1}}}
\newrobustcmd{\barbmmathfrakMz}[2][]{\ensuremath{\subp{\bar{\bm{\mathfrak{M}}}}{}{#2}{}{#1}}}
\newrobustcmd{\vecbmmathfrakMz}[2][]{\ensuremath{\subp{\vec{\bm{\mathfrak{M}}}}{}{#2}{}{#1}}}
\newrobustcmd{\mathfrakNz}[2][]{\ensuremath{\subp{\mathfrak{N}}{}{#2}{}{#1}}}
\newrobustcmd{\hatmathfrakNz}[2][]{\ensuremath{\subp{\hat{\mathfrak{N}}}{}{#2}{}{#1}}}
\newrobustcmd{\widehatmathfrakNz}[2][]{\ensuremath{\subp{\widehat{\mathfrak{N}}}{}{#2}{}{#1}}}
\newrobustcmd{\checkmathfrakNz}[2][]{\ensuremath{\subp{\check{\mathfrak{N}}}{}{#2}{}{#1}}}
\newrobustcmd{\tildemathfrakNz}[2][]{\ensuremath{\subp{\tilde{\mathfrak{N}}}{}{#2}{}{#1}}}
\newrobustcmd{\widetildemathfrakNz}[2][]{\ensuremath{\subp{\widetilde{\mathfrak{N}}}{}{#2}{}{#1}}}
\newrobustcmd{\acutemathfrakNz}[2][]{\ensuremath{\subp{\acute{\mathfrak{N}}}{}{#2}{}{#1}}}
\newrobustcmd{\gravemathfrakNz}[2][]{\ensuremath{\subp{\grave{\mathfrak{N}}}{}{#2}{}{#1}}}
\newrobustcmd{\dotmathfrakNz}[2][]{\ensuremath{\subp{\dot{\mathfrak{N}}}{}{#2}{}{#1}}}
\newrobustcmd{\ddotmathfrakNz}[2][]{\ensuremath{\subp{\ddot{\mathfrak{N}}}{}{#2}{}{#1}}}
\newrobustcmd{\brevemathfrakNz}[2][]{\ensuremath{\subp{\breve{\mathfrak{N}}}{}{#2}{}{#1}}}
\newrobustcmd{\barmathfrakNz}[2][]{\ensuremath{\subp{\bar{\mathfrak{N}}}{}{#2}{}{#1}}}
\newrobustcmd{\vecmathfrakNz}[2][]{\ensuremath{\subp{\vec{\mathfrak{N}}}{}{#2}{}{#1}}}
\newrobustcmd{\bmmathfrakNz}[2][]{\ensuremath{\subp{\bm{\mathfrak{N}}}{}{#2}{}{#1}}}
\newrobustcmd{\hatbmmathfrakNz}[2][]{\ensuremath{\subp{\hat{\bm{\mathfrak{N}}}}{}{#2}{}{#1}}}
\newrobustcmd{\widehatbmmathfrakNz}[2][]{\ensuremath{\subp{\widehat{\bm{\mathfrak{N}}}}{}{#2}{}{#1}}}
\newrobustcmd{\checkbmmathfrakNz}[2][]{\ensuremath{\subp{\check{\bm{\mathfrak{N}}}}{}{#2}{}{#1}}}
\newrobustcmd{\tildebmmathfrakNz}[2][]{\ensuremath{\subp{\tilde{\bm{\mathfrak{N}}}}{}{#2}{}{#1}}}
\newrobustcmd{\widetildebmmathfrakNz}[2][]{\ensuremath{\subp{\widetilde{\bm{\mathfrak{N}}}}{}{#2}{}{#1}}}
\newrobustcmd{\acutebmmathfrakNz}[2][]{\ensuremath{\subp{\acute{\bm{\mathfrak{N}}}}{}{#2}{}{#1}}}
\newrobustcmd{\gravebmmathfrakNz}[2][]{\ensuremath{\subp{\grave{\bm{\mathfrak{N}}}}{}{#2}{}{#1}}}
\newrobustcmd{\dotbmmathfrakNz}[2][]{\ensuremath{\subp{\dot{\bm{\mathfrak{N}}}}{}{#2}{}{#1}}}
\newrobustcmd{\ddotbmmathfrakNz}[2][]{\ensuremath{\subp{\ddot{\bm{\mathfrak{N}}}}{}{#2}{}{#1}}}
\newrobustcmd{\brevebmmathfrakNz}[2][]{\ensuremath{\subp{\breve{\bm{\mathfrak{N}}}}{}{#2}{}{#1}}}
\newrobustcmd{\barbmmathfrakNz}[2][]{\ensuremath{\subp{\bar{\bm{\mathfrak{N}}}}{}{#2}{}{#1}}}
\newrobustcmd{\vecbmmathfrakNz}[2][]{\ensuremath{\subp{\vec{\bm{\mathfrak{N}}}}{}{#2}{}{#1}}}
\newrobustcmd{\mathfrakOz}[2][]{\ensuremath{\subp{\mathfrak{O}}{}{#2}{}{#1}}}
\newrobustcmd{\hatmathfrakOz}[2][]{\ensuremath{\subp{\hat{\mathfrak{O}}}{}{#2}{}{#1}}}
\newrobustcmd{\widehatmathfrakOz}[2][]{\ensuremath{\subp{\widehat{\mathfrak{O}}}{}{#2}{}{#1}}}
\newrobustcmd{\checkmathfrakOz}[2][]{\ensuremath{\subp{\check{\mathfrak{O}}}{}{#2}{}{#1}}}
\newrobustcmd{\tildemathfrakOz}[2][]{\ensuremath{\subp{\tilde{\mathfrak{O}}}{}{#2}{}{#1}}}
\newrobustcmd{\widetildemathfrakOz}[2][]{\ensuremath{\subp{\widetilde{\mathfrak{O}}}{}{#2}{}{#1}}}
\newrobustcmd{\acutemathfrakOz}[2][]{\ensuremath{\subp{\acute{\mathfrak{O}}}{}{#2}{}{#1}}}
\newrobustcmd{\gravemathfrakOz}[2][]{\ensuremath{\subp{\grave{\mathfrak{O}}}{}{#2}{}{#1}}}
\newrobustcmd{\dotmathfrakOz}[2][]{\ensuremath{\subp{\dot{\mathfrak{O}}}{}{#2}{}{#1}}}
\newrobustcmd{\ddotmathfrakOz}[2][]{\ensuremath{\subp{\ddot{\mathfrak{O}}}{}{#2}{}{#1}}}
\newrobustcmd{\brevemathfrakOz}[2][]{\ensuremath{\subp{\breve{\mathfrak{O}}}{}{#2}{}{#1}}}
\newrobustcmd{\barmathfrakOz}[2][]{\ensuremath{\subp{\bar{\mathfrak{O}}}{}{#2}{}{#1}}}
\newrobustcmd{\vecmathfrakOz}[2][]{\ensuremath{\subp{\vec{\mathfrak{O}}}{}{#2}{}{#1}}}
\newrobustcmd{\bmmathfrakOz}[2][]{\ensuremath{\subp{\bm{\mathfrak{O}}}{}{#2}{}{#1}}}
\newrobustcmd{\hatbmmathfrakOz}[2][]{\ensuremath{\subp{\hat{\bm{\mathfrak{O}}}}{}{#2}{}{#1}}}
\newrobustcmd{\widehatbmmathfrakOz}[2][]{\ensuremath{\subp{\widehat{\bm{\mathfrak{O}}}}{}{#2}{}{#1}}}
\newrobustcmd{\checkbmmathfrakOz}[2][]{\ensuremath{\subp{\check{\bm{\mathfrak{O}}}}{}{#2}{}{#1}}}
\newrobustcmd{\tildebmmathfrakOz}[2][]{\ensuremath{\subp{\tilde{\bm{\mathfrak{O}}}}{}{#2}{}{#1}}}
\newrobustcmd{\widetildebmmathfrakOz}[2][]{\ensuremath{\subp{\widetilde{\bm{\mathfrak{O}}}}{}{#2}{}{#1}}}
\newrobustcmd{\acutebmmathfrakOz}[2][]{\ensuremath{\subp{\acute{\bm{\mathfrak{O}}}}{}{#2}{}{#1}}}
\newrobustcmd{\gravebmmathfrakOz}[2][]{\ensuremath{\subp{\grave{\bm{\mathfrak{O}}}}{}{#2}{}{#1}}}
\newrobustcmd{\dotbmmathfrakOz}[2][]{\ensuremath{\subp{\dot{\bm{\mathfrak{O}}}}{}{#2}{}{#1}}}
\newrobustcmd{\ddotbmmathfrakOz}[2][]{\ensuremath{\subp{\ddot{\bm{\mathfrak{O}}}}{}{#2}{}{#1}}}
\newrobustcmd{\brevebmmathfrakOz}[2][]{\ensuremath{\subp{\breve{\bm{\mathfrak{O}}}}{}{#2}{}{#1}}}
\newrobustcmd{\barbmmathfrakOz}[2][]{\ensuremath{\subp{\bar{\bm{\mathfrak{O}}}}{}{#2}{}{#1}}}
\newrobustcmd{\vecbmmathfrakOz}[2][]{\ensuremath{\subp{\vec{\bm{\mathfrak{O}}}}{}{#2}{}{#1}}}
\newrobustcmd{\mathfrakPz}[2][]{\ensuremath{\subp{\mathfrak{P}}{}{#2}{}{#1}}}
\newrobustcmd{\hatmathfrakPz}[2][]{\ensuremath{\subp{\hat{\mathfrak{P}}}{}{#2}{}{#1}}}
\newrobustcmd{\widehatmathfrakPz}[2][]{\ensuremath{\subp{\widehat{\mathfrak{P}}}{}{#2}{}{#1}}}
\newrobustcmd{\checkmathfrakPz}[2][]{\ensuremath{\subp{\check{\mathfrak{P}}}{}{#2}{}{#1}}}
\newrobustcmd{\tildemathfrakPz}[2][]{\ensuremath{\subp{\tilde{\mathfrak{P}}}{}{#2}{}{#1}}}
\newrobustcmd{\widetildemathfrakPz}[2][]{\ensuremath{\subp{\widetilde{\mathfrak{P}}}{}{#2}{}{#1}}}
\newrobustcmd{\acutemathfrakPz}[2][]{\ensuremath{\subp{\acute{\mathfrak{P}}}{}{#2}{}{#1}}}
\newrobustcmd{\gravemathfrakPz}[2][]{\ensuremath{\subp{\grave{\mathfrak{P}}}{}{#2}{}{#1}}}
\newrobustcmd{\dotmathfrakPz}[2][]{\ensuremath{\subp{\dot{\mathfrak{P}}}{}{#2}{}{#1}}}
\newrobustcmd{\ddotmathfrakPz}[2][]{\ensuremath{\subp{\ddot{\mathfrak{P}}}{}{#2}{}{#1}}}
\newrobustcmd{\brevemathfrakPz}[2][]{\ensuremath{\subp{\breve{\mathfrak{P}}}{}{#2}{}{#1}}}
\newrobustcmd{\barmathfrakPz}[2][]{\ensuremath{\subp{\bar{\mathfrak{P}}}{}{#2}{}{#1}}}
\newrobustcmd{\vecmathfrakPz}[2][]{\ensuremath{\subp{\vec{\mathfrak{P}}}{}{#2}{}{#1}}}
\newrobustcmd{\bmmathfrakPz}[2][]{\ensuremath{\subp{\bm{\mathfrak{P}}}{}{#2}{}{#1}}}
\newrobustcmd{\hatbmmathfrakPz}[2][]{\ensuremath{\subp{\hat{\bm{\mathfrak{P}}}}{}{#2}{}{#1}}}
\newrobustcmd{\widehatbmmathfrakPz}[2][]{\ensuremath{\subp{\widehat{\bm{\mathfrak{P}}}}{}{#2}{}{#1}}}
\newrobustcmd{\checkbmmathfrakPz}[2][]{\ensuremath{\subp{\check{\bm{\mathfrak{P}}}}{}{#2}{}{#1}}}
\newrobustcmd{\tildebmmathfrakPz}[2][]{\ensuremath{\subp{\tilde{\bm{\mathfrak{P}}}}{}{#2}{}{#1}}}
\newrobustcmd{\widetildebmmathfrakPz}[2][]{\ensuremath{\subp{\widetilde{\bm{\mathfrak{P}}}}{}{#2}{}{#1}}}
\newrobustcmd{\acutebmmathfrakPz}[2][]{\ensuremath{\subp{\acute{\bm{\mathfrak{P}}}}{}{#2}{}{#1}}}
\newrobustcmd{\gravebmmathfrakPz}[2][]{\ensuremath{\subp{\grave{\bm{\mathfrak{P}}}}{}{#2}{}{#1}}}
\newrobustcmd{\dotbmmathfrakPz}[2][]{\ensuremath{\subp{\dot{\bm{\mathfrak{P}}}}{}{#2}{}{#1}}}
\newrobustcmd{\ddotbmmathfrakPz}[2][]{\ensuremath{\subp{\ddot{\bm{\mathfrak{P}}}}{}{#2}{}{#1}}}
\newrobustcmd{\brevebmmathfrakPz}[2][]{\ensuremath{\subp{\breve{\bm{\mathfrak{P}}}}{}{#2}{}{#1}}}
\newrobustcmd{\barbmmathfrakPz}[2][]{\ensuremath{\subp{\bar{\bm{\mathfrak{P}}}}{}{#2}{}{#1}}}
\newrobustcmd{\vecbmmathfrakPz}[2][]{\ensuremath{\subp{\vec{\bm{\mathfrak{P}}}}{}{#2}{}{#1}}}
\newrobustcmd{\mathfrakQz}[2][]{\ensuremath{\subp{\mathfrak{Q}}{}{#2}{}{#1}}}
\newrobustcmd{\hatmathfrakQz}[2][]{\ensuremath{\subp{\hat{\mathfrak{Q}}}{}{#2}{}{#1}}}
\newrobustcmd{\widehatmathfrakQz}[2][]{\ensuremath{\subp{\widehat{\mathfrak{Q}}}{}{#2}{}{#1}}}
\newrobustcmd{\checkmathfrakQz}[2][]{\ensuremath{\subp{\check{\mathfrak{Q}}}{}{#2}{}{#1}}}
\newrobustcmd{\tildemathfrakQz}[2][]{\ensuremath{\subp{\tilde{\mathfrak{Q}}}{}{#2}{}{#1}}}
\newrobustcmd{\widetildemathfrakQz}[2][]{\ensuremath{\subp{\widetilde{\mathfrak{Q}}}{}{#2}{}{#1}}}
\newrobustcmd{\acutemathfrakQz}[2][]{\ensuremath{\subp{\acute{\mathfrak{Q}}}{}{#2}{}{#1}}}
\newrobustcmd{\gravemathfrakQz}[2][]{\ensuremath{\subp{\grave{\mathfrak{Q}}}{}{#2}{}{#1}}}
\newrobustcmd{\dotmathfrakQz}[2][]{\ensuremath{\subp{\dot{\mathfrak{Q}}}{}{#2}{}{#1}}}
\newrobustcmd{\ddotmathfrakQz}[2][]{\ensuremath{\subp{\ddot{\mathfrak{Q}}}{}{#2}{}{#1}}}
\newrobustcmd{\brevemathfrakQz}[2][]{\ensuremath{\subp{\breve{\mathfrak{Q}}}{}{#2}{}{#1}}}
\newrobustcmd{\barmathfrakQz}[2][]{\ensuremath{\subp{\bar{\mathfrak{Q}}}{}{#2}{}{#1}}}
\newrobustcmd{\vecmathfrakQz}[2][]{\ensuremath{\subp{\vec{\mathfrak{Q}}}{}{#2}{}{#1}}}
\newrobustcmd{\bmmathfrakQz}[2][]{\ensuremath{\subp{\bm{\mathfrak{Q}}}{}{#2}{}{#1}}}
\newrobustcmd{\hatbmmathfrakQz}[2][]{\ensuremath{\subp{\hat{\bm{\mathfrak{Q}}}}{}{#2}{}{#1}}}
\newrobustcmd{\widehatbmmathfrakQz}[2][]{\ensuremath{\subp{\widehat{\bm{\mathfrak{Q}}}}{}{#2}{}{#1}}}
\newrobustcmd{\checkbmmathfrakQz}[2][]{\ensuremath{\subp{\check{\bm{\mathfrak{Q}}}}{}{#2}{}{#1}}}
\newrobustcmd{\tildebmmathfrakQz}[2][]{\ensuremath{\subp{\tilde{\bm{\mathfrak{Q}}}}{}{#2}{}{#1}}}
\newrobustcmd{\widetildebmmathfrakQz}[2][]{\ensuremath{\subp{\widetilde{\bm{\mathfrak{Q}}}}{}{#2}{}{#1}}}
\newrobustcmd{\acutebmmathfrakQz}[2][]{\ensuremath{\subp{\acute{\bm{\mathfrak{Q}}}}{}{#2}{}{#1}}}
\newrobustcmd{\gravebmmathfrakQz}[2][]{\ensuremath{\subp{\grave{\bm{\mathfrak{Q}}}}{}{#2}{}{#1}}}
\newrobustcmd{\dotbmmathfrakQz}[2][]{\ensuremath{\subp{\dot{\bm{\mathfrak{Q}}}}{}{#2}{}{#1}}}
\newrobustcmd{\ddotbmmathfrakQz}[2][]{\ensuremath{\subp{\ddot{\bm{\mathfrak{Q}}}}{}{#2}{}{#1}}}
\newrobustcmd{\brevebmmathfrakQz}[2][]{\ensuremath{\subp{\breve{\bm{\mathfrak{Q}}}}{}{#2}{}{#1}}}
\newrobustcmd{\barbmmathfrakQz}[2][]{\ensuremath{\subp{\bar{\bm{\mathfrak{Q}}}}{}{#2}{}{#1}}}
\newrobustcmd{\vecbmmathfrakQz}[2][]{\ensuremath{\subp{\vec{\bm{\mathfrak{Q}}}}{}{#2}{}{#1}}}
\newrobustcmd{\mathfrakRz}[2][]{\ensuremath{\subp{\mathfrak{R}}{}{#2}{}{#1}}}
\newrobustcmd{\hatmathfrakRz}[2][]{\ensuremath{\subp{\hat{\mathfrak{R}}}{}{#2}{}{#1}}}
\newrobustcmd{\widehatmathfrakRz}[2][]{\ensuremath{\subp{\widehat{\mathfrak{R}}}{}{#2}{}{#1}}}
\newrobustcmd{\checkmathfrakRz}[2][]{\ensuremath{\subp{\check{\mathfrak{R}}}{}{#2}{}{#1}}}
\newrobustcmd{\tildemathfrakRz}[2][]{\ensuremath{\subp{\tilde{\mathfrak{R}}}{}{#2}{}{#1}}}
\newrobustcmd{\widetildemathfrakRz}[2][]{\ensuremath{\subp{\widetilde{\mathfrak{R}}}{}{#2}{}{#1}}}
\newrobustcmd{\acutemathfrakRz}[2][]{\ensuremath{\subp{\acute{\mathfrak{R}}}{}{#2}{}{#1}}}
\newrobustcmd{\gravemathfrakRz}[2][]{\ensuremath{\subp{\grave{\mathfrak{R}}}{}{#2}{}{#1}}}
\newrobustcmd{\dotmathfrakRz}[2][]{\ensuremath{\subp{\dot{\mathfrak{R}}}{}{#2}{}{#1}}}
\newrobustcmd{\ddotmathfrakRz}[2][]{\ensuremath{\subp{\ddot{\mathfrak{R}}}{}{#2}{}{#1}}}
\newrobustcmd{\brevemathfrakRz}[2][]{\ensuremath{\subp{\breve{\mathfrak{R}}}{}{#2}{}{#1}}}
\newrobustcmd{\barmathfrakRz}[2][]{\ensuremath{\subp{\bar{\mathfrak{R}}}{}{#2}{}{#1}}}
\newrobustcmd{\vecmathfrakRz}[2][]{\ensuremath{\subp{\vec{\mathfrak{R}}}{}{#2}{}{#1}}}
\newrobustcmd{\bmmathfrakRz}[2][]{\ensuremath{\subp{\bm{\mathfrak{R}}}{}{#2}{}{#1}}}
\newrobustcmd{\hatbmmathfrakRz}[2][]{\ensuremath{\subp{\hat{\bm{\mathfrak{R}}}}{}{#2}{}{#1}}}
\newrobustcmd{\widehatbmmathfrakRz}[2][]{\ensuremath{\subp{\widehat{\bm{\mathfrak{R}}}}{}{#2}{}{#1}}}
\newrobustcmd{\checkbmmathfrakRz}[2][]{\ensuremath{\subp{\check{\bm{\mathfrak{R}}}}{}{#2}{}{#1}}}
\newrobustcmd{\tildebmmathfrakRz}[2][]{\ensuremath{\subp{\tilde{\bm{\mathfrak{R}}}}{}{#2}{}{#1}}}
\newrobustcmd{\widetildebmmathfrakRz}[2][]{\ensuremath{\subp{\widetilde{\bm{\mathfrak{R}}}}{}{#2}{}{#1}}}
\newrobustcmd{\acutebmmathfrakRz}[2][]{\ensuremath{\subp{\acute{\bm{\mathfrak{R}}}}{}{#2}{}{#1}}}
\newrobustcmd{\gravebmmathfrakRz}[2][]{\ensuremath{\subp{\grave{\bm{\mathfrak{R}}}}{}{#2}{}{#1}}}
\newrobustcmd{\dotbmmathfrakRz}[2][]{\ensuremath{\subp{\dot{\bm{\mathfrak{R}}}}{}{#2}{}{#1}}}
\newrobustcmd{\ddotbmmathfrakRz}[2][]{\ensuremath{\subp{\ddot{\bm{\mathfrak{R}}}}{}{#2}{}{#1}}}
\newrobustcmd{\brevebmmathfrakRz}[2][]{\ensuremath{\subp{\breve{\bm{\mathfrak{R}}}}{}{#2}{}{#1}}}
\newrobustcmd{\barbmmathfrakRz}[2][]{\ensuremath{\subp{\bar{\bm{\mathfrak{R}}}}{}{#2}{}{#1}}}
\newrobustcmd{\vecbmmathfrakRz}[2][]{\ensuremath{\subp{\vec{\bm{\mathfrak{R}}}}{}{#2}{}{#1}}}
\newrobustcmd{\mathfrakSz}[2][]{\ensuremath{\subp{\mathfrak{S}}{}{#2}{}{#1}}}
\newrobustcmd{\hatmathfrakSz}[2][]{\ensuremath{\subp{\hat{\mathfrak{S}}}{}{#2}{}{#1}}}
\newrobustcmd{\widehatmathfrakSz}[2][]{\ensuremath{\subp{\widehat{\mathfrak{S}}}{}{#2}{}{#1}}}
\newrobustcmd{\checkmathfrakSz}[2][]{\ensuremath{\subp{\check{\mathfrak{S}}}{}{#2}{}{#1}}}
\newrobustcmd{\tildemathfrakSz}[2][]{\ensuremath{\subp{\tilde{\mathfrak{S}}}{}{#2}{}{#1}}}
\newrobustcmd{\widetildemathfrakSz}[2][]{\ensuremath{\subp{\widetilde{\mathfrak{S}}}{}{#2}{}{#1}}}
\newrobustcmd{\acutemathfrakSz}[2][]{\ensuremath{\subp{\acute{\mathfrak{S}}}{}{#2}{}{#1}}}
\newrobustcmd{\gravemathfrakSz}[2][]{\ensuremath{\subp{\grave{\mathfrak{S}}}{}{#2}{}{#1}}}
\newrobustcmd{\dotmathfrakSz}[2][]{\ensuremath{\subp{\dot{\mathfrak{S}}}{}{#2}{}{#1}}}
\newrobustcmd{\ddotmathfrakSz}[2][]{\ensuremath{\subp{\ddot{\mathfrak{S}}}{}{#2}{}{#1}}}
\newrobustcmd{\brevemathfrakSz}[2][]{\ensuremath{\subp{\breve{\mathfrak{S}}}{}{#2}{}{#1}}}
\newrobustcmd{\barmathfrakSz}[2][]{\ensuremath{\subp{\bar{\mathfrak{S}}}{}{#2}{}{#1}}}
\newrobustcmd{\vecmathfrakSz}[2][]{\ensuremath{\subp{\vec{\mathfrak{S}}}{}{#2}{}{#1}}}
\newrobustcmd{\bmmathfrakSz}[2][]{\ensuremath{\subp{\bm{\mathfrak{S}}}{}{#2}{}{#1}}}
\newrobustcmd{\hatbmmathfrakSz}[2][]{\ensuremath{\subp{\hat{\bm{\mathfrak{S}}}}{}{#2}{}{#1}}}
\newrobustcmd{\widehatbmmathfrakSz}[2][]{\ensuremath{\subp{\widehat{\bm{\mathfrak{S}}}}{}{#2}{}{#1}}}
\newrobustcmd{\checkbmmathfrakSz}[2][]{\ensuremath{\subp{\check{\bm{\mathfrak{S}}}}{}{#2}{}{#1}}}
\newrobustcmd{\tildebmmathfrakSz}[2][]{\ensuremath{\subp{\tilde{\bm{\mathfrak{S}}}}{}{#2}{}{#1}}}
\newrobustcmd{\widetildebmmathfrakSz}[2][]{\ensuremath{\subp{\widetilde{\bm{\mathfrak{S}}}}{}{#2}{}{#1}}}
\newrobustcmd{\acutebmmathfrakSz}[2][]{\ensuremath{\subp{\acute{\bm{\mathfrak{S}}}}{}{#2}{}{#1}}}
\newrobustcmd{\gravebmmathfrakSz}[2][]{\ensuremath{\subp{\grave{\bm{\mathfrak{S}}}}{}{#2}{}{#1}}}
\newrobustcmd{\dotbmmathfrakSz}[2][]{\ensuremath{\subp{\dot{\bm{\mathfrak{S}}}}{}{#2}{}{#1}}}
\newrobustcmd{\ddotbmmathfrakSz}[2][]{\ensuremath{\subp{\ddot{\bm{\mathfrak{S}}}}{}{#2}{}{#1}}}
\newrobustcmd{\brevebmmathfrakSz}[2][]{\ensuremath{\subp{\breve{\bm{\mathfrak{S}}}}{}{#2}{}{#1}}}
\newrobustcmd{\barbmmathfrakSz}[2][]{\ensuremath{\subp{\bar{\bm{\mathfrak{S}}}}{}{#2}{}{#1}}}
\newrobustcmd{\vecbmmathfrakSz}[2][]{\ensuremath{\subp{\vec{\bm{\mathfrak{S}}}}{}{#2}{}{#1}}}
\newrobustcmd{\mathfrakTz}[2][]{\ensuremath{\subp{\mathfrak{T}}{}{#2}{}{#1}}}
\newrobustcmd{\hatmathfrakTz}[2][]{\ensuremath{\subp{\hat{\mathfrak{T}}}{}{#2}{}{#1}}}
\newrobustcmd{\widehatmathfrakTz}[2][]{\ensuremath{\subp{\widehat{\mathfrak{T}}}{}{#2}{}{#1}}}
\newrobustcmd{\checkmathfrakTz}[2][]{\ensuremath{\subp{\check{\mathfrak{T}}}{}{#2}{}{#1}}}
\newrobustcmd{\tildemathfrakTz}[2][]{\ensuremath{\subp{\tilde{\mathfrak{T}}}{}{#2}{}{#1}}}
\newrobustcmd{\widetildemathfrakTz}[2][]{\ensuremath{\subp{\widetilde{\mathfrak{T}}}{}{#2}{}{#1}}}
\newrobustcmd{\acutemathfrakTz}[2][]{\ensuremath{\subp{\acute{\mathfrak{T}}}{}{#2}{}{#1}}}
\newrobustcmd{\gravemathfrakTz}[2][]{\ensuremath{\subp{\grave{\mathfrak{T}}}{}{#2}{}{#1}}}
\newrobustcmd{\dotmathfrakTz}[2][]{\ensuremath{\subp{\dot{\mathfrak{T}}}{}{#2}{}{#1}}}
\newrobustcmd{\ddotmathfrakTz}[2][]{\ensuremath{\subp{\ddot{\mathfrak{T}}}{}{#2}{}{#1}}}
\newrobustcmd{\brevemathfrakTz}[2][]{\ensuremath{\subp{\breve{\mathfrak{T}}}{}{#2}{}{#1}}}
\newrobustcmd{\barmathfrakTz}[2][]{\ensuremath{\subp{\bar{\mathfrak{T}}}{}{#2}{}{#1}}}
\newrobustcmd{\vecmathfrakTz}[2][]{\ensuremath{\subp{\vec{\mathfrak{T}}}{}{#2}{}{#1}}}
\newrobustcmd{\bmmathfrakTz}[2][]{\ensuremath{\subp{\bm{\mathfrak{T}}}{}{#2}{}{#1}}}
\newrobustcmd{\hatbmmathfrakTz}[2][]{\ensuremath{\subp{\hat{\bm{\mathfrak{T}}}}{}{#2}{}{#1}}}
\newrobustcmd{\widehatbmmathfrakTz}[2][]{\ensuremath{\subp{\widehat{\bm{\mathfrak{T}}}}{}{#2}{}{#1}}}
\newrobustcmd{\checkbmmathfrakTz}[2][]{\ensuremath{\subp{\check{\bm{\mathfrak{T}}}}{}{#2}{}{#1}}}
\newrobustcmd{\tildebmmathfrakTz}[2][]{\ensuremath{\subp{\tilde{\bm{\mathfrak{T}}}}{}{#2}{}{#1}}}
\newrobustcmd{\widetildebmmathfrakTz}[2][]{\ensuremath{\subp{\widetilde{\bm{\mathfrak{T}}}}{}{#2}{}{#1}}}
\newrobustcmd{\acutebmmathfrakTz}[2][]{\ensuremath{\subp{\acute{\bm{\mathfrak{T}}}}{}{#2}{}{#1}}}
\newrobustcmd{\gravebmmathfrakTz}[2][]{\ensuremath{\subp{\grave{\bm{\mathfrak{T}}}}{}{#2}{}{#1}}}
\newrobustcmd{\dotbmmathfrakTz}[2][]{\ensuremath{\subp{\dot{\bm{\mathfrak{T}}}}{}{#2}{}{#1}}}
\newrobustcmd{\ddotbmmathfrakTz}[2][]{\ensuremath{\subp{\ddot{\bm{\mathfrak{T}}}}{}{#2}{}{#1}}}
\newrobustcmd{\brevebmmathfrakTz}[2][]{\ensuremath{\subp{\breve{\bm{\mathfrak{T}}}}{}{#2}{}{#1}}}
\newrobustcmd{\barbmmathfrakTz}[2][]{\ensuremath{\subp{\bar{\bm{\mathfrak{T}}}}{}{#2}{}{#1}}}
\newrobustcmd{\vecbmmathfrakTz}[2][]{\ensuremath{\subp{\vec{\bm{\mathfrak{T}}}}{}{#2}{}{#1}}}
\newrobustcmd{\mathfrakUz}[2][]{\ensuremath{\subp{\mathfrak{U}}{}{#2}{}{#1}}}
\newrobustcmd{\hatmathfrakUz}[2][]{\ensuremath{\subp{\hat{\mathfrak{U}}}{}{#2}{}{#1}}}
\newrobustcmd{\widehatmathfrakUz}[2][]{\ensuremath{\subp{\widehat{\mathfrak{U}}}{}{#2}{}{#1}}}
\newrobustcmd{\checkmathfrakUz}[2][]{\ensuremath{\subp{\check{\mathfrak{U}}}{}{#2}{}{#1}}}
\newrobustcmd{\tildemathfrakUz}[2][]{\ensuremath{\subp{\tilde{\mathfrak{U}}}{}{#2}{}{#1}}}
\newrobustcmd{\widetildemathfrakUz}[2][]{\ensuremath{\subp{\widetilde{\mathfrak{U}}}{}{#2}{}{#1}}}
\newrobustcmd{\acutemathfrakUz}[2][]{\ensuremath{\subp{\acute{\mathfrak{U}}}{}{#2}{}{#1}}}
\newrobustcmd{\gravemathfrakUz}[2][]{\ensuremath{\subp{\grave{\mathfrak{U}}}{}{#2}{}{#1}}}
\newrobustcmd{\dotmathfrakUz}[2][]{\ensuremath{\subp{\dot{\mathfrak{U}}}{}{#2}{}{#1}}}
\newrobustcmd{\ddotmathfrakUz}[2][]{\ensuremath{\subp{\ddot{\mathfrak{U}}}{}{#2}{}{#1}}}
\newrobustcmd{\brevemathfrakUz}[2][]{\ensuremath{\subp{\breve{\mathfrak{U}}}{}{#2}{}{#1}}}
\newrobustcmd{\barmathfrakUz}[2][]{\ensuremath{\subp{\bar{\mathfrak{U}}}{}{#2}{}{#1}}}
\newrobustcmd{\vecmathfrakUz}[2][]{\ensuremath{\subp{\vec{\mathfrak{U}}}{}{#2}{}{#1}}}
\newrobustcmd{\bmmathfrakUz}[2][]{\ensuremath{\subp{\bm{\mathfrak{U}}}{}{#2}{}{#1}}}
\newrobustcmd{\hatbmmathfrakUz}[2][]{\ensuremath{\subp{\hat{\bm{\mathfrak{U}}}}{}{#2}{}{#1}}}
\newrobustcmd{\widehatbmmathfrakUz}[2][]{\ensuremath{\subp{\widehat{\bm{\mathfrak{U}}}}{}{#2}{}{#1}}}
\newrobustcmd{\checkbmmathfrakUz}[2][]{\ensuremath{\subp{\check{\bm{\mathfrak{U}}}}{}{#2}{}{#1}}}
\newrobustcmd{\tildebmmathfrakUz}[2][]{\ensuremath{\subp{\tilde{\bm{\mathfrak{U}}}}{}{#2}{}{#1}}}
\newrobustcmd{\widetildebmmathfrakUz}[2][]{\ensuremath{\subp{\widetilde{\bm{\mathfrak{U}}}}{}{#2}{}{#1}}}
\newrobustcmd{\acutebmmathfrakUz}[2][]{\ensuremath{\subp{\acute{\bm{\mathfrak{U}}}}{}{#2}{}{#1}}}
\newrobustcmd{\gravebmmathfrakUz}[2][]{\ensuremath{\subp{\grave{\bm{\mathfrak{U}}}}{}{#2}{}{#1}}}
\newrobustcmd{\dotbmmathfrakUz}[2][]{\ensuremath{\subp{\dot{\bm{\mathfrak{U}}}}{}{#2}{}{#1}}}
\newrobustcmd{\ddotbmmathfrakUz}[2][]{\ensuremath{\subp{\ddot{\bm{\mathfrak{U}}}}{}{#2}{}{#1}}}
\newrobustcmd{\brevebmmathfrakUz}[2][]{\ensuremath{\subp{\breve{\bm{\mathfrak{U}}}}{}{#2}{}{#1}}}
\newrobustcmd{\barbmmathfrakUz}[2][]{\ensuremath{\subp{\bar{\bm{\mathfrak{U}}}}{}{#2}{}{#1}}}
\newrobustcmd{\vecbmmathfrakUz}[2][]{\ensuremath{\subp{\vec{\bm{\mathfrak{U}}}}{}{#2}{}{#1}}}
\newrobustcmd{\mathfrakVz}[2][]{\ensuremath{\subp{\mathfrak{V}}{}{#2}{}{#1}}}
\newrobustcmd{\hatmathfrakVz}[2][]{\ensuremath{\subp{\hat{\mathfrak{V}}}{}{#2}{}{#1}}}
\newrobustcmd{\widehatmathfrakVz}[2][]{\ensuremath{\subp{\widehat{\mathfrak{V}}}{}{#2}{}{#1}}}
\newrobustcmd{\checkmathfrakVz}[2][]{\ensuremath{\subp{\check{\mathfrak{V}}}{}{#2}{}{#1}}}
\newrobustcmd{\tildemathfrakVz}[2][]{\ensuremath{\subp{\tilde{\mathfrak{V}}}{}{#2}{}{#1}}}
\newrobustcmd{\widetildemathfrakVz}[2][]{\ensuremath{\subp{\widetilde{\mathfrak{V}}}{}{#2}{}{#1}}}
\newrobustcmd{\acutemathfrakVz}[2][]{\ensuremath{\subp{\acute{\mathfrak{V}}}{}{#2}{}{#1}}}
\newrobustcmd{\gravemathfrakVz}[2][]{\ensuremath{\subp{\grave{\mathfrak{V}}}{}{#2}{}{#1}}}
\newrobustcmd{\dotmathfrakVz}[2][]{\ensuremath{\subp{\dot{\mathfrak{V}}}{}{#2}{}{#1}}}
\newrobustcmd{\ddotmathfrakVz}[2][]{\ensuremath{\subp{\ddot{\mathfrak{V}}}{}{#2}{}{#1}}}
\newrobustcmd{\brevemathfrakVz}[2][]{\ensuremath{\subp{\breve{\mathfrak{V}}}{}{#2}{}{#1}}}
\newrobustcmd{\barmathfrakVz}[2][]{\ensuremath{\subp{\bar{\mathfrak{V}}}{}{#2}{}{#1}}}
\newrobustcmd{\vecmathfrakVz}[2][]{\ensuremath{\subp{\vec{\mathfrak{V}}}{}{#2}{}{#1}}}
\newrobustcmd{\bmmathfrakVz}[2][]{\ensuremath{\subp{\bm{\mathfrak{V}}}{}{#2}{}{#1}}}
\newrobustcmd{\hatbmmathfrakVz}[2][]{\ensuremath{\subp{\hat{\bm{\mathfrak{V}}}}{}{#2}{}{#1}}}
\newrobustcmd{\widehatbmmathfrakVz}[2][]{\ensuremath{\subp{\widehat{\bm{\mathfrak{V}}}}{}{#2}{}{#1}}}
\newrobustcmd{\checkbmmathfrakVz}[2][]{\ensuremath{\subp{\check{\bm{\mathfrak{V}}}}{}{#2}{}{#1}}}
\newrobustcmd{\tildebmmathfrakVz}[2][]{\ensuremath{\subp{\tilde{\bm{\mathfrak{V}}}}{}{#2}{}{#1}}}
\newrobustcmd{\widetildebmmathfrakVz}[2][]{\ensuremath{\subp{\widetilde{\bm{\mathfrak{V}}}}{}{#2}{}{#1}}}
\newrobustcmd{\acutebmmathfrakVz}[2][]{\ensuremath{\subp{\acute{\bm{\mathfrak{V}}}}{}{#2}{}{#1}}}
\newrobustcmd{\gravebmmathfrakVz}[2][]{\ensuremath{\subp{\grave{\bm{\mathfrak{V}}}}{}{#2}{}{#1}}}
\newrobustcmd{\dotbmmathfrakVz}[2][]{\ensuremath{\subp{\dot{\bm{\mathfrak{V}}}}{}{#2}{}{#1}}}
\newrobustcmd{\ddotbmmathfrakVz}[2][]{\ensuremath{\subp{\ddot{\bm{\mathfrak{V}}}}{}{#2}{}{#1}}}
\newrobustcmd{\brevebmmathfrakVz}[2][]{\ensuremath{\subp{\breve{\bm{\mathfrak{V}}}}{}{#2}{}{#1}}}
\newrobustcmd{\barbmmathfrakVz}[2][]{\ensuremath{\subp{\bar{\bm{\mathfrak{V}}}}{}{#2}{}{#1}}}
\newrobustcmd{\vecbmmathfrakVz}[2][]{\ensuremath{\subp{\vec{\bm{\mathfrak{V}}}}{}{#2}{}{#1}}}
\newrobustcmd{\mathfrakWz}[2][]{\ensuremath{\subp{\mathfrak{W}}{}{#2}{}{#1}}}
\newrobustcmd{\hatmathfrakWz}[2][]{\ensuremath{\subp{\hat{\mathfrak{W}}}{}{#2}{}{#1}}}
\newrobustcmd{\widehatmathfrakWz}[2][]{\ensuremath{\subp{\widehat{\mathfrak{W}}}{}{#2}{}{#1}}}
\newrobustcmd{\checkmathfrakWz}[2][]{\ensuremath{\subp{\check{\mathfrak{W}}}{}{#2}{}{#1}}}
\newrobustcmd{\tildemathfrakWz}[2][]{\ensuremath{\subp{\tilde{\mathfrak{W}}}{}{#2}{}{#1}}}
\newrobustcmd{\widetildemathfrakWz}[2][]{\ensuremath{\subp{\widetilde{\mathfrak{W}}}{}{#2}{}{#1}}}
\newrobustcmd{\acutemathfrakWz}[2][]{\ensuremath{\subp{\acute{\mathfrak{W}}}{}{#2}{}{#1}}}
\newrobustcmd{\gravemathfrakWz}[2][]{\ensuremath{\subp{\grave{\mathfrak{W}}}{}{#2}{}{#1}}}
\newrobustcmd{\dotmathfrakWz}[2][]{\ensuremath{\subp{\dot{\mathfrak{W}}}{}{#2}{}{#1}}}
\newrobustcmd{\ddotmathfrakWz}[2][]{\ensuremath{\subp{\ddot{\mathfrak{W}}}{}{#2}{}{#1}}}
\newrobustcmd{\brevemathfrakWz}[2][]{\ensuremath{\subp{\breve{\mathfrak{W}}}{}{#2}{}{#1}}}
\newrobustcmd{\barmathfrakWz}[2][]{\ensuremath{\subp{\bar{\mathfrak{W}}}{}{#2}{}{#1}}}
\newrobustcmd{\vecmathfrakWz}[2][]{\ensuremath{\subp{\vec{\mathfrak{W}}}{}{#2}{}{#1}}}
\newrobustcmd{\bmmathfrakWz}[2][]{\ensuremath{\subp{\bm{\mathfrak{W}}}{}{#2}{}{#1}}}
\newrobustcmd{\hatbmmathfrakWz}[2][]{\ensuremath{\subp{\hat{\bm{\mathfrak{W}}}}{}{#2}{}{#1}}}
\newrobustcmd{\widehatbmmathfrakWz}[2][]{\ensuremath{\subp{\widehat{\bm{\mathfrak{W}}}}{}{#2}{}{#1}}}
\newrobustcmd{\checkbmmathfrakWz}[2][]{\ensuremath{\subp{\check{\bm{\mathfrak{W}}}}{}{#2}{}{#1}}}
\newrobustcmd{\tildebmmathfrakWz}[2][]{\ensuremath{\subp{\tilde{\bm{\mathfrak{W}}}}{}{#2}{}{#1}}}
\newrobustcmd{\widetildebmmathfrakWz}[2][]{\ensuremath{\subp{\widetilde{\bm{\mathfrak{W}}}}{}{#2}{}{#1}}}
\newrobustcmd{\acutebmmathfrakWz}[2][]{\ensuremath{\subp{\acute{\bm{\mathfrak{W}}}}{}{#2}{}{#1}}}
\newrobustcmd{\gravebmmathfrakWz}[2][]{\ensuremath{\subp{\grave{\bm{\mathfrak{W}}}}{}{#2}{}{#1}}}
\newrobustcmd{\dotbmmathfrakWz}[2][]{\ensuremath{\subp{\dot{\bm{\mathfrak{W}}}}{}{#2}{}{#1}}}
\newrobustcmd{\ddotbmmathfrakWz}[2][]{\ensuremath{\subp{\ddot{\bm{\mathfrak{W}}}}{}{#2}{}{#1}}}
\newrobustcmd{\brevebmmathfrakWz}[2][]{\ensuremath{\subp{\breve{\bm{\mathfrak{W}}}}{}{#2}{}{#1}}}
\newrobustcmd{\barbmmathfrakWz}[2][]{\ensuremath{\subp{\bar{\bm{\mathfrak{W}}}}{}{#2}{}{#1}}}
\newrobustcmd{\vecbmmathfrakWz}[2][]{\ensuremath{\subp{\vec{\bm{\mathfrak{W}}}}{}{#2}{}{#1}}}
\newrobustcmd{\mathfrakXz}[2][]{\ensuremath{\subp{\mathfrak{X}}{}{#2}{}{#1}}}
\newrobustcmd{\hatmathfrakXz}[2][]{\ensuremath{\subp{\hat{\mathfrak{X}}}{}{#2}{}{#1}}}
\newrobustcmd{\widehatmathfrakXz}[2][]{\ensuremath{\subp{\widehat{\mathfrak{X}}}{}{#2}{}{#1}}}
\newrobustcmd{\checkmathfrakXz}[2][]{\ensuremath{\subp{\check{\mathfrak{X}}}{}{#2}{}{#1}}}
\newrobustcmd{\tildemathfrakXz}[2][]{\ensuremath{\subp{\tilde{\mathfrak{X}}}{}{#2}{}{#1}}}
\newrobustcmd{\widetildemathfrakXz}[2][]{\ensuremath{\subp{\widetilde{\mathfrak{X}}}{}{#2}{}{#1}}}
\newrobustcmd{\acutemathfrakXz}[2][]{\ensuremath{\subp{\acute{\mathfrak{X}}}{}{#2}{}{#1}}}
\newrobustcmd{\gravemathfrakXz}[2][]{\ensuremath{\subp{\grave{\mathfrak{X}}}{}{#2}{}{#1}}}
\newrobustcmd{\dotmathfrakXz}[2][]{\ensuremath{\subp{\dot{\mathfrak{X}}}{}{#2}{}{#1}}}
\newrobustcmd{\ddotmathfrakXz}[2][]{\ensuremath{\subp{\ddot{\mathfrak{X}}}{}{#2}{}{#1}}}
\newrobustcmd{\brevemathfrakXz}[2][]{\ensuremath{\subp{\breve{\mathfrak{X}}}{}{#2}{}{#1}}}
\newrobustcmd{\barmathfrakXz}[2][]{\ensuremath{\subp{\bar{\mathfrak{X}}}{}{#2}{}{#1}}}
\newrobustcmd{\vecmathfrakXz}[2][]{\ensuremath{\subp{\vec{\mathfrak{X}}}{}{#2}{}{#1}}}
\newrobustcmd{\bmmathfrakXz}[2][]{\ensuremath{\subp{\bm{\mathfrak{X}}}{}{#2}{}{#1}}}
\newrobustcmd{\hatbmmathfrakXz}[2][]{\ensuremath{\subp{\hat{\bm{\mathfrak{X}}}}{}{#2}{}{#1}}}
\newrobustcmd{\widehatbmmathfrakXz}[2][]{\ensuremath{\subp{\widehat{\bm{\mathfrak{X}}}}{}{#2}{}{#1}}}
\newrobustcmd{\checkbmmathfrakXz}[2][]{\ensuremath{\subp{\check{\bm{\mathfrak{X}}}}{}{#2}{}{#1}}}
\newrobustcmd{\tildebmmathfrakXz}[2][]{\ensuremath{\subp{\tilde{\bm{\mathfrak{X}}}}{}{#2}{}{#1}}}
\newrobustcmd{\widetildebmmathfrakXz}[2][]{\ensuremath{\subp{\widetilde{\bm{\mathfrak{X}}}}{}{#2}{}{#1}}}
\newrobustcmd{\acutebmmathfrakXz}[2][]{\ensuremath{\subp{\acute{\bm{\mathfrak{X}}}}{}{#2}{}{#1}}}
\newrobustcmd{\gravebmmathfrakXz}[2][]{\ensuremath{\subp{\grave{\bm{\mathfrak{X}}}}{}{#2}{}{#1}}}
\newrobustcmd{\dotbmmathfrakXz}[2][]{\ensuremath{\subp{\dot{\bm{\mathfrak{X}}}}{}{#2}{}{#1}}}
\newrobustcmd{\ddotbmmathfrakXz}[2][]{\ensuremath{\subp{\ddot{\bm{\mathfrak{X}}}}{}{#2}{}{#1}}}
\newrobustcmd{\brevebmmathfrakXz}[2][]{\ensuremath{\subp{\breve{\bm{\mathfrak{X}}}}{}{#2}{}{#1}}}
\newrobustcmd{\barbmmathfrakXz}[2][]{\ensuremath{\subp{\bar{\bm{\mathfrak{X}}}}{}{#2}{}{#1}}}
\newrobustcmd{\vecbmmathfrakXz}[2][]{\ensuremath{\subp{\vec{\bm{\mathfrak{X}}}}{}{#2}{}{#1}}}
\newrobustcmd{\mathfrakYz}[2][]{\ensuremath{\subp{\mathfrak{Y}}{}{#2}{}{#1}}}
\newrobustcmd{\hatmathfrakYz}[2][]{\ensuremath{\subp{\hat{\mathfrak{Y}}}{}{#2}{}{#1}}}
\newrobustcmd{\widehatmathfrakYz}[2][]{\ensuremath{\subp{\widehat{\mathfrak{Y}}}{}{#2}{}{#1}}}
\newrobustcmd{\checkmathfrakYz}[2][]{\ensuremath{\subp{\check{\mathfrak{Y}}}{}{#2}{}{#1}}}
\newrobustcmd{\tildemathfrakYz}[2][]{\ensuremath{\subp{\tilde{\mathfrak{Y}}}{}{#2}{}{#1}}}
\newrobustcmd{\widetildemathfrakYz}[2][]{\ensuremath{\subp{\widetilde{\mathfrak{Y}}}{}{#2}{}{#1}}}
\newrobustcmd{\acutemathfrakYz}[2][]{\ensuremath{\subp{\acute{\mathfrak{Y}}}{}{#2}{}{#1}}}
\newrobustcmd{\gravemathfrakYz}[2][]{\ensuremath{\subp{\grave{\mathfrak{Y}}}{}{#2}{}{#1}}}
\newrobustcmd{\dotmathfrakYz}[2][]{\ensuremath{\subp{\dot{\mathfrak{Y}}}{}{#2}{}{#1}}}
\newrobustcmd{\ddotmathfrakYz}[2][]{\ensuremath{\subp{\ddot{\mathfrak{Y}}}{}{#2}{}{#1}}}
\newrobustcmd{\brevemathfrakYz}[2][]{\ensuremath{\subp{\breve{\mathfrak{Y}}}{}{#2}{}{#1}}}
\newrobustcmd{\barmathfrakYz}[2][]{\ensuremath{\subp{\bar{\mathfrak{Y}}}{}{#2}{}{#1}}}
\newrobustcmd{\vecmathfrakYz}[2][]{\ensuremath{\subp{\vec{\mathfrak{Y}}}{}{#2}{}{#1}}}
\newrobustcmd{\bmmathfrakYz}[2][]{\ensuremath{\subp{\bm{\mathfrak{Y}}}{}{#2}{}{#1}}}
\newrobustcmd{\hatbmmathfrakYz}[2][]{\ensuremath{\subp{\hat{\bm{\mathfrak{Y}}}}{}{#2}{}{#1}}}
\newrobustcmd{\widehatbmmathfrakYz}[2][]{\ensuremath{\subp{\widehat{\bm{\mathfrak{Y}}}}{}{#2}{}{#1}}}
\newrobustcmd{\checkbmmathfrakYz}[2][]{\ensuremath{\subp{\check{\bm{\mathfrak{Y}}}}{}{#2}{}{#1}}}
\newrobustcmd{\tildebmmathfrakYz}[2][]{\ensuremath{\subp{\tilde{\bm{\mathfrak{Y}}}}{}{#2}{}{#1}}}
\newrobustcmd{\widetildebmmathfrakYz}[2][]{\ensuremath{\subp{\widetilde{\bm{\mathfrak{Y}}}}{}{#2}{}{#1}}}
\newrobustcmd{\acutebmmathfrakYz}[2][]{\ensuremath{\subp{\acute{\bm{\mathfrak{Y}}}}{}{#2}{}{#1}}}
\newrobustcmd{\gravebmmathfrakYz}[2][]{\ensuremath{\subp{\grave{\bm{\mathfrak{Y}}}}{}{#2}{}{#1}}}
\newrobustcmd{\dotbmmathfrakYz}[2][]{\ensuremath{\subp{\dot{\bm{\mathfrak{Y}}}}{}{#2}{}{#1}}}
\newrobustcmd{\ddotbmmathfrakYz}[2][]{\ensuremath{\subp{\ddot{\bm{\mathfrak{Y}}}}{}{#2}{}{#1}}}
\newrobustcmd{\brevebmmathfrakYz}[2][]{\ensuremath{\subp{\breve{\bm{\mathfrak{Y}}}}{}{#2}{}{#1}}}
\newrobustcmd{\barbmmathfrakYz}[2][]{\ensuremath{\subp{\bar{\bm{\mathfrak{Y}}}}{}{#2}{}{#1}}}
\newrobustcmd{\vecbmmathfrakYz}[2][]{\ensuremath{\subp{\vec{\bm{\mathfrak{Y}}}}{}{#2}{}{#1}}}
\newrobustcmd{\mathfrakZz}[2][]{\ensuremath{\subp{\mathfrak{Z}}{}{#2}{}{#1}}}
\newrobustcmd{\hatmathfrakZz}[2][]{\ensuremath{\subp{\hat{\mathfrak{Z}}}{}{#2}{}{#1}}}
\newrobustcmd{\widehatmathfrakZz}[2][]{\ensuremath{\subp{\widehat{\mathfrak{Z}}}{}{#2}{}{#1}}}
\newrobustcmd{\checkmathfrakZz}[2][]{\ensuremath{\subp{\check{\mathfrak{Z}}}{}{#2}{}{#1}}}
\newrobustcmd{\tildemathfrakZz}[2][]{\ensuremath{\subp{\tilde{\mathfrak{Z}}}{}{#2}{}{#1}}}
\newrobustcmd{\widetildemathfrakZz}[2][]{\ensuremath{\subp{\widetilde{\mathfrak{Z}}}{}{#2}{}{#1}}}
\newrobustcmd{\acutemathfrakZz}[2][]{\ensuremath{\subp{\acute{\mathfrak{Z}}}{}{#2}{}{#1}}}
\newrobustcmd{\gravemathfrakZz}[2][]{\ensuremath{\subp{\grave{\mathfrak{Z}}}{}{#2}{}{#1}}}
\newrobustcmd{\dotmathfrakZz}[2][]{\ensuremath{\subp{\dot{\mathfrak{Z}}}{}{#2}{}{#1}}}
\newrobustcmd{\ddotmathfrakZz}[2][]{\ensuremath{\subp{\ddot{\mathfrak{Z}}}{}{#2}{}{#1}}}
\newrobustcmd{\brevemathfrakZz}[2][]{\ensuremath{\subp{\breve{\mathfrak{Z}}}{}{#2}{}{#1}}}
\newrobustcmd{\barmathfrakZz}[2][]{\ensuremath{\subp{\bar{\mathfrak{Z}}}{}{#2}{}{#1}}}
\newrobustcmd{\vecmathfrakZz}[2][]{\ensuremath{\subp{\vec{\mathfrak{Z}}}{}{#2}{}{#1}}}
\newrobustcmd{\bmmathfrakZz}[2][]{\ensuremath{\subp{\bm{\mathfrak{Z}}}{}{#2}{}{#1}}}
\newrobustcmd{\hatbmmathfrakZz}[2][]{\ensuremath{\subp{\hat{\bm{\mathfrak{Z}}}}{}{#2}{}{#1}}}
\newrobustcmd{\widehatbmmathfrakZz}[2][]{\ensuremath{\subp{\widehat{\bm{\mathfrak{Z}}}}{}{#2}{}{#1}}}
\newrobustcmd{\checkbmmathfrakZz}[2][]{\ensuremath{\subp{\check{\bm{\mathfrak{Z}}}}{}{#2}{}{#1}}}
\newrobustcmd{\tildebmmathfrakZz}[2][]{\ensuremath{\subp{\tilde{\bm{\mathfrak{Z}}}}{}{#2}{}{#1}}}
\newrobustcmd{\widetildebmmathfrakZz}[2][]{\ensuremath{\subp{\widetilde{\bm{\mathfrak{Z}}}}{}{#2}{}{#1}}}
\newrobustcmd{\acutebmmathfrakZz}[2][]{\ensuremath{\subp{\acute{\bm{\mathfrak{Z}}}}{}{#2}{}{#1}}}
\newrobustcmd{\gravebmmathfrakZz}[2][]{\ensuremath{\subp{\grave{\bm{\mathfrak{Z}}}}{}{#2}{}{#1}}}
\newrobustcmd{\dotbmmathfrakZz}[2][]{\ensuremath{\subp{\dot{\bm{\mathfrak{Z}}}}{}{#2}{}{#1}}}
\newrobustcmd{\ddotbmmathfrakZz}[2][]{\ensuremath{\subp{\ddot{\bm{\mathfrak{Z}}}}{}{#2}{}{#1}}}
\newrobustcmd{\brevebmmathfrakZz}[2][]{\ensuremath{\subp{\breve{\bm{\mathfrak{Z}}}}{}{#2}{}{#1}}}
\newrobustcmd{\barbmmathfrakZz}[2][]{\ensuremath{\subp{\bar{\bm{\mathfrak{Z}}}}{}{#2}{}{#1}}}
\newrobustcmd{\vecbmmathfrakZz}[2][]{\ensuremath{\subp{\vec{\bm{\mathfrak{Z}}}}{}{#2}{}{#1}}}

\newrobustcmd{\mathcalAz}[2][]{\ensuremath{\subp{\mathcal{A}}{}{#2}{}{#1}}}
\newrobustcmd{\hatmathcalAz}[2][]{\ensuremath{\subp{\hat{\mathcal{A}}}{}{#2}{}{#1}}}
\newrobustcmd{\widehatmathcalAz}[2][]{\ensuremath{\subp{\widehat{\mathcal{A}}}{}{#2}{}{#1}}}
\newrobustcmd{\checkmathcalAz}[2][]{\ensuremath{\subp{\check{\mathcal{A}}}{}{#2}{}{#1}}}
\newrobustcmd{\tildemathcalAz}[2][]{\ensuremath{\subp{\tilde{\mathcal{A}}}{}{#2}{}{#1}}}
\newrobustcmd{\widetildemathcalAz}[2][]{\ensuremath{\subp{\widetilde{\mathcal{A}}}{}{#2}{}{#1}}}
\newrobustcmd{\acutemathcalAz}[2][]{\ensuremath{\subp{\acute{\mathcal{A}}}{}{#2}{}{#1}}}
\newrobustcmd{\gravemathcalAz}[2][]{\ensuremath{\subp{\grave{\mathcal{A}}}{}{#2}{}{#1}}}
\newrobustcmd{\dotmathcalAz}[2][]{\ensuremath{\subp{\dot{\mathcal{A}}}{}{#2}{}{#1}}}
\newrobustcmd{\ddotmathcalAz}[2][]{\ensuremath{\subp{\ddot{\mathcal{A}}}{}{#2}{}{#1}}}
\newrobustcmd{\brevemathcalAz}[2][]{\ensuremath{\subp{\breve{\mathcal{A}}}{}{#2}{}{#1}}}
\newrobustcmd{\barmathcalAz}[2][]{\ensuremath{\subp{\bar{\mathcal{A}}}{}{#2}{}{#1}}}
\newrobustcmd{\vecmathcalAz}[2][]{\ensuremath{\subp{\vec{\mathcal{A}}}{}{#2}{}{#1}}}
\newrobustcmd{\bmmathcalAz}[2][]{\ensuremath{\subp{\bm{\mathcal{A}}}{}{#2}{}{#1}}}
\newrobustcmd{\hatbmmathcalAz}[2][]{\ensuremath{\subp{\hat{\bm{\mathcal{A}}}}{}{#2}{}{#1}}}
\newrobustcmd{\widehatbmmathcalAz}[2][]{\ensuremath{\subp{\widehat{\bm{\mathcal{A}}}}{}{#2}{}{#1}}}
\newrobustcmd{\checkbmmathcalAz}[2][]{\ensuremath{\subp{\check{\bm{\mathcal{A}}}}{}{#2}{}{#1}}}
\newrobustcmd{\tildebmmathcalAz}[2][]{\ensuremath{\subp{\tilde{\bm{\mathcal{A}}}}{}{#2}{}{#1}}}
\newrobustcmd{\widetildebmmathcalAz}[2][]{\ensuremath{\subp{\widetilde{\bm{\mathcal{A}}}}{}{#2}{}{#1}}}
\newrobustcmd{\acutebmmathcalAz}[2][]{\ensuremath{\subp{\acute{\bm{\mathcal{A}}}}{}{#2}{}{#1}}}
\newrobustcmd{\gravebmmathcalAz}[2][]{\ensuremath{\subp{\grave{\bm{\mathcal{A}}}}{}{#2}{}{#1}}}
\newrobustcmd{\dotbmmathcalAz}[2][]{\ensuremath{\subp{\dot{\bm{\mathcal{A}}}}{}{#2}{}{#1}}}
\newrobustcmd{\ddotbmmathcalAz}[2][]{\ensuremath{\subp{\ddot{\bm{\mathcal{A}}}}{}{#2}{}{#1}}}
\newrobustcmd{\brevebmmathcalAz}[2][]{\ensuremath{\subp{\breve{\bm{\mathcal{A}}}}{}{#2}{}{#1}}}
\newrobustcmd{\barbmmathcalAz}[2][]{\ensuremath{\subp{\bar{\bm{\mathcal{A}}}}{}{#2}{}{#1}}}
\newrobustcmd{\vecbmmathcalAz}[2][]{\ensuremath{\subp{\vec{\bm{\mathcal{A}}}}{}{#2}{}{#1}}}
\newrobustcmd{\mathcalBz}[2][]{\ensuremath{\subp{\mathcal{B}}{}{#2}{}{#1}}}
\newrobustcmd{\hatmathcalBz}[2][]{\ensuremath{\subp{\hat{\mathcal{B}}}{}{#2}{}{#1}}}
\newrobustcmd{\widehatmathcalBz}[2][]{\ensuremath{\subp{\widehat{\mathcal{B}}}{}{#2}{}{#1}}}
\newrobustcmd{\checkmathcalBz}[2][]{\ensuremath{\subp{\check{\mathcal{B}}}{}{#2}{}{#1}}}
\newrobustcmd{\tildemathcalBz}[2][]{\ensuremath{\subp{\tilde{\mathcal{B}}}{}{#2}{}{#1}}}
\newrobustcmd{\widetildemathcalBz}[2][]{\ensuremath{\subp{\widetilde{\mathcal{B}}}{}{#2}{}{#1}}}
\newrobustcmd{\acutemathcalBz}[2][]{\ensuremath{\subp{\acute{\mathcal{B}}}{}{#2}{}{#1}}}
\newrobustcmd{\gravemathcalBz}[2][]{\ensuremath{\subp{\grave{\mathcal{B}}}{}{#2}{}{#1}}}
\newrobustcmd{\dotmathcalBz}[2][]{\ensuremath{\subp{\dot{\mathcal{B}}}{}{#2}{}{#1}}}
\newrobustcmd{\ddotmathcalBz}[2][]{\ensuremath{\subp{\ddot{\mathcal{B}}}{}{#2}{}{#1}}}
\newrobustcmd{\brevemathcalBz}[2][]{\ensuremath{\subp{\breve{\mathcal{B}}}{}{#2}{}{#1}}}
\newrobustcmd{\barmathcalBz}[2][]{\ensuremath{\subp{\bar{\mathcal{B}}}{}{#2}{}{#1}}}
\newrobustcmd{\vecmathcalBz}[2][]{\ensuremath{\subp{\vec{\mathcal{B}}}{}{#2}{}{#1}}}
\newrobustcmd{\bmmathcalBz}[2][]{\ensuremath{\subp{\bm{\mathcal{B}}}{}{#2}{}{#1}}}
\newrobustcmd{\hatbmmathcalBz}[2][]{\ensuremath{\subp{\hat{\bm{\mathcal{B}}}}{}{#2}{}{#1}}}
\newrobustcmd{\widehatbmmathcalBz}[2][]{\ensuremath{\subp{\widehat{\bm{\mathcal{B}}}}{}{#2}{}{#1}}}
\newrobustcmd{\checkbmmathcalBz}[2][]{\ensuremath{\subp{\check{\bm{\mathcal{B}}}}{}{#2}{}{#1}}}
\newrobustcmd{\tildebmmathcalBz}[2][]{\ensuremath{\subp{\tilde{\bm{\mathcal{B}}}}{}{#2}{}{#1}}}
\newrobustcmd{\widetildebmmathcalBz}[2][]{\ensuremath{\subp{\widetilde{\bm{\mathcal{B}}}}{}{#2}{}{#1}}}
\newrobustcmd{\acutebmmathcalBz}[2][]{\ensuremath{\subp{\acute{\bm{\mathcal{B}}}}{}{#2}{}{#1}}}
\newrobustcmd{\gravebmmathcalBz}[2][]{\ensuremath{\subp{\grave{\bm{\mathcal{B}}}}{}{#2}{}{#1}}}
\newrobustcmd{\dotbmmathcalBz}[2][]{\ensuremath{\subp{\dot{\bm{\mathcal{B}}}}{}{#2}{}{#1}}}
\newrobustcmd{\ddotbmmathcalBz}[2][]{\ensuremath{\subp{\ddot{\bm{\mathcal{B}}}}{}{#2}{}{#1}}}
\newrobustcmd{\brevebmmathcalBz}[2][]{\ensuremath{\subp{\breve{\bm{\mathcal{B}}}}{}{#2}{}{#1}}}
\newrobustcmd{\barbmmathcalBz}[2][]{\ensuremath{\subp{\bar{\bm{\mathcal{B}}}}{}{#2}{}{#1}}}
\newrobustcmd{\vecbmmathcalBz}[2][]{\ensuremath{\subp{\vec{\bm{\mathcal{B}}}}{}{#2}{}{#1}}}
\newrobustcmd{\mathcalCz}[2][]{\ensuremath{\subp{\mathcal{C}}{}{#2}{}{#1}}}
\newrobustcmd{\hatmathcalCz}[2][]{\ensuremath{\subp{\hat{\mathcal{C}}}{}{#2}{}{#1}}}
\newrobustcmd{\widehatmathcalCz}[2][]{\ensuremath{\subp{\widehat{\mathcal{C}}}{}{#2}{}{#1}}}
\newrobustcmd{\checkmathcalCz}[2][]{\ensuremath{\subp{\check{\mathcal{C}}}{}{#2}{}{#1}}}
\newrobustcmd{\tildemathcalCz}[2][]{\ensuremath{\subp{\tilde{\mathcal{C}}}{}{#2}{}{#1}}}
\newrobustcmd{\widetildemathcalCz}[2][]{\ensuremath{\subp{\widetilde{\mathcal{C}}}{}{#2}{}{#1}}}
\newrobustcmd{\acutemathcalCz}[2][]{\ensuremath{\subp{\acute{\mathcal{C}}}{}{#2}{}{#1}}}
\newrobustcmd{\gravemathcalCz}[2][]{\ensuremath{\subp{\grave{\mathcal{C}}}{}{#2}{}{#1}}}
\newrobustcmd{\dotmathcalCz}[2][]{\ensuremath{\subp{\dot{\mathcal{C}}}{}{#2}{}{#1}}}
\newrobustcmd{\ddotmathcalCz}[2][]{\ensuremath{\subp{\ddot{\mathcal{C}}}{}{#2}{}{#1}}}
\newrobustcmd{\brevemathcalCz}[2][]{\ensuremath{\subp{\breve{\mathcal{C}}}{}{#2}{}{#1}}}
\newrobustcmd{\barmathcalCz}[2][]{\ensuremath{\subp{\bar{\mathcal{C}}}{}{#2}{}{#1}}}
\newrobustcmd{\vecmathcalCz}[2][]{\ensuremath{\subp{\vec{\mathcal{C}}}{}{#2}{}{#1}}}
\newrobustcmd{\bmmathcalCz}[2][]{\ensuremath{\subp{\bm{\mathcal{C}}}{}{#2}{}{#1}}}
\newrobustcmd{\hatbmmathcalCz}[2][]{\ensuremath{\subp{\hat{\bm{\mathcal{C}}}}{}{#2}{}{#1}}}
\newrobustcmd{\widehatbmmathcalCz}[2][]{\ensuremath{\subp{\widehat{\bm{\mathcal{C}}}}{}{#2}{}{#1}}}
\newrobustcmd{\checkbmmathcalCz}[2][]{\ensuremath{\subp{\check{\bm{\mathcal{C}}}}{}{#2}{}{#1}}}
\newrobustcmd{\tildebmmathcalCz}[2][]{\ensuremath{\subp{\tilde{\bm{\mathcal{C}}}}{}{#2}{}{#1}}}
\newrobustcmd{\widetildebmmathcalCz}[2][]{\ensuremath{\subp{\widetilde{\bm{\mathcal{C}}}}{}{#2}{}{#1}}}
\newrobustcmd{\acutebmmathcalCz}[2][]{\ensuremath{\subp{\acute{\bm{\mathcal{C}}}}{}{#2}{}{#1}}}
\newrobustcmd{\gravebmmathcalCz}[2][]{\ensuremath{\subp{\grave{\bm{\mathcal{C}}}}{}{#2}{}{#1}}}
\newrobustcmd{\dotbmmathcalCz}[2][]{\ensuremath{\subp{\dot{\bm{\mathcal{C}}}}{}{#2}{}{#1}}}
\newrobustcmd{\ddotbmmathcalCz}[2][]{\ensuremath{\subp{\ddot{\bm{\mathcal{C}}}}{}{#2}{}{#1}}}
\newrobustcmd{\brevebmmathcalCz}[2][]{\ensuremath{\subp{\breve{\bm{\mathcal{C}}}}{}{#2}{}{#1}}}
\newrobustcmd{\barbmmathcalCz}[2][]{\ensuremath{\subp{\bar{\bm{\mathcal{C}}}}{}{#2}{}{#1}}}
\newrobustcmd{\vecbmmathcalCz}[2][]{\ensuremath{\subp{\vec{\bm{\mathcal{C}}}}{}{#2}{}{#1}}}
\newrobustcmd{\mathcalDz}[2][]{\ensuremath{\subp{\mathcal{D}}{}{#2}{}{#1}}}
\newrobustcmd{\hatmathcalDz}[2][]{\ensuremath{\subp{\hat{\mathcal{D}}}{}{#2}{}{#1}}}
\newrobustcmd{\widehatmathcalDz}[2][]{\ensuremath{\subp{\widehat{\mathcal{D}}}{}{#2}{}{#1}}}
\newrobustcmd{\checkmathcalDz}[2][]{\ensuremath{\subp{\check{\mathcal{D}}}{}{#2}{}{#1}}}
\newrobustcmd{\tildemathcalDz}[2][]{\ensuremath{\subp{\tilde{\mathcal{D}}}{}{#2}{}{#1}}}
\newrobustcmd{\widetildemathcalDz}[2][]{\ensuremath{\subp{\widetilde{\mathcal{D}}}{}{#2}{}{#1}}}
\newrobustcmd{\acutemathcalDz}[2][]{\ensuremath{\subp{\acute{\mathcal{D}}}{}{#2}{}{#1}}}
\newrobustcmd{\gravemathcalDz}[2][]{\ensuremath{\subp{\grave{\mathcal{D}}}{}{#2}{}{#1}}}
\newrobustcmd{\dotmathcalDz}[2][]{\ensuremath{\subp{\dot{\mathcal{D}}}{}{#2}{}{#1}}}
\newrobustcmd{\ddotmathcalDz}[2][]{\ensuremath{\subp{\ddot{\mathcal{D}}}{}{#2}{}{#1}}}
\newrobustcmd{\brevemathcalDz}[2][]{\ensuremath{\subp{\breve{\mathcal{D}}}{}{#2}{}{#1}}}
\newrobustcmd{\barmathcalDz}[2][]{\ensuremath{\subp{\bar{\mathcal{D}}}{}{#2}{}{#1}}}
\newrobustcmd{\vecmathcalDz}[2][]{\ensuremath{\subp{\vec{\mathcal{D}}}{}{#2}{}{#1}}}
\newrobustcmd{\bmmathcalDz}[2][]{\ensuremath{\subp{\bm{\mathcal{D}}}{}{#2}{}{#1}}}
\newrobustcmd{\hatbmmathcalDz}[2][]{\ensuremath{\subp{\hat{\bm{\mathcal{D}}}}{}{#2}{}{#1}}}
\newrobustcmd{\widehatbmmathcalDz}[2][]{\ensuremath{\subp{\widehat{\bm{\mathcal{D}}}}{}{#2}{}{#1}}}
\newrobustcmd{\checkbmmathcalDz}[2][]{\ensuremath{\subp{\check{\bm{\mathcal{D}}}}{}{#2}{}{#1}}}
\newrobustcmd{\tildebmmathcalDz}[2][]{\ensuremath{\subp{\tilde{\bm{\mathcal{D}}}}{}{#2}{}{#1}}}
\newrobustcmd{\widetildebmmathcalDz}[2][]{\ensuremath{\subp{\widetilde{\bm{\mathcal{D}}}}{}{#2}{}{#1}}}
\newrobustcmd{\acutebmmathcalDz}[2][]{\ensuremath{\subp{\acute{\bm{\mathcal{D}}}}{}{#2}{}{#1}}}
\newrobustcmd{\gravebmmathcalDz}[2][]{\ensuremath{\subp{\grave{\bm{\mathcal{D}}}}{}{#2}{}{#1}}}
\newrobustcmd{\dotbmmathcalDz}[2][]{\ensuremath{\subp{\dot{\bm{\mathcal{D}}}}{}{#2}{}{#1}}}
\newrobustcmd{\ddotbmmathcalDz}[2][]{\ensuremath{\subp{\ddot{\bm{\mathcal{D}}}}{}{#2}{}{#1}}}
\newrobustcmd{\brevebmmathcalDz}[2][]{\ensuremath{\subp{\breve{\bm{\mathcal{D}}}}{}{#2}{}{#1}}}
\newrobustcmd{\barbmmathcalDz}[2][]{\ensuremath{\subp{\bar{\bm{\mathcal{D}}}}{}{#2}{}{#1}}}
\newrobustcmd{\vecbmmathcalDz}[2][]{\ensuremath{\subp{\vec{\bm{\mathcal{D}}}}{}{#2}{}{#1}}}
\newrobustcmd{\mathcalEz}[2][]{\ensuremath{\subp{\mathcal{E}}{}{#2}{}{#1}}}
\newrobustcmd{\hatmathcalEz}[2][]{\ensuremath{\subp{\hat{\mathcal{E}}}{}{#2}{}{#1}}}
\newrobustcmd{\widehatmathcalEz}[2][]{\ensuremath{\subp{\widehat{\mathcal{E}}}{}{#2}{}{#1}}}
\newrobustcmd{\checkmathcalEz}[2][]{\ensuremath{\subp{\check{\mathcal{E}}}{}{#2}{}{#1}}}
\newrobustcmd{\tildemathcalEz}[2][]{\ensuremath{\subp{\tilde{\mathcal{E}}}{}{#2}{}{#1}}}
\newrobustcmd{\widetildemathcalEz}[2][]{\ensuremath{\subp{\widetilde{\mathcal{E}}}{}{#2}{}{#1}}}
\newrobustcmd{\acutemathcalEz}[2][]{\ensuremath{\subp{\acute{\mathcal{E}}}{}{#2}{}{#1}}}
\newrobustcmd{\gravemathcalEz}[2][]{\ensuremath{\subp{\grave{\mathcal{E}}}{}{#2}{}{#1}}}
\newrobustcmd{\dotmathcalEz}[2][]{\ensuremath{\subp{\dot{\mathcal{E}}}{}{#2}{}{#1}}}
\newrobustcmd{\ddotmathcalEz}[2][]{\ensuremath{\subp{\ddot{\mathcal{E}}}{}{#2}{}{#1}}}
\newrobustcmd{\brevemathcalEz}[2][]{\ensuremath{\subp{\breve{\mathcal{E}}}{}{#2}{}{#1}}}
\newrobustcmd{\barmathcalEz}[2][]{\ensuremath{\subp{\bar{\mathcal{E}}}{}{#2}{}{#1}}}
\newrobustcmd{\vecmathcalEz}[2][]{\ensuremath{\subp{\vec{\mathcal{E}}}{}{#2}{}{#1}}}
\newrobustcmd{\bmmathcalEz}[2][]{\ensuremath{\subp{\bm{\mathcal{E}}}{}{#2}{}{#1}}}
\newrobustcmd{\hatbmmathcalEz}[2][]{\ensuremath{\subp{\hat{\bm{\mathcal{E}}}}{}{#2}{}{#1}}}
\newrobustcmd{\widehatbmmathcalEz}[2][]{\ensuremath{\subp{\widehat{\bm{\mathcal{E}}}}{}{#2}{}{#1}}}
\newrobustcmd{\checkbmmathcalEz}[2][]{\ensuremath{\subp{\check{\bm{\mathcal{E}}}}{}{#2}{}{#1}}}
\newrobustcmd{\tildebmmathcalEz}[2][]{\ensuremath{\subp{\tilde{\bm{\mathcal{E}}}}{}{#2}{}{#1}}}
\newrobustcmd{\widetildebmmathcalEz}[2][]{\ensuremath{\subp{\widetilde{\bm{\mathcal{E}}}}{}{#2}{}{#1}}}
\newrobustcmd{\acutebmmathcalEz}[2][]{\ensuremath{\subp{\acute{\bm{\mathcal{E}}}}{}{#2}{}{#1}}}
\newrobustcmd{\gravebmmathcalEz}[2][]{\ensuremath{\subp{\grave{\bm{\mathcal{E}}}}{}{#2}{}{#1}}}
\newrobustcmd{\dotbmmathcalEz}[2][]{\ensuremath{\subp{\dot{\bm{\mathcal{E}}}}{}{#2}{}{#1}}}
\newrobustcmd{\ddotbmmathcalEz}[2][]{\ensuremath{\subp{\ddot{\bm{\mathcal{E}}}}{}{#2}{}{#1}}}
\newrobustcmd{\brevebmmathcalEz}[2][]{\ensuremath{\subp{\breve{\bm{\mathcal{E}}}}{}{#2}{}{#1}}}
\newrobustcmd{\barbmmathcalEz}[2][]{\ensuremath{\subp{\bar{\bm{\mathcal{E}}}}{}{#2}{}{#1}}}
\newrobustcmd{\vecbmmathcalEz}[2][]{\ensuremath{\subp{\vec{\bm{\mathcal{E}}}}{}{#2}{}{#1}}}
\newrobustcmd{\mathcalFz}[2][]{\ensuremath{\subp{\mathcal{F}}{}{#2}{}{#1}}}
\newrobustcmd{\hatmathcalFz}[2][]{\ensuremath{\subp{\hat{\mathcal{F}}}{}{#2}{}{#1}}}
\newrobustcmd{\widehatmathcalFz}[2][]{\ensuremath{\subp{\widehat{\mathcal{F}}}{}{#2}{}{#1}}}
\newrobustcmd{\checkmathcalFz}[2][]{\ensuremath{\subp{\check{\mathcal{F}}}{}{#2}{}{#1}}}
\newrobustcmd{\tildemathcalFz}[2][]{\ensuremath{\subp{\tilde{\mathcal{F}}}{}{#2}{}{#1}}}
\newrobustcmd{\widetildemathcalFz}[2][]{\ensuremath{\subp{\widetilde{\mathcal{F}}}{}{#2}{}{#1}}}
\newrobustcmd{\acutemathcalFz}[2][]{\ensuremath{\subp{\acute{\mathcal{F}}}{}{#2}{}{#1}}}
\newrobustcmd{\gravemathcalFz}[2][]{\ensuremath{\subp{\grave{\mathcal{F}}}{}{#2}{}{#1}}}
\newrobustcmd{\dotmathcalFz}[2][]{\ensuremath{\subp{\dot{\mathcal{F}}}{}{#2}{}{#1}}}
\newrobustcmd{\ddotmathcalFz}[2][]{\ensuremath{\subp{\ddot{\mathcal{F}}}{}{#2}{}{#1}}}
\newrobustcmd{\brevemathcalFz}[2][]{\ensuremath{\subp{\breve{\mathcal{F}}}{}{#2}{}{#1}}}
\newrobustcmd{\barmathcalFz}[2][]{\ensuremath{\subp{\bar{\mathcal{F}}}{}{#2}{}{#1}}}
\newrobustcmd{\vecmathcalFz}[2][]{\ensuremath{\subp{\vec{\mathcal{F}}}{}{#2}{}{#1}}}
\newrobustcmd{\bmmathcalFz}[2][]{\ensuremath{\subp{\bm{\mathcal{F}}}{}{#2}{}{#1}}}
\newrobustcmd{\hatbmmathcalFz}[2][]{\ensuremath{\subp{\hat{\bm{\mathcal{F}}}}{}{#2}{}{#1}}}
\newrobustcmd{\widehatbmmathcalFz}[2][]{\ensuremath{\subp{\widehat{\bm{\mathcal{F}}}}{}{#2}{}{#1}}}
\newrobustcmd{\checkbmmathcalFz}[2][]{\ensuremath{\subp{\check{\bm{\mathcal{F}}}}{}{#2}{}{#1}}}
\newrobustcmd{\tildebmmathcalFz}[2][]{\ensuremath{\subp{\tilde{\bm{\mathcal{F}}}}{}{#2}{}{#1}}}
\newrobustcmd{\widetildebmmathcalFz}[2][]{\ensuremath{\subp{\widetilde{\bm{\mathcal{F}}}}{}{#2}{}{#1}}}
\newrobustcmd{\acutebmmathcalFz}[2][]{\ensuremath{\subp{\acute{\bm{\mathcal{F}}}}{}{#2}{}{#1}}}
\newrobustcmd{\gravebmmathcalFz}[2][]{\ensuremath{\subp{\grave{\bm{\mathcal{F}}}}{}{#2}{}{#1}}}
\newrobustcmd{\dotbmmathcalFz}[2][]{\ensuremath{\subp{\dot{\bm{\mathcal{F}}}}{}{#2}{}{#1}}}
\newrobustcmd{\ddotbmmathcalFz}[2][]{\ensuremath{\subp{\ddot{\bm{\mathcal{F}}}}{}{#2}{}{#1}}}
\newrobustcmd{\brevebmmathcalFz}[2][]{\ensuremath{\subp{\breve{\bm{\mathcal{F}}}}{}{#2}{}{#1}}}
\newrobustcmd{\barbmmathcalFz}[2][]{\ensuremath{\subp{\bar{\bm{\mathcal{F}}}}{}{#2}{}{#1}}}
\newrobustcmd{\vecbmmathcalFz}[2][]{\ensuremath{\subp{\vec{\bm{\mathcal{F}}}}{}{#2}{}{#1}}}
\newrobustcmd{\mathcalGz}[2][]{\ensuremath{\subp{\mathcal{G}}{}{#2}{}{#1}}}
\newrobustcmd{\hatmathcalGz}[2][]{\ensuremath{\subp{\hat{\mathcal{G}}}{}{#2}{}{#1}}}
\newrobustcmd{\widehatmathcalGz}[2][]{\ensuremath{\subp{\widehat{\mathcal{G}}}{}{#2}{}{#1}}}
\newrobustcmd{\checkmathcalGz}[2][]{\ensuremath{\subp{\check{\mathcal{G}}}{}{#2}{}{#1}}}
\newrobustcmd{\tildemathcalGz}[2][]{\ensuremath{\subp{\tilde{\mathcal{G}}}{}{#2}{}{#1}}}
\newrobustcmd{\widetildemathcalGz}[2][]{\ensuremath{\subp{\widetilde{\mathcal{G}}}{}{#2}{}{#1}}}
\newrobustcmd{\acutemathcalGz}[2][]{\ensuremath{\subp{\acute{\mathcal{G}}}{}{#2}{}{#1}}}
\newrobustcmd{\gravemathcalGz}[2][]{\ensuremath{\subp{\grave{\mathcal{G}}}{}{#2}{}{#1}}}
\newrobustcmd{\dotmathcalGz}[2][]{\ensuremath{\subp{\dot{\mathcal{G}}}{}{#2}{}{#1}}}
\newrobustcmd{\ddotmathcalGz}[2][]{\ensuremath{\subp{\ddot{\mathcal{G}}}{}{#2}{}{#1}}}
\newrobustcmd{\brevemathcalGz}[2][]{\ensuremath{\subp{\breve{\mathcal{G}}}{}{#2}{}{#1}}}
\newrobustcmd{\barmathcalGz}[2][]{\ensuremath{\subp{\bar{\mathcal{G}}}{}{#2}{}{#1}}}
\newrobustcmd{\vecmathcalGz}[2][]{\ensuremath{\subp{\vec{\mathcal{G}}}{}{#2}{}{#1}}}
\newrobustcmd{\bmmathcalGz}[2][]{\ensuremath{\subp{\bm{\mathcal{G}}}{}{#2}{}{#1}}}
\newrobustcmd{\hatbmmathcalGz}[2][]{\ensuremath{\subp{\hat{\bm{\mathcal{G}}}}{}{#2}{}{#1}}}
\newrobustcmd{\widehatbmmathcalGz}[2][]{\ensuremath{\subp{\widehat{\bm{\mathcal{G}}}}{}{#2}{}{#1}}}
\newrobustcmd{\checkbmmathcalGz}[2][]{\ensuremath{\subp{\check{\bm{\mathcal{G}}}}{}{#2}{}{#1}}}
\newrobustcmd{\tildebmmathcalGz}[2][]{\ensuremath{\subp{\tilde{\bm{\mathcal{G}}}}{}{#2}{}{#1}}}
\newrobustcmd{\widetildebmmathcalGz}[2][]{\ensuremath{\subp{\widetilde{\bm{\mathcal{G}}}}{}{#2}{}{#1}}}
\newrobustcmd{\acutebmmathcalGz}[2][]{\ensuremath{\subp{\acute{\bm{\mathcal{G}}}}{}{#2}{}{#1}}}
\newrobustcmd{\gravebmmathcalGz}[2][]{\ensuremath{\subp{\grave{\bm{\mathcal{G}}}}{}{#2}{}{#1}}}
\newrobustcmd{\dotbmmathcalGz}[2][]{\ensuremath{\subp{\dot{\bm{\mathcal{G}}}}{}{#2}{}{#1}}}
\newrobustcmd{\ddotbmmathcalGz}[2][]{\ensuremath{\subp{\ddot{\bm{\mathcal{G}}}}{}{#2}{}{#1}}}
\newrobustcmd{\brevebmmathcalGz}[2][]{\ensuremath{\subp{\breve{\bm{\mathcal{G}}}}{}{#2}{}{#1}}}
\newrobustcmd{\barbmmathcalGz}[2][]{\ensuremath{\subp{\bar{\bm{\mathcal{G}}}}{}{#2}{}{#1}}}
\newrobustcmd{\vecbmmathcalGz}[2][]{\ensuremath{\subp{\vec{\bm{\mathcal{G}}}}{}{#2}{}{#1}}}
\newrobustcmd{\mathcalHz}[2][]{\ensuremath{\subp{\mathcal{H}}{}{#2}{}{#1}}}
\newrobustcmd{\hatmathcalHz}[2][]{\ensuremath{\subp{\hat{\mathcal{H}}}{}{#2}{}{#1}}}
\newrobustcmd{\widehatmathcalHz}[2][]{\ensuremath{\subp{\widehat{\mathcal{H}}}{}{#2}{}{#1}}}
\newrobustcmd{\checkmathcalHz}[2][]{\ensuremath{\subp{\check{\mathcal{H}}}{}{#2}{}{#1}}}
\newrobustcmd{\tildemathcalHz}[2][]{\ensuremath{\subp{\tilde{\mathcal{H}}}{}{#2}{}{#1}}}
\newrobustcmd{\widetildemathcalHz}[2][]{\ensuremath{\subp{\widetilde{\mathcal{H}}}{}{#2}{}{#1}}}
\newrobustcmd{\acutemathcalHz}[2][]{\ensuremath{\subp{\acute{\mathcal{H}}}{}{#2}{}{#1}}}
\newrobustcmd{\gravemathcalHz}[2][]{\ensuremath{\subp{\grave{\mathcal{H}}}{}{#2}{}{#1}}}
\newrobustcmd{\dotmathcalHz}[2][]{\ensuremath{\subp{\dot{\mathcal{H}}}{}{#2}{}{#1}}}
\newrobustcmd{\ddotmathcalHz}[2][]{\ensuremath{\subp{\ddot{\mathcal{H}}}{}{#2}{}{#1}}}
\newrobustcmd{\brevemathcalHz}[2][]{\ensuremath{\subp{\breve{\mathcal{H}}}{}{#2}{}{#1}}}
\newrobustcmd{\barmathcalHz}[2][]{\ensuremath{\subp{\bar{\mathcal{H}}}{}{#2}{}{#1}}}
\newrobustcmd{\vecmathcalHz}[2][]{\ensuremath{\subp{\vec{\mathcal{H}}}{}{#2}{}{#1}}}
\newrobustcmd{\bmmathcalHz}[2][]{\ensuremath{\subp{\bm{\mathcal{H}}}{}{#2}{}{#1}}}
\newrobustcmd{\hatbmmathcalHz}[2][]{\ensuremath{\subp{\hat{\bm{\mathcal{H}}}}{}{#2}{}{#1}}}
\newrobustcmd{\widehatbmmathcalHz}[2][]{\ensuremath{\subp{\widehat{\bm{\mathcal{H}}}}{}{#2}{}{#1}}}
\newrobustcmd{\checkbmmathcalHz}[2][]{\ensuremath{\subp{\check{\bm{\mathcal{H}}}}{}{#2}{}{#1}}}
\newrobustcmd{\tildebmmathcalHz}[2][]{\ensuremath{\subp{\tilde{\bm{\mathcal{H}}}}{}{#2}{}{#1}}}
\newrobustcmd{\widetildebmmathcalHz}[2][]{\ensuremath{\subp{\widetilde{\bm{\mathcal{H}}}}{}{#2}{}{#1}}}
\newrobustcmd{\acutebmmathcalHz}[2][]{\ensuremath{\subp{\acute{\bm{\mathcal{H}}}}{}{#2}{}{#1}}}
\newrobustcmd{\gravebmmathcalHz}[2][]{\ensuremath{\subp{\grave{\bm{\mathcal{H}}}}{}{#2}{}{#1}}}
\newrobustcmd{\dotbmmathcalHz}[2][]{\ensuremath{\subp{\dot{\bm{\mathcal{H}}}}{}{#2}{}{#1}}}
\newrobustcmd{\ddotbmmathcalHz}[2][]{\ensuremath{\subp{\ddot{\bm{\mathcal{H}}}}{}{#2}{}{#1}}}
\newrobustcmd{\brevebmmathcalHz}[2][]{\ensuremath{\subp{\breve{\bm{\mathcal{H}}}}{}{#2}{}{#1}}}
\newrobustcmd{\barbmmathcalHz}[2][]{\ensuremath{\subp{\bar{\bm{\mathcal{H}}}}{}{#2}{}{#1}}}
\newrobustcmd{\vecbmmathcalHz}[2][]{\ensuremath{\subp{\vec{\bm{\mathcal{H}}}}{}{#2}{}{#1}}}
\newrobustcmd{\mathcalIz}[2][]{\ensuremath{\subp{\mathcal{I}}{}{#2}{}{#1}}}
\newrobustcmd{\hatmathcalIz}[2][]{\ensuremath{\subp{\hat{\mathcal{I}}}{}{#2}{}{#1}}}
\newrobustcmd{\widehatmathcalIz}[2][]{\ensuremath{\subp{\widehat{\mathcal{I}}}{}{#2}{}{#1}}}
\newrobustcmd{\checkmathcalIz}[2][]{\ensuremath{\subp{\check{\mathcal{I}}}{}{#2}{}{#1}}}
\newrobustcmd{\tildemathcalIz}[2][]{\ensuremath{\subp{\tilde{\mathcal{I}}}{}{#2}{}{#1}}}
\newrobustcmd{\widetildemathcalIz}[2][]{\ensuremath{\subp{\widetilde{\mathcal{I}}}{}{#2}{}{#1}}}
\newrobustcmd{\acutemathcalIz}[2][]{\ensuremath{\subp{\acute{\mathcal{I}}}{}{#2}{}{#1}}}
\newrobustcmd{\gravemathcalIz}[2][]{\ensuremath{\subp{\grave{\mathcal{I}}}{}{#2}{}{#1}}}
\newrobustcmd{\dotmathcalIz}[2][]{\ensuremath{\subp{\dot{\mathcal{I}}}{}{#2}{}{#1}}}
\newrobustcmd{\ddotmathcalIz}[2][]{\ensuremath{\subp{\ddot{\mathcal{I}}}{}{#2}{}{#1}}}
\newrobustcmd{\brevemathcalIz}[2][]{\ensuremath{\subp{\breve{\mathcal{I}}}{}{#2}{}{#1}}}
\newrobustcmd{\barmathcalIz}[2][]{\ensuremath{\subp{\bar{\mathcal{I}}}{}{#2}{}{#1}}}
\newrobustcmd{\vecmathcalIz}[2][]{\ensuremath{\subp{\vec{\mathcal{I}}}{}{#2}{}{#1}}}
\newrobustcmd{\bmmathcalIz}[2][]{\ensuremath{\subp{\bm{\mathcal{I}}}{}{#2}{}{#1}}}
\newrobustcmd{\hatbmmathcalIz}[2][]{\ensuremath{\subp{\hat{\bm{\mathcal{I}}}}{}{#2}{}{#1}}}
\newrobustcmd{\widehatbmmathcalIz}[2][]{\ensuremath{\subp{\widehat{\bm{\mathcal{I}}}}{}{#2}{}{#1}}}
\newrobustcmd{\checkbmmathcalIz}[2][]{\ensuremath{\subp{\check{\bm{\mathcal{I}}}}{}{#2}{}{#1}}}
\newrobustcmd{\tildebmmathcalIz}[2][]{\ensuremath{\subp{\tilde{\bm{\mathcal{I}}}}{}{#2}{}{#1}}}
\newrobustcmd{\widetildebmmathcalIz}[2][]{\ensuremath{\subp{\widetilde{\bm{\mathcal{I}}}}{}{#2}{}{#1}}}
\newrobustcmd{\acutebmmathcalIz}[2][]{\ensuremath{\subp{\acute{\bm{\mathcal{I}}}}{}{#2}{}{#1}}}
\newrobustcmd{\gravebmmathcalIz}[2][]{\ensuremath{\subp{\grave{\bm{\mathcal{I}}}}{}{#2}{}{#1}}}
\newrobustcmd{\dotbmmathcalIz}[2][]{\ensuremath{\subp{\dot{\bm{\mathcal{I}}}}{}{#2}{}{#1}}}
\newrobustcmd{\ddotbmmathcalIz}[2][]{\ensuremath{\subp{\ddot{\bm{\mathcal{I}}}}{}{#2}{}{#1}}}
\newrobustcmd{\brevebmmathcalIz}[2][]{\ensuremath{\subp{\breve{\bm{\mathcal{I}}}}{}{#2}{}{#1}}}
\newrobustcmd{\barbmmathcalIz}[2][]{\ensuremath{\subp{\bar{\bm{\mathcal{I}}}}{}{#2}{}{#1}}}
\newrobustcmd{\vecbmmathcalIz}[2][]{\ensuremath{\subp{\vec{\bm{\mathcal{I}}}}{}{#2}{}{#1}}}
\newrobustcmd{\mathcalJz}[2][]{\ensuremath{\subp{\mathcal{J}}{}{#2}{}{#1}}}
\newrobustcmd{\hatmathcalJz}[2][]{\ensuremath{\subp{\hat{\mathcal{J}}}{}{#2}{}{#1}}}
\newrobustcmd{\widehatmathcalJz}[2][]{\ensuremath{\subp{\widehat{\mathcal{J}}}{}{#2}{}{#1}}}
\newrobustcmd{\checkmathcalJz}[2][]{\ensuremath{\subp{\check{\mathcal{J}}}{}{#2}{}{#1}}}
\newrobustcmd{\tildemathcalJz}[2][]{\ensuremath{\subp{\tilde{\mathcal{J}}}{}{#2}{}{#1}}}
\newrobustcmd{\widetildemathcalJz}[2][]{\ensuremath{\subp{\widetilde{\mathcal{J}}}{}{#2}{}{#1}}}
\newrobustcmd{\acutemathcalJz}[2][]{\ensuremath{\subp{\acute{\mathcal{J}}}{}{#2}{}{#1}}}
\newrobustcmd{\gravemathcalJz}[2][]{\ensuremath{\subp{\grave{\mathcal{J}}}{}{#2}{}{#1}}}
\newrobustcmd{\dotmathcalJz}[2][]{\ensuremath{\subp{\dot{\mathcal{J}}}{}{#2}{}{#1}}}
\newrobustcmd{\ddotmathcalJz}[2][]{\ensuremath{\subp{\ddot{\mathcal{J}}}{}{#2}{}{#1}}}
\newrobustcmd{\brevemathcalJz}[2][]{\ensuremath{\subp{\breve{\mathcal{J}}}{}{#2}{}{#1}}}
\newrobustcmd{\barmathcalJz}[2][]{\ensuremath{\subp{\bar{\mathcal{J}}}{}{#2}{}{#1}}}
\newrobustcmd{\vecmathcalJz}[2][]{\ensuremath{\subp{\vec{\mathcal{J}}}{}{#2}{}{#1}}}
\newrobustcmd{\bmmathcalJz}[2][]{\ensuremath{\subp{\bm{\mathcal{J}}}{}{#2}{}{#1}}}
\newrobustcmd{\hatbmmathcalJz}[2][]{\ensuremath{\subp{\hat{\bm{\mathcal{J}}}}{}{#2}{}{#1}}}
\newrobustcmd{\widehatbmmathcalJz}[2][]{\ensuremath{\subp{\widehat{\bm{\mathcal{J}}}}{}{#2}{}{#1}}}
\newrobustcmd{\checkbmmathcalJz}[2][]{\ensuremath{\subp{\check{\bm{\mathcal{J}}}}{}{#2}{}{#1}}}
\newrobustcmd{\tildebmmathcalJz}[2][]{\ensuremath{\subp{\tilde{\bm{\mathcal{J}}}}{}{#2}{}{#1}}}
\newrobustcmd{\widetildebmmathcalJz}[2][]{\ensuremath{\subp{\widetilde{\bm{\mathcal{J}}}}{}{#2}{}{#1}}}
\newrobustcmd{\acutebmmathcalJz}[2][]{\ensuremath{\subp{\acute{\bm{\mathcal{J}}}}{}{#2}{}{#1}}}
\newrobustcmd{\gravebmmathcalJz}[2][]{\ensuremath{\subp{\grave{\bm{\mathcal{J}}}}{}{#2}{}{#1}}}
\newrobustcmd{\dotbmmathcalJz}[2][]{\ensuremath{\subp{\dot{\bm{\mathcal{J}}}}{}{#2}{}{#1}}}
\newrobustcmd{\ddotbmmathcalJz}[2][]{\ensuremath{\subp{\ddot{\bm{\mathcal{J}}}}{}{#2}{}{#1}}}
\newrobustcmd{\brevebmmathcalJz}[2][]{\ensuremath{\subp{\breve{\bm{\mathcal{J}}}}{}{#2}{}{#1}}}
\newrobustcmd{\barbmmathcalJz}[2][]{\ensuremath{\subp{\bar{\bm{\mathcal{J}}}}{}{#2}{}{#1}}}
\newrobustcmd{\vecbmmathcalJz}[2][]{\ensuremath{\subp{\vec{\bm{\mathcal{J}}}}{}{#2}{}{#1}}}
\newrobustcmd{\mathcalKz}[2][]{\ensuremath{\subp{\mathcal{K}}{}{#2}{}{#1}}}
\newrobustcmd{\hatmathcalKz}[2][]{\ensuremath{\subp{\hat{\mathcal{K}}}{}{#2}{}{#1}}}
\newrobustcmd{\widehatmathcalKz}[2][]{\ensuremath{\subp{\widehat{\mathcal{K}}}{}{#2}{}{#1}}}
\newrobustcmd{\checkmathcalKz}[2][]{\ensuremath{\subp{\check{\mathcal{K}}}{}{#2}{}{#1}}}
\newrobustcmd{\tildemathcalKz}[2][]{\ensuremath{\subp{\tilde{\mathcal{K}}}{}{#2}{}{#1}}}
\newrobustcmd{\widetildemathcalKz}[2][]{\ensuremath{\subp{\widetilde{\mathcal{K}}}{}{#2}{}{#1}}}
\newrobustcmd{\acutemathcalKz}[2][]{\ensuremath{\subp{\acute{\mathcal{K}}}{}{#2}{}{#1}}}
\newrobustcmd{\gravemathcalKz}[2][]{\ensuremath{\subp{\grave{\mathcal{K}}}{}{#2}{}{#1}}}
\newrobustcmd{\dotmathcalKz}[2][]{\ensuremath{\subp{\dot{\mathcal{K}}}{}{#2}{}{#1}}}
\newrobustcmd{\ddotmathcalKz}[2][]{\ensuremath{\subp{\ddot{\mathcal{K}}}{}{#2}{}{#1}}}
\newrobustcmd{\brevemathcalKz}[2][]{\ensuremath{\subp{\breve{\mathcal{K}}}{}{#2}{}{#1}}}
\newrobustcmd{\barmathcalKz}[2][]{\ensuremath{\subp{\bar{\mathcal{K}}}{}{#2}{}{#1}}}
\newrobustcmd{\vecmathcalKz}[2][]{\ensuremath{\subp{\vec{\mathcal{K}}}{}{#2}{}{#1}}}
\newrobustcmd{\bmmathcalKz}[2][]{\ensuremath{\subp{\bm{\mathcal{K}}}{}{#2}{}{#1}}}
\newrobustcmd{\hatbmmathcalKz}[2][]{\ensuremath{\subp{\hat{\bm{\mathcal{K}}}}{}{#2}{}{#1}}}
\newrobustcmd{\widehatbmmathcalKz}[2][]{\ensuremath{\subp{\widehat{\bm{\mathcal{K}}}}{}{#2}{}{#1}}}
\newrobustcmd{\checkbmmathcalKz}[2][]{\ensuremath{\subp{\check{\bm{\mathcal{K}}}}{}{#2}{}{#1}}}
\newrobustcmd{\tildebmmathcalKz}[2][]{\ensuremath{\subp{\tilde{\bm{\mathcal{K}}}}{}{#2}{}{#1}}}
\newrobustcmd{\widetildebmmathcalKz}[2][]{\ensuremath{\subp{\widetilde{\bm{\mathcal{K}}}}{}{#2}{}{#1}}}
\newrobustcmd{\acutebmmathcalKz}[2][]{\ensuremath{\subp{\acute{\bm{\mathcal{K}}}}{}{#2}{}{#1}}}
\newrobustcmd{\gravebmmathcalKz}[2][]{\ensuremath{\subp{\grave{\bm{\mathcal{K}}}}{}{#2}{}{#1}}}
\newrobustcmd{\dotbmmathcalKz}[2][]{\ensuremath{\subp{\dot{\bm{\mathcal{K}}}}{}{#2}{}{#1}}}
\newrobustcmd{\ddotbmmathcalKz}[2][]{\ensuremath{\subp{\ddot{\bm{\mathcal{K}}}}{}{#2}{}{#1}}}
\newrobustcmd{\brevebmmathcalKz}[2][]{\ensuremath{\subp{\breve{\bm{\mathcal{K}}}}{}{#2}{}{#1}}}
\newrobustcmd{\barbmmathcalKz}[2][]{\ensuremath{\subp{\bar{\bm{\mathcal{K}}}}{}{#2}{}{#1}}}
\newrobustcmd{\vecbmmathcalKz}[2][]{\ensuremath{\subp{\vec{\bm{\mathcal{K}}}}{}{#2}{}{#1}}}
\newrobustcmd{\mathcalLz}[2][]{\ensuremath{\subp{\mathcal{L}}{}{#2}{}{#1}}}
\newrobustcmd{\hatmathcalLz}[2][]{\ensuremath{\subp{\hat{\mathcal{L}}}{}{#2}{}{#1}}}
\newrobustcmd{\widehatmathcalLz}[2][]{\ensuremath{\subp{\widehat{\mathcal{L}}}{}{#2}{}{#1}}}
\newrobustcmd{\checkmathcalLz}[2][]{\ensuremath{\subp{\check{\mathcal{L}}}{}{#2}{}{#1}}}
\newrobustcmd{\tildemathcalLz}[2][]{\ensuremath{\subp{\tilde{\mathcal{L}}}{}{#2}{}{#1}}}
\newrobustcmd{\widetildemathcalLz}[2][]{\ensuremath{\subp{\widetilde{\mathcal{L}}}{}{#2}{}{#1}}}
\newrobustcmd{\acutemathcalLz}[2][]{\ensuremath{\subp{\acute{\mathcal{L}}}{}{#2}{}{#1}}}
\newrobustcmd{\gravemathcalLz}[2][]{\ensuremath{\subp{\grave{\mathcal{L}}}{}{#2}{}{#1}}}
\newrobustcmd{\dotmathcalLz}[2][]{\ensuremath{\subp{\dot{\mathcal{L}}}{}{#2}{}{#1}}}
\newrobustcmd{\ddotmathcalLz}[2][]{\ensuremath{\subp{\ddot{\mathcal{L}}}{}{#2}{}{#1}}}
\newrobustcmd{\brevemathcalLz}[2][]{\ensuremath{\subp{\breve{\mathcal{L}}}{}{#2}{}{#1}}}
\newrobustcmd{\barmathcalLz}[2][]{\ensuremath{\subp{\bar{\mathcal{L}}}{}{#2}{}{#1}}}
\newrobustcmd{\vecmathcalLz}[2][]{\ensuremath{\subp{\vec{\mathcal{L}}}{}{#2}{}{#1}}}
\newrobustcmd{\bmmathcalLz}[2][]{\ensuremath{\subp{\bm{\mathcal{L}}}{}{#2}{}{#1}}}
\newrobustcmd{\hatbmmathcalLz}[2][]{\ensuremath{\subp{\hat{\bm{\mathcal{L}}}}{}{#2}{}{#1}}}
\newrobustcmd{\widehatbmmathcalLz}[2][]{\ensuremath{\subp{\widehat{\bm{\mathcal{L}}}}{}{#2}{}{#1}}}
\newrobustcmd{\checkbmmathcalLz}[2][]{\ensuremath{\subp{\check{\bm{\mathcal{L}}}}{}{#2}{}{#1}}}
\newrobustcmd{\tildebmmathcalLz}[2][]{\ensuremath{\subp{\tilde{\bm{\mathcal{L}}}}{}{#2}{}{#1}}}
\newrobustcmd{\widetildebmmathcalLz}[2][]{\ensuremath{\subp{\widetilde{\bm{\mathcal{L}}}}{}{#2}{}{#1}}}
\newrobustcmd{\acutebmmathcalLz}[2][]{\ensuremath{\subp{\acute{\bm{\mathcal{L}}}}{}{#2}{}{#1}}}
\newrobustcmd{\gravebmmathcalLz}[2][]{\ensuremath{\subp{\grave{\bm{\mathcal{L}}}}{}{#2}{}{#1}}}
\newrobustcmd{\dotbmmathcalLz}[2][]{\ensuremath{\subp{\dot{\bm{\mathcal{L}}}}{}{#2}{}{#1}}}
\newrobustcmd{\ddotbmmathcalLz}[2][]{\ensuremath{\subp{\ddot{\bm{\mathcal{L}}}}{}{#2}{}{#1}}}
\newrobustcmd{\brevebmmathcalLz}[2][]{\ensuremath{\subp{\breve{\bm{\mathcal{L}}}}{}{#2}{}{#1}}}
\newrobustcmd{\barbmmathcalLz}[2][]{\ensuremath{\subp{\bar{\bm{\mathcal{L}}}}{}{#2}{}{#1}}}
\newrobustcmd{\vecbmmathcalLz}[2][]{\ensuremath{\subp{\vec{\bm{\mathcal{L}}}}{}{#2}{}{#1}}}
\newrobustcmd{\mathcalMz}[2][]{\ensuremath{\subp{\mathcal{M}}{}{#2}{}{#1}}}
\newrobustcmd{\hatmathcalMz}[2][]{\ensuremath{\subp{\hat{\mathcal{M}}}{}{#2}{}{#1}}}
\newrobustcmd{\widehatmathcalMz}[2][]{\ensuremath{\subp{\widehat{\mathcal{M}}}{}{#2}{}{#1}}}
\newrobustcmd{\checkmathcalMz}[2][]{\ensuremath{\subp{\check{\mathcal{M}}}{}{#2}{}{#1}}}
\newrobustcmd{\tildemathcalMz}[2][]{\ensuremath{\subp{\tilde{\mathcal{M}}}{}{#2}{}{#1}}}
\newrobustcmd{\widetildemathcalMz}[2][]{\ensuremath{\subp{\widetilde{\mathcal{M}}}{}{#2}{}{#1}}}
\newrobustcmd{\acutemathcalMz}[2][]{\ensuremath{\subp{\acute{\mathcal{M}}}{}{#2}{}{#1}}}
\newrobustcmd{\gravemathcalMz}[2][]{\ensuremath{\subp{\grave{\mathcal{M}}}{}{#2}{}{#1}}}
\newrobustcmd{\dotmathcalMz}[2][]{\ensuremath{\subp{\dot{\mathcal{M}}}{}{#2}{}{#1}}}
\newrobustcmd{\ddotmathcalMz}[2][]{\ensuremath{\subp{\ddot{\mathcal{M}}}{}{#2}{}{#1}}}
\newrobustcmd{\brevemathcalMz}[2][]{\ensuremath{\subp{\breve{\mathcal{M}}}{}{#2}{}{#1}}}
\newrobustcmd{\barmathcalMz}[2][]{\ensuremath{\subp{\bar{\mathcal{M}}}{}{#2}{}{#1}}}
\newrobustcmd{\vecmathcalMz}[2][]{\ensuremath{\subp{\vec{\mathcal{M}}}{}{#2}{}{#1}}}
\newrobustcmd{\bmmathcalMz}[2][]{\ensuremath{\subp{\bm{\mathcal{M}}}{}{#2}{}{#1}}}
\newrobustcmd{\hatbmmathcalMz}[2][]{\ensuremath{\subp{\hat{\bm{\mathcal{M}}}}{}{#2}{}{#1}}}
\newrobustcmd{\widehatbmmathcalMz}[2][]{\ensuremath{\subp{\widehat{\bm{\mathcal{M}}}}{}{#2}{}{#1}}}
\newrobustcmd{\checkbmmathcalMz}[2][]{\ensuremath{\subp{\check{\bm{\mathcal{M}}}}{}{#2}{}{#1}}}
\newrobustcmd{\tildebmmathcalMz}[2][]{\ensuremath{\subp{\tilde{\bm{\mathcal{M}}}}{}{#2}{}{#1}}}
\newrobustcmd{\widetildebmmathcalMz}[2][]{\ensuremath{\subp{\widetilde{\bm{\mathcal{M}}}}{}{#2}{}{#1}}}
\newrobustcmd{\acutebmmathcalMz}[2][]{\ensuremath{\subp{\acute{\bm{\mathcal{M}}}}{}{#2}{}{#1}}}
\newrobustcmd{\gravebmmathcalMz}[2][]{\ensuremath{\subp{\grave{\bm{\mathcal{M}}}}{}{#2}{}{#1}}}
\newrobustcmd{\dotbmmathcalMz}[2][]{\ensuremath{\subp{\dot{\bm{\mathcal{M}}}}{}{#2}{}{#1}}}
\newrobustcmd{\ddotbmmathcalMz}[2][]{\ensuremath{\subp{\ddot{\bm{\mathcal{M}}}}{}{#2}{}{#1}}}
\newrobustcmd{\brevebmmathcalMz}[2][]{\ensuremath{\subp{\breve{\bm{\mathcal{M}}}}{}{#2}{}{#1}}}
\newrobustcmd{\barbmmathcalMz}[2][]{\ensuremath{\subp{\bar{\bm{\mathcal{M}}}}{}{#2}{}{#1}}}
\newrobustcmd{\vecbmmathcalMz}[2][]{\ensuremath{\subp{\vec{\bm{\mathcal{M}}}}{}{#2}{}{#1}}}
\newrobustcmd{\mathcalNz}[2][]{\ensuremath{\subp{\mathcal{N}}{}{#2}{}{#1}}}
\newrobustcmd{\hatmathcalNz}[2][]{\ensuremath{\subp{\hat{\mathcal{N}}}{}{#2}{}{#1}}}
\newrobustcmd{\widehatmathcalNz}[2][]{\ensuremath{\subp{\widehat{\mathcal{N}}}{}{#2}{}{#1}}}
\newrobustcmd{\checkmathcalNz}[2][]{\ensuremath{\subp{\check{\mathcal{N}}}{}{#2}{}{#1}}}
\newrobustcmd{\tildemathcalNz}[2][]{\ensuremath{\subp{\tilde{\mathcal{N}}}{}{#2}{}{#1}}}
\newrobustcmd{\widetildemathcalNz}[2][]{\ensuremath{\subp{\widetilde{\mathcal{N}}}{}{#2}{}{#1}}}
\newrobustcmd{\acutemathcalNz}[2][]{\ensuremath{\subp{\acute{\mathcal{N}}}{}{#2}{}{#1}}}
\newrobustcmd{\gravemathcalNz}[2][]{\ensuremath{\subp{\grave{\mathcal{N}}}{}{#2}{}{#1}}}
\newrobustcmd{\dotmathcalNz}[2][]{\ensuremath{\subp{\dot{\mathcal{N}}}{}{#2}{}{#1}}}
\newrobustcmd{\ddotmathcalNz}[2][]{\ensuremath{\subp{\ddot{\mathcal{N}}}{}{#2}{}{#1}}}
\newrobustcmd{\brevemathcalNz}[2][]{\ensuremath{\subp{\breve{\mathcal{N}}}{}{#2}{}{#1}}}
\newrobustcmd{\barmathcalNz}[2][]{\ensuremath{\subp{\bar{\mathcal{N}}}{}{#2}{}{#1}}}
\newrobustcmd{\vecmathcalNz}[2][]{\ensuremath{\subp{\vec{\mathcal{N}}}{}{#2}{}{#1}}}
\newrobustcmd{\bmmathcalNz}[2][]{\ensuremath{\subp{\bm{\mathcal{N}}}{}{#2}{}{#1}}}
\newrobustcmd{\hatbmmathcalNz}[2][]{\ensuremath{\subp{\hat{\bm{\mathcal{N}}}}{}{#2}{}{#1}}}
\newrobustcmd{\widehatbmmathcalNz}[2][]{\ensuremath{\subp{\widehat{\bm{\mathcal{N}}}}{}{#2}{}{#1}}}
\newrobustcmd{\checkbmmathcalNz}[2][]{\ensuremath{\subp{\check{\bm{\mathcal{N}}}}{}{#2}{}{#1}}}
\newrobustcmd{\tildebmmathcalNz}[2][]{\ensuremath{\subp{\tilde{\bm{\mathcal{N}}}}{}{#2}{}{#1}}}
\newrobustcmd{\widetildebmmathcalNz}[2][]{\ensuremath{\subp{\widetilde{\bm{\mathcal{N}}}}{}{#2}{}{#1}}}
\newrobustcmd{\acutebmmathcalNz}[2][]{\ensuremath{\subp{\acute{\bm{\mathcal{N}}}}{}{#2}{}{#1}}}
\newrobustcmd{\gravebmmathcalNz}[2][]{\ensuremath{\subp{\grave{\bm{\mathcal{N}}}}{}{#2}{}{#1}}}
\newrobustcmd{\dotbmmathcalNz}[2][]{\ensuremath{\subp{\dot{\bm{\mathcal{N}}}}{}{#2}{}{#1}}}
\newrobustcmd{\ddotbmmathcalNz}[2][]{\ensuremath{\subp{\ddot{\bm{\mathcal{N}}}}{}{#2}{}{#1}}}
\newrobustcmd{\brevebmmathcalNz}[2][]{\ensuremath{\subp{\breve{\bm{\mathcal{N}}}}{}{#2}{}{#1}}}
\newrobustcmd{\barbmmathcalNz}[2][]{\ensuremath{\subp{\bar{\bm{\mathcal{N}}}}{}{#2}{}{#1}}}
\newrobustcmd{\vecbmmathcalNz}[2][]{\ensuremath{\subp{\vec{\bm{\mathcal{N}}}}{}{#2}{}{#1}}}
\newrobustcmd{\mathcalOz}[2][]{\ensuremath{\subp{\mathcal{O}}{}{#2}{}{#1}}}
\newrobustcmd{\hatmathcalOz}[2][]{\ensuremath{\subp{\hat{\mathcal{O}}}{}{#2}{}{#1}}}
\newrobustcmd{\widehatmathcalOz}[2][]{\ensuremath{\subp{\widehat{\mathcal{O}}}{}{#2}{}{#1}}}
\newrobustcmd{\checkmathcalOz}[2][]{\ensuremath{\subp{\check{\mathcal{O}}}{}{#2}{}{#1}}}
\newrobustcmd{\tildemathcalOz}[2][]{\ensuremath{\subp{\tilde{\mathcal{O}}}{}{#2}{}{#1}}}
\newrobustcmd{\widetildemathcalOz}[2][]{\ensuremath{\subp{\widetilde{\mathcal{O}}}{}{#2}{}{#1}}}
\newrobustcmd{\acutemathcalOz}[2][]{\ensuremath{\subp{\acute{\mathcal{O}}}{}{#2}{}{#1}}}
\newrobustcmd{\gravemathcalOz}[2][]{\ensuremath{\subp{\grave{\mathcal{O}}}{}{#2}{}{#1}}}
\newrobustcmd{\dotmathcalOz}[2][]{\ensuremath{\subp{\dot{\mathcal{O}}}{}{#2}{}{#1}}}
\newrobustcmd{\ddotmathcalOz}[2][]{\ensuremath{\subp{\ddot{\mathcal{O}}}{}{#2}{}{#1}}}
\newrobustcmd{\brevemathcalOz}[2][]{\ensuremath{\subp{\breve{\mathcal{O}}}{}{#2}{}{#1}}}
\newrobustcmd{\barmathcalOz}[2][]{\ensuremath{\subp{\bar{\mathcal{O}}}{}{#2}{}{#1}}}
\newrobustcmd{\vecmathcalOz}[2][]{\ensuremath{\subp{\vec{\mathcal{O}}}{}{#2}{}{#1}}}
\newrobustcmd{\bmmathcalOz}[2][]{\ensuremath{\subp{\bm{\mathcal{O}}}{}{#2}{}{#1}}}
\newrobustcmd{\hatbmmathcalOz}[2][]{\ensuremath{\subp{\hat{\bm{\mathcal{O}}}}{}{#2}{}{#1}}}
\newrobustcmd{\widehatbmmathcalOz}[2][]{\ensuremath{\subp{\widehat{\bm{\mathcal{O}}}}{}{#2}{}{#1}}}
\newrobustcmd{\checkbmmathcalOz}[2][]{\ensuremath{\subp{\check{\bm{\mathcal{O}}}}{}{#2}{}{#1}}}
\newrobustcmd{\tildebmmathcalOz}[2][]{\ensuremath{\subp{\tilde{\bm{\mathcal{O}}}}{}{#2}{}{#1}}}
\newrobustcmd{\widetildebmmathcalOz}[2][]{\ensuremath{\subp{\widetilde{\bm{\mathcal{O}}}}{}{#2}{}{#1}}}
\newrobustcmd{\acutebmmathcalOz}[2][]{\ensuremath{\subp{\acute{\bm{\mathcal{O}}}}{}{#2}{}{#1}}}
\newrobustcmd{\gravebmmathcalOz}[2][]{\ensuremath{\subp{\grave{\bm{\mathcal{O}}}}{}{#2}{}{#1}}}
\newrobustcmd{\dotbmmathcalOz}[2][]{\ensuremath{\subp{\dot{\bm{\mathcal{O}}}}{}{#2}{}{#1}}}
\newrobustcmd{\ddotbmmathcalOz}[2][]{\ensuremath{\subp{\ddot{\bm{\mathcal{O}}}}{}{#2}{}{#1}}}
\newrobustcmd{\brevebmmathcalOz}[2][]{\ensuremath{\subp{\breve{\bm{\mathcal{O}}}}{}{#2}{}{#1}}}
\newrobustcmd{\barbmmathcalOz}[2][]{\ensuremath{\subp{\bar{\bm{\mathcal{O}}}}{}{#2}{}{#1}}}
\newrobustcmd{\vecbmmathcalOz}[2][]{\ensuremath{\subp{\vec{\bm{\mathcal{O}}}}{}{#2}{}{#1}}}
\newrobustcmd{\mathcalPz}[2][]{\ensuremath{\subp{\mathcal{P}}{}{#2}{}{#1}}}
\newrobustcmd{\hatmathcalPz}[2][]{\ensuremath{\subp{\hat{\mathcal{P}}}{}{#2}{}{#1}}}
\newrobustcmd{\widehatmathcalPz}[2][]{\ensuremath{\subp{\widehat{\mathcal{P}}}{}{#2}{}{#1}}}
\newrobustcmd{\checkmathcalPz}[2][]{\ensuremath{\subp{\check{\mathcal{P}}}{}{#2}{}{#1}}}
\newrobustcmd{\tildemathcalPz}[2][]{\ensuremath{\subp{\tilde{\mathcal{P}}}{}{#2}{}{#1}}}
\newrobustcmd{\widetildemathcalPz}[2][]{\ensuremath{\subp{\widetilde{\mathcal{P}}}{}{#2}{}{#1}}}
\newrobustcmd{\acutemathcalPz}[2][]{\ensuremath{\subp{\acute{\mathcal{P}}}{}{#2}{}{#1}}}
\newrobustcmd{\gravemathcalPz}[2][]{\ensuremath{\subp{\grave{\mathcal{P}}}{}{#2}{}{#1}}}
\newrobustcmd{\dotmathcalPz}[2][]{\ensuremath{\subp{\dot{\mathcal{P}}}{}{#2}{}{#1}}}
\newrobustcmd{\ddotmathcalPz}[2][]{\ensuremath{\subp{\ddot{\mathcal{P}}}{}{#2}{}{#1}}}
\newrobustcmd{\brevemathcalPz}[2][]{\ensuremath{\subp{\breve{\mathcal{P}}}{}{#2}{}{#1}}}
\newrobustcmd{\barmathcalPz}[2][]{\ensuremath{\subp{\bar{\mathcal{P}}}{}{#2}{}{#1}}}
\newrobustcmd{\vecmathcalPz}[2][]{\ensuremath{\subp{\vec{\mathcal{P}}}{}{#2}{}{#1}}}
\newrobustcmd{\bmmathcalPz}[2][]{\ensuremath{\subp{\bm{\mathcal{P}}}{}{#2}{}{#1}}}
\newrobustcmd{\hatbmmathcalPz}[2][]{\ensuremath{\subp{\hat{\bm{\mathcal{P}}}}{}{#2}{}{#1}}}
\newrobustcmd{\widehatbmmathcalPz}[2][]{\ensuremath{\subp{\widehat{\bm{\mathcal{P}}}}{}{#2}{}{#1}}}
\newrobustcmd{\checkbmmathcalPz}[2][]{\ensuremath{\subp{\check{\bm{\mathcal{P}}}}{}{#2}{}{#1}}}
\newrobustcmd{\tildebmmathcalPz}[2][]{\ensuremath{\subp{\tilde{\bm{\mathcal{P}}}}{}{#2}{}{#1}}}
\newrobustcmd{\widetildebmmathcalPz}[2][]{\ensuremath{\subp{\widetilde{\bm{\mathcal{P}}}}{}{#2}{}{#1}}}
\newrobustcmd{\acutebmmathcalPz}[2][]{\ensuremath{\subp{\acute{\bm{\mathcal{P}}}}{}{#2}{}{#1}}}
\newrobustcmd{\gravebmmathcalPz}[2][]{\ensuremath{\subp{\grave{\bm{\mathcal{P}}}}{}{#2}{}{#1}}}
\newrobustcmd{\dotbmmathcalPz}[2][]{\ensuremath{\subp{\dot{\bm{\mathcal{P}}}}{}{#2}{}{#1}}}
\newrobustcmd{\ddotbmmathcalPz}[2][]{\ensuremath{\subp{\ddot{\bm{\mathcal{P}}}}{}{#2}{}{#1}}}
\newrobustcmd{\brevebmmathcalPz}[2][]{\ensuremath{\subp{\breve{\bm{\mathcal{P}}}}{}{#2}{}{#1}}}
\newrobustcmd{\barbmmathcalPz}[2][]{\ensuremath{\subp{\bar{\bm{\mathcal{P}}}}{}{#2}{}{#1}}}
\newrobustcmd{\vecbmmathcalPz}[2][]{\ensuremath{\subp{\vec{\bm{\mathcal{P}}}}{}{#2}{}{#1}}}
\newrobustcmd{\mathcalQz}[2][]{\ensuremath{\subp{\mathcal{Q}}{}{#2}{}{#1}}}
\newrobustcmd{\hatmathcalQz}[2][]{\ensuremath{\subp{\hat{\mathcal{Q}}}{}{#2}{}{#1}}}
\newrobustcmd{\widehatmathcalQz}[2][]{\ensuremath{\subp{\widehat{\mathcal{Q}}}{}{#2}{}{#1}}}
\newrobustcmd{\checkmathcalQz}[2][]{\ensuremath{\subp{\check{\mathcal{Q}}}{}{#2}{}{#1}}}
\newrobustcmd{\tildemathcalQz}[2][]{\ensuremath{\subp{\tilde{\mathcal{Q}}}{}{#2}{}{#1}}}
\newrobustcmd{\widetildemathcalQz}[2][]{\ensuremath{\subp{\widetilde{\mathcal{Q}}}{}{#2}{}{#1}}}
\newrobustcmd{\acutemathcalQz}[2][]{\ensuremath{\subp{\acute{\mathcal{Q}}}{}{#2}{}{#1}}}
\newrobustcmd{\gravemathcalQz}[2][]{\ensuremath{\subp{\grave{\mathcal{Q}}}{}{#2}{}{#1}}}
\newrobustcmd{\dotmathcalQz}[2][]{\ensuremath{\subp{\dot{\mathcal{Q}}}{}{#2}{}{#1}}}
\newrobustcmd{\ddotmathcalQz}[2][]{\ensuremath{\subp{\ddot{\mathcal{Q}}}{}{#2}{}{#1}}}
\newrobustcmd{\brevemathcalQz}[2][]{\ensuremath{\subp{\breve{\mathcal{Q}}}{}{#2}{}{#1}}}
\newrobustcmd{\barmathcalQz}[2][]{\ensuremath{\subp{\bar{\mathcal{Q}}}{}{#2}{}{#1}}}
\newrobustcmd{\vecmathcalQz}[2][]{\ensuremath{\subp{\vec{\mathcal{Q}}}{}{#2}{}{#1}}}
\newrobustcmd{\bmmathcalQz}[2][]{\ensuremath{\subp{\bm{\mathcal{Q}}}{}{#2}{}{#1}}}
\newrobustcmd{\hatbmmathcalQz}[2][]{\ensuremath{\subp{\hat{\bm{\mathcal{Q}}}}{}{#2}{}{#1}}}
\newrobustcmd{\widehatbmmathcalQz}[2][]{\ensuremath{\subp{\widehat{\bm{\mathcal{Q}}}}{}{#2}{}{#1}}}
\newrobustcmd{\checkbmmathcalQz}[2][]{\ensuremath{\subp{\check{\bm{\mathcal{Q}}}}{}{#2}{}{#1}}}
\newrobustcmd{\tildebmmathcalQz}[2][]{\ensuremath{\subp{\tilde{\bm{\mathcal{Q}}}}{}{#2}{}{#1}}}
\newrobustcmd{\widetildebmmathcalQz}[2][]{\ensuremath{\subp{\widetilde{\bm{\mathcal{Q}}}}{}{#2}{}{#1}}}
\newrobustcmd{\acutebmmathcalQz}[2][]{\ensuremath{\subp{\acute{\bm{\mathcal{Q}}}}{}{#2}{}{#1}}}
\newrobustcmd{\gravebmmathcalQz}[2][]{\ensuremath{\subp{\grave{\bm{\mathcal{Q}}}}{}{#2}{}{#1}}}
\newrobustcmd{\dotbmmathcalQz}[2][]{\ensuremath{\subp{\dot{\bm{\mathcal{Q}}}}{}{#2}{}{#1}}}
\newrobustcmd{\ddotbmmathcalQz}[2][]{\ensuremath{\subp{\ddot{\bm{\mathcal{Q}}}}{}{#2}{}{#1}}}
\newrobustcmd{\brevebmmathcalQz}[2][]{\ensuremath{\subp{\breve{\bm{\mathcal{Q}}}}{}{#2}{}{#1}}}
\newrobustcmd{\barbmmathcalQz}[2][]{\ensuremath{\subp{\bar{\bm{\mathcal{Q}}}}{}{#2}{}{#1}}}
\newrobustcmd{\vecbmmathcalQz}[2][]{\ensuremath{\subp{\vec{\bm{\mathcal{Q}}}}{}{#2}{}{#1}}}
\newrobustcmd{\mathcalRz}[2][]{\ensuremath{\subp{\mathcal{R}}{}{#2}{}{#1}}}
\newrobustcmd{\hatmathcalRz}[2][]{\ensuremath{\subp{\hat{\mathcal{R}}}{}{#2}{}{#1}}}
\newrobustcmd{\widehatmathcalRz}[2][]{\ensuremath{\subp{\widehat{\mathcal{R}}}{}{#2}{}{#1}}}
\newrobustcmd{\checkmathcalRz}[2][]{\ensuremath{\subp{\check{\mathcal{R}}}{}{#2}{}{#1}}}
\newrobustcmd{\tildemathcalRz}[2][]{\ensuremath{\subp{\tilde{\mathcal{R}}}{}{#2}{}{#1}}}
\newrobustcmd{\widetildemathcalRz}[2][]{\ensuremath{\subp{\widetilde{\mathcal{R}}}{}{#2}{}{#1}}}
\newrobustcmd{\acutemathcalRz}[2][]{\ensuremath{\subp{\acute{\mathcal{R}}}{}{#2}{}{#1}}}
\newrobustcmd{\gravemathcalRz}[2][]{\ensuremath{\subp{\grave{\mathcal{R}}}{}{#2}{}{#1}}}
\newrobustcmd{\dotmathcalRz}[2][]{\ensuremath{\subp{\dot{\mathcal{R}}}{}{#2}{}{#1}}}
\newrobustcmd{\ddotmathcalRz}[2][]{\ensuremath{\subp{\ddot{\mathcal{R}}}{}{#2}{}{#1}}}
\newrobustcmd{\brevemathcalRz}[2][]{\ensuremath{\subp{\breve{\mathcal{R}}}{}{#2}{}{#1}}}
\newrobustcmd{\barmathcalRz}[2][]{\ensuremath{\subp{\bar{\mathcal{R}}}{}{#2}{}{#1}}}
\newrobustcmd{\vecmathcalRz}[2][]{\ensuremath{\subp{\vec{\mathcal{R}}}{}{#2}{}{#1}}}
\newrobustcmd{\bmmathcalRz}[2][]{\ensuremath{\subp{\bm{\mathcal{R}}}{}{#2}{}{#1}}}
\newrobustcmd{\hatbmmathcalRz}[2][]{\ensuremath{\subp{\hat{\bm{\mathcal{R}}}}{}{#2}{}{#1}}}
\newrobustcmd{\widehatbmmathcalRz}[2][]{\ensuremath{\subp{\widehat{\bm{\mathcal{R}}}}{}{#2}{}{#1}}}
\newrobustcmd{\checkbmmathcalRz}[2][]{\ensuremath{\subp{\check{\bm{\mathcal{R}}}}{}{#2}{}{#1}}}
\newrobustcmd{\tildebmmathcalRz}[2][]{\ensuremath{\subp{\tilde{\bm{\mathcal{R}}}}{}{#2}{}{#1}}}
\newrobustcmd{\widetildebmmathcalRz}[2][]{\ensuremath{\subp{\widetilde{\bm{\mathcal{R}}}}{}{#2}{}{#1}}}
\newrobustcmd{\acutebmmathcalRz}[2][]{\ensuremath{\subp{\acute{\bm{\mathcal{R}}}}{}{#2}{}{#1}}}
\newrobustcmd{\gravebmmathcalRz}[2][]{\ensuremath{\subp{\grave{\bm{\mathcal{R}}}}{}{#2}{}{#1}}}
\newrobustcmd{\dotbmmathcalRz}[2][]{\ensuremath{\subp{\dot{\bm{\mathcal{R}}}}{}{#2}{}{#1}}}
\newrobustcmd{\ddotbmmathcalRz}[2][]{\ensuremath{\subp{\ddot{\bm{\mathcal{R}}}}{}{#2}{}{#1}}}
\newrobustcmd{\brevebmmathcalRz}[2][]{\ensuremath{\subp{\breve{\bm{\mathcal{R}}}}{}{#2}{}{#1}}}
\newrobustcmd{\barbmmathcalRz}[2][]{\ensuremath{\subp{\bar{\bm{\mathcal{R}}}}{}{#2}{}{#1}}}
\newrobustcmd{\vecbmmathcalRz}[2][]{\ensuremath{\subp{\vec{\bm{\mathcal{R}}}}{}{#2}{}{#1}}}
\newrobustcmd{\mathcalSz}[2][]{\ensuremath{\subp{\mathcal{S}}{}{#2}{}{#1}}}
\newrobustcmd{\hatmathcalSz}[2][]{\ensuremath{\subp{\hat{\mathcal{S}}}{}{#2}{}{#1}}}
\newrobustcmd{\widehatmathcalSz}[2][]{\ensuremath{\subp{\widehat{\mathcal{S}}}{}{#2}{}{#1}}}
\newrobustcmd{\checkmathcalSz}[2][]{\ensuremath{\subp{\check{\mathcal{S}}}{}{#2}{}{#1}}}
\newrobustcmd{\tildemathcalSz}[2][]{\ensuremath{\subp{\tilde{\mathcal{S}}}{}{#2}{}{#1}}}
\newrobustcmd{\widetildemathcalSz}[2][]{\ensuremath{\subp{\widetilde{\mathcal{S}}}{}{#2}{}{#1}}}
\newrobustcmd{\acutemathcalSz}[2][]{\ensuremath{\subp{\acute{\mathcal{S}}}{}{#2}{}{#1}}}
\newrobustcmd{\gravemathcalSz}[2][]{\ensuremath{\subp{\grave{\mathcal{S}}}{}{#2}{}{#1}}}
\newrobustcmd{\dotmathcalSz}[2][]{\ensuremath{\subp{\dot{\mathcal{S}}}{}{#2}{}{#1}}}
\newrobustcmd{\ddotmathcalSz}[2][]{\ensuremath{\subp{\ddot{\mathcal{S}}}{}{#2}{}{#1}}}
\newrobustcmd{\brevemathcalSz}[2][]{\ensuremath{\subp{\breve{\mathcal{S}}}{}{#2}{}{#1}}}
\newrobustcmd{\barmathcalSz}[2][]{\ensuremath{\subp{\bar{\mathcal{S}}}{}{#2}{}{#1}}}
\newrobustcmd{\vecmathcalSz}[2][]{\ensuremath{\subp{\vec{\mathcal{S}}}{}{#2}{}{#1}}}
\newrobustcmd{\bmmathcalSz}[2][]{\ensuremath{\subp{\bm{\mathcal{S}}}{}{#2}{}{#1}}}
\newrobustcmd{\hatbmmathcalSz}[2][]{\ensuremath{\subp{\hat{\bm{\mathcal{S}}}}{}{#2}{}{#1}}}
\newrobustcmd{\widehatbmmathcalSz}[2][]{\ensuremath{\subp{\widehat{\bm{\mathcal{S}}}}{}{#2}{}{#1}}}
\newrobustcmd{\checkbmmathcalSz}[2][]{\ensuremath{\subp{\check{\bm{\mathcal{S}}}}{}{#2}{}{#1}}}
\newrobustcmd{\tildebmmathcalSz}[2][]{\ensuremath{\subp{\tilde{\bm{\mathcal{S}}}}{}{#2}{}{#1}}}
\newrobustcmd{\widetildebmmathcalSz}[2][]{\ensuremath{\subp{\widetilde{\bm{\mathcal{S}}}}{}{#2}{}{#1}}}
\newrobustcmd{\acutebmmathcalSz}[2][]{\ensuremath{\subp{\acute{\bm{\mathcal{S}}}}{}{#2}{}{#1}}}
\newrobustcmd{\gravebmmathcalSz}[2][]{\ensuremath{\subp{\grave{\bm{\mathcal{S}}}}{}{#2}{}{#1}}}
\newrobustcmd{\dotbmmathcalSz}[2][]{\ensuremath{\subp{\dot{\bm{\mathcal{S}}}}{}{#2}{}{#1}}}
\newrobustcmd{\ddotbmmathcalSz}[2][]{\ensuremath{\subp{\ddot{\bm{\mathcal{S}}}}{}{#2}{}{#1}}}
\newrobustcmd{\brevebmmathcalSz}[2][]{\ensuremath{\subp{\breve{\bm{\mathcal{S}}}}{}{#2}{}{#1}}}
\newrobustcmd{\barbmmathcalSz}[2][]{\ensuremath{\subp{\bar{\bm{\mathcal{S}}}}{}{#2}{}{#1}}}
\newrobustcmd{\vecbmmathcalSz}[2][]{\ensuremath{\subp{\vec{\bm{\mathcal{S}}}}{}{#2}{}{#1}}}
\newrobustcmd{\mathcalTz}[2][]{\ensuremath{\subp{\mathcal{T}}{}{#2}{}{#1}}}
\newrobustcmd{\hatmathcalTz}[2][]{\ensuremath{\subp{\hat{\mathcal{T}}}{}{#2}{}{#1}}}
\newrobustcmd{\widehatmathcalTz}[2][]{\ensuremath{\subp{\widehat{\mathcal{T}}}{}{#2}{}{#1}}}
\newrobustcmd{\checkmathcalTz}[2][]{\ensuremath{\subp{\check{\mathcal{T}}}{}{#2}{}{#1}}}
\newrobustcmd{\tildemathcalTz}[2][]{\ensuremath{\subp{\tilde{\mathcal{T}}}{}{#2}{}{#1}}}
\newrobustcmd{\widetildemathcalTz}[2][]{\ensuremath{\subp{\widetilde{\mathcal{T}}}{}{#2}{}{#1}}}
\newrobustcmd{\acutemathcalTz}[2][]{\ensuremath{\subp{\acute{\mathcal{T}}}{}{#2}{}{#1}}}
\newrobustcmd{\gravemathcalTz}[2][]{\ensuremath{\subp{\grave{\mathcal{T}}}{}{#2}{}{#1}}}
\newrobustcmd{\dotmathcalTz}[2][]{\ensuremath{\subp{\dot{\mathcal{T}}}{}{#2}{}{#1}}}
\newrobustcmd{\ddotmathcalTz}[2][]{\ensuremath{\subp{\ddot{\mathcal{T}}}{}{#2}{}{#1}}}
\newrobustcmd{\brevemathcalTz}[2][]{\ensuremath{\subp{\breve{\mathcal{T}}}{}{#2}{}{#1}}}
\newrobustcmd{\barmathcalTz}[2][]{\ensuremath{\subp{\bar{\mathcal{T}}}{}{#2}{}{#1}}}
\newrobustcmd{\vecmathcalTz}[2][]{\ensuremath{\subp{\vec{\mathcal{T}}}{}{#2}{}{#1}}}
\newrobustcmd{\bmmathcalTz}[2][]{\ensuremath{\subp{\bm{\mathcal{T}}}{}{#2}{}{#1}}}
\newrobustcmd{\hatbmmathcalTz}[2][]{\ensuremath{\subp{\hat{\bm{\mathcal{T}}}}{}{#2}{}{#1}}}
\newrobustcmd{\widehatbmmathcalTz}[2][]{\ensuremath{\subp{\widehat{\bm{\mathcal{T}}}}{}{#2}{}{#1}}}
\newrobustcmd{\checkbmmathcalTz}[2][]{\ensuremath{\subp{\check{\bm{\mathcal{T}}}}{}{#2}{}{#1}}}
\newrobustcmd{\tildebmmathcalTz}[2][]{\ensuremath{\subp{\tilde{\bm{\mathcal{T}}}}{}{#2}{}{#1}}}
\newrobustcmd{\widetildebmmathcalTz}[2][]{\ensuremath{\subp{\widetilde{\bm{\mathcal{T}}}}{}{#2}{}{#1}}}
\newrobustcmd{\acutebmmathcalTz}[2][]{\ensuremath{\subp{\acute{\bm{\mathcal{T}}}}{}{#2}{}{#1}}}
\newrobustcmd{\gravebmmathcalTz}[2][]{\ensuremath{\subp{\grave{\bm{\mathcal{T}}}}{}{#2}{}{#1}}}
\newrobustcmd{\dotbmmathcalTz}[2][]{\ensuremath{\subp{\dot{\bm{\mathcal{T}}}}{}{#2}{}{#1}}}
\newrobustcmd{\ddotbmmathcalTz}[2][]{\ensuremath{\subp{\ddot{\bm{\mathcal{T}}}}{}{#2}{}{#1}}}
\newrobustcmd{\brevebmmathcalTz}[2][]{\ensuremath{\subp{\breve{\bm{\mathcal{T}}}}{}{#2}{}{#1}}}
\newrobustcmd{\barbmmathcalTz}[2][]{\ensuremath{\subp{\bar{\bm{\mathcal{T}}}}{}{#2}{}{#1}}}
\newrobustcmd{\vecbmmathcalTz}[2][]{\ensuremath{\subp{\vec{\bm{\mathcal{T}}}}{}{#2}{}{#1}}}
\newrobustcmd{\mathcalUz}[2][]{\ensuremath{\subp{\mathcal{U}}{}{#2}{}{#1}}}
\newrobustcmd{\hatmathcalUz}[2][]{\ensuremath{\subp{\hat{\mathcal{U}}}{}{#2}{}{#1}}}
\newrobustcmd{\widehatmathcalUz}[2][]{\ensuremath{\subp{\widehat{\mathcal{U}}}{}{#2}{}{#1}}}
\newrobustcmd{\checkmathcalUz}[2][]{\ensuremath{\subp{\check{\mathcal{U}}}{}{#2}{}{#1}}}
\newrobustcmd{\tildemathcalUz}[2][]{\ensuremath{\subp{\tilde{\mathcal{U}}}{}{#2}{}{#1}}}
\newrobustcmd{\widetildemathcalUz}[2][]{\ensuremath{\subp{\widetilde{\mathcal{U}}}{}{#2}{}{#1}}}
\newrobustcmd{\acutemathcalUz}[2][]{\ensuremath{\subp{\acute{\mathcal{U}}}{}{#2}{}{#1}}}
\newrobustcmd{\gravemathcalUz}[2][]{\ensuremath{\subp{\grave{\mathcal{U}}}{}{#2}{}{#1}}}
\newrobustcmd{\dotmathcalUz}[2][]{\ensuremath{\subp{\dot{\mathcal{U}}}{}{#2}{}{#1}}}
\newrobustcmd{\ddotmathcalUz}[2][]{\ensuremath{\subp{\ddot{\mathcal{U}}}{}{#2}{}{#1}}}
\newrobustcmd{\brevemathcalUz}[2][]{\ensuremath{\subp{\breve{\mathcal{U}}}{}{#2}{}{#1}}}
\newrobustcmd{\barmathcalUz}[2][]{\ensuremath{\subp{\bar{\mathcal{U}}}{}{#2}{}{#1}}}
\newrobustcmd{\vecmathcalUz}[2][]{\ensuremath{\subp{\vec{\mathcal{U}}}{}{#2}{}{#1}}}
\newrobustcmd{\bmmathcalUz}[2][]{\ensuremath{\subp{\bm{\mathcal{U}}}{}{#2}{}{#1}}}
\newrobustcmd{\hatbmmathcalUz}[2][]{\ensuremath{\subp{\hat{\bm{\mathcal{U}}}}{}{#2}{}{#1}}}
\newrobustcmd{\widehatbmmathcalUz}[2][]{\ensuremath{\subp{\widehat{\bm{\mathcal{U}}}}{}{#2}{}{#1}}}
\newrobustcmd{\checkbmmathcalUz}[2][]{\ensuremath{\subp{\check{\bm{\mathcal{U}}}}{}{#2}{}{#1}}}
\newrobustcmd{\tildebmmathcalUz}[2][]{\ensuremath{\subp{\tilde{\bm{\mathcal{U}}}}{}{#2}{}{#1}}}
\newrobustcmd{\widetildebmmathcalUz}[2][]{\ensuremath{\subp{\widetilde{\bm{\mathcal{U}}}}{}{#2}{}{#1}}}
\newrobustcmd{\acutebmmathcalUz}[2][]{\ensuremath{\subp{\acute{\bm{\mathcal{U}}}}{}{#2}{}{#1}}}
\newrobustcmd{\gravebmmathcalUz}[2][]{\ensuremath{\subp{\grave{\bm{\mathcal{U}}}}{}{#2}{}{#1}}}
\newrobustcmd{\dotbmmathcalUz}[2][]{\ensuremath{\subp{\dot{\bm{\mathcal{U}}}}{}{#2}{}{#1}}}
\newrobustcmd{\ddotbmmathcalUz}[2][]{\ensuremath{\subp{\ddot{\bm{\mathcal{U}}}}{}{#2}{}{#1}}}
\newrobustcmd{\brevebmmathcalUz}[2][]{\ensuremath{\subp{\breve{\bm{\mathcal{U}}}}{}{#2}{}{#1}}}
\newrobustcmd{\barbmmathcalUz}[2][]{\ensuremath{\subp{\bar{\bm{\mathcal{U}}}}{}{#2}{}{#1}}}
\newrobustcmd{\vecbmmathcalUz}[2][]{\ensuremath{\subp{\vec{\bm{\mathcal{U}}}}{}{#2}{}{#1}}}
\newrobustcmd{\mathcalVz}[2][]{\ensuremath{\subp{\mathcal{V}}{}{#2}{}{#1}}}
\newrobustcmd{\hatmathcalVz}[2][]{\ensuremath{\subp{\hat{\mathcal{V}}}{}{#2}{}{#1}}}
\newrobustcmd{\widehatmathcalVz}[2][]{\ensuremath{\subp{\widehat{\mathcal{V}}}{}{#2}{}{#1}}}
\newrobustcmd{\checkmathcalVz}[2][]{\ensuremath{\subp{\check{\mathcal{V}}}{}{#2}{}{#1}}}
\newrobustcmd{\tildemathcalVz}[2][]{\ensuremath{\subp{\tilde{\mathcal{V}}}{}{#2}{}{#1}}}
\newrobustcmd{\widetildemathcalVz}[2][]{\ensuremath{\subp{\widetilde{\mathcal{V}}}{}{#2}{}{#1}}}
\newrobustcmd{\acutemathcalVz}[2][]{\ensuremath{\subp{\acute{\mathcal{V}}}{}{#2}{}{#1}}}
\newrobustcmd{\gravemathcalVz}[2][]{\ensuremath{\subp{\grave{\mathcal{V}}}{}{#2}{}{#1}}}
\newrobustcmd{\dotmathcalVz}[2][]{\ensuremath{\subp{\dot{\mathcal{V}}}{}{#2}{}{#1}}}
\newrobustcmd{\ddotmathcalVz}[2][]{\ensuremath{\subp{\ddot{\mathcal{V}}}{}{#2}{}{#1}}}
\newrobustcmd{\brevemathcalVz}[2][]{\ensuremath{\subp{\breve{\mathcal{V}}}{}{#2}{}{#1}}}
\newrobustcmd{\barmathcalVz}[2][]{\ensuremath{\subp{\bar{\mathcal{V}}}{}{#2}{}{#1}}}
\newrobustcmd{\vecmathcalVz}[2][]{\ensuremath{\subp{\vec{\mathcal{V}}}{}{#2}{}{#1}}}
\newrobustcmd{\bmmathcalVz}[2][]{\ensuremath{\subp{\bm{\mathcal{V}}}{}{#2}{}{#1}}}
\newrobustcmd{\hatbmmathcalVz}[2][]{\ensuremath{\subp{\hat{\bm{\mathcal{V}}}}{}{#2}{}{#1}}}
\newrobustcmd{\widehatbmmathcalVz}[2][]{\ensuremath{\subp{\widehat{\bm{\mathcal{V}}}}{}{#2}{}{#1}}}
\newrobustcmd{\checkbmmathcalVz}[2][]{\ensuremath{\subp{\check{\bm{\mathcal{V}}}}{}{#2}{}{#1}}}
\newrobustcmd{\tildebmmathcalVz}[2][]{\ensuremath{\subp{\tilde{\bm{\mathcal{V}}}}{}{#2}{}{#1}}}
\newrobustcmd{\widetildebmmathcalVz}[2][]{\ensuremath{\subp{\widetilde{\bm{\mathcal{V}}}}{}{#2}{}{#1}}}
\newrobustcmd{\acutebmmathcalVz}[2][]{\ensuremath{\subp{\acute{\bm{\mathcal{V}}}}{}{#2}{}{#1}}}
\newrobustcmd{\gravebmmathcalVz}[2][]{\ensuremath{\subp{\grave{\bm{\mathcal{V}}}}{}{#2}{}{#1}}}
\newrobustcmd{\dotbmmathcalVz}[2][]{\ensuremath{\subp{\dot{\bm{\mathcal{V}}}}{}{#2}{}{#1}}}
\newrobustcmd{\ddotbmmathcalVz}[2][]{\ensuremath{\subp{\ddot{\bm{\mathcal{V}}}}{}{#2}{}{#1}}}
\newrobustcmd{\brevebmmathcalVz}[2][]{\ensuremath{\subp{\breve{\bm{\mathcal{V}}}}{}{#2}{}{#1}}}
\newrobustcmd{\barbmmathcalVz}[2][]{\ensuremath{\subp{\bar{\bm{\mathcal{V}}}}{}{#2}{}{#1}}}
\newrobustcmd{\vecbmmathcalVz}[2][]{\ensuremath{\subp{\vec{\bm{\mathcal{V}}}}{}{#2}{}{#1}}}
\newrobustcmd{\mathcalWz}[2][]{\ensuremath{\subp{\mathcal{W}}{}{#2}{}{#1}}}
\newrobustcmd{\hatmathcalWz}[2][]{\ensuremath{\subp{\hat{\mathcal{W}}}{}{#2}{}{#1}}}
\newrobustcmd{\widehatmathcalWz}[2][]{\ensuremath{\subp{\widehat{\mathcal{W}}}{}{#2}{}{#1}}}
\newrobustcmd{\checkmathcalWz}[2][]{\ensuremath{\subp{\check{\mathcal{W}}}{}{#2}{}{#1}}}
\newrobustcmd{\tildemathcalWz}[2][]{\ensuremath{\subp{\tilde{\mathcal{W}}}{}{#2}{}{#1}}}
\newrobustcmd{\widetildemathcalWz}[2][]{\ensuremath{\subp{\widetilde{\mathcal{W}}}{}{#2}{}{#1}}}
\newrobustcmd{\acutemathcalWz}[2][]{\ensuremath{\subp{\acute{\mathcal{W}}}{}{#2}{}{#1}}}
\newrobustcmd{\gravemathcalWz}[2][]{\ensuremath{\subp{\grave{\mathcal{W}}}{}{#2}{}{#1}}}
\newrobustcmd{\dotmathcalWz}[2][]{\ensuremath{\subp{\dot{\mathcal{W}}}{}{#2}{}{#1}}}
\newrobustcmd{\ddotmathcalWz}[2][]{\ensuremath{\subp{\ddot{\mathcal{W}}}{}{#2}{}{#1}}}
\newrobustcmd{\brevemathcalWz}[2][]{\ensuremath{\subp{\breve{\mathcal{W}}}{}{#2}{}{#1}}}
\newrobustcmd{\barmathcalWz}[2][]{\ensuremath{\subp{\bar{\mathcal{W}}}{}{#2}{}{#1}}}
\newrobustcmd{\vecmathcalWz}[2][]{\ensuremath{\subp{\vec{\mathcal{W}}}{}{#2}{}{#1}}}
\newrobustcmd{\bmmathcalWz}[2][]{\ensuremath{\subp{\bm{\mathcal{W}}}{}{#2}{}{#1}}}
\newrobustcmd{\hatbmmathcalWz}[2][]{\ensuremath{\subp{\hat{\bm{\mathcal{W}}}}{}{#2}{}{#1}}}
\newrobustcmd{\widehatbmmathcalWz}[2][]{\ensuremath{\subp{\widehat{\bm{\mathcal{W}}}}{}{#2}{}{#1}}}
\newrobustcmd{\checkbmmathcalWz}[2][]{\ensuremath{\subp{\check{\bm{\mathcal{W}}}}{}{#2}{}{#1}}}
\newrobustcmd{\tildebmmathcalWz}[2][]{\ensuremath{\subp{\tilde{\bm{\mathcal{W}}}}{}{#2}{}{#1}}}
\newrobustcmd{\widetildebmmathcalWz}[2][]{\ensuremath{\subp{\widetilde{\bm{\mathcal{W}}}}{}{#2}{}{#1}}}
\newrobustcmd{\acutebmmathcalWz}[2][]{\ensuremath{\subp{\acute{\bm{\mathcal{W}}}}{}{#2}{}{#1}}}
\newrobustcmd{\gravebmmathcalWz}[2][]{\ensuremath{\subp{\grave{\bm{\mathcal{W}}}}{}{#2}{}{#1}}}
\newrobustcmd{\dotbmmathcalWz}[2][]{\ensuremath{\subp{\dot{\bm{\mathcal{W}}}}{}{#2}{}{#1}}}
\newrobustcmd{\ddotbmmathcalWz}[2][]{\ensuremath{\subp{\ddot{\bm{\mathcal{W}}}}{}{#2}{}{#1}}}
\newrobustcmd{\brevebmmathcalWz}[2][]{\ensuremath{\subp{\breve{\bm{\mathcal{W}}}}{}{#2}{}{#1}}}
\newrobustcmd{\barbmmathcalWz}[2][]{\ensuremath{\subp{\bar{\bm{\mathcal{W}}}}{}{#2}{}{#1}}}
\newrobustcmd{\vecbmmathcalWz}[2][]{\ensuremath{\subp{\vec{\bm{\mathcal{W}}}}{}{#2}{}{#1}}}
\newrobustcmd{\mathcalXz}[2][]{\ensuremath{\subp{\mathcal{X}}{}{#2}{}{#1}}}
\newrobustcmd{\hatmathcalXz}[2][]{\ensuremath{\subp{\hat{\mathcal{X}}}{}{#2}{}{#1}}}
\newrobustcmd{\widehatmathcalXz}[2][]{\ensuremath{\subp{\widehat{\mathcal{X}}}{}{#2}{}{#1}}}
\newrobustcmd{\checkmathcalXz}[2][]{\ensuremath{\subp{\check{\mathcal{X}}}{}{#2}{}{#1}}}
\newrobustcmd{\tildemathcalXz}[2][]{\ensuremath{\subp{\tilde{\mathcal{X}}}{}{#2}{}{#1}}}
\newrobustcmd{\widetildemathcalXz}[2][]{\ensuremath{\subp{\widetilde{\mathcal{X}}}{}{#2}{}{#1}}}
\newrobustcmd{\acutemathcalXz}[2][]{\ensuremath{\subp{\acute{\mathcal{X}}}{}{#2}{}{#1}}}
\newrobustcmd{\gravemathcalXz}[2][]{\ensuremath{\subp{\grave{\mathcal{X}}}{}{#2}{}{#1}}}
\newrobustcmd{\dotmathcalXz}[2][]{\ensuremath{\subp{\dot{\mathcal{X}}}{}{#2}{}{#1}}}
\newrobustcmd{\ddotmathcalXz}[2][]{\ensuremath{\subp{\ddot{\mathcal{X}}}{}{#2}{}{#1}}}
\newrobustcmd{\brevemathcalXz}[2][]{\ensuremath{\subp{\breve{\mathcal{X}}}{}{#2}{}{#1}}}
\newrobustcmd{\barmathcalXz}[2][]{\ensuremath{\subp{\bar{\mathcal{X}}}{}{#2}{}{#1}}}
\newrobustcmd{\vecmathcalXz}[2][]{\ensuremath{\subp{\vec{\mathcal{X}}}{}{#2}{}{#1}}}
\newrobustcmd{\bmmathcalXz}[2][]{\ensuremath{\subp{\bm{\mathcal{X}}}{}{#2}{}{#1}}}
\newrobustcmd{\hatbmmathcalXz}[2][]{\ensuremath{\subp{\hat{\bm{\mathcal{X}}}}{}{#2}{}{#1}}}
\newrobustcmd{\widehatbmmathcalXz}[2][]{\ensuremath{\subp{\widehat{\bm{\mathcal{X}}}}{}{#2}{}{#1}}}
\newrobustcmd{\checkbmmathcalXz}[2][]{\ensuremath{\subp{\check{\bm{\mathcal{X}}}}{}{#2}{}{#1}}}
\newrobustcmd{\tildebmmathcalXz}[2][]{\ensuremath{\subp{\tilde{\bm{\mathcal{X}}}}{}{#2}{}{#1}}}
\newrobustcmd{\widetildebmmathcalXz}[2][]{\ensuremath{\subp{\widetilde{\bm{\mathcal{X}}}}{}{#2}{}{#1}}}
\newrobustcmd{\acutebmmathcalXz}[2][]{\ensuremath{\subp{\acute{\bm{\mathcal{X}}}}{}{#2}{}{#1}}}
\newrobustcmd{\gravebmmathcalXz}[2][]{\ensuremath{\subp{\grave{\bm{\mathcal{X}}}}{}{#2}{}{#1}}}
\newrobustcmd{\dotbmmathcalXz}[2][]{\ensuremath{\subp{\dot{\bm{\mathcal{X}}}}{}{#2}{}{#1}}}
\newrobustcmd{\ddotbmmathcalXz}[2][]{\ensuremath{\subp{\ddot{\bm{\mathcal{X}}}}{}{#2}{}{#1}}}
\newrobustcmd{\brevebmmathcalXz}[2][]{\ensuremath{\subp{\breve{\bm{\mathcal{X}}}}{}{#2}{}{#1}}}
\newrobustcmd{\barbmmathcalXz}[2][]{\ensuremath{\subp{\bar{\bm{\mathcal{X}}}}{}{#2}{}{#1}}}
\newrobustcmd{\vecbmmathcalXz}[2][]{\ensuremath{\subp{\vec{\bm{\mathcal{X}}}}{}{#2}{}{#1}}}
\newrobustcmd{\mathcalYz}[2][]{\ensuremath{\subp{\mathcal{Y}}{}{#2}{}{#1}}}
\newrobustcmd{\hatmathcalYz}[2][]{\ensuremath{\subp{\hat{\mathcal{Y}}}{}{#2}{}{#1}}}
\newrobustcmd{\widehatmathcalYz}[2][]{\ensuremath{\subp{\widehat{\mathcal{Y}}}{}{#2}{}{#1}}}
\newrobustcmd{\checkmathcalYz}[2][]{\ensuremath{\subp{\check{\mathcal{Y}}}{}{#2}{}{#1}}}
\newrobustcmd{\tildemathcalYz}[2][]{\ensuremath{\subp{\tilde{\mathcal{Y}}}{}{#2}{}{#1}}}
\newrobustcmd{\widetildemathcalYz}[2][]{\ensuremath{\subp{\widetilde{\mathcal{Y}}}{}{#2}{}{#1}}}
\newrobustcmd{\acutemathcalYz}[2][]{\ensuremath{\subp{\acute{\mathcal{Y}}}{}{#2}{}{#1}}}
\newrobustcmd{\gravemathcalYz}[2][]{\ensuremath{\subp{\grave{\mathcal{Y}}}{}{#2}{}{#1}}}
\newrobustcmd{\dotmathcalYz}[2][]{\ensuremath{\subp{\dot{\mathcal{Y}}}{}{#2}{}{#1}}}
\newrobustcmd{\ddotmathcalYz}[2][]{\ensuremath{\subp{\ddot{\mathcal{Y}}}{}{#2}{}{#1}}}
\newrobustcmd{\brevemathcalYz}[2][]{\ensuremath{\subp{\breve{\mathcal{Y}}}{}{#2}{}{#1}}}
\newrobustcmd{\barmathcalYz}[2][]{\ensuremath{\subp{\bar{\mathcal{Y}}}{}{#2}{}{#1}}}
\newrobustcmd{\vecmathcalYz}[2][]{\ensuremath{\subp{\vec{\mathcal{Y}}}{}{#2}{}{#1}}}
\newrobustcmd{\bmmathcalYz}[2][]{\ensuremath{\subp{\bm{\mathcal{Y}}}{}{#2}{}{#1}}}
\newrobustcmd{\hatbmmathcalYz}[2][]{\ensuremath{\subp{\hat{\bm{\mathcal{Y}}}}{}{#2}{}{#1}}}
\newrobustcmd{\widehatbmmathcalYz}[2][]{\ensuremath{\subp{\widehat{\bm{\mathcal{Y}}}}{}{#2}{}{#1}}}
\newrobustcmd{\checkbmmathcalYz}[2][]{\ensuremath{\subp{\check{\bm{\mathcal{Y}}}}{}{#2}{}{#1}}}
\newrobustcmd{\tildebmmathcalYz}[2][]{\ensuremath{\subp{\tilde{\bm{\mathcal{Y}}}}{}{#2}{}{#1}}}
\newrobustcmd{\widetildebmmathcalYz}[2][]{\ensuremath{\subp{\widetilde{\bm{\mathcal{Y}}}}{}{#2}{}{#1}}}
\newrobustcmd{\acutebmmathcalYz}[2][]{\ensuremath{\subp{\acute{\bm{\mathcal{Y}}}}{}{#2}{}{#1}}}
\newrobustcmd{\gravebmmathcalYz}[2][]{\ensuremath{\subp{\grave{\bm{\mathcal{Y}}}}{}{#2}{}{#1}}}
\newrobustcmd{\dotbmmathcalYz}[2][]{\ensuremath{\subp{\dot{\bm{\mathcal{Y}}}}{}{#2}{}{#1}}}
\newrobustcmd{\ddotbmmathcalYz}[2][]{\ensuremath{\subp{\ddot{\bm{\mathcal{Y}}}}{}{#2}{}{#1}}}
\newrobustcmd{\brevebmmathcalYz}[2][]{\ensuremath{\subp{\breve{\bm{\mathcal{Y}}}}{}{#2}{}{#1}}}
\newrobustcmd{\barbmmathcalYz}[2][]{\ensuremath{\subp{\bar{\bm{\mathcal{Y}}}}{}{#2}{}{#1}}}
\newrobustcmd{\vecbmmathcalYz}[2][]{\ensuremath{\subp{\vec{\bm{\mathcal{Y}}}}{}{#2}{}{#1}}}
\newrobustcmd{\mathcalZz}[2][]{\ensuremath{\subp{\mathcal{Z}}{}{#2}{}{#1}}}
\newrobustcmd{\hatmathcalZz}[2][]{\ensuremath{\subp{\hat{\mathcal{Z}}}{}{#2}{}{#1}}}
\newrobustcmd{\widehatmathcalZz}[2][]{\ensuremath{\subp{\widehat{\mathcal{Z}}}{}{#2}{}{#1}}}
\newrobustcmd{\checkmathcalZz}[2][]{\ensuremath{\subp{\check{\mathcal{Z}}}{}{#2}{}{#1}}}
\newrobustcmd{\tildemathcalZz}[2][]{\ensuremath{\subp{\tilde{\mathcal{Z}}}{}{#2}{}{#1}}}
\newrobustcmd{\widetildemathcalZz}[2][]{\ensuremath{\subp{\widetilde{\mathcal{Z}}}{}{#2}{}{#1}}}
\newrobustcmd{\acutemathcalZz}[2][]{\ensuremath{\subp{\acute{\mathcal{Z}}}{}{#2}{}{#1}}}
\newrobustcmd{\gravemathcalZz}[2][]{\ensuremath{\subp{\grave{\mathcal{Z}}}{}{#2}{}{#1}}}
\newrobustcmd{\dotmathcalZz}[2][]{\ensuremath{\subp{\dot{\mathcal{Z}}}{}{#2}{}{#1}}}
\newrobustcmd{\ddotmathcalZz}[2][]{\ensuremath{\subp{\ddot{\mathcal{Z}}}{}{#2}{}{#1}}}
\newrobustcmd{\brevemathcalZz}[2][]{\ensuremath{\subp{\breve{\mathcal{Z}}}{}{#2}{}{#1}}}
\newrobustcmd{\barmathcalZz}[2][]{\ensuremath{\subp{\bar{\mathcal{Z}}}{}{#2}{}{#1}}}
\newrobustcmd{\vecmathcalZz}[2][]{\ensuremath{\subp{\vec{\mathcal{Z}}}{}{#2}{}{#1}}}
\newrobustcmd{\bmmathcalZz}[2][]{\ensuremath{\subp{\bm{\mathcal{Z}}}{}{#2}{}{#1}}}
\newrobustcmd{\hatbmmathcalZz}[2][]{\ensuremath{\subp{\hat{\bm{\mathcal{Z}}}}{}{#2}{}{#1}}}
\newrobustcmd{\widehatbmmathcalZz}[2][]{\ensuremath{\subp{\widehat{\bm{\mathcal{Z}}}}{}{#2}{}{#1}}}
\newrobustcmd{\checkbmmathcalZz}[2][]{\ensuremath{\subp{\check{\bm{\mathcal{Z}}}}{}{#2}{}{#1}}}
\newrobustcmd{\tildebmmathcalZz}[2][]{\ensuremath{\subp{\tilde{\bm{\mathcal{Z}}}}{}{#2}{}{#1}}}
\newrobustcmd{\widetildebmmathcalZz}[2][]{\ensuremath{\subp{\widetilde{\bm{\mathcal{Z}}}}{}{#2}{}{#1}}}
\newrobustcmd{\acutebmmathcalZz}[2][]{\ensuremath{\subp{\acute{\bm{\mathcal{Z}}}}{}{#2}{}{#1}}}
\newrobustcmd{\gravebmmathcalZz}[2][]{\ensuremath{\subp{\grave{\bm{\mathcal{Z}}}}{}{#2}{}{#1}}}
\newrobustcmd{\dotbmmathcalZz}[2][]{\ensuremath{\subp{\dot{\bm{\mathcal{Z}}}}{}{#2}{}{#1}}}
\newrobustcmd{\ddotbmmathcalZz}[2][]{\ensuremath{\subp{\ddot{\bm{\mathcal{Z}}}}{}{#2}{}{#1}}}
\newrobustcmd{\brevebmmathcalZz}[2][]{\ensuremath{\subp{\breve{\bm{\mathcal{Z}}}}{}{#2}{}{#1}}}
\newrobustcmd{\barbmmathcalZz}[2][]{\ensuremath{\subp{\bar{\bm{\mathcal{Z}}}}{}{#2}{}{#1}}}
\newrobustcmd{\vecbmmathcalZz}[2][]{\ensuremath{\subp{\vec{\bm{\mathcal{Z}}}}{}{#2}{}{#1}}}

\newrobustcmd{\ellz}[2][]{\ensuremath{\subp{\ell}{}{#2}{}{#1}}}

\newcommand{\TSR}[3]{   \mbox{\ensuremath{\subp{   \left\{ #1 \right\}}{}{#2}{}{#3}}}
}

\newcommand{\lgcor}[2][]{
  \rhoz{#2
    \IfStrEq{#1}{\empty}{\!}{|#1}}}

\newcommand{\hatlgcor}[2][]{
  \hatrhoz{#2
    \IfStrEq{#1}{\empty}{\!}{|#1}}}

\newcommand{\lgacr}[3][]{
  \rhoz{#2
    \IfStrEq{#1}{\empty}{\!}{|#1}}(#3)}

\newcommand{\lgccr}[4][]{
  \rhoz{#2:#3
    \IfStrEq{#1}{\empty}{\!}{|#1}}(#4)}

\newcommand{\hatlgacr}[3][]{
  \hatrhoz{#2
    \IfStrEq{#1}{\empty}{\!}{|#1}}(#3)}
\newcommand{\hatlgccr}[4][]{
  \hatrhoz{#2:#3
    \IfStrEq{#1}{\empty}{\!}{|#1}}(#4)}

\newcommand{\hatlgacrb}[4][]{
  \hatrhoz{\!#2
    \IfStrEq{#1}{\empty}{\!}{|#1}
  }(#3|\scalebox{.7}{$#4$})}

\newcommand{\hatlgccrb}[5][]{
  \hatrhoz{#2:#3
    \IfStrEq{#1}{\empty}{\!}{|#1}
  }(#4|\scalebox{.7}{$#5$})}

\newcommand{\hatlgthetab}[4][]{
  \hatbmthetaz{\!#2
    \IfStrEq{#1}{\empty}{\!}{|#1}
  }(#3|\scalebox{.7}{$#4$})}

\newcommand{\lgsd}[3][]{
  \fz{\!#2
    \IfStrEq{#1}{\empty}{\!}{|#1}}(#3)}

\newcommand{\lgcsd}[4][]{
  \fz{#2:#3
    \IfStrEq{#1}{\empty}{\!}{|#1}}(#4)}

\newcommand{\lgcsdCo}[4][]{
  \cz{#2:#3
    \IfStrEq{#1}{\empty}{\!}{|#1}}(#4)}
\newcommand{\lgcsdQuad}[4][]{
  \qz{#2:#3
    \IfStrEq{#1}{\empty}{\!}{|#1}}(#4)}
\newcommand{\lgcsdAmplitude}[4][]{
  \alphaz{#2:#3
    \IfStrEq{#1}{\empty}{\!}{|#1}}(#4)}
\newcommand{\lgcsdPhase}[4][]{
  \phiz{#2:#3
    \IfStrEq{#1}{\empty}{\!}{|#1}}(#4)}

\newcommand{\lgcsdCoM}[5][]{
  \cz[#5]{#2:#3
    \IfStrEq{#1}{\empty}{\!}{|#1}}(#4)}
\newcommand{\lgcsdQuadM}[5][]{
  \qz[#5]{#2:#3
    \IfStrEq{#1}{\empty}{\!}{|#1}}(#4)}
\newcommand{\lgcsdAmplitudeM}[5][]{
  \alphaz[#5]{#2:#3
    \IfStrEq{#1}{\empty}{\!}{|#1}}(#4)}
\newcommand{\lgcsdPhaseM}[5][]{
  \phiz[#5]{#2:#3
    \IfStrEq{#1}{\empty}{\!}{|#1}}(#4)}

\newcommand{\hatlgcsdCoM}[5][]{
  \hatcz[\ #5]{#2:#3
    \IfStrEq{#1}{\empty}{\!}{|#1}}(#4)}
\newcommand{\hatlgcsdQuadM}[5][]{
  \hatqz[\ #5]{#2:#3
    \IfStrEq{#1}{\empty}{\!}{|#1}}(#4)}
\newcommand{\hatlgcsdAmplitudeM}[5][]{
  \hatalphaz[\ #5]{#2:#3
    \IfStrEq{#1}{\empty}{\!}{|#1}}(#4)}
\newcommand{\hatlgcsdPhaseM}[5][]{
  \hatphiz[\ #5]{#2:#3
    \IfStrEq{#1}{\empty}{\!}{|#1}}(#4)}

\newcommand{\lgcsdSQ}[4][]{
  \mathcalKz[asc]{#2:#3
    \IfStrEq{#1}{\empty}{\!}{|#1}}(#4)}

\newcommand{\lgsdM}[4][]{
  \fz[#4]{\!#2
    \IfStrEq{#1}{\empty}{\!}{|#1}}(#3)}
\newcommand{\lgcsdM}[5][]{
  \fz[#5]{\!#2:#3
    \IfStrEq{#1}{\empty}{\!}{|#1}}(#4)}

\newcommand{\hatlgsd}[3][]{
  \hatfz{\!#2
    \IfStrEq{#1}{\empty}{\!}{|#1}}(#3)}

\newcommand{\hatlgcsd}[4][]{
  \hatfz{\!#2:#3
    \IfStrEq{#1}{\empty}{\!}{|#1}}(#4)}

\newcommand{\hatlgsdM}[4][]{
  \hatfz[#4]{\!#2
    \IfStrEq{#1}{\empty}{\!}{|#1}}(#3)}
\newcommand{\hatlgcsdM}[5][]{
  \hatfz[#5]{\!#2:#3
    \IfStrEq{#1}{\empty}{\!}{|#1}}(#4)}

\newcommand{\lgsdRE}[3][]{
  \cz{#2
    \IfStrEq{#1}{\empty}{\!}{|#1}}(#3)}
\newcommand{\hatlgsdREM}[4][]{
  \hatcz[\ #4]{#2
    \IfStrEq{#1}{\empty}{\!}{|#1}}(#3)}

\newcommand{\lgsdIM}[3][]{
  \qz{#2
    \IfStrEq{#1}{\empty}{\!}{|#1}}(#3)}
\newcommand{\hatlgsdIMM}[4][]{
  \hatqz[\ #4]{#2
    \IfStrEq{#1}{\empty}{\!}{|#1}}(#3)}

\newcommand{\overbar}[1]{\mkern 2.5mu\overline{\mkern-2.5mu#1\mkern-2.5mu}\mkern 2.5mu}

\newcommand{\Ubar}[1]{\mkern 2.5mu\underline{\mkern-2.5mu#1\mkern-2.5mu}\mkern 2.5mu}

\newcommand{\OUbar}[1]{\ensuremath{\overbar{\Ubar{#1}}}}

\newcommand{\Mblock}[1]{
  \overbar{#1}}

\newcommand{\MMAS}[8]{
  \ensuremath{\subp{#1}{}{                 #2
                \ifstrempty{#4}{#3}{\Mblock{#3}}
                        #5#6}{}{}{}
            \IfStrEq{#8}{}{}{#7\left(#8\right)} 
          }}

\newcommand{\LGpsymbol}{v}
\newcommand{\LGp}{\ensuremath{\bm{\LGpsymbol}}}
\newcommand{\LGpd}{\ensuremath{\breve{\LGp}}}
\newcommand{\LGpi}[2][]{
  \subp{\LGpsymbol}{}{#2}{}{#1}
}
\newcommand{\LGpoint}{
  \ensuremath{\parenR{\LGpi{1},\LGpi{2}}}}
\newcommand{\LGpointd}{
  \ensuremath{\parenR{\LGpi{2},\LGpi{1}}}}

\newrobustcmd{\Yht}[2]{
  \MMAS{\bm{Y}}{}{#1}{}{:}{#2}{}{} }
\newrobustcmd{\Yhtc}[3]{
  \MMAS{\bm{Y}}{}{#1:#2}{}{:}{#3}{}{} }
\newrobustcmd{\YMt}[2]{
  \MMAS{\bm{Y}}{}{#1}{-}{:}{#2}{}{} }
\newrobustcmd{\yh}[1]{
  \MMAS{\bm{y}}{}{#1}{}{}{}{}{} }
\newrobustcmd{\yht}[2]{
  \MMAS{\bm{y}}{}{#1}{}{:}{#2}{}{} }
\newrobustcmd{\yM}[1]{
  \MMAS{\bm{y}}{}{#1}{-}{}{}{}{} }
\newrobustcmd{\yMt}[2]{
  \MMAS{\bm{y}}{}{#1}{-}{:}{#2}{}{} }

\newrobustcmd{\yhc}[2]{
  \MMAS{\bm{y}}{}{#1:#2}{}{}{}{}{} }

\newrobustcmd{\dyh}[1]{
  \MMAS{\d{\bm{y}}}{}{#1}{}{}{}{}{} }
\newrobustcmd{\dyhc}[2]{
  \MMAS{\d{\bm{y}}}{}{#1:#2}{}{}{}{}{} }
\newrobustcmd{\dyM}[1]{
  \MMAS{\d{\bm{y}}}{}{#1}{-}{}{}{}{} }
\newrobustcmd{\dyMh}[2]{
  \MMAS{\d{\bm{y}}}{}{#1}{-}{/}{#2}{}{} }
\newrobustcmd{\gh}[2][]{
  \MMAS{g}{}{#2}{}{}{}{\!}{#1} }
\newrobustcmd{\ght}[3][]{
  \MMAS{g}{}{#2}{}{:}{#3}{\!}{#1} }
\newrobustcmd{\gM}[2][]{
  \MMAS{g}{}{#2}{-}{}{}{\!}{#1} }
\newrobustcmd{\gMt}[3][]{
  \MMAS{g}{}{#2}{-}{:}{#3}{\!}{#1} }

\newrobustcmd{\ghc}[3][]{
  \MMAS{g}{}{#2:#3}{}{}{}{\!}{#1} }

\newrobustcmd{\Xhit}[3]{
  \MMAS{X}{}{#1#2}{}{:}{#3}{}{} }
\newrobustcmd{\Xat}[2][]{
  \MMAS{X}{}{#2}{}{}{}{\!}{\bm{#1}} }
\newrobustcmd{\Xt}[1]{
  \MMAS{\bm{X}}{}{#1}{}{}{}{\!}{} }
\newrobustcmd{\Xht}[2]{
  \MMAS{\bm{X}}{}{#1}{}{:}{#2}{\!}{} }

\newcommand{\absp}[2][]{  \ensuremath{\subp{\left|#2\right|}{}{}{}{#1}} }

\newcommand{\thetahcomp}[3][]{
  \thetaz{
    \IfStrEq{#1}{\empty}{}{\!#1|}
    #3}}

\newcommand{\thetah}[2][]{
  {\bmthetaz{
    \IfStrEq{#1}{\empty}{}{\!#1|}
    #2}}}

\newcommand{\thetahc}[4]{
  {\bmthetaz{
      \IfStrEq{#2}{\empty}{}{\!#2|}
      #3:#4
      \IfStrEq{#1}{\empty}{}{|#1}
}}}

\newcommand{\hatthetahc}[4]{
  {\hatbmthetaz{
      \IfStrEq{#2}{\empty}{}{\!#2|}
      #3:#4
      \IfStrEq{#1}{\empty}{}{|#1}
}}
}

\newcommand{\thetahbc}[5]{
  \MMAS{\bm{\theta}}{}{#2|#3:#4}{}{:}{\bm{#5}
    \IfStrEq{#1}{\empty}{}{|#1}
  }{}{} }

\newcommand{\hatthetahNc}[6][]{
    \MMAS{\hat{\bm{\theta}}}{}{#3|#4:#5}{}{:}{#6
      \IfStrEq{#2}{\empty}{}{|#2}
    }{\!}{#1} }
\newcommand{\hatthetaMN}[3][]{
  \MMAS{\hat{\bm{\theta}}}{}{#2}{-}
  {\IfStrEq{#3}{}{}{:}}
  {\IfStrEq{#3}{}{}{#3}}
  {\!}{#1} }

\newcommand{\nablahc}[4][]{
  \ensuremath{\subp{\bm{\nabla}}{\!\!}{#3:#4
      \IfStrEq{#2}{\empty}{}{|#2}}{}{#1}}}

\newcommand{\nablaMc}[4][]{
  \ensuremath{\subp{\bm{\nabla}}{\!\!}{#3:\OUbar{#4}
      \IfStrEq{#2}{\empty}{}{|#2}}{}{#1}}}

\newcommand{\Uh}[2][]{
  \MMAS{u}{}{#2}{}{}{}{\!}{#1} }
\newrobustcmd{\UhL}[3][]{
  \MMAS{u}{}{#2}{}{}{}{\!}{#1}^{{}_{#3}} }

\newcommand{\Kh}[2][]{\ensuremath{\Kz{\!#2}
    \IfStrEq{#1}{}{}{\!\left(#1\right)} }}
\newcommand{\Khb}[3][]{\ensuremath{K_{\overset{#2:#3}{}}
    \IfStrEq{#1}{}{}{\!\left(#1\right)} }}
\newcommand{\Khbc}[4][]{\ensuremath{K_{\overset{#2:#3:#4}{}}
    \IfStrEq{#1}{}{}{\!\left(#1\right)} }}
\newcommand{\KhN}[3][]{\ensuremath{K_{\overset{#2:#3}{}}
    \IfStrEq{#1}{}{}{\!\left(#1\right)} }}

\newcommand{\Khbdefc}{\ensuremath{\Khbc[\yhc{k\ell}{h}-\LGp]{k\ell}{h}{\bm{b}}}}
\newcommand{\KhbDEFc}{\ensuremath{\Khbc[\Yhtc{k\ell}{h}{t}-\LGp]{k\ell}{h}{\bm{b}}}}

\newcommand{\qhc}[5]{
  \MMAS{q}{}{#2|#3:#4}{}{:}{\bm{#5}
    \IfStrEq{#1}{}{}{|#1}}{}{} }

\newcommand{\LhNc}[6][]{
  \MMAS{L}{}{#3|#4:#5}{}{:}{#6
    \IfStrEq{#2}{\empty}{}{|#2}
  }{\!}{#1} }

\newcommand{\QhNc}[6][]{
  \MMAS{Q}{}{#3|#4:#5}{}{:}{#6
    \IfStrEq{#2}{\empty}{}{|#2}
  }{\!}{#1} }

\newcommand{\tildeQhNc}[6][]{
  \MMAS{\tilde{Q}}{}{#3|#4:#5}{}{:}{#6
    \IfStrEq{#2}{\empty}{}{|#2}
  }{\!}{#1} }

\newcommand{\QMNc}[7][]{
  \MMAS{Q}{}{#3:#4|#5:\OUbar{#6}}{}{:}{#7
    \IfStrEq{#2}{\empty}{}{|#2}
  }{\!}{#1} }

\newcommand{\ThN}[3][]{
    \ensuremath{\Tz{#2
      \IfStrEq{#3}{}{}{:#3}}
    \IfStrEq{#1}{}{}{\!\left(#1\right)}}
}
\newcommand{\ThNc}[6][]{
  \ensuremath{\Tz{#3|#4:#5
      \IfStrEq{#6}{}{}{:#6
        \IfStrEq{#2}{\empty}{}{|#2}
    }}
    \IfStrEq{#1}{}{}{\!\left(#1\right)}}
}

\newcommand{\TMbN}[4][]{
  \ensuremath{\Tz{\overbar{#2}
      \IfStrEq{#3}{}{}{|#3}
      \IfStrEq{#4}{}{}{:#4}}
    \IfStrEq{#1}{}{}{\!\left(#1\right)}}
}

\newcommand{\WhN}[4][]{
  \ensuremath{\Wz[
    \IfStrEq{#2}{}{}{#2}]{\!#3
            \IfStrEq{#4}{}{}{:#4}}
    \IfStrEq{#1}{}{}{\!\left(#1\right)}}
}
\newcommand{\WhNc}[7][]{
  \ensuremath{\Wz[
      \IfStrEq{#5}{}{}{#5}]{\!#3|#4:#6
      \IfStrEq{#7}{}{}{:#7}
      \IfStrEq{#2}{\empty}{}{|#2}
    }
    \IfStrEq{#1}{}{}{\!\left(#1\right)}}
}

\newcommand{\WMbN}[4][]{
  \ensuremath{\Wz{\!\overbar{#2}
      \IfStrEq{#3}{}{}{|#3}
      \IfStrEq{#4}{}{}{:#4}}
    \IfStrEq{#1}{}{}{\!\left(#1\right)}}
}

\newcommand{\VhN}[4][]{
  \ensuremath{\Vz[
    \IfStrEq{#2}{}{}{#2}]{\!#3
            \IfStrEq{#4}{}{}{:#4}}
    \IfStrEq{#1}{}{}{\!\left(#1\right)}}
}
\newcommand{\VhNc}[7][]{
  \ensuremath{\Vz[
      \IfStrEq{#5}{}{}{#5}]{\!#3|#4:#6
      \IfStrEq{#7}{}{}{:#7}
      \IfStrEq{#2}{\empty}{}{|#2}
    }
    \IfStrEq{#1}{}{}{\!\left(#1\right)}}
}

\newcommand{\VMbN}[4][]{
  \ensuremath{\Vz{\overbar{#2}
      \IfStrEq{#3}{}{}{|#3}
      \IfStrEq{#4}{}{}{:#4}}
    \IfStrEq{#1}{}{}{\!\left(#1\right)}}
}

\newcommand{\RRn}[1]{\ensuremath{\subp{\RR}{}{}{}{#1}}}

\crefname{algorithm}{Algorithm}{Algorithms}
\crefname{appendix}{Appendix}{Appendices}
\crefname{assumption}{Assumption}{Assumptions}
\crefname{corollary}{Corollary}{Corollaries}
\crefname{definition}{Definition}{Definitions}
\crefname{enumi}{Item}{Items}
\crefname{equation}{Equation}{Equations}
\crefname{figure}{Figure}{Figures}
\crefname{lemma}{Lemma}{Lemmas}
\crefname{remark}{Remark}{Remarks}
\crefname{table}{Table}{Tables}
\crefname{theorem}{Theorem}{Theorems}

\newrobustcmd{\uhc}[4][]{
  \MMAS{\bm{u}}{}{#3:#4
    \IfStrEq{#2}{\empty}{}{|#2}
  }{}{}{}{\!}{#1} }
\newrobustcmd{\bNi}[2][N]{
  \MMAS{b}{}{#1}{}{}{#2}{}{} }
\newrobustcmd{\aNi}[2][]{
  \MMAS{a}{}{#1}{}{}{#2}{}{} }
\newrobustcmd{\cNi}[2][]{
  \MMAS{c}{}{#1}{}{}{#2}{}{} }
\newrobustcmd{\rNi}[2][]{
  \MMAS{r}{}{#1}{}{}{#2}{}{} }
\newrobustcmd{\sNi}[2][]{
  \MMAS{s}{}{#1}{}{}{#2}{}{} }
\newrobustcmd{\ellNi}[2][]{
  \MMAS{\ell}{}{#1}{}{}{#2}{}{} }
\newrobustcmd{\XNht}[4][]{\Xz[#2\IfStrEq{#1}{\empty}{}{|#1}]{#3:#4}} 
\newrobustcmd{\tildeXNht}[4][]{\tildeXz[#2\IfStrEq{#1}{\empty}{}{|#1}]{#3:#4}} 
\newrobustcmd{\XNMt}[4][]{\Xz[#2\IfStrEq{#1}{\empty}{}{|#1}]{\overbar{#3}:#4}}

\newrobustcmd{\ZNht}[4][]{\Zz[#2\IfStrEq{#1}{\empty}{}{|#1}]{#3:#4}} 
\newrobustcmd{\ZNhtvec}[3][]{\bmZz[#2\IfStrEq{#1}{\empty}{}{|#1}]{#3}} 
\newrobustcmd{\ZNMt}[4][]{\Zz[#2\IfStrEq{#1}{\empty}{}{|#1}]{\overbar{#3}:#4}}

\newrobustcmd{\QNh}[3][]{\subp{\reflectbox{$Q$}}{}{#3}{}{#2\IfStrEq{#1}{\empty}{}{|#1}}} 
\newrobustcmd{\QNhvec}[3][]{\subp{\reflectbox{$\bm{Q}$}}{}{#3}{}{#2\IfStrEq{#1}{\empty}{}{|#1}}} 
\newrobustcmd{\QNM}[3][]{\subp{\reflectbox{$Q$}}{}{\overbar{#3}}{}{#2\IfStrEq{#1}{\empty}{}{|#1}}}

\newrobustcmd{\muNh}[3][]{
  \MMAS{\mu}{}{#3}{}{}{}{}{}^{{}_{#2\IfStrEq{#1}{\empty}{}{|#1}}} }
\newrobustcmd{\muNM}[3][]{
  \MMAS{\mu}{}{#3}{-}{}{}{}{}^{{}_{#2\IfStrEq{#1}{\empty}{}{|#1}}} }

\newrobustcmd{\zetaNj}[3][*]{
  \MMAS{\zeta}{}{}{}{}{#3}{}{}^{{}_{#2|#1}} }

\newrobustcmd{\pnt}[2]{
  \ensuremath{p\!\left(
      \IfStrEq{#1}{#2}{#2}{#2|#1}\right)}}
\newrobustcmd{\qnt}[2]{
  \ensuremath{q\!\left(
      \IfStrEq{#1}{#2}{#2}{#2|#1}\right)}}

\newrobustcmd{\tildeqnt}[2]{
  \ensuremath{\tilde{q}\!\left(
      \IfStrEq{#1}{#2}{#2}{#2|#1}\right)}}

\newrobustcmd{\innt}[2]{
  \ensuremath{i\!\left(
      \IfStrEq{#1}{#2}{#2}{#2|#1}\right)}}
  
\newrobustcmd{\tildeinnt}[2]{
  \ensuremath{\tilde{\imath}\!\left(
      \IfStrEq{#1}{#2}{#2}{#2|#1}\right)}}

\newrobustcmd{\anp}[3][]{
  \ensuremath{\subp{a}{
      \IfStrEq{#2}{}{}{#2:}
      }{#3}{}{#1}}}

\newrobustcmd{\vnp}[3][]{
  \ensuremath{\subp{v}{#2:}{#3}{}{#1}}}

\newrobustcmd{\tp}[1]{
  \ensuremath{\subp{t}{}{}{}{#1}}}

\newrobustcmd{\lambdaiM}[2]{
  \ensuremath{\subp{\lambda}{}{#1}{}{}\!\left(#2 \right)}}

\newrobustcmd{\lambdazM}[3][]{
  \lambdaz[#1]{#3\!}(#2)}

\newrobustcmd{\bmbzh}[1]{\bmbz{#1}}

\newrobustcmd{\rni}[3][]{
  \ensuremath{\subp{r}{
      \IfStrEq{#2}{}{}{#2:}}{#3}{}{}        
    \IfStrEq{#1}{}
    {}{\!\left(#1\right) }}}

\newrobustcmd{\unit}[2][]{
  \ensuremath{\subp{\bm{e}}{}{#2}{}{#1}}}

\newrobustcmd{\baM}[2][]{
  \ensuremath{\subp{\bm{a}}{}{\overbar{#2}}{}{#1}}}

\newrobustcmd{\mathfrakUhb}[3][]{
  \ensuremath{\subp{\mathfrak{U}}{#2}{
      \IfStrEq{#3}{}{}{:\bm{#3}}
    }{}{#1}}}

\newrobustcmd{\mathcalCnki}[3]{
  \ensuremath{\subp{\mathcal{C}}{
      \IfStrEq{#1}{}{}{#1:}
    }{
      \IfStrEq{#2}{}{}{#2|}
      #3}{}{}}}

\newrobustcmd{\mathcalVnki}[3]{
  \ensuremath{\subp{\mathcal{V}}{#1:}{#2|#3}{}{}}}

\newrobustcmd{\mathcalWnki}[3]{
  \ensuremath{\subp{\mathcal{W}}{#1:}{#2|#3}{}{}}}

\newrobustcmd{\QMx}[2]{
  \ensuremath{\subp{Q}{}{#1}{}{}\!\!\left(#2\right)}}

\newrobustcmd{\sigmah}[2]{
  \MMAS{\sigma}{}{#1}{}{|}{#2}{}{}^{\,{}_{2}} }

\newrobustcmd{\sigmaM}[3][]{
  \MMAS{\sigma}{}{#2}{-}{|}{#3}{}{}^{\,{}_{2}}\!\!\left(#1\right) }

\newrobustcmd{\sigmaMlimit}[2][]{
  \MMAS{\sigma}{}{\bullet}{}{|}{#2}{}{}^{\,{}_{2}}\!\!\left(#1\right) }

\newrobustcmd{\SigmaM}[2]{
  \MMAS{\bm{\Sigma}}{}{#1}{-}{|}{#2}{}{} }

\newrobustcmd{\INhjl}[5][]{\Iz[#2\IfStrEq{#1}{\empty}{}{|#1}]{#3#4:#5}}
\newrobustcmd{\INMl}[4][]{\Iz[#2\IfStrEq{#1}{\empty}{}{|#1}]{\overbar{#3}:#4}}

\newcommand{\oh}[2][]{   \ensuremath{\subp{o}{}{#1}{}{}\!\left(#2\right)} }

\newcommand{\blimit}{
  \ensuremath{\bm{b}\rightarrow\subp{\bm{0}}{}{}{}{+}}
}

\newcommand{\hlimit}{
  \ensuremath{h\rightarrow\infty}
}

\newcommand{\mlimit}{
  \ensuremath{m\rightarrow\infty}
}

\newcommand{\nlimit}{
  \ensuremath{n\rightarrow\infty}
}

\newcommand{\operatorss}[3]{  \relax   \ifmmode   #1_{#2}^{#3}   \else   \ensuremath{\subp{#1}{}{#2}{}{#3}}   \fi }

\newcommand{\intss}[2]{
  \operatorss{\int}{#1}{#2}
}

\newcommand{\oplusss}[2]{
  \operatorss{\bigoplus}{#1}{#2}
}

\newcommand{\sumss}[2]{
  \operatorss{\sum}{#1}{#2}
}

\expandarg

\newrobustcmd{\parenType}[2][L]{
    \StrLeft{#2}{1}[\parSelect]    \IfStrEq{#1}{L}{
    \def\parDirection{\left}
         \IfStrEq{\parSelect}{.}{\def\parType{.}}{}
    \IfStrEq{\parSelect}{|}{\def\parType{|}}{}
    \IfStrEq{\parSelect}{1}{\def\parType{\|}}{}
    \IfStrEq{\parSelect}{u}{\def\parType{\uparrow}}{}
    \IfStrEq{\parSelect}{d}{\def\parType{\downarrow}}{}
    \IfStrEq{\parSelect}{U}{\def\parType{\Uparrow}}{}
    \IfStrEq{\parSelect}{D}{\def\parType{\Downarrow}}{}
        \IfStrEq{\parSelect}{2}{\def\parType{(}}{}
    \IfStrEq{\parSelect}{3}{\def\parType{)}}{}
        \IfStrEq{\parSelect}{4}{\def\parType{[}}{}
    \IfStrEq{\parSelect}{5}{\def\parType{]}}{}
        \IfStrEq{\parSelect}{6}{\def\parType{\{}}{}
    \IfStrEq{\parSelect}{7}{\def\parType{\}}}{}
        \IfStrEq{\parSelect}{8}{\def\parType{<}}{}
    \IfStrEq{\parSelect}{9}{\def\parType{>}}{}
                    \IfStrEq{\parSelect}{b}{\def\parType{\lfloor}}{}
    \IfStrEq{\parSelect}{B}{\def\parType{\rfloor}}{}
        \IfStrEq{\parSelect}{b}{\def\parType{\lfloor}}{}
    \IfStrEq{\parSelect}{B}{\def\parType{\rfloor}}{}
        \IfStrEq{\parSelect}{t}{\def\parType{\lceil}}{}
    \IfStrEq{\parSelect}{T}{\def\parType{\rceil}}{}
  }{
    \def\parDirection{\right}
         \IfStrEq{\parSelect}{.}{\def\parType{.}}{}
    \IfStrEq{\parSelect}{|}{\def\parType{|}}{}
    \IfStrEq{\parSelect}{1}{\def\parType{\|}}{}
    \IfStrEq{\parSelect}{u}{\def\parType{\uparrow}}{}
    \IfStrEq{\parSelect}{d}{\def\parType{\downarrow}}{}
    \IfStrEq{\parSelect}{U}{\def\parType{\Uparrow}}{}
    \IfStrEq{\parSelect}{D}{\def\parType{\Downarrow}}{}
        \IfStrEq{\parSelect}{2}{\def\parType{)}}{}
    \IfStrEq{\parSelect}{3}{\def\parType{(}}{}
        \IfStrEq{\parSelect}{4}{\def\parType{]}}{}
    \IfStrEq{\parSelect}{5}{\def\parType{[}}{}
        \IfStrEq{\parSelect}{6}{\def\parType{\}}}{}
    \IfStrEq{\parSelect}{7}{\def\parType{\{}}{}
        \IfStrEq{\parSelect}{8}{\def\parType{>}}{}
    \IfStrEq{\parSelect}{9}{\def\parType{<}}{}
                    \IfStrEq{\parSelect}{b}{\def\parType{\rfloor}}{}
    \IfStrEq{\parSelect}{B}{\def\parType{\lfloor}}{}
        \IfStrEq{\parSelect}{T}{\def\parType{\rceil}}{}
    \IfStrEq{\parSelect}{t}{\Def\parType{\lceil}}{}
  }
    \IfStrEq{#2}{}{}{\parDirection\parType}
}

\newrobustcmd{\parenthesisGenerator}[5][]{
  \begingroup
        \IfStrEq{#1}{}{\def\firstArg{#2}}{\def\firstArg{#1}}
    \IfStrEq{#2}{}{\def\firstArg{#2}}{}
    \subp{
    \IfStrEq{#2}{}{}{\parenType[L]{#2}}
    \IfStrEq{#2}{}{
      \mbox{\ensuremath{#3}}}{#3}
    \IfStrEq{\firstArg}{}{}{\parenType[R]{\firstArg}}}{}{#4}{}{#5}
  \endgroup
}

\newrobustcmd{\parenR}[1]{
  \parenthesisGenerator{2}{#1}{}{}
}
\newrobustcmd{\parenRz}[3][]{
  \parenthesisGenerator{2}{#3}{\!#2}{\!#1}
}
\newrobustcmd{\parenA}[1]{
  \parenthesisGenerator{8}{#1}{}{}
}
\newrobustcmd{\parenAz}[3][]{
  \parenthesisGenerator{8}{#3}{\!#2}{\!#1}
}
\newrobustcmd{\parenC}[1]{
  \parenthesisGenerator{6}{#1}{}{}
} 
\newrobustcmd{\parenCz}[3][]{
  \parenthesisGenerator{6}{#3}{#2}{#1}
}
\newrobustcmd{\parenS}[1]{
  \parenthesisGenerator{4}{#1}{}{}
} 
\newrobustcmd{\parenSz}[3][]{
  \parenthesisGenerator{4}{#3}{\!#2}{\!#1}
}
\newrobustcmd{\parenAbs}[1]{
  \parenthesisGenerator{|}{#1}{}{}
} 
\newrobustcmd{\parenAbsz}[3][]{
  \parenthesisGenerator{|}{#3}{#2}{#1}
}
\newrobustcmd{\parenABS}[1]{
  \parenthesisGenerator{1}{#1}{}{}
} 
\newrobustcmd{\parenABSz}[3][]{
  \parenthesisGenerator{1}{#3}{#2}{#1}
}

\newrobustcmd{\power}[3][]{
  \parenthesisGenerator{}{#3}{#2}{#1}
}

\newrobustcmd{\Vector}[2][]{
  \parenSz[\IfStrEq{#1}{}{}{\,#1}]{}{#2}
}

\newrobustcmd{\Co}{\texttt{Co}\xspace}
\newrobustcmd{\Quad}{\texttt{Quad}\xspace}
\newrobustcmd{\Phase}{\texttt{Phase}\xspace}
\newrobustcmd{\Amplitude}{\texttt{Amplitude}\xspace}

\def\EuStockMarkets{\texttt{EuStockMarkets}\xspace}

\newrobustcmd{\eval}[3][]{
  \subp{\left.#2\right|}{}{#3}{}{}
}

\makeatletter{}

\renewcommand{\lgcor}[2][]{
  \rhoz{#2
    \IfStrEq{#1}{\empty}{\!}{}}}

\renewcommand{\hatlgcor}[2][]{
  \hatrhoz{#2
    \IfStrEq{#1}{\empty}{\!}{}}}

\renewcommand{\lgacr}[3][]{
  \rhoz{#2
    \IfStrEq{#1}{\empty}{\!}{}}(#3)}

\renewcommand{\lgccr}[4][]{
  \rhoz{#2:#3
    \IfStrEq{#1}{\empty}{\!}{}}(#4)}

\renewcommand{\hatlgacr}[3][]{
  \hatrhoz{#2
    \IfStrEq{#1}{\empty}{\!}{}}(#3)}
\renewcommand{\hatlgccr}[4][]{
  \hatrhoz{#2:#3
    \IfStrEq{#1}{\empty}{\!}{}}(#4)}

\renewcommand{\hatlgacrb}[4][]{
  \hatrhoz{\!#2
    \IfStrEq{#1}{\empty}{\!}{}
  }(#3|\scalebox{.7}{$#4$})}

\renewcommand{\hatlgccrb}[5][]{
  \hatrhoz{#2:#3
    \IfStrEq{#1}{\empty}{\!}{}
  }(#4|\scalebox{.7}{$#5$})}

\renewcommand{\hatlgthetab}[4][]{
  \hatbmthetaz{\!#2
    \IfStrEq{#1}{\empty}{\!}{}
  }(#3|\scalebox{.7}{$#4$})}

\renewcommand{\lgsd}[3][]{
  \fz{\!#2
    \IfStrEq{#1}{\empty}{\!}{}}(#3)}

\renewcommand{\lgcsd}[4][]{
  \fz{#2:#3
    \IfStrEq{#1}{\empty}{\!}{}}(#4)}

\renewcommand{\lgcsdCo}[4][]{
  \cz{#2:#3
    \IfStrEq{#1}{\empty}{\!}{}}(#4)}
\renewcommand{\lgcsdQuad}[4][]{
  \qz{#2:#3
    \IfStrEq{#1}{\empty}{\!}{}}(#4)}
\renewcommand{\lgcsdAmplitude}[4][]{
  \alphaz{#2:#3
    \IfStrEq{#1}{\empty}{\!}{}}(#4)}
\renewcommand{\lgcsdPhase}[4][]{
  \phiz{#2:#3
    \IfStrEq{#1}{\empty}{\!}{}}(#4)}

\renewcommand{\lgcsdCoM}[5][]{
  \cz[#5]{#2:#3
    \IfStrEq{#1}{\empty}{\!}{}}(#4)}
\renewcommand{\lgcsdQuadM}[5][]{
  \qz[#5]{#2:#3
    \IfStrEq{#1}{\empty}{\!}{}}(#4)}
\renewcommand{\lgcsdAmplitudeM}[5][]{
  \alphaz[#5]{#2:#3
    \IfStrEq{#1}{\empty}{\!}{}}(#4)}
\renewcommand{\lgcsdPhaseM}[5][]{
  \phiz[#5]{#2:#3
    \IfStrEq{#1}{\empty}{\!}{}}(#4)}

\renewcommand{\hatlgcsdCoM}[5][]{
  \hatcz[\ #5]{#2:#3
    \IfStrEq{#1}{\empty}{\!}{}}(#4)}
\renewcommand{\hatlgcsdQuadM}[5][]{
  \hatqz[\ #5]{#2:#3
    \IfStrEq{#1}{\empty}{\!}{}}(#4)}
\renewcommand{\hatlgcsdAmplitudeM}[5][]{
  \hatalphaz[\ #5]{#2:#3
    \IfStrEq{#1}{\empty}{\!}{}}(#4)}
\renewcommand{\hatlgcsdPhaseM}[5][]{
  \hatphiz[\ #5]{#2:#3
    \IfStrEq{#1}{\empty}{\!}{}}(#4)}

\renewcommand{\lgcsdSQ}[4][]{
  \mathcalKz[asc]{#2:#3
    \IfStrEq{#1}{\empty}{\!}{}}(#4)}

\renewcommand{\lgsdM}[4][]{
  \fz[#4]{\!#2
    \IfStrEq{#1}{\empty}{\!}{}}(#3)}
\renewcommand{\lgcsdM}[5][]{
  \fz[#5]{\!#2:#3
    \IfStrEq{#1}{\empty}{\!}{}}(#4)}

\renewcommand{\hatlgsd}[3][]{
  \hatfz{\!#2
    \IfStrEq{#1}{\empty}{\!}{}}(#3)}

\renewcommand{\hatlgcsd}[4][]{
  \hatfz{\!#2:#3
    \IfStrEq{#1}{\empty}{\!}{}}(#4)}

\renewcommand{\hatlgsdM}[4][]{
  \hatfz[#4]{\!#2
    \IfStrEq{#1}{\empty}{\!}{}}(#3)}
\renewcommand{\hatlgcsdM}[5][]{
  \hatfz[#5]{\!#2:#3
    \IfStrEq{#1}{\empty}{\!}{}}(#4)}

\renewcommand{\lgsdRE}[3][]{
  \cz{#2
    \IfStrEq{#1}{\empty}{\!}{}}(#3)}
\renewcommand{\hatlgsdREM}[4][]{
  \hatcz[\ #4]{#2
    \IfStrEq{#1}{\empty}{\!}{}}(#3)}

\renewcommand{\lgsdIM}[3][]{
  \qz{#2
    \IfStrEq{#1}{\empty}{\!}{}}(#3)}
\renewcommand{\hatlgsdIMM}[4][]{
  \hatqz[\ #4]{#2
    \IfStrEq{#1}{\empty}{\!}{}}(#3)}

\usepackage{marginnote}

\usepackage[
linewidth=1pt,
middlelinecolor= black,
middlelinewidth=0.4pt,
roundcorner=1pt,
topline = false,
rightline = false,
bottomline = false,
rightmargin=0pt,
skipabove=0pt,
skipbelow=0pt,
leftmargin=-.3cm,
innerleftmargin=.3cm,
innerrightmargin=0pt,
innertopmargin=0pt,
innerbottommargin=0pt,
]{mdframed}

\usepackage{bibunits}
\defaultbibliography{../Bibliography}
\defaultbibliographystyle{elsarticle-harv}

\usepackage{xr}

\usepackage{placeins}

\newcommand{\JT}[1][]{%
  JT22\IfStrEq{#1}{\empty}{}{~[#1]}\xspace}

\newrobustcmd{\lgsdRef}[2][]{\JT[\IfStrEq{#1}{\empty}{\cref{#2}}{\cref{#2}, p.~\pageref{#2}}]}

\begin{document}

\def\spacingset#1{\renewcommand{\baselinestretch}
  {#1}\small\normalsize} \spacingset{1}

\begin{bibunit}

  \if1\blind
  {
  \title{\bf   {Local Gaussian cross-spectrum analysis}}
  
  \author{      Lars Arne Jordanger\thanks{     Western Norway University of Applied Sciences, Faculty of
      Engineering and Science, P.B 7030, 5020 Bergen, Norway
      E-mail: \textrm{Lars.Arne.Jordanger@hvl.no}}  \and       Dag Tj{\o}stheim\thanks{     University of Bergen, Department of Mathematics, P.B. 7803, 5020
      Bergen, Norway} }
  \date{\vspace{-5ex}}
  
  \maketitle
  } \fi

  \if0\blind
  {
    \bigskip
    \bigskip
    \bigskip
    \begin{center}
      {\LARGE\bf {Local Gaussian cross-spectrum analysis}}
    \end{center}
    \medskip
  } \fi

  \bigskip
  \begin{abstract}

  The ordinary spectrum is restricted in its applications, since it is
  based on the second order moments (auto and cross-covariances).
  Alternative approaches to spectrum analysis have been investigated
  based on other measures of dependence.  One such approach was
  developed for univariate time series by the authors of this paper
  using the \textit{local Gaussian auto-spectrum} based on the
  \textit{local Gaussian auto-correlations}.  This makes it possible
  to detect local structures in univariate time series that looks like
  white noise when investigated by the ordinary auto-spectrum.  In
  this paper the \textit{local Gaussian approach} is extended to a
  \textit{local Gaussian cross-spectrum} for multivariate time series.
  The local Gaussian cross-spectrum has the desirable property that it
  coincides with the ordinary cross-spectrum for Gaussian time series,
  which implies that it can be used to detect non-Gaussian traits in
  the time series under investigation.  In particular: If the ordinary
  spectrum is flat, then peaks and troughs of \textit{the local
    Gaussian spectrum} can indicate nonlinear traits, which
  potentially might reveal \textit{local periodic phenomena} that goes
  undetected in an ordinary spectral analysis.

  \end{abstract}


\section{Introduction}
\label{sec:lgch_Introduction}

The auto- and cross-spectra are tools that can extract temporal
information from a time series, and they can be used to detect
interesting periodic phenomena from the data under investigation.  A
disadvantage of these tools is that they are based on the auto- and
cross-correlations.  This implies that non-linear temporal
dependencies could be invisible from the point of view of the
classical spectral theory, and this has, as discussed below, motivated
quite a few modifications of the classical approach.

One such modification, the local Gaussian spectral approach, was
introduced for the univariate case in \citet{jordanger17:_lgsd}
(hereafter referred to as \JT).  A key feature of the local Gaussian
approach is that the resulting local Gaussian spectra coincide
completely with the classical spectra for Gaussian time series.  The
local Gaussian auto-spectra can be used to investigate how the
temporal dependency structure for a given univariate time series
deviates from the linear Gaussian dependency structure, and a local
Gaussian analysis can thus detect nonlinear temporal dependencies,
including local periodic phenomena, that are hidden from the classical
approach.
 
This paper presents the extension of the local Gaussian approach to
the multivariate case, introducing local Gaussian variants of the
classical (global) co-spectrum, the quadrature spectrum, the amplitude
spectrum and the phase spectrum.  Using these new local Gaussian
tools, it is possible to investigate on a deeper than linear
dependency level how different time series are interconnected.  The
local Gaussian approach can detect cross-temporal interdependency
relationships beyond what the classical (global) cross-spectrum can
do, and this can, e.g., be used to examine whether one series is
leading-lagging another one at extreme amplitudes, but not at moderate
amplitudes, say.

The local Gaussian spectral-estimation algorithm has in this paper
been sanity-tested on simulated examples (a nonlinearly connected
sinusoidal example, as explained later on), and it has also been used
on the French CAC index versus the German DAX index in order to show
how it can be applied to investigate nonlinear cross-temporal
interdependency structures from the \EuStockMarkets-data.

In addition to this, a GARCH-type model has been fitted to the
\EuStockMarkets-data, and samples from the fitted model have been used
to show how the local Gaussian sanity testing of parametric models
from \JT can be extended to the multivariate case.  This local
Gaussian sanity testing of the fitted models could supplement other
model selection techniques, and it could help detect models whose
local Gaussian properties clearly deviates from the properties seen in
the sample.  Note that the focus for this particular part, as in \JT,
is on the concepts themselves. A precise analysis of a test of fit is
left to future work.

\textbf{The classical (global) approach:}
The auto and cross-covariances
$\TSR{\TSR{\gammaz{k\ell}(h)}{h\in\ZZ}{}}{k,\ell=1}{d}$ from a time
series $\TSR{\bmYz{t}=\parenR{\Yz{1,t},\dotsc,\Yz{d,t}}}{t\in\ZZ}{}$,
can range from determining it completely (Gaussian time series) to
containing no information at all (GARCH-type models).  The auto- and
cross-spectral densities $\TSR{\fz{\!k\ell}(\omega)}{k,\ell=1}{d}$ are
(assuming they exist) the Fourier transformations of the auto and
cross-covariances, and these tools thus share the same limitations.
This implies that the auto- and cross-spectra might be inadequate
tools when the task of interest is to investigate non-Gaussian time
series containing asymmetries or other nonlinear structures --- like
those observed in stock returns, cf.\ e.g.\
\citet{hong07:_asymm_stock_retur}.

In ordinary spectral analysis, if $\TSR{\Yz{k,t}}{t\in\ZZ}{}$ and
$\TSR{\Yz{\ell,t}}{t\in\ZZ}{}$ are jointly weakly stationary, and if
the cross-covariances
$\gammaz{k\ell}(h)\defeq \Cov{\Yz{k,t+h}}{\Yz{\ell,t}}$ are absolutely
summable, then the cross-spectrum $\fz{\!k\ell}(\omega)$ is defined as
the Fourier transform of the cross-covariances, i.e.,
\begin{align}
  \label{eq:global_cross_spectrum}
  \fz{\!k\ell}(\omega)
  &= \sumss{h\in\ZZ}{} \gammaz{k\ell}(h)\cdot \ez[-2\pi i \omega h]{}.
\end{align}
A simple scaling connects the auto- and cross-covariances $\gammaz{k\ell}(h)$
and the corresponding auto- and cross-correlations $\rhoz{k\ell}(h)$,
in particular
$\gammaz{k\ell}(h)=\sqrt{\gammaz{kk}(0)\cdot\gammaz{\ell\ell}(0)}\cdot\rhoz{k\ell}(h)$,
and this implies that the cross-spectrum $\fz{k\ell}(\omega)$ given in
\cref{eq:global_cross_spectrum} also can be written as
\begin{align}
  \label{eq:global_cross_spectrum_II}
  \fz{\!k\ell}(\omega)
  &= \sqrt{\gammaz{kk}(0)\cdot\gammaz{\ell\ell}(0)}\cdot \sumss{h\in\ZZ}{} \rhoz{k\ell}(h)\cdot \ez[-2\pi i \omega h]{}.
\end{align}
It is the sum of the cross-correlations that are generalised in the
local Gaussian approach.

The expression for the inverse Fourier transform reveals, when $h=0$,
that the covariance $\Cov{\Yz{k,t}}{\Yz{\ell,t}}=\gammaz{k\ell}(0)$
can be expressed as the integral
$\intss{-1/2}{1/2} \fz{\!k\ell}(\omega) \d{\omega}$.  This makes it
possible to inspect how the (linear) interaction between the marginal
time series varies with the frequency $\omega$.  An inspection of the
cross-spectrum $\fz{\!k\ell}(\omega)$ is a bit more complicated than
that of the auto-spectrum, since $\fz{\!k\ell}(\omega)$ in general
will be a complex-valued function.  It is thus usually the following
real valued functions that are investigated,
\begin{subequations}
  \label{eq:global_co_quadrature_amplitude_phase}
  \begin{align}
    \label{eq:global_co_quadrature}
    \cz{k\ell}(\omega) 
    &= \operatorname{Re}\parenR{\fz{\!k\ell}(\omega)},
    &\qz{k\ell}(\omega) 
    &= -
      \operatorname{Im}\parenR{\fz{\!k\ell}(\omega)}, \\
    \label{eq:global_amplitude_phase}
    \alphaz{k\ell}(\omega) 
    &= \operatorname{Mod}\parenR{\fz{\!k\ell}(\omega)},
    &\phiz{k\ell}(\omega) 
     &= \phantom{-} \operatorname{Arg}\parenR{\fz{\!k\ell}(\omega)},
  \end{align}
\end{subequations}
where $\cz{k\ell}(\omega)$, $\qz{k\ell}(\omega)$,
$\alphaz{k\ell}(\omega)$ and $\phiz{k\ell}(\omega)$, respectively, are
referred to as the cospectrum, quadrature spectrum, amplitude spectrum
and phase spectrum.  Note that $\cz{k\ell}(\omega)$ always integrates
to $\gammaz{k\ell}(0)$ over one period, whereas $\qz{k\ell}(\omega)$
always integrates to zero.

The \textit{coherence}
$\mathcalKz{k\ell}(\omega) \defeq {\fz{\!k\ell}(\omega)}/{
  \sqrt{\fz{kk}(\omega)\fz{\ell\ell}(\omega)}}$ is an important tool
when a spectral analysis is performed on a multivariate time series,
in particular since $\mathcalKz{k\ell}(\omega)$ can be realised as the
correlation of ${\d{\Zz{k}(\omega)}}$ and ${\d{\Zz{\ell}(\omega)}}$,
where ${{\Zz{k}(\omega)}}$ and ${{\Zz{\ell}(\omega)}}$ are the right
continuous orthogonal-increment processes that by the Spectral
Representation Theorem correspond to the weakly stationary time series
$\TSR{\Yz{k,t}}{t\in\ZZ}{}$ and $\TSR{\Yz{\ell,t}}{t\in\ZZ}{}$, see
e.g.\ \citet[p.~436]{Brockwell:1986:TST:17326} for details.  The
\textit{squared coherence}
$\parenAbsz[2]{}{\mathcalKz{k\ell}(\omega)}$ is of interest since its
value (in the interval $[0,1]$) reveals to what extent the two time
series $\TSR{\Yz{k,t}}{t\in\ZZ}{}$ and $\TSR{\Yz{\ell,t}}{t\in\ZZ}{}$
can be related by a linear~filter.

\textbf{Modifications of the classical approach:} Other spectral
approaches, involving different generalisations of the auto-spectrum
$f(\omega)$ were discussed in
\JT[section~1], and several of the approaches
were based on the following idea: The second order moments captured by
the autocovariances $\TSR{\gamma(h)}{h\in\ZZ}{}$ can be replaced by
alternative dependence measures $\xiz{h}$ computed from the bivariate
random variables $\parenR{\Yz{t+h},\Yz{t}}$, and a spectral density
approach can then (under suitable regularity conditions) be defined as
the Fourier transform of $\TSR{\xiz{h}}{h\in\ZZ}{}$.  For multivariate
time series, the natural extension is then to define similar measures
$\xiz{k\ell:h}$ for the bivariate random variables
$\parenR{\Yz{k,t+h},\Yz{\ell,t}}$, and then use the corresponding
Fourier-transform as an alternative to the cross-spectrum
$\fz{\!k\ell}(\omega)$.

It does not seem to be the case (yet) that multivariate versions have
been investigated for all of the possible generalisations of the
auto-spectrum $f(\omega)$, but some generalisations do exist.  The
first extension of the cross-spectrum $\fz{\!k\ell}(\omega)$ along
these lines is the \textit{polyspectra} introduced in
\citet{brillinger1965introduction}, which is the multivariate version
of the higher order moments/cumulants approach to spectral analysis,
see \citet{Tukey:1959:IMS,Brillinger:1984:CWJa,brillinger1991some}.
Another generalisation of $\fz{\!k\ell}(\omega)$ is given in
\citet{chung07:_model}, where the generalised function approach
introduced in \citet{hong1999hypothesis} is used to set up a
cross-spectrum that can be used for the testing of directional
predictability in foreign exchange markets.  Some recent work applies
copula-techniques for the investigation of multivariate time series,
like \citet{li2023robust} using a \textit{copula spectra density
  kernel} and \citet{zhao22:_model_multiv_time_series_with} using
\textit{Copula-linked Univariate D-vines}.  A quantile based approach
can be found in \citet{barunik19:_quant} where the
\textit{quantile-coherency} concept is introduced, and quantile
approaches based on the \textit{quantile periodogram} can be seen in 
\citet{li2021quantile,li22:_quant_fourier_trans_quant_series}.

In many cases it can be difficult to interpret the various
spectra/visualisations which results from these modifications of the
classical approach, but machine learning methods can be combined with
these techniques and help reduce/remove this problem --- see, e.g.,
\citet{li22:_quant_frequen_analy_deep_learn_signal_class,li/zero_crossing_to_QFA,chen2019semi}
for some examples of this approach.  The review-paper
\citet{ciaburro21:_machin_learn_based_algor_knowl} contains more
information about machine-learning based methods applied to time
series data.

\textbf{The local Gaussian approach:} The local Gaussian spectral
density $\lgsd{\LGp}{\omega}$ for univariate strictly stationary time
series that was defined in \JT, is based on the
local Gaussian auto-correlations $\lgacr{\LGp}{h}$ from
\citet{Tjostheim201333}.  A simple adjustment gives the local Gaussian
cross-correlations $\lgccr{k\ell}{\LGp}{h}$ for multivariate strictly
stationary time series, from which a local Gaussian analogue
$\lgcsd{\!k\ell}{\LGp}{\omega}$ of the cross-spectrum
$\fz{\!k\ell}(\omega)$ can be constructed using the Fourier transform.
The local Gaussian version of the cross-spectrum enables local
Gaussian alternatives to be defined of the cospectrum, quadrature
spectrum, amplitude spectrum, and phase spectrum, by simply copying
the setup used in the ordinary (global) case.  Local Gaussian
analogues of the coherence and squared coherence were investigated in
the preparation for this paper, but then discarded, see the discussion
at the end of \cref{sec:related_LG_entities}
for further details.

It should be mentioned that the local Gaussian approach to statistical
modelling has been applied recently to a number of nonlinear
statistical problems.  See, e.g., \citet{tjostheim2021statistical}.

\textbf{An overview of the paper:} \Cref{seq:lgch_definitions} defines
the \textit{local Gaussian cross-spectrum}
$\lgcsd{\!k\ell}{\LGp}{\omega}$, which immediately gives the related
local Gaussian variants of the cospectrum, quadrature spectrum,
amplitude spectrum and phase spectrum from
\cref{eq:global_co_quadrature_amplitude_phase}.  The asymptotic theory
for the estimators are then presented (some technical details and
proofs are relegated to the appendices).  The real and simulated
examples in \cref{sec:lgch_examples} shows that estimates of
$\lgcsd{\!k\ell}{\LGp}{\omega}$ can be used to detect and investigate
nonlinear structures in non-Gaussian white noise, and in particular
that $\lgcsd{\!k\ell}{\LGp}{\omega}$ can detect local periodic
phenomena that goes undetected in an ordinary spectral
analysis.
Note that the scripts needed for the reproduction of these
examples are contained in the \Rpackage
\lgsdRpackage,\footnote{ \label{lgch_footnote:install_package} Use
  \Rref{remotes:\!:install\_github("LAJordanger/localgaussSpec")} to
  install the package.
  See \cref{P2.app:The.scripts.in.lgsdRpackage} in the online
  Supplementary Material for further details.}  where it in addition
is possible to use an interactive solution to see how adjustments of
the tuning parameters (used in the estimation algorithms) influence
the estimates of~$\lgcsd{\!k\ell}{\LGp}{\omega}$.  %
Some concluding remarks are then gathered in
\cref{sec:lgch_conclusion}.

\textbf{Supplementary Material:} The appendices are available in the
online Supplementary Material, and these contain the proofs of the
theoretical results, discussions related to the sensitivity of the
different input parameters in the estimation algorithm, and also a
discussion of the sensitivity of the block length in the bootstrapped
based resampling algorithm that was developed in \JT for the local
Gaussian spectral investigation.  The Supplementary Material also
contains a discussion related to how the local Gaussian spectra (based
on samples) of a parametric model can be compared to the local
Gaussian spectra from the sample the model was fitted to.  The
reproduction of all the results in this paper can be done by the help
of the scripts contained in the \Rpackage \lgsdRpackage, and details
about this is also available in the appendices.  The Supplementary
Material ends with a discussion related to the construction of the
test-data used for the sanity testing of the estimation algorithm, and
this part also highlights some limitations that a practitioner should
keep in mind when the local Gaussian approach is used in the extreme
tails of the sample at hand.

\section{Definitions}
\label{seq:lgch_definitions}
This section will present the formal definitions of the local Gaussian
versions of the cross-correlation $\fz{k\ell}(\omega)$ and its derived
entities.  The details are almost identical to those encountered when
the local Gaussian spectral density was introdued in
\JT, and the present discussion will thus only
give short summaries of descriptions and arguments already undertaken
in~\JT.

\subsection{The local Gaussian correlations}
\label{sec:LGC_lgch}

At the core of the generalisation of \cref{eq:global_cross_spectrum}
lies the local Gaussian correlation $\rhoz{\LGp}$ from
\citet{Tjostheim201333}.  A construction that originally was used for
density estimation, see \citet{hjort96:_local}, was in
\citet{Tjostheim201333} adjusted a bit in order to define
$\rhoz{\LGp}$ as a new local measure of dependence for a bivariate
random variable $\bm{W}$.  With $g(\bm{w})$ being the density function
of $\bm{W}$, the idea is to find a local Gaussian approximation
$\psi(\bm{w};\bmthetaz{})$ in the neighbourhood of $\LGp$~--- i.e.,
start with
\begin{align}
  \label{eq:psi_density_five_or_one}
  \psiz{}(\bm{w}; \bmthetaz{}) 
  &\defeq \frac{1}{2\pi \cdot     \sigmaz{1} \sigmaz{2} \sqrt{1 -
    \rhoz[2]{}}}   \exp \parenC{- \tfrac{      \sigmaz[2]{2} \parenRz[2]{}{\wz{1} - \muz{1}}      - 2\sigmaz{1}\sigmaz{2}\rho       \left(\wz{1} - \muz{1} \right)      \left(\wz{2} - \muz{2} \right)      + \sigmaz[2]{1} \parenRz[2]{}{\wz{2} - \muz{2}}    }{      \sigmaz[2]{1} \sigmaz[2]{2} \left( 1 -
    \rhoz[2]{} \right) } },
\end{align}

then find a parameter-vector
$\bmthetaz{\LGp}=\parenSz[\,\prime]{}{\muz{1\LGp},\muz{2\LGp},\sigmaz{1\LGp},\sigmaz{2\LGp},\rhoz{\LGp}}$
that gives the best match to $g(\bm{w})$ in a neighbourhood of $\LGp$,
and then extract the correlation component $\rhoz{\LGp}$ from the
parameter-vector $\bmthetaz{\LGp}$.

As discussed in \lgsdRef{sec:LGC_revised}, the theoretical treatment
can be based directly on the construction from
\citet{Tjostheim201333}, but the numerical convergence of the
estimation-algorithm might then sometimes fail.  This
estimation-problem can be countered if a \textit{normalisation of the
  marginals} (see \cref{def:LG_cross-and-auto-spectra_def_main_part}
below) are performed before the estimation-algorithm is used.

\subsection{The local Gaussian cross-spectrum}
\label{sec:lgch_and_friends_definition}
The definition of the local Gaussian cross-spectrum density is almost
identical to the definition of the local Gaussian spectral density
from \lgsdRef{sec:LGSD_definition}, which in this paper henceforth
will be referred to as the local Gaussian auto-spectrum.

\begin{definition}
  \label{def:LG_cross-and-auto-spectra_def_main_part}
  For a strictly stationary multivariate time series
  $\TSR{\bmYz{t}}{t\in\ZZ}{}$, where
  
  {$\bmYz{t}=\parenR{\Yz{1,t},\dotsc,\Yz{d,t}}$}, the local
  Gaussian cross-spectrum of the marginal time series
  $\TSR{\Yz{k,t}}{t\in\ZZ}{}$ and $\TSR{\Yz{\ell,t}}{t\in\ZZ}{}$ is
  constructed in the following manner.
  \begin{enumerate}%[label=(\alph*)]
  \item
    With $\Gz{k}$ and $\Gz{\ell}$ the univariate \textit{marginal}
    cumulative distributions of respectively
    $\TSR{\Yz{k,t}}{t\in\ZZ}{}$ and $\TSR{\Yz{\ell,t}}{t\in\ZZ}{}$,
    and $\Phi$ the cumulative distribution of the univariate standard
    normal distribution, define normalised versions
    $\TSR{\Zz{k,t}}{t\in\ZZ}{}$ and $\TSR{\Zz{\ell,t}}{t\in\ZZ}{}$ by
    \begin{align}
      \TSR{\Zz{k,t} \defeq
      \Phiz[-1]{}\!\parenR{\Gz{k}\!\parenR{\Yz{k,t}}}}{t\in\ZZ}{}, \qquad
      \TSR{\Zz{\ell,t} \defeq
      \Phiz[-1]{}\!\parenR{\Gz{\ell}\!\parenR{\Yz{\ell,t}}}}{t\in\ZZ}{}.
    \end{align}
  \item For a given point \mbox{$\LGp=\LGpoint$} and for each
    \textit{bivariate lag~$h$ pair}
    \mbox{$\bmZz{k\ell:h:t}\defeq \parenR{\Zz{k:t+h},\Zz{\ell:t}}$}, a
    \textit{local Gaussian cross-correlation}
    $\lgccr{k\ell}{\LGp}{h}$ can be computed based on a five
    parameter local Gaussian approximation of the bivariate density of
    $\bmZz{k\ell:h:t}$ at~$\LGpoint$.
  \item
    \label{def:lgcs_requirement}
    When \mbox{$\sumss{h\in\ZZ}{} \parenAbs{\lgccr{k\ell}{\LGp}{h}}
      < \infty$}, the \textit{local Gaussian cross-spectrum} at the
    point~$\LGp$ is defined~as
    \begin{align}
      \label{eq:lgcs_definition_main_document}
      \lgcsd{k\ell}{\LGp}{\omega} \defeq \sumss{h=-\infty}{\infty}
      \lgccr{k\ell}{\LGp}{h} \cdot \ez[-2\pi i\omega h]{}.% ,
    \end{align}
  \end{enumerate}
\end{definition}

The definition of the local Gaussian auto-spectrum is in essence the
same as the one given here for the local Gaussian cross-spectrum, with
the minor adjustment that \mbox{$k=\ell$} in the auto-spectrum
case~--- which requires the added convention that
\mbox{$\lgccr{kk}{\LGp}{0}\equiv 1$} for all points~$\LGp$.

The normalisation of the marginals in
\cref{def:LG_cross-and-auto-spectra_def_main_part} reduces the amount
of numerical convergence issues in the estimation algorithm, and the
resulting local Gaussian cross-spectrum $\lgcsd{k\ell}{\LGp}{\omega}$
is thus a feature of the collection of bivariate copula-structures
that are present inside the overall cross-temporal dependency
structure of the strictly stationary time series
$\TSR{\bmYz{t}}{t\in\ZZ}{}$.  This implies that investigations based
on the local Gaussian cross-spectrum $\lgcsd{k\ell}{\LGp}{\omega}$
should be complemented with methods that can extract relevant features
from the probability density functions of the $d$ marginal
distributions.

The basic properties of the local Gaussian cross-spectrum are quite
similar to those encountered for the local Gaussian auto-spectrum in
\lgsdRef{th:lgsd_properties}.
\begin{lemma}
  \label{th:lgcs_properties}
  The following properties holds for $\lgcsd{k\ell}{\LGp}{\omega}$.
  \begin{enumerate}%[label=(\alph*)]
  \item
    \label{th:lgcs_equal_to_osd_when_Gaussian}
        $\lgcsd{k\ell}{\LGp}{\omega}$ coincides with
    $\fz{k\ell}(\omega)$ for all \mbox{$\LGp\in\RRn{2}$} when
    $\TSR{\bmYz{t}}{t\in\ZZ}{}$ is a multivariate Gaussian time
    series.
  \item
    \label{th:lgcs_reflection_property}
        The following holds when \mbox{$\LGpd\defeq\LGpointd$} is the
    diagonal reflection of \mbox{$\LGp=\LGpoint$};
    \begin{subequations}
      \label{eq:th:lgcs_properties}
      \begin{align}
        \label{eq:th:lgcs_conjugate_property}
        \lgcsd{k\ell}{\LGp}{\omega} &=
        \overline{\lgcsd{\ell k}{\LGpd}{\omega}}, \\
        \label{eq:th:lgcs_folding_property}
        \lgcsd{k\ell}{\LGp}{\omega} &= \lgccr{k\ell}{\LGp}{0} +
        \sumss{h=1}{\infty} \lgccr{\ell k}{\LGpd}{h} \cdot \ez[+2\pi
          i\omega h]{} + \sumss{h=1}{\infty} \lgccr{k\ell}{\LGp}{h}
        \cdot \ez[-2\pi i\omega h]{}.
      \end{align}
    \end{subequations}
  \end{enumerate}
\end{lemma}

\begin{proof}
  \Cref{th:lgcs_equal_to_osd_when_Gaussian} follows since the local
  Gaussian cross-correlations $\lgccr{k\ell}{\LGp}{h}$ by
  construction coincides with the ordinary (global) cross-correlations
  $\rho(h)$ in the Gaussian case.  For the proof of
  \cref{th:lgcs_reflection_property}, the key observation is that the
  \textit{diagonal folding property} that was observed for the local
  Gaussian auto-spectrum, see
  \lgsdRef{th:diagonal_folding_property}, extends directly to the
  present case, i.e.,
  \mbox{$\rhoz{k\ell|\LGp}(-h) = \rhoz{\ell k|\LGpd}(h)$}, where
  \mbox{$\LGpd = \parenR{\LGpi{2},\LGpi{1}}$} is the
  \textit{diagonally reflected point corresponding to $\LGp$}.  This
  implies that
  $\lgcsd{k\ell}{\LGp}{\omega} =
  \overline{\lgcsd{k\ell}{\LGp}{-\omega}} =
  \overline{\lgcsd{\ell k}{\LGpd}{\omega}}$, and it also follows
  that \cref{eq:lgcs_definition_main_document} can be reexpressed as
  \cref{eq:th:lgcs_folding_property}.
\end{proof}

\subsection{Related local Gaussian entities}
\label{sec:related_LG_entities}

From the definition of the local Gaussian cross-spectrum, it is
possible to define related spectra in the same manner as those
mentioned for the ordinary spectrum in
\cref{eq:global_co_quadrature_amplitude_phase}.

\begin{definition}
  \label{def:local_co_quad_amplitude_phase}
  The local Gaussian versions of the cospectrum~$\cz{k\ell}(\omega)$,
  the quadrature spectrum~$\qz{k\ell}(\omega)$, the amplitude
  spectrum~$\alphaz{k\ell}(\omega)$ and the phase
  spectrum~$\phiz{k\ell}(\omega)$, are given by
  \begin{subequations}
    \label{eq:local_co_quad_amplitude_phase}
    \begin{align}
      \label{eq:local_co-spectrum}
      \lgcsdCo{k\ell}{\LGp}{\omega} &\defeq
      \operatorname{Re}\parenR{\lgcsd{k\ell}{\LGp}{\omega}} =
      \lgccr{k\ell}{\LGp}{0} + \sumss{h=1}{\infty}
      \cos\parenR{2\pi\omega h} \parenS{\lgccr{k\ell}{\LGp}{h} +
        \lgccr{k\ell}{\LGpd}{h}}, \\
      \label{eq:local_quad-spectrum} 
      \lgcsdQuad{k\ell}{\LGp}{\omega} &\defeq -
      \operatorname{Im}\parenR{\lgcsd{k\ell}{\LGp}{\omega}} =
      \sumss{h=1}{\infty} \sin\parenR{2\pi\omega h}
      \parenS{\lgccr{k\ell}{\LGp}{h} -
        \lgccr{k\ell}{\LGpd}{h}}, \\
      \label{eq:local_amplitude-spectrum}
      \lgcsdAmplitude{k\ell}{\LGp}{\omega} &\defeq
      \operatorname{Mod}\parenR{\lgcsd{k\ell}{\LGp}{\omega}} =
      \sqrt{\lgcsdCoM{k\ell}{\LGp}{\omega}{\,2} +
        \lgcsdQuadM{k\ell}{\LGp}{\omega}{\,2}}, \\
      \label{eq:local_phase-spectrum}
      \lgcsdPhase{k\ell}{\LGp}{\omega} &\defeq
      \operatorname{Arg}\parenR{\lgcsd{k\ell}{\LGp}{\omega}} \in (-\pi,
      \pi].
    \end{align}
  \end{subequations}
\end{definition}

The sums occurring in
\cref{eq:local_co-spectrum,eq:local_quad-spectrum} follows from
\cref{eq:th:lgcs_folding_property}.
\Cref{eq:th:lgcs_conjugate_property} gives
\mbox{$\lgcsdCo{k\ell}{\LGp}{\omega} = \cz{\ell
    k|\LGpd}(\omega)$},
\mbox{$\lgcsdQuad{k\ell}{\LGp}{\omega} = - \qz{\ell
    k|\LGpd}(\omega)$},
\mbox{$\lgcsdAmplitude{k\ell}{\LGp}{\omega} = \alphaz{\ell
    k|\LGpd}(\omega)$} and
{$\lgcsdPhase{k\ell}{\LGp}{\omega} = - \phiz{\ell
    k|\LGpd}(\omega)$}.

For Gaussian distributions, the local Gaussian correlations will
always be equal to the ordinary (global) correlations,\footnote{
  This is due to the way the local Gaussian correlation is defined,
  see \citet{Tjostheim201333} for details.}  and the local Gaussian
constructions in
\cref{def:LG_cross-and-auto-spectra_def_main_part,def:local_co_quad_amplitude_phase}
will thus coincide with the ordinary (global) versions for
multivariate Gaussian time series.  A comparison of the local and
global estimates in the same plot is thus of interest when a given
sample is considered, since this could detect nonlinear dependencies
of the time series under investigation.

It is possible to define a local Gaussian anaologue of the squared
coherence mentioned in \cref{sec:lgch_Introduction} by replacing the
ordinary cross- and auto-spectra with the corresponding local Gaussian
versions, i.e., the object of interest would be
$\mathcalQz{k\ell:\LGp}(\omega) \defeq \lgcsd{k\ell}{\LGp}{\omega}
\lgcsd{\ell k}{\LGp}{\omega}/
\lgcsd{kk}{\LGp}{\omega}\lgcsd{\ell\ell}{\LGp}{\omega}$.  This
approach was investigated in the preparation of this paper, but it has
not been included here since $\mathcalQz{k\ell:\LGp}(\omega)$ in
general lacked the nice properties known from the ordinary global
case.  In particular, the local Gaussian auto-spectra
$\lgcsd{kk}{\LGp}{\omega}$ and $\lgcsd{\ell\ell}{\LGp}{\omega}$ will
in general be complex valued functions, so an inspection of
$\mathcalQz{k\ell:\LGp}(\omega)$ must thus be based on plots of its
real and imaginary parts (or its amplitude and phase).  Moreover,
these plots did more often than not turn out to be rather hard to
investigate, since the estimates of $\lgcsd{kk}{\LGp}{\omega}$ and
$\lgcsd{\ell\ell}{\LGp}{\omega}$ (for some distributions and some
frequencies~$\omega$) gave values very close to zero in the
denominator.

\makeatletter{}

\subsection{Estimation}
\label{sec2:Estimation}

The estimation of the local Gaussian cross-spectrum
$\lgcsd{k\ell}{\LGp}{\omega}$ from
\cref{sec:lgch_and_friends_definition} follows the same setup that was
used in \lgsdRef{sec:Estimation} for the estimation of the
local Gaussian auto-spectrum, with the obvious difference that some
extra indices are needed in the present case.  The estimation of the
related spectra $\lgcsdCo{k\ell}{\LGp}{\omega}$,
$\lgcsdQuad{k\ell}{\LGp}{\omega}$, $\lgcsdAmplitude{k\ell}{\LGp}{\omega}$
and $\lgcsdPhase{k\ell}{\LGp}{\omega}$ from
\cref{sec:related_LG_entities} is then obtained from the estimate of
$\lgcsd{k\ell}{\LGp}{\omega}$ in an obvious manner.

\begin{algorithm}
  \label{def:lgsd_estimator}
        For a sample
  $\TSR{\bmyz{t}=\parenR{\yz{1,t},\dotsc,\yz{d,t}}}{t=1}{n}$ of
  size~$n$ from a multivariate time series, an
    \mbox{$m$-truncated} estimate
  $\hatlgcsdM{k\ell}{\LGp}{\omega}{m}$ of
  $\lgcsd{k\ell}{\LGp}{\omega}$ is constructed by means of the
  following procedure.
  \begin{enumerate}%[label=(\alph*)]
  \item
    \label{lgcsd_estimator_ecdf}
    Use the univariate marginals $\TSR{\yz{k,t}}{t=1}{n}$ and
    $\TSR{\yz{\ell,t}}{t=1}{n}$ to find estimates $\widehatGz{\!k:n}$
    and $\widehatGz{\!\ell:n}$ of the corresponding marginal
    cumulative distribution functions, and compute from this the
    \textit{pseudo-normalised observations}
    $\TSR{\widehatzz{k,t} \defeq
      \Phiz[-1]{}\!\parenR{\widehatGz{\!k:n}\!\parenR{\yz{k,t}}}}{t=1}{n}$
    and
    $\TSR{\widehatzz{\ell,t} \defeq
      \Phiz[-1]{}\!\parenR{\widehatGz{\!\ell:n}\!\parenR{\yz{\ell,t}}}}{t=1}{n}$.
  \item
    \label{lgcsd_estimator_pseudo_normalisation_I}
        Create the lag~$h$ pseudo-normalised pairs
    $\TSR{\parenR{\widehatzz{k,t+h},\widehatzz{\ell,t}}}{t=1}{n-h}$
    for \mbox{$h=0,\dotsc,m$}, and estimate for the point
    \mbox{$\LGp=\LGpoint$} the local Gaussian cross-correlations
    $\TSR{\hatlgccrb{k\ell}{\LGp}{h}{\bmbzh{h}}}{h=0}{m}$, where
    the $\TSR{\bmbzh{h}}{h=0}{m}$ is the bandwidths that are used for
    the different lags.
  \item
    \label{lgcsd_estimator_pseudo_normalisation_II}
        Create the lag~$h$ pseudo-normalised pairs
    $\TSR{\parenR{\widehatzz{\ell,t+h},\widehatzz{k,t}}}{t=1}{n-h}$
    for \mbox{$h=1,\dotsc,m$}, and estimate for the diagonally
    reflected point \mbox{$\LGpd=\LGpointd$} the local Gaussian
    cross-correlations
    $\TSR{\hatlgccrb{\ell k}{\LGpd}{h}{\bmbzh{h}}}{h=0}{m}$.
  \item
    \label{lgcsd_estimator_folded}
        Adjust \cref{eq:th:lgcs_folding_property} from
    \myref{th:lgcs_properties}{th:lgcs_reflection_property} with some
    lag-window function $\lambdazM{h}{m}$ to get the estimate
    \begin{align}
      \label{eq:App:hatlgsd_definition_main_document}
                                                \hatlgcsdM{k\ell}{\LGp}{\omega}{m} \defeq
      \hatlgccr{k\ell}{\LGp}{0} + \sumss{h=1}{m}
      \lambdazM{h}{m}\cdot \hatlgccr{\ell k}{\LGpd}{h}
      \cdot \ez[+2\pi i\omega h]{} + \sumss{h=1}{m}
      \lambdazM{h}{m}\cdot \hatlgccr{k\ell}{\LGp}{h}
      \cdot \ez[-2\pi i\omega h]{},
    \end{align}
    where the $\TSR{\bmbzh{h}}{h=0}{m}$ has been suppressed from the
    notation in order to get a more compact formula.
  \end{enumerate}
\end{algorithm}

\begin{definition}
  \label{lgsd_related_estimators}
  For a multivariate sample $\TSR{\bmyz{t}}{t=1}{n}$ of size~$n$, as
  described in \cref{def:lgsd_estimator}, the \mbox{$m$-truncated}
  estimates of the local Gaussian versions of the cospectrum,
  quadrature spectrum, amplitude spectrum and phase spectrum is given
  by
  \begin{subequations}
    \label{eq:estimate_local_co_quad_amplitude_phase}
    \begin{align}
      \label{eq:estimate_local_co-spectrum}
      \hatlgcsdCoM{k\ell}{\LGp}{\omega}{m} &\defeq
      \operatorname{Re}\parenR{\hatlgcsdM{k\ell}{\LGp}{\omega}{m}} =
      \hatlgccr{k\ell}{\LGp}{0} + \sumss{h=1}{m}
      \cos\parenR{2\pi\omega h}\lambdazM{h}{m} \parenS{\hatlgccr{k\ell}{\LGp}{h} +
        \hatlgccr{k\ell}{\LGpd}{h}}, \\
      \label{eq:estimate_local_quad-spectrum} 
      \hatlgcsdQuadM{k\ell}{\LGp}{\omega}{m} &\defeq -
      \operatorname{Im}\parenR{\hatlgcsdM{k\ell}{\LGp}{\omega}{m}} =
      \sumss{h=1}{m} \sin\parenR{2\pi\omega h}\lambdazM{h}{m}
      \parenS{\hatlgccr{k\ell}{\LGp}{h} -
        \hatlgccr{k\ell}{\LGpd}{h}}, \\
      \label{eq:estimate_local_amplitude-spectrum}
      \hatlgcsdAmplitudeM{k\ell}{\LGp}{\omega}{m} &\defeq
      \operatorname{Mod}\parenR{\hatlgcsdM{k\ell}{\LGp}{\omega}{m}} =
      \sqrt{\subp{\parenR{\hatlgcsdCoM{k\ell}{\LGp}{\omega}{m}}}{}{}{}{2}
        + \subp{\parenR{\hatlgcsdQuadM{k\ell}{\LGp}{\omega}{m}}}{}{}{}{2}}, \\
      \label{eq:estimate_local_phase-spectrum}
      \hatlgcsdPhaseM{k\ell}{\LGp}{\omega}{m} &\defeq
      \operatorname{Arg}\parenR{\hatlgcsdM{k\ell}{\LGp}{\omega}{m}} \in (-\pi,
      \pi].
    \end{align} 
  \end{subequations}
\end{definition}

The comments in \lgsdRef{sec:Estimation} holds for the present case
too.  In particular, the estimated marginal cumulative distributions
$\widehatGz{\!k:n}$ and $\widehatGz{\!\ell:n}$ from
\myref{def:lgsd_estimator}{lgcsd_estimator_ecdf} can either be based
on the (rescaled) empirical cumulative distribution functions or they
could be built upon a logspline technique like the one implemented in
\citet{otneim2017locally}.  Furthermore, for the asymptotic
investigation, the arguments in \citet[Section~3]{otneim2017locally}
reveals that the pseudo-normalisation of the marginals does not affect
the final convergence rates, which (as was done in
\JT) implies that the present theoretical
analysis can ignore the distinction between the original observations
and the pseudo-normalised observations.

\makeatletter{}
\subsection{Asymptotic theory for
  $\hatlgcsdM{k\ell}{\LGp}{\omega}{m}$}
\label{sec:lgcsd_Asymptotic_theory}

The asymptotic theory for the local Gaussian cross-spectrum
$\hatlgcsdM{k\ell}{\LGp}{\omega}{m}$ follows from a few minor
adjustments of the asymptotic theory that was developed for the local
Gaussian auto-spectra.  As in \lgsdRef{sec:Asymptotic_theory}, the
assumptions and results will be stated for the original observations
instead of the pseudo-normalised observations, since this makes the
analysis easier and since the final convergence rates are unaffected
by this distinction, see
the discussion after \cref{lgsd_related_estimators} for details.

\subsubsection{Some definitions and an assumption for
  $\bmYz{t}=\parenR{\Yz{1,t},\dotsc,\Yz{d,t}}$}
\label{app:assumptions_upon_Yt}

As for the univariate case in \JT, the
assumptions to be imposed on the $k$ and $\ell$ components of the
multivariate times series $\TSR{\bmYz{t}}{t\in\ZZ}{}$ need to be
phrased relative to the bivariate pairs that can be created as
different combinations of elements from the univariate marginals
$\TSR{\Yz{k,t}}{t\in\ZZ}{}$ and $\TSR{\Yz{\ell,t}}{t\in\ZZ}{}$.  Note
that the \textit{folding property} from \cref{lgcsd_estimator_folded}
of \cref{def:lgsd_estimator} implies that it is sufficient to
formulate the assumption based on non-negative values of the lag~$h$.

\begin{definition}
  \label{def:lgcsd_Y_kl_ht_and_Y_lk_ht}
  For a strictly stationary multivariate time series
  $\TSR{\bmYz{t}}{t\in\ZZ}{}$, with
  {$\bmYz{t}=\parenR{\Yz{1,t},\dotsc,\Yz{d,t}}$}, and for a
  selected pair of indices $k$ and $\ell$, define the following
  bivariate pairs from the univariate marginals
  $\TSR{\Yz{k,t}}{t\in\ZZ}{}$ and $\TSR{\Yz{\ell,t}}{t\in\ZZ}{}$.
  \begin{subequations}
    \label{eq:lgcsd_bivariate_pairs_Y_kl_ht_and_Y_lk_ht}
    \begin{align}
      \label{eq:lgcsd_bivariate_pairs_Y_kl_ht}
      \Yht{\!k\ell:h}{t} \defeq \Vector[\prime]{\Yz{k,t+h}, \Yz{\ell,t}},  \qquad
      h \geq 0, \\
      \label{eq:lgcsd_bivariate_pairs_Y_lk_ht}
      \Yht{\!\ell k:h}{t} \defeq \Vector[\prime]{\Yz{\ell,t+h},
        \Yz{k,t}}, \qquad h \geq 1,
    \end{align}
  \end{subequations}
  and let $\ghc[\yhc{k\ell}{h}]{k\ell}{h}$ and $\ghc[\yhc{\ell
      k}{h}]{\ell k}{h}$ denote the respective probability density
  functions.
\end{definition}

The basic idea for the construction of
$\lgcsd{k\ell}{\LGp}{\omega}$ is that a point
\mbox{$\LGp=\LGpoint$} should be selected at which for all~$h$ the
density functions $\ghc[\yhc{k\ell}{h}]{k\ell}{h}$ of
$\Yht{\!k\ell:h}{t}$ will be approximated by
$\psi\!\left(\yhc{k\ell}{h}; \thetahc{}{\LGp}{k\ell}{h}\right)$,
where $\psi$ is the bivariate Gaussian density function from
\cref{eq:psi_density_five_or_one}.

This local investigation requires a bandwidth vector
\mbox{$\bm{b}=\left(\bz{1},\bz{2}\right)$} and a kernel
function~$K(\bm{w})$, which is used to define
$\Khbdefc\defeq \tfrac{1}{\bz{1}\bz{2}} K\!\left( \tfrac{\yz{k,h} -
    \LGpi{1}}{\bz{1}}, \tfrac{\yz{\ell,0} - \LGpi{2}}{\bz{2}}\right)$,
which in turn is used in
\begin{align}
  \label{eq:lgcsd:penalty_qh}
  \qhc{}{\LGp}{k\ell}{h}{b} &\defeq \intss{\RRn{2}}{} \Khbdefc \left[
    \psi\!\left(\yhc{k\ell}{h}; \thetahc{}{\LGp}{k\ell}{h}\right)
    - \log\psi\!\left(\yhc{k\ell}{h};
    \thetahc{}{\LGp}{k\ell}{h}\right) \ghc[\yhc{k\ell}{h}]{k\ell}{h}
    \right] \dyhc{k\ell}{h},
        \end{align}
a minimiser of which should satisfy the vector equation
\begin{align}
  \label{eq:lgcsd:def_of_theta_hb} 
  \intss{\RRn{2}}{} \Khbdefc
  \uhc[\yhc{k\ell}{h};\thetahc{}{\LGp}{k\ell}{h}]{}{k\ell}{h}
  \left[
    \ghc{k\ell}{h}\!\left(\yhc{k\ell}{h}\right) - \psi\!\left(\yhc{k\ell}{h};
    \thetahc{}{\LGp}{k\ell}{h}\right)\right] \dyhc{k\ell}{h} &= \bm{0},
\end{align}
where
\mbox{$\uhc[\yhc{k\ell}{h};\thetahc{}{\LGp}{k\ell}{h}]{}{k\ell}{h}
  \defeq \nablahc{}{k\ell}{h} \log \psi\!\left(\yhc{k\ell}{h};
  \thetahc{}{\LGp}{k\ell}{h}\right)$} is the score function of
$\psi\!\left(\yhc{k\ell}{h}; \thetahc{}{\LGp}{k\ell}{h}\right)$
(with {$\nablahc{}{k\ell}{h} \defeq \partial/\partial
  \thetahc{}{\LGp}{k\ell}{h}$}).  Under the assumption that there is
a bandwidth $\bmbz{k\ell:h:0}$ such that there exists a minimiser
$\thetahbc{}{\LGp}{k\ell}{h}{b}$ of \cref{eq:lgcsd:penalty_qh} which
satisfies \cref{eq:lgcsd:def_of_theta_hb} for any $\bm{b}$ with
\mbox{$\bm{0} < \bm{b} < \bmbz{k\ell:h:0}$}, this
$\thetahbc{}{\LGp}{k\ell}{h}{b}$ will be referred to as the
population~value for the given bandwidth~$\bm{b}$. 

This approach was introduced in a more general context in
\citet{hjort96:_local}, where it was used to define a local approach
to density estimation, and the new idea in \citet{Tjostheim201333} was
to focus upon the estimated local Gaussian parameters
$\hatthetahc{}{\LGp}{k\ell}{h}$ (instead of the estimated densities).
The asymptotic properties of the estimated parameters was investigated
in \citet{Tjostheim201333} by the help of the Klimko-Nelson
approach\footnote{The Klimko-Nelson approach (see \citet{klimko1978})
  shows how the asymptotic properties of \textit{an estimate of the
    parameters of a penalty function $Q$} can be expressed relative to
  the asymptotic properties of (entities related to) the penalty
  function itself.  The interested reader can consult
  \lgsdRef{App:local_penalty_function_Klimko_Nelson_approach} for a
  more detailed presentation of the Klimko-Nelson approach when a
  local penalty-function is used.} and a suitably defined local
penalty function
$\QhNc[\thetahc{}{\LGp}{k\ell}{h}]{}{\LGp}{k\ell}{h}{n}$ (see
\cref{eq:lgcsd:QhN} in \cref{sec:the_bivariate_case}).

The assumptions to be imposed on $\bmYz{t}$ is related to the
estimation of \cref{eq:lgcsd:def_of_theta_hb}, and thus requires a few
additional definitions.

\begin{definition}
  \label{def:Uh}
  For \mbox{$\psi\!\left(\yhc{k\ell}{h}; \bmthetaz{}\right)$} the
  local Gaussian density used when
  approximating~$\ghc[\yhc{k\ell}{h}]{k\ell}{h}$ at the point
  \mbox{$\LGp=\LGpoint$}, and for $\thetahbc{}{\LGp}{k\ell}{h}{b}$
  the population value that minimises the penalty
  function~$\qhc{}{\LGp}{k\ell}{h}{b}$ from
  \cref{eq:lgcsd:penalty_qh}, define for all \mbox{$h \in \NN$} and
  all \mbox{$q \in \parenC{1,\dotsc,5}$}
  \begin{align}
    \label{eq_def:Uh}
    \Uh[\bm{w}]{k\ell:h:q:\bm{b}} &\defeq \left.
    \frac{\partial}{\partial \thetahcomp{}{q}} \log
    \left(\psi\!\left(\yhc{k\ell}{h}; \bmthetaz{}\right)\right)
    \right|_{\left(\yhc{k\ell}{h};\, \bmthetaz{}\right) =
      \left(\bm{w};\, \thetahbc{}{\LGp}{k\ell}{h}{b}\right) },
  \end{align}
  where $\partial/\partial\thetahcomp{}{q}$ is the
  $q^{\operatorname{th}}$ partial derivative (with respect to
  $\bmthetaz{}$).
\end{definition}

The following requirements on the kernel function are identical to
those given in \lgsdRef{def:kernel}.

\begin{definition}
  \label{def:lgcsd:kernel}
  From a bivariate, non-negative, and bounded kernel function
  $K(\bm{w})$, that satisfies
  \begin{subequations}
    \label{eq:kernel_integrals}
    \begin{align}
      \label{eq:kernel_integral_one}
      &\intss{\RRn{2}}{} \!\! K\!\left(\wz{1},\wz{2}\right)
      \d{\wz{1}}\!\d{\wz{2}} = 1,\\
      \label{eq:kernel_integrals_left_wellbehaved}
      &\mathcalKz{1:k} \! \left(\wz{2}\right) \defeq \intss{\RRn{1}}{}
      \!\!  K\!\left(\wz{1},\wz{2}\right) \wz[k]{1} \d{\wz{1}} \qquad
      \text{is bounded for } k \in \TSR{0,1,2}{}{}, \\
      \label{eq:kernel_integrals_right_wellbehaved}
      &\mathcalKz{2:\ell} \! \left(\wz{1}\right) \defeq
      \intss{\RRn{1}}{} \!\!  K\!\left(\wz{1},\wz{2}\right)
      \wz[\ell]{2} \d{\wz{2}} \qquad \text{is bounded for } \ell \in
      \TSR{0,1,2}{}{}, \\
      \label{eq:kernel_integrals_finite}
      &\intss{\RRn{2}}{} \!\! K\!\left(\wz{1},\wz{2}\right)
      \absp{\wz[k]{1} \wz[\ell]{2}} \d{\wz{1}}\!\d{\wz{2}} < \infty,
      \qquad k, \ell \geq 0 \text{ and } k + \ell \leq
      2\cdot\ceil{\nu},
    \end{align}
  \end{subequations}
  where \mbox{$\nu>2$} is from
  \myref{assumption_lgcsd_Yt}{assumption_lgcsd_Yt_strong_mixing} (and
  $\ceil{\cdot}$ is the ceiling function), define
  \begin{align}
    \label{eq:definition_of_K}
    \Khb[\yhc{k\ell}{h}-\LGp]{h}{\bm{b}} &\defeq \frac{1}{\bz{1}\bz{2}}
    K\!\left( \frac{\yz{h}-\vz{1}}{\bz{1}}, \frac{\yz{0}-\vz{2}}{\bz{2}}
    \right).
  \end{align}
\end{definition}

\begin{definition}
  \label{def:Xklht}
  Based on $\Yhtc{k\ell}{h}{t}$, $\Uh[\bm{w}]{k\ell:h:q:\bm{b}}$ and
  $\Khb[\yhc{k\ell}{h}-\LGp]{h}{\bm{b}}$, define the new bivariate
  random variables $\XNht[\LGp]{n}{k\ell:h:q}{t}$ as follows,
  \begin{align}
    \label{eq:definition_of_XklNht}
    \XNht[\LGp]{n}{k\ell:h:q}{t} &\defeq \sqrt{\bz{1}\bz{2}}
    \Khb[\Yhtc{k\ell}{h}{t}-\LGp]{h}{\bm{b}}
    \Uh[\Yhtc{k\ell}{h}{t}]{k\ell:h:q:\bm{b}}.
  \end{align}
\end{definition}

Note that a product of the random variables
$\XNht[\LGp]{n}{k\ell:h:q}{t}$ and $\XNht[\LGp]{n}{k\ell:i:r}{s}$ will
be a function of $\Yz{k,t+h}$, $\Yz{\ell,t}$, $\Yz{k,s+i}$ and
$\Yz{\ell,s}$, which depending on the configuration of the indices
$h, i, s, t$ will be either a bivariate, trivariate or tetravariate
function.  The expectation of the product
$\XNht[\LGp]{n}{k\ell:h:q}{t}\cdot\XNht[\LGp]{n}{k\ell:i:r}{s}$ will
thus (depending on these indices) either require a bivariate,
trivariate or tetravariate density function.

\begin{assumption}
  \label{assumption_lgcsd_Yt}
  The multivariate process $\TSR{\bmYz{t}}{t\in\ZZ}{}$ will be assumed
  to satisfy the following properties, with \mbox{$\LGp=\LGpoint$} (in
  \cref{assumption:lgcsd:gh_differentiable_at_LGp} below) the point at which
  $\hatlgcsdM{k\ell}{\LGp}{\omega}{m}$, the estimate of
  $\lgcsd{k\ell}{\LGp}{\omega}$, is to be computed.
  \begin{enumerate}%[label=(\alph*)]
  \item
    \label{assumption_lgcsd_Yt_strictly_stationary}
        $\TSR{\bmYz{t}}{t\in\ZZ}{}$ is strictly stationary.
  \item
    \label{assumption_lgcsd_Yt_strong_mixing}
        $\TSR{\bmYz{t}}{t\in\ZZ}{}$ is strongly mixing, with mixing
    coefficient $\alpha(j)$ satisfying
    \begin{align}
      \label{eq:alpha_requirement}
      \sumss{j=1}{\infty} \jz[a]{} \parenSz[1-2/\nu]{}{\alpha(j)} <
      \infty \qquad \text{for some $\nu>2$ and $a > 1 - 2/\nu$}.
    \end{align}
  \item
    \label{assumption_lgcsd_Yt_finite_variance}
        $\E{\subp{||\bmYz{t}||}{}{}{}{2}} < \infty$, where $||\cdot||$
    is the Euclidean norm.
  \end{enumerate}
  The bivariate density functions $\ghc[\yhc{k\ell}{h}]{k\ell}{h}$ and
  $\ghc[\yhc{\ell k}{h}]{\ell k}{h}$, corresponding to the lag~$h$
  pairs introduced in
  \cref{eq:lgcsd_bivariate_pairs_Y_kl_ht_and_Y_lk_ht}, must satisfy
  the following requirements for a given point \mbox{$\LGp=\LGpoint$}.
  \begin{enumerate}[resume]%[label=(\alph*),resume]
  \item
    \label{assumption:lgcsd:gh_differentiable_at_LGp}
        $\ghc[\yhc{k\ell}{h}]{k\ell}{h}$ is differentiable at $\LGp$, such that 
    Taylor's theorem can be used to write $\ghc[\yhc{k\ell}{h}]{k\ell}{h}$~as
    \begin{align}
      \label{eq_assumption:lgcsd:gh_differentiable_at_LGp}
      \ghc[\yhc{k\ell}{h}]{k\ell}{h}
      & = \gh[\LGp]{h} +
        \power[\prime]{}{\bmmathfrakgz{h}(\LGp)} \left[\yhc{k\ell}{h}-\LGp
        \right] + \power[\prime]{}{\bmmathfrakRz{h}(\yhc{k\ell}{h})}
        \left[\yhc{k\ell}{h}-\LGp \right], \\
      \nonumber
      &\text{where }         \bmmathfrakgz{h}(\LGp)
        = \Vector[\prime]{
        \left.\tfrac{\partial}{\partial \yz{h}}
        \ghc[\yhc{k\ell}{h}]{k\ell}{h}\right|_{\yhc{k\ell}{h} =\, \LGp},
        \left.\tfrac{\partial}{\partial \yz{0}}
        \ghc[\yhc{k\ell}{h}]{k\ell}{h}\right|_{\yhc{k\ell}{h} =\,
        \LGp}} \\
      \nonumber
      &\text{and } \lim_{\yhc{k\ell}{h}\longrightarrow\, \LGp}
        \bmmathfrakRz{h}(\yhc{k\ell}{h}) = 0,
    \end{align}
    with the same requirement for $\ghc[\yhc{\ell k}{h}]{\ell k}{h}$
    at the diagonally reflected point
    \mbox{$\LGpd=\parenR{\LGpi{2},\LGpi{1}}$}.
  \item 
    \label{assumption:lgcsd:gh_bh0}
        There exists a bandwidth $\bmbz{k\ell:h:0}$ such that there for
    every \mbox{$\bm{0} < \bm{b} < \bmbz{k\ell:h:0}$} is a unique
    minimiser $\thetahbc{}{\LGp}{k\ell}{h}{b}$ of the penalty
    function~$\qhc{}{\LGp}{k\ell}{h}{b}$ from
    \cref{eq:lgcsd:penalty_qh}.
  \item
    \label{assumption:lgcsd:gh_b0_infimum_of_bh0}
        The collection of bandwidths $\TSR{\bmbz{k\ell:h:0}}{h\in\ZZ}{}$
    has a positive infimum, i.e., there exists a $\bmbz{k\ell:0}$ such
    that\footnote{Inequalities involving vectors are to be interpreted
      in a component-wise manner.}
    \begin{align}
      \bm{0} < \bmbz{k\ell:0} \defeq \inf_{h\in\ZZ} \bmbz{k\ell:h:0},
    \end{align}
    which implies that this $\bmbz{k\ell:0}$ can be used
    simultaneously for all the lags.
  \item
    \label{assumption:lgcsd:h_x:y_x:y:z_finite_expectations}
                                                        For $\XNht[\LGp]{n}{k\ell:h:q}{t}$ from \cref{def:Xklht}, the
    related bivariate, trivariate and tetravariate density functions
    must be such that the expectations
    $\E{\XNht[\LGp]{n}{k\ell:h:q}{t}}$,
    $\E{\absp[\nu]{\XNht[\LGp]{n}{k\ell:h:q}{t}}}$ and
    $\E{\XNht[\LGp]{n}{k\ell:h:q}{t}\cdot\XNht[\LGp]{n}{k\ell:i:r}{s}}$
    all are finite.
  \end{enumerate}
\end{assumption}

The present \cref{assumption_lgcsd_Yt} is in essence identical to
\lgsdRef{assumption_Yt} with some extra indices, so the remarks from
\JT is of intereset here too.  In particular, the
\mbox{$\alpha$-mixing} requirement in
\cref{assumption_lgcsd_Yt_strong_mixing} implies that $\Yz{k,t+h}$
and~$\Yz{\ell,t}$ will be asymptotically independent as $\hlimit$,
i.e., the bivariate density functions $\ghc[\yhc{k\ell}{h}]{k\ell}{h}$
will for large lags~$h$ approach the product of the marginal
densities, and the situation will thus stabilise when $h$ is large
enough.  This is in particular of importance for
\cref{assumption:lgcsd:gh_b0_infimum_of_bh0}, since it implies that it
will be possible to find a nonzero $\bmbz{k\ell:0}$ that works for
all~$h$.  Moreover, the finiteness assumptions in
\cref{assumption:lgcsd:h_x:y_x:y:z_finite_expectations} are always
trivially satisfied if the required density-functions are~finite.

\subsubsection{An assumption for $\parenR{\Yz{k,t},\Yz{\ell,t}}$ and
  the score function $\bmuz{}(\bm{w};\bmthetaz{})$ of
  $\psi(\bm{w};\bmthetaz{})$}
\label{sec:lgcsd:assumption_score_function}

The following assumption is in essence identical to
\lgsdRef{assumption_score_function}, which was included due to the
need for the asymptotic results from \citet{Tjostheim201333} to be
applied for all the different lags~$h$.

\begin{assumption}
  \label{assumption:lgcsd:score_function}
  The collection of local Gaussian parameters
  $\TSR{\thetahc{}{\LGp}{k\ell}{h}}{}{}$ at the point $\LGp$ for the
  bivariate probability density functions $\gh[\yh{k\ell:h}]{k\ell:h}$, must all
  be such that
  \begin{enumerate}%[label=(\alph*)]
  \item     \label{assumption:lgcsd:score_function_finite_h}
        $\bmuz{}(\LGp;\thetahc{}{\LGp}{k\ell}{h})\neq\bm{0}$ for
    all finite $h$.
  \item     \label{assumption:lgcsd:score_function_limit_h}
        $\lim_{h\rightarrow\infty} \bmuz{}(\LGp;\thetahc{}{\LGp}{k\ell}{h})\neq\bm{0}$.
  \end{enumerate}
\end{assumption}

Note that an inspection of the $5$ equations in
\mbox{$\bmuz{}(\bm{w};\bmthetaz{})=\bm{0}$} can be used to identify
when
\cref{assumption:lgcsd:score_function_finite_h,assumption:lgcsd:score_function_limit_h}
of \cref{assumption:lgcsd:score_function} might fail, cf.\ the
discussion in \lgsdRef{sec:assumption_score_function} for further
details.

\subsubsection{Assumptions for $n$, $m$ and \mbox{$\bm{b}=\parenR{\bz{1},\bz{2}}$}}
\label{app:lgcsd:assumptions_upon_mNb}

The following assumption is identical to \lgsdRef{assumption_Nmb}.
The internal consistency of it was verified in
\lgsdRef{th:block_sizes_for_main_result}.

\begin{assumption}
  \label{assumption:lgcsd:Nmb}
  Let \mbox{$m \defeq \mz{n} \rightarrow \infty$} be a sequence of
  integers denoting the number of lags to include, and let
  \mbox{$\bm{b} \defeq \bmbz{n} \rightarrow \bmzeroz[+]{}$} be the
  bandwidths used when estimating the local Gaussian correlations for
  the lags \mbox{$h = 1,\dotsc,m$} (based on $n$ observations).  Let
  $\bz{1}$ and $\bz{2}$ refer to the two components of $\bm{b}$, and
  let $\alpha$, $\nu$ and~$a$ be as introduced in
  \myref{assumption_lgcsd_Yt}{assumption_lgcsd_Yt_strong_mixing}.  Let
  \mbox{$s\defeq \sz{n} \rightarrow\infty$} be a sequence of integers
  such that \mbox{$s=\oh{\sqrt{n\bz{1}\bz{2}/m}}$}, and let $\tau$ be
  a positive constant.  The following requirements must be satisfied
  for these entities.\footnote{ Notational convention:
    \enquote{$\vee$} denotes the maximum of two numbers, whereas
    \enquote{$\wedge$} denotes the minimum.}
  \begin{enumerate}%[label=(\alph*)]
  \item
    \label{eq:lgcsd:assumption_N_and_b1b2}
    $\log n / n\!  \parenRz[5]{}{\bz{1}\bz{2}} \longrightarrow 0$.
  \item
    \label{eq:lgcsd:assumption_Nb1b2/m}
    $n\bz{1}\bz{2}/m \longrightarrow \infty$.
  \item
    \label{eq:lgcsd:assumption_m(b1 join b2)}
    $\mz[\delta]{}\!\left(\bz{1}\vee\bz{2}\right) \longrightarrow 0,
    \text{ where } \delta = 2 \vee \tfrac{\nu(a+1)}{\nu(a-1)-2}$.
  \item
    \label{eq:lgcsd:assumption_alpha(s)_and_(N/b1b2)^1/2}
    $\sqrt{nm/\bz{1}\bz{2}}\cdot \sz[\tau]{}\!\cdot\alpha(s-m+1)
    \longrightarrow \infty$.
  \item
    \label{eq:lgcsd:assumption_m=o((Nb1b2)^tau/(2+5tau)-lambda)}
    $m = \oh{\parenRz[\tau/(2+5\tau)-\lambda]{}{n\bz{1}\bz{2}}},
    \text{ for some } \lambda \in \left(0, \tau/(2+5\tau)\right)$.
  \item
    \label{eq:lgcsd:assumption_m=o(s)}
    $m = \oh{s}$.
  \end{enumerate}
\end{assumption}

\makeatletter{}
\subsubsection{Convergence theorems for
  $\hatlgcsdM{k\ell}{\LGp}{\omega}{m}$,
  $\hatlgcsdAmplitudeM{k\ell}{\LGp}{\omega}{m}$ and
  $\hatlgcsdPhaseM{k\ell}{\LGp}{\omega}{m}$}
\label{seq:lgcsd:convergence_theorems}
See \cref{sec:proofs} in the Supplementary Material for the proofs of
the theorems stated below.

\begin{theorem}
  \label{th:lgcsd:asymptotics_for_hatlgcsd}
  The estimate $\hatlgcsdM{k\ell}{\LGp}{\omega}{m} =
  \hatlgcsdCoM{k\ell}{\LGp}{\omega}{m} -
  i\cdot\hatlgcsdQuadM{k\ell}{\LGp}{\omega}{m}$ of the local
  Gaussian cross-spectrum $\lgcsd{k\ell}{\LGp}{\omega} =
  \lgcsdCo{k\ell}{\LGp}{\omega} -
  i\cdot\lgcsdQuad{k\ell}{\LGp}{\omega}$, will under
  \cref{assumption_lgcsd_Yt,assumption:lgcsd:score_function,assumption:lgcsd:Nmb} satisfy
    \begin{align}
    \label{eq:lgcsd:th:asymptotics_for_hatlgsd_off_diagonal_Yt_not_reversible}
    \sqrt{n \!\parenRz[3]{}{\bz{1}\bz{2}} \!/ m} \cdot     \parenR{
      \begin{bmatrix}
        \hatlgcsdCoM{k\ell}{\LGp}{\omega}{m} \\
        \hatlgcsdQuadM{k\ell}{\LGp}{\omega}{m}
      \end{bmatrix}
      -     \begin{bmatrix}
        \lgcsdCo{k\ell}{\LGp}{\omega} \\
        \lgcsdQuad{k\ell}{\LGp}{\omega}
    \end{bmatrix} }
    \stackrel{\scriptscriptstyle d}{\longrightarrow}
    \UVN{
      \begin{bmatrix}
        0 \\
        0
    \end{bmatrix}}{
      \begin{bmatrix}
        \sigmaz[2]{\!c|k\ell:\LGp}(\omega) & 0 \\
        0 &\sigmaz[2]{\!q|k\ell:\LGp}(\omega) 
    \end{bmatrix}},
  \end{align}
  when \mbox{$\omega \not\in \tfrac{1}{2}\cdot\ZZ = \parenC{\dotsc,-1, -\tfrac{1}{2}, 0,
      \tfrac{1}{2}, 1, \dotsc}$},
  where the variances $\sigmaz[2]{\!c|k\ell:\LGp}(\omega)$ and
  $\sigmaz[2]{\!q|k\ell:\LGp}(\omega)$ are given by
  \begin{subequations}
    \label{eq:lgcsd:th:asymptotics_for_hatlgsd_variance_RE_and_IM}
    \begin{align}
    \label{eq:lgcsd:th:asymptotics_for_hatlgsd_variance_RE}
      \sigmaz[2]{\!c|k\ell:\LGp}(\omega) &= \lim_{\mlimit} \frac{1}{m}
      \parenR{ \tildesigmaz[2]{\!\LGp|k\ell}(0) + \sumss{h=1}{m}
        \lambdazM[2]{h}{m} \cdot \subp{\cos}{}{}{}{2} (2\pi\omega h)
        \cdot \parenC{ \tildesigmaz[2]{\!\LGp|k\ell}(h) +
          \tildesigmaz[2]{\!\LGpd|\ell k}(h) }} \\
    \label{eq:lgcsd:th:asymptotics_for_hatlgsd_variance_IM}
      \sigmaz[2]{\!q|k\ell:\LGp}(\omega) &= \lim_{\mlimit} \frac{1}{m} \parenR{
      \sumss{h=1}{m} \lambdazM[2]{h}{m} \cdot \subp{\sin}{}{}{}{2}
      (2\pi\omega h) \cdot \parenC{ \tildesigmaz[2]{\!\LGp|k\ell}(h) +
        \tildesigmaz[2]{\!\LGpd|\ell k}(h) }},
    \end{align}
  \end{subequations}
  with $\tildesigmaz[2]{\!\LGp|k\ell}(h)$ and
  $\tildesigmaz[2]{\!\LGpd|\ell k}(h)$ the asymptotic variances
  related to the estimates $\hatlgccr{k\ell}{\LGp}{h}$ and
  $\hatlgccr{\ell k}{\LGpd}{h}$, see
  \cref{th:asymptotic_result_cross_correlation} in the Supplementary
  Material for the details.

  The local Gaussian quadrature spectrum is identical to zero when
  \mbox{$\omega \in \tfrac{1}{2}\cdot\ZZ$}, and for those frequencies
  the following asymptotic result holds under the given assumptions
  \begin{align}
    \label{eq:lgcsd:th:asymptotics_for_hatlgsd_off_diagonal_Yt_not_reversible_real_valued_case}
    \sqrt{n \!\parenRz[3]{}{\bz{1}\bz{2}} \!/ m} \cdot     \parenR{
    \hatlgcsdM{k\ell}{\LGp}{\omega}{m} -
    \lgcsd{k\ell}{\LGp}{\omega} }
    \stackrel{\scriptscriptstyle d}{\longrightarrow}
    \UVN{0}{\sigmaz[2]{\!c|k\ell:\LGp}(\omega)}, \qquad \omega \in \tfrac{1}{2}\cdot\ZZ.
  \end{align}
  
\end{theorem}

The asymptotic results for the local Gaussian amplitude- and
phase-spectra is a direct consequence of
\cref{th:lgcsd:asymptotics_for_hatlgcsd} and \citet[proposition~6.4.3,
p.~211]{Brockwell:1986:TST:17326}.

\begin{theorem}
  \label{th:lgcsd:asymptotics_for_amplitude_spectrum}
                Under \cref{assumption_lgcsd_Yt,assumption:lgcsd:score_function,assumption:lgcsd:Nmb}, when
  \mbox{$\lgcsdAmplitude{k\ell}{\LGp}{\omega}>0$} and
  \mbox{$\omega \not\in \tfrac{1}{2}\cdot\ZZ$}, the estimate
  $\hatlgcsdAmplitudeM{k\ell}{\LGp}{\omega}{m} =
  \sqrt{\subp{\parenR{\hatlgcsdCoM{k\ell}{\LGp}{\omega}{m}}}{}{}{}{2}
    +
    \subp{\parenR{\hatlgcsdQuadM{k\ell}{\LGp}{\omega}{m}}}{}{}{}{2}}$
  satisfies
    \begin{align}
    \label{eq:th:asymptotics_for_amplitude_spectrum}
    \sqrt{n \!\parenRz[3]{}{\bz{1}\bz{2}} \!/ m} \cdot \parenR{
      \hatlgcsdAmplitudeM{k\ell}{\LGp}{\omega}{m} -
      \lgcsdAmplitude{k\ell}{\LGp}{\omega} }
    \stackrel{\scriptscriptstyle d}{\longrightarrow}
    \UVN{0}{\sigmaz[2]{\!\alpha}(\omega)},
  \end{align}
  where $\sigmaz[2]{\!\alpha}(\omega)$ is given relative to
  $\sigmaz[2]{\!c|k\ell:\LGp}(\omega)$ and
  $\sigmaz[2]{\!q|k\ell:\LGp}(\omega)$ (from
  \cref{eq:lgcsd:th:asymptotics_for_hatlgsd_variance_RE_and_IM},
  \cref{th:lgcsd:asymptotics_for_hatlgcsd}) as
  \begin{align}
    \label{eq:variance_for_amplitude}
    \sigmaz[2]{\!\alpha} &= \parenR{
      \lgcsdCoM{k\ell}{\LGp}{\omega}{2}\cdot\sigmaz[2]{\!c|k\ell:\LGp}(\omega)
      +
      \lgcsdQuadM{k\ell}{\LGp}{\omega}{2}\cdot\sigmaz[2]{\!q|k\ell:\LGp}(\omega)
    } / \lgcsdAmplitudeM{k\ell}{\LGp}{\omega}{2}.
  \end{align}
\end{theorem}

\begin{theorem}
  \label{th:lgcsd:asymptotics_for_phase_spectrum}
                Under \cref{assumption_lgcsd_Yt,assumption:lgcsd:score_function,assumption:lgcsd:Nmb}, when
  \mbox{$\lgcsdAmplitude{k\ell}{\LGp}{\omega}>0$} and
  \mbox{$\omega \not\in \tfrac{1}{2}\cdot\ZZ$}, the estimate
  $\hatlgcsdPhaseM{k\ell}{\LGp}{\omega}{m} =
  \operatorname{args}\!\parenR{\hatlgcsdCoM{k\ell}{\LGp}{\omega}{m}
    - i\cdot\hatlgcsdQuadM{k\ell}{\LGp}{\omega}{m}}$ satisfies
    \begin{align}
    \label{eq:th:asymptotics_for_phase_spectrum}
    \sqrt{n \!\parenRz[3]{}{\bz{1}\bz{2}} \!/ m} \cdot \parenR{
      \hatlgcsdPhaseM{k\ell}{\LGp}{\omega}{m} -
      \lgcsdPhase{k\ell}{\LGp}{\omega} }
    \stackrel{\scriptscriptstyle d}{\longrightarrow}
    \UVN{0}{\sigmaz[2]{\!\phi}(\omega)},
  \end{align}
  where $\sigmaz[2]{\!\phi}(\omega)$ is given relative to
  $\sigmaz[2]{\!c|k\ell:\LGp}(\omega)$ and
  $\sigmaz[2]{\!q|k\ell:\LGp}(\omega)$ (from
  \cref{eq:lgcsd:th:asymptotics_for_hatlgsd_variance_RE_and_IM},
  \cref{th:lgcsd:asymptotics_for_hatlgcsd}) as
  \begin{align}
    \label{eq:variance_for_phase}
    \sigmaz[2]{\!\phi}(\omega) &= \parenR{
      \lgcsdQuadM{k\ell}{\LGp}{\omega}{2}\cdot\sigmaz[2]{\!c|k\ell:\LGp}(\omega)
      +
      \lgcsdCoM{k\ell}{\LGp}{\omega}{2}\cdot\sigmaz[2]{\!q|k\ell:\LGp}(\omega)
    } / \lgcsdAmplitudeM{k\ell}{\LGp}{\omega}{4}.
  \end{align}
\end{theorem}

The asymptotic normality results in
\cref{th:lgcsd:asymptotics_for_hatlgcsd,th:lgcsd:asymptotics_for_amplitude_spectrum,th:lgcsd:asymptotics_for_phase_spectrum}
do not necessarily help much if computations of pointwise confidence
intervals for the estimated local Gaussian estimates are of interest,
since it in practice may be unfeasible to find decent estimates of the
variances $\sigmaz[2]{\!c|k\ell:\LGp}(\omega)$ and
$\sigmaz[2]{\!q|k\ell:\LGp}(\omega)$ that occurs in
\cref{th:lgcsd:asymptotics_for_hatlgcsd}.  The pointwise confidence
intervals will thus later in the paper either be estimated based on
suitable quantiles obtained by repeated sampling from a known
distribution, or they will be based on bootstrapping techniques for
those cases where real data has been investigated.  Confer
\citet[ch.~7.2.5 and 7.2.6]{terasvirta2010modelling} for further
details with regard to the need for bootstrapping in such situations.

\section{Visualisations and interpretations}
\label{sec:lgch_examples}

This section will show how different visualizations of the estimates
$\hatlgcsdM{k\ell}{\LGp}{\omega}{m}$ of the \mbox{$m$-truncated} local
Gaussian cross-spectrum $\lgcsdM{k\ell}{\LGp}{\omega}{m}$ can be used to
detect nonlinear dependency structures in a multivariate time series.
The plots encountered here are natural extensions of those introduced
in \JT, but as $\hatlgcsdM{k\ell}{\LGp}{\omega}{m}$ is complex-valued,
the actual investigation will be based on plots of the corresponding
local Gaussian versions of the cospectrum, quadrature spectrum, phase
spectrum and amplitude spectrum.

Technical details, and the description of the selected tuning
parameters of $\hatlgcsdM{k\ell}{\LGp}{\omega}{m}$, are given in
\cref{sec:lgcsd_Examples_the_parameters}.  This material has also been
covered in \lgsdRef{sec:Examples}, but it is included here for the
convenience of the reader.

A sanity test of the implemented estimation algorithm is presented in
\cref{sec:lgcsd_Some_simulations}, and it is there seen that
$\hatlgcsdM{k\ell}{\LGp}{\omega}{m}$ can detect local periodic
structures in an example where a heuristic argument enables the
prediction of the anticipated result.
\Cref{sec:lgcsd_Some_simulations} also contains a procedure, based on
combined \textit{heatmap and distance}-plots, that can help an
investigator detect local regions of interest.
\Cref{sec:lgch:EuStockMarkets+cGARCH}
applies the local Gaussian machinery to the log-returns of the
\EuStockMarkets-data, and it also contains the results from a
GARCH-type model fitted to these log-returns.

A comparison of the results from the original data and the fitted
model can reveal to what extent the internal dependency structure of
the fitted model actually reflects the dependency structure of the
original sample, and this might be of interest with regard to model
selection.  As in \JT, it will be seen that
plots based on $\hatlgcsdM{k\ell}{\LGp}{\omega}{m}$ can be useful
as an exploratory tool, and this approach might detect nonlinear
dependencies and periodicities between the variables, which can not be
detected by ordinary cross spectral analysis.

\subsection{The input parameters and some other technical details}
\label{sec:lgcsd_Examples_the_parameters}

Several tuning parameters must be selected in order to compute the
$m$-truncated local Gaussian cross-spectrum density estimates
$\hatlgcsdM{k\ell}{\LGp}{\omega}{m}$, and the values used for the
plots in this section are given below.  The parameters are mostly the
same as those used when estimating the local Gaussian auto-spectra in
\JT, with the main exception that the value
$0.6$ is used instead of $0.5$ in the bandwidth vector $\bm{b}$.  This
adjustment is due to the fact that the time series in this paper are
sligthly shorter than those used in \JT, i.e.,
length 1859 versus length~1974.  (All simulated time series have the
same length as the one encountered for the real data example, i.e.,
the log-returns of the \EuStockMarkets-data.)

Note that these tuning parameters have been selected in order to
provide a \textit{proof of concept} for the fact that nonlinear
dependency structures can be detected by this approach, and the quest
for \enquote{optimal parameters} is a topic for further work.  The
interested reader can consult \cref{P2.app:sensitivity_analysis} in
the Supplementary Material for a sensitivity analysis of the different
tuning parameters.

\textbf{The pseudo-normalization:} The initial step of the computation
of $\hatlgcsdM{k\ell}{\LGp}{\omega}{m}$ is to replace the
observations
$\TSR{\parenR{\yz{1,t},\dotsc,\yz{d,t}}}{t=1}{n}$ with the
corresponding pseudo-normalized observations
$\TSR{\parenR{\hatzz{1,t},\dotsc,\hatzz{d,t}}}{t=1}{n}$,
cf.\ \cref{def:lgsd_estimator}, that is, estimates of the $d$ marginal
cumulative density functions $\TSR{\Gz{\ell}}{\ell=1}{d}$ are needed.
The present analysis has used the rescaled empirical cumulative
density functions $\hatGz{\ell:n}$ for this purpose, but the
computations could also have been based on logspline-estimates
of~$\Gz{\ell}$.  A preliminary test revealed that the two
normalization procedures created strikingly similar estimates of
$\hatlgcsdM{k\ell}{\LGp}{\omega}{m}$, so the computationally faster
approach based on the rescaled empirical cumulative density-function
has thus been applied for the present investigation.
 
\textbf{The points $\LGp$ of investigation:} Three diagonal points,
with coordinates corresponding to the 10\%, 50\% and 90\% percentiles
of the standard normal distribution,\footnote{The corresponding
  coordinates are $(-1.28,-1.28)$, $(0,0)$ and $(1.28,1.28)$.} will be
used in the basic plots in this section.
Confer \cref{P2.app:Bandwidth_sensitivity} for further details related
to the selection of $\LGp$, and see
\cref{fig:heatmap_co_quad_dmt_bivariate_constant_phases,fig:heatmap_co_quad_dmt_bivariate_different_phases,fig:heatmap_distance_plot_EuStockMarkets_DAX_CAC}
for some heatmap-based plots.

\textbf{The lag-window function $\lambdazM{h}{m}$:} The smoothing of the estimated local Gaussian autocorrelations, cf.\
\myref{def:lgsd_estimator}{def:lgsd_esitimator_folded}, was done by
the Tukey-Hanning lag-window kernel:
$\lambdazM{h}{m} = \tfrac{1}{2} \cdot \left(1 + \cos\left(\pi\cdot
    \tfrac{h}{m}\right) \right)$ for $|h| \leq m $,
$\lambdazM{h}{m} =0$ for $|h|>m$.

\textbf{The bandwidth $\bm{b}$:} The estimation of the local Gaussian
autocorrelations requires the selection of a bandwidth-vector
$\bm{b} = \parenR{\bz{1},\bz{2}}$, and the plots in this section have
used \mbox{$\bm{b}=(0.6,0.6)$}.  It was noted in
\JT that it was natural to require that the
bandwidth \mbox{$\bm{b}=\parenR{\bz{1},\bz{2}}$} should satisfy
\mbox{$\bz{1}=\bz{2}$} when the local Gaussian autocorrelations
$\lgccr{kk}{\LGp}{h}$ should be estimated, since both of the
components in the lag~$h$ pseudo-normalised pairs originated from the
same univariate time series.  For the estimation of local Gaussian
cross-correlations $\lgccr{k\ell}{\LGp}{h}$, it is the
pseudo-normalisation of the marginals that justifies the assumed
equality of $\bz{1}$ and~$\bz{2}$.  A discussion related to the choice
of bandwidth can be found in
\cref{P2.How.to.select.the.tuning.parameters?}

\textbf{The truncation level $m$:} The value \mbox{$m=10$} was used
for the truncation level, since it was possible to detect nonlinear
dependency structures even for that low truncation level.

\textbf{The number of replicates $R$:} The estimated values (means and
90\% pointwise confidence intervals) have been based on $R=100$
replicates.  Simulations were used for the cases with known parametric
models, whereas the bootstrapped based resampling strategy developed
in \JT was used for the real data example.  The
relevant details about this resampling strategy are for the
convenience of the reader included in
\cref{P2.app:regarding_resampling} in the Supplementary Material.

\textbf{Numerical convergence:} The \Rpackage \Rref{localgauss}, see
\citet{Berentsen:2014:ILR}, estimates the local Gaussian
cross-correlations $\lgccr{k\ell}{\LGp}{h}$ and returns them
together with an attribute that reveals whether or not the estimation
algorithm converged numerically.  The $m$-truncated estimates
$\hatlgcsdM{k\ell}{\LGp}{\omega}{m}$ inherits the
convergence-attributes from the estimates
$\TSR{\hatlgccr{k\ell}{\LGp}{h}}{h=-m}{m}$ , and either
\enquote{\texttt{NC = OK}} or \enquote{\texttt{NC = FAIL}} will be
added to the plot depending on the convergence status.  Note that
convergence-problems hardly occur when the computations are based on
pseudo-normalized observations.

\textbf{Estimation aspects for a given parameter configuration:} The
estimation of $\hatlgcsdM{k\ell}{\LGp}{\omega}{m}$ for a point
\mbox{$\LGp = \LGpoint$} is based on the estimates
$\TSR{\hatlgccr{k\ell}{\LGp}{h}}{h=-m}{m}$, and it can thus be of
interest to see how these estimates depend on the configuration of the
above mentioned tuning parameters.  A detailed discussion can be found
in
\lgsdRef{sec:estimation_aspect_for_the_given_parameter_configuration}.
The key observation is that the amount of pseudo-normalised
observations close to the point $\LGp$ will be much lower when $\LGp$
lies in one of the tails.  This implies a higher variability of the
estimated values $\TSR{\hatlgccr{k\ell}{\LGp}{h}}{h=-m}{m}$ for points
in the tails, and this is the reason that the estimated pointwise
confidence intervals for $\hatlgcsdM{k\ell}{\LGp}{\omega}{m}$ are much
wider for points in the tails.

\textbf{Reproducibility and interactive investigations:} The \Rpackage
\lgsdRpackage\ contains scripts that can be used to reproduce all the
examples and figures encountered in this paper, cf.\
\cref{P2.app:The.scripts.in.lgsdRpackage} in the Supplementary
Material for further details.  This package also allows the user to
interactively investigate the results, which is of interest since
$\hatlgcsdM{k\ell}{\LGp}{\omega}{m}$ can be estimated for a wide
range of tuning parameters.

\subsection{Sanity testing the implemented estimation algorithm}
\label{sec:lgcsd_Some_simulations}

This section will check that the estimates of
$\lgcsdM{k\ell}{\LGp}{\omega}{m}$ behaves as expected for a few simple
simulated bivariate examples.  The approach is similar to the one used
in \JT, as the examples are bivariate extensions of those used in that
paper, but this paper highlights how an initial investigation of
\textit{heatmap and distance}-plots can help identify regions that it
might be of interest to investigate further, cf.\
\cref{fig:heatmap_co_quad_dmt_bivariate_constant_phases,fig:heatmap_co_quad_dmt_bivariate_different_phases,fig:heatmap_distance_plot_EuStockMarkets_DAX_CAC}.
It will be seen that an inspection of the \Co-, \Quad- and
\Phase-plots might be useful exploratory tools for the identification
of nonlinear dependency-structures in multivariate time
series.\footnote{%
  The \Amplitude-plots have not been included here since the
  interesting details (in most cases) already would have been detected
  by the other plots.}

The strategy used to create the plots for the simulated data works as
follows: First draw a given number of independent replicates from the
specified model, and compute $\hatlgcsdM{k\ell}{\LGp}{\omega}{m}$ and
$\hatfz[m]{\!k\ell}(\omega)$ for each of the replicates.  Extract the
relevant \Co-, \Quad- and \Phase-spectra, and use the mean of these
$m$-truncated estimates as estimates of the true (and in general
unknown) $m$-truncated spectra.  Suitable upper and lower percentiles
of the estimates can be used to produce estimates of the pointwise
confidence intervals.

Note that the plots have been annotated with the following
information: A stamp at the center that specifies the type of spectrum
investigated; the numerical convergence status {NC} in the lower left
corner; the truncation level $m$ in the upper left corner; the
percentiles of the point $\LGp$ of investigation, and the bandwidth
$\bm{b}$ in the upper right corner; the length $n$ and the number of
replicates $R$ in the lower right corner.  Later on, for plots based
on resampling from a given sample, the plots will also include the
block length $L$ in the lower right corner.

\subsubsection{Bivariate Gaussian white noise}
\label{sec:lgcsd_Gaussian_white_noise}

For the univariate time series considered in \JT, the sanity testing
of the implemented estimation algorithm started with the trivial
Gaussian case, since it for this case is known that the local Gaussian
auto-spectrum coincides with the ordinary auto-spectrum for Gaussian
time series.  The local Gaussian cross-spectrum also coincide with the
ordinary cross-spectrum for multivariate Gaussian time series, cf.\
\myref{th:lgcs_properties}{th:lgcs_equal_to_osd_when_Gaussian}, and it
is thus natural to test the implemented algorithms on a bivariate
Gaussian example.
 
The explicit details for this bivariate Gaussian test-case are not
that relevant for the overall discussion, and those have thus been
relegated to \cref{sec:lgch-bivariate-Gaussian-white-noise} in the
Supplementary Material.  For the present discussion, it suffices to
say that the resulting plots are as expected, cf.\
\cref{fig:Bivariate_Gaussian_WN}.  In particular; the estimates
$\hatlgcsdM{k\ell}{\LGp}{\omega}{m}$ and $\hatfz[m]{\!k\ell}(\omega)$
of the $m=10$ truncated local Gaussian and ordinary cross-spectra (based
on 100 replicates) are close to the values they are supposed to have,
and the estimated pointwise 90\% confidence intervals for
$\hatlgcsdM{k\ell}{\LGp}{\omega}{m}$ are as expected wider when the
point $\LGp$ lies in the periphery of the observations.

\subsubsection{Bivariate local trigonometric examples}
\label{sec:lgch-some-simulations}

With the exception of the Gaussian case, it is not known what the true
local Gaussian cross-spectrum should look like.  This implies that is
hard to know whether or not the $m$-truncated \Co-, \Quad- and
\Phase-plots looks like they are supposed to do, or if what is seen
might be the result of an erroneous implementation of the estimation
algorithm.

The sanity testing of the implemented estimation algorithm was in
\JT done by the means of an artificially
constructed \textit{local trigonometric} time series, for which it at
least could be reasonably argued what the expected outcome should be
for some specially designated points~$\LGp$ (given a suitable
bandwidth~$\bm{b}$).  This approach will here be extended to the
bivariate case, that is, \textit{bivariate local trigonometric} time
series will be constructed for which it, at some designated
points~$\LGp$, can be given a heuristic argument for the expected
shape of the estimated local Gaussian \Co-, \Quad- and \Phase-spectra.

These artificial time series will not satisfy the requirements needed
for the asymptotic theory to hold true (as is also the case for
standard global spectral analysis), but they can still be used to show
how an exploratory tool based on the local Gaussian spectral density
can detect local structures that the ordinary spectral density fails
to detect.  Details related to the two cases investigated in this
section are given below, whereas \cref{P2.app:fig:trigonometric} in
the Supplementary Material presents an in depth explanation of the
artificial time series construction and the motivation for the
heuristic arguments given in this section.

\textbf{The reference case:} The heuristic argument needed for the
bivariate case is identical in structure to the one used in the
univariate case, and for the present investigation the reference for
the plots later on is based on the following simple bivariate model,
\begin{align}
  \label{eq:bivariate_cosine_example}
  \Yz{1,t}=\cos\parenR{2\pi\alpha t + \phi} + \wz{1,t} \text{ and }
  \Yz{2,t}=\cos\parenR{2\pi\alpha t + \phi + \theta} + \wz{2,t},
\end{align}
where $\wz{i,t}$ is Gaussian white noise with mean zero and standard
deviation~$\sigma$, with $\wz{1,t}$ and $\wz{2,t}$ independent, and
where it in addition is such that~$\alpha$ and~$\theta$ are fixed for
all the replicates whereas $\phi$ is drawn uniformly from
\mbox{$[0,2\pi)$} for each individual replicate.  A realisation with
\mbox{$\sigma = 0.75$}, \mbox{$\alpha=0.302$} and
\mbox{$\mbox{$\theta = \pi/3$}$} has been used for the \Co-, \Quad-,
and \Phase-plots shown in \cref{fig:Bivariate_global_cosine}, where
100 independent samples of length 1859 were used to get the estimates
of the $m$-truncated spectra and their corresponding 90\% pointwise
confidence~intervals (based on the bandwidth
\mbox{$\bm{b}=\parenR{0.6 ,0.6}$}).  Some useful remarks can be based
on this plot, before \textit{the bivariate local trigonometric case}
is defined and investigated.

\begin{figure}[h]

{\centering \includegraphics[width=1\linewidth]{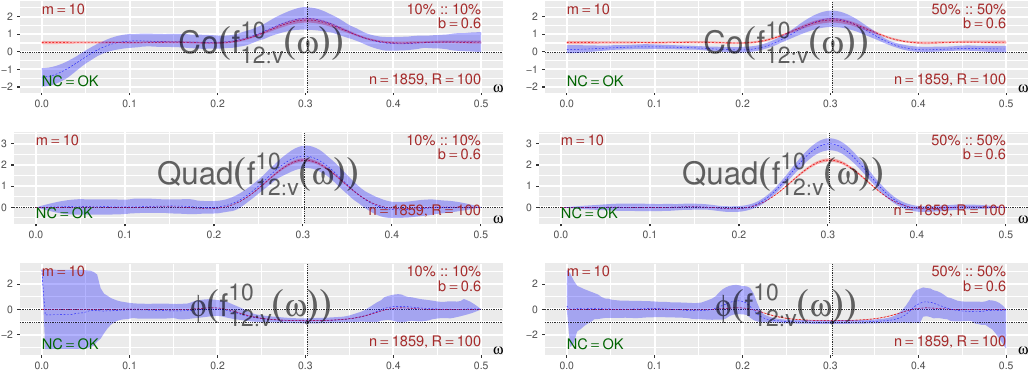} 

}

\caption{Realisation of \cref{eq:bivariate_cosine_example},
  \enquote{bivariate global cosine}, with $\sigma=0.75$,
  $\alpha = 0.302$ and \mbox{$\theta = \pi/3$}.  Similar behaviour in
  lower tail and center, both the global and local Gaussian
  cross-spectra (truncated at $m=10$) detects the frequency $\alpha$
  and the phase-difference $\theta$.}
\label{fig:Bivariate_global_cosine}
\end{figure}

In this particular case, the local Gaussian spectra (in
blue)\footnote{If you have a black and white copy of this paper, then
  read \enquote{blue} as \enquote{light} and \enquote{red} as
  \enquote{dark}} in \cref{fig:Bivariate_global_cosine} shares many
similarities with the corresponding global spectra (in red).\footnote{
  The dotted lines represents the means of the estimated values,
  whereas the 90\% pointwise confidence intervals are based on the 5\%
  and 95\% quantiles of these samples.} %
In particular, the peak of the \Co- and \Quad-plots lies both for the
local and global spectra at the frequency~\mbox{$\omega=\alpha$}
(shown in the plots as a vertical line), and the corresponding
\Phase-plots at this frequency lies quite close to the
phase-adjustment \mbox{$\mbox{$\theta = \pi/3$}$} (shown as a
horizontal line, positioned with an appropriate sign adjustment).
This phenomenon is present both for the point at \texttt{10\%::10\%}
and the point at \texttt{50\%::50\%}, but it should be noted that this
nice match does not hold for all values of~$\sigma$.  In fact,
experiments with different values for~$\sigma$ (plots not included in
this paper) indicates that the difference between the local and global
spectra becomes larger (in particular for the point
\texttt{10\%::10\%}) when~$\sigma$ becomes smaller.  See
\lgsdRef{sec:lgch:beware_of_global_structures} for some further
comments related to this issue.

Note that the wide pointwise confidence band observed for $\omega$
near~0 in the \texttt{10\%::10\%}-\Phase-plot in
\cref{fig:Bivariate_global_cosine}, is related to the branch-cut that
occurs at $-\pi$ in the definition of the phase-spectrum, cf.\
\cref{def:local_co_quad_amplitude_phase}.  The simple algorithm used
for the creation of the pointwise confidence intervals has not been
tweaked to properly cover the case where the majority of the estimates
lies in the second and third quadrants of the complex plane, which
implies that the \Co- and \Quad-plots should be consulted instead when
the \Phase-plot misbehaves in this manner.

The values of the \Co- and \Quad-plots (for a given
frequency~$\omega$) are (for each frequency) related to the
corresponding values of the \Amplitude- and \Phase-plots by the
following simple relations,
\begin{subequations}
  \label{note:lgch:Phase_Co_Quad_connection_equations}
  \begin{align}
    - \text{\Quad-plot}/\text{\Co-plot} 
    &= \tan\!\parenR{\text{\Phase-plot}}\!, \\
    \text{\Co-plot} 
    &= \texttt{\Amplitude-plot} \cdot \cos\!\parenR{\text{\Phase-plot}}\!, \\
    - \text{\Quad-plot} 
    &= \texttt{\Amplitude-plot} \cdot \sin\!\parenR{\text{\Phase-plot}}\!, 
  \end{align}
\end{subequations}
which follows trivially from the way these spectra are defined
relative to Cartesian or polar representations of the complex-valued
cross-spectra, cf.\
\cref{eq:global_co_quadrature_amplitude_phase,def:local_co_quad_amplitude_phase}.
For the example investigated in \cref{fig:Bivariate_global_cosine},
where the \Phase plot is close to $-\pi/3$ at \mbox{$\alpha=0.302$},
it thus follows that the peak for the \Quad-spectrum should be
approximately $\sqrt{3}$ times larger than the peak of the
\Co-spectrum.

The plots encountered later in this paper will be based on the above
mentioned Cartesian or polar representations of
$\hatlgcsdM{k\ell}{\LGp}{\omega}{m}$, but it is in principle also
possible to make plots (and animations) that, for a given value of
$\omega$ shows the estimated values of
$\hatlgcsdM{k\ell}{\LGp}{\omega}{m}$ as point-clouds in the complex
plane.  A discussion related to this complex-valued plot approach is
given in \cref{app:Plots_of_the_complex-valued_spectra} in the
Supplementary Material, see in particular
\cref{fig:Bivariate_global_cosine_lag_frequency}.

\textbf{The bivariate local trigonometric case:} Two bivariate
extensions of the artificial \textit{local trigonometric} time series
from \lgsdRef{sec:Deterministic_trigonometric_models} will now be
considered.  The key idea is that an artificial bivariate time series
$\TSR{\parenR{\Yz{1,t},\Yz{2,t}}}{t\in\ZZ}{}$ can be constructed by
the following scheme:
\begin{enumerate}
\item Select $r\geq 2$ bivariate time series
  $\TSR{\parenR{\Cz{1,i}(t),\Cz{2,i}(t)}}{i=1}{r}$.
\item Select a random variable $I$ with values in the set
  $\parenC{1,\dotsc,r}$, and use this to sample a collection of
  indices $\TSR{\Iz{t}}{t\in\ZZ}{}$ (that is, for each $t$ an
  independent realization of $I$ is taken).  Let
  $\pz{i}\defeq\Prob{\Iz{i}=i}$ denote the probabilities for the
  different outcomes.
\item Define $\Yz{t}$ by
  means of the equation
  \begin{subequations}
    \label{eq:lgch:Y1Y2_local_trigonometric}
    \begin{align}
      \label{eq:lgch:Y1_local_trigonometric}
      \Yz{1,t} &\defeq \sumss{i=1}{r} \Cz{1,i}(t) \cdot\Ind{\Iz{t} = i},\\
      \label{eq:lgch:Y2_local_trigonometric}
      \Yz{2,t} &\defeq \sumss{i=1}{r} \Cz{2,i}(t) \cdot\Ind{\Iz{t} = i}.
    \end{align}
  \end{subequations}
  The indicator function $\Ind{\cdot}$ ensures that only one of the
  bivariate $\parenR{\Cz{1,i}(t),\Cz{2,i}(t)}$-components contributes
  for a given value~$t$, that is, it is also possible to write
  $\parenR{\Yz{1,t},\Yz{2,t}} =
  \parenR{\Cz{1,\Iz{t}}(t),\Cz{2,\Iz{t}}(t)}$.
\end{enumerate}

The \textit{bivariate local trigonometric} time series (needed for the
sanity testing of the implemented estimation algorithm) can now be
constructed by selecting $r$ cosine-functions that oscillate around
different horizontal base-lines $\Lz{i}$, that is,
\begin{subequations}
  \label{eq:lgch:local_cosines}
  \begin{align}
    \label{eq:lgch:local_cosines_first}
    \Cz{1,i}(t)
    &= \Lz{i} + \Az{i}(t) \cdot \cos
      \left(2\pi\alphaz{i} t + \phiz{i} \right),
    &i = 1,\dotsc,r, \\
    \label{eq:lgch:local_cosines_second}
    \Cz{2,i}(t)
    &= \Lz{i} + \Az{i}(t) \cdot \cos
      \left(2\pi\alphaz{i} t + \phiz{i} + \thetaz{i}
      \right),
    &i = 1,\dotsc,r,
  \end{align}
\end{subequations}
where $\alphaz{i}$, $\phiz{i}$ and $\thetaz{i}$ respectively represent
frequency and phase-adjustments occurring in the cosine-function, and
where the amplitudes $\Az{i}(t)$ are uniformly distributed in some
interval $\parenS{\az{i},\bz{i}}$.  Note that it is assumed that the
phases $\phiz{i}$ are uniformly drawn (one time for each realisation)
from the interval between~$0$ and~$2\pi$, whereas the phases
$\thetaz{i}$ are constants.  It is also assumed that the stochastic
processes $\phiz{i}$, $\Az{i}(t)$ and $\Iz{t}$ are independent of
each other.

The auto- and cross-correlations $\rhoz{k\ell}(h)$ of
$\TSR{\parenR{\Yz{1,t},\Yz{2,t}}}{t\in\ZZ}{}$, as given by
\cref{eq:lgch:Y1Y2_local_trigonometric,eq:lgch:local_cosines}, are
functions of $\Lz{i}$ and $\pz{i}$.  The autocorrelations
$\rhoz{11}(h)$ of the marginal time series $\TSR{\Yz{1,t}}{t\in\ZZ}{}$
was computed in \lgsdRef{app:fig:trigonometric}, and the corresponding
result for the cross-correlation $\rhoz{12}(h)$ is given in
\cref{P2.app:fig:trigonometric} in the Supplementary Material.  For
the purpose of the present section, it is sufficient to know that it
is possible to find parameter-configurations for which the global
spectrum (based on the pseudo-normalised observations) is rather flat
(when truncated at $m=10$), which implies that it cannot detect the
frequencies $\alphaz{i}$ of the underlying structure.

Note that neither $\fz{12}(\omega)$ nor $\lgcsd{12}{\LGp}{\omega}$ are
well defined for the \textit{bivariate local trigonometric} times
series, but this is not important since it still is possible to
{predict} (cf.\ \cref{P2.app:fig:trigonometric} for details) that the
$m$-truncated estimates $\hatlgcsdM{12}{\LGp}{\omega}{m}$ for some
points $\LGp$ (and a given bandwidth $\bm{b}$) should resemble
\cref{fig:Bivariate_global_cosine} --- and this can be used, cf.\
\cref{fig:Bivariate_local_trigonometric_A,fig:Bivariate_local_trigonometric_C},
to test the sanity of the implemented estimation algorithm.  These
plots also reveal that the local Gaussian versions of the \Co-, \Quad-
and \Phase-plots can detect local properties that goes undetected by
the ordinary version of these spectra.

\textbf{Parameter setup:} The models for the two time series presented
in this section are based on an extension of the univariate case seen
in \JT, that is, both have \mbox{$r=4$}
components, the probabilities $\pz{i}$ are given by
$(0.05, 1/3-0.05, 1/3, 1/3)$, the frequencies $\alphaz{i}$ are given
by $(0.267, 0.091, 0.431, 0.270)$, the base-lines $\Lz{i}$ are given
by the values $(-2, -1, 0, 1)$, and the lower and upper ranges for the
uniforms sampling of the amplitudes $\Az{i}(t)$ are respectively given
by {$(0.5, 0.2, 0.2, 0.5)$} and {$(1.0, 0.5, 0.3, 0.6)$}.  Note that
$\Lz{i}$ and $\Az{i}(t)$ should be selected in order to give a minimal
amount of overlap between the different components, cf.\
\cref{P2.app:fig:trigonometric} for further details.

The distinction between the two models are due to the selection of the
additional phase-adjustments $\thetaz{i}$.  The model investigated in
\cref{fig:heatmap_co_quad_dmt_bivariate_constant_phases,fig:heatmap_co_quad_dmt_bivariate_different_phases,fig:heatmap_distance_plot_EuStockMarkets_DAX_CAC}
have a constant phase adjustment of \mbox{$\theta = \pi/3$}, whereas
the model investigated in
\cref{fig:heatmap_co_quad_dmt_bivariate_different_phases,fig:Bivariate_local_trigonometric_C}
have individual phase-adjustments given as
\mbox{$\parenR{\thetaz{1}, \thetaz{2}, \thetaz{3},
    \thetaz{4}}=\parenR{\pi/3, \pi/4, 0, \pi/2}$}.

To complete the specification of the setup, note that the
90\% pointwise confidence intervals in
\cref{fig:Bivariate_local_trigonometric_A,fig:Bivariate_local_trigonometric_C}
all are based on 100 independent samples
of length 1859 from the above described
models, and that the bandwidth
\mbox{$\bm{b} = \parenR{0.6,0.6}$}
was used in the computation of the local Gaussian cross-correlations.

\textbf{Constant phase adjustment:} The case where the phase
difference \mbox{$\theta = \pi/3$} was used for all the $\thetaz{i}$
is investigated in
\cref{fig:heatmap_co_quad_dmt_bivariate_constant_phases,fig:Bivariate_local_trigonometric_A}.
This example has been created in order to test the sanity of the
estimation algorithm, and it is (by construction) known in advance
that some points $\LGp$ will be of interest to inspect.  In
particular, the heuristic arguments (cf.\
\cref{P2.app:fig:trigonometric.heuristic.argument} in the
Supplementary Material) imply that the resulting local Gaussian
spectra
should look a bit like those encountered in
\cref{fig:Bivariate_global_cosine} for the three designated points
\texttt{10\%::10\%}, \texttt{50\%::50\%} and \texttt{90\%::90\%}
(which turns out to be the case, cf.\
\cref{fig:Bivariate_local_trigonometric_A}).

For a general case it will be necessary with a strategy that can help
identify interesting regions, and for this purpose an adjusted version
of the combined \textit{heatmap and distance plots} introduced in \JT
will be used, as seen in
\cref{fig:heatmap_co_quad_dmt_bivariate_constant_phases} where the
point $\LGp$ varies along the diagonal.  The heatmap-part of the plot
must be adjusted a bit since the estimated $m$-truncated local
Gaussian cross-spectrum $\hatlgcsdM{k\ell}{\LGp}{\omega}{m}$ is a
complex-valued entity, and the solution seen in
\cref{fig:heatmap_co_quad_dmt_bivariate_constant_phases} shows a
decomposition into the \Co- and \Quad-plots.  It is of course possible
to use a polar decomposition into \Amplitude- and \Phase-plots too,
but it seems to be a bit easier to digest the information from the
\Co- and \Quad-plots, cf.\
\cref{P2.app:Extension_to_the_multivariate_case} in the Supplementary
Material for further details.  The distance-part of these plots can,
as explained in \cref{P2.app:method_for_sensitivity_analysis}, help an
investigator see how far the time series of interest is from being
i.i.d.\ observations~--- and it enables a comparison with the global
spectrum that is not possible only based on the heatmap-plots.

Keep in mind that the pseudo-normalisation step in
\cref{def:lgsd_estimator} implies that the plots related to
$\hatlgcsdM{k\ell}{\LGp}{\omega}{m}$ only reveals information about
the cross-temporal interdependency structures of the sample.  An
investigator must thus use some supplementing technique in order to
extract information from the marginal distributions.

\begin{figure}[h]
  {\centering
    \includegraphics[width=1\linewidth]{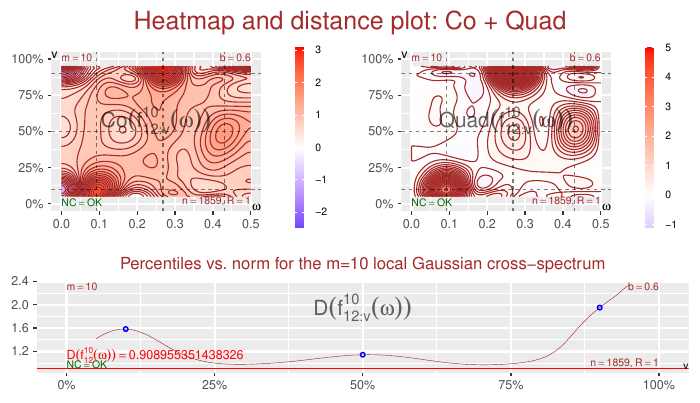}}
  \caption{\Co- and \Quad-heatmaps, distance-plot:  Sample
    from the \textit{bivariate local trigonometric} model in
    \cref{eq:lgch:Y1Y2_local_trigonometric,eq:lgch:local_cosines},
    constant phase-changes $\thetaz{i}=\pi/3, i=1,2,3,4$.  This shows
    how $\hatlgcsdM{k\ell}{\LGp}{\omega}{m}$ varies along the
    diagonal-points $\LGp$. The frequencies $\alphaz{i}$ shown as
    vertical lines.  The points used in
    \cref{fig:Bivariate_local_trigonometric_A} have been indicated
    with lines/points.  }
  \label{fig:heatmap_co_quad_dmt_bivariate_constant_phases}
\end{figure}

The contours in the heatmap plots seen in
\cref{fig:heatmap_co_quad_dmt_bivariate_constant_phases} reveal that
different peaks occur at different combinations of points $\LGp$ and
frequencies $\omega$, in agreement with the heuristic argument that
motivated the construction of this example.  Together with the
distance-plot part, which shows how much the $m$-truncated local
Gaussian cross-spectrum differs from the corresponding $m$-truncated
ordinary cross-spectrum, this shows that it could be of interest to
take a closer look at the three points investigated in
\cref{fig:Bivariate_local_trigonometric_A}.

\begin{figure}[h]
  
  {\centering
    \includegraphics[width=1\linewidth]{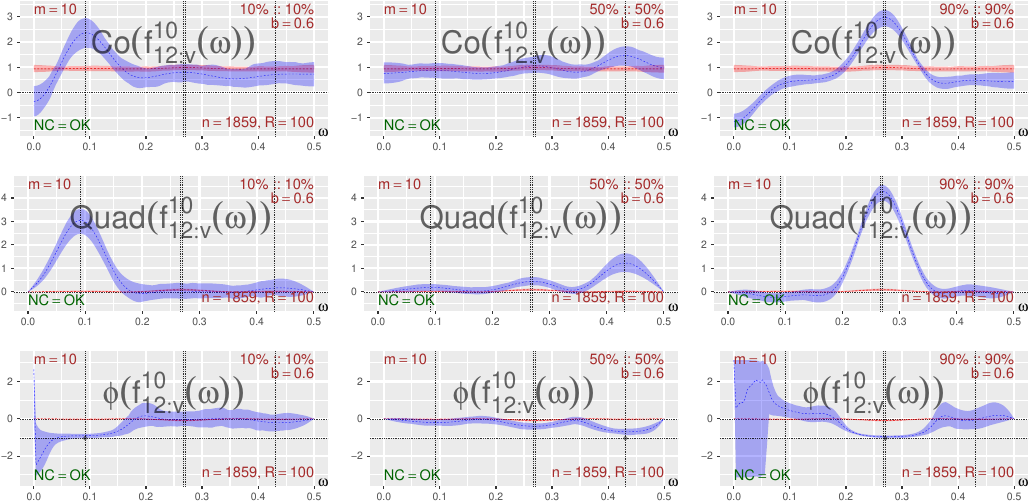}  }
  \caption{\Co-, \Quad- and \Phase-plots for three diagonal points:
    The \textit{bivariate local trigonometric} model in
    \cref{eq:lgch:Y1Y2_local_trigonometric,eq:lgch:local_cosines},
    constant phase-changes $\thetaz{i}=\pi/3, i=1,2,3,4$ shown as
    sign-adjusted horizontal lines.  The frequencies $\alphaz{i}$
    shown as vertical lines.  The local Gaussian spectra
    detect structures that are not detected by the ordinary
    spectrum.}

\label{fig:Bivariate_local_trigonometric_A}
\end{figure}

The three points investigated in
\cref{fig:Bivariate_local_trigonometric_A} correspond to the
function-components in \cref{eq:lgch:local_cosines} with indices
\mbox{$i=2,3,4$}.  The point that corresponds to the \mbox{$i=1$}
component would, due to the combination of the low probability
$\pz{1}$ and the placement of the level $\Lz{1}$, lie too far out in
the tail to be properly investigated by the present sample size.  Note
that the \mbox{$i=1$} component has been included here in order to
show that an exploratory approach based on local Gaussian spectra can
fail to detect local signals that are much weaker than the dominating
ones, cf.\ the discussion in \cref{P2.app:fig:trigonometric}.

For the points investigated in
\cref{fig:Bivariate_local_trigonometric_A}, it seems to be the case
that the local Gaussian part of the \Co-, \Quad- and \Phase-plots
together reveal local properties in accordance with the outcome
expected from the knowledge of the generating model~--- and these
local structures are not detected by the ordinary global spectra,
which in this case (due to the values used for $\Lz{i}$ and $\pz{i}$)
are quite close to being flat (i.e., information about the specified
frequencies can not be extracted from the global spectrum, cf.\
\cref{P2.app:fig:trigonometric.general.properties}).  The left column
investigates a point at the lower tail of the diagonal, and it can
there be observed that both the \Co- and \Quad-plots have a peak close
to the leftmost $\alpha$-value~--- and the value of the corresponding
\Phase-plot for frequencies close to this $\alpha$-value lies quite
close to the phase difference between the first and second component.
A similar situation is present for the three plots shown in the right
column, where a point at the upper tail of the diagonal are
investigated.  Moreover, in accordance with the general discussion
related to \cref{note:lgch:Phase_Co_Quad_connection_equations},
the peaks of the \Quad-plots are higher than those of the \Co-plots in
this case due to the phase-difference $\theta$ that was used in the
input~parameters.

For the center column of \cref{fig:Bivariate_local_trigonometric_A},
which investigates the point at the center of the diagonal, it can be
seen that the \Quad- and \Phase-plots in addition to the expected
$\alpha$-value also detects the presence of the other $\alpha$-values.
The \Phase-plot is for this point not that close to the expected
value, but that situation changes if the truncation is performed at a
higher lag than \mbox{$m=10$}.  The center column thus shows the
importance of considering a range of values for the truncation point
when such plots are investigated.  The additional peaks that are
detected in the center column are due to \textit{contamination} from
the neighbouring regions.

\textbf{Individual phase adjustment:} In this example, the samples
have been generated with the following individual phase-adjustments:
\mbox{$\parenR{\thetaz{1}, \thetaz{2}, \thetaz{3},
    \thetaz{4}}=\parenR{\pi/3, \pi/4, 0, \pi/2}$}.  The
\textit{heatmap and distance}-plot seen in
\cref{fig:heatmap_co_quad_dmt_bivariate_different_phases} reveals once
more, in agreement with the way the example has been constructed, that
it is natural to look at the three points shown in
\cref{fig:Bivariate_local_trigonometric_C}.  The horizontal lines seen
in the \Phase-plots part of \cref{fig:Bivariate_local_trigonometric_C}
show the $\thetaz{i}$-values (adjusted to have the correct sign), and
for each of the designated points the intersection with the relevant
vertical $\alphaz{i}$-line has been highlighted to show the expected
outcome based on the knowledge of the model.

\begin{figure}[h]
  {\centering
    \includegraphics[width=1\linewidth]{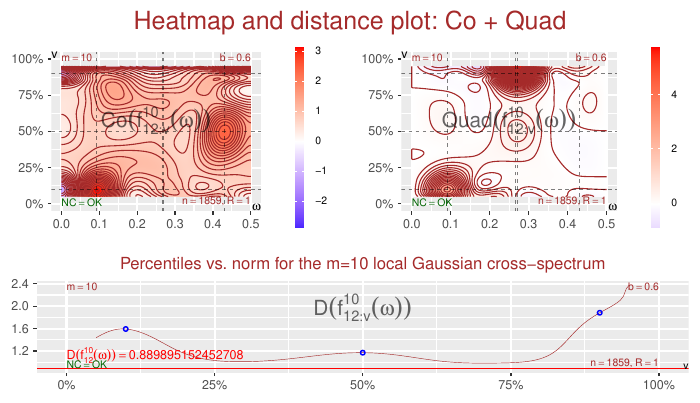}}
  \caption{\Co- and \Quad-heatmaps, distance-plot:  Sample from
    the \textit{bivariate local trigonometric} model in
    \cref{eq:lgch:Y1Y2_local_trigonometric,eq:lgch:local_cosines},
    phase-changes
    $\parenR{\thetaz{1}, \thetaz{2}, \thetaz{3},
      \thetaz{4}}=\parenR{\pi/3, \pi/4, 0, \pi/2}$.  This shows how
    $\hatlgcsdM{k\ell}{\LGp}{\omega}{m}$ varies along the
    diagonal-points $\LGp$. The frequencies $\alphaz{i}$ shown as
    vertical lines.  The points used in
    \cref{fig:Bivariate_local_trigonometric_C} have been indicated
    with lines/points.  }
  \label{fig:heatmap_co_quad_dmt_bivariate_different_phases}
\end{figure}

The \Co-, \Quad- and \Phase-plots in
\cref{fig:Bivariate_local_trigonometric_C} behave in accordance with
what was observed in \cref{fig:Bivariate_local_trigonometric_A}, i.e.,
the \Phase-plots lies close to the expected $\thetaz{i}$-value when
the frequency $\omega$ is near the corresponding $\alphaz{i}$-value,
and the height of the corresponding \Co- and \Quad-peaks are in
accordance with the values of the \Phase-plots.  In particular, the
phase-adjustment is \mbox{$\thetaz{2} = \pi/4$} for the point
\texttt{10\%::10\%}, which implies that the \Co- and \Quad-peaks
should rise approximately to the same height above their respective
baselines, which seems to be fairly close to the observed result.  For
the points \texttt{50\%::50\%} and \texttt{90\%::90\%} the situation
is clearer since the respective local frequencies
\mbox{$\thetaz{3} = 0$} and \mbox{$\thetaz{4} = \pi/2$} then implies
that only the \Co-plot should have a peak for the point
\texttt{50\%::50\%} and only the \Quad-plot should have a peak for the
point \texttt{90\%::90\%}, again in agreement with the impression
based on~\cref{fig:Bivariate_local_trigonometric_C}.  The global
spectrum in \cref{fig:Bivariate_local_trigonometric_C} (red line) does
not reveal any of this local spectral structure.

\begin{figure}[h]

{\centering \includegraphics[width=1\linewidth]{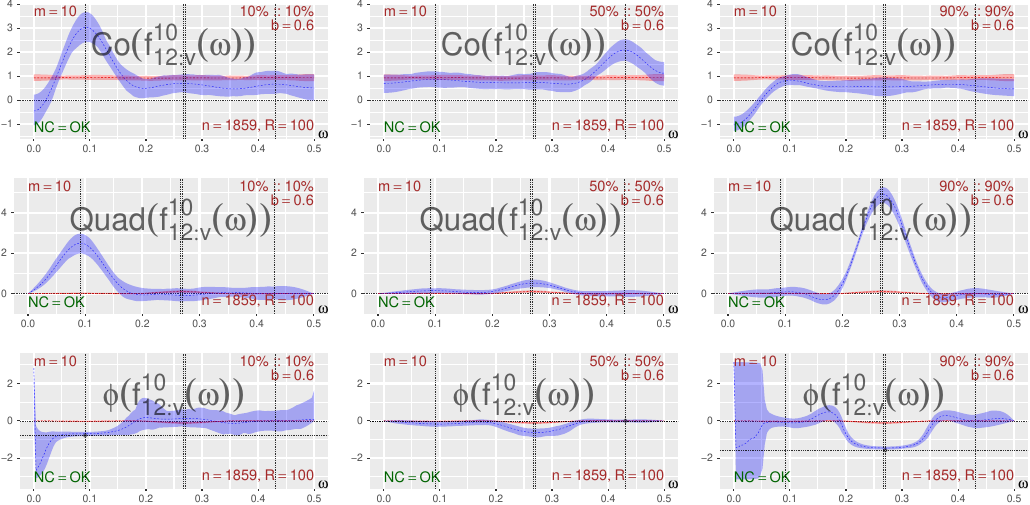} 

}

  \caption{\Co-, \Quad- and \Phase-plots for three diagonal points:
    The \textit{bivariate local trigonometric} model in
    \cref{eq:lgch:Y1Y2_local_trigonometric,eq:lgch:local_cosines},
    phase-changes
    $\parenR{\thetaz{1}, \thetaz{2}, \thetaz{3},
      \thetaz{4}}=\parenR{\pi/3, \pi/4, 0, \pi/2}$, shown as
    sign-adjusted horizontal lines.  The frequencies $\alphaz{i}$
    shown as vertical lines.  The local Gaussian spectra
    detect structures that are not detected by the ordinary
    spectrum.}

\label{fig:Bivariate_local_trigonometric_C}
\end{figure}

The examples investigated in
\cref{fig:Bivariate_global_cosine,fig:heatmap_co_quad_dmt_bivariate_constant_phases,fig:Bivariate_local_trigonometric_A,fig:heatmap_co_quad_dmt_bivariate_different_phases,fig:Bivariate_local_trigonometric_C}
do not satisfy the requirements needed for the asymptotic results
(both for the global and local cases) to hold true, in particular the
local Gaussian cross-correlations will in these cases not be
absolutely summable.\footnote{ In this respect the situation is
  similar to the detection of a pure sinusoidal for the global
  spectrum.}  Despite this, the examples are still of interest since
they show that an exploratory tool based on the local Gaussian spectra
in this case (with a truncation at $m=10$) can detect information that
is in agreement with what could be expected based on the parameters
used in the models for $\parenR{\Yz{1,t},\Yz{2,t}}$.

The underlying model will of course not be known when a real
multivariate time series is encountered, so it is important to
estimate the local Gaussian cross-spectrum at a wide range of points
$\LGp$ in the plane and a wide range of truncation levels~$m$.  This
kind of investigation can be done by the \textit{heatmap and
  distance}-plots seen in
\cref{fig:heatmap_co_quad_dmt_bivariate_constant_phases,fig:heatmap_co_quad_dmt_bivariate_different_phases,fig:heatmap_distance_plot_EuStockMarkets_DAX_CAC}.

Even when it might not be obvious how to interpret the results shown
in the \Co-, \Quad- and \Phase-plots, it should be noted that they can
be used as an exploratory tool that can detect nonlinear traits in the
observations.  Moreover, these plots can also be used to investigate
if a model fitted to the data contains elements that can mimic the
observed features.  The recipe for this approach would then be to
first select a model, then estimate parameters based on the available
sample, and finally use the resulting fitted model to generate
independent samples of the same length as the sample.
\Cref{sec:lgch:EuStockMarkets+cGARCH} will show an example of this
approach, cf.\ \cref{fig:EuStockMarkets,fig:cGARCH}.  Further details
related to this are given in
\cref{P2.app:resampling_parametric_bootstrap}.

\subsection{A real multivariate time series and a poorly fitted
  GARCH-type model}
\label{sec:lgch:EuStockMarkets+cGARCH}

This section will show how the \Co-, \Quad- and \Phase-plots can be
used as an exploratory tool on a financial data set, and then it will
be seen how this approach also can be used to get a visual impression
of the quality of a multivariate GARCH-type model fitted to
these~data.  

The multivariate time series sample to be considered in this section
will be a bivariate subset of the (log-returns of the) tetravariate
\EuStockMarkets-sample from the \texttt{datasets}-package of \Rref{R},
\citet{R_manual}.  This data-set has been selected since it has a
length that should be large enough to justify the assumption that the
observed features in the \Co-, \Quad- and \Phase-plots are not solely
there due to small sample variation.

The \EuStockMarkets contains 1860 daily closing prices collected in
the period 1991-1998, from the following four major European stock
indices: Germany DAX (Ibis), Switzerland SMI, France CAC, and UK FTSE.
The data was sampled in business time, i.e., weekends and holidays
were omitted.

The log-returns of the \EuStockMarkets values gives a tetravariate
data-set that it seems natural to model with some multivariate
GARCH-type model, and the \Rpackage \Rref{rmgarch},
\citet{rmgarch_Ghalanos}, was used for that purpose.  Note that the
present paper only aims at showing how this kind of analysis can be
performed, so only one very simple model was investigated --- which
thus gave a rather poor model for the data at hand.

The local Gaussian approach presented here is based on a
pseudo-normalisation of the marginal distributions.  It is thus in
essence properties of the copula-structures of the time series that
are revealed here, and a practitioner should supplement this approach
by methods that also extracts information from the original marginals.

\subsubsection{The DAX-CAC subset of the \EuStockMarkets-log-returns}
\label{sec:lgch:EuStockMarkets}

The log-returns of the bivariate \EuStockMarkets-subset
\mbox{$\parenR{\Yz{1},\Yz{3}}=\parenR{\text{DAX},\text{CAC}}$}, of
length 1859, will now be investigated.  The
individual pseudo-normalised traces of these observations are shown in
\cref{fig:EuStockMarkets_YiYj}, and it will be from these
pseudo-normalised observations that the local Gaussian
cross-correlations will be computed.

\begin{figure}[h]

{\centering \includegraphics[width=1\linewidth]{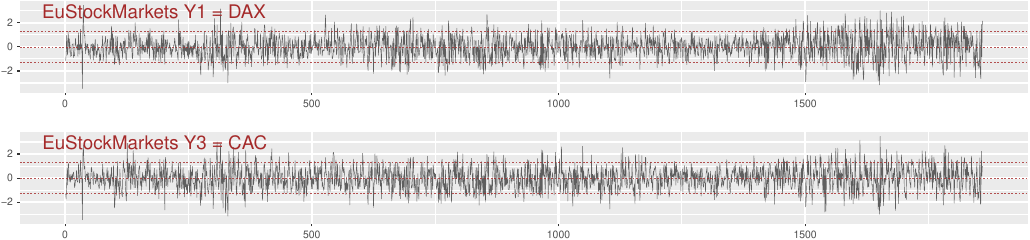} 

}

\caption{The pseudo-normalised log-reurns of the DAX- and
  CAC-components from the \EuStockMarkets-data, which will be
  investigated further in
  \cref{fig:heatmap_distance_plot_EuStockMarkets_DAX_CAC,fig:EuStockMarkets_lags200,fig:EuStockMarkets}. }

\label{fig:EuStockMarkets_YiYj}
\end{figure}

A local Gaussian investigation requires that some points $\LGp$ must
be selected, and thereafter an in depth investigation can be performed
for the selected points.  The \textit{heatmap} and \textit{distance}
plots can be useful for the task of identifying interesting points
along the diagonal, as seen in the
\cref{fig:heatmap_distance_plot_EuStockMarkets_DAX_CAC} where the
$m=10$ truncated \Co- and \Quad-spectra are presented for the DAX- and
CAC-components of the log-returns of the \EuStockMarkets-data.

\begin{figure}[h]

{\centering \includegraphics[width=1\linewidth]{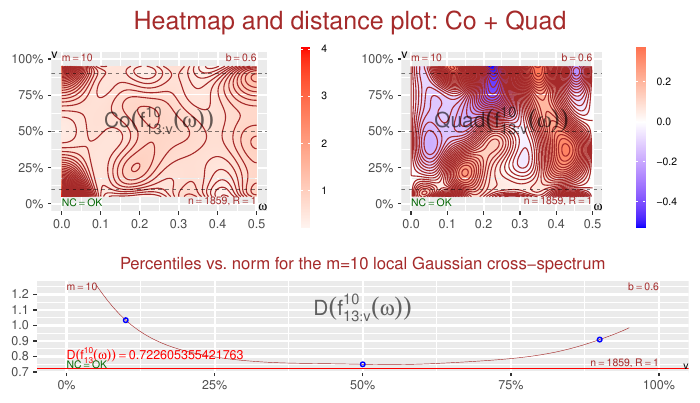} 

}
\caption{\Co- and \Quad-heatmaps, distance-plot: Based on
  \mbox{$\parenR{\Yz{1,t+h},\Yz{3,t}}$} with $\Yz{1,t+h}$ and
  $\Yz{3,t}$ being respectively the DAX- and CAC-components shown in
  \cref{fig:EuStockMarkets_YiYj}. 
  Asymmetric situation, with higher values in the lower tail.  The
  diagonal points used in
  \cref{fig:EuStockMarkets_lags200,fig:EuStockMarkets} have been
  indicated with lines/points.  }

\label{fig:heatmap_distance_plot_EuStockMarkets_DAX_CAC}
\end{figure}

The scale-indicators for the heatmap-plots in
\cref{fig:heatmap_distance_plot_EuStockMarkets_DAX_CAC} reveal that it
is the \Co-spectrum that dominates, and it is clear that the situation
in the tails are different from the one in the center.  The asymmetry
between the upper and lower tail is not immediate to see from the
heatmap-part of
\cref{fig:heatmap_distance_plot_EuStockMarkets_DAX_CAC}, but it is
easy to see from the distance-part of the plot that the local Gaussian
interdependency structure is stronger in the lower tail.

From the information in
\cref{fig:heatmap_distance_plot_EuStockMarkets_DAX_CAC}, it can now be
seen that it could be of interest to
compare the behaviour in the tails with that at the center.  The
selected points $\LGp$ have been highlighted in
\cref{fig:heatmap_distance_plot_EuStockMarkets_DAX_CAC} by the help of
added lines (the heatmap-part) and circles (the distance-part).  These
points $\LGp$ are the same as those used for
\cref{fig:Bivariate_local_trigonometric_A,fig:Bivariate_local_trigonometric_C},
i.e., the diagonal points $\LGp$ whose coordinates correspond to the
10\%, 50\% and 90\% percentiles of the standard normal distribution.

The estimates of the local Gaussian cross-correlations might
degenerate toward $+1$ or $-1$ if the points $\LGp$ are too far out in
the tails.  It is thus of interest to check the behaviour of these
estimates for the given sample before the computationally costly
production of pointwise confidence intervals (based on resampling) is
undertaken.  The result of this investigation of the estimated local
Gaussian cross-correlations $\lgccr{k\ell}{\LGp}{h}$ is seen in
\cref{fig:EuStockMarkets_lags200}, where a wide range of values for
$h$ has been included.  The observed values of
$\hatlgccr{k\ell}{\LGp}{h}$ indicate that degeneration of the
estimates are not a problem for these points.

\begin{figure}[h]

{\centering \includegraphics[width=1\linewidth]{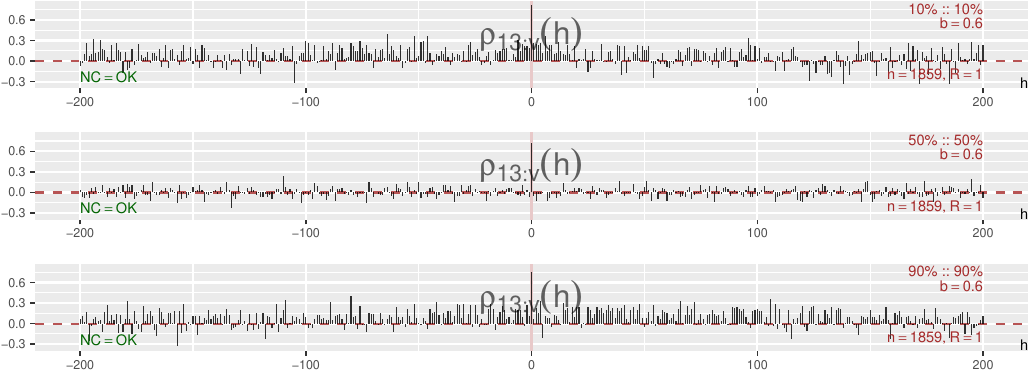} 

}

\caption{Local Gaussian cross-correlations $\hatlgccr{k\ell}{\LGp}{h}$
  for three diagonal points: Between
  \mbox{$\parenR{\Yz{1,t+h},\Yz{3,t}}$}, where $\Yz{1,t+h}$ and
  $\Yz{3,t}$ are the DAX- and CAC-components shown in
  \cref{fig:EuStockMarkets_YiYj}.  }

\label{fig:EuStockMarkets_lags200}
\end{figure}

It seems to be the case that the local Gaussian spectrum can detect
nonlinear structures even with a rather low truncation level $m$, but
the value $m=10$ used in this paper is solely selected in order to
present a proof-of-principle.  For an actual investigation it would be
natural to investigate a wider range of different lags $h$, like seen
in \cref{fig:EuStockMarkets_lags200}, in order to check what kind of
behaviour that is observed.  The interested reader will find a
sensitivity analysis of $m$ in
\cref{P2.app:Truncation_level_sensitivity}, which is based on the
values seen in \cref{fig:EuStockMarkets_lags200}.

\Cref{fig:EuStockMarkets} shows the \Co-, \Quad-, and \Phase-spectra
obtained from the $m=10$ truncated global and local spectra for the
three selected diagonal points.  A solid red line represents the
estimate of the ordinary cross-spectrum $\fz[m]{k\ell}(\omega)$,
whereas a solid blue line represent the estimate of the local Gaussian
cross-spectrum $\lgcsdM{k\ell}{\LGp}{\omega}{m}$.  The 90\% pointwise
confidence intervals have been created based on $R=100$ block-bootstrap
replicates using a block length of $L=25$ for the \textit{circular
  index-based block bootstrap for tuples} resampling strategy
developed in \JT.  Details related to this resampling strategy are
given in \cref{P2.app:regarding_resampling} in the Supplementary
Material, and the sensitivity analysis in
\cref{P2.app:Block_length_sensitivity} reveals that the results in
this case
are stable over a wide range of block lengths.

\begin{figure}[h]

{\centering \includegraphics[width=1\linewidth]{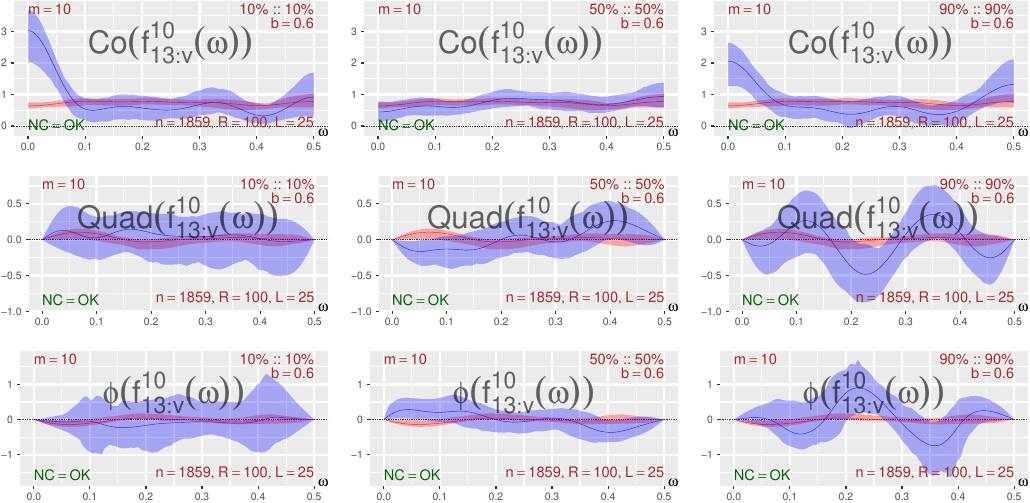} 

}

\caption{\Co-, \Quad- and \Phase-plots for three diagonal points:
  Between \mbox{$\parenR{\Yz{1,t+h},\Yz{3,t}}$}, where $\Yz{1,t+h}$
  and $\Yz{3,t}$ are the DAX- and CAC-components shown in
  \cref{fig:EuStockMarkets_YiYj}.  The local Gaussian spectra detect
  structures that are not detected by the ordinary spectrum.}

\label{fig:EuStockMarkets}

\end{figure}

The three points considered in \cref{fig:EuStockMarkets} all lie on
the diagonal, since it seems easier to give an interpretation for
those points.  In particular, the point \texttt{10\%::10\%} represent
a situation where the market goes down both in Germany and France,
whereas the points \texttt{50\%::50\%} and \texttt{90\%::90\%}
similarly represent cases where the market either is stable or goes~up
in both countries.
For the purpose of the present paper, it suffices to point out that
the \Co-, \Quad- and \Phase-plots of \cref{fig:EuStockMarkets}
indicates that the data contains nonlinear traits, which in particular
is visible in the \Co-plot for the point \texttt{10\%::10\%} and for
all the plots related to the point \texttt{90\%::90\%}.  It should be
noted that the \Co-plots at the frequency \mbox{$\omega=0$} simply
gives a weighted sum of the local Gaussian cross-correlations (between
\mbox{$\parenR{\Yz{1,t+h},\Yz{3,t}}$}) seen in
\cref{fig:EuStockMarkets_lags200}, so the \Co-plot peaks at
\mbox{$\omega=0$} for the points \texttt{10\%::10\%} and
\texttt{90\%::90\%} are thus as expected, and the lack of a \Co-plot
peak at \mbox{$\omega=0$} for the point \texttt{50\%::50\%} also seems
natural in view of \cref{fig:EuStockMarkets_lags200}.  It should also
be noted that the \mbox{$\omega=0$} peak for the \texttt{90\%::90\%}
\Co-plot is lower than the corresponding peak for \texttt{10\%::10\%},
but this seems in this case to be due to the low truncation level used
for the plots, i.e., these two peaks attain approximately the same
height when a higher truncation level is applied.

It seems, for the particular parameter-configuration that generated
the plots in \cref{fig:EuStockMarkets}, to be the case that the point
\texttt{90\%::90\%} has the most interesting \Quad- and \Phase-plots,
but again, as noted above, it may be premature to put too much
emphasis on this particular plot given the uncertainties involved in
the selection of the bandwidth~$\bm{b}$ and the truncation level~$m$.
As mentioned before, the effect of changes to the block length $L$,
are quite minimal, cf.\ the discussion in
\cref{P2.app:Block_length_sensitivity} in the Supplementary Material.

It should also be noted that a low number of bootstrapped replicates
can be a source of small sample variation for the width of the
estimated pointwise confidence intervals, and this is important to
keep in mind if a minor gap is observed between the pointwise
confidence intervals for the local and global spectra.  Such gaps
could appear or disappear when the algorithm is used to generate new
computations based on the same number of bootstrapped replicates,
a behaviour that in particular has been observed for the rightmost
peak/trough of the \Quad- and \Phase-plots at the point
\texttt{90\%::90\%} in \cref{fig:EuStockMarkets}. %
This kind of ambiguity can be countered by increasing the number of
bootstrapped replicates, but that has not been done for the present
example due to the increased computational cost.

The peak at the center of the \texttt{90\%::90\%} \Phase-plot in
\cref{fig:EuStockMarkets} seems to be significantly different from the
global spectrum, which strengthens the impression that something of
interest might be present at that frequency.  However, it is important
to keep in mind that this impression is based on the present
combination of bandwidth and truncation level~--- and there are at the
moment no data-driven method for the selection of these parameters.
(A positive phase difference is consistent with the
\texttt{90\%::90\%} cross-correlation plot in
\cref{fig:EuStockMarkets_lags200}, which might indicate that the DAX
is leading over CAC when the market is going~up.)

Note that the \Rref{shiny}-interface in the \Rpackage \lgsdRpackage\
should be used if it is of interest to pursue a further analysis of
the local Gaussian spectra of the log-returns of the
\EuStockMarkets-data, since that enables an interactive investigation
that shows how the estimates vary based on different
bandwidths~$\bm{b}$ and truncation levels~$m$.  Moreover, as discussed
in \cref{P2.app:Block_length_sensitivity}, it is also possible to
check how the selection of the block length $L$ for the
bootstrap-procedure developed in \JT influences
the widths of the pointwise confidence intervals in the \Co-, \Quad-
and \Phase-plots.

\subsubsection{A simple copula GARCH-model fitted to the
  \EuStockMarkets log-returns}
\label{sec:lgch:cGARCH}

It might not be obvious how to interpret the \Co-, \Quad- and \Phase
spectra based on the log-returns of the \EuStockMarkets-data, but they
do at least provide an approach where nonlinear dependencies might be
detected from a visual inspection of the plots.

Furthermore, it is possible to use this as an exploratory tool in
order to investigate whether a model fitted to the original data is
capable of reproducing nonlinear traits that match those observed for
the data.  The procedure is straightforward: 

\begin{enumerate}
\item   \label{algorithm:lgcd:fit_model}
  Fit the selected model to the data.
\item   \label{algorithm:lgcd:select_parameters} 
  Perform a local Gaussian spectrum investigation based on simulated
  samples from the fitted model.  The parameters should match those
  used in the investigation of the original data.
\item   Compare the plots based on the original data with corresponding
  \label{algorithm:lgcd:show_plots} 
  plots based on the simulated data from the model.  It can be of
  interest to not only compare the \Co-, \Quad- and \Phase-plots, but
  also include plots that show the traces and the estimated local
  Gaussian auto- and cross-spectra.
\end{enumerate}

For the present case of interest,
\cref{algorithm:lgcd:select_parameters} of the list above implies that
100 independent samples of length 1859 will be used as the basis for
the construction of the \Co-, \Quad- and \Phase-plots of the fitted
model, and the bandwidth \mbox{$\bm{b} = \parenR{0.6,0.6}$} will be
used for the estimation of the local Gaussian cross-correlations at
the three diagonal points \texttt{10\%::10\%}, \texttt{50\%::50\%}
and~\texttt{90\%::90\%}.

\textbf{The model:} The \Rpackage \Rref{rmgarch} was used to fit a
simple multivariate GARCH-type model to the log-returns of the
\EuStockMarkets-data, in order to exemplify the procedure outlined
above, i.e., a copula GARCH-model (cGARCH) with the simplest available
univariate models for the marginals\footnote{ See
  \citet{ghalanos15:_rmgarch} for details about the cGARCH-model and
  other options available in the \Rref{rmgarch}-package.} was fitted
to the data, and the resulting model was then used to produce
\cref{fig:cGARCH_YiYj,fig:cGARCH_lags20,fig:cGARCH}.  See
\cref{P2.app:fig:GARCH} for further details.  Note that this simple
model was selected simply in order to provide a proof-of-principle
example for the investigations encountered later on.  It was in
particular of interest, as discussed in
\cref{P2.app:resampling_parametric_bootstrap}, to use a \enquote{too
  simple model} in order to highlight the ideas related to the local
Gaussian sanity testing of parametric models fitted to a sample.

\textbf{The traces:} \Cref{fig:cGARCH_YiYj} shows the pseudo-normalised trace of the
$\Yz{1}$- and $\Yz{3}$-variables for one sample from the tetravariate
cGARCH-model, and this can be compared with the corresponding
pseudo-normalised trace of the DAX and CAC plot for the
pseudo-normalised log-returns of the \EuStockMarkets-data, see
\cref{fig:EuStockMarkets_YiYj}.  Obviously, a comparison of
one single simulated trace with the trace of the original data might
not reveal much, and it should also be noted that it in general might
be preferable to compare the traces before they are subjected to the
pseudo-normalisation, since that could detect if the model might fail
to produce sufficiently extreme outliers.

\begin{figure}[h]

{\centering \includegraphics[width=1\linewidth]{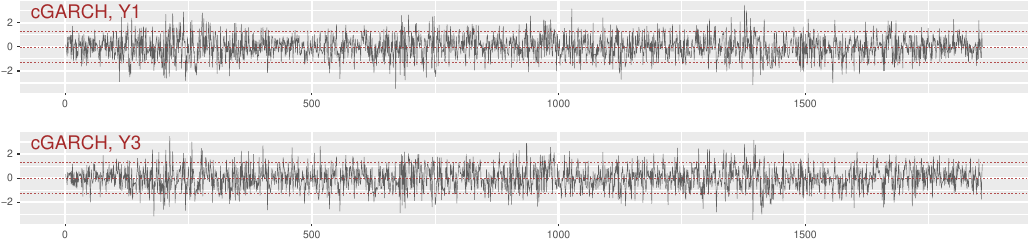} 

}

\caption{The pseudo-normalised pairs
  \mbox{$\parenR{\Yz{1,t+h},\Yz{3,t}}$} from a sample from the cGARCH
  model fitted to the \EuStockMarkets-data, where $\Yz{1,t+h}$ and
  $\Yz{3,t}$ correspond to the DAX- and CAC-components shown in
  \cref{fig:EuStockMarkets_YiYj}.  }

\label{fig:cGARCH_YiYj}
\end{figure}

\textbf{The local Gaussian correlations:} %
Box-plots, based on the 100 independent estimates of the local
Gaussian cross-correlations from the cGARCH-model, are shown in
\cref{fig:cGARCH_lags20}.  These can be compared with the local
Gaussian cross-correlations estimated from the original sample, shown
in \cref{fig:EuStockMarkets_lags200}.  It should be noted that the
computational cost for the production of the box-plots in
\cref{fig:cGARCH_lags20} is substantially larger than the cost for the
production of the simpler plots shown in
\cref{fig:EuStockMarkets_lags200}, so it is preferable to restrict the
attention to a shorter range of lags in \cref{fig:cGARCH_lags20}.
Note also that the wide range of lags included in
\cref{fig:EuStockMarkets_lags200} is related to the desire for an
initial investigation of the sensitivity of the truncation length $m$,
and it is as seen here possible to judge the suitability of the fitted
model from the shorter range of lags included
in~\cref{fig:cGARCH_lags20}.

\begin{figure}[h]

{\centering \includegraphics[width=1\linewidth]{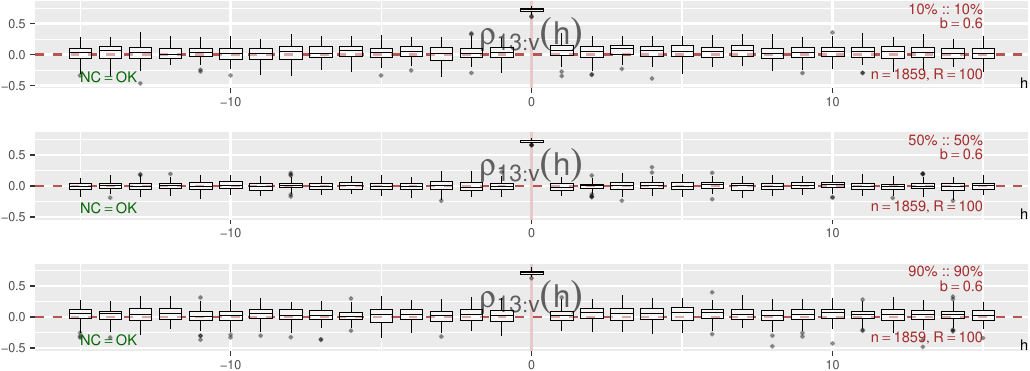} 

}

\caption{Local Gaussian cross-correlations $\hatlgccr{k\ell}{\LGp}{h}$
  for three diagonal points, based on samples of the pseudo-normalised
  pairs \mbox{$\parenR{\Yz{1,t+h},\Yz{3,t}}$} from the cGARCH model
  fitted to the \EuStockMarkets-data, where $\Yz{1,t+h}$ and
  $\Yz{3,t}$ correspond to the DAX- and CAC-components shown in
  \cref{fig:EuStockMarkets_YiYj}.  These values are used for the
  construction of \cref{fig:cGARCH}. }

\label{fig:cGARCH_lags20}
\end{figure}

The impression from the lags included in \cref{fig:cGARCH_lags20} is
that the medians of the estimated local Gaussian cross-correlations
for the point \texttt{50\%::50\%} are quite close to zero, whereas the
medians for the points \texttt{10\%::10\%} and \texttt{90\%::90\%}
mostly are slightly above zero.  Almost none of the boxes for the two
latter points appear to be positioned in a manner consistent with the
desired outcome for a good match with the corresponding estimated
values in \cref{fig:EuStockMarkets_lags200}, and it might thus be
ample reason to suspect that this cGARCH-model might better be
replaced with another model instead.

\textbf{The \Co- \Quad- and \Phase-plots:} \Cref{fig:cGARCH} shows the
local Gaussian spectra for the same points~$\LGp$ and the same
configuration of parameters as those used in
\cref{fig:EuStockMarkets}.

\begin{figure}[h]

{\centering \includegraphics[width=1\linewidth]{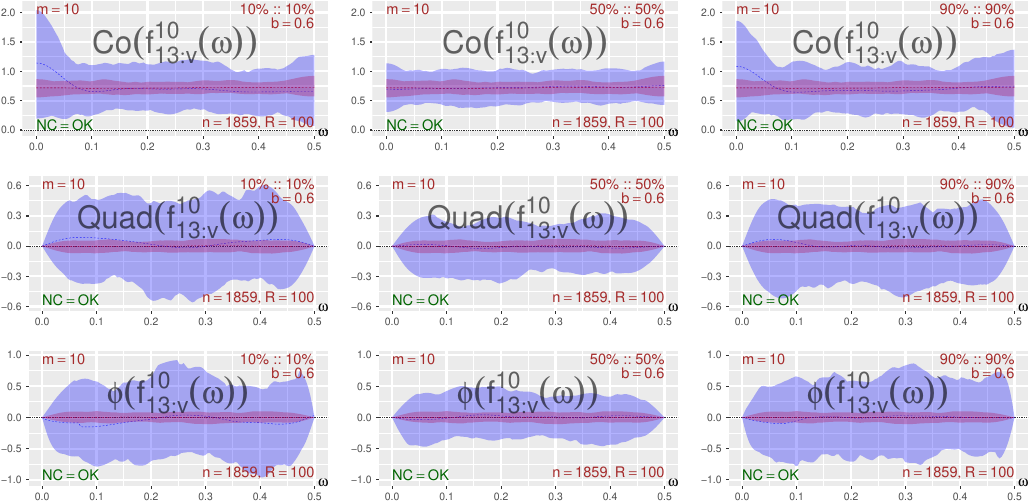} 

}

\caption{\Co-, \Quad- and \Phase-plots for three diagonal points,
  based on samples of the pseudo-normalised pairs
  \mbox{$\parenR{\Yz{1,t+h},\Yz{3,t}}$} from the cGARCH model fitted
  to the \EuStockMarkets-data, where $\Yz{1,t+h}$ and $\Yz{3,t}$
  correspond to the DAX- and CAC-components shown in
  \cref{fig:EuStockMarkets_YiYj}.  A comparison with
  \cref{fig:EuStockMarkets} reveals that this model failed to pick up
  some of the local interdependency-structures seen in the
  \EuStockMarkets-data.}

\label{fig:cGARCH}
\end{figure}

The \Co-, \Quad- and \Phase-plots of
\cref{fig:cGARCH} look like they could originate from i.i.d.\ white
noise~--- which comes as no surprise in view of the information about
the local Gaussian cross-correlations in \cref{fig:cGARCH_lags20}.
The \Co-plots for the two points \texttt{10\%::10\%} and
\texttt{90\%::90\%} show that the estimates of the $m$-truncated local
Gaussian cospectra, i.e., the blue dashed lines, might have a peak at
\mbox{$\omega=0$}~--- but the 90\% pointwise confidence intervals are
too wide to support a claim that these peaks are significant.  A
further comparison of these two \Co-plots with the corresponding
\Co-plots in \cref{fig:EuStockMarkets} (beware of different scales for
the axes), shows that the confidence intervals from \cref{fig:cGARCH}
are too narrow (at \mbox{$\omega=0$}) to encompass the peaks observed
in \cref{fig:EuStockMarkets}~--- which indicates that the selected
model might be a rather bad approximation to the log-returns of the
\EuStockMarkets-data.  It thus seems advisable to look for some other
model instead, a natural conclusion given that no effort whatsoever
was made with regard to finding reasonable marginal distributions for
the copula GARCH-model used in this discussion.

It should be noted that a \textit{local Gaussian spectra comparison}
of the original data and the fitted model in practise also should
include a comparison of the local Gaussian auto-spectra of the
marginals, as was done in \JT.  These
auto-spectra plots (not included in this paper, but see
\cref{P2.app:Point_sensitivity.Revisiting.the.univariate.case} for
some related \textit{heatmap and distance}-plots) can provide some
additional information useful for the model-selection process.  In
particular, if a model-selection algorithm for GARCH-type models has
been used to pick one marginal model from a given collection of
marginal models, then an investigation based on the \textit{local
  Gaussian auto-spectrum} might reveal if the selected marginal model
captures the local traits of the corresponding marginal observations
in a reasonable manner.

The preceding discussion was restricted to three diagonal points, but
the local Gaussian sanity testing of a fitted parametric model can
also be done for points outside of the diagonal.  Such an approach was
discussed for the univariate case in
\lgsdRef{app:resampling_parametric_bootstrap}, and a natural extension
to the multivariate case is presented in
\cref{P2.app:resampling_parametric_bootstrap} in the online
Supplementary Material.

\section{Conclusion}
\label{sec:lgch_conclusion}

The \textit{local Gaussian auto-spectrum} from \JT can, as seen in
this paper, be extended to cover the multivariate case too --- and
estimates of the $m$-truncated \textit{local Gaussian cross-spectrum}
$\lgcsdM{k\ell}{\LGp}{\omega}{m}$ can be used to detect nonlinear
cross-temporal dependencies between the marginals $\Yz{k,t}$ and
$\Yz{\ell,t}$ of a multivariate time series
$\TSR{\bmYz{t}=\parenR{\Yz{1,t},\cdots,\Yz{d,t}}}{t\in\ZZ}{}$.

The resulting \textit{local Gaussian approach to spectral analysis}
can in particular be used to detect if the interdependency structure
of the time series under investigation deviates from being Gaussian.
For time series whose ordinary auto- and cross-spectra are flat, any
peaks and troughs from the \textit{local Gaussian approach} can then
be considered indicators of \textit{local nonlinear traits} and
\textit{local periodicities}.

The $m$-truncated estimates $\hatlgcsdM{k\ell}{\LGp}{\omega}{m}$ can,
as discussed for GARCH-models in
\cref{sec:lgch:EuStockMarkets+cGARCH}, also be of interest with regard
to \textit{local comparisons} of models fitted to a given sample.  In
particular, it is possible to use this approach to check if a model
fitted to the data can reproduce local traits detected in the original
sample, and an investigator can thus use a \textit{local Gaussian
  sanity-test} of the fitted models as a supplement to other model
selection methods.

All the examples in this paper can be reproduced by the help of the
scripts that are included as a part of the \Rpackage
\lgsdRpackage. The interested reader can run these scripts, and then
use the integrated \Rref{shiny}-application from this \Rpackage to
investigate how the resulting plots change when the tuning parameters
in the estimation algorithm are modified.  This enables a much deeper
inspection of the results than the static plots contained in this
paper.  These scripts can also be used as templates for new
investigations, and they can hopefully help any interested readers to
test out the local Gaussian approach to spectral analysis on their own
data.

\section*{Supplementary Material}
The online Supplementary Material contains the appendices.  The
scripts needed for the reproduction of the examples in this paper is
contained in the \Rpackage \lgsdRpackage, cf.\
\cref{app:data_details} for further details.

\clearpage{}

  \putbib     
\end{bibunit}

\begin{bibunit}[elsarticle-harv]
  \spacingset{1.5}

\newpage

\appendix

\fancyhead[C]{Supplementary Material --- {Local Gaussian
    cross-spectrum analysis}} %
\thispagestyle{empty}

\renewcommand\thefigure{\thesection.\arabic{figure}}
\setcounter{figure}{0} 

\begin{center}
  {\Large SUPPLEMENTARY MATERIAL}
  \setcounter{page}{1}
\end{center}

This is the supplementary material to the paper \textit{Local Gaussian
  cross-spectrum analysis}, which extends the univariate theory
developed in \citet{jordanger17:_lgsd} (hereafter referred to as \JT)
to the multivariate case.

\Cref{sec:proofs} contains the proofs of
\cref{th:lgcsd:asymptotics_for_hatlgcsd,th:lgcsd:asymptotics_for_amplitude_spectrum,th:lgcsd:asymptotics_for_phase_spectrum}
(from the main part), whereas \cref{sec:underlying_asymptotic_theory}
gathers the theoretical framework required for these proofs.  Note
that the multivariate local Gaussian spectral theory is quite similar
to the theory developed for the univariate case, and the discussion in
\cref{sec:underlying_asymptotic_theory} thus frequently refers to the
theory developed in the Supplementary Material of \JT and it primarily
focuses on the adjustments needed to cope with the present
multivariate case.

The sensitivity of the tuning parameters (and the effect of varying
the point $\LGp$) are discussed in \cref{P2.app:sensitivity_analysis},
whereas \cref{P2.How.to.select.the.tuning.parameters?} briefly
comments upon the task of selecting these parameters.  Note that this
discussion in essence is the same as the one given for the univariate
case, see
\lgsdRef{app:sensitivity_analysis,How.to.select.the.tuning.parameters?},
but it has been included here for the convenience of the reader.  The
scrips needed for the reproduction of these multivariate examples are
included in the \Rpackage \lgsdRpackage, and an investigator can use
these scrips as templates for similar investigations.

\Cref{P2.app:regarding_resampling} is devoted to sampling and
resampling strategies.  The first part of
\cref{P2.app:regarding_resampling} applies the univariate local
Gaussian sanity testing of parametric models from \JT to the DAX- and
CAC-margins of the \EuStockMarkets example, and after this new plots
are presented that can be used for the sanity testing of the
cross-interaction of the multivariate parametric models fitted to the
\EuStockMarkets data.  The second part of
\cref{P2.app:regarding_resampling} explains the motivation for the
\textit{Circular index-based block bootstrap for tuples} resampling
strategy that was developed in \JT,
and this part also contains a sensitivity analysis of the block length
$L$, which reveals that this tuning parameter plays a minor role for
the time series that are long enough for a local Gaussian spectral
investigation to be of interest.

All the examples/investigations in this paper can be reproduced by the
scripts contained in the \Rpackage \lgsdRpackage, and the interested
reader will find some practical information related to that in
\cref{P2:app:scripts_and_details_related_to_the_examples}.  Some
further discussions of the examples in \cref{sec:lgch_examples} have
also been collected in
\cref{P2:app:scripts_and_details_related_to_the_examples} in order to
improve the flow of the main part.  The interested reader will here
find the details related to the bivariate Gaussian test-example, a
short explanation of how animations of complex-valued plots can be
used to investigate the estimated local Gaussian cross-spectra, and
some additional details related to the construction of the local
trigonometric examples used for the sanity-testing of the estimation
algorithm.  Some details are also included at the end related to
limitations of the local Gaussian approach for points $\LGp$ that lies
in the extreme tails of the observations.

\section{The proofs of \cref{th:lgcsd:asymptotics_for_hatlgcsd,th:lgcsd:asymptotics_for_amplitude_spectrum,th:lgcsd:asymptotics_for_phase_spectrum}}
\label{sec:proofs}
\setcounter{figure}{0} 

This appendix presents the proofs of the asymptotic results stated in
the main part of the paper.  The proof of the result for the
\mbox{$m$-truncated} estimate of the local Gaussian cross-spectrum
$\lgcsd[p]{k\ell}{\LGp}{\omega}$ is in essence identical to the one
encountered in \JT for the local Gaussian
auto-spectrum $\lgcsd[p]{kk}{\LGp}{\omega}$, whereas the proofs for
the estimates of the local Gaussian amplitude- and phase-spectra are
identical in structure to those encountered in the ordinary global
case.  Some technical details needed for the proof of
\cref{th:lgcsd:asymptotics_for_hatlgcsd} are covered in 
\cref{sec:underlying_asymptotic_theory}.

\begin{proof}[Proof of \cref{th:lgcsd:asymptotics_for_hatlgcsd},
  page~\pageref{th:lgcsd:asymptotics_for_hatlgcsd}] \ \\
          The case $\omega\not\in\tfrac{1}{2}\cdot\ZZ$ will be treated first,
  since the other case follows from a trivial adjustment of the setup.
    The key observation for this case is that the sum that defines
  $\hatlgcsdM[p]{k\ell}{\LGp}{\omega}{m}$, see
  \cref{eq:App:hatlgsd_definition_main_document} in
  \myref{def:lgsd_estimator}{lgcsd_estimator_folded}, implies that
  $\hatlgcsdCoM[p]{k\ell}{\LGp}{\omega}{m}$ and
  $\hatlgcsdQuadM[p]{k\ell}{\LGp}{\omega}{m}$ can be realised as the
  following inner products,
  \begin{subequations}
    \label{eq:M_truncated_lgsdCo_and_lgsdQuad_vectors}
    \begin{align}
      \label{eq:M_truncated_lgsdCo_vector}
      \hatlgcsdCoM[p]{k\ell}{\LGp}{\omega}{m}
      &= \bmLambdaz[\!\prime]{c|\OUbar{m}}(\omega)\cdot
        \widehatbmPz{k\ell:\OUbar{m}|\bm{b}}\!\parenR{\LGp,\LGpd} \\
      \label{eq:M_truncated_lgsdQuad_vector}
      \hatlgcsdQuadM[p]{k\ell}{\LGp}{\omega}{m}
      &= \bmLambdaz[\!\prime]{q|\OUbar{m}}(\omega)\cdot
        \widehatbmPz{k\ell:\OUbar{m}|\bm{b}}\!\parenR{\LGp,\LGpd},
    \end{align}
  \end{subequations}
  where $\bmLambdaz[\!\prime]{c|\OUbar{m}}(\omega)$ and
  $\bmLambdaz[\!\prime]{q|\OUbar{m}}(\omega)$ respectively contains the
  coefficients $\lambdazM{h}{m}\cos(2\pi\omega h)$ and
  $\lambdazM{h}{m}\sin(2\pi\omega h)$, and where
  $\widehatbmPz{k\ell:\OUbar{m}|\bm{b}}\!\parenR{\LGp,\LGpd} =
  \Vector[\prime]{\hatlgccr[p]{\ell k}{\LGpd}{m},\dotsc,\hatlgccr[p]{\ell
      k}{\LGpd}{1},
    \hatlgccr[p]{k\ell}{\LGp}{0},\dotsc,\hatlgccr[p]{k\ell}{\LGp}{m}}$
  contains the \mbox{$2m+1$} estimates of the local Gaussian
  cross-correlations (based on the bandwidth~$\bm{b}$).
  The vector
  $\widehatbmPz{k\ell:\OUbar{m}|\bm{b}}\!\parenR{\LGp,\LGpd}$ can by
  the help of a suitable \mbox{$(2m+1)\times5(2m+1)$} matrix
  $\bmEz[\prime]{\OUbar{m}}$ (based on the vectors
  $\bmez[\prime]{5}$ that gives
  $\lgccr[p]{k\ell}{\LGp}{h} =
  \bmez[\prime]{5}\cdot\thetahc{}{\LGp}{k\ell}{h}$) be expressed as
  \begin{align}
    \label{eq:lgcsd:rho_from_theta}
    \widehatbmPz{k\ell:\OUbar{m}|\bm{b}}\!\parenR{\LGp,\LGpd} =
    \bmEz[\prime]{\OUbar{m}} \cdot
    \widehatbmThetaz{k\ell:\OUbar{m}|\bm{b}}\!\parenR{\LGp,\LGpd},
  \end{align}
  where
  $\widehatbmThetaz{k\ell:\OUbar{m}|\bm{b}}\!\parenR{\LGp,\LGpd} =
  \Vector[\prime]{\hatthetahNc{}{\LGpd}{\ell
      k}{m}{n},\dotsc,\hatthetahNc{}{\LGpd}{\ell
      k}{1}{n},\hatthetahNc{}{\LGp}{k\ell}{0}{n},\dotsc,\hatthetahNc{}{\LGp}{k\ell}{m}{n}}$
  is the length \mbox{$5(2m+1)$} vector created by stacking into one
  vector all the estimated parameters from the local Gaussian
  approximations.  It follows from this that the target of interest
  can be written as
  \begin{align}
    \begin{bmatrix}
      \hatlgcsdCoM[p]{k\ell}{\LGp}{\omega}{m} \\
      \hatlgcsdQuadM[p]{k\ell}{\LGp}{\omega}{m}
    \end{bmatrix} =
    \begin{bmatrix}
      \bmLambdaz[\!\prime]{c|\OUbar{m}}(\omega)
      \\ \bmLambdaz[\!\prime]{q|\OUbar{m}}(\omega)
    \end{bmatrix}
    \cdot
    \bmEz[\prime]{\OUbar{m}} \cdot
    \widehatbmThetaz{k\ell:\OUbar{m}|\bm{b}}\!\parenR{\LGp,\LGpd},
  \end{align}
  which together with the asymptotic normality result from
  \cref{th:lgcsd:asymptotics_for_hatlgtheta_LGp_and_LGpd}
  (page~\pageref{th:lgcsd:asymptotics_for_hatlgtheta_LGp_and_LGpd}),
  i.e.,\
  \begin{align}
    \sqrt{n \!\parenRz[3]{}{\bz{1}\bz{2}}} \cdot
    \left(\widehatbmThetaz{k\ell:\OUbar{m}|\bm{b}}\!\parenR{\LGp,\LGpd}
    - \bmThetaz{k\ell:\OUbar{m}|\bm{b}}\!\parenR{\LGp,\LGpd} \right)
    \stackrel{\scriptscriptstyle d}{\longrightarrow}
    \UVN{\bm{0}}{\Sigmaz{\LGp|k\ell:\OUbar{m}}},
  \end{align}
  and \citet[proposition~6.4.2, p.~211]{Brockwell:1986:TST:17326}
  gives that
  \begin{align}
    \sqrt{n \!\parenRz[3]{}{\bz{1}\bz{2}} \!/ m} \cdot     \parenR{
      \begin{bmatrix}
        \hatlgcsdCoM[p]{k\ell}{\LGp}{\omega}{m} \\
        \hatlgcsdQuadM[p]{k\ell}{\LGp}{\omega}{m}
      \end{bmatrix}
      -     \begin{bmatrix}
        \lgcsdCo[p]{k\ell}{\LGp}{\omega} \\
        \lgcsdQuad[p]{k\ell}{\LGp}{\omega}
    \end{bmatrix} }
  \end{align}
  is asymptotically bivariate normally distributed with mean $\bm{0}$
  and covariance matrix
  \begin{align}
    \frac{1}{m} \cdot \parenR{
    \begin{bmatrix}
      \bmLambdaz[\!\prime]{c|\OUbar{m}}(\omega)
      \\ \bmLambdaz[\!\prime]{q|\OUbar{m}}(\omega)
    \end{bmatrix}
    \cdot
    \bmEz[\prime]{\OUbar{m}} \cdot
    \Sigmaz{\LGp|k\ell:\OUbar{m}} \cdot
    \bmEz{\OUbar{m}} \cdot
    \begin{bmatrix}
      \bmLambdaz{c|\OUbar{m}},
      \bmLambdaz{q|\OUbar{m}}
    \end{bmatrix} }.
  \end{align}
                      The specified form of the covariance matrix given in
  \cref{th:lgcsd:asymptotics_for_hatlgcsd} now follows by the help of
  some linear algebra, the observation in
  \cref{th:lgcsd:asymptotics_for_hatlgtheta_LGp_and_LGpd} that
  \begin{align}
    \label{eq:covariance_matrix_general_case}
    \Sigmaz{\LGp|k\ell:\OUbar{m}}
    &\defeq
      \parenR{\oplusss{h=m}{1}\Sigmaz{\LGpd|\ell k:h}} \bigoplus
      \parenR{\oplusss{h=0}{m}\Sigmaz{\LGp|k\ell:h}},
  \end{align}
  and the definition
  \mbox{$\tildesigmaz[2]{\!\LGp|k\ell}(h) \defeq
    \bmez[\prime]{5}\cdot\Sigmaz{\LGp|k\ell:h}\cdot\bmez{5}$} from
  \cref{th:asymptotic_result_cross_correlation}.

  It is easy to see that both
  $\sigmaz[2]{\!c|k\ell:\LGp}(\omega)$ and
  $\sigmaz[2]{\!q|k\ell:\LGp}(\omega)$ from
  \cref{eq:lgcsd:th:asymptotics_for_hatlgsd_variance_RE_and_IM} are
  nonzero when $\omega\not\in\tfrac{1}{2}\cdot\ZZ$, as required for
  the validity of \citet[proposition~6.4.2,
  p.~211]{Brockwell:1986:TST:17326}.  The proof for the case
  $\omega\in\tfrac{1}{2}\cdot\ZZ$ can be constructed in the same
  manner, simply ignoring the components having sine-terms.
\end{proof}

The key observation for the proof of
\cref{th:lgcsd:asymptotics_for_amplitude_spectrum,th:lgcsd:asymptotics_for_phase_spectrum}
is that they both follow as a consequence of
\cref{th:lgcsd:asymptotics_for_hatlgcsd} and \citet[proposition~6.4.3,
  p.~211]{Brockwell:1986:TST:17326}.  Note that these arguments are
quite similar to those used for the investigation of the estimates of
the ordinary amplitude and phase spectra in
\citet[p.448-449]{Brockwell:1986:TST:17326}.

\begin{proof}[Proof of
    \cref{th:lgcsd:asymptotics_for_amplitude_spectrum},
    page~\pageref{th:lgcsd:asymptotics_for_amplitude_spectrum}] \ \\
    First observe that the function $h(\xz{1},\xz{2}) = \sqrt{\xz[2]{1}
    + \xz[2]{2}}$ implies that
  \begin{align}
    \hatlgcsdAmplitudeM[p]{k\ell}{\LGp}{\omega}{m} -
    \lgcsdAmplitude[p]{k\ell}{\LGp}{\omega} =
    h\!\parenR{\hatlgcsdCoM[p]{k\ell}{\LGp}{\omega}{m},
      \hatlgcsdQuadM[p]{k\ell}{\LGp}{\omega}{m}} -
    h\!\parenR{\lgcsdCo[p]{k\ell}{\LGp}{\omega},
      \lgcsdQuad[p]{k\ell}{\LGp}{\omega}},
  \end{align}
  and then observe that the asymptotic covariance matrix in
  \cref{th:lgcsd:asymptotics_for_hatlgcsd}
  \begin{align}
    \Sigmaz{k\ell:\LGp}(\omega) \defeq
    \begin{pmatrix}
      \sigmaz[2]{\!c|k\ell:\LGp}(\omega) & 0 \\
      0 &\sigmaz[2]{\!q|k\ell:\LGp}(\omega) 
    \end{pmatrix},
  \end{align}
  obviously is a symmetric non-negative definite matrix.

  It now follows from \citet[proposition~6.4.2,
    p.~211]{Brockwell:1986:TST:17326} that
  \begin{align}
    \sqrt{n \!\parenRz[3]{}{\bz{1}\bz{2}} \!/ m} \cdot \parenC{
      h\!\parenR{\hatlgcsdCoM[p]{k\ell}{\LGp}{\omega}{m},
        \hatlgcsdQuadM[p]{k\ell}{\LGp}{\omega}{m}} -
      h\!\parenR{\lgcsdCo[p]{k\ell}{\LGp}{\omega},
        \lgcsdQuad[p]{k\ell}{\LGp}{\omega}} }
    \stackrel{\scriptscriptstyle d}{\longrightarrow}
    \UVN{0}{\sigmaz[2]{\!\alpha}(\omega)},
  \end{align}
  where $\sigmaz[2]{\!\alpha}(\omega) =
  \bmDz{}\cdot\Sigmaz{k\ell:\LGp}(\omega)\cdot\bmDz[\prime]{}$, with
  \begin{align}
    \bmDz{} = \Vector{ \frac{\partial}{\partial\xz{1}}h(\xz{1},\xz{2}),
    \frac{\partial}{\partial\xz{2}}h(\xz{1},\xz{2})} =
  \Vector{\xz{1}/\sqrt{\xz[2]{1} + \xz[2]{2}},\xz{2}/\sqrt{\xz[2]{1} +
      \xz[2]{2}}}
  \end{align}
          evaluated in
  $(\xz{1},\xz{2})=\parenR{\lgcsdCo[p]{k\ell}{\LGp}{\omega},
    \lgcsdQuad[p]{k\ell}{\LGp}{\omega}}$.

  A simple calculation gives $\bmDz{}=
  \Vector{\lgcsdCo[p]{k\ell}{\LGp}{\omega}/\lgcsdAmplitude[p]{k\ell}{\LGp}{\omega},
    \lgcsdQuad[p]{k\ell}{\LGp}{\omega}/\lgcsdAmplitude[p]{k\ell}{\LGp}{\omega}}$,
  from which it follows that $\sigmaz[2]{\!\alpha}(\omega) = \parenR{
    \lgcsdCoM[p]{k\ell}{\LGp}{\omega}{2}\cdot\sigmaz[2]{\!c|k\ell:\LGp}(\omega)
    +
    \lgcsdQuadM[p]{k\ell}{\LGp}{\omega}{2}\cdot\sigmaz[2]{\!q|k\ell:\LGp}(\omega)
  } / \lgcsdAmplitudeM[p]{k\ell}{\LGp}{\omega}{2}$.
\end{proof}

\begin{proof}[Proof of
    \cref{th:lgcsd:asymptotics_for_phase_spectrum},
    page~\pageref{th:lgcsd:asymptotics_for_phase_spectrum}] \ \\
    This argument is quite similar to the proof of
  \cref{th:lgcsd:asymptotics_for_amplitude_spectrum}, and only the
  details that are different will thus be included.  In this case the
  function of interest is 
          $h(\xz{1},\xz{2}) = \subp{\tan}{}{}{}{-1}\!\parenR{\xz{2}/\xz{1}}$,
  from which it follows that 
  \begin{align}
    \Vector{
    \frac{\partial}{\partial\xz{1}}h(\xz{1},\xz{2}),
    \frac{\partial}{\partial\xz{2}}h(\xz{1},\xz{2})} =
    \Vector{-\xz{2}/\parenR{\xz[2]{1} +
    \xz[2]{2}},\xz{1}/\parenR{\xz[2]{1} + \xz[2]{2}}}.  
  \end{align}
            This
  implies that $\bmDz{}=
  \Vector{-\lgcsdQuad[p]{k\ell}{\LGp}{\omega}/\lgcsdAmplitudeM[p]{k\ell}{\LGp}{\omega}{2},
    \lgcsdCo[p]{k\ell}{\LGp}{\omega}/\lgcsdAmplitudeM[p]{k\ell}{\LGp}{\omega}{2}}$
  and a simple calculation now gives
  \begin{align}
    \sigmaz[2]{\!\phi}(\omega) = \parenR{
    \lgcsdQuadM[p]{k\ell}{\LGp}{\omega}{2}\cdot\sigmaz[2]{\!c|k\ell:\LGp}(\omega)
    +
    \lgcsdCoM[p]{k\ell}{\LGp}{\omega}{2}\cdot\sigmaz[2]{\!q|k\ell:\LGp}(\omega)
    } / \lgcsdAmplitudeM[p]{k\ell}{\LGp}{\omega}{4},
  \end{align}
  which completes the proof.
            \end{proof}

\section{The underlying asymptotic results}
\label{sec:underlying_asymptotic_theory}
\setcounter{figure}{0}

\makeatletter{}
\subsection{The bivariate case, a brief overview and the $\hatlgccr[p]{k\ell}{\LGp}{h}$-case}
\label{sec:the_bivariate_case}

The main ingredient for the theoretical setup is a translation of the
bivariate results from \citet{Tjostheim201333} into the multivariate
framework, and this is almost identical to the discussion that was
given in \lgsdRef{app:bivariate_penalty_functions}.  The main
difference is that two extra indices ($k$ and~$\ell$) now are needed
in order to specify which components from
\mbox{$\bmYz{t}=\parenR{\Yz{1,t},\dotsc,\Yz{d,t}}$} that are
investigated.

The basic building-blocks was given in \cref{app:assumptions_upon_Yt},
see
\cref{def:lgcsd_Y_kl_ht_and_Y_lk_ht,def:Uh,def:lgcsd:kernel,def:Xklht},
and the first target of interest is to define a suitable bivariate
penalty function relative to the requirement of
\cref{eq:lgcsd:def_of_theta_hb}.  For a sample of size~$n$ from
$\TSR{\Yht{k\ell:h}{t}}{t\in\ZZ}{}$, and with the present notation,
the local penalty function from \citet{Tjostheim201333} can be
described as
\begin{align}
  \nonumber
  \QhNc[\thetahc{}{\LGp}{k\ell}{h}]{}{\LGp}{k\ell}{h}{n}
                &\defeq - \sumss{t=1}{n} \KhbDEFc \log \psi\!\left( \Yht{k\ell:h}{t};
  \thetahc{}{\LGp}{k\ell}{h} \right) \\
  \label{eq:lgcsd:QhN}
  &\phantom{= -} + n \intss{\RRn{2}}{} \Khbdefc \psi\!\left(\yhc{k\ell}{h};
  \thetahc{}{\LGp}{k\ell}{h}\right) \dyhc{k\ell}{h},
\end{align}
and from this, under suitable regularity conditions and by the help of
the Klimko-Nelson approach, the following asymptotic normality result
can be obtained for the estimated parameters,
\begin{align}
  \label{eq:lgcsd:CLT_bivariate_general_case}
  \sqrt{n \!\parenRz[3]{}{\bz{1}\bz{2}}} \cdot
  \left(\hatthetahNc{}{\LGp}{k\ell}{h}{n} -
  \thetahc{}{\LGp}{k\ell}{h}\right) \stackrel{\scriptscriptstyle
  d}{\longrightarrow} \UVN{\bm{0}}{\Sigmaz{\LGp|k\ell:h}}.
\end{align}
See \citet[Th.~3]{Tjostheim201333} for the details.

Notice that it only is the correlation component of $\hatthetahNc{}{\LGp}{k\ell}{h}{n} -
\thetahc{}{\LGp}{k\ell}{h}$ that is of interest for the present
paper.   The relevant result for that part is stated below in order to
give a reference for the statements in 
\cref{th:lgcsd:asymptotics_for_hatlgcsd}.

\begin{theorem}
  \label{th:asymptotic_result_cross_correlation}
  Under \cref{assumption_lgcsd_Yt,assumption:lgcsd:score_function,assumption:lgcsd:Nmb}, the following
  univariate asymptotic result holds for the estimates
  $\hatlgccr[p]{k\ell}{\LGp}{h}$ of the local Gaussian
  cross-correlations $\lgccr[p]{k\ell}{\LGp}{h}$.
  \begin{align}
    \label{eq:lgcsd:CLT_bivariate_general_case_rho}
    \sqrt{n \!\parenRz[3]{}{\bz{1}\bz{2}}} \cdot
    \left(\hatlgccr[p]{k\ell}{\LGp}{h} - \lgccr[p]{k\ell}{\LGp}{h}
    \right) \stackrel{\scriptscriptstyle d}{\longrightarrow}
    \UVN{0}{\tildesigmaz[2]{\!\LGp|k\ell}(h)},
  \end{align}
  where \mbox{$\tildesigmaz[2]{\!\LGp|k\ell}(h) \defeq
    \bmez[\prime]{5}\cdot\Sigmaz{\LGp|k\ell:h}\cdot\bmez{5}$}.
\end{theorem}

\begin{proof}
  With $\bmez[\prime]{5}$ the unit vector that picks out the
  correlation part from $\hatthetahNc{}{\LGp}{k\ell}{h}{n} -
  \thetahc{}{\LGp}{k\ell}{h}$, it follows from
  \citet[proposition~6.4.2, p.~211]{Brockwell:1986:TST:17326} that
  \begin{align}
    \label{eq:lgcsd:CLT_bivariate_general_case_p=5}
    \sqrt{n \!\parenRz[\,3]{}{\bz{1}\bz{2}}} \cdot
    \left(\hatlgccr[5]{k\ell}{\LGp}{h} - \lgccr[5]{k\ell}{\LGp}{h}
    \right) \stackrel{\scriptscriptstyle d}{\longrightarrow}
    \UVN{0}{\bmez[\prime]{5}\cdot\Sigmaz{\LGp|k\ell:h|5}\cdot\bmez{5}}.
  \end{align}
\end{proof}

\makeatletter{}

\subsection{The asymptotic result for $\widehatbmThetaz{k\ell:\OUbar{m}|\bm{b}}\!\parenR{\LGp,\LGpd}$}
\label{sec:asymptotic_theory}

The Klimko-Nelson approach from the bivariate case can be extended to
the present case of interest in the same manner as it was done for the
local Gaussian auto-spectrum in \lgsdRef{sec:Technical_Results}.

\begin{definition}
  \label{def:lgcsd:QMN}
  For each bivariate penalty function
  $\QhNc[\thetahc{}{\LGp}{k\ell}{h}]{}{\LGp}{k\ell}{h}{n}$ (as given
  in \cref{eq:lgcsd:QhN}), denote by
  $\tildeQhNc[\thetahc{}{\LGp}{k\ell}{h}]{}{\LGp}{k\ell}{h}{n}$ the
  extension of it from a function of
  \mbox{$\Yht{\!k\ell:h}{t}\defeq\Vector[\prime]{\Yz{k,t+h},
      \Yz{\ell,t}}$} to a function of
  $\Vector[\prime]{\Yz{k,t+m},\dotsc,\Yz{k,t},\Yz{\ell,t+m},\dotsc,
    \Yz{\ell,t}}$.  Use these extensions do define the new penalty
  function
    \begin{align}
    \label{eq:lgcsd:QMN_definition}
      \QMNc[\bmThetaz{k\ell:\OUbar{m}|\bm{b}}\!\parenR{\LGp,\LGpd}]{}{\LGp}{\LGpd}{k\ell}{m}{n}
    &\defeq \sumss{h=m}{1} \tildeQhNc[\thetahc{}{\LGpd}{\ell
        k}{h}]{}{\LGpd}{\ell k}{h}{n} + \sumss{h=0}{m}
    \tildeQhNc[\thetahc{}{\LGp}{k\ell}{h}]{}{\LGp}{k\ell}{h}{n} .
                          \end{align}
\end{definition}

The \mbox{$2m+1$} bivariate components in the sum that defines
$\QMNc[\widehatbmThetaz{k\ell:\OUbar{m}|\bm{b}}\!\parenR{\LGp,\LGpd}]{}{\LGp}{\LGpd}{k\ell}{m}{n}$
have no common parameters, so the optimisation of the parameters for
the different summands can be performed~independently.  The optimal
parameter
vector~$\widehatbmThetaz{k\ell:\OUbar{m}|\bm{b}}\!\parenR{\LGp,\LGpd}$
for the penalty function $\QMNc{}{\LGp}{\LGpd}{k\ell}{m}{n}$ (for a
given sample) can thus be constructed by stacking on top of each other
the parameter vectors $\hatthetahc{}{\LGp}{k\ell}{h}$ and
$\hatthetahc{}{\LGpd}{\ell k}{h}$ that optimise the individual
summands in~\cref{eq:lgcsd:QMN_definition}.

The Klimko-Nelson approach can now be used on the penalty function
from \cref{eq:lgcsd:QMN_definition}
i.e.,\ four requirements related to the penalty function must be
verified before the desired asymptotic result for the paramater-vector
$\widehatbmThetaz{k\ell:\OUbar{m}|\bm{b}}\!\parenR{\LGp,\LGpd}$ is
obtained.  The following cross-spectrum analogue of
\lgsdRef{th:asymptotics_for_hatlgtheta_LGp_and_LGpd} can now be stated
for the present case of interest.

\begin{theorem}
  \label{th:lgcsd:asymptotics_for_hatlgtheta_LGp_and_LGpd}
  Under \cref{assumption_lgcsd_Yt,assumption:lgcsd:score_function,assumption:lgcsd:Nmb}, the following
  asymptotic behaviour holds for the estimated parameters
  $\widehatbmThetaz{k\ell:\OUbar{m}|\bm{b}}\!\parenR{\LGp,\LGpd} =
  \Vector[\prime]{\hatthetahNc{}{\LGpd}{\ell
      k}{m}{n},\dotsc,\hatthetahNc{}{\LGpd}{\ell
      k}{1}{n},\hatthetahNc{}{\LGp}{k\ell}{0}{n},\dotsc,\hatthetahNc{}{\LGp}{k\ell}{m}{n}}$,
    \begin{align}
    \label{eq:th:lgcsd:asymptotics_for_hatlgtheta_LGp_and_LGpd}
    \sqrt{n \!\parenRz[3]{}{\bz{1}\bz{2}}} \cdot
    \left(\widehatbmThetaz{k\ell:\OUbar{m}|\bm{b}}\!\parenR{\LGp,\LGpd}
    - \bmThetaz{k\ell:\OUbar{m}|\bm{b}}\!\parenR{\LGp,\LGpd} \right)
    \stackrel{\scriptscriptstyle d}{\longrightarrow}
    \UVN{\bm{0}}{\Sigmaz{\LGp|k\ell:\OUbar{m}}},
  \end{align}
  where the matrix $\Sigmaz{\LGp|k\ell:\OUbar{m}}$ is the direct sum
  of the matrices from \cref{eq:lgcsd:CLT_bivariate_general_case} that
  occurs when the individual bivariate components of the penalty
  function is investigated, i.e.,\
  \begin{align}
    \Sigmaz{\LGp|k\ell:\OUbar{m}} \defeq
    \parenR{\oplusss{h=m}{1}\Sigmaz{\LGpd|\ell k:h}} \bigoplus
    \parenR{\oplusss{h=0}{m}\Sigmaz{\LGp|k\ell:h}}.
  \end{align}
\end{theorem}

\begin{proof}
  This result follows when the Klimko-Nelson approach is used with the
  local penalty-function
  $\QMNc[\bmThetaz{k\ell:\OUbar{m}|\bm{b}}\!\parenR{\LGp,\LGpd}]{}{\LGp}{\LGpd}{k\ell}{m}{n}$
  from \cref{eq:lgcsd:QMN_definition}, and the proof is in essence
  identical to the proof of
  \lgsdRef{th:asymptotics_for_hatlgtheta_LGp_and_LGpd}.  The three
  first requirements of the Klimko-Nelson approach follows trivially
  from the corresponding investigation for the bivariate case, whereas
  the proof of the fourth requirement must take into account how
  $\mlimit$ and $\blimit$ as $\nlimit$.
 
  The investigation of the fourth requirement of the Klimko-Nelson
  approach can be done in the exact same manner that was employed in
  \JT, i.e.,\ first construct a collection of
  simple random variables whose interaction and asymptotic properties
  are easy to investigate, then use these basic building blocks to
  construct a more complicated random variable
  $\QNhvec{n}{\LGp|\OUbar{m}}$ that has the same limiting distribution
  as the estimator of \mbox{$\sqrt{\bz{1}\bz{2}}\nablaMc{}{k\ell}{m}
    \QMNc[\bmThetaz{k\ell:\OUbar{m}|\bm{b}}\!\parenR{\LGp,\LGpd}]{}{\LGp}{\LGpd}{k\ell}{m}{n}$}
  (where $\nablaMc{}{k\ell}{m}$ is obtained by stacking
  together~$\nablahc{}{k\ell}{h}$).  After this, it is sufficient to
  use standard methods to prove that the limiting distribution of
  $\QNhvec{n}{\LGp|\OUbar{m}}$ is the desired multivariate normal
  distribution, and the statement for the parameter vectors then
  follows from the Klimko-Nelson theorem and some linear algebra.
\end{proof}

\section{Sensitivity analysis of the tuning parameters}
\label{P2.app:sensitivity_analysis}
\setcounter{figure}{0} 

An investigation of how sensitive the estimates of the local Gaussian
auto-spectra $\hatlgcsdM{\ell\ell}{\LGp}{\omega}{m}$ are to changes in
the tuning parameters (and the point $\LGp$) can be found in
\lgsdRef{app:sensitivity_analysis}.  This section presents a similar
investigation for the estimates of the local Gaussian cross-spectra
$\hatlgcsdM{k\ell}{\LGp}{\omega}{m}$, and it is seen here that the
features observed for the univariate case are also present in the
multivariate case.

Different plots based on estimates of the local Gaussian
cross-correlations and the corresponding local Gaussian cross-spectra,
i.e.,\ $\hatlgccr{k\ell}{\LGp}{h}$ and
$\hatlgcsdM{k\ell}{\LGp}{\omega}{m}$, will be encountered, and some of
these plots are based on the distance function $D$ from
\lgsdRef{app:method_for_sensitivity_analysis}.  The definition of $D$,
and a short discussion of it, is for the convenience of the reader
repeated in \cref{P2.app:method_for_sensitivity_analysis} in the
present Supplementary Material.

\cref{P2.app:Point_sensitivity,P2.app:Bandwidth_sensitivity}
respectively consider the sensitivity of the point $\LGp$ and the
bandwidth $\bm{b}$, whereas the sensitivity of the truncation level
$m$ is discussed in \cref{P2.app:Truncation_level_sensitivity}.  The
effect the value of the block length $L$ has upon the bootstrap-based
pointwise confidence intervals is discussed in
\cref{P2.app:Block_length_sensitivity}, since that gives the most
natural flow.

The scripts required for the replication of the results in this
section are contained in the \Rpackage \lgsdRpackage, and these
scripts can be used as templates for those that would like to
investigate other time series in a similar manner.  See
\cref{P2.app:The.scripts.in.lgsdRpackage} for details.

\subsection{Sensitivity analysis - the distance function}
\label{P2.app:method_for_sensitivity_analysis}

The purpose of this section is to repeat \textit{the motivation for}
and \textit{the definition of} the distance function $D$ given in
\lgsdRef{app:method_for_sensitivity_analysis}.  The discussion of
possible alternatives will not be included here.

\textbf{Motivation:} An investigation of the sensitivity of the
different tuning parameters requires a tool that can measure the
differences that occur when these tuning parameters are adjusted.  The
distance function $D$ used in this paper is the one inherited from the
complex Hilbert space of Fourier series on the interval
$\parenS{-\tfrac{1}{2},\tfrac{1}{2}}$, cf.\ e.g.\
\citet[Ch.~2.8]{Brockwell:1986:TST:17326}, i.e.,\ for
$f(\omega)=\sumss{h=-\infty}{\infty}\rho(h)\ez[-2\pi i h]{}$ the norm
is defined by
$\subp{||f(\omega)||}{}{}{}{2} = \intss{-1/2}{1/2}
f(\omega)\overbar{f(\omega)}\d{\omega} = \sumss{h=-\infty}{\infty}
\subp{\rho(h)}{}{}{}{2}$.  This motivates the following definition.

\begin{definition}
  \label{P2.def:sensitivity_analysis_distance_function}
  Given two spectra
  $\fz{1}(\omega)=\sumss{h=-\infty}{\infty}\rhoz{1}(h)\ez[-2\pi i
  h]{}$ and
  $\fz{2}(\omega)=\sumss{h=-\infty}{\infty}\rhoz{2}(h)\ez[-2\pi i
  h]{}$, the distance between them is denoted by
  \begin{equation}
    \label{P2.eq:def:sensitivity_analysis_distance_function}
    D(\fz{1}(\omega),\fz{2}(\omega)) \defeq \sqrt{\sumss{h=-\infty}{\infty}
      \subp{\parenR{\rhoz{1}(h)-\rhoz{2}(h)}}{}{}{}{2} }.
  \end{equation}
  Furthermore: The notation $D(\fz{1}(\omega))$ will be interpreted as
  $D(\fz{1}(\omega),0)$, which implies that
  $D(\fz{1}(\omega),\fz{2}(\omega))$ also can be written as
  $D(\fz{1}(\omega) -\fz{2}(\omega))$ (which is used in
  \cref{fig:m_sensitivity_bivariate}).
\end{definition}

The obvious adjustment must be applied when the distance function $D$
is used on $m$-truncated estimates
$\hatlgcsdM{k\ell}{\LGp}{\omega}{m}$, i.e.,\ $\rho(h)$ should be
replaced with $\lambdazM{h}{m}\cdot\hatlgccr{k\ell}{\LGp}{h}$ when
$|h|\leq m$, and with 0 when $|h|>m$.

\textbf{Regarding the frequency-dimension:} The distance measure $D$
given in \cref{P2.def:sensitivity_analysis_distance_function} does not
contain any information about the frequencies, and completely
different spectral densities can have the same distance-value.  It is
thus, for the purpose of sensitivity analysis, important to combine
distance-based plots with plots that reveal something about the
frequency-component too.

\subsection{Sensitivity analysis: The point $\LGp$}
\label{P2.app:Point_sensitivity}

The point $\LGp$, contrary to the bandwidth $\bm{b}$ and the
truncation level $m$, is not a tuning parameter of the estimation
algorithm --- but it is natural to investigate how the local Gaussian
spectra $\hatlgcsdM[p]{k\ell}{\LGp}{\omega}{m}$ varies with $\LGp$.
Due to this it is also of interest to consider some plots that can
reveal how $\hatlgcsdM[p]{k\ell}{\LGp}{\omega}{m}$ behaves as a
function of the selected point $\LGp$.

\Cref{fig:heatmap_co_quad_dmt_bivariate_constant_phases,fig:heatmap_co_quad_dmt_bivariate_different_phases,fig:heatmap_distance_plot_EuStockMarkets_DAX_CAC}
of the main part used heatmap- and distance-plots to show how
$\hatlgcsdM[p]{k\ell}{\LGp}{\omega}{m}$ varied with the selected point
$\LGp$ (restricted to the diagonal), in particular
\cref{fig:heatmap_co_quad_dmt_bivariate_constant_phases,fig:heatmap_co_quad_dmt_bivariate_different_phases}
considered the \textit{bivariate local trigonometrical} examples used
for the sanity testing of the implemented estimation algorithm,
whereas \cref{fig:heatmap_distance_plot_EuStockMarkets_DAX_CAC} looked
at the DAX- and CAC-components of the \EuStockMarkets-data.
These plots are extensions to the multivariate case of the
corresponding heatmap- and distance-plots that was introduced for the
univariate case in \JT.

\subsubsection{Revisiting the univariate case}
\label{P2.app:Point_sensitivity.Revisiting.the.univariate.case}

For a local Gaussian investigation of a multivariate sample, like the
\EuStockMarkets-data, it is of course important to investigate both
the local Gaussian cross-spectra
$\hatlgcsdM[p]{k\ell}{\LGp}{\omega}{m}$ and the marginal local
Gaussian auto-spectra $\hatlgcsdM[p]{kk}{\LGp}{\omega}{m}$.
\Cref{fig:heatmap_distance_plot_EuStockMarkets_DAX_CAC} from the
main-part should thus be accompanied with the marginal heatmap- and
distance-plots of the DAX- and CAC-components of \EuStockMarkets, as
seen in \cref{fig:EuStockMarkets_marginals_DAX_CAC}.

\begin{figure}
  {\centering
    \includegraphics[width=1\linewidth]{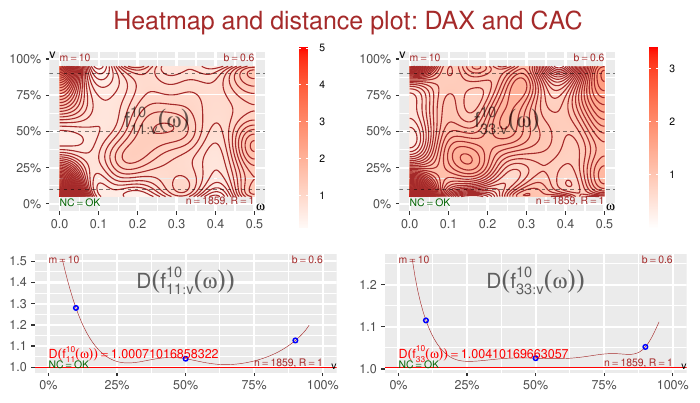}}
  \caption{Heatmap- and distance-plots for the marginal DAX- and
    CAC-components of the \EuStockMarkets-data, the
    DAX-component to the left.  Investigation with focus on how the
    local Gaussian auto-spectra $\hatlgcsdM{kk}{\LGp}{\omega}{10}$
    changes when the point $\LGp$ varies along the diagonal. The three
    points used in \cref{fig:EuStockMarkets} have been highlighted
    with lines/points.}
  \label{fig:EuStockMarkets_marginals_DAX_CAC}
\end{figure}

Note that the points $\LGp=\LGpoint$ seen in
\cref{fig:EuStockMarkets_marginals_DAX_CAC} only varies along the
diagonal, i.e.,\ $\LGpi{1} = \LGpi{2}$.  This was used in
\JT since it simplified the graphical inspection
(the local Gaussian auto-spectra are always real-valued on the
diagonal).  The points $\LGp=\LGpoint$ can of course also vary along
non-diagonal lines, but the resulting visualisation for the local
Gaussian auto-spectra must then deal with complex-valued results
instead of real-valued results.

\Cref{fig:EuStockMarkets_marginals_DAX_CAC} shows how heatmap-plots
can be used to see how the local Gaussian auto-spectra
$\hatlgcsdM[p]{kk}{\LGp}{\omega}{m}$ varies with diagonal points
$\LGp$ (for a fixed bandwidth $\bm{b}$ and a fixed truncation
level~$m$).  The distance function $D$ from
\cref{P2.def:sensitivity_analysis_distance_function} has been used to
create the distance-based plots, and these show how the norms
$D\!\parenR{\hatlgcsdM[p]{kk}{\LGp}{\omega}{m}}$ varies with $\LGp$.
Note that the scales are different for the DAX- and CAC-components in
\cref{fig:EuStockMarkets_marginals_DAX_CAC}, with the DAX-component
(left side) having the highest values.

The distance-parts of \cref{fig:EuStockMarkets_marginals_DAX_CAC} also
includes the distance-based value computed from the ordinary
($m$-truncated) autospectra (red horizontal line).  A comparison of
the distances from the ordinary spectra and the local Gaussian spectra
can provide an indication of the presence of non-Gaussian dependency
structures in the time series under investigation.
Note that a comparison of the global spectra and the local Gaussian
spectra is not possible based on the heatmaps seen in
\cref{fig:EuStockMarkets_marginals_DAX_CAC} (since they only reveal
information based on the local Gaussian spectra), and it is thus also
necessary to consider plots like those seen in
\cref{fig:Bivariate_global_cosine,fig:Bivariate_local_trigonometric_A,fig:Bivariate_local_trigonometric_C,fig:EuStockMarkets,fig:cGARCH}
in the main part.

The points $\LGp$ in \cref{fig:EuStockMarkets_marginals_DAX_CAC}
ranges from the 5\% percentile to the 95\% percentile of the standard
normal distribution, increasing in steps of 0.5\% (altogether 91
different points).  This percentile based selection implies that the
corresponding points are not equally spaced along the actual diagonal,
and the plots in \cref{fig:EuStockMarkets_marginals_DAX_CAC} (and all
similar $\LGp$-investigating plots) have thus used the option that the
points $\LGp$ have been presented according to their underlying
percentile-values --- which implies that these plots primarily reveals
information about the copula-structure of the time series under
investigation.

The features seen for the DAX- and CAC-parts of
\cref{fig:EuStockMarkets_marginals_DAX_CAC} are similar to those
observed for the
\texttt{dmbp}-example\footnote{\label{P2.footnote:dmbp} The
  Deutschemark/British pound Exchange Rate (\texttt{dmbp}) data from
  \citet{bollerslev96:_period}, which is a common benchmark data set
  for GARCH-type models.  The data used in \JT
  was found in the \Rpackage \Rref{rugarch}, see
  \citet{ghalanos15:_rugarch}, where the following description was
  given: \enquote{The daily percentage nominal returns computed as
    \mbox{$100\left[\ln\left(P_t\right) -
        \ln\left(P_t-1\right)\right]$}, where $P_t$ is the bilateral
    Deutschemark/British pound rate constructed from the corresponding
    U.S.\ dollar rates.}  } investigated in \JT.
In particular, the local Gaussian auto-spectra
$\hatlgcsdM[p]{kk}{\LGp}{\omega}{m}$ are rather flat near near the
50\% percentile, and in this region the spectra are quite similar to
those encountered from a sample from an i.i.d.\ white noise situation.
Moreover, it is also here seen (as it was for the \texttt{dmbp}-case)
that there is a clear asymmetric behaviour between the lower tail and
upper tail, and this asymmetry is easier to see from the
distance-parts of
\cref{fig:EuStockMarkets_marginals_DAX_CAC}.\footnote{%
  It can also be seen that the DAX-component (left side of
  \cref{fig:EuStockMarkets_marginals_DAX_CAC}) have a slightly more
  extreme behaviour in the lower tail than the one observed for the
  CAC-component (right side).}

The observed asymmetry, with a higher peak at the lower tail, are in
agreement with the asymmetry between a \textit{bear market} (going
down) and a \textit{bull market} (going up).

It must be added that the 5\% and 95\% percentiles are quite far out
in the tails of the distribution, and it is thus natural to assume
that the selected bandwidth in those cases might fail to work properly
--- and the small sample variation of the points closest to the point
$\LGp$ might then render the estimated local Gaussian autocorrelations
rather dubious.  It is possible to counter this problem by selecting a
larger bandwidth for percentiles in the tails, but it is then
important to keep in mind that a too large bandwidth might completely
miss the desired local structure at the point of investigation.

\subsubsection{Extension to the multivariate case}
\label{P2.app:Extension_to_the_multivariate_case}

The restriction to diagonal points, as seen in
\cref{fig:EuStockMarkets_marginals_DAX_CAC}, ensures that the
resulting local Gaussian auto-spectra
$\hatlgcsdM[p]{kk}{\LGp}{\omega}{m}$ always are real-valued.  The
local Gaussian cross-spectra $\hatlgcsdM[p]{k\ell}{\LGp}{\omega}{m}$
will however be complex-valued, and some modifications of the plotting
procedure are thus needed when these plots are extended to the
multivariate case.

The distance function $D$ from
\cref{P2.def:sensitivity_analysis_distance_function} works equally
well for real-valued and complex-valued spectra, so the distance-part
of the plots remains exactly the same.\footnote{%
  It is worth noticing that the lowest value that can occur for the
  distance is 1 when an auto-spectrum is investigated, whereas it is 0
  when a cross-spectrum is investigated, i.e.,\
  $D\left(\hatlgcsdM[p]{kk}{\LGp}{\omega}{m}\right)\geq 1$ and
  $D\left(\hatlgcsdM[p]{k\ell}{\LGp}{\omega}{m}\right)\geq 0$ when
  $k\neq \ell$.  The values seen in the distance-plot thus gives an
  investigator an idea with regard to how far away from an i.i.d.\
  white noise situation the resulting auto- and cross-spectra are.} %
The heatmap-part of the plots must however be updated, and a natural
strategy for this endeavour is to split the complex-valued local
Gaussian cross-spectra $\hatlgcsdM[p]{k\ell}{\LGp}{\omega}{m}$ into a
combination of the related real-valued spectra given in
\cref{def:local_co_quad_amplitude_phase} in the main part

A complex number $z$ can be represented in a Cartesian form, i.e.,\
$z=x+i\cdot y$, where $x$ and $y$ respectively are the real and
imaginary parts of the number. From the definition of the \Co-spectrum
$\lgcsdCo[p]{k\ell}{\LGp}{\omega}$ and \Quad-spectrum
$\lgcsdQuad[p]{k\ell}{\LGp}{\omega}$, it is clear that
$\hatlgcsdM[p]{k\ell}{\LGp}{\omega}{m} =
\lgcsdCo[p]{k\ell}{\LGp}{\omega} -
i\cdot\lgcsdQuad[p]{k\ell}{\LGp}{\omega}$.  This implies that one way
to adjust the heatmap-plot to the bivariate case is to present a pair
of plots, where one plot represent the \Co-spectrum and the other
represent the \Quad-spectrum.  This approach was used in
\cref{fig:heatmap_co_quad_dmt_bivariate_constant_phases,fig:heatmap_co_quad_dmt_bivariate_different_phases,fig:heatmap_distance_plot_EuStockMarkets_DAX_CAC}
in the main part.

The complex number $z$ can also be represented in a polar form, i.e.,\
$z = \alpha \cdot e^{i\cdot \phi}$, where $\alpha$ and $\phi$
respectively are the modulus and the phase of the complex number.
From the definition of the \Amplitude-spectrum
$\lgcsdAmplitude[p]{k\ell}{\LGp}{\omega}$ and \Phase-spectrum
$\lgcsdPhase[p]{k\ell}{\LGp}{\omega}$, it is clear that
$\hatlgcsdM[p]{k\ell}{\LGp}{\omega}{m} =
\lgcsdAmplitude[p]{k\ell}{\LGp}{\omega} \cdot e^{
  i\cdot\lgcsdPhase[p]{k\ell}{\LGp}{\omega}}$.  This implies that
another possible extension of the heatmap-plot to the bivariate case
is to use a pair of plots, where one plot represent the
\Amplitude-spectrum and the other plot represent the \Phase-spectrum.

\Cref{fig:EuStockMarkets_diagonal_points_LS_c_amplitude} shows a
heatmap- and distance plot for the local Gaussian cross-spectrum for
the DAX- and CAC-components of the \EuStockMarkets-data, where
the heatmap-part use the \Amplitude + \Phase decomposition.  A
comparison with
\cref{fig:heatmap_distance_plot_EuStockMarkets_DAX_CAC} from the main
part reveals that the information gained from the \Amplitude-spectrum
part of \cref{fig:EuStockMarkets_diagonal_points_LS_c_amplitude} in
this case is similar to the one gained from the \Co-spectrum part of
\cref{fig:heatmap_distance_plot_EuStockMarkets_DAX_CAC}.  A similar
comparison of the \Quad-spectrum and the \Phase-spectrum from these
reveals that these also are quite similar (at least after taking into
account that the definition of the \Quad-spectrum is minus one times
the imaginary part of the cross-spectrum explains why positive and
negative regions have been inverted).

\begin{figure}
  {\centering
    \includegraphics[width=1\linewidth]{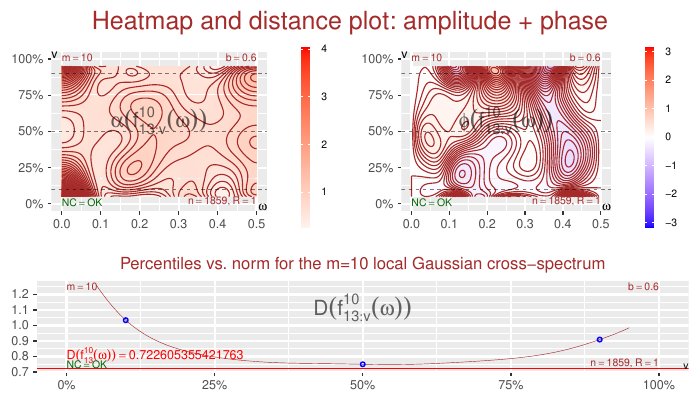}}
  \caption{Heatmap and corresponding distance-based plots based on the
    DAX- and CAC-components of the \EuStockMarkets-data,
    showing how the complex-valued local Gaussian cross-spectrum
    $\hatlgcsdM{k\ell}{\LGp}{\omega}{10}$ changes when when the point
    $\LGp$ varies along the diagonal.  Investigation based on the
    \Amplitude- and \Phase-spectra, cf.\
    \cref{fig:heatmap_distance_plot_EuStockMarkets_DAX_CAC} in the
    main part for the \Co- and \Quad-spectra version.  The three
    points used in \cref{fig:EuStockMarkets} have been highlighted
    with lines/points.}
  \label{fig:EuStockMarkets_diagonal_points_LS_c_amplitude}
\end{figure}

Based on the details seen in
\Cref{fig:EuStockMarkets_diagonal_points_LS_c_amplitude,fig:heatmap_distance_plot_EuStockMarkets_DAX_CAC},
it might be tempting to conclude that it hardly matters what kind of
decomposition (Cartesian or polar) that is used for the inspection of
the complex-valued local Gaussian spectra, but it should be noted that
the \Quad-spectrum in this particular case have much smaller values
than the \Co-spectrum, and that could be the source of the similarity.

The two cases investigated in
\cref{fig:heatmap_co_quad_dmt_bivariate_constant_phases,fig:heatmap_co_quad_dmt_bivariate_different_phases},
i.e.,\ the \textit{bivariate local trigonometrical} examples used for
the sanity testing of the implemented estimation algorithm, can also
be investigated by the help of a polar composition of the
complex-valued local Gaussian cross-spectrum, and the resulting plots
are presented in
\cref{fig:bivariate_constant_phases_Heatmap1,fig:bivariate_constant_phases_Heatmap2}.

\begin{figure}
  {\centering
    \includegraphics[width=1\linewidth]{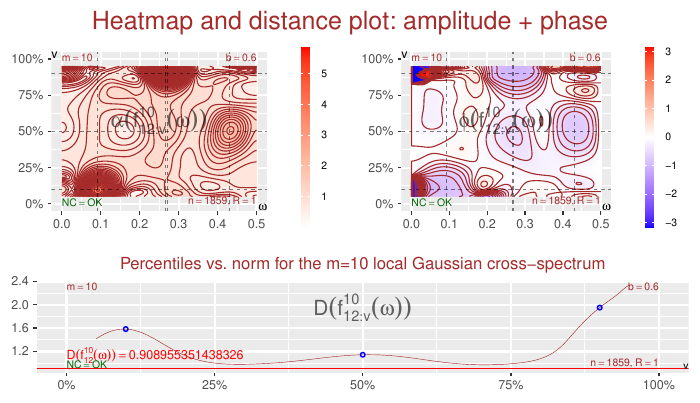}}
  \caption{Heatmap and corresponding distance-based plots based on the
    \textit{local trigonometric constant phase-adjustment} example,
    showing how the complex-valued local Gaussian cross-spectrum
    $\hatlgcsdM{k\ell}{\LGp}{\omega}{10}$ changes when when the point
    $\LGp$ varies along the diagonal.  Investigation based on the
    \Amplitude- and \Phase-spectra, cf.\
    \cref{fig:heatmap_co_quad_dmt_bivariate_constant_phases} in the
    main part for the \Co- and \Quad-spectra version.  The three
    points used in \cref{fig:Bivariate_local_trigonometric_A} have
    been highlighted with lines/points.    Vertical lines have been
    added to the heatmap-plots to highlight the frequencies used to
    generate the \textit{local trigonometric example}. }
  \label{fig:bivariate_constant_phases_Heatmap1}
\end{figure}

An inspection of the \Co- and \Quad- spectra in
\cref{fig:heatmap_co_quad_dmt_bivariate_constant_phases,fig:heatmap_co_quad_dmt_bivariate_different_phases}
clearly reveals that the two \textit{bivariate local trigonometrical}
examples have different local properties, i.e.,\ the peaks and troughs
occur at different frequencies for different points $\LGp$.  The
distance-parts of
\cref{fig:heatmap_co_quad_dmt_bivariate_constant_phases,fig:heatmap_co_quad_dmt_bivariate_different_phases}
did, however, not show this difference --- and this highlights why it
is important to include a visualisation of the frequency component.

As seen in the previous paragraph, the Cartesian-type decomposition of
the local Gaussian cross-spectrum, i.e.,\
$\hatlgcsdM[p]{k\ell}{\LGp}{\omega}{m} =
\lgcsdCo[p]{k\ell}{\LGp}{\omega} -
i\cdot\lgcsdQuad[p]{k\ell}{\LGp}{\omega}$, made it easy to spot the
differences between the two \textit{bivariate local trigonometrical}
examples.  What about the strategy based on the polar-type
decomposition, i.e.,\
$\hatlgcsdM[p]{k\ell}{\LGp}{\omega}{m} =
\lgcsdAmplitude[p]{k\ell}{\LGp}{\omega} \cdot e^{
  i\cdot\lgcsdPhase[p]{k\ell}{\LGp}{\omega}}$?

\begin{figure}
  {\centering
    \includegraphics[width=1\linewidth]{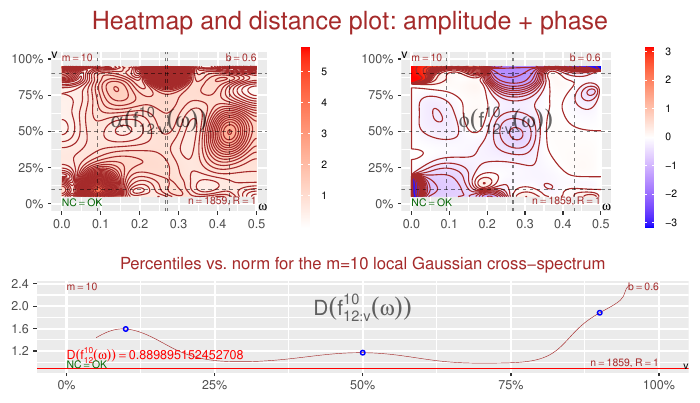}}
  \caption{Heatmap and corresponding distance-based plots based on the
    \textit{local trigonometric different phases-adjustment} example,
    showing how the complex-valued local Gaussian cross-spectrum
    $\hatlgcsdM{k\ell}{\LGp}{\omega}{10}$ changes when when the point
    $\LGp$ varies along the diagonal.  Investigation based on the
    \Amplitude- and \Phase-spectra, cf.\
    \cref{fig:heatmap_co_quad_dmt_bivariate_different_phases} in the
    main part for the \Co- and \Quad-spectra version.  The three
    points used in \cref{fig:Bivariate_local_trigonometric_C} have
    been highlighted with lines/points.  Vertical lines have been
    added to the heatmap-plots to highlight the frequencies used to
    generate the \textit{local trigonometric example}. }
  \label{fig:bivariate_constant_phases_Heatmap2}
\end{figure}

The complex-valued data used for the construction of
\cref{fig:heatmap_co_quad_dmt_bivariate_constant_phases,fig:heatmap_co_quad_dmt_bivariate_different_phases}
was also used for the construction of
\cref{fig:bivariate_constant_phases_Heatmap1,fig:bivariate_constant_phases_Heatmap2},
but it is quite a bit harder to use the latter pair of plots to spot
the differences between the two \textit{bivariate local
  trigonometrical} examples.  In fact the \Amplitude-spectra parts of
\cref{fig:bivariate_constant_phases_Heatmap1,fig:bivariate_constant_phases_Heatmap2}
are so similar that it is hard to tell them apart.  The \Phase-spectra
parts are similarly quite hard to investigate, in particular since the
values in this case are restricted to the range $(-\pi,\pi]$. Unless
the \Phase-spectra values have opposite signs, it is very hard to
compare the colour-gradients in the \Phase-spectra parts of
\cref{fig:bivariate_constant_phases_Heatmap1,fig:bivariate_constant_phases_Heatmap2}
and conclude that there are significant differences between the two
\textit{bivariate local trigonometrical} examples that are
investigated in these plots.

The preceding discussion shows that it is easier to digest the
information stored in the complex-valued local Gaussian cross-spectrum
when the heatmap-parts are decomposed into \Co- and \Quad-spectra ---
and this motivated the use of this version for the heatmap- and
distance-plots included in the main part, i.e.,\
\cref{fig:heatmap_co_quad_dmt_bivariate_constant_phases,fig:heatmap_co_quad_dmt_bivariate_different_phases,fig:heatmap_distance_plot_EuStockMarkets_DAX_CAC}.

\subsubsection{Heatmap-plots for the estimates $\hatlgccr{k\ell}{\LGp}{h}$}

The computation of the complex-valued local Gaussian cross-spectra
$\hatlgcsdM[p]{k\ell}{\LGp}{\omega}{m}$ requires the computation of
the $2m+1$ local Gaussian cross-correlations
$\TSR{\hatlgccr{k\ell}{\LGp}{h}}{h=-m}{m}$, and it is thus (as a
supplement to the plots seen in
\cref{fig:heatmap_co_quad_dmt_bivariate_constant_phases,fig:heatmap_co_quad_dmt_bivariate_different_phases,fig:heatmap_distance_plot_EuStockMarkets_DAX_CAC}
in the main part and
\cref{fig:EuStockMarkets_diagonal_points_LS_c_amplitude,fig:bivariate_constant_phases_Heatmap1,fig:bivariate_constant_phases_Heatmap2}
in the preceding discussion) also possible to create heatmap-based
plots that can visualise how these estimates change as the point
$\LGp$ moves along the diagonal from the 5\% to the 95\% percentile.
\Cref{fig:EuStockMarkets_diagonal_points_P1_fig_D3} presents a
heatmap-plot for the local Gaussian cross-correlation between the DAX-
and CAC- components of the \EuStockMarkets-data.

\begin{figure}
  {\centering
    \includegraphics[width=1\linewidth]{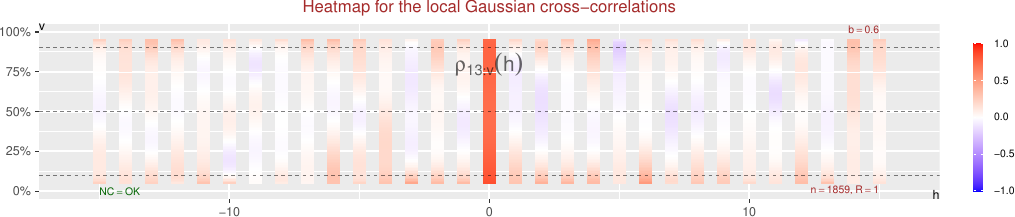}}
  \caption{Heatmap of the local Gaussian cross-correlations
    $\hatlgccr[p]{k\ell}{\LGp}{h}$ based on the DAX- and
    CAC-components of the \EuStockMarkets-data, showing how
    these estimates change when the point $\LGp$ varies along the
    diagonal.  The three points used in \cref{fig:EuStockMarkets} have
    been highlighted with lines.}
  \label{fig:EuStockMarkets_diagonal_points_P1_fig_D3}
\end{figure}

Note that
the values $\hatlgccr{k\ell}{\LGp}{h}$ by construction lies in the
interval $[-1,1]$, and the colour-gradient has thus been predefined to
use this range.  It is quite easy to see from
\cref{fig:EuStockMarkets_diagonal_points_P1_fig_D3} if the
$\hatlgccr{k\ell}{\LGp}{h}$-values are close to zero (white) or if
they are closer to the endpoints~$\pm 1$.  It can from such an
investigation also be seen if some lags behave differently from the
others, and it is also possible to see if it might be an asymmetric
situation between the behaviour in lower and upper tails.  Finally, it
can also help reveal if the estimation of the values in the tails
might be somewhat dubious, since that often will reveal itself by
values that becomes quite close to~$\pm 1$, or suddenly changes from
one extreme to the other.

\subsection{Sensitivity analysis: The bandwidth $\bm{b}$}
\label{P2.app:Bandwidth_sensitivity}

The bandwidth $\bm{b}=\parenR{\bz{1},\bz{2}}$ is a bivariate
parameter, but it was noted in \JT that it is
quite natural to assume $\bz{1}=\bz{2}$ for univariate time series
$\TSR{\Yz{t}}{t=1}{n}$.  This is due to the local Gaussian
autocorrelations $\lgccr[p]{kk}{\LGp}{h}$ being computed based on
(pseudo-normalised versions of the) lag $h$ pairs
$\parenR{\Yz{t+h},\Yz{t}}$, and the marginal time series in this case
thus coincide with the one we started with.

The justification from the univariate case for $\bz{1}=\bz{2}$ does
not work when the local Gaussian cross-correlations
$\lgccr[p]{k\ell}{\LGp}{h}$ are to be estimated for a multivariate
time series, but (as mentioned in
\cref{sec:lgcsd_Examples_the_parameters} in the main part) in this
case it is possible to use the pseudo-normalisation of the marginals
as a justification for the assumed equality of $\bz{1}$ and~$\bz{2}$.

With this restriction it follows that the sensitivity of the
complex-valued local Gaussian cross-spectrum
$\hatlgcsdM{k\ell}{\LGp}{\omega}{m}$, due to changes in the bandwidth
$\bm{b}$, can be investigated in a similar manner to the one used in
the preceding section for the diagonal points $\LGp$.
\Cref{fig:EuStockMarkets_bandwidth_P1_fig_D5__LS_c_Co_II,fig:EuStockMarkets_bandwidth_HCD_Amplitude+Phase}
show the results of such an investigation (at the point
\texttt{10\%::10\%}) for the complex-valued local Gaussian
cross-spectrum $\hatlgcsdM{k\ell}{\LGp}{\omega}{m}$ based on the DAX-
and CAC-components of the \EuStockMarkets-data.
\Cref{fig:EuStockMarkets_bandwidth_P1_fig_D5__LS_c_Co_II} gives a
Cartesian decomposition into the \Co- and \Quad-spectra, whereas
\cref{fig:EuStockMarkets_bandwidth_HCD_Amplitude+Phase} presents the
polar decomposition into the \Amplitude- and \Phase-spectra.  As
mentioned before: Since the scale of the imaginary part of
$\hatlgcsdM{k\ell}{\LGp}{\omega}{m}$ is much smaller than the real
part in this particular case, the \Amplitude-spectrum is in essence
the absolute value of the \Co-spectrum --- and since the \Co-spectrum
is positive for this example it follows that the plots of the
\Amplitude-spectrum is very similar to the plot of the \Co-spectrum in
this particular case.

\begin{figure}
  {\centering
    \includegraphics[width=1\linewidth]{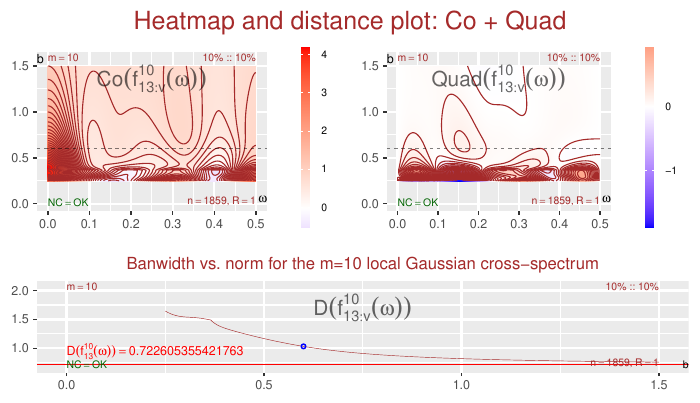}}
  \caption{Heatmap and corresponding distance-based plots based on the
    DAX- and CAC-components of the \EuStockMarkets-data,
    showing how the complex-valued local Gaussian cross-spectrum
    $\hatlgcsdM{k\ell}{\LGp}{\omega}{10}$ changes when when the
    bandwidth $\bm{b}$ varies.  Investigation based on \Co- and
    \Quad-spectra, cf.\
    \cref{fig:EuStockMarkets_bandwidth_HCD_Amplitude+Phase} for the
    the \Amplitude- and \Phase-spectra version.  The bandwidth $b=0.6$
    used in \cref{fig:EuStockMarkets} has been highlighted with a
    line/point.}
  \label{fig:EuStockMarkets_bandwidth_P1_fig_D5__LS_c_Co_II}
\end{figure}

\begin{figure}
  {\centering
    \includegraphics[width=1\linewidth]{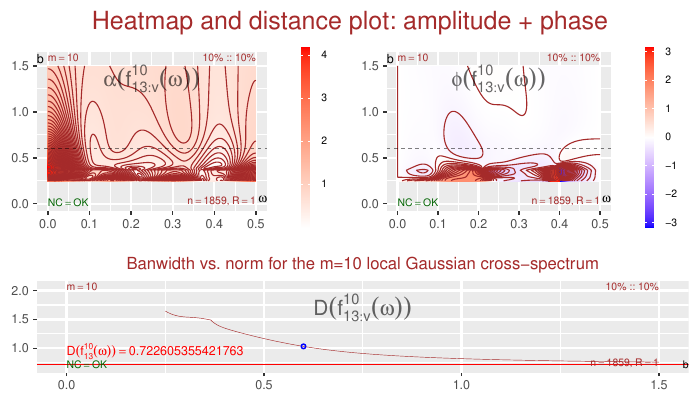}}
  \caption{Heatmap and corresponding distance-based plots based on the
    DAX- and CAC-components of the \EuStockMarkets-data,
    showing how the complex-valued local Gaussian cross-spectrum
    $\hatlgcsdM{k\ell}{\LGp}{\omega}{10}$ changes when when the
    bandwidth $\bm{b}$ varies.  Investigation based on the \Amplitude-
    and \Phase-spectra, cf.\
    \cref{fig:EuStockMarkets_bandwidth_P1_fig_D5__LS_c_Co_II} for the
    \Co- and \Quad-spectra version.  The bandwidth $b=0.6$ used in
    \cref{fig:EuStockMarkets} has been highlighted with a line/point.}
  \label{fig:EuStockMarkets_bandwidth_HCD_Amplitude+Phase}
\end{figure}

As mentioned in \lgsdRef{app:Bandwidth_sensitivity}: \enquote{The
  bandwidth $\bm{b}$ should be selected according to the Goldilocks
  principle, i.e.\ it should neither be \enquote{too low} nor
  \enquote{too high}, it must be \enquote{just right}.}  %
An inspection of the heatmap and distance-based plots in
\cref{fig:EuStockMarkets_bandwidth_P1_fig_D5__LS_c_Co_II,fig:EuStockMarkets_bandwidth_HCD_Amplitude+Phase}
can provide some information with regard to what kind of problems that
can occur when the bandwidth does not belong to the \enquote{just
  right} region.  The bandwidths used in
\cref{fig:EuStockMarkets_bandwidth_P1_fig_D5__LS_c_Co_II,fig:EuStockMarkets_bandwidth_HCD_Amplitude+Phase}
starts at 0.25 and ends at 1.5 (increasing in steps of 0.005,
altogether 251 different bandwidths), and the bandwidth $b=0.6$ has
been highlighted since it was that value that was used in
\cref{fig:EuStockMarkets_lags200,fig:heatmap_distance_plot_EuStockMarkets_DAX_CAC,fig:EuStockMarkets}
in the main part.

The problems occurring for the multivariate case are the same as those
discussed in \lgsdRef{app:Bandwidth_sensitivity} for the univariate
case.  If $\bm{b}$ becomes too large, then the estimated local
Gaussian cross-correlations $\hatlgccr{k\ell}{\LGp}{h}$ will no longer
capture the local structure of interest, and the corresponding
estimated local Gaussian spectral density
$\hatlgcsdM{k\ell}{\LGp}{\omega}{m}$ (which no longer deserves to be
referred to as \enquote{local}) will then be indistinguishable from
the ordinary spectral density.  It is clear from
\cref{fig:EuStockMarkets_bandwidth_P1_fig_D5__LS_c_Co_II,fig:EuStockMarkets_bandwidth_HCD_Amplitude+Phase}
that a bandwidth greater than $b=1.0$ will be too large for the
investigation of the \EuStockMarkets-data.

On the other side, it is expected that a too low bandwidth will
trigger a degeneration of the estimated local Gaussian
cross-correlations, i.e.,\ $\hatlgccr[p]{k\ell}{\LGp}{h}$ will tend
towards either $+1$ or $-1$ regardless of the actual structure of the
underlying density distributions.  The reason for this is (as
previously explained in \JT for the local
Gaussian auto-correlation case) that $\hatlgccr[p]{k\ell}{\LGp}{h}$
will, due to the kernel function from the density estimation
algorithm, become increasingly sensitive to the position of the
(pseudo-normalised) $h$-lagged pairs $\parenR{\Yz{k,t+h},\Yz{\ell,t}}$
that lies nearest to the point $\LGp=\LGpoint$.  To clarify, for a
given point $\LGp$ there will be a collection of Euclidean distances
to the (pseudo-normalised versions of the) $h$-lagged pairs
$\parenR{\Yz{k,t+h},\Yz{\ell,t}}$ in the sample, and these distances
could (after a re-indexing) be sorted in ascending order
$\TSR{\dz{i}}{i=1}{n-h}$.

Under the assumption that it is the product normal kernel that is
used, the contribution from a lag-$h$ pair
$\parenR{\Yz{k,t+h},\Yz{\ell,t}}$ that lies a distance of $\dz{i}$
from $\LGp$ will be weighted by
$\wz{i:\bm{b}} \defeq \tfrac{1}{2\pi \bz[2]{}}\ez[-d_i^2/2b^2]{}$ ---
and it is now natural to consider the set of all the weights
$\mathcalWz{\LGp:\bm{b}}\defeq\TSR{\wz{i:\bm{b}}}{i=1}{n-h}$.

The primary detail of interest is how much larger the weights are for
the pairs that are closest to $\LGp$, and it is thus necessary to
consider the fraction
$\rz{ij:\bm{b}}\defeq\wz{i:\bm{b}}/\wz{j:\bm{b}} =
\parenRz[1/b^2]{}{\ez[d_j^2-d_i^2]{}}$.  The number $\rz{ij:\bm{b}}$
will, when $\dz{i}<\dz{j}$, grow to $\infty$ when $\blimit$, and this
implies that the estimation algorithm for small $b$-values will become
increasingly sensitive to the $h$-lagged pairs that are  closest to
the point $\LGp$ when the bandwidth shrinks --- and in the end it
would thus be natural to have a degeneration of the estimated value
$\lgacr{\LGp}{h}$ to either $+1$ or $-1$.

Note that the corresponding
$D\!\parenR{\hatlgcsdM[p]{k\ell}{\LGp}{\omega}{m}}$ will grow when
this degeneration happens, as can be seen for $b=0.25$ in the
distance-part plots in
\cref{fig:EuStockMarkets_bandwidth_P1_fig_D5__LS_c_Co_II,fig:EuStockMarkets_bandwidth_HCD_Amplitude+Phase}.

\textbf{Heatmap-plots for the estimates $\hatlgccr{k\ell}{\LGp}{h}$:}
It is here, as it was for the investigation of the diagonal points
$\LGp$, possible to also consider a heatmap-based investigation of the
underlying estimates $\TSR{\hatlgccr{k\ell}{\LGp}{h}}{h=-m}{m}$.  Such
a plot is given in \cref{fig:EuStockMarkets_bandwidth_P1_fig_D6}, and
it can there be observed how some of the
$\hatlgccr{k\ell}{\LGp}{h}$-estimates changes from having a value
close to zero to a value that must be close to $+$ or $-1$ as the
bandwidth $\bm{b}$ decreases.  This is, as mentioned above, not a
surprising observation --- since this behaviour is expected when the
bandwidth $\bm{b}$ has shrunk to a level where the kernel function in
the local penalty function gives quite high weights to the few
observations $\parenR{\Yz{k,t+h},\Yz{\ell,t}}$ nearest $\LGp$, and
very low weights elsewhere.

\begin{figure}
  {\centering
    \includegraphics[width=1\linewidth]{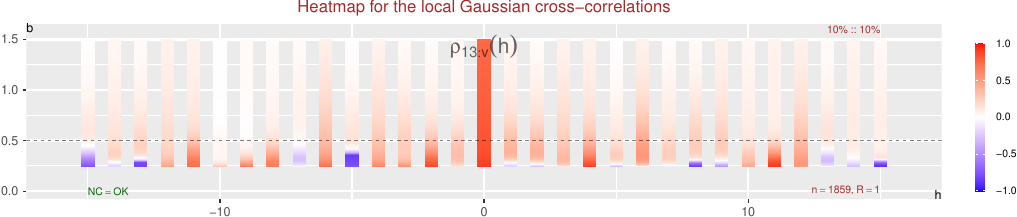}}
  \caption{Heatmap of the local Gaussian cross-correlations
    $\hatlgccr[p]{k\ell}{\LGp}{h}$ based on the DAX- and
    CAC-components of the \EuStockMarkets-data, showing how
    these estimates change when the bandwidth $\bm{b}$ varies.  The
    bandwidth $b=0.6$ used in \cref{fig:EuStockMarkets} has been
    highlighted with a line.}
  \label{fig:EuStockMarkets_bandwidth_P1_fig_D6}
\end{figure}

Note that
\cref{fig:EuStockMarkets_bandwidth_P1_fig_D5__LS_c_Co_II,fig:EuStockMarkets_bandwidth_HCD_Amplitude+Phase,fig:EuStockMarkets_bandwidth_P1_fig_D6}
consider the situation where $\LGp$ is the \enquote{lower tail}
diagonal point at \texttt{10\%::10\%}.  Similar heatmap- and
distance-plots could have been included for the cases where $\LGp$
corresponds to either the center or the upper tail.\footnote{The
  interested reader can use the scripts in the \Rpackage
  \lgsdRpackage\ to get access to these plots for the center and upper
  tail, cf.\ \cref{P2.app:The.scripts.in.lgsdRpackage} for details.}
A comparison of the distance-based plots for these three points is
presented in \cref{fig:EuStockMarkets_bandwidth_P1_fig_D7}, in order
to show how the bandwidth-sensitivity of
$\hatlgsdM[p]{\LGp}{\omega}{m}$ also depends on the selected point
$\LGp$.  A common scale has been used for the three subplots in order
to emphasise the asymmetry between the lower and upper tail.

\begin{figure}
  {\centering
    \includegraphics[width=1\linewidth]{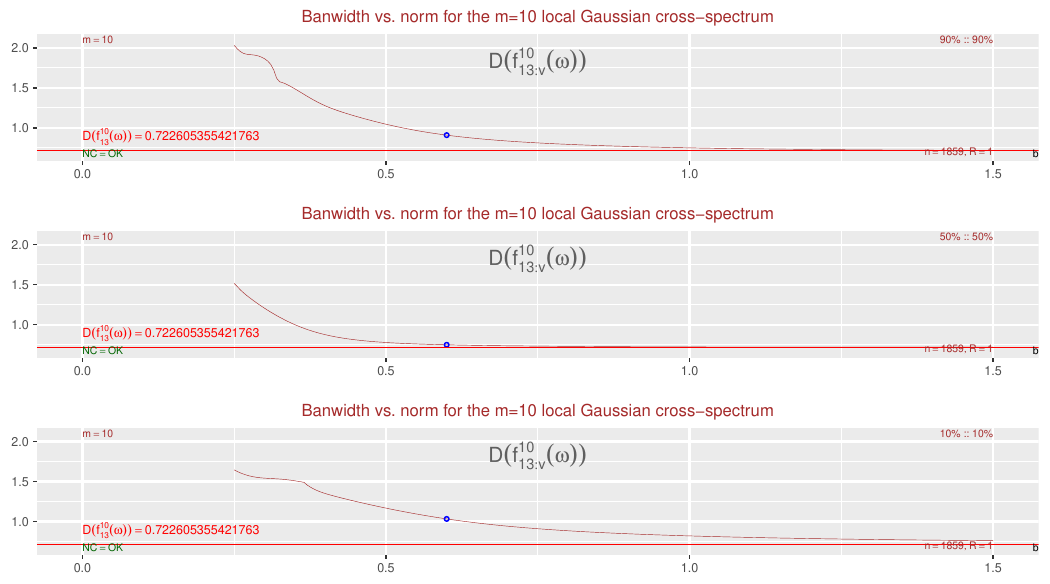}}
  \caption{Distance-based plots based on the DAX- and CAC-components
    of the \EuStockMarkets-data, that shows how the distance
    $D\parenR{\hatlgcsdM{k\ell}{\LGp}{\omega}{10}}$ varies with the
    bandwidth $\bm{b}$ for the three percentiles used in
    \cref{fig:EuStockMarkets}, i.e.,\ 10\%, 50\% and 90\%.  The
    bandwidth $b=0.6$ used in \cref{fig:EuStockMarkets} has been
    highlighted with a point.}
  \label{fig:EuStockMarkets_bandwidth_P1_fig_D7}
\end{figure}

The center plot of \cref{fig:EuStockMarkets_bandwidth_P1_fig_D7}
reveals that the \enquote{too low bandwidth problem} occurs a bit
slower in a high density region, but it will even there eventually
create a situation where the estimated local Gaussian autocorrelations
degenerate towards either $+1$ or $-1$.  This is the exact same
behaviour that was observed in \JT for the local
Gaussian auto-spectra $\lgcsdM[p]{kk}{\LGp}{\omega}{m}$, and this is
as expected since the higher concentration of (pseudo-normalised)
points at the center implies that the degeneration of the estimates
first will occur at shorter bandwidths.

The heatmap and distance-based plots in
\cref{fig:EuStockMarkets_bandwidth_P1_fig_D5__LS_c_Co_II,fig:EuStockMarkets_bandwidth_HCD_Amplitude+Phase,fig:EuStockMarkets_bandwidth_P1_fig_D6,fig:EuStockMarkets_bandwidth_P1_fig_D7}
can detect the clearly undesirable regions for the bandwidth $\bm{b}$,
but they do not reveal what the \enquote{just right} value for the
bandwidth should be.  Nevertheless, it is still possible to gain some
insight into how sensitive the estimate of
$\lgcsdM[p]{k\ell}{\LGp}{\omega}{m}$ will be for minor variations of the
bandwidth $\bm{b}$, and that can be useful with regard to the
selection of a few bandwidths that can be used when e.g.\ a
bootstrap-investigation is to be performed.

The framework used in the \Rpackage \lgsdRpackage\ ensures that it is
trivial to compute and investigate/compare a wide range of bandwidths
simultaneously, and the key idea is that knowledge of the local
dependency structure can still be obtained even if the selected
bandwidths are not spot on the \enquote{just right} value for the
bandwidth.

\subsection{Sensitivity analysis: The truncation level $m$}
\label{P2.app:Truncation_level_sensitivity}

This section is in essence the same as
\lgsdRef{app:Truncation_level_sensitivity}.  The main differences are
primarily (1) that the notation have been updated in order to reflect
that the target of the investigation now is the local Gaussian
cross-spectrum $\lgcsdM{k\ell}{\LGp}{\omega}{m}$ (instead of the local
Gaussian auto-spectrum), and (2) that the sample under investigation
in this case is the multivariate \EuStockMarkets-data instead of the
univariate \texttt{dmbp}-data.

The shape of $\lgcsdM{k\ell}{\LGp}{\omega}{m}$ for a low truncation level
can be different from the shape seen when a higher truncation level
is used.  It is thus of interest to investigate how sensitive the
estimates $\hatlgcsdM{k\ell}{\LGp}{\omega}{m}$ are to changes in the
truncation level $m$.

This issue can easily be probed by performing an initial investigation
with a high value for the maximum lag to be computed, since the
computational cost is not too large when only a single sample (like
the \EuStockMarkets-data) is investigated. It did e.g.\ not
take a long time to estimate $\lgccr{k\ell}{\LGp}{h}$ for
$h=-200,\dotsc,-1,0.1,\dotsc,200$, which was needed for the
construction of \cref{fig:EuStockMarkets_lags200} in the main document
--- and with these estimates it is trivial to compare
$\hatfz[m]{k\ell}(\omega)$ and $\hatlgcsdM{k\ell}{\LGp}{\omega}{m}$
for $m$ up to 200, since the integrated \Rref{shiny}-application of
the \Rpackage \lgsdRpackage\ can animate the changes that occur in the
spectra when $m$ grows from 0 to 200.

The computational costs can become rather large when it is necessary
to find pointwise confidence intervals, since a high number of
replicates then must be investigated with the same configuration of
tuning parameters.  It is then important to figure out a sufficient
truncation level $m$, and restrict the attention to the estimates of
$\lgccr{k\ell}{\LGp}{h}$ for $h=-m,\dotsc,-1,0,1,\dotsc,m$.

A drawback with the \Rref{shiny}-based approach in \lgsdRpackage\ is
that it requires an inspection of many different plots.  It could thus
be of interest to also consider summary-plots that either use the
distance function $D$ from
\cref{P2.def:sensitivity_analysis_distance_function}, or some
heatmap-based alternative visualisation of
$\hatlgcsdM{k\ell}{\LGp}{\omega}{m}$, similar to those used for
$\hatlgccr{k\ell}{\LGp}{h}$ in
\cref{fig:EuStockMarkets_diagonal_points_P1_fig_D3,fig:EuStockMarkets_bandwidth_P1_fig_D6}.

\textbf{Distance plots:} It is possible to investigate the
$m$-sensitivity by distance-based plots, but those plots are less
useful in this case.  One reason for this is that the norms
$D\!\parenR{\hatlgcsdM{k\ell}{\LGp}{\omega}{m}}$ are monotonically
increasing as functions of $m$. This can easily be seen by first
recalling (cf.\ \myref{def:lgsd_estimator}{lgcsd_estimator_folded})
that the estimates $\hatlgcsdM{k\ell}{\LGp}{\omega}{m}$ are given by
$$\hatlgcsdM[p]{k\ell}{\LGp}{\omega}{m} \defeq
\hatlgccr[p]{k\ell}{\LGp}{0} + \sumss{h=1}{m} \lambdazM{h}{m}\cdot
\hatlgccr[p]{\ell k}{\LGpd}{h} \cdot \ez[+2\pi i\omega h]{} +
\sumss{h=1}{m} \lambdazM{h}{m}\cdot \hatlgccr[p]{k\ell}{\LGp}{h} \cdot
\ez[-2\pi i\omega h]{},$$ and then keeping in mind that the lag-window
function $\lambdazM{h}{m}$ satisfies
$\lambdazM{h}{m+1}\geq\lambdazM{h}{m}$.  It follows that
$D\!\parenR{\hatlgcsdM{k\ell}{\LGp}{\omega}{m+1}} \geq
D\!\parenR{\hatlgcsdM{k\ell}{\LGp}{\omega}{m}}$, which does not
provide any useful new information.

Instead of a plot of the norms
$D\!\parenR{\hatlgcsdM{k\ell}{\LGp}{\omega}{m}}$, it is slightly more
interesting to consider a plot that shows
$D\!\parenR{\hatlgcsdM{k\ell}{\LGp}{\omega}{m+1}-\hatlgcsdM{k\ell}{\LGp}{\omega}{m}}$,
i.e.,\ the distances between $\hatlgcsdM{k\ell}{\LGp}{\omega}{m+1}$ and
$\hatlgcsdM{k\ell}{\LGp}{\omega}{m}$ in the Hilbert space of Fourier
series.  This idea is shown in \cref{fig:m_sensitivity_bivariate} for
the three diagonal points and 200 lags that was included in
\cref{fig:EuStockMarkets_lags200}.  Note that
\cref{fig:m_sensitivity_bivariate} takes into account the scaling due
to the lag-window function $\lambdazM{h}{m}$, and as such it does
provide some new information compared to that contained in the plot
showing the estimated local Gaussian cross-correlations.

\begin{figure}
  {\centering
    \includegraphics[width=1\linewidth]{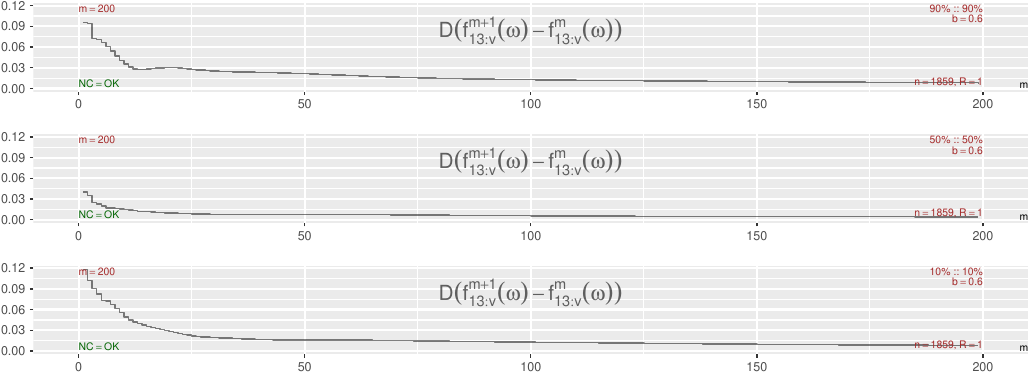}}
  \caption{Distances between successive $m$-truncations of the local
    cross-spectra $\hatlgcsdM{k\ell}{\LGp}{\omega}{m}$ between the
    DAX- and CAC-components of the \EuStockMarkets-data.}
  \label{fig:m_sensitivity_bivariate}
\end{figure}

The three subplots of \cref{fig:m_sensitivity_dmbp_percentages} shows
that
$D\!\parenR{\hatlgcsdM{k\ell}{\LGp}{\omega}{m+1}-\hatlgcsdM{k\ell}{\LGp}{\omega}{m}}$
rather quickly starts to decrease monotonically, which is as expected
given the presence of the lag-window function $\lambdazM{h}{m}$.  This
decrease implies that the effect of a change in the truncation level
from $m$ to $m+1$ becomes smaller as $m$ grows, and the sensitivity is
thus largest when $m$ is small.

\cref{fig:m_sensitivity_dmbp_percentages} might indicate that the
$m=10$ used in the main part is a bit too small.  However, the purpose
of that particular truncation level was simply to show that even a low
truncation level could be used to detect the presence of nonlinear
dependency structures in the time series under investigation, i.e.,\
structures not detected by the ordinary spectrum.

It is natural to assume that two \textit{successive} local Gaussian
spectra $\hatlgcsdM{k\ell}{\LGp}{\omega}{m}$ and
$\hatlgcsdM{k\ell}{\LGp}{\omega}{m+1}$ should be similar in shape when $m$
has grown a bit, but this does not imply that the
\textit{accumulated} changes to $\hatlgcsdM{k\ell}{\LGp}{\omega}{m}$ are
negligible.  It is thus important to also inspect the
frequency-dimension, and this can as mentioned above easily be done
by the interactive \Rref{shiny}-application in the
\lgsdRpackage-package.

\textbf{Heatmap plots:} The truncation level $m$ is a discrete tuning
parameter, and an inspection based on a heatmap-based approach could
thus follow the setup used for the estimated
$\hatlgccr{k\ell}{\LGp}{h}$-values seen in
\cref{fig:EuStockMarkets_diagonal_points_P1_fig_D3,fig:EuStockMarkets_bandwidth_P1_fig_D6}.
The \Rpackage \lgsdRpackage\ contains a script that can be used to
create such a heatmap-based plot for
$\hatlgcsdM{k\ell}{\LGp}{\omega}{m}$, with the frequencies $\omega$
along one axis and the truncation levels $m$ along the other.

The resulting heatmap-based plot clearly showed that the peak seen in
the \Co-spectra part of \cref{fig:EuStockMarkets} at $\omega=0$ (for
$m=10$ and a point either in the lower or upper tail) became even more
dominating as $m$ increased, and the peak dominated to such an extent
that the heatmap-based plot did not reveal anything about the other
frequencies.  This plot has thus not been included here, but the
script is available in \lgsdRpackage, cf.\
\cref{P2.app:The.scripts.in.lgsdRpackage} for details.

\section{How to select the tuning parameters?}
\label{P2.How.to.select.the.tuning.parameters?}
\setcounter{figure}{0} 

The sensitivity analysis in \cref{P2.app:sensitivity_analysis}
considered the effect of minor changes to the tuning parameters
$\bm{b}$ and $m$, and it did also discuss the sensitivity of
$\hatlgcsdM{k\ell}{\LGp}{\omega}{m}$ that is due to the position of
the point $\LGp$ --- which is of interest to know when a given
sample/model is to be investigated.  But what should the tuning
parameters be for an actual investigation?

This topic was discussed in
\lgsdRef{How.to.select.the.tuning.parameters?} for the local Gaussian
auto-spectra $\hatlgcsdM{kk}{\LGp}{\omega}{m}$, and this discussion
extends without any changes to the present multivariate case.  The
proposed strategy has been summarised below.

First of all, the \Rpackage \lgsdRpackage\ can compute
$\hatlgcsdM{k\ell}{\LGp}{\omega}{m}$ for a wide range of tuning
parameters, and for a large selection of different points $\LGp$.  The
computation will then also provide estimates of the univariate
components, and all the relevant plots (of local Gaussian auto- and
cross-correlations, and local Gaussian auto- and cross-spectra) can be
investigated interactively by the \Rref{shiny}-application that is a
part of \lgsdRpackage.

The computational cost for a single estimate
$\hatlgcsdM{k\ell}{\LGp}{\omega}{m}$ is normally not too large
(assuming the truncation level $m$ is not very high), but the total
time can become quite large when a wide range of points $\LGp$ and
bandwidths $\bm{b}$ are to be investigated at the same time.  The
creation of pointwise confidence intervals might on top of this
require a large number of replicates $R$ (or bootstrapped samples).

An investigation can start by looking at a single sample (real data,
or a sample generated by a parametric model).  It is reasonable to
first look at a coarse grid based on a low number of points $\LGp$.
The bandwidths $\bm{b}$ should also be restricted to only a few values
in this part.  The truncation level $m$ can be rather low (range
10-30), since it seems to be the case that interesting features can be
detected even for this range of truncation levels.

If the initial investigation detects the presence of non-Gaussian
dependency structures in the sample, then the next step is to select a
(reasonable) bandwidth $\bm{b}$ and a (not too high) truncation level
$m$, and then create heatmap- and distance-based plots of the
estimates $\hatlgcsdM{k\ell}{\LGp}{\omega}{m}$ for a wide range of
points $\LGp$.

Select a collection of points $\LGp$ that it could be of particular
interest to investigate further, and then investigate e.g.\ $R=100$
replicates (or bootstrap samples) in order to produce pointwise
confidence intervals for the resulting estimates of the local Gaussian
auto- and cross-spectra.  (See the discussion in
\cref{P2.app:regarding_resampling} for details related to the
bootstrapping.)

\textbf{Some comments regarding the bandwidth $\bm{b}$:} %
The discussion about the $\bm{b}$-sensitivity of the estimated local
Gaussian cross-spectra $\hatlgcsdM{k\ell}{\LGp}{\omega}{m}$ in
\cref{P2.app:Bandwidth_sensitivity}, see
\cref{fig:EuStockMarkets_bandwidth_P1_fig_D5__LS_c_Co_II,fig:EuStockMarkets_bandwidth_HCD_Amplitude+Phase,fig:EuStockMarkets_bandwidth_P1_fig_D6,fig:EuStockMarkets_bandwidth_P1_fig_D7},
indicate that the value used for the bandwidth $\bm{b}$ might not be
that critical, unless it is \enquote{too short} or \enquote{too long}.
Insofar the purpose of the investigation is the detection of
non-Gaussian dependency structures, it seems to be the case that a
reasonable range of bandwidths will return the same conclusion.  This
implies that it instead of a \enquote{select the optimal
  bandwidth}-algorithm, should be sufficient to perform an
investigation for a few bandwidths in order to see if any findings are
consistent across those.

The interested reader can find a more in detail discussion related to
the selection of the bandwidth $\bm{b}$ in
\lgsdRef{sec:Some.comments.regarding.the.bandwidth}, including a short
discussion of some papers that have discussed the bandwidth selection
for the estimation of the local Gaussian correlation~$\lgcor{\LGp}$.

\section{Regarding sampling and resampling}
\label{P2.app:regarding_resampling}
\setcounter{figure}{0} 

The section corresponds to \lgsdRef{app:regarding_resampling}, and the
topic of interest is how sampling and resampling strategies can be
used for different local Gaussian spectral investigations of a time
series.

\Cref{P2.app:resampling_parametric_bootstrap} focuses on the
parametric bootstrap approach, where repeated independent simulations
are done from given parametric models fitted to a given sample. It is
here seen how the univariate local Gaussian sanity testing of
parametric models from \lgsdRef{app:resampling_parametric_bootstrap}
can be extended to cover the multivariate case.  This part includes an
investigation of the DAX- and CAC-margins of the \EuStockMarkets
example, and then new plots are presented for the sanity testing of
the cross-interaction of the multivariate parametric models fitted to
the \EuStockMarkets data.

The discussion in
\cref{P2.app:nonparametric_resampling_strategies,P2.app:Block_length_sensitivity}
focuses on the resampling problem, which is critical for the
production of the pointwise confidence intervals based on a given
sample.  \Cref{P2.app:nonparametric_resampling_strategies} gives a
summary of the discussion in \lgsdRef{app:regarding_resampling} that
motivated the introduction of the \textit{Circular index-based block
  bootstrap for tuples}, with some additional comments about the role
of the block length $L$.  Finally,
\cref{P2.app:Block_length_sensitivity} presents a sensitivity analysis
of the block length $L$ for a multivariate example, which shows that
this parameter have a minor effect on the resulting estimated
pointwise confidence intervals.

\subsection{Parametric bootstrap and local sanity-testing of models}
\label{P2.app:resampling_parametric_bootstrap}

The discussion in \lgsdRef{app:resampling_parametric_bootstrap}
(motivated by a similar approach in \citet{BIRR2019122}) explained how
a parametric bootstrap approach could be used to investigate models
fitted to real data.  The key idea is that a parametric model first is
fitted to the original sample, and samples from that particular fitted
model is thereafter used to produce the local Gaussian estimates
needed for the comparison with the actual sample under investigation.
This enables a local sanity-test (based on the $m$-truncated local
Gaussian spectra) of the fitted model, since it can be used to
identify if there are any points/frequencies with a clear mismatch
between the local structures detected in the original sample and those
seen in samples from the fitted model.

A key restriction in the above mentioned comparison is that the number
of observations simulated from the fitted model should equal the
number of observations in the sample. It is of course also possible to
compare different fitted models against each other directly, in which
case the length of the simulated samples can be selected freely.
Such an approach might be useful in order to check if different
model-fitting approaches might have deviating local Gaussian
dependency structures.  The plots from the main part (and
\cref{P2.app:sensitivity_analysis}) can be used when comparing
different fitted models.

It might also be of interest to simulate samples of different lengths
from one particular known parametric model, and then fit different
parametric models to those samples.  The plots discussed in
\cref{P2.app:the_univariate_marginals,P2.app:the_multivariate_marginals}
might be helpful for such an investigation.

\subsubsection{The example under consideration}
\label{P2.app:parametric_bootstrap_the_example_under_consideration}

The univariate example in
\lgsdRef{app:resampling_parametric_bootstrap} used the
\texttt{dmbp}-data as the known sample.  The fitted model was a
GARCH-type model, more precisely an \textit{asymmetric power
  ARCH-model} (apARCH) of order $(2,3)$, with parameters based on a
fitting to the \texttt{dmbp}-data.\footnote{%
  The \Rpackage \Rref{rugarch}, \citet{ghalanos15:_rugarch} was used to
  find the parameters of several GARCH-models, and the asymmetric
  power ARCH model with the best fit was then selected.  See
  \lgsdRef{sec:GARCH_model} for further details.}  The univariate
approach 
will now be extended to the multivariate case, with a repetition of
the setup from the univariate case in
\cref{P2.app:the_univariate_marginals} (since any investigation of a
multivariate parametric model fitted to a multivariate sample should
include a comparison of the univariate components).

The multivariate example considered in this paper uses the
\EuStockMarkets-data as the known sample, and in particular the
restriction to the bivariate subset based on the DAX- and
CAC-components (see
\cref{fig:BIRR_P1_fig_F1_Y1Y1,fig:BIRR_P1_fig_F1_Y3Y3,fig:BIRR_Co_P1_fig_F1_Y1Y3,fig:BIRR_Quad_P1_fig_F1_Y1Y3,fig:BIRR_amplitude_P1_fig_F1_Y1Y3,fig:BIRR_phase_P1_fig_F1_Y1Y3}).
As described in \cref{sec:lgch:cGARCH}, the fitted parametric model in
this case is a basic multivariate copula GARCH-model --- which was
selected only in order to provide a proof of concept.  From this
starting point
it is natural to expect that the present \textit{local Gaussian sanity
  investigation} should reveal that this particular multivariate
copula GARCH-model does not properly capture the local Gaussian
structures seen in the original sample.

\subsubsection{The univariate marginals}
\label{P2.app:the_univariate_marginals}

\Cref{fig:BIRR_P1_fig_F1_Y1Y1,fig:BIRR_P1_fig_F1_Y3Y3} consider
respectively the univariate DAX- and CAC-components of the
\EuStockMarkets-data, and these use the comparison of fitted models
and sample that was introduced in
\lgsdRef{fig:aparch_dmbp_comparison}.  The description below is in
essence the same as the one given in
\lgsdRef{app:resampling_parametric_bootstrap}, and it explains how
these plots are to be interpreted.

\begin{figure}
  {\centering
    \includegraphics[width=1\linewidth]{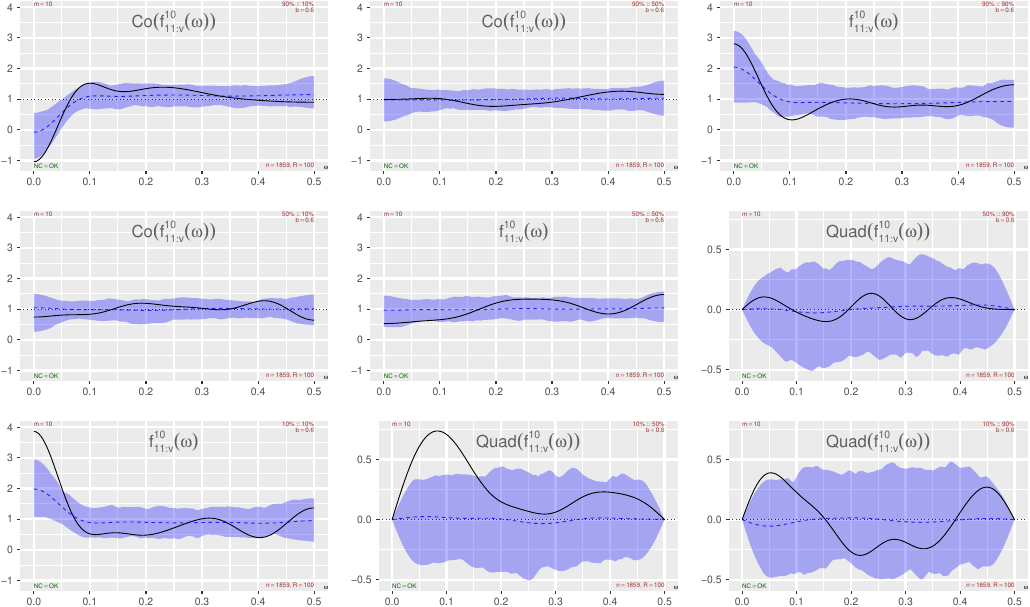}}

  \caption{Marginal investigation: cGARCH-model fitted to
    \EuStockMarkets, the part corresponding to the DAX-component.}

  \label{fig:BIRR_P1_fig_F1_Y1Y1}
\end{figure}

\begin{figure}
  {\centering
    \includegraphics[width=1\linewidth]{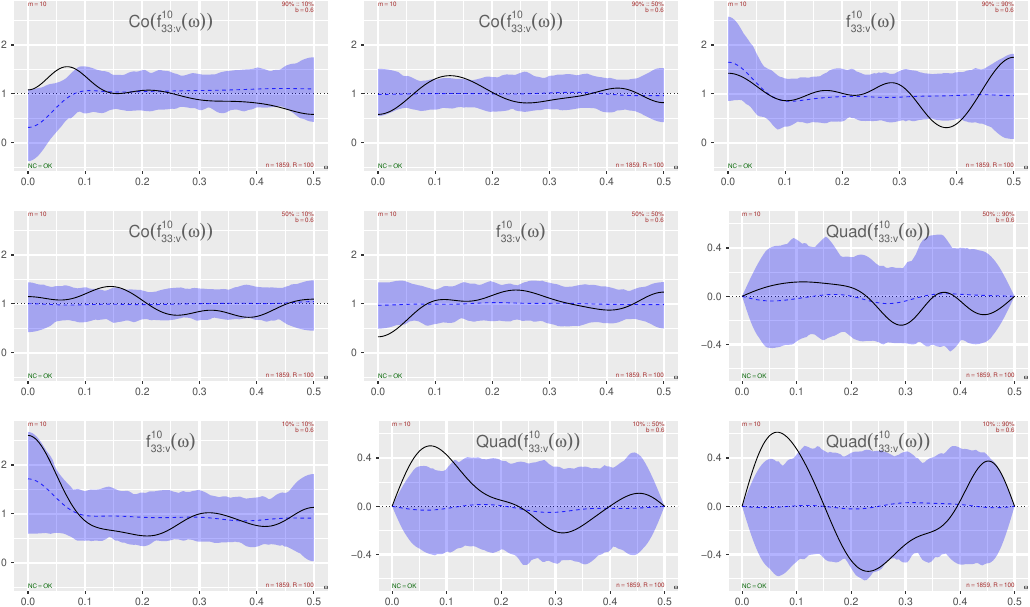}}

    \caption{Marginal investigation: cGARCH-model fitted to
    \EuStockMarkets, the part corresponding to the CAC-component.}

  \label{fig:BIRR_P1_fig_F1_Y3Y3}
\end{figure}

The key idea in \cref{fig:BIRR_P1_fig_F1_Y1Y1,fig:BIRR_P1_fig_F1_Y3Y3}
is that estimates of the $m$-truncated local Gaussian auto-spectra
$\lgcsdM{kk}{\LGp}{\omega}{m}$, based on the original sample, can be
superimposed on the plots based on parametric bootstrapping from the
fitted model, and this makes it easy to compare them.

Nine different points $\LGp=\LGpoint$ are considered in
\cref{fig:BIRR_P1_fig_F1_Y1Y1,fig:BIRR_P1_fig_F1_Y3Y3}, and these are
based on the combinations that can be created when $\LGpi{1}$ and
$\LGpi{2}$ varies over the 10\%, 50\% and 90\% percentiles of the
standard normal distribution.  The corresponding plots are ordered in
a grid in accordance with the position of these nine points in the
plane, as can be seen by the information about $\LGp$ in the upper
right corner of the respective plots.

The estimates $\lgcsdM{kk}{\LGp}{\omega}{m}$ are always real-valued on
the diagonal (since $\lgccr{kk}{\LGp}{-h}=\lgccr{kk}{\LGp}{h}$ when
the point $\LGp$ is on the diagonal), and the three diagonal plots are
thus marked with $\lgcsdM{kk}{\LGp}{\omega}{m}$ in order to signal
that the content is the local Gaussian autospectrum.  The solid lines
represents the estimates based on the original sample, whereas the
dotted lines and the 95~\% pointwise confidence intervals are based on
samples from the fitted model.

The estimates $\lgcsdM{kk}{\LGp}{\omega}{m}$ are complex-valued for the
six off-diagonal points, and in this case the \Rpackage
\lgsdRpackage\ follows the convention used for the complex-valued
cross-spectra, viz.,
$\operatorname{Co}\!\parenR{\lgcsdM{kk}{\LGp}{\omega}{m}} =
\operatorname{Re}\!\parenR{\lgcsdM{kk}{\LGp}{\omega}{m}}$ and
$\operatorname{Quad}\!\parenR{\lgcsdM{kk}{\LGp}{\omega}{m}} =
- \operatorname{Im}\!\parenR{\lgcsdM{kk}{\LGp}{\omega}{m}}$. 

The off-diagonal points are symmetric around the diagonal, i.e.,\ both
$\LGp=\LGpoint$ and its diagonal reflection $\LGpd=\LGpointd$ are
present.  The property
$\lgcsd{kk}{\LGp}{\omega} = \overline{\lgcsd{kk}{\LGpd}{\omega}}$,
implies that it is sufficient to plot
$\operatorname{Co}\!\parenR{\lgcsdM{kk}{\LGp}{\omega}{m}}$ on one side
of the diagonal and
$\operatorname{Quad}\!\parenR{\lgcsdM{kk}{\LGp}{\omega}{m}}$ on the
other side.

Finally, the same scale is used for all plots showing real values,
whereas another scale is used for the plots related to the imaginary
parts.  This distinction is natural since the scale needed for the
imaginary part can be much smaller,
as can be seen in \cref{fig:BIRR_P1_fig_F1_Y1Y1,fig:BIRR_P1_fig_F1_Y3Y3}.

It was in \lgsdRef{app:resampling_parametric_bootstrap} seen that the
local Gaussian dependency structure of the fitted apARCH$(2,3)$-model
agreed quite well with the one observed in the \texttt{dmbp}-data.
This did not come as a surprise, since the selected
apARCH$(2,3)$-model in that case was obtained after a testing
procedure that tried out several thousand different variations of the
GARCH-type models implemented in the \Rref{rugarch}-package.

The situation in this paper is completely contrary to the one in
\JT, since a default model from the
\Rref{rmgarch}-package, see \citet{ghalanos15:_rmgarch}, was fitted to
the \EuStockMarkets-data in order to simply provide a
proof-of-principle parametric model for how this kind of investigation
can be performed.  It is thus natural to expect that an investigation
of \cref{fig:BIRR_P1_fig_F1_Y1Y1,fig:BIRR_P1_fig_F1_Y3Y3} would
indicate that the local Gaussian dependency structures of the fitted
marginal models do not properly match those seen in the DAX- and
CAC-components of the \EuStockMarkets-data --- and this is
indeed what is observed when the dashed and solid lines in
\cref{fig:BIRR_P1_fig_F1_Y1Y1,fig:BIRR_P1_fig_F1_Y3Y3} are compared
against each other.

The plots related to the real parts in
\cref{fig:BIRR_P1_fig_F1_Y1Y1,fig:BIRR_P1_fig_F1_Y3Y3} show that there
is a mismatch for the lower tail for low frequencies, in particular
for the DAX-component as seen in \cref{fig:BIRR_P1_fig_F1_Y1Y1}.  The
plots related to the imaginary parts also indicate that the fitted
models do not properly reflect the dependency structure from the DAX-
and CAC-components --- but it should here be noted that different
scales are used for the two groups of plots (real-valued versus
complex valued-parts), and it is more critical when a deviation is
observed within the plots having the dominating scale.

The observed mismatch between the fitted marginal models and the DAX-
and CAC-components of the \EuStockMarkets-data seen in
\cref{fig:BIRR_P1_fig_F1_Y1Y1,fig:BIRR_P1_fig_F1_Y3Y3} clearly
indicates that the selected marginal models should be replaced.  This
is (as mentioned above) not a surprising revelation, since the present
marginal models were selected simply in order to present a
proof-of-principle for the methods discussed in this paper.

\subsubsection{The multivariate interdependency structure}
\label{P2.app:the_multivariate_marginals}

This section provides an extension to the multivariate case of the
univariate local Gaussian sanity testing seen in
\cref{P2.app:the_univariate_marginals}.  The target of interest in
this case is the cross-temporal interdependency structure of the
parametric model that has been fitted to a given multivariate sample.
The idea is again to use visual inspections to see if clear deviations
can be seen when the local Gaussian cross-spectra based on the fitted
model are compared with the local Gaussian cross-spectra from the
original sample.

The univariate investigation in
\cref{fig:BIRR_P1_fig_F1_Y1Y1,fig:BIRR_P1_fig_F1_Y3Y3} focused on a
grid of nine points $\LGp=\LGpoint$, where $\LGpi{1}$ and $\LGpi{2}$
varied over the 10\%, 50\% and 90\% percentiles of the standard normal
distribution --- and the multivariate investigation will use the exact
same grid (other points can of course be used if it is of interest to
focus on a particular region).

The \enquote{diagonal folding} seen in
\cref{fig:BIRR_P1_fig_F1_Y1Y1,fig:BIRR_P1_fig_F1_Y3Y3} can not be used
for the investigation of the cross-temporal case, since the local
Gaussian cross-spectrum $\lgcsdM{k\ell}{\LGp}{\omega}{m}$ are complex
valued even for the diagonal points.  This implies that it, similarly
to the situation encountered for the heatmap- and distance plots (see
discussion in \cref{P2.app:Point_sensitivity}), can be natural to
include pair of plots that either focus on the \textit{Cartesian
  decomposition} or the \textit{polar decomposition} of the
complex-valued local Gaussian cross-spectrum
$\lgcsdM{k\ell}{\LGp}{\omega}{m}$.

A visual investigation based on the \textit{Cartesian decomposition}
are given in
\cref{fig:BIRR_Co_P1_fig_F1_Y1Y3,fig:BIRR_Quad_P1_fig_F1_Y1Y3}, which
respectively present the \Co- and \Quad-spectra components of
$\lgcsdM{k\ell}{\LGp}{\omega}{m}$.  These figures gathers local
Gaussian spectra from the original data (shown with solid lines) and
local Gaussian spectra based on simulations from the fitted models
(shown with dashed lines and pointwise confidence intervals).  The
same internal scale is used for all the subplots, since that
simplifies the task of detecting different behaviour between the
points.

\begin{figure}
  {\centering
    \includegraphics[width=1\linewidth]{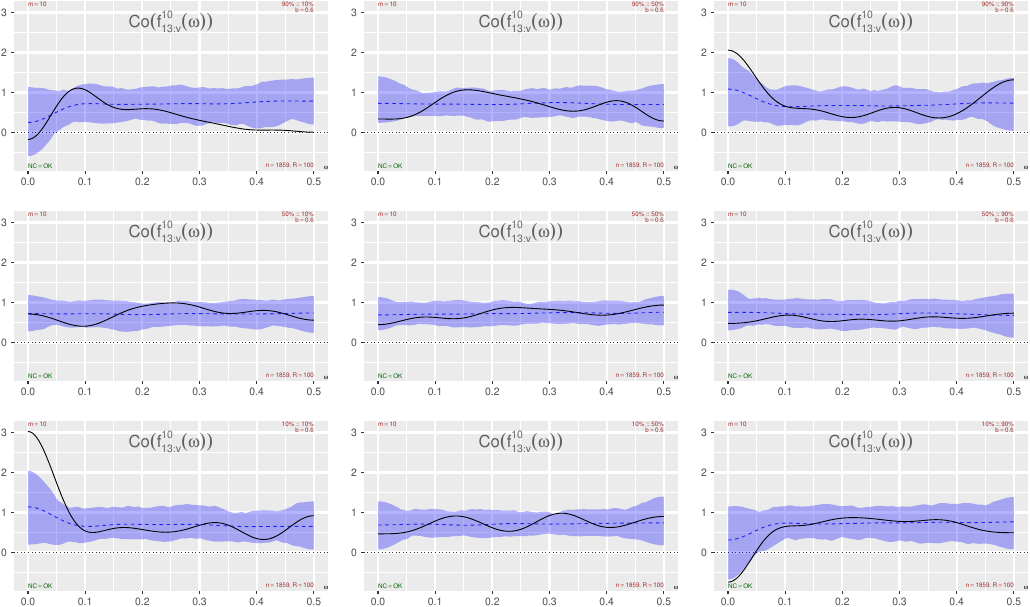}}

  \caption{\Co-spectra based local Gaussian investigation of the
    cGARCH-model fitted to the \EuStockMarkets-data.  Focus on the
    interaction between the DAX- and CAC-components of
    \EuStockMarkets}
  
  \label{fig:BIRR_Co_P1_fig_F1_Y1Y3}
\end{figure}

\begin{figure}
  {\centering
    \includegraphics[width=1\linewidth]{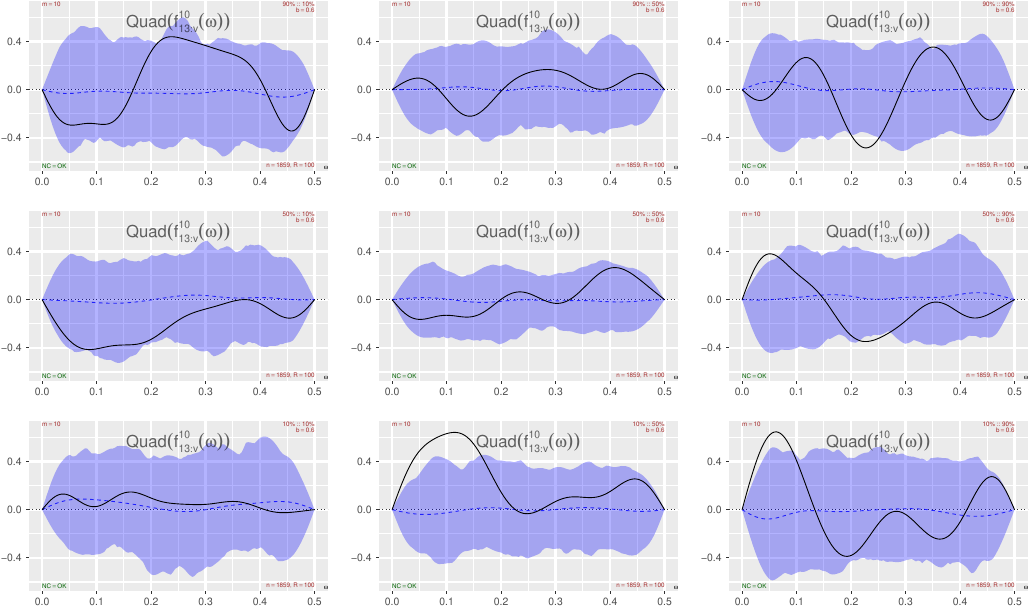}}
  \caption{\Quad-spectra based local Gaussian investigation of the
    cGARCH-model fitted to the \EuStockMarkets-data.  Focus on the
    interaction between the DAX- and CAC-components of
    \EuStockMarkets}
  \label{fig:BIRR_Quad_P1_fig_F1_Y1Y3}
\end{figure}

The scales for the \Co-plots in \cref{fig:BIRR_Co_P1_fig_F1_Y1Y3} are
much larger than those for the \Quad-plots in
\cref{fig:BIRR_Quad_P1_fig_F1_Y1Y3}, and it is thus natural to focus on the
\Quad-plots in this case.  An inspection of
\cref{fig:BIRR_Co_P1_fig_F1_Y1Y3}, in particular the points in the
tails, shows that the local Gaussian cross-spectra based on the fitted
parametric model deviates quite a bit from those computed from the
original \EuStockMarkets-sample.  This is not surprising at all, since
the most basic parametric model was used simply in order to provide a
proof-of-concept for the present investigation.

The \textit{polar decomposition} of the complex-values local Gaussian
cross-spectra $\lgcsdM{k\ell}{\LGp}{\omega}{m}$ into \Amplitude- and
\Phase-spectra are given in
\cref{fig:BIRR_amplitude_P1_fig_F1_Y1Y3,fig:BIRR_phase_P1_fig_F1_Y1Y3}.
The \Amplitude-plots in \cref{fig:BIRR_amplitude_P1_fig_F1_Y1Y3} are
quite similar to the \Co-plots seen in
\cref{fig:BIRR_Co_P1_fig_F1_Y1Y3}, which is as expected since the
scales of the \Co-plots are much larger than those for the
\Quad-plots.  For this particular example it follows that the same
conclusions can be drawn from
\cref{fig:BIRR_Co_P1_fig_F1_Y1Y3,fig:BIRR_amplitude_P1_fig_F1_Y1Y3},
i.e.\ that it would be preferable to look for a better model than the
trivial one used as a proof-of-concept example in this paper.

\begin{figure}
  {\centering
    \includegraphics[width=1\linewidth]{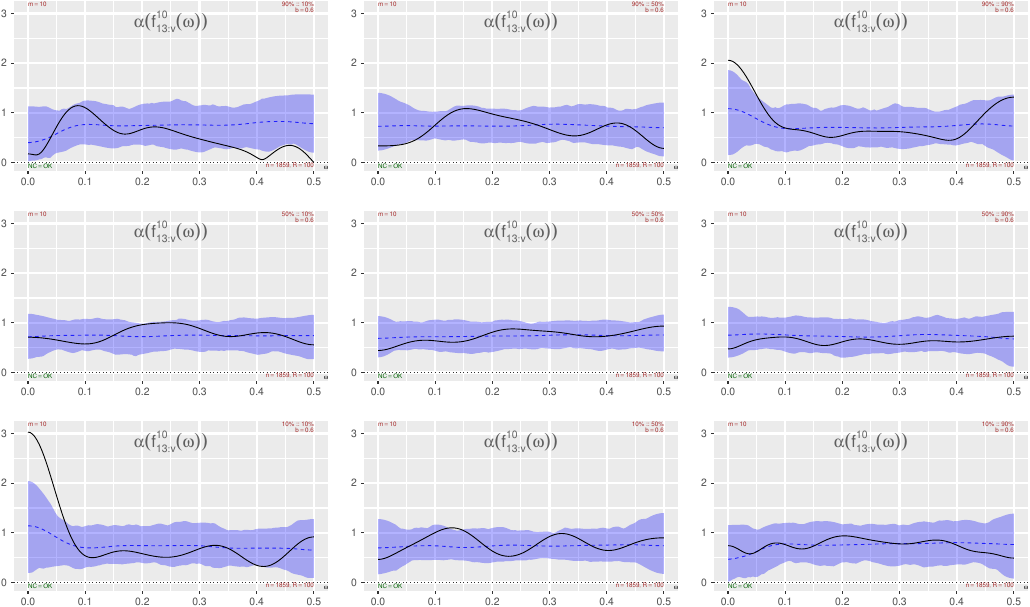}}

  \caption{\Amplitude-spectra based local Gaussian investigation of the
    cGARCH-model fitted to the \EuStockMarkets-data.  Focus on the
    interaction between the DAX- and CAC-components of
    \EuStockMarkets}
  \label{fig:BIRR_amplitude_P1_fig_F1_Y1Y3}
\end{figure}

\begin{figure}
  {\centering
    \includegraphics[width=1\linewidth]{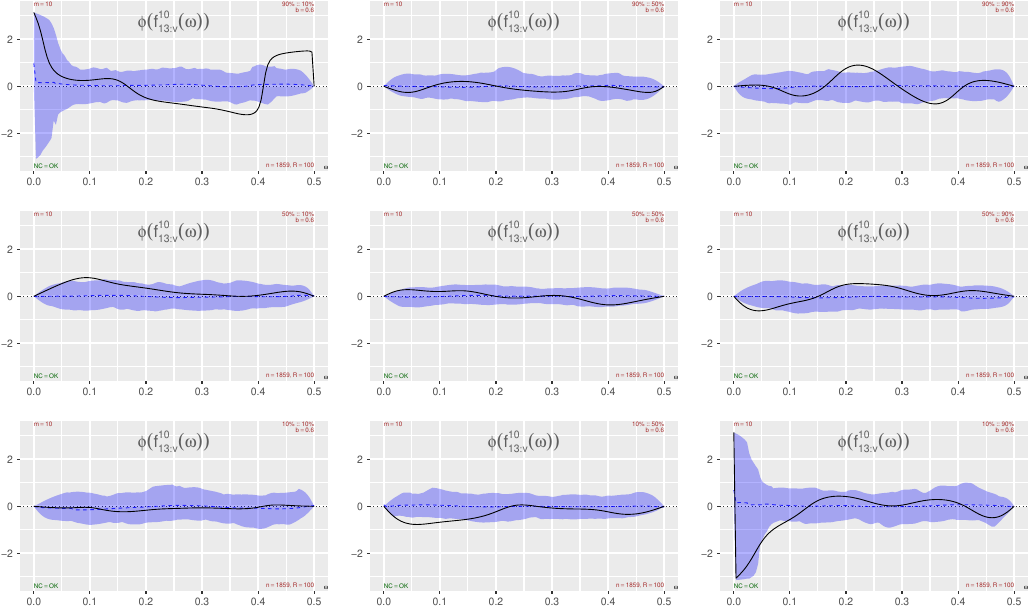}}
  \caption{\Phase-spectra based local Gaussian investigation of the
    cGARCH-model fitted to the \EuStockMarkets-data.  Focus on the
    interaction between the DAX- and CAC-components of
    \EuStockMarkets}
  \label{fig:BIRR_phase_P1_fig_F1_Y1Y3}
\end{figure}

Note that the similarity between
\cref{fig:BIRR_Co_P1_fig_F1_Y1Y3,fig:BIRR_amplitude_P1_fig_F1_Y1Y3} is
due to the real part of $\lgcsdM{k\ell}{\LGp}{\omega}{m}$ are the
dominating component, and it could be that would not always be the
case.  Furthermore, the \Phase-spectra seen in
\cref{fig:BIRR_phase_P1_fig_F1_Y1Y3} can be quite hard to interpret if
the observations of $\lgcsdM{k\ell}{\LGp}{\omega}{m}$ lies in the
third and fourth quadrant of the complex plane.  Due to this it seems
natural to consider a \Co- and \Quad- based investigation of
$\lgcsdM{k\ell}{\LGp}{\omega}{m}$ instead of an \Amplitude- an
\Phase-investigation when a local Gaussian sanity testing is performed
for a parametric model of the the cross-temporal interdependency
structure for a given multivariate sample.

\subsection{Nonparametric bootstrapping techniques}
\label{P2.app:nonparametric_resampling_strategies}

The local Gaussian estimates $\hatlgcsdM{k\ell}{\LGp}{\omega}{m}$,
with pointwise confidence intervals, are straightforward to construct
when a given parametric model is investigated --- simply generate a
sufficiently high number of samples from the model, compute
$\hatlgcsdM{k\ell}{\LGp}{\omega}{m}$ for each sample, and then (for
each frequency $\omega$) use suitable quantiles (like, e.g., the 5\%
and 95\%) to construct the pointwise confidence intervals.

Pointwise confidence intervals are also of interest when the local
Gaussian estimates $\hatlgcsdM{k\ell}{\LGp}{\omega}{m}$ are based on a
given sample (like the \EuStockMarkets-data), and such confidence
intervals can, e.g., be seen in \cref{fig:EuStockMarkets} in the main
part.  These intervals are again (for each frequency $\omega$) based
on quantiles of $\hatlgcsdM{k\ell}{\LGp}{\omega}{m}$ from a collection
of samples, but these samples must now be constructed by a suitable
resampling strategy.

The initial resampling strategy was the block bootstrap from
\citet{kuensch89:_jackk_boots_gener_station_obser}.  The block
bootstrap can (as explained below) be applied for the estimators
underlying the local Gaussian estimates
$\hatlgcsdM{k\ell}{\LGp}{\omega}{m}$, but it was early on observed
that it was necessary with rather heuristic arguments regarding the
selection of the block length argument $L$ for the time series samples
of interest (cases where $Y_{t}$ and $Y_{t+h}$ are dependent but not
correlated).

Comments received during the review-process of \JT initiated an
investigation of how estimates based on the block bootstrap method can
suffer from edge-effect noise when it is used on smaller sample sizes
--- and this resulted in the development of the \textit{circular
  index-based block bootstrap for tuples} resampling strategy in
\lgsdRef{def:index_based_block_bootstrap_for_tuples}

A summary of this story will be given below. The interested reader
should consult
\lgsdRef{app:nonparametric_resampling_strategies,app:adjusted_resampling_algorithm}
for an in depth discussion of different bootstrapping methods, with
emphasis on the problems that occur due to edge-effect noise.  The
revised resampling algorithm is based on a tweaking of the block
bootstrap, and it thus inherits the block length argument $L$.  Note
that the block length $L$ plays a different role in this algorithm,
and that the sensitivity analysis (based on the \EuStockMarkets-data)
in \cref{P2.app:Block_length_sensitivity} shows that the resulting
pointwise confidence intervals for
$\hatlgcsdM{k\ell}{\LGp}{\omega}{m}$ are quite stable.

\subsubsection{Summary of \lgsdRef{app:nonparametric_resampling_strategies}}
\label{P2.summary.P1.App.F3}

The discussion in \lgsdRef{app:nonparametric_resampling_strategies}
started out with an explanation of why resampling based on the block
bootstrap should be an acceptable strategy for the estimators needed
for the computation of $m$-truncated estimates of the local Gaussian
spectra $\lgcsdM{k\ell}{\LGp}{\omega}{m}$, then it addressed the
problem that edge-effects between the blocks could potentially create
serious problems for the estimates
$\hatlgcsdM{k\ell}{\LGp}{\omega}{m}$ when used on the \enquote{short}
\texttt{dmbp}-sample (length 1974) that was investigated in \JT.
Potential solutions to these problems were addressed, including the
block-of-blocks bootstrap (where edge-effects do not occur), but
these solutions did not provide an optimal approach for a statistic
computed by means of an algorithm that contains a kernel function.

\textbf{Justification for the block bootstrap:} The estimates
$\hatlgcsdM{k\ell}{\LGp}{\omega}{m}$ of the local Gaussian spectral
densities are based on linear combinations of the estimates
$\TSR{\hatlgccr{k\ell}{\LGp}{h}}{h=-m}{m}$ of the local Gaussian
cross-correlations. These are all estimated by a local likelihood
approach --- and the asymptotic properties of the estimates
$\hatlgcsdM{k\ell}{\LGp}{\omega}{m}$ were developed in \JT using the
procedure from \citet{klimko1978}.
A statistic obtained from the Klimko-Nelson procedure was explicitly
mentioned by K{\"u}nsch as an example for which the block bootstrap
method would be applicable, cf.\ \citet[Example 2.4, p.\
1219-20]{kuensch89:_jackk_boots_gener_station_obser}, and a resampling based on the block bootstrap was thus initially
used for the construction of the pointwise confidence intervals for
the \texttt{dmbp}-example seen in \lgsdRef{fig:dmbp}.

\textbf{Edge-effect noise:} The blocks in the block bootstrap
algorithm will introduce noise into the estimation procedure, since
sequences will occur in the resampled data that was not present in the
original data.  The following paragraphs from
\lgsdRef{sec:edge_effects_between_blocks} introduce the
\enquote{corrupt tuple}-concept and some notation needed later on.

\begin{blockquote} %
  {For example, if a time series $\TSR{\Yz{t}}{t=1}{n}$ of length $n$
  is given, then an estimate of $\lgacr{\LGp}{h}$ will be based on the
  bivariate set
  $\mathcalYz{h}\defeq\TSR{\parenR{\Yz{t+h},\Yz{t}}}{t=1}{n-h}$ of
  size $n-h$.  When the block bootstrap is used with some block length
  $L$, then there will be a resampled sequence
  $\TSR{\Yz[*]{t}}{t=1}{n}$ and the idea is that an estimate of
  $\lgacr{\LGp}{h}$ now should be computed based on the bivariate set
  $\mathcalYz[*]{h:L}\defeq\TSR{\parenR{\Yz[*]{t+h},\Yz[*]{t}}}{t=1}{n-h}$.

  However, the set $\mathcalYz[*]{h:L}$ will contain \textit{corrupt
    tuples} that do not exist in $\mathcalYz{h}$, i.e.,\ the first and
  second component of $\parenR{\Yz[*]{t+h},\Yz[*]{t}}$ can belong to
  different blocks, and this will add a bit of \textit{edge-effect
    noise} into the estimation process.  The edge-effect noise is
  negligible in the asymptotic situation (very large sample sizes $n$
  and large block lengths $L$), but it can make an impact when smaller
  samples are investigated.}
\end{blockquote}

The example used in \JT was the \texttt{dmbp}-data ($n=1974$ unique
observations), and the plots of the estimated local Gaussian
autospectra $\hatlgsdM{\LGp}{\omega}{m}$ used the truncation level
$m=10$.  It was thus of particular interest to compute (for different
block lengths $L$) the expected amount of \textit{corrupt tuples} when
estimating the local Gaussian autocorrelations $\lgacr{\LGp}{h}$ for
$h=1,\dotsc,10$ --- and a copy of the resulting table from \JT is
given in \cref{P2.app:table:corrupt.tuples.block.bootstrap}.

\begin{table}[ht]
  \centering
  \begingroup\tiny
  \begin{tabular}{r|llllllllll}
    \hline
    $L$ \textbackslash\ $h$ & 1 & 2 & 3 & 4 & 5 & 6 & 7 & 8 & 9 & 10 \\ 
    \hline
    25 & 4.0\% & 7.9\% & 11.9\% & 15.8\% & 19.8\% & 23.8\% & 27.8\% & 31.7\% & 35.7\% & 39.7\% \\ 
    100 & 1.0\% & 1.9\% & 2.9\% & 3.9\% & 4.8\% & 5.8\% & 6.8\% & 7.7\% & 8.7\% & 9.7\% \\
    \hline
  \end{tabular}
  \endgroup
  \caption{Copy of \lgsdRef{app:table:corrupt.tuples.block.bootstrap}.
    The expected fraction of corrupt tuples when $\lgacr{\LGp}{h}$ are
    estimated from block bootstrap replicates of the
    \texttt{dmbp}-data ($n=1974$), when $L\in\parenC{25,100}$ and
    $h\in\parenC{1,\dotsc,10}$.}
  \label{P2.app:table:corrupt.tuples.block.bootstrap}
\end{table}

\Cref{P2.app:table:corrupt.tuples.block.bootstrap} shows that the
amount of noise due to edge-effects can be severe when the block
bootstrap is used with a too short block length, and this was also
seen when attempts were made at making pointwise confidence intervals
for $\lgsdM{\LGp}{\omega}{m}$ based on this.  The resulting
\enquote{confidence intervals} for small block lengths $L$ did in some
cases (for low frequencies $\omega$) not even contain the actual
estimate $\hatlgsdM{\LGp}{\omega}{m}$ based on the sample under
investigation.

The observations in \cref{P2.app:table:corrupt.tuples.block.bootstrap}
indicated that an adjusted resampling technique was needed, preferably
one that completely (or at least partially) removed the corrupt tuples
from the estimation algorithm.  Two different approaches that
completely avoids the corrupt tuples was discussed in \JT, but they
had some issues that made them less interesting to implement --- as
briefly explained below.

Resampling from the tuples in
$\mathcalYz{h}\defeq\TSR{\parenR{\Yz{t+h},\Yz{t}}}{t=1}{n-h}$, for
$h=1$ to $h=m$, was considered in
\lgsdRef{sec:natural_solution_to_edge_effect_issue?} as one option for
finding an estimate $\hatlgcsdM{k\ell}{\LGp}{\omega}{m}$ of the $m$
truncated local Gaussian autospectrum
$\lgcsdM{k\ell}{\LGp}{\omega}{m}$.  The resampling must be done in a
manner that properly takes into account the time-connection between
the bivariate pairs $\parenR{\Yz{t+h},\Yz{t}}$ for different
$h$-values, but it is easy to find an algorithm that does this for a
selected truncation level $m$: Simply consider the bivariate tuples in
$\mathcalYz{h}$ as projections from the $\parenR{m+1}$-variate tuples
$\mathcalYz{m:0} =
\TSR{\parenR{\Yz{t+m},\dotsc,\Yz{t+1},\Yz{t}}}{t=1}{n-m}$, and start
out by sampling the desired number of elements from $\mathcalYz{m:0}$.
A shortcoming with this approach is that the resulting estimates
$\hatlgacr[p]{\LGp}{h}$ will depend on the selected $m$-value --- and
this dependency on the truncation level $m$ implies that all the local
Gaussian autocorrelations must be recomputed if the initial truncation
level $m$ is changed to $m+1$, and this is computational costly and
not desirable to implement.

The block-of-blocks bootstrap introduced in
\citet{politis92:_gener_resam_schem_trian_array} was discussed in
\lgsdRef{sec:block_of_blocks_bootstrap}.  This method completely
avoids the edge-effect issue that occurs for the block bootstrap,
since the statistic of interest are computed on the individual blocks.
However, as mentioned in \JT:

\begin{blockquote} %
  {This restriction to individual blocks can be an excellent idea for
    many statistics of interest, but it is a somewhat questionable
    approach for the estimates $\hatlgacr[p]{\LGp}{h}$ of the local
    Gaussian autocorrelations.  The reason for this is that the
    bandwidth argument $\bm{b}$ in the kernel function
    $\Kh[\bm{w}-\LGp]{\bm{b}}$ must be much larger if the estimation
    algorithm is to be used on only a subset of the observations ---
    and the local structures of interest might then not be detected at
    all.

    It would of course be of interest to implement the block-of-block
    bootstrap for the estimates of the local Gaussian spectra if very
    large samples are encountered, i.e.,\ when the individual blocks
    contains several thousand consecutive observations --- but for
    shorter samples (like the \texttt{dmbp}-example) it seems better to
    use something else.}
\end{blockquote}

\subsubsection{Summary of \lgsdRef{app:adjusted_resampling_algorithm}}
\label{P2.summary.P1.App.F4}

The \textit{circular index-based block bootstrap for tuples} was
introduced in \lgsdRef{app:adjusted_resampling_algorithm}.  This is a
minor adjustment of the ordinary block bootstrap, and it inherits the
block length argument $L$ as a tuning parameter (see
\cref{P2.app:Block_length_sensitivity} for a sensitivity analysis).
The key motivation for the revised approach is that an estimate
$\hatlgcsdM{k\ell}{\LGp}{\omega}{m}$ of the $m$-truncated local
Gaussian autospectrum $\lgcsdM{k\ell}{\LGp}{\omega}{m}$ should be
based on $m$ \enquote{linked samples} from
$\mathcalYz{h}\defeq\TSR{\parenR{\Yz{t+h},\Yz{t}}}{t=1}{n-h}$, for
$h=1$ to $h=m$, and that these \enquote{linked samples} should be used
for the estimation of the local Gaussian autocorrelations
$\lgacr{\LGp}{h}$.

This approach avoids the computational cost problem that made the
method mentioned in
\lgsdRef{sec:natural_solution_to_edge_effect_issue?}) less interesting
to implement, i.e., the estimates $\hatlgacr{\LGp}{h}$ do not depend
on the truncation level $m$ --- but the price to pay for this is that
some \enquote{corrupt tuples}
will still be present, and the issue with noise due to edge-effects
does not completely disappear.  However, the expected amount of this
noise for the revised resampling strategy (see
\cref{P2.app:table:corrupt.tuples.ibb.bootstrap} on page
\pageref{P2.app:table:corrupt.tuples.ibb.bootstrap}) is reduced a lot
from the one seen for the resampling based on the block bootstrap (see
\cref{P2.app:table:corrupt.tuples.block.bootstrap} on page
\pageref{P2.app:table:corrupt.tuples.block.bootstrap}).

The following toy-example are copied from \JT in order to explain the
underlying motivation for the \textit{circular index-based block
  bootstrap for tuples} resampling algorithm, i.e., that it is the
indices of the observations that are the key detail of interest to
keep track of.

\begin{blockquote} %
  {It will be a bit easier to digest the definitions and the algorithm
    that are given later on in this section, if a simple toy-example
    is investigated first: Consider a situation with a time series
    having five unique observations
    $\Yz{1},\Yz{2},\Yz{3},\Yz{4},\Yz{5}$ and assume that there is an
    interest for an estimate based on the four lag-1 tuples in
    $\mathcalYz{1} = \TSR{\parenR{\Yz{t+1},\Yz{t}}}{t=1}{4}$.  If a
    block bootstrap with block length $L=2$ is used, the resampled
    time series might e.g.\ look like
    $\Yz[*]{1}=\Yz{4},\Yz[*]{2}=\Yz{5},\Yz[*]{3}=\Yz{3},\Yz[*]{4}=\Yz{4},\Yz[*]{5}=\Yz{2}$,
    and the corresponding set of lag-1 tuples would be
    $\mathcalYz[*]{1:2} =
    \TSR{\parenR{\Yz[*]{t+1},\Yz[*]{t}}}{t=1}{4}$.  It is easy to see
    that $\mathcalYz[*]{1:2}$ in this case will contain the two
    corrupt tuples $\parenR{\Yz{3},\Yz{5}}$ and
    $\parenR{\Yz{2},\Yz{4}}$, i.e.,\ tuples that are not present in
    $\mathcalYz{1}$.

    The key idea in the adjusted algorithm is to move the focus to the
    indices of the original sample, i.e.,\ $1,2,3,4,5$, and then use
    the block bootstrap to sample from these.  The resampled set of
    indices for the example above would be $4,5,3,4,2$, and from these
    it is possible to construct the \textit{cyclically $h=1$ shifted}
    set of indices $5, 1, 4, 5, 3$.  The method is simply to add the
    lag $h=1$ to all the resampled indices --- and to start back on 1
    if a value exceeds $n=5$.  The four desired lag-1 tuples
    $\mathcalYz[\sharp]{1:2} =
    \TSR{\parenR{\Yz[\sharp]{t+1},\Yz[\sharp]{t}}}{t=1}{4}$ are now
    created by using the resampled set of indices in the
    $\Yz[\sharp]{t}$-component, whereas the cyclically $h=1$ shifted
    indices are used for the $\Yz[\sharp]{t+h}$-component.  This
    results in the following four tuples,
    $\mathcalYz[\sharp]{1:2}=\parenC{\parenR{\Yz{5},\Yz{4}},
      \parenR{\Yz{1},\Yz{5}}, \parenR{\Yz{4},\Yz{3}},
      \parenR{\Yz{5},\Yz{4}}}$, and it is easy to see that the only
    corrupt tuple in $\mathcalYz[\sharp]{1:2}$ is
    $\parenR{\Yz{1},\Yz{5}}$.  Note: It could in principle now also be
    added a fifth tuple $\parenR{\Yz{3},\Yz{2}}$ to
    $\mathcalYz[\sharp]{1:2}$, but that is not of interest since there
    are only four tuples in $\mathcalYz{1}$.
  }
\end{blockquote}

The adjusted resampling algorithm is, as seen above, based on a simple
tweaking of the ordinary block bootstrap.  The formal definition of
the algorithm requires some rather technical steps, and these will not
be repeated in this paper.  The interested reader can consult
\lgsdRef{app:adjusted_resampling_algorithm} for further details, see
in particular
\lgsdRef{def:slightly_tweaked_modulo_function,def:indices_of_tuple_given_starting_index,def:actual_tuple_given_starting_index,def:index_based_block_bootstrap_for_tuples}.

Note that the resampling algorithm from
\lgsdRef{def:index_based_block_bootstrap_for_tuples} by construction
will return the same results as those obtained from the ordinary block
bootstrap when the sample size $n$ and the block length $L$ are large.
The situation is different for smaller sample sizes, since the
adjusted approach then will remove the majority of the corrupt tuples
that adds edge-effect noise into the estimation of the local Gaussian
autocorrelations $\lgacr{\LGp}{h}$.  This can be seen clearly by
comparing the expected amount of corrupt tuples when the block
bootstrap (see \cref{P2.app:table:corrupt.tuples.block.bootstrap} on
page \pageref{P2.app:table:corrupt.tuples.block.bootstrap}) with the
corresponding numbers for the revised resampling algorithm, as given
in \cref{P2.app:table:corrupt.tuples.ibb.bootstrap}.

\begin{table}[ht]
  \centering
  \begingroup\tiny
  \begin{tabular}{r|llllllllll}
    \hline
    $L$ \textbackslash\ $h$ & 1 & 2 & 3 & 4 & 5 & 6 & 7 & 8 & 9 & 10 \\ 
    \hline
    25 & 0.002\% & 0.006\% & 0.012\% & 0.020\% & 0.030\% & 0.043\% & 0.057\% & 0.073\% & 0.092\% & 0.112\% \\ 
    100 & 0.001\% & 0.002\% & 0.003\% & 0.005\% & 0.008\% & 0.011\% & 0.014\% & 0.019\% & 0.023\% & 0.028\% \\ 
    \hline
  \end{tabular}
  \endgroup
  \caption{Copy of
    \lgsdRef{app:table:corrupt.tuples.ibb.bootstrap}. The expected
    amount of corrupt tuples when $\lgacr{\LGp}{h}$ are estimated for
    the \texttt{dmbp}-data by the \textit{circular index-based block
      bootstrap for tuples}, cf.\
    \lgsdRef{def:index_based_block_bootstrap_for_tuples}.}
  \label{P2.app:table:corrupt.tuples.ibb.bootstrap}
\end{table}

A comparison of the entries in
\cref{P2.app:table:corrupt.tuples.block.bootstrap,P2.app:table:corrupt.tuples.ibb.bootstrap}
reveals that it, for a time series of length $n=1974$, is quite a
drastic reduction in the number of corrupt tuples when the block
bootstrap is replaced with the \textit{circular index-based block
  bootstrap for tuples} resampling strategy in
\lgsdRef{def:index_based_block_bootstrap_for_tuples}.  From the $h=10$
column it can be seen that the expected amount have been reduced from
39.7\% to 0.112\% when $L=25$, and it has been a reduction from 9.7\%
to 0.028\% when $L=100$.  The edge-effect noise are thus in this case
rather negligible when the adjusted resampling strategy is used.

\Cref{P2.app:table:corrupt.tuples.ibb.bootstrap} indicates that the
expected amount of corrupted tuples will be low even for short block
lengths $L$.  This implies that the pointwise confidence intervals for
the $m$-truncated local Gaussian autocorrelations
$\lgsdM{\LGp}{\omega}{m}$ should not be too sensitive to the block
length $L$ when the \textit{circular index-based block bootstrap for
  tuples} resampling strategy is used.

\subsection{Sensitivity analysis: The block length $L$}
\label{P2.app:Block_length_sensitivity}

The block length sensitivity for the adjusted resampling strategy from
\lgsdRef{def:index_based_block_bootstrap_for_tuples} was investigated
for two different cases in \lgsdRef{app:Block_length_sensitivity}.
The first case was based on the \texttt{dmbp}-data, whereas the second
case was based on resampling of a single sample generated by the local
trigonometric model initially used for the sanity testing of the
estimation algorithm.  For each case the $m=10$ truncated local
Gaussian spectral density was estimated, and then pointwise confidence
intervals were generated for each case based on resampling using the
\textit{circular index-based block bootstrap for tuples} strategy. 
The exact same results for both of the cases, i.e., the resulting
pointwise confidence intervals hardly changed when the block length
$L$ increased from 10 to 69.

The topic of interest for the present paper is the extension of the
local Gaussian spectral theory to the multivariate case, and the
script used for the univariate \texttt{dmbp}-case in
\lgsdRef{app:Block_length_sensitivity} was thus modified to cope with
the multivariate \EuStockMarkets-case --- and the results are shown in 
\cref{fig:EuStockMarkets_blocklength_P1_fig_F2,fig:EuStockMarkets_blocklength_P1_fig_F3}.
Note that
these figures are strikingly similar to those encountered in
\lgsdRef{app:Block_length_sensitivity}, and the discussion below is
with some minor adjustments the same as the one found there.

\Cref{fig:EuStockMarkets_blocklength_P1_fig_F2} contains two sequences
of boxplots, indexed by the block length $L$ which increases in steps
of 1 from $L=10$ to $L=69$.  The distance function $D$ mentioned in
\cref{P2.app:method_for_sensitivity_analysis} is used to generate the
values the boxplots are based on.  Keep in mind that the distance
function does not reveal anything about the frequency-component of the
cases under investigation, so it is also necessary to include a plot
that focus on that aspect for a few of the block lengths $L$, as seen
in \cref{fig:EuStockMarkets_blocklength_P1_fig_F3}.

\begin{figure}
  {\centering
    \includegraphics[width=1\linewidth]{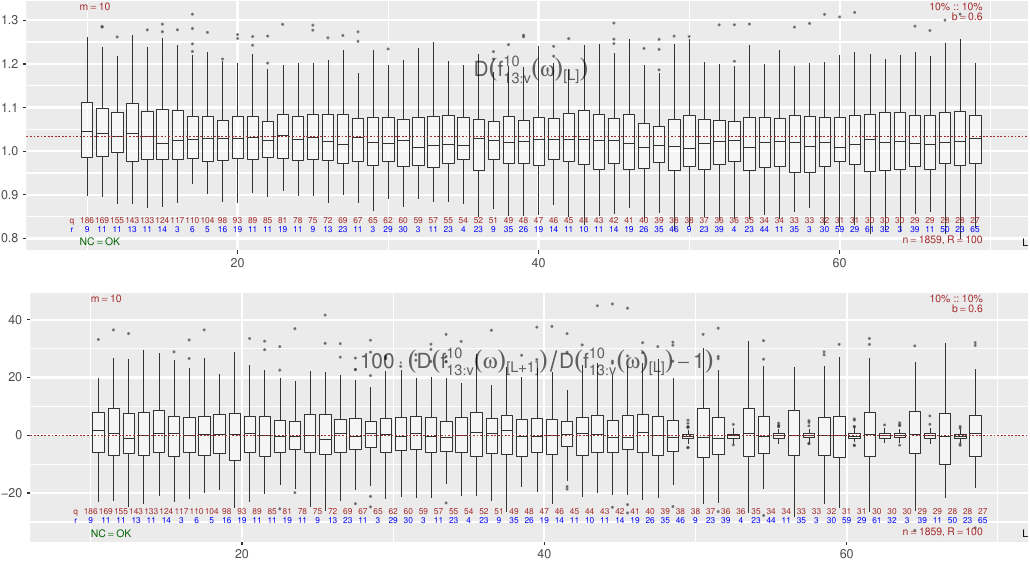}}
  \caption{Distance based box-plots for the investigation of the
    sensitivity of the block length $L$ for the adjusted resampling
    strategy from \lgsdRef{def:index_based_block_bootstrap_for_tuples}.
    The numbers in the two bottom rows show $q=\ceil{n/L}$ and
    $r=n-(q-1)\cdot L$, i.e.,\ the number of blocks and the length of
    the last block. }
  \label{fig:EuStockMarkets_blocklength_P1_fig_F2}
\end{figure}

\textbf{The panel at the top of
  \cref{fig:EuStockMarkets_blocklength_P1_fig_F2}:} A box-plot for the
$D\!\parenR{\lgcsdM{k\ell}{\LGp}{\omega}{10}_ {[L]}}$-values (based on
$R=100$ replicates) is given for each block length $L$.  A horizontal
red dashed line has been added that shows the
$D\!\parenR{\lgcsdM{k\ell}{\LGp}{\omega}{10}}$-value for the original
sample.  It can be seen that the medians of the box-plots tend to be
slightly smaller than the horizontal line that corresponds to the
value based on the original sample, but these medians are based on
$R=100$ replicates --- and another realisation might thus look
slightly different.  It does not seem to be any pattern here with
regard to how these box-plots changes when $L$ increases.

\textbf{The panel at the bottom of
  \cref{fig:EuStockMarkets_blocklength_P1_fig_F2}:} These box-plots
shows the percent-wise changes in the distances when the block length
goes from $L$ to $L+1$, and everything else is kept identical, i.e.,\
$100 \cdot \parenR{D\!\parenR{\lgcsdM{k\ell}{\LGp}{\omega}{10}_
    {[L+1]}}/D\!\parenR{\lgcsdM{k\ell}{\LGp}{\omega}{10}_ {[L]}}-1}$.  This is
possible to do since the reproducibility setup enables a tracking for
each individual realisation.

A horizontal red dashed line has been added at 0, and it is clear
that the median-part of these box-plots are quite close to this
horizontal line.  It can also be observed that some of these
box-plots are more compact than the other ones, and a simple
investigation of the numbers given at the bottom of the plots
reveals that this phenomenon occurs when an increase from $L$ to
$L+1$ does not reduce the number of blocks that are needed, i.e.,\
they occur when $\ceil{n/L} = \ceil{n/(L+1)}$.

For the individual bootstrapped time series, this indicates that the
changes are minimal when the number of blocks remains the same ---
whereas the changes are much larger when the increase of $L$ triggers
a reduction in the number of blocks.  However, as is evident from an
inspection of the panel at the top of
\cref{fig:EuStockMarkets_blocklength_P1_fig_F2}, this effect is only
on the level of the individual replicates, and it is averaged away
when a collection of replicates is considered.

Note that the effect noticed in the bottom panel of
\cref{fig:EuStockMarkets_blocklength_P1_fig_F2} also is present for
the global spectral densities (based on these bootstrapped samples),
so this phenomenon is thus not an artefact of the way the local
Gaussian spectral densities are estimated.

\textbf{The frequency-component:}
\cref{fig:EuStockMarkets_blocklength_P1_fig_F2} indicates that the
block length sensitivity, as measured by
$D\!\parenR{\lgcsdM{k\ell}{\LGp}{\omega}{m}}$, for the
\textit{circular index-based block bootstrap for tuples} resampling
strategy from \lgsdRef{def:index_based_block_bootstrap_for_tuples} is
rather small.  But does this imply that these block lengths should be
considered \textit{equally good} or \textit{equally bad}?  That can
not be concluded from \cref{fig:EuStockMarkets_blocklength_P1_fig_F2}
alone, and it is thus necessary to also consider a plot that takes the
frequency-dimension into account.
\Cref{fig:EuStockMarkets_blocklength_P1_fig_F3} does this, using the
\Co-spectra component, for the four block lengths
$L\in\parenC{10,25,50,69}$.

\begin{figure}
  {\centering
    \includegraphics[width=1\linewidth]{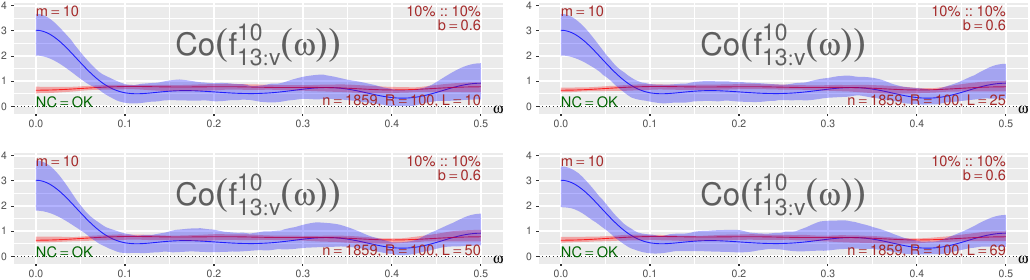}}

  \caption{Four different block lengths $L$ (from those investigated
    in \cref{fig:EuStockMarkets_blocklength_P1_fig_F2}) have here been
    used in the \textit{circular index-based block bootstrap for
      tuples} resampling strategy given in second part of
    \cref{P2.app:regarding_resampling}.  The values of $L$ are 10, 25,
    50 and 69, and this information is plotted at the lower right
    corner of the plots.  The observed differences are tiny.  }

  \label{fig:EuStockMarkets_blocklength_P1_fig_F3}
\end{figure}

It is clear from \cref{fig:EuStockMarkets_blocklength_P1_fig_F3} that
the differences between these estimates are rather small, and it is
necessary to look closely in order to see that the pointwise
confidence intervals are slightly wider near $\omega=0$ for the case
$L=50$.

\textbf{Conclusion:} These results, are as previously mentioned, in
complete agreement with the corresponding results from the univariate
case in \lgsdRef{app:Block_length_sensitivity}, i.e., the effect of
changing the block length $L$ is minuscule.  The interested reader can
in \lgsdRef{app:block.length.and.expected.content.of.resampled.data}
find an in depth discussion of why the block length
only plays a role as a tuning parameter in the \textit{circular
  index-based block bootstrap for tuples} resampling strategy.

\textbf{Reproducibility:} The scripts needed for the reproduction of
the plots in this section
are included in the \Rpackage \lgsdRpackage\ (see \cref{P2.app:The.scripts.in.lgsdRpackage}
for further details), and the interested reader can there easily
adjust the range of the block lengths to be used.  It is also
possible to adjust all the other tuning parameters needed for the
estimation of $\lgcsdM{k\ell}{\LGp}{\omega}{m}$, and it is even possible to
perform the computations with the ordinary block bootstrap if so
should be desired.

\section{Scripts and details related to the examples}
\label{P2:app:scripts_and_details_related_to_the_examples}
\setcounter{figure}{0} 

The reproducibility of all the examples in this paper can be done by
the scripts contained in the \Rpackage \lgsdRpackage, and
\cref{P2.app:The.scripts.in.lgsdRpackage} explains how the interested
reader can extract these scripts. 
This section also contains some further details/discussions related to
the examples in \cref{sec:lgch_examples} in the main paper.  Details
related to the bivariate Gaussian test-example are given in
\cref{sec:lgch-bivariate-Gaussian-white-noise}, whereas
\cref{P2.app:fig:GARCH} presents the R-code used to estimate the
parameters used for the cGARCH-example seen in \cref{fig:cGARCH}.

\Cref{app:Plots_of_the_complex-valued_spectra} contains a short
explanation of how animations of complex-valued plots can be used to
investigate the estimated local Gaussian cross-spectra, a part that
was moved here in order to improve the flow of the main part.

\cref{P2.app:fig:trigonometric} investigates the \textit{local
  bivariate trigonometric} example seen in
\cref{fig:heatmap_co_quad_dmt_bivariate_constant_phases,fig:Bivariate_local_trigonometric_A,fig:heatmap_co_quad_dmt_bivariate_different_phases,fig:Bivariate_local_trigonometric_C}
of \cref{sec:lgcsd_Some_simulations} in the main part.
This is quite similar to the corresponding discussion for the
univariate case in \lgsdRef{app:fig:trigonometric}, it adjusts to the
bivariate case the theoretical investigation of the general
construction of which the \textit{local trigonometric} examples are
particular realisations, and it also briefly summarises the heuristic
arguments from \JT that was used to sanity test the implemented
estimation algorithm.

Finally, once more mirroring the discussion in
\lgsdRef{app:data_details}, the last part of
\cref{P2.app:fig:trigonometric} verifies that it for a large sample is
possible to detect an elusive component that only occurs with
probability 0.05 in the \textit{local bivariate trigonometric}
example, and it ends with some comments related to issues that can
occur (under specific circumstances) when the local Gaussian machinery
is used on a time series whose global spectrum does not look like
white noise.

\subsection{The scripts in the \Rpackage \lgsdRpackage}
\label{P2.app:The.scripts.in.lgsdRpackage}

All the examples in this paper (and all the examples in \JT) can be
reproduced by the scripts in the \Rpackage \lgsdRpackage.  This
\Rpackage can be installed by using
\enquote{\devtoolgithublgsdRpackage}.  The simplest way to extract the
scripts from the internal storage of this \Rpackage is to use the
\Rfunction \enquote{\Rref{LG\_extract\_scripts()}} after the package
has been installed.

These scripts can either be used as they are (reproduction of the
examples in this paper), or they can be used as templates for similar
investigations of other samples/models that the user would like to
investigate.

The reproduction of the figures requires two different scripts.  The
first scripts contain the code needed for the estimation of the local
Gaussian auto-correlations $\lgccr{kk}{\LGp}{h}$ and
cross-correlations $\lgccr{k\ell}{\LGp}{h}$ for all the specified
combinations of the tuning parameters, whereas the second scripts
contain the code that creates the particular visualisations seen in
the figures in this paper. Note that it is sufficient to use the first
type of scripts in order to use the integrated
\Rref{shiny}-application that enables an easy interactive
investigation of the resulting estimates.  The second type of scripts
is first needed when one wants to put many figures into one larger
grid.

\subsection{The bivariate Gaussian white noise example}
\label{sec:lgch-bivariate-Gaussian-white-noise}

This section contains the details related to the bivariate Gaussian
white noise example, used for the sanity testing of the estimation
algorithm for the local Gaussian cross-spectra
$\lgcsd{k\ell}{\LGp}{\omega}$, that was referred to in the initial
discussion of \cref{sec:lgch_examples} in the main part.

It is natural to start with a Gaussian example, since it is the one
case where it is explicitly known what the true value of the the local
Gaussian cross-spectra $\lgcsd{k\ell}{\LGp}{\omega}$ is, in
particular, it is known from
\myref{th:lgcs_properties}{th:lgcs_equal_to_osd_when_Gaussian} that
the local Gaussian cross-spectrum coincide with the ordinary
cross-spectrum when the time series under investigation is Gaussian.

The plots in \cref{fig:Bivariate_Gaussian_WN} shows the \Co-, \Quad-
and \Phase-plots based on 100 independent samples of length 1859 from
a bivariate Gaussian distribution with standard normal marginals and
correlation~0.35.  The left column of \cref{fig:Bivariate_Gaussian_WN}
shows the situation for a point off the diagonal, whereas the right
column shows the situation for a point at the center of the~diagonal,
i.e., \mbox{$\LGpi{1}=\LGpi{2}=0$}.  Note that the global spectra are
identical for all the points, i.e.,\ the red components\footnote{%
  If you have a black and white copy of this paper, then read
  \enquote{red} as \enquote{dark} and \enquote{blue} as
  \enquote{light}.} are the same for each row
of~\cref{fig:Bivariate_Gaussian_WN}.

In this simple case, where the true values of the local Gaussian
versions of the spectra coincides with the ordinary global spectra, it
follows that the \Co-, \Quad- and \Phase-spectra (for any truncation
level~$m$) respectively should be the constants 0.35, 0 and~0.
\Cref{fig:Bivariate_Gaussian_WN} shows that the red and blue dotted
lines, that respectively represents the estimates of the global and
local $m$-truncated spectra,\footnote{%
  The dotted lines represents the means of the estimated values,
  whereas the 90\% pointwise confidence intervals are based on the 5\%
  and 95\% quantiles of these samples.} seems to match these true
values quite reasonably --- and this provides a sanity check of the
code that generated these plots.  Note that the 90\% pointwise
confidence interval for the local Gaussian versions (blue ribbons) are
wider than those for the ordinary spectra (red ribbons) since the
bandwidth used for the estimation of the local Gaussian
cross-correlations, in this case \mbox{$\bm{b} = \parenR{0.6,0.6}$},
reduces the number of observations that effectively contributes to the
computation of the local Gaussian~spectra.

\begin{figure}

{\centering \includegraphics[width=1\linewidth]{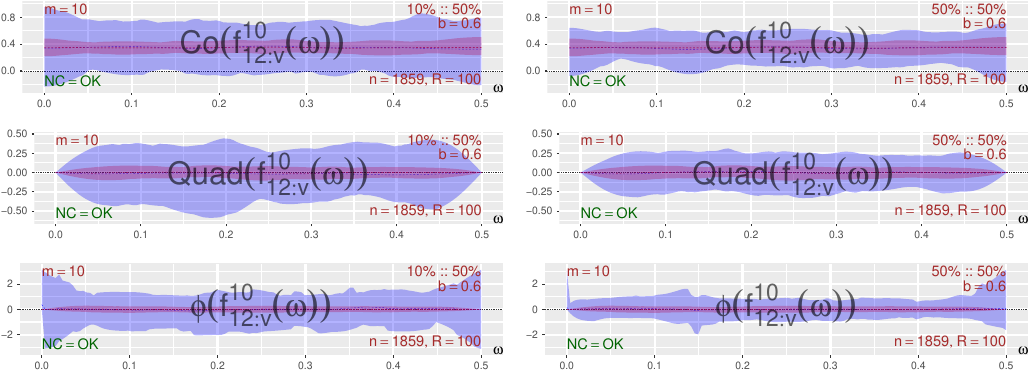} 

}

\caption[i.i.d]{i.i.d.\ bivariate Gaussian white noise.}

\caption{\Co-, \Quad- and \Phase-plots for two points: Samples from
  the bivariate Gaussian white noise model, for which it is known that
  the local Gaussian spectra are equal to the ordinary spectra.  The
  equality of the estimates are clear from this figure, the wider
  pointwise confidence intervals for the local Gaussian spectra is a
  consequence of the kernel-function in the estimation algorithm.  }

\label{fig:Bivariate_Gaussian_WN}
\end{figure}

\subsection{The cGARCH-example in \cref{fig:cGARCH}}
\label{P2.app:fig:GARCH}

The cGARCH-example seen in \cref{fig:cGARCH} (see also
\cref{fig:cGARCH_YiYj,fig:cGARCH_lags20}) had coefficients that were
fitted to the pseudo-normalised log-returns of the
\EuStockMarkets-data by the help of the \Rref{rmgarch}-package
\citet{ghalanos15:_rmgarch}.  The \enquote{c} in cGARCH implies that a
Copula-based approach is used for the modelling of the
interdependency-structure, and the following steps can be used to fit
a Copula-GARCH model to a multivariate sample:
\begin{enumerate}
\item %
  For each univariate subset of the sample, create a
  \Rref{uGARCHspec}-object (univariate GARCH specification) by using
  the \Rref{ugarchspec}-function from the \Rref{rugarch}-package.
\item %
  Create a \Rref{uGARCHmultispec}-object (univariate multiple GARCH
  specification) by using the \Rref{multispec}-function from the
  \Rref{rugarch}-package on the list of univariate GARCH
  specifications obtained in the first step.
\item %
  Create a \Rref{cGARCHspec}-object (Copula-GARCH specification) by
  using the \Rref{cgarchspec}-function from the \Rref{rmgarch}-package
  on the \Rref{uGARCHmultispec}-object obtained in the second step.
\item %
  Create a \Rref{cGARCHfit}-object (Copula-GARCH fit object) by using
  the \Rref{cgarchfit}-function from the \Rref{rmgarch}-package.  The
  \Rref{spec}-argument should be the \Rref{cGARCHspec}-object from the
  third step, and the \Rref{data}-argument should be the multivariate
  sample under investigation.
\end{enumerate}

The model and parameters stored in the \Rref{cGARCHfit}-object can now
be used as the \Rref{fit}-argument of the \Rref{cgarchsim}-function
from the \Rref{rmgarch}-package.  This implies that samples from the
fitted model can be generated, from which it is possible to find
estimates and pointwise confidence intervals for the $m$-truncated
local Gaussian auto- and cross-spectra for this particular model.

Note that the \Rref{ugarchspec}, \Rref{uGARCHmultispec} and
\Rref{cgarchfit} all have many arguments, and tweaking of these
arguments enables a plethora of different GARCH-type models to be
investigated.  The cGARCH-model used for the creation of
\cref{fig:cGARCH} used the default arguments in the previously
mentioned functions, and this was done since the desired target was to
produce a simple proof-of-concept model for the local Gaussian sanity
testing presented in \cref{P2.app:resampling_parametric_bootstrap}.
It was no surprise that
\cref{fig:BIRR_P1_fig_F1_Y1Y1,fig:BIRR_P1_fig_F1_Y3Y3,fig:BIRR_Co_P1_fig_F1_Y1Y3,fig:BIRR_Quad_P1_fig_F1_Y1Y3,fig:BIRR_amplitude_P1_fig_F1_Y1Y3,fig:BIRR_phase_P1_fig_F1_Y1Y3}
revealed that the local Gaussian spectra based on this
\enquote{trivial} cGARCH-model had some clear deviations from the
local Gaussian spectra based on the \EuStockMarkets data.

The full list of default arguments for the \Rfunction{s}
\Rref{ugarchspec}, \Rref{uGARCHmultispec} and \Rref{cgarchfit} will
not be included here, since that would require several pages of code
and explanations. The interested reader can look these details up in
the documentation of the \Rref{rugarch}- and \Rref{rmgarch}-packages.

A final comment: It is of course possible to follow the strategy used
in \JT, where an apARCH(2,3)-model was fitted to the
\texttt{dmbp}-data after a procedure that tested several thousand
different variations of the GARCH-type models implemented in the
\Rref{rugarch}-package.  This strategy was not adopted in the present
paper since the multivariate nature of the \EuStockMarkets-data, and
the amount of arguments to tweak in the functions \Rref{ugarchspec},
\Rref{uGARCHmultispec} and \Rref{cgarchfit}, would have required a
rather big computational investment for a sample that ended out being
used simply due to it being available in R.

\subsection{Plots of the complex-valued spectra}
\label{app:Plots_of_the_complex-valued_spectra}

This section contains a discussion related to animations of
complex-valued plots, and how these might be used to investigate the
estimated local Gaussian cross-spectra.  This was initially contained
in the main part, but it was then moved here in order to improve the
flow of that part.  Before reading on, recall the following from the
main part:

\begin{blockquote} %
  {\textbf{The reference case:} The heuristic argument needed for the
    bivariate case is identical in structure to the one used in the
    univariate case, and for the present investigation the reference
    for the plots later on is based on the following simple bivariate
    model,
    \begin{align*}
      \Yz{1,t}=\cos\parenR{2\pi\alpha t + \phi} + \wz{1,t} \text{ and }
      \Yz{2,t}=\cos\parenR{2\pi\alpha t + \phi + \theta} + \wz{2,t},
    \end{align*}
    where $\wz{i,t}$ is Gaussian white noise with mean zero and
    standard deviation~$\sigma$, with $\wz{1,t}$ and $\wz{2,t}$
    independent, and where it in addition is such that~$\alpha$
    and~$\theta$ are fixed for all the replicates whereas $\phi$ is
    drawn uniformly from \mbox{$[0,2\pi)$} for each individual
    replicate.  A realisation with \mbox{$\sigma = 0.75$},
    \mbox{$\alpha=0.302$} and \mbox{$\mbox{$\theta = \pi/3$}$} has
    been used for the \Co-, \Quad-, and \Phase-plots shown in
    \cref{fig:Bivariate_global_cosine}, where 100 independent samples
    of length 1859 were used to get the estimates of the $m$-truncated
    spectra and their corresponding 90\% pointwise
    confidence~intervals (based on the bandwidth
    \mbox{$\bm{b}=\parenR{0.6 ,0.6}$}).  Some useful remarks can be
    based on this plot, before \textit{the bivariate local
      trigonometric case} is defined and investigated.  }
\end{blockquote}

It can be enlightening to compare the \Co-, \Quad- and \Phase-plots in
\cref{fig:Bivariate_global_cosine} with a plot that shows the
underlying estimates upon which the pointwise
confidence intervals were based.  Such a plot is shown in
\cref{fig:Bivariate_global_cosine_lag_frequency}, where the left panel
presents the complex-valued estimates of the local Gaussian
cross-spectrum at the frequency~\mbox{$\omega=\alpha$}, and where
means and quantiles relative to a polar representation, i.e.,\
\mbox{$z = re^{i\theta}$}, have been added to the plot.

\begin{figure}

{\centering \includegraphics[width=1\linewidth]{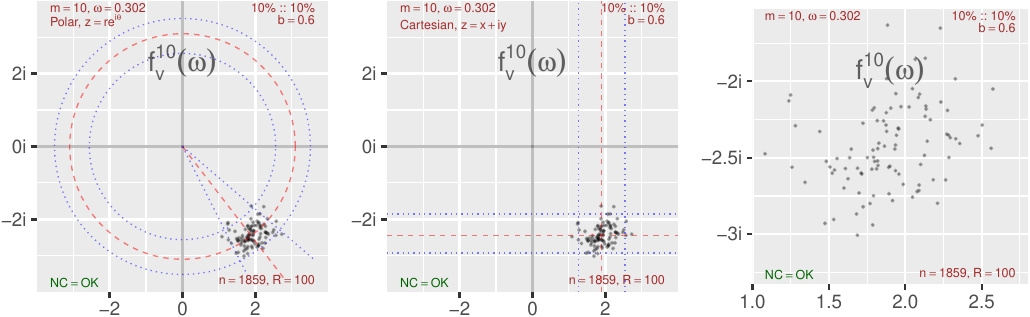} 

}

\caption{Complex-valued representation of 100 samples of
  $\lgcsd{12}{\LGp}{\omega}$ from \cref{eq:bivariate_cosine_example},
  at the peak frequency $\omega = 0.302$. Left panel: Pointwise 90\%
  confidence bands based on polar representation.  Center panel:
  Pointwise 90\% confidence bands based on Cartesian representation.
  Right panel: Zoomed in
  plot.}
\label{fig:Bivariate_global_cosine_lag_frequency}
\end{figure}

The center panel of \cref{fig:Bivariate_global_cosine_lag_frequency}
shows the same estimated values, but this time the means/quantiles are
based on a Cartesian representation \mbox{$z=x+iy$}.  These two panels
gives a geometrical view to the observations presented above.
The third panel of \cref{fig:Bivariate_global_cosine_lag_frequency}
presents a zoomed in version of the estimated values of the local
Gaussian cross-spectrum, and it gives a reminder that it in principle
is possible to extract more information from these estimates than what
has been done so~far.  A closer inspection of these estimates could
e.g.\ be used to see how much they (for the given
\mbox{$m$-truncation}) deviate from the expected asymptotic
distributions that was given in~\cref{seq:lgcsd:convergence_theorems}.

\subsection{The local bivariate trigonometric examples in
  \cref{fig:heatmap_co_quad_dmt_bivariate_constant_phases,fig:Bivariate_local_trigonometric_A,fig:heatmap_co_quad_dmt_bivariate_different_phases,fig:Bivariate_local_trigonometric_C}}
\label{P2.app:fig:trigonometric}

This section will discuss some topics related to the \textit{local
  bivariate trigonometric} examples, whose local Gaussian spectral
density were investigated in
\cref{fig:heatmap_co_quad_dmt_bivariate_constant_phases,fig:Bivariate_local_trigonometric_A,fig:heatmap_co_quad_dmt_bivariate_different_phases,fig:Bivariate_local_trigonometric_C}
of \cref{sec:lgcsd_Some_simulations} in the main part.  %
This discussion is quite similar to the one given for the univariate
\textit{local trigonometric} examples in
\lgsdRef{app:fig:trigonometric}, and the interested reader should
consult this source for some of the in depth discussions.

The \textit{local bivariate trigonometric} examples are a special case
of a more general bivariate construction, which follows from an
extension of the corresponding univariate construction given in
\lgsdRef{app:fig:trigonometric.general.properties}.  The details are
slightly technical, and they have for the benefit of the reader been
included in \cref{P2.app:fig:trigonometric.general.properties}.

A brief discussion of the heuristic arguments that explains why the
\textit{local bivariate trigonometric} examples can be used to sanity
test the implemented estimation algorithm are given in
\cref{P2.app:fig:trigonometric.heuristic.argument}.  An in depth
discussion can be found in
\lgsdRef{app:fig:trigonometric.heuristic.argument}.

The \textit{local bivariate trigonometric} examples that was used as
the basis for the plots in
\cref{fig:heatmap_co_quad_dmt_bivariate_constant_phases,fig:Bivariate_local_trigonometric_A,fig:heatmap_co_quad_dmt_bivariate_different_phases,fig:Bivariate_local_trigonometric_C},
contained four components, of which one occurred with a lower
probability.  This setup was inherited from the \textit{local
  unvariate trigonometric} examples in \JT, and it was there seen
(cf.\ \cref{fig:local.trigonometric.C1-component} in
\cref{app:fig:trigonometric.C1.component}) that the \enquote{lower
  tail}-component could be detected when the sample size was large
enough.  \Cref{P2.app:fig:trigonometric.C1.component} presents a
similar investigation for the \textit{local bivariate trigonometric}
examples.

The discussion in \JT also considered issues that could occur if the
$m$-truncated local Gaussian spectrum was estimated for samples from a
deterministic function perturbed by very low random fluctuations, and
it was there seen (cf.\
\lgsdRef{sec:lgch:beware_of_global_structures}) that the local
Gaussian machinery did not work well for such cases.  This discussion
will not be duplicated in this paper, but the key observations can be
mentioned here: It might be safest to apply the local Gaussian
spectral machinery to time series with rather flat global spectra.  If
a \enquote{clear} global structure is present, then it might be
preferable to use the local Gaussian spectral machinery on the
residuals that are left after some suitable model have been fitted to
the data.

\subsubsection{Some properties of the general construction}
\label{P2.app:fig:trigonometric.general.properties}

The \textit{local bivariate trigonometric} examples considered in
\cref{sec:lgcsd_Some_simulations} are (as mentioned there) particular
cases of a general construction that is based on the key idea that an
artificial bivariate time series
$\TSR{\parenR{\Yz{1,t},\Yz{2,t}}}{t\in\ZZ}{}$ can be constructed by
the following scheme:
\begin{enumerate}
\item Select $r\geq 2$ bivariate time series
  $\TSR{\parenR{\Cz{1,i}(t),\Cz{2,i}(t)}}{i=1}{r}$.
\item Select a random variable $I$ with values in the set
  $\parenC{1,\dotsc,r}$, and use this to sample a collection of
  indices $\TSR{\Iz{t}}{t\in\ZZ}{}$ (that is, for each $t$ an
  independent realization of $I$ is taken).  Let
  $\pz{i}\defeq\Prob{\Iz{i}=i}$ denote the probabilities for the
  different outcomes.
\item Define $\Yz{t}$ by
  means of the equation
  \begin{subequations}
    \label{app.eq:lgch:Y1Y2_local_trigonometric}
    \begin{align}
      \label{app.eq:lgch:Y1_local_trigonometric}
      \Yz{1,t} &\defeq \sumss{i=1}{r} \Cz{1,i}(t) \cdot\Ind{\Iz{t} = i},\\
      \label{app.eq:lgch:Y2_local_trigonometric}
      \Yz{2,t} &\defeq \sumss{i=1}{r} \Cz{2,i}(t) \cdot\Ind{\Iz{t} = i}.
    \end{align}
  \end{subequations}
\end{enumerate}

The basic properties of $\TSR{\parenR{\Yz{1,t},\Yz{2,t}}}{t\in\ZZ}{}$
can be expressed relatively those of
$\TSR{\parenR{\Cz{1,i}(t),\Cz{2,i}(t)}}{i=1}{r}$, as seen in the
following result: 

\begin{lemma} %
  \label{P2.app:lemma:properties.of.the.artificial.time.series}
  With $\TSR{\parenR{\Yz{1,t},\Yz{2,t}}}{t\in\ZZ}{}$ as defined above,
  it follows that:
  \begin{enumerate}%[label=(\alph*)]
  \item \label{P2.app:lemma:properties.of.the.artificial.time.series..E}
    $\E{\Yz{k,t}} = \sumss{i=1}{r} \pz{i}\cdot\E{\Cz{k,i}(t)}$
  \item \label{P2.app:lemma:properties.of.the.artificial.time.series..E2}
    $\E{\Yz{k,t+h}\cdot\Yz{\ell,t}} =
    \begin{cases}
      \sumss{i=1}{r}\sumss{j=1}{r} \pz{i}\cdot\pz{j}\cdot\E{\Cz{k,i}(t+h)\cdot\Cz{\ell,j}(t)}
      & h\neq 0 \\
      \sumss{i=1}{r} \pz{i}\cdot\E{\Cz{k,i}(t)\cdot\Cz{\ell,j}(t)} & h=0
    \end{cases}$
  \item \label{P2.app:lemma:properties.of.the.artificial.time.series..Cov}
    $\Cov{\Yz{k,t+h}}{\Yz{\ell,t}} =
    \begin{cases}
      \sumss{i=1}{r} \sumss{j=1}{r} \pz{i}\cdot\pz{j}\cdot
      \Cov{\Cz{k,i}(t+h)}{\Cz{\ell,j}(t)} & h \neq 0 \\
      \sumss{i=1}{r} \pz{i}\cdot\E{\Cz{k,i}(t)\cdot\Cz{\ell,i}(t)} - \\
      \qquad \parenR{\sumss{i=1}{r}
        \pz{i}\cdot\E{\Cz{k,i}(t)}}\cdot\parenR{\sumss{j=1}{r}
        \pz{j}\cdot\E{\Cz{\ell,j}(t)}} & h = 0
    \end{cases}$
  \item \label{P2.app:lemma:properties.of.the.artificial.time.series..Cov..independent}
    For $k,\ell\in\parenC{1,2}$, the additional assumption that
    $\parenR{\Cz{1,i}(t),\Cz{2,i}(t)}$ and
    $\parenR{\Cz{1,j}(t),\Cz{2,j}(t)}$ are independent when $i\neq j$,
    simplifies the $h\neq0$ case to:
    $\Cov{\Yz{k,t+h}}{\Yz{\ell,t}}=\sumss{i=1}{r} \pz[2]{i}\cdot
    \Cov{\Cz{k,i}(t+h)}{\Cz{\ell,i}(t)}$.
  \end{enumerate}
\end{lemma}

\begin{proof}
  The results in
  \cref{P2.app:lemma:properties.of.the.artificial.time.series..E,P2.app:lemma:properties.of.the.artificial.time.series..E2,P2.app:lemma:properties.of.the.artificial.time.series..Cov,P2.app:lemma:properties.of.the.artificial.time.series..Cov..independent}
  are in essence the same as those seen in
  \lgsdRef{app:lemma:properties.of.the.artificial.time.series}, with
  some extra indices $k$ and $\ell$ added to it.  The proof is short,
  and it is repeated below.
  
  The random variable $\Iz{t}$ that produces the set of indices is
  independent of $\Cz{k,i}(t)$ and $\Cz{\ell,j}(t)$, and
  \cref{app:lemma:properties.of.the.artificial.time.series..E} thus
  follows without further ado.  For the $h\neq0$ case of
  \cref{app:lemma:properties.of.the.artificial.time.series..E2} it is
  sufficient to note that $\Iz{t+h}$ and $\Iz{t}$ then are
  independent, and it follows that
  $\E{\Ind{\Iz{t+h}=i}\cdot\Ind{\Iz{t}=j}} =
  \E{\Ind{\Iz{t+h}=i}}\cdot\E{\Ind{\Iz{t}=j}} =
  \Prob{\Iz{t+h}=i}\cdot\Prob{\Iz{t}=j} = \pz{i}\cdot\pz{j}$. For the
  $h=0$ case of
  \cref{app:lemma:properties.of.the.artificial.time.series..E2} it is
  enough to note that $\Ind{\Iz{t}=i}\cdot\Ind{\Iz{t}=j}=0$ when
  $i\neq j$, which together with
  $\Ind{\Iz{t}=i}\cdot\Ind{\Iz{t}=i}=\Ind{\Iz{t}=i}$ gives the
  required expression. The statements in
  \cref{app:lemma:properties.of.the.artificial.time.series..Cov,app:lemma:properties.of.the.artificial.time.series..Cov..independent}
  follows trivially from those in
  \cref{app:lemma:properties.of.the.artificial.time.series..E,app:lemma:properties.of.the.artificial.time.series..E2}.
\end{proof}

The key idea in the local bivariate trigonometric example is that the
$r$ bivariate time series $\parenR{\Cz{1,i}(t),\Cz{2,i}(t)}$ all should
be \enquote{connected cosine-pairs with some noise}, since this
implies (given a reasonable parameter configuration) that it should be
possible to present a decent guesstimate with regard to the expected
shape of the $m$-truncated local Gaussian cross-spectrum density (for
some carefully selected tuning parameters of the estimation
algorithm).  The global cross-spectrum in this case will not be flat,
but it will for low truncation levels be \enquote{flat enough} for the
purpose of showing that the global spectrum does not detect the
underlying frequencies whereas the local Gaussian cross-spectrum can
do that task.

The following result reiterates the
$\parenR{\Cz{1,i}(t),\Cz{2,i}(t)}$-definition used in the local
bivariate trigonometric example, and it presents some basic properties
related to this definition.

\begin{lemma}
  \label{P2.app:lemma:properties.of.the.artificial.time.series..cosine.part}
  Let the bivariate random variables
  $\parenR{\Cz{1,i}(t),\Cz{2,i}(t)}$, for $i=1,\dotsc,r$, be defined by
  $\Cz{1,i}(t) = \Lz{i} + \Az{i}(t) \cdot \cos \left(2\pi\alphaz{i} t
    + \phiz{i} \right)$ and
  $\Cz{2,i}(t) = \Lz{i} + \Az{i}(t) \cdot \cos \left(2\pi\alphaz{i} t
    + \phiz{i} + \thetaz{i} \right)$ in the following manner: $\Lz{i}$
  and $\alphaz{i}$ are constants that respectively defines the
  horizontal base-line and the frequency.  The amplitude $\Az{i}(t)$
  are for each $t$ uniformly distributed on an interval
  $\parenS{\az{i},\bz{i}}$, and $\Az{i}(t+h)$ and $\Az{i}(t)$ are
  independent when $h\neq0$. The phase-adjustment $\phiz{i}$ are
  uniformly drawn (one time for each realisation) from the interval
  between~$0$ and~$2\pi$, whereas the phases $\thetaz{i}$ are
  constants.  It is also assumed that the stochastic processes
  $\phiz{i}$ and $\Az{i}(t)$ are independent of each other.
  {\small  
    \begin{enumerate}%[label=(\alph*)]
    \item \label{P2.app:lemma:properties.of.the.artificial.time.series..cosine.part.E}
      $\E{\Cz{k,i}(t)} = \Lz{i}$, for $k\in\parenC{1,2}$.
    \item \label{P2.app:lemma:properties.of.the.artificial.time.series..cosine.part.E.05}
      $\E{\Cz{k,i}(t+h)\Cz{\ell,j}(t)} = \Lz{i}\Lz{j}$, for
      $k,\ell\in\parenC{1,2}$ when $i\neq j$.
    \item \label{P2.app:lemma:properties.of.the.artificial.time.series..cosine.part.Cov.05}
      $\Cov{\Cz{k,i}(t+h)}{\Cz{\ell,j}(t)} = 0$, for
      $k,\ell\in\parenC{1,2}$ when $i\neq j$.
    \item \label{P2.app:lemma:properties.of.the.artificial.time.series..cosine.part.E2_k=l}
      $\E{\Cz{k,i}(t+h)\cdot\Cz{k,i}(t)} =
      \begin{cases}
        \Lz[2]{i} + \frac{1}{4}\cdot
        \parenR{\az[2]{i}+2\az{i}\bz{i}+\bz[2]{i}}\cdot
        \cos(2\pi\alphaz{i} h) & h \neq 0, k\in\parenC{1,2}\\
        \Lz[2]{i} + \frac{1}{3}\cdot
        \parenR{\az[2]{i}+\az{i}\bz{i}+\bz[2]{i}} & h = 0, k\in\parenC{1,2}
      \end{cases}$
    \item  \label{P2.app:lemma:properties.of.the.artificial.time.series..cosine.part.Cov_k=l}
      $\Cov{\Cz{k,i}(t+h)}{\Cz{k,i}(t)} =
      \begin{cases}
        \frac{1}{4}\cdot
        \parenR{\az[2]{i}+2\az{i}\bz{i}+\bz[2]{i}}\cdot
        \cos(2\pi\alphaz{i} h) &h\neq 0, k\in\parenC{1,2}\\
        \frac{1}{3}\cdot
        \parenR{\az[2]{i}+\az{i}\bz{i}+\bz[2]{i}} &h= 0, k\in\parenC{1,2}
      \end{cases}$
    \item \label{P2.app:lemma:properties.of.the.artificial.time.series..cosine.part.E2_12}
      $\E{\Cz{k,i}(t+h)\cdot\Cz{\ell,i}(t)} =
      \begin{cases}
        \Lz[2]{i} + \frac{1}{4}\cdot
        \parenR{\az[2]{i}+2\az{i}\bz{i}+\bz[2]{i}}\cdot
        \cos(2\pi\alphaz{i} h+\thetaz{i}) & h \neq 0, k=1, \ell=2\\
        \Lz[2]{i} + \frac{1}{4}\cdot
        \parenR{\az[2]{i}+2\az{i}\bz{i}+\bz[2]{i}}\cdot
        \cos(2\pi\alphaz{i} h-\thetaz{i}) & h \neq 0, k=2, \ell=1\\
        \Lz[2]{i} + \frac{1}{3}\cdot
        \parenR{\az[2]{i}+\az{i}\bz{i}+\bz[2]{i}}\cdot
        \cos(\thetaz{i}) & h = 0, k,\ell\in\parenC{1,2}
      \end{cases}$
    \item  \label{P2.app:lemma:properties.of.the.artificial.time.series..cosine.part.Cov_21}
      $\Cov{\Cz{k,i}(t+h)}{\Cz{\ell,i}(t)} =
      \begin{cases}
        \frac{1}{4}\cdot
        \parenR{\az[2]{i}+2\az{i}\bz{i}+\bz[2]{i}}\cdot
        \cos(2\pi\alphaz{i} h+\thetaz{i}) & h \neq 0, k=1, \ell=2\\
        \frac{1}{4}\cdot
        \parenR{\az[2]{i}+2\az{i}\bz{i}+\bz[2]{i}}\cdot
        \cos(2\pi\alphaz{i} h-\thetaz{i}) & h \neq 0, k=2, \ell=1\\
        \frac{1}{3}\cdot
        \parenR{\az[2]{i}+\az{i}\bz{i}+\bz[2]{i}}\cdot
        \cos(\thetaz{i}) & h = 0, k,\ell\in\parenC{1,2}
      \end{cases}$
    \end{enumerate}
  }
\end{lemma}

\begin{proof}
  The independence of $\Az{i}(t)$ and $\phiz{i}$, together with some
  simple results related to the expectations of functions based on the
  uniformly distributed random variables $\Az{i}$ and $\phi$ is needed
  in order to obtain these results.  

  The assumption that $\Az{i}$ is uniformly distributed on the
  interval $\parenR{\az{i},\bz{i}}$ implies that its density function
  is given by $\tfrac{1}{\bz{i}-\az{i}}$ on this interval (and 0
  outside of it).  From this it is easy to find
  $\E{\Az{i}(t)}=\tfrac{1}{2}\cdot\parenR{\az{i}+\bz{i}}$ and
  $\E{\Az[2]{i}(t)}=\tfrac{1}{3} \cdot
  \parenR{\az[2]{i}+\az{i}\bz{i}+\bz[2]{i}}$.
  Moreover, since $\Az{i}(t+h)$ and $\Az{i}(t)$ are independent when
  $h\neq0$, it also follows that
  $\E{\Az{i}(t+h)\Az{i}(t)}=\frac{1}{4}\parenR{\az[2]{i}+2\az{i}\bz{i}+\bz[2]{i}}$
  in this case.
  
  Similarly, the assumption that $\phiz{i}$ is uniformly distributed
  on the interval $\parenR{0,2\pi}$ implies that its density function
  is given by $\tfrac{1}{2\pi}$ on the interval $\parenR{0,2\pi}$ (and
  0 outside of it). From this it follows that
  $\E{\cos\left(\phiz{i}\right)} = \E{\sin\left(\phiz{i}\right)} =
  \E{\cos\left(\phiz{i}\right)\sin\left(\phiz{i}\right)} = 0$ and
  $\E{\cos\left(\phiz{i}\right)^2} = \E{\sin\left(\phiz{i}\right)^2} =
  \tfrac{1}{2}$.  In addition to this, since $\phiz{i}$ and $\phiz{j}$
  are independent when $i\neq j$, it also follows that
  $\E{\cos\left(\phiz{i}\right)\sin\left(\phiz{j}\right)} = 0$.  From
  this it follows that
  $\E{\cos \left(2\pi\alphaz{i} t + \phiz{i} \right)} = \E{\cos
    \left(2\pi\alphaz{i} t + \phiz{i} +\thetaz{i}\right)} = 0$.

  The trigonometric formula
  $\cos(u\pm v)=\cos(u)\cos(v)\mp\sin(u)\sin(v)$ can be used to show
  that the expectation of a product of the form
  $\cos\parenR{\uz{1}+\phiz{i}}\cos\parenR{\uz{2}+\phiz{j}}$ is equal
  to 0 when $i\neq j$, whereas it is equal to
  $\tfrac{1}{2}\cos\parenR{\uz{1}-\uz{2}}$ when $i=j$.  From this last
  observation it is easy to conclude the following:
  \begin{subequations}
    \label{eq:expectations_for_Ci_lemma}
    \begin{align}
      \label{eq:expectations_for_Ci_lemma_h}
      \E{\cos\parenR{2\pi\alphaz{i}(t+h)+\phiz{i}}
      \cos\parenR{2\pi\alphaz{i}t+\phiz{i}}}
      &=\frac{1}{2}\cos\parenR{2\pi\alphaz{i}h}\\
      \label{eq:expectations_for_Ci_lemma_h2}
      \E{\cos\parenR{2\pi\alphaz{i}(t+h)+\phiz{i}+\thetaz{i}}
      \cos\parenR{2\pi\alphaz{i}t+\phiz{i}+\thetaz{i}}}
      &=\frac{1}{2}\cos\parenR{2\pi\alphaz{i}h}\\
      \label{eq:expectations_for_Ci_lemma_h_theta_not_same_part}
      \E{\cos\parenR{2\pi\alphaz{i}(t+h)+\phiz{i}}
      \cos\parenR{2\pi\alphaz{i}+\phiz{i}+\thetaz{i}}}
      &=\frac{1}{2}\cos\parenR{2\pi\alphaz{i}h - \thetaz{i}} \\
      \label{eq:expectations_for_Ci_lemma_h_theta_same_part}
      \E{\cos\parenR{2\pi\alphaz{i}(t+h)+\phiz{i}+\thetaz{i}}
      \cos\parenR{2\pi\alphaz{i}+\phiz{i}}}
      &=\frac{1}{2}\cos\parenR{2\pi\alphaz{i}h + \thetaz{i}} 
    \end{align}
  \end{subequations}
  
  From these observations it is trivial to verify
  \cref{P2.app:lemma:properties.of.the.artificial.time.series..cosine.part.E,P2.app:lemma:properties.of.the.artificial.time.series..cosine.part.E.05},
  and
  \cref{P2.app:lemma:properties.of.the.artificial.time.series..cosine.part.Cov.05}
  follows immediately from these two.  The statements in
  \cref{P2.app:lemma:properties.of.the.artificial.time.series..cosine.part.E2_k=l,P2.app:lemma:properties.of.the.artificial.time.series..cosine.part.E2_12}
  are easily verified by multiplying the relevant expressions for
  $\Cz{1,i}(t) = \Lz{i} + \Az{i}(t) \cdot \cos \left(2\pi\alphaz{i} t
    + \phiz{i} \right)$ and
  $\Cz{2,i}(t) = \Lz{i} + \Az{i}(t) \cdot \cos \left(2\pi\alphaz{i} t
    + \phiz{i} + \thetaz{i} \right)$, and then using the above
  mentioned results related to $\Az{i}(t)$ and $\phiz{i}$ to establish
  the expected values.  The covariance-results in
  \cref{P2.app:lemma:properties.of.the.artificial.time.series..cosine.part.Cov_k=l,P2.app:lemma:properties.of.the.artificial.time.series..cosine.part.Cov_21}
  now follows trivially.
\end{proof}

Finally, the \textit{local bivariate trigonometric} example is
obtained by using $r$ bivariate time series
$\parenR{\Cz{1,i}(t),\Cz{2,i}(t)}$, of the form given in
\cref{P2.app:lemma:properties.of.the.artificial.time.series..cosine.part},
in the construction of the bivariate time series
$\parenR{\Yz{1,t},\Yz{2,t}}$, i.e.,\
\begin{subequations}
  \label{P2.eq:THE.local.trigonometric.example}
  \begin{align}
    \label{P2.eq:THE.local.trigonometric.example_1}
    \Yz{1,t}
    &\defeq \sumss{i=1}{r} \Ind{\Iz{t} = i}
      \cdot\parenS{\Lz{i} + \Az{i}(t) \cdot \cos \left(2\pi\alphaz{i} t
      + \phiz{i} \right)},\\
    \label{P2.eq:THE.local.trigonometric.example_1}
    \Yz{2,t}
    &\defeq \sumss{i=1}{r} \cdot\Ind{\Iz{t} = i}
      \cdot\parenS{\Lz{i} + \Az{i}(t) \cdot \cos \left(2\pi\alphaz{i} t
      + \phiz{i} + \thetaz{i} \right)},
  \end{align}
\end{subequations}
where it furthermore is assumed that the $i$-indexed stochastic
variables $\Az{i}(t)$ and $\varphiz{i}$ are independent of the
$j$-indexed variants when $i\neq j$.  It now follows from
\cref{P2.app:lemma:properties.of.the.artificial.time.series,P2.app:lemma:properties.of.the.artificial.time.series..cosine.part}
that the $h\neq 0$ auto-correlations of the bivariate time series
$\parenR{\Yz{1,t},\Yz{2,t}}$ in
\cref{P2.eq:THE.local.trigonometric.example} is given by
\begin{align}
  \label{P2.app:eq:auto-correlation.of.the.artificial.time.series}
  \rhoz{11}(h) = \rhoz{22}(h)
  &= \frac{\frac{1}{4}\cdot\sumss{i=1}{r}
    \pz[2]{i}\cdot\parenR{\az[2]{i} + 2\az{i}\bz{i} +
    \bz[2]{i}}\cdot\cos(2\pi\alphaz{i} h) }{\sumss{i=1}{r} \pz{i}\cdot \parenS{\Lz[2]{i} + \frac{1}{3}\cdot
    \parenR{\az[2]{i}+\az{i}\bz{i}+\bz[2]{i}}} -
    \parenRz[2]{}{\sumss{i=1}{r} \pz{i}\cdot\Lz{i}}}.
\end{align}
The cross-correlations $\rhoz{\Yz{12}}(h)$ and $\rhoz{\Yz{12}}(h)$ can
for $h\neq 0$ similarly be written out as
\begin{subequations}
  \label{P2.app:eq:correlation.of.the.artificial.time.series}
  \begin{align}
    \rhoz{\Yz{12}}(h)
    &= \frac{\frac{1}{4}\cdot\sumss{i=1}{r}
      \pz[2]{i}\cdot\parenR{\az[2]{i} + 2\az{i}\bz{i} +
      \bz[2]{i}}\cdot\cos(2\pi\alphaz{i} h + \thetaz{i}) }{\sumss{i=1}{r} \pz{i}\cdot \parenS{\Lz[2]{i} + \frac{1}{3}\cdot
      \parenR{\az[2]{i}+\az{i}\bz{i}+\bz[2]{i}}} -
      \parenRz[2]{}{\sumss{i=1}{r} \pz{i}\cdot\Lz{i}}},\\
    \rhoz{\Yz{21}}(h)
    &= \frac{\frac{1}{4}\cdot\sumss{i=1}{r}
      \pz[2]{i}\cdot\parenR{\az[2]{i} + 2\az{i}\bz{i} +
      \bz[2]{i}}\cdot\cos(2\pi\alphaz{i} h - \thetaz{i}) }{\sumss{i=1}{r} \pz{i}\cdot \parenS{\Lz[2]{i} + \frac{1}{3}\cdot
      \parenR{\az[2]{i}+\az{i}\bz{i}+\bz[2]{i}}} -
      \parenRz[2]{}{\sumss{i=1}{r} \pz{i}\cdot\Lz{i}}},
  \end{align}
\end{subequations}
whereas the $h=0$ case is given by:
\begin{align}
  \label{P2.app:eq:correlation.of.the.artificial.time.series_lag_zero}
  \rhoz{\Yz{12}}(0) = \rhoz{21}(0) 
  &= \frac{\sumss{i=1}{r} \pz{i}\cdot \parenS{\Lz[2]{i} + \frac{1}{3}\cdot
    \parenR{\az[2]{i}+\az{i}\bz{i}+\bz[2]{i}}}\cdot\cos(\thetaz{i}) -
    \parenRz[2]{}{\sumss{i=1}{r} \pz{i}\cdot\Lz{i}}}{\sumss{i=1}{r} \pz{i}\cdot \parenS{\Lz[2]{i} + \frac{1}{3}\cdot
    \parenR{\az[2]{i}+\az{i}\bz{i}+\bz[2]{i}}} -
    \parenRz[2]{}{\sumss{i=1}{r} \pz{i}\cdot\Lz{i}}}.
\end{align}

An inspection of
\cref{P2.app:eq:auto-correlation.of.the.artificial.time.series,P2.app:eq:correlation.of.the.artificial.time.series}
reveals that it is fairly easy to find a parameter configuration for
which the numerator is rather small compared to the denominator.  This
is of course not white noise, but the key idea is that it is close
enough to white noise to make it impossible to deduce anything about
the underlying frequencies $\alphaz{i}$ and the phase-differences
$\thetaz{i}$ based on the ordinary cross-spectrum.

\subsubsection{The heuristic argument underlying the local bivariate
  trigonometric example}
\label{P2.app:fig:trigonometric.heuristic.argument}

The purpose of this section is to give a short summary of the
corresponding discussion from
\lgsdRef{app:fig:trigonometric.heuristic.argument}, i.e., to briefly
explain the reasoning that motivated the construction of the
\textit{local bivariate trigonometric} example in
\cref{P2.app:fig:trigonometric.general.properties}.  Only the general
ideas will be described here, and the interested reader should consult
\lgsdRef{app:fig:trigonometric.heuristic.argument} for an in depth
discussion.

The key motivation for the construction of the \textit{local bivariate
  trigonometric} models is the sanity testing of the estimation
algorithm for the $m$-truncated local Gaussian cross-spectrum
$\hatlgcsdM{k\ell}{\LGp}{\omega}{m}$, i.e., to construct models in
such a manner that heuristic arguments can be used to give a
\enquote{decent guesstimate of the result} when the $m$-truncated
estimates $\hatlgcsdM{k\ell}{\LGp}{\omega}{m}$ are created based on
pseudo-normalised samples from these models.

In order for this to work, it is necessary to select several
parameters in a structured way, both on the side of the \textit{local
  bivariate trigonometric} models (probabilities $\pz{i}$, amplitudes
$\Az{i}$, base-lines $\Lz{i}$, frequencies $\phiz{i}$ and phases
$\thetaz{i}$) and on the side of the local Gaussian estimation
algorithm (truncation level $m$, bandwidth $\bm{b}$ and point of
investigation $\LGp$).

The problem is to figure out how the parameters $\pz{i}$, $\Az{i}$ and
$\Lz{i}$ should \enquote{map to} points $\LGp$ and bandwidths $\bm{b}$
when the samples are pseudo-normalised.  This problem was discussed in
detail in \lgsdRef{app:fig:trigonometric.heuristic.argument}, and this
discussion will not be repeated here.

The following observation is the pivotal one for the heuristic
argument that motivates the sanity testing.  For a \textit{local
  bivariate trigonometric} random variable
$\parenR{\Yz{1,t},\Yz{2,t}}$, as given in
\cref{P2.eq:THE.local.trigonometric.example}, the idea goes as
follows: For specified indices $i$ and $j$, it should be possible to
find a \enquote{good combination} of point $\LGp$ and bandwidth
$\bmbz{}$, such that the estimated $m$-truncated local Gaussian
spectrum $\hatlgcsdM{k\ell}{\LGp}{\omega}{m}$ looks similar to the
spectrum that would occur if only the $\Cz{k,i}(t)$ and
$\Cz{\ell,j}(t)$ components had been used to generate the sample under
investigation.

The argument for this part is actually quite simple: For the
\enquote{correct combination} of $\LGp$ and $\bmbz{}$, the local
kernel function in the estimation algorithm of the local Gaussian
cross-correlations $\lgccr{k\ell}{\LGp}{h}$ will in essence
\enquote{ignore} the contributions from other components than
$\Cz{k,i}(t)$ and $\Cz{\ell,j}(t)$ when the lag $h$ pairs from
$\Yz{k,t}$ and $\Yz{\ell,t}$ are considered (since low weights are
assigned to the lag $h$ pairs that lies \enquote{far away} from the
point $\LGp$).  The lag $h$ pairs that contribute most to the
estimated value $\lgccr{k\ell}{\LGp}{h}$, can be seen as a
\enquote{randomly selected subset} of those lag $h$ pairs that would
have occurred if all the observations had been from the $\Cz{k,i}(t)$
and $\Cz{\ell,j}(t)$ components.

If this \enquote{randomly selected subset} from $\Cz{k,i}(t)$ and
$\Cz{\ell,j}(t)$ is sufficiently large, then it can be expected that
the local Gaussian cross-correlation $\hatlgccr{k\ell}{\LGp}{h}$
estimated from this subset should be \enquote{fairly close to} the one
estimated from the full set.  This argument does not depend on the lag
$h$ value, so the expected result is that the estimates
$\hatlgccr{k\ell}{\LGp}{h}$ all should be \enquote{fairly close to} to
those estimated from the full set.

Since the estimate $\hatlgcsdM{k\ell}{\LGp}{\omega}{m}$ of the local
Gaussian cross-spectrum $\lgcsdM{k\ell}{\LGp}{\omega}{m}$ is a linear
combination of the estimates $\hatlgccr{k\ell}{\LGp}{h}$, it now seems
reasonable to expect that the shape of it should share some clear
similarities with the one based on the full set. This is the reasoning
that enabled \textit{local bivariate trigonometric} models to be used
for the sanity testing of the estimation algorithm for the
$m$-truncated local Gaussian cross-spectrum
$\hatlgcsdM{k\ell}{\LGp}{\omega}{m}$, as seen in
\cref{sec:lgch-some-simulations} in the main part.

Strictly speaking, neither $\fz{k\ell}(\omega)$ nor
$\lgcsd{k\ell}{\LGp}{\omega}$ are well defined for the \textit{local
  bivariate trigonometric} times series, but this is not a problem
when the $m$-truncated spectra $\fz[m]{k\ell}(\omega)$ and
$\lgcsdM{k\ell}{\LGp}{\omega}{m}$ are investigated. The important
detail is that it (for suitably selected combinations of point $\LGp$
and bandwidth $\bmbz{}$) is possible to \enquote{predict} that the
$m$-truncated estimates $\lgcsdM{k\ell}{\LGp}{\omega}{m}$ seen in
\cref{fig:Bivariate_local_trigonometric_A} should look a bit like like
the \enquote{global bivariate cosine combination} seen in
\cref{fig:Bivariate_global_cosine}.

\subsubsection{What points to investigate for the bivariate local
  trigonometric examples?}

The \textit{local bivariate trigonometric} models used for the sanity
testing in the main part used points located along the diagonal, like
seen in \cref{fig:Bivariate_local_trigonometric_A}.  The diagonal
points were selected since it for them was \enquote{fairly clear} that
the resulting local Gaussian spectra would share some similarities
with the basic reference example seen in
\cref{fig:Bivariate_global_cosine}, which enabled an easy visual
comparison of the results.

The heatmap+distance plot seen in
\cref{fig:heatmap_co_quad_dmt_bivariate_constant_phases} also focused
on points along the diagonal, and from this it can be seen that the
spectra will (as expected per construction) change a lot as the point
$\LGp$ varies along the diagonal.  It is of course possible to
consider points outside of the diagonal too, and an example of this
can be seen in \cref{fig:Bivariate_local_trigonometric_B}.  The same
sample that was used for \cref{fig:Bivariate_local_trigonometric_A},
has here been investigated for the three points \texttt{10\%::90\%},
\texttt{10\%::50\%} and \texttt{50\%::90\%} (percentages relates to
quantiles of the standard normal distribution).  This plot is, as was
done for \cref{fig:Bivariate_local_trigonometric_A}, based on a
presentation of the \Co-, \Quad- and \Phase-plots of the $m=10$
truncated local Gaussian cross-spectrum.
\begin{figure}
  
  {\centering \includegraphics[width=1\linewidth]{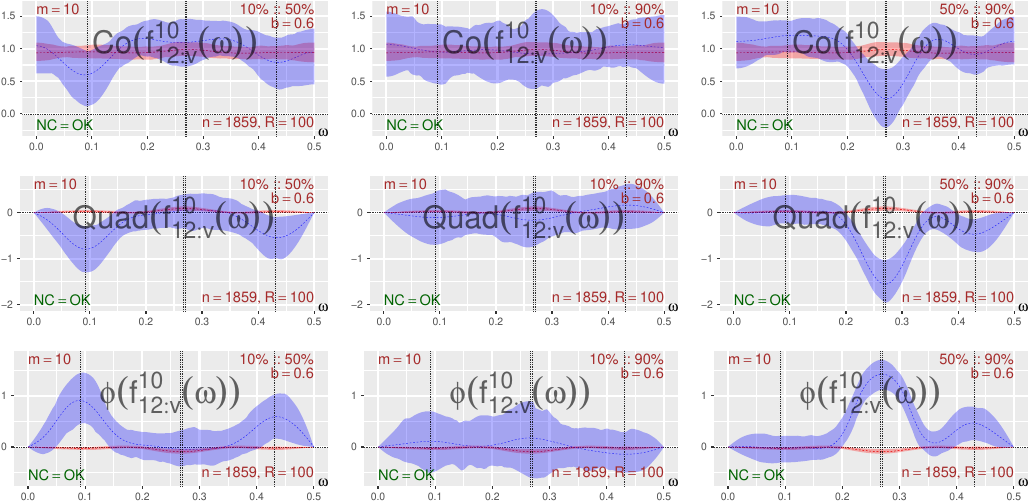} 
  }

  \caption{\Co-, \Quad- and \Phase-plots for points away from the
    diagonal: three diagonal points: The \textit{bivariate local
      trigonometric} model in
    \cref{app.eq:lgch:Y1Y2_local_trigonometric}, constant
    phase-changes $\thetaz{i}=\pi/3, i=1,2,3,4$.
    % shown as sign-adjusted horizontal lines.
    The frequencies $\alphaz{i}$ shown as vertical lines.  The local
    Gaussian spectra detects structures that are not detected by the
    ordinary spectrum.}

  \label{fig:Bivariate_local_trigonometric_B}
\end{figure}

The plots in \cref{fig:Bivariate_local_trigonometric_B} shows that the
\Co-, \Quad- and \Phase-plots at the point \texttt{10\%::90\%} (middle
column) looks more like the i.i.d.\ white noise that was encountered
in~\cref{fig:Bivariate_Gaussian_WN} (see page
\pageref{fig:Bivariate_Gaussian_WN}), whereas the plots for the two
points \texttt{10\%::50\%} and \texttt{50\%::90\%} do detect the
presence of local phenomena.  It might not be any obvious
interpretation of these plots when seen isolated, but it should at
least be noted that the plots for the two points \texttt{10\%::50\%}
and \texttt{50\%::90\%} have troughs/peaks for the $\alpha$-values
that corresponds to the first and second coordinates of these points
--- and this seen in conjunction with the previously investigated
points in \cref{fig:Bivariate_local_trigonometric_A} supports the idea
that there are local features in the data that depends on these
$\alpha$-values.  These features are not detected by the $m=10$
truncated ordinary cross-spectrum $\fz[m]{k\ell}(\omega)$, but they
are detected by the
local Gaussian spectrum $\lgcsdM{k\ell}{\LGp}{\omega}{m}$.

Note that it is possible to create new \enquote{global trigonometric
  examples}, like the one seen in \cref{fig:Bivariate_global_cosine},
in order to see if the figures in
\cref{fig:Bivariate_local_trigonometric_B} are similar to those ---
but that is only part of the work to do.  The bandwidth $\bmbz{}=0.6$
used for he investigation in
\cref{fig:Bivariate_local_trigonometric_A} might not be the
\enquote{correct one} for the points investigated in
\cref{fig:Bivariate_local_trigonometric_B}.  If the bandwidth
$\bmbz{}$ is too big, then \enquote{contamination} might occur, i.e,
the resulting local Gaussian spectra will pick up features from the
other components used in the construction of the \textit{local
  bivariate trigonometric} models, see the discussion in
\lgsdRef{app:fig:trigonometric.heuristic.argument} for further
details.

\subsubsection{Heatmap+distance plots, bivariate local trigonometric and
  extreme tail}
\label{P2.app:extreme_tail_example}
\label{P2.app:fig:trigonometric.C1.component}

Another aspect that can be discussed based on the \textit{local
  bivariate trigonometric} examples is the capability of the local
Gaussian estimation machinery to pick up local structures that occurs
far out in the tails.  This section will mirror the discussion in
\lgsdRef{app:fig:trigonometric.C1.component}, and it will also here be
investigated if a known local structure in the periphery of the data
might be detected if the sample size is large.

The \textit{local bivariate trigonometric} example in this paper is
based on a simple bivariate extension of the corresponding
\textit{local trigonometric} example from \JT.  It contains 4
different components $\parenR{\Cz{1,i}(t),\Cz{2,i}(t)}$,
$i=1,\dotsc,4$, of which the first component
$\parenR{\Cz{1,1}(t),\Cz{2,1}(t)}$ occurs with probability
$\pz{1}=0.05$.  For a sample of length 1859 (length inherited from the
investigation of the \EuStockMarkets-example) it is then expected that
$0.05\cdot 1859 = 92.95$ observations are from this first component
--- and these observations will (by construction) all be in the lower
tail.

As explained in \lgsdRef{app:fig:trigonometric.C1.component}, the
\enquote{border} between the observations from the $i=1$ and $i=2$
components should occur near the 5\% percentile, but it is necessary
to \enquote{zoom in} on a diagonal point $\LGp$ that lies farther out
in the tail than $\pz{1}/2=0.025$.  This requirement occurs since the
estimate should avoid \enquote{contamination} from the observations
from the $i=2$ component.

The parameters selected for the present investigation are all
inherited from the following discussion in
\lgsdRef{app:fig:trigonometric.C1.component}:

\begin{blockquote} %
  {Based on the idea that it might be necessary to go all
    the way out to the 1\%, it seemed natural to attempt an investigation
    based on $n=25000$ observations.  Since the point $\LGp$ now is far
    out in the lower tail, e.g.\ the 1\% percentile of the standard normal
    distribution is $-2.326$, it seemed reasonable to use the bandwidth
    $\bm{b} = (0.4, 0.4)$.

    The heatmap and distance plots in
    \cref{fig:dmt_for_extreme_tail_heatmap_and__levels_vs_norm} is based
    on an investigating of a single realisation, that included
    percentiles based on values starting from 2 bandwidths below the 5\%
    percentile and ending at 1/2 bandwidth below the 5\% percentile,
    i.e.,\ the diagonal points starts at approximately the 0.72\%
    percentile and ends at the 3.25\% percentile.
  }
\end{blockquote}

Note that the reference to
\cref{fig:dmt_for_extreme_tail_heatmap_and__levels_vs_norm} in this
quoted text is to
\lgsdRef{fig:dmt_for_extreme_tail_heatmap_and__levels_vs_norm}.  The
corresponding heatmap and distance plots for the \textit{local
  bivariate trigonometric} example are given in
\cref{fig:bivariate_extreme_tail_Heatmap1} (\Co- and \Quad-version)
and \cref{fig:bivariate_extreme_tail_Heatmap3} (\Amplitude- and
\Phase-version).

\begin{figure}
  {\centering
    \includegraphics[width=1\linewidth]{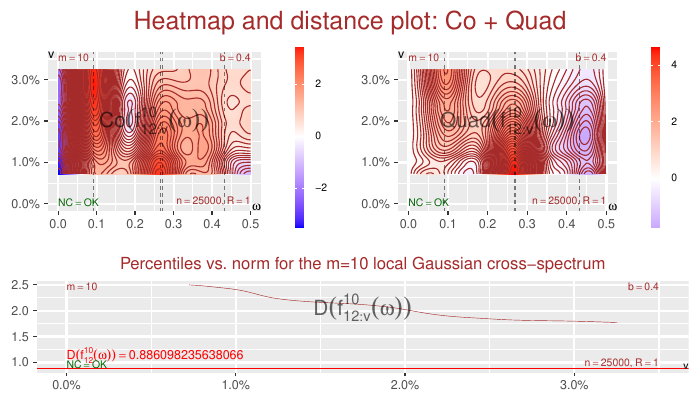}}

  \caption{\Co- and \Quad-heatmaps, distance-plot: Sample from the
    \textit{bivariate local trigonometric} model in
    \cref{eq:lgch:Y1Y2_local_trigonometric,eq:lgch:local_cosines},
    constant phase-changes $\thetaz{i}=\pi/3, i=1,2,3,4$.  The
    frequencies $\alphaz{i}$ shown as vertical lines.  This can be
    used to search for an \enquote{optimal} percentile that can reveal
    the elusive $\alphaz{1}$ frequency in the lower tail.  }
  \label{fig:bivariate_extreme_tail_Heatmap1}
\end{figure}

\begin{figure}
  {\centering
    \includegraphics[width=1\linewidth]{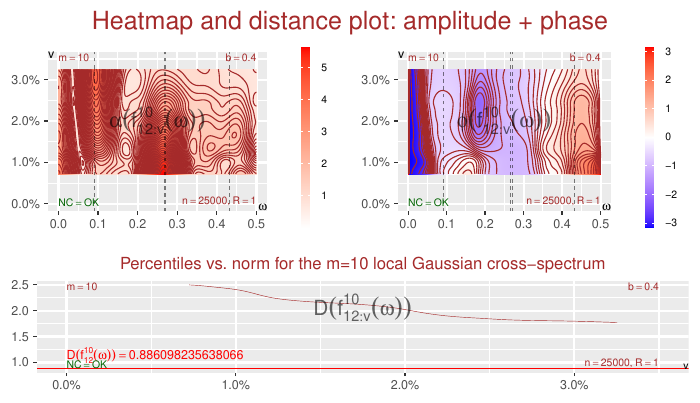}}

  \caption{\Amplitude- and \Phase-heatmaps, distance-plot: Sample from the
    \textit{bivariate local trigonometric} model in
    \cref{eq:lgch:Y1Y2_local_trigonometric,eq:lgch:local_cosines},
    constant phase-changes $\thetaz{i}=\pi/3, i=1,2,3,4$.  The
    frequencies $\alphaz{i}$ shown as vertical lines.  This can be
    used to search for an \enquote{optimal} percentile that can reveal
    the elusive $\alphaz{1}$ frequency in the lower tail.  }

  \label{fig:bivariate_extreme_tail_Heatmap3}
\end{figure}

The heatmap-parts of
\cref{fig:bivariate_extreme_tail_Heatmap1,fig:bivariate_extreme_tail_Heatmap3}
reveal that the contribution from the $i=2$ component completely
dominates at the 3.25\% percentile, and it can also be seen that it is
necessary to go down to at least the 1\% percentile in order to detect
the local structures due to the $i=1$ component (frequency
$\alphaz{1}=0.267$).  Note that
\cref{fig:bivariate_extreme_tail_Heatmap1,fig:bivariate_extreme_tail_Heatmap3}
are based on only 1 single realisation, and other realisations might
look slightly different.

The content of
\cref{fig:bivariate_extreme_tail_Heatmap1,fig:bivariate_extreme_tail_Heatmap3}
are based on estimates $\hatlgccr{k\ell}{\LGp}{h}$ of the local
Gaussian cross-correlations $\lgccr{k\ell}{\LGp}{h}$.  Note that these
local Gaussian estimates focus on points in the periphery of the
observations, and the kernel function in the estimation algorithm
might then, as explained in the sensitivity analysis in
\cref{P2.app:Bandwidth_sensitivity}, give estimates that degenerate
towards $+1$ or $-1$.  The presence/absence of this
degeneration-problem can be investigated using the heatmaps in
\cref{fig:dmt_bivariate_P1_fig_D3}, and the conclusion in this case is
that it seems to be some (but not all) estimates that have degenerated
in this manner.

\begin{figure}
  {\centering
    \includegraphics[width=1\linewidth]{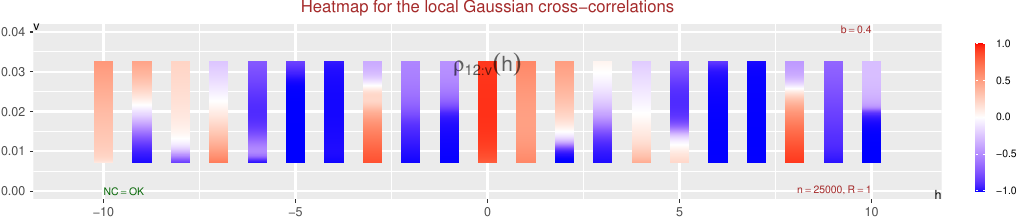}}

  \caption{Heatmap plots for the local Gaussian cross-correlations
    $\lgccr{k\ell}{\LGp}{h}$, for diagonal points $\LGp$ in the lower tail.
    Values close to $+1$ and $-1$ can indicate that the estimates have
    started to degenerate.}
  \label{fig:dmt_bivariate_P1_fig_D3}
\end{figure}

\Cref{fig:dmt_bivariate_extreme_tail_version_of_P2_fig_04} shows the
situation when $R=100$ replicates are used to estimate the local
Gaussian cross-spectrum $\lgcsdM{k\ell}{\LGp}{\omega}{m}$ at the
diagonal point $\LGp$ that corresponds to the 1\% percentile.  The
conclusion in \JT, when a similar investigation was performed for the
univariate \textit{local trigonometric} example, see in particular
\lgsdRef{fig:local.trigonometric.C1-component}, was that the local
Gaussian estimates resulted in a plot with the expected
\enquote{cosine-shape} and a peak near the frequency
$\alphaz{1}=0.267$ of the $i=1$ component.

The situation for the bivariate case investigated in this section, see
\cref{fig:dmt_bivariate_extreme_tail_version_of_P2_fig_04}, diverges
from the one seen for the univariate case.  There is also for this
case a peak near $\alphaz{1}=0.267$ due to the $i=1$ component, but
there is also a similar clear peak near the frequency
$\alphaz{2}=0.091$ due to the $i=2$ component.  This indicates that the
bandwidth $\bmbz{}=0.4$ is too large in this case, since it is a clear
contamination from the $i=2$ component.

\begin{figure}
  {\centering
    \includegraphics[width=1\linewidth]{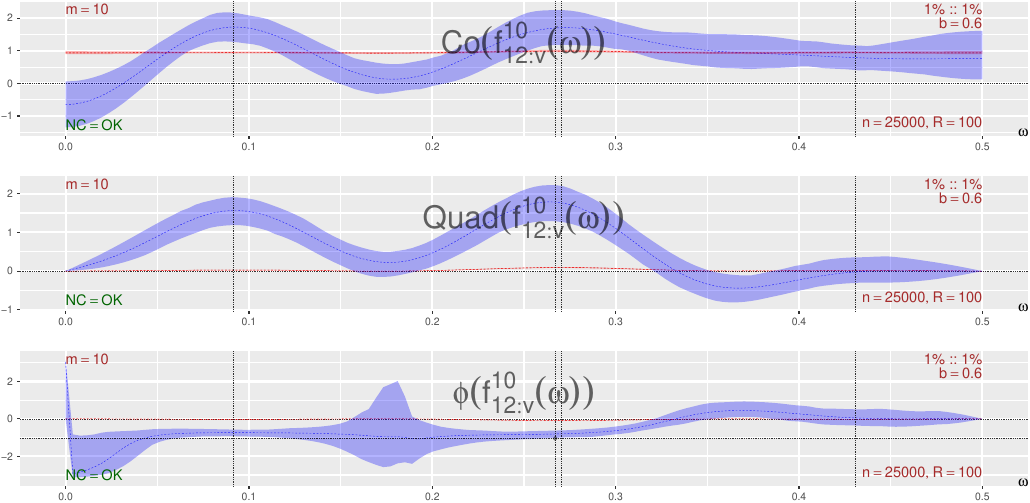}}

  \caption{The detection of the $\alphaz{1}=0.267$ frequency in the
    lower tail of the \textit{local trigonometric} example requires a
    large sample and an investigation far out in the lower tail.  In
    this case the result ended out being contaminated by the
    $\alphaz{2}=0.091$ frequency.}
  \label{fig:dmt_bivariate_extreme_tail_version_of_P2_fig_04}
\end{figure}

The $i=1$ component $\parenR{\Cz{1,1}(t),\Cz{2,1}(t)}$ was included in
the \textit{local bivariate trigonometric} example in order to
emphasise that extra care is needed when investigating points $\LGp$
in the extreme tails of a sample --- and the present discussion
reveals that it even for a known model can be quite tricky to detect
local interdependency-structures, in particular if they do not occur
frequently in the observations.  The size of the sample is pivotal for
the level of details that can be extracted, but the selection of the
bandwidth must also take into account the point $\LGp$ under
investigation.

Keep in mind, as explained in \cref{P2.app:Bandwidth_sensitivity},
that a too small bandwidth $\bmbz{}$ can trigger degenerate estimates
close to $+1$ and $-1$, whereas a too big bandwidth returns estimates
that can be contaminated by neighbouring structures.  A local Gaussian
investigation should thus investigate a range of bandwidths, in order
to see what kind of behaviour that occurs.  The heatmap and distance
plots can be useful for this task.

  \putbib     
\end{bibunit}


\begin{thebibliography}{28}
\expandafter\ifx\csname natexlab\endcsname\relax\def\natexlab#1{#1}\fi
\providecommand{\url}[1]{\texttt{#1}}
\providecommand{\href}[2]{#2}
\providecommand{\path}[1]{#1}
\providecommand{\DOIprefix}{doi:}
\providecommand{\ArXivprefix}{arXiv:}
\providecommand{\URLprefix}{URL: }
\providecommand{\Pubmedprefix}{pmid:}
\providecommand{\doi}[1]{\href{http://dx.doi.org/#1}{\path{#1}}}
\providecommand{\Pubmed}[1]{\href{pmid:#1}{\path{#1}}}
\providecommand{\bibinfo}[2]{#2}
\ifx\xfnm\relax \def\xfnm[#1]{\unskip,\space#1}\fi
%Type = Article
\bibitem[{Barun{\'\i}k and Kley(2019)}]{barunik19:_quant}
\bibinfo{author}{Barun{\'\i}k, J.}, \bibinfo{author}{Kley, T.},
  \bibinfo{year}{2019}.
\newblock \bibinfo{title}{{Quantile coherency: A general measure for dependence
  between cyclical economic variables}}.
\newblock \bibinfo{journal}{The Econometrics Journal} \bibinfo{volume}{22},
  \bibinfo{pages}{131--152}.
\newblock \URLprefix \url{https://doi.org/10.1093/ectj/utz002},
  \DOIprefix\doi{10.1093/ectj/utz002},
  \href{http://arxiv.org/abs/https://academic.oup.com/ectj/article-pdf/22/2/131/37967340/utz002.pdf}{{\tt
  arXiv:https://academic.oup.com/ectj/article-pdf/22/2/131/37967340/utz002.pdf}}.
%Type = Article
\bibitem[{Berentsen et~al.(2014)Berentsen, Kleppe and
  Tj{\o}stheim}]{Berentsen:2014:ILR}
\bibinfo{author}{Berentsen, G.D.}, \bibinfo{author}{Kleppe, T.S.},
  \bibinfo{author}{Tj{\o}stheim, D.B.}, \bibinfo{year}{2014}.
\newblock \bibinfo{title}{{Introducing \texttt{localgauss}, an R Package for
  Estimating and Visualizing Local Gaussian Correlation}}.
\newblock \bibinfo{journal}{j-J-STAT-SOFT} \bibinfo{volume}{56}.
\newblock \URLprefix \url{http://www.jstatsoft.org/v56/i12},
  \DOIprefix\doi{10.18637/jss.v056.i12}.
%Type = Article
\bibitem[{Brillinger(1965)}]{brillinger1965introduction}
\bibinfo{author}{Brillinger, D.R.}, \bibinfo{year}{1965}.
\newblock \bibinfo{title}{An {I}ntroduction to {P}olyspectra}.
\newblock \bibinfo{journal}{The Annals of Mathematical Statistics}
  \bibinfo{volume}{36}, \bibinfo{pages}{1351--1374}.
\newblock \URLprefix \url{http://www.jstor.org/stable/2238424}.
%Type = Book
\bibitem[{Brillinger(1984)}]{Brillinger:1984:CWJa}
\bibinfo{editor}{Brillinger, D.R.} (Ed.), \bibinfo{year}{1984}.
\newblock \bibinfo{title}{The collected works of {John W. Tukey}. {Volume I}.
  {Time} series: 1949--1964}.
\newblock Wadsworth Statistics\slash Probability Series,
  \bibinfo{publisher}{Wadsworth}, \bibinfo{address}{Pacific Grove, CA, USA}.
\newblock \bibinfo{note}{With introductory material by William S. Cleveland and
  Frederick Mosteller.}
%Type = Article
\bibitem[{Brillinger(1991)}]{brillinger1991some}
\bibinfo{author}{Brillinger, D.R.}, \bibinfo{year}{1991}.
\newblock \bibinfo{title}{Some history of the study of higher-order moments and
  spectra}.
\newblock \bibinfo{journal}{Statistica Sinica} \bibinfo{volume}{1},
  \bibinfo{pages}{24J}.
\newblock \URLprefix
  \url{http://www3.stat.sinica.edu.tw/statistica/j1n2/j1n23/..\\j1n210\\j1n210.htm}.
%Type = Book
\bibitem[{Brockwell and Davis(1986)}]{Brockwell:1986:TST:17326}
\bibinfo{author}{Brockwell, P.J.}, \bibinfo{author}{Davis, R.A.},
  \bibinfo{year}{1986}.
\newblock \bibinfo{title}{Time {S}eries: {T}heory and {M}ethods}.
\newblock \bibinfo{publisher}{Springer-Verlag New York, Inc.},
  \bibinfo{address}{New York, NY, USA}.
%Type = Article
\bibitem[{Chen et~al.(2021)Chen, Sun and Li}]{chen2019semi}
\bibinfo{author}{Chen, T.}, \bibinfo{author}{Sun, Y.}, \bibinfo{author}{Li,
  T.H.}, \bibinfo{year}{2021}.
\newblock \bibinfo{title}{A semi-parametric estimation method for the quantile
  spectrum with an application to earthquake classification using convolutional
  neural network}.
\newblock \bibinfo{journal}{Computational Statistics \& Data Analysis}
  \bibinfo{volume}{154}, \bibinfo{pages}{107069}.
\newblock \URLprefix
  \url{https://www.sciencedirect.com/science/article/pii/S0167947320301602},
  \DOIprefix\doi{https://doi.org/10.1016/j.csda.2020.107069}.
%Type = Article
\bibitem[{Chung and Hong(2007)}]{chung07:_model}
\bibinfo{author}{Chung, J.}, \bibinfo{author}{Hong, Y.}, \bibinfo{year}{2007}.
\newblock \bibinfo{title}{Model-free evaluation of directional predictability
  in foreign exchange markets}.
\newblock \bibinfo{journal}{Journal of Applied Econometrics}
  \bibinfo{volume}{22}, \bibinfo{pages}{855--889}.
\newblock \URLprefix \url{http://dx.doi.org/10.1002/jae.965},
  \DOIprefix\doi{10.1002/jae.965}.
%Type = Article
\bibitem[{Ciaburro and
  Iannace(2021)}]{ciaburro21:_machin_learn_based_algor_knowl}
\bibinfo{author}{Ciaburro, G.}, \bibinfo{author}{Iannace, G.},
  \bibinfo{year}{2021}.
\newblock \bibinfo{title}{{Machine Learning-Based Algorithms to Knowledge
  Extraction from Time Series Data: A Review}}.
\newblock \bibinfo{journal}{Data} \bibinfo{volume}{6}.
\newblock \URLprefix \url{https://www.mdpi.com/2306-5729/6/6/55},
  \DOIprefix\doi{10.3390/data6060055}.
%Type = Manual
\bibitem[{Ghalanos(2015)}]{rmgarch_Ghalanos}
\bibinfo{author}{Ghalanos, A.}, \bibinfo{year}{2015}.
\newblock \bibinfo{title}{rmgarch: {M}ultivariate {GARCH} models.}
\newblock \URLprefix \url{https://CRAN.R-project.org/package=rmgarch}.
  \bibinfo{note}{r package version 1.3-0.}
%Type = Manual
\bibitem[{Ghalanos(2019)}]{ghalanos15:_rmgarch}
\bibinfo{author}{Ghalanos, A.}, \bibinfo{year}{2019}.
\newblock \bibinfo{title}{{rmgarch: Multivariate GARCH models.}}
\newblock \URLprefix \url{https://cran.r-project.org/package=rmgarch}.
  \bibinfo{note}{r package version 1.3-7.}
%Type = Article
\bibitem[{Hjort and Jones(1996)}]{hjort96:_local}
\bibinfo{author}{Hjort, N.L.}, \bibinfo{author}{Jones, M.C.},
  \bibinfo{year}{1996}.
\newblock \bibinfo{title}{Locally parametric nonparametric density estimation}.
\newblock \bibinfo{journal}{Ann. Statist.} \bibinfo{volume}{24},
  \bibinfo{pages}{1619--1647}.
\newblock \URLprefix \url{http://dx.doi.org/10.1214/aos/1032298288},
  \DOIprefix\doi{10.1214/aos/1032298288}.
%Type = Article
\bibitem[{Hong(1999)}]{hong1999hypothesis}
\bibinfo{author}{Hong, Y.}, \bibinfo{year}{1999}.
\newblock \bibinfo{title}{{Hypothesis Testing in Time Series via the Empirical
  Characteristic Function: A Generalized Spectral Density Approach}}.
\newblock \bibinfo{journal}{Journal of the American Statistical Association}
  \bibinfo{volume}{94}, \bibinfo{pages}{1201--1220}.
\newblock \URLprefix
  \url{http://tandfonline.com/doi/abs/10.1080/01621459.1999.10473874},
  \DOIprefix\doi{10.1080/01621459.1999.10473874},
  \href{http://arxiv.org/abs/http://tandfonline.com/doi/pdf/10.1080/01621459.1999.10473874}{{\tt
  arXiv:http://tandfonline.com/doi/pdf/10.1080/01621459.1999.10473874}}.
%Type = Article
\bibitem[{Hong et~al.(2007)Hong, Tu and Zhou}]{hong07:_asymm_stock_retur}
\bibinfo{author}{Hong, Y.}, \bibinfo{author}{Tu, J.}, \bibinfo{author}{Zhou,
  G.}, \bibinfo{year}{2007}.
\newblock \bibinfo{title}{{Asymmetries in Stock Returns: Statistical Tests and
  Economic Evaluation}}.
\newblock \bibinfo{journal}{The Review of Financial Studies}
  \bibinfo{volume}{20}, \bibinfo{pages}{1547--1581}.
\newblock \URLprefix \url{http://www.jstor.org/stable/4494812}.
%Type = Article
\bibitem[{Jordanger and Tj{\o}stheim(2022)}]{jordanger17:_lgsd}
\bibinfo{author}{Jordanger, L.A.}, \bibinfo{author}{Tj{\o}stheim, D.},
  \bibinfo{year}{2022}.
\newblock \bibinfo{title}{{Nonlinear Spectral Analysis: A Local Gaussian
  Approach}}.
\newblock \bibinfo{journal}{Journal of the American Statistical Association}
  \bibinfo{volume}{117}, \bibinfo{pages}{1010--1027}.
\newblock \URLprefix \url{https://doi.org/10.1080/01621459.2020.1840991},
  \DOIprefix\doi{10.1080/01621459.2020.1840991},
  \href{http://arxiv.org/abs/https://doi.org/10.1080/01621459.2020.1840991}{{\tt
  arXiv:https://doi.org/10.1080/01621459.2020.1840991}}.
%Type = Article
\bibitem[{Klimko and Nelson(1978)}]{klimko1978}
\bibinfo{author}{Klimko, L.A.}, \bibinfo{author}{Nelson, P.I.},
  \bibinfo{year}{1978}.
\newblock \bibinfo{title}{{On Conditional Least Squares Estimation for
  Stochastic Processes}}.
\newblock \bibinfo{journal}{Ann. Statist.} \bibinfo{volume}{6},
  \bibinfo{pages}{629--642}.
\newblock \URLprefix \url{http://dx.doi.org/10.1214/aos/1176344207},
  \DOIprefix\doi{10.1214/aos/1176344207}.
%Type = Article
\bibitem[{Li(2020)}]{li/zero_crossing_to_QFA}
\bibinfo{author}{Li, T.H.}, \bibinfo{year}{2020}.
\newblock \bibinfo{title}{From zero crossings to quantile-frequency analysis of
  time series with an application to nondestructive evaluation}.
\newblock \bibinfo{journal}{Applied Stochastic Models in Business and Industry}
  \bibinfo{volume}{36}, \bibinfo{pages}{1111--1130}.
\newblock \URLprefix
  \url{https://onlinelibrary.wiley.com/doi/abs/10.1002/asmb.2499},
  \DOIprefix\doi{10.1002/asmb.2499},
  \href{http://arxiv.org/abs/https://onlinelibrary.wiley.com/doi/pdf/10.1002/asmb.2499}{{\tt
  arXiv:https://onlinelibrary.wiley.com/doi/pdf/10.1002/asmb.2499}}.
%Type = Article
\bibitem[{Li(2021)}]{li2021quantile}
\bibinfo{author}{Li, T.H.}, \bibinfo{year}{2021}.
\newblock \bibinfo{title}{Quantile-frequency analysis and spectral measures for
  diagnostic checks of time series with nonlinear dynamics}.
\newblock \bibinfo{journal}{{Journal of the Royal Statistical Society: Series C
  (Applied Statistics)}} \bibinfo{volume}{2}, \bibinfo{pages}{270--290}.
\newblock \URLprefix
  \url{https://rss.onlinelibrary.wiley.com/doi/abs/10.1111/rssc.12458},
  \DOIprefix\doi{10.1111/rssc.12458}.
%Type = Misc
\bibitem[{Li(2022a)}]{li22:_quant_fourier_trans_quant_series}
\bibinfo{author}{Li, T.H.}, \bibinfo{year}{2022}a.
\newblock \bibinfo{title}{{Quantile Fourier Transform, Quantile Series, and
  Nonparametric Estimation of Quantile Spectra}}.
\newblock \URLprefix \url{https://arxiv.org/abs/2211.05844},
  \DOIprefix\doi{10.48550/ARXIV.2211.05844}.
%Type = Unpublished
\bibitem[{Li(2022b)}]{li22:_quant_frequen_analy_deep_learn_signal_class}
\bibinfo{author}{Li, T.H.}, \bibinfo{year}{2022}b.
\newblock \bibinfo{title}{{Quantile-Frequency Analysis and Deep Learning for
  Signal Classification}}.
\newblock \URLprefix
  \url{https://assets.researchsquare.com/files/rs-1855496/v1_covered.pdf},
  \DOIprefix\doi{https://doi.org/10.21203/rs.3.rs-1855496/v1}.
  \bibinfo{note}{preprint}.
%Type = Article
\bibitem[{Li(2023)}]{li2023robust}
\bibinfo{author}{Li, Z.}, \bibinfo{year}{2023}.
\newblock \bibinfo{title}{Robust conditional spectral analysis of replicated
  time series}.
\newblock \bibinfo{journal}{Statistics and Its Interface} \bibinfo{volume}{16},
  \bibinfo{pages}{81--96}.
\newblock \URLprefix
  \url{https://www.intlpress.com/site/pub/pages/journals/items/sii/content/vols/0016/0001/a007/index.php}.
%Type = Article
\bibitem[{Otneim and Tj{\o}stheim(2017)}]{otneim2017locally}
\bibinfo{author}{Otneim, H.}, \bibinfo{author}{Tj{\o}stheim, D.},
  \bibinfo{year}{2017}.
\newblock \bibinfo{title}{{The locally Gaussian density estimator for
  multivariate data}}.
\newblock \bibinfo{journal}{Statistics and Computing} \bibinfo{volume}{27},
  \bibinfo{pages}{1595--1616}.
\newblock \URLprefix \url{https://doi.org/10.1007/s11222-016-9706-6},
  \DOIprefix\doi{10.1007/s11222-016-9706-6}.
%Type = Manual
\bibitem[{{R Core Team}(2020)}]{R_manual}
\bibinfo{author}{{R Core Team}}, \bibinfo{year}{2020}.
\newblock \bibinfo{title}{{R: A Language and Environment for Statistical
  Computing}}.
\newblock \bibinfo{organization}{{R Foundation for Statistical Computing}}.
  \bibinfo{address}{Vienna, Austria}.
\newblock \URLprefix \url{https://www.R-project.org/}.
%Type = Book
\bibitem[{Ter{\"a}svirta et~al.(2010)Ter{\"a}svirta, Tj{\o}stheim and
  Granger}]{terasvirta2010modelling}
\bibinfo{author}{Ter{\"a}svirta, T.}, \bibinfo{author}{Tj{\o}stheim, D.},
  \bibinfo{author}{Granger, C.W.J.}, \bibinfo{year}{2010}.
\newblock \bibinfo{title}{Modelling nonlinear economic time series}.
\newblock \bibinfo{publisher}{Oxford University Press, Oxford}.
%Type = Article
\bibitem[{Tj{\o}stheim and Hufthammer(2013)}]{Tjostheim201333}
\bibinfo{author}{Tj{\o}stheim, D.}, \bibinfo{author}{Hufthammer, K.O.},
  \bibinfo{year}{2013}.
\newblock \bibinfo{title}{{Local Gaussian correlation: A new measure of
  dependence}}.
\newblock \bibinfo{journal}{Journal of Econometrics} \bibinfo{volume}{172},
  \bibinfo{pages}{33 -- 48}.
\newblock \URLprefix
  \url{http://www.sciencedirect.com/science/article/pii/S0304407612001741},
  \DOIprefix\doi{10.1016/j.jeconom.2012.08.001}.
%Type = Book
\bibitem[{Tj{\o}stheim et~al.(2021)Tj{\o}stheim, Otneim and
  St{\o}ve}]{tjostheim2021statistical}
\bibinfo{author}{Tj{\o}stheim, D.}, \bibinfo{author}{Otneim, H.},
  \bibinfo{author}{St{\o}ve, B.}, \bibinfo{year}{2021}.
\newblock \bibinfo{title}{{Statistical Modeling Using Local Gaussian
  Approximation}}.
\newblock \bibinfo{publisher}{Academic Press}.
%Type = Inproceedings
\bibitem[{Tukey(1959)}]{Tukey:1959:IMS}
\bibinfo{author}{Tukey, J.W.}, \bibinfo{year}{1959}.
\newblock \bibinfo{title}{An introduction to the measurement of spectra}, in:
  \bibinfo{editor}{Grenander, U.} (Ed.), \bibinfo{booktitle}{Probability and
  Statistics, The {Harald Cram{\'e}r} Volume}, \bibinfo{publisher}{Almqvist and
  Wiksell}, \bibinfo{address}{Stockholm, Sweden}. pp.
  \bibinfo{pages}{300--330}.
%Type = Article
\bibitem[{Zhao et~al.(2022)Zhao, Shi and
  Zhang}]{zhao22:_model_multiv_time_series_with}
\bibinfo{author}{Zhao, Z.}, \bibinfo{author}{Shi, P.}, \bibinfo{author}{Zhang,
  Z.}, \bibinfo{year}{2022}.
\newblock \bibinfo{title}{Modeling multivariate time series with copula-linked
  univariate d-vines}.
\newblock \bibinfo{journal}{{Journal of Business \& Economic Statistics}}
  \bibinfo{volume}{40}, \bibinfo{pages}{690--704}.
\newblock \URLprefix \url{https://doi.org/10.1080/07350015.2020.1859381},
  \DOIprefix\doi{10.1080/07350015.2020.1859381},
  \href{http://arxiv.org/abs/https://doi.org/10.1080/07350015.2020.1859381}{{\tt
  arXiv:https://doi.org/10.1080/07350015.2020.1859381}}.

\end{thebibliography}


\begin{thebibliography}{10}
\expandafter\ifx\csname natexlab\endcsname\relax\def\natexlab#1{#1}\fi
\providecommand{\url}[1]{\texttt{#1}}
\providecommand{\href}[2]{#2}
\providecommand{\path}[1]{#1}
\providecommand{\DOIprefix}{doi:}
\providecommand{\ArXivprefix}{arXiv:}
\providecommand{\URLprefix}{URL: }
\providecommand{\Pubmedprefix}{pmid:}
\providecommand{\doi}[1]{\href{http://dx.doi.org/#1}{\path{#1}}}
\providecommand{\Pubmed}[1]{\href{pmid:#1}{\path{#1}}}
\providecommand{\bibinfo}[2]{#2}
\ifx\xfnm\relax \def\xfnm[#1]{\unskip,\space#1}\fi
%Type = Article
\bibitem[{Birr et~al.(2019)Birr, Kley and Volgushev}]{BIRR2019122}
\bibinfo{author}{Birr, S.}, \bibinfo{author}{Kley, T.},
  \bibinfo{author}{Volgushev, S.}, \bibinfo{year}{2019}.
\newblock \bibinfo{title}{{Model assessment for time series dynamics using
  copula spectral densities: A graphical tool}}.
\newblock \bibinfo{journal}{{Journal of Multivariate Analysis}}
  \bibinfo{volume}{172}, \bibinfo{pages}{122 -- 146}.
\newblock \URLprefix
  \url{http://www.sciencedirect.com/science/article/pii/S0047259X18301842},
  \DOIprefix\doi{10.1016/j.jmva.2019.03.003}. \bibinfo{note}{{Dependence
  Models}}.
%Type = Article
\bibitem[{Bollerslev and Ghysels(1996)}]{bollerslev96:_period}
\bibinfo{author}{Bollerslev, T.}, \bibinfo{author}{Ghysels, E.},
  \bibinfo{year}{1996}.
\newblock \bibinfo{title}{{Periodic Autoregressive Conditional
  Heteroscedasticity}}.
\newblock \bibinfo{journal}{Journal of Business \& Economic Statistics}
  \bibinfo{volume}{14}, \bibinfo{pages}{139--151}.
\newblock \URLprefix
  \url{http://amstat.tandfonline.com/doi/abs/10.1080/07350015.1996.10524640},
  \DOIprefix\doi{10.1080/07350015.1996.10524640},
  \href{http://arxiv.org/abs/http://amstat.tandfonline.com/doi/pdf/10.1080/07350015.1996.10524640}{{\tt
  arXiv:http://amstat.tandfonline.com/doi/pdf/10.1080/07350015.1996.10524640}}.
%Type = Book
\bibitem[{Brockwell and Davis(1986)}]{Brockwell:1986:TST:17326}
\bibinfo{author}{Brockwell, P.J.}, \bibinfo{author}{Davis, R.A.},
  \bibinfo{year}{1986}.
\newblock \bibinfo{title}{Time {S}eries: {T}heory and {M}ethods}.
\newblock \bibinfo{publisher}{Springer-Verlag New York, Inc.},
  \bibinfo{address}{New York, NY, USA}.
%Type = Manual
\bibitem[{Ghalanos(2019)}]{ghalanos15:_rmgarch}
\bibinfo{author}{Ghalanos, A.}, \bibinfo{year}{2019}.
\newblock \bibinfo{title}{{rmgarch: Multivariate GARCH models.}}
\newblock \URLprefix \url{https://cran.r-project.org/package=rmgarch}.
  \bibinfo{note}{r package version 1.3-7.}
%Type = Manual
\bibitem[{Ghalanos(2020)}]{ghalanos15:_rugarch}
\bibinfo{author}{Ghalanos, A.}, \bibinfo{year}{2020}.
\newblock \bibinfo{title}{Introduction to the rugarch package (Version 1.4-4)}.
\newblock \URLprefix
  \url{https://cran.r-project.org/web/packages/rugarch/vignettes/Introduction_to_the_rugarch_package.pdf}.
  \bibinfo{note}{r package version 1.4-4.}
%Type = Article
\bibitem[{Jordanger and Tj{\o}stheim(2022)}]{jordanger17:_lgsd}
\bibinfo{author}{Jordanger, L.A.}, \bibinfo{author}{Tj{\o}stheim, D.},
  \bibinfo{year}{2022}.
\newblock \bibinfo{title}{{Nonlinear Spectral Analysis: A Local Gaussian
  Approach}}.
\newblock \bibinfo{journal}{Journal of the American Statistical Association}
  \bibinfo{volume}{117}, \bibinfo{pages}{1010--1027}.
\newblock \URLprefix \url{https://doi.org/10.1080/01621459.2020.1840991},
  \DOIprefix\doi{10.1080/01621459.2020.1840991},
  \href{http://arxiv.org/abs/https://doi.org/10.1080/01621459.2020.1840991}{{\tt
  arXiv:https://doi.org/10.1080/01621459.2020.1840991}}.
%Type = Article
\bibitem[{Klimko and Nelson(1978)}]{klimko1978}
\bibinfo{author}{Klimko, L.A.}, \bibinfo{author}{Nelson, P.I.},
  \bibinfo{year}{1978}.
\newblock \bibinfo{title}{{On Conditional Least Squares Estimation for
  Stochastic Processes}}.
\newblock \bibinfo{journal}{Ann. Statist.} \bibinfo{volume}{6},
  \bibinfo{pages}{629--642}.
\newblock \URLprefix \url{http://dx.doi.org/10.1214/aos/1176344207},
  \DOIprefix\doi{10.1214/aos/1176344207}.
%Type = Article
\bibitem[{K{\"u}nsch(1989)}]{kuensch89:_jackk_boots_gener_station_obser}
\bibinfo{author}{K{\"u}nsch, H.R.}, \bibinfo{year}{1989}.
\newblock \bibinfo{title}{{The Jackknife and the Bootstrap for General
  Stationary Observations}}.
\newblock \bibinfo{journal}{The Annals of Statistics} \bibinfo{volume}{17},
  \bibinfo{pages}{1217--1241}.
\newblock \URLprefix \url{http://www.jstor.org/stable/2241719},
  \DOIprefix\doi{10.1214/aos/1176347265}.
%Type = Article
\bibitem[{Politis and Romano(1992)}]{politis92:_gener_resam_schem_trian_array}
\bibinfo{author}{Politis, D.N.}, \bibinfo{author}{Romano, J.P.},
  \bibinfo{year}{1992}.
\newblock \bibinfo{title}{{A General Resampling Scheme for Triangular Arrays of
  $\alpha$-Mixing Random Variables with Application to the Problem of Spectral
  Density Estimation}}.
\newblock \bibinfo{journal}{The Annals of Statistics} \bibinfo{volume}{20},
  \bibinfo{pages}{1985--2007}.
\newblock \URLprefix \url{https://www.jstor.org/stable/2242377},
  \DOIprefix\doi{10.1214/aos/1176348899}.
%Type = Article
\bibitem[{Tj{\o}stheim and Hufthammer(2013)}]{Tjostheim201333}
\bibinfo{author}{Tj{\o}stheim, D.}, \bibinfo{author}{Hufthammer, K.O.},
  \bibinfo{year}{2013}.
\newblock \bibinfo{title}{{Local Gaussian correlation: A new measure of
  dependence}}.
\newblock \bibinfo{journal}{Journal of Econometrics} \bibinfo{volume}{172},
  \bibinfo{pages}{33 -- 48}.
\newblock \URLprefix
  \url{http://www.sciencedirect.com/science/article/pii/S0304407612001741},
  \DOIprefix\doi{10.1016/j.jeconom.2012.08.001}.

\end{thebibliography}
\end{document}